\documentclass[%
pagesize,twoside,%
titlepage=false,%
toc=listof,toc=bibliography,%
chapterprefix=true,%
bibliography=totocnumbered%
]{scrbook}

\usepackage[english]{babel}

\usepackage{amsmath,amssymb,amsthm,amsfonts}
\usepackage{mathtools}
\usepackage{exscale}

\usepackage{bbm}
\usepackage{mathrsfs} 
\usepackage[T1]{fontenc}
\usepackage{lmodern}

\usepackage{graphicx}

\usepackage{slashed}
\usepackage{shuffle}

\usepackage[binary-units]{siunitx}

\usepackage{ifthen}

\usepackage{listings}
\usepackage{xcolor}

\usepackage{numprint}
\usepackage{geometry}

\usepackage[norefpage,intoc]{nomencl}

\makenomenclature

\usepackage[pageanchor]{hyperref}

\usepackage{tabularx}
\usepackage{booktabs}
\usepackage{subfig}

\numberwithin{equation}{section}

\usepackage[numbers,sort&compress]{natbib}
\definecolor{mapleinput}{rgb}{0.5,0.0,0.0}
\definecolor{maplemath}{rgb}{0.0,0.0,1.0}
\definecolor{maplewarning}{cmyk}{0.0,1.0,.0.0,0.0}

\lstnewenvironment{MapleInput}{%
\lstset{
	basicstyle=\bfseries\upshape\ttfamily\color{mapleinput},
	numberstyle=\bfseries\upshape\ttfamily\color{mapleinput},
	breaklines=true,
	columns=flexible,
	breakautoindent=false,
	breakindent=0pt,
	xleftmargin=4ex,
	numbers=left,
	stepnumber=1,
	numbersep=1.3ex,
	aboveskip=\smallskipamount,
	belowskip=\smallskipamount,
}%
}{%
}

\lstnewenvironment{MapleError}{%
\lstset{
	basicstyle=\upshape\ttfamily\color{maplewarning},
	breaklines=true,
	columns=flexible,
	breakautoindent=false,
	breakindent=0pt,
	xleftmargin=2ex,
	numbers=none,
	aboveskip=0pt,
	belowskip=\smallskipamount,
}%
}{%
}

\lstnewenvironment{MapleOutput}{%
\lstset{
	basicstyle=\upshape\ttfamily\color{maplemath},
	breaklines=true,
	columns=flexible,
	breakautoindent=false,
	breakindent=0pt,
	xleftmargin=2ex,
	numbers=none,
	aboveskip=0pt,
	belowskip=\smallskipamount,
}%
}{%
}

\newenvironment{MapleMath}{%
\color{maplemath}\upshape\rmfamily%
\setlength{\abovedisplayskip}{0ex}%
\setlength{\abovedisplayshortskip}{\abovedisplayskip}%
\setlength{\belowdisplayskip}{\medskipamount}%
\setlength{\belowdisplayshortskip}{0ex}%
\csname gather*\endcsname}{\csname endgather*\endcsname%
{\hrule height 0pt}%
\ignorespacesafterend}

\newcommand{\Hlog}[2]{\operatorname{Hlog}\left( #1, \left[ #2 \right] \right)}
\newcommand{\Mpl}[2]{\operatorname{Mpl}\left( \left[ #1 \right], \left[ #2 \right] \right)}

\newcommand{\K}{
\mathbbm{K}
}
\newcommand{\Q}{
\mathbbm{Q}
}
\newcommand{\N}{
\mathbbm{N}
}
\newcommand{\OddN}{
\mathbbm{O}
}
\newcommand{\Z}{
\mathbbm{Z}
}
\newcommand{\R}{
\mathbbm{R}
}
\newcommand{\C}{
\mathbbm{C}
}
\newcommand{\F}{
\mathbbm{F}
}
\newcommand{\Halfplane}{
\mathbbm{H}%
}
\newcommand{\RP}{
\mathbbm{RP}%
}
\newcommand{\CP}{
\mathbbm{CP}%
}
\newcommand{\Projective}{
\mathbbm{P}%
}

\newcommand{\prdu}[1]{#1'}

\newcommand{\Kilo}{\,\mathrm{k}}
\newcommand{\Mega}{\,\mathrm{M}}

\newcommand{\homotop}{\simeq}

\newcommand{\SymDiff}{\bigtriangleup}

\newcommand{\Transpose}{\intercal}

\newcommand{\Catalan}{G}

\newcommand{\RSphere}{\widehat{\C}}

\newcommand{\tp}{
\otimes
}

\newcommand{\imag}{i}

\newlength{\tempLength}

\newlength{\wurelwidth}
\newcommand{\urel}[2][=]{\mathrel{\mathop{#1}\limits_{\makebox[2.8ex]{\clap{\scalebox{0.5}{#2}}}}}}
\newcommand{\wurel}[2][=]{\mathrel{\mathop{#1}\limits_{\!\scalebox{0.5}{\makebox[\the\wurelwidth]{#2}}\!}}}

\newcommand{\conjugate}[1]{#1^{\ast}}

\newcommand{\bigo}[1]{\mathcal{O}\left(#1\right)}

\newcommand{\dd}[1][]{\mathrm{d}^{#1}}

\newcommand{\cupdot}{\mathbin{\dot{\cup}}}
\newcommand{\bigcupdot}{\mathbin{\dot{\bigcup}}}

\newcommand{\restrict}[2]{%
{\left. #1 \right|}_{#2}%
}

\DeclareMathOperator{\dist}{dist}

\DeclareMathOperator*{\Res}{Res}

\DeclareMathOperator{\lin}{lin}
\DeclareMathOperator{\Aut}{Aut}

\DeclareMathOperator{\im}{im}
\DeclareMathOperator{\Imaginaerteil}{Im}
\DeclareMathOperator{\Realteil}{Re}
\DeclareMathOperator{\sgn}{sgn}
\DeclareMathOperator{\supp}{supp}

\newcommand{\id}{
\mathrm{id}
}

\newcommand{\1}{
\mathbbm{1}
}

\newcommand{\isomorph}{
\cong
}
\newcommand{\defas}{
\mathrel{\mathop:}=
}
\newcommand{\safed}{
=\mathrel{\mathop:}
}
\newcommand{\gdw}{
\ensuremath{\Leftrightarrow}
}

\newcommand{\set}[1]{
\left\{ #1 \right\}
}
\newcommand{\setexp}[2]{%
\left\{ #1\!:\ #2 \right\}
}
\newcommand{\tsetexp}[2]{%
\big\{ #1\!:\ #2 \big\}
}
\newcommand{\abs}[1]{%
\left\lvert #1 \right\rvert
}
\newcommand{\norm}[1]{%
\left\lVert #1 \right\rVert
}

\newcommand{\gf}[1][]{\hyperref[eq:position-space-graphical-function]{f_{#1}}}

\newcommand{\Gscheme}{G_0}

\newcommand{\source}{\partial^{-}}
\newcommand{\target}{\partial^{+}}

\newcommand{\posprop}[1][]{\Delta\ifthenelse{\equal{#1}{}}{}{^{(#1)}}}

\newcommand{\FR}{\hyperref[eq:feynman-integral-momentum]{\Phi}}

\newcommand{\onemaster}[2]{\hyperref[eq:oneloop-master]{L\left( #1, #2 \right)}}
\newcommand{\FRren}{\Phi_+}
\newcommand{\FRdim}{\Phi_{\varepsilon}}

\newcommand{\scalelog}{\ell}
\newcommand{\period}{\hyperref[def:period]{\mathcal{P}}}

\newcommand{\FeynHopf}{%
\mathcal{H}%
}
\newcommand{\FeynGraphs}{%
\mathcal{G}%
}

\newcommand{\ExMom}{p}

\newcommand{\contract}{\mathop{/}}

\newcommand{\forests}{\mathcal{F}}

\newcommand{\convolution}{\star}
\newcommand{\counit}{\varepsilon}
\newcommand{\antipode}{S}

\newcommand{\cop}[1][]{\Delta\ifthenelse{\equal{#1}{}}{}{^{(#1)}}}
\newcommand{\copred}[1][]{\widetilde{\Delta}\ifthenelse{\equal{#1}{}}{}{^{(#1)}}}

\newcommand{\alg}[1][A]{\mathcal{#1}}

\newcommand{\hide}[1]{}

\newcommand{\dimension}{D}

\newcommand{\fieldphi}{\phi}

\newcommand{\sdd}{\omega}
\newcommand{\loops}[1]{\hyperref[eq:loop-number]{h_1}\!\left( #1 \right)}

\newcommand{\vertices}{V}
\newcommand{\actives}{A} 
\newcommand{\comps}{\pi_0}
\newcommand{\edges}{E}

\newcommand{\Kinematics}{\hyperlink{eqkinematics}{\Theta}}
\newcommand{\ScaleVec}{\varrho}

\newcommand{\Tint}{\text{\upshape{int}}}
\newcommand{\Text}{\text{\upshape{ext}}}

\newcommand{\Graph}[2][1.0]{%
\vcenter{\hbox{\includegraphics[scale=#1]{Graphs/#2}}}%
}
\newcommand{\CloseProp}[1]{\widehat{#1}}
\newcommand{\GComp}[1]{\widehat{#1}}

\DeclareMathOperator{\vw}{\hyperref[def:vw]{vw}}
\newcommand{\vconn}{\kappa}

\newcommand{\LI}{\mu}
\newcommand{\AUX}{\xi}
\newcommand{\SP}{\alpha}
\newcommand{\EP}{a}
	\newcommand{\EPZ}{n}
	\newcommand{\EPE}{\nu}
\newcommand{\IM}{
\newcommand{\LM}{\hyperref[eq:laplace-matrix]{\mathcal{L}}}
\newcommand{\GM}{\hyperref[eq:graph-matrix]{M}}
\newcommand{\DM}{\hyperref[eq:graph-matrix]{\Lambda}}
\DeclareMathOperator{\diag}{diag}

\newcommand{\ZZ}[1]{\hyperref[fig:zigzags]{\mathord{\text{\upshape{ZZ}}}_{#1}}}
\newcommand{\WS}[1]{\mathord{\text{\upshape{WS}}}_{#1}}

\newcommand{\BBr}[2]{\hyperref[fig:bubble-graphs-def]{B_{#1,#2}}}
\newcommand{\BBt}[2]{\hyperref[fig:bubble-graphs-def]{\hat{B}_{#1,#2}}}
\newcommand{\smallbubble}{\Graph[0.2]{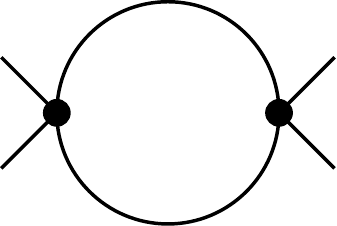}}

\newcommand{\anapartial}[1]{\mathcal{D}_{#1}}
\newcommand{\Pscale}{\lambda}

\newcommand{\phipsi}{\raisebox{1mm}{\scalebox{0.7}{{\raisebox{0.7mm}{$\varphi$}\hspace{-3.15mm}\rotatebox{180}{$\psi$}}}}}
\newcommand{\phipol}{
\newcommand{\psipol}{\psi}
\newcommand{\dodgson}{\hyperref[def:dodgson]{\Psi}}
\newcommand{\fiveinv}[1]{{^{5}\Psi\left(#1\right)}}
\newcommand{\forestpolynom}[2][]{\hyperref[eq:def-forestpolynom]{\Phi_{#1}^{#2}}}
\newcommand{\forestpolynomDual}[2][]{\hyperref[eq:def-forestpolynom]{\widehat{\Phi}_{#1}^{#2}}}

\newcommand{\minoreq}{\preceq}

\newcommand{\subdiv}{\prec}
\newcommand{\subdiveq}{\preceq}
\newcommand{\supdiv}{\succ}
\newcommand{\supdiveq}{\succeq}
\newcommand{\nosubdiv}{\nprec}
\newcommand{\nosubdiveq}{\npreceq}
\newcommand{\nosupdiv}{\nsucc}
\newcommand{\nosupdiveq}{\nsucceq}

\newcommand{\GfunStar}[1]{\hyperref[def:star-triangle-functions]{f^{\StarSymbol}_{#1}}}
\newcommand{\GfunTriangle}[1]{\hyperref[def:star-triangle-functions]{f^{\TriangleSymbol}_{#1}}}
\newcommand{\GfunForest}[1]{\hyperref[eq:def-GfunForest]{f^{\ForestDeltaSymbol}_{#1}}}
\newcommand{\SunrisePsi}{\hyperref[eq:SunrisePsi]{%
\psipol_{\includegraphics[width=1ex]{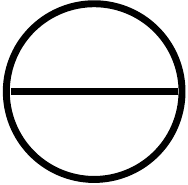}}%
}}
\newcommand{\BoxPoly}{\hyperref[eq:def-BoxPoly]{Q}}
\newcommand{\GfunForestBox}[1]{\hyperref[eq:def-GfunForestBox]{f^{\ForestDeltaBoxSymbol}_{#1}}}
\newcommand{\StarSymbol}{%
\includegraphics[scale=0.28]{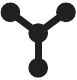}%
}
\newcommand{\TriangleSymbol}{%
\includegraphics[scale=0.28]{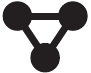}%
}
\newcommand{\ForestDeltaSymbol}{%
\includegraphics[scale=0.28]{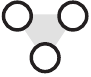}%
}
\newcommand{\ForestDeltaBoxSymbol}{%
\includegraphics[scale=0.28]{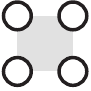}%
}

\newcommand{\Lyndons}{\hyperref[eq:def-lyndon-words]{\operatorname{Lyn}}}

\newcommand{\Hyper}[1]{\hyperref[def:Hlog]{L}_{#1}}
\newcommand{\HyperLappo}[3]{L_{#1}\left( #2 | #3 \right)}
\newcommand{\BlochWigner}{\hyperref[eq:BlochWigner]{D_2}}
\DeclareMathOperator{\Li}{\hyperref[eq:def-MPL-intro]{Li}}
\DeclareMathOperator{\Glaisher}{Gl}

\newcommand{\mzv}[2][]{\hyperref[eq:def-MZV-intro]{\zeta^{#1}_{#2}}}
\newcommand{\MZV}[1][]{\hyperref[def:MZV-N]{\mathcal{Z}\ifthenelse{\equal{#1}{}}{}{^{(#1)}}}}
\newcommand{\Deligne}{\hyperref[theorem:Deligne-N=6]{\mathcal{Z}^{(6)}_D}}

\newcommand{\Moduli}[2]{\hyperlink{eqmodulispace}{\mathfrak{M}_{#1,#2}}}

\newcommand{\WeightDepth}[2]{\mathcal{W}_{#1}^{#2}}

\newcommand{\concat}{\star}
\newcommand{\Shuffles}[2]{S_{#1,#2}}

\newcommand{\letter}[1]{\omega_{#1}}

\newcommand{\AnaReg}[2]{\hyperref[eq:def-reglim]{\Reglim_{#1 \rightarrow #2}}}
\newcommand{\WordReg}[2]{\hyperref[def:shuffle-regularization]{\WordReglim\nolimits_{#1}^{#2}}}
\newcommand{\ReglimWord}[2]{\hyperref[eq:def-reglim-word]{\WordReglim_{#1 \rightarrow #2}}}
\DeclareMathOperator*{\Reglim}{Reg}
\DeclareMathOperator*{\WordReglim}{reg}

\newcommand{\ProjectOn}[1]{P_{#1}}

\newcommand{\WordTransformation}[1]{\hyperref[eq:Moebius-transformation]{\Phi}_{#1}}
\newcommand{\HlogProjection}[1]{#1^{*}}
\newcommand{\emptyWord}{1}
\DeclareMathOperator{\leadCoeff}{\hyperref[def:letter-leading-coefficient]{lead}}
\DeclareMathOperator{\subleadCoeff}{slead}
\DeclareMathOperator{\sdeg}{sdeg}
\renewcommand{\deg}{

\newcommand{\Zsum}[3]{\hyperref[def:zsum]{Z}\left(#1;#2;#3\right)}

\newcommand{\Vanishing}{\mathbf{V}}
\newcommand{\Affine}{\mathbb{A}}
\newcommand{\regulars}{\mathcal{O}}

\newcommand{\BarObjects}{\hyperref[eq:def-bar-objects]{B}}
\newcommand{\BarIntegrals}[1][]{\hyperref[eq:def-barintegrals]{\mathscr{B}_{#1}}}
\newcommand{\BarIntegralsRegulars}[1][]{\hyperlink{inline:barregulars}{\mathscr{B}^{\regulars}_{#1}}}

\newcommand{\HlogAlgebra}{\hyperref[eq:def-HlogAlgebra]{\mathscr{L}}}

\newcommand{\ZerosWrt}[2]{\hyperref[eq:def-zeros-wrt]{\Sigma_{#2}\!\left(#1\right)}}
\newcommand{\resultant}[3]{\left[ #1, #2\right]_{#3}}
\newcommand{\discriminant}[1]{D_{#1}}
\newcommand{\degree}[2][]{\deg_{#1} #2}
\newcommand{\divides}{\ |\ }
\DeclareMathOperator{\rank}{rank}

\newcommand{\PointCount}[2]{%
\big[ X_{#1} \big]_{#2} %
}

\theoremstyle{plain}
\newtheorem{theorem}{Theorem}[section]
\newtheorem{lemma}[theorem]{Lemma}
\newtheorem{corollary}[theorem]{Corollary}
\newtheorem{proposition}[theorem]{Proposition}
\newtheorem{conjecture}[theorem]{Conjecture}

\theoremstyle{definition}
\newtheorem{example}[theorem]{Example}
\newtheorem{definition}[theorem]{Definition}

\theoremstyle{remark}
\newtheorem{remark}[theorem]{Remark}

\newcommand{\Filename}[1]{{\upshape\ttfamily #1}}
\newcommand{\Maple}{\texttt{\textup{Maple}}}
\newcommand{\JaxoDraw}{\texttt{\textup{JaxoDraw}}}
\newcommand{\Axodraw}{\texttt{\textup{Axodraw}}}
\newcommand{\zetaprocedures}{\texttt{\textup{zeta\_procedures}}}
\newcommand{\polylogprocedures}{\texttt{\textup{polylog\_procedures}}}
\newcommand{\HyperInt}{\texttt{\textup{HyperInt}}}

\newcommand{\code}[1]{\mbox{\texttt{\textup{#1}}}}

\newcommand{\MyTitle}{Feynman integrals and hyperlogarithms}
\title{\MyTitle}
\author{Erik Panzer}
\date{\today}

\hypersetup{%
	pdftitle = {\MyTitle},
	pdfauthor = {Erik Panzer}
}

\setcounter{secnumdepth}{2}

\begin{document}

\addtokomafont{chapterprefix}{\raggedleft}
\addtokomafont{chapter}{\rmfamily}
\renewcommand*{\chapterformat}{
	\mbox{\color{gray}\rmfamily\upshape\scalebox{2.0}{\chapappifchapterprefix{\nobreakspace}}%
	\scalebox{3.5}{\thechapter}\enskip}\\[-2pt]
	\noindent{\color{gray}%
\rule{0.3\textwidth}{2.4pt}%
\smash{\rlap{\rule{5pt}{10pt}}}}%
}

\pagenumbering{roman}
\thispagestyle{empty}
{\centering
{\Huge
	\MyTitle
}
\\[\baselineskip]
\rule{0.7\textwidth}{1.4pt}%
\\[\baselineskip]
\href{mailto:erikpanzer@ihes.fr}{Erik Panzer}, \today

}
\vspace{1.0cm}
\noindent This is an updated version of my PhD thesis in mathematics, which I defended at
Humboldt-Universit\"{a}t zu Berlin on the 5th of February, 2015. The referees were
\begin{itemize}
	\item Dr. David Broadhurst (Open University),
	\item Dr. Francis Brown (Institute des Hautes \'{E}tudes Scientifiques) and my supervisor
	\item Prof. Dr. Dirk Kreimer (Humboldt-Universit\"{a}t zu Berlin).
\end{itemize}
Comments are welcome.
\vspace{2mm}
\section*{Abstract}
We study Feynman integrals in the representation with Schwinger parameters and derive recursive integral formulas for massless $3$- and $4$-point functions. Properties of analytic (including dimensional) regularization are summarized and we prove that in the Euclidean region, each Feynman integral can be written as a linear combination of convergent Feynman integrals. This means that one can choose a basis of convergent master integrals and need not evaluate any divergent Feynman graph directly.

Secondly we give a self-contained account of hyperlogarithms and explain in detail the algorithms needed for their application to the evaluation of multivariate integrals. We define a new method to track singularities of such integrals and present a computer program that implements the integration method.

As our main result, we prove the existence of infinite families of massless $3$- and $4$-point graphs (including the ladder box graphs with arbitrary loop number and their minors) whose Feynman integrals can be expressed in terms of multiple polylogarithms, to all orders in the $\varepsilon$-expansion. These integrals can be computed effectively with the presented program.

We include interesting examples of explicit results for Feynman integrals with up to $6$ loops. In particular we present the first exactly computed counterterm in massless $\phi^4$ theory which is not a multiple zeta value, but a linear combination of multiple polylogarithms at primitive sixth roots of unity (and divided by $\sqrt{3}$). To this end we derive a parity result on the reducibility of the real- and imaginary parts of such numbers into products and terms of lower depth.

\raisebox{-1.8cm}[0mm][0mm]{\centering $\smash{\hspace{-6mm}%
	\Graph[0.33]{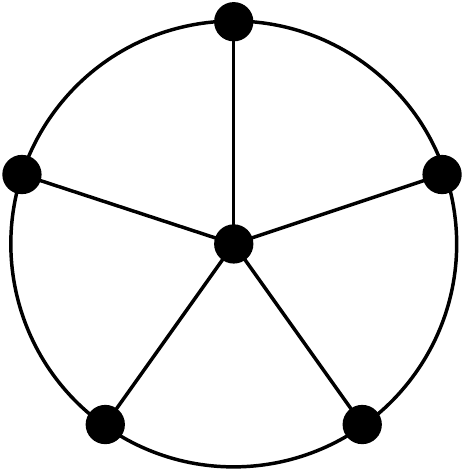}\qquad
	\Graph[0.4]{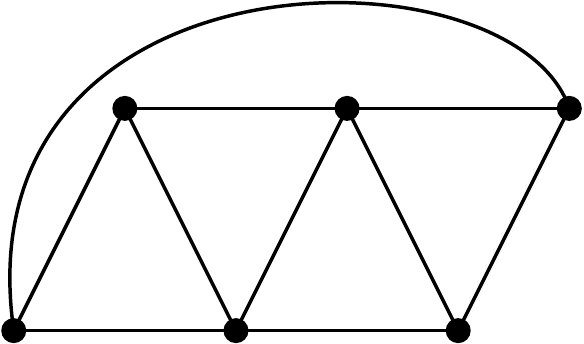}\qquad
	\Graph[0.4]{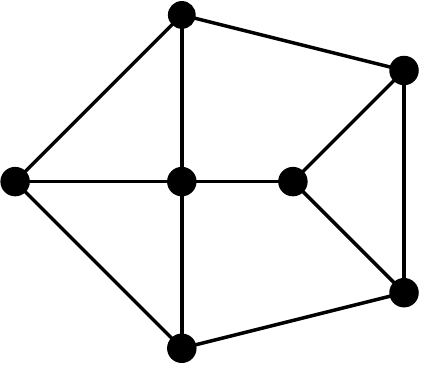}\qquad
	\Graph[0.37]{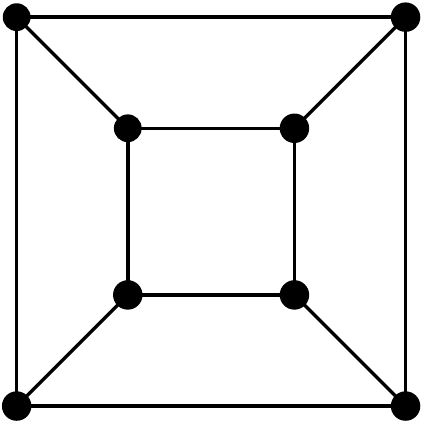}\qquad
	\Graph[0.4]{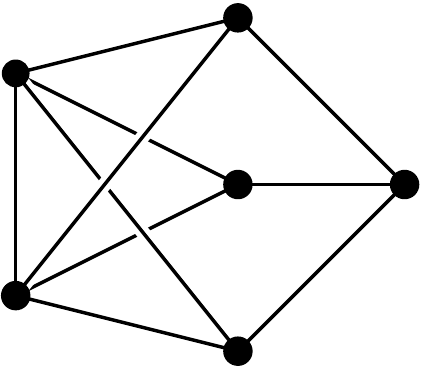}\qquad
	\Graph[0.35]{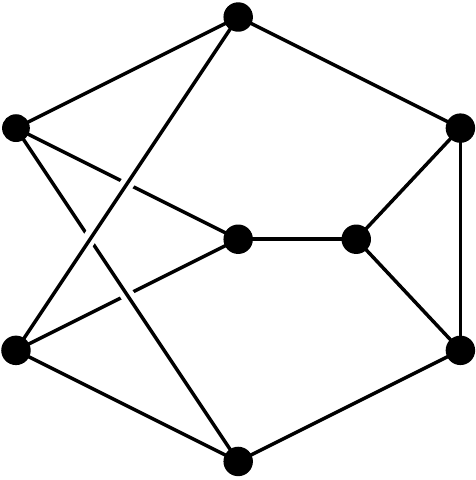}\qquad
	\Graph[0.38]{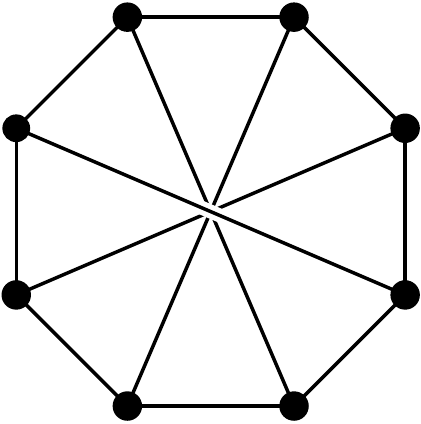}}$}

\clearpage
\thispagestyle{empty}
\begin{center}
\vspace*{\fill}
	\ \clap{$\Graph[2.6]{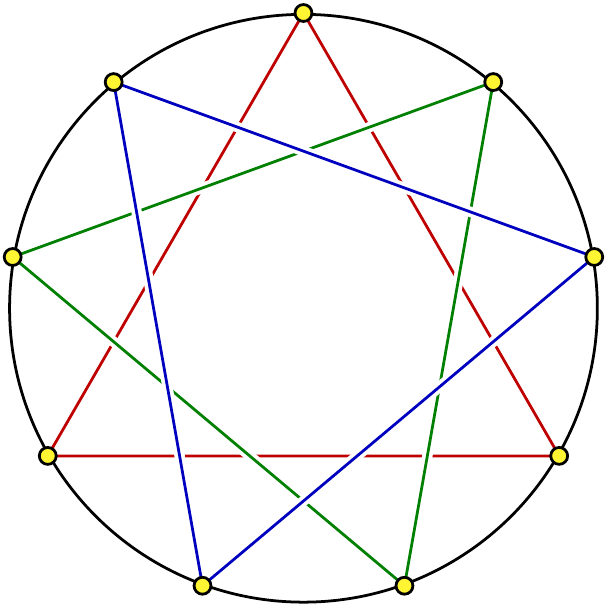}$}
	\clap{$
		\Graph[0.5]{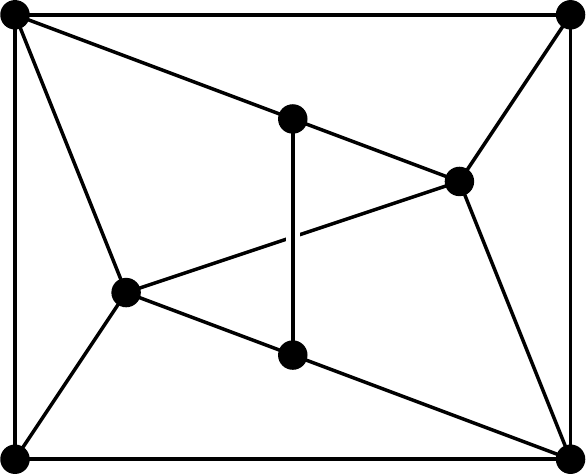}
	$} \ 
\vspace*{\fill}

\raisebox{-1.5cm}[0cm][0cm]{\color{gray}\scalebox{3}{$P_{7,11}$}}
\end{center}

\tableofcontents
\listoffigures

\chapter{Introduction}
\label{chap:intro}%
\pagenumbering{arabic}

\section{Motivation}
This thesis addresses a problem from physics: the computation of Feynman integrals.\footnote{These should not be confused with the completely different \emph{path integrals} that also go back to Feynman.}
These arise in perturbative quantum field theory as contributions to scattering amplitudes, which describe interactions of elementary particles and must be computed in order to predict the cross-sections that can be measured in experiments. Very high accuracies obtained for example at the Large Hadron Collider demand the evaluation of ever more Feynman integrals to assess the validity of the Standard Model.

Such calculations have reached an extreme level of complexity and constantly probe (often exceed) the very edge of knowledge of special functions, analytic methods, algebraic tools, algorithm design and available computational power.
Immense efforts are being invested to overcome these problems and led to impressive progress.
By now, the Feynman integral is appreciated as a rich mathematical structure that interrelates different disciplines such as algebraic geometry, complex analysis and number theory.

A striking feature of all known results for Feynman integrals is the prevalence of multiple polylogarithms and related periods like multiple zeta values, which raises
\begin{itemize}
	\item\textbf{Question 1.} Which Feynman integrals can be expressed in terms of multiple polylogarithms and their special values? How does this property relate to the combinatorial structure of the diagrams?
\end{itemize}
Apart from a huge number of explicit examples, only very little is known for this question on a conceptual level and it seems to be a hard one to answer. In practice however, the pure knowledge of a simple result is not enough, one must actually do the computation.
\begin{itemize}
	\item\textbf{Question 2.} If a Feynman integral does evaluate to multiple polylogarithms, how can it be computed explicitly?
\end{itemize}
In practice, calculations can involve thousands and even more individual Feynman integrals. It is therefore crucial to develop and provide efficient algorithms and programs to compute them in an automatized way.

\begin{figure}
	\centering
	$ \Graph[0.6]{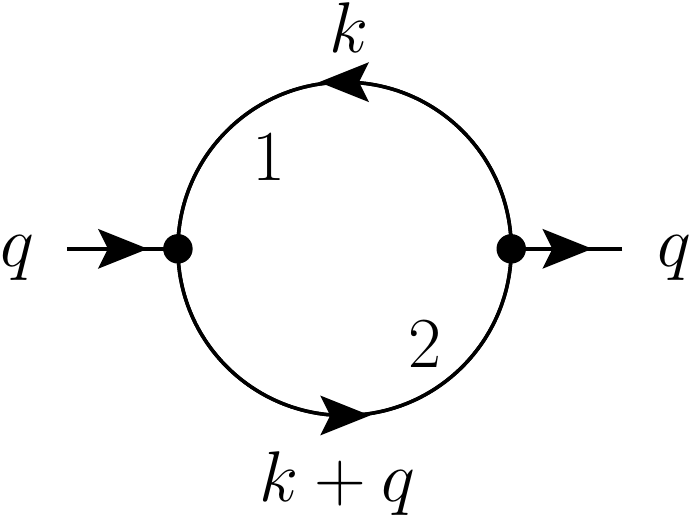} $
	\qquad
	$ F \defas \Graph[0.45]{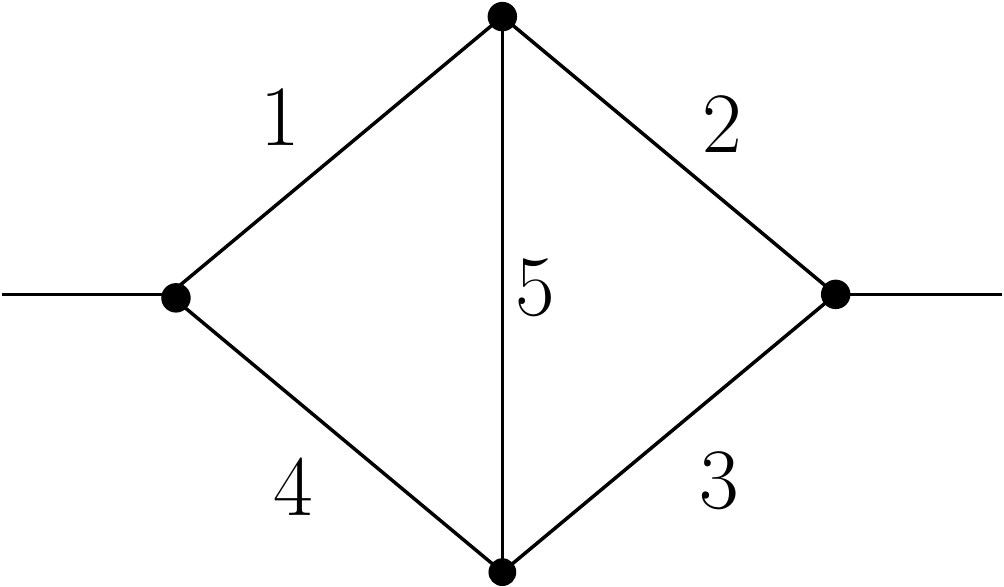} $
	\caption{Propagator diagrams with one and two loops.}%
	\label{fig:low-loop-masters}%
\end{figure}%
\section{Background}
\subsection{Feynman integrals}
Feynman introduced diagrams (graphs) as mnemonics for individual contributions to a perturbation series. Each of them corresponds to an integral determined by the Feynman rules $\FR$. For example, the scalar massless $1$-loop propagator of figure~\ref{fig:low-loop-masters} gives
\begin{equation*}
	\FR\left(\Graph[0.4]{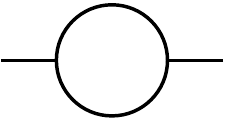}\right)
	= \int_{ \R^{\dimension}} \frac{\dd[\dimension] k}{\pi^{\dimension/2}}
	\frac{1}{k^{2\EP_1}(q+k)^{2\EP_2}}
\end{equation*}
in the \emph{momentum space representation} of $\dimension$-dimensional space-time and depends on exponents $\EP_i$ (called \emph{indices}) and an external momentum $q^2$. This function is just
\begin{equation}
	\onemaster{\EP_1}{\EP_2}
	\defas
	\FR\left(\Graph[0.4]{1master}\right)
	=
		q^{D-2\EP_1 - 2\EP_2}
		\frac{
			\Gamma\left( \frac{\dimension}{2} - \EP_1 \right) 
			\Gamma\left( \frac{\dimension}{2} - \EP_2 \right) 
			\Gamma\left( \EP_1+\EP_2 - \frac{\dimension}{2} \right)
		}{
			\Gamma(\EP_1)\Gamma(\EP_2)\Gamma(\dimension-\EP_1-\EP_2)
		},
	\label{eq:oneloop-master}%
\end{equation}
but for a generic diagram these integrals are exceedingly complicated and cannot be evaluated with elementary methods. They may depend on several external momenta and internal masses and it is not understood what kind of special functions and numbers can arise this way. Certainly this class is very rich, for it includes involved objects like elliptic polylogarithms \cite{BlochKerrVanhove:3loopSunrise,BlochVanhove:Sunset,AdamsBognerWeinzierl:SunriseELI} and special values of $L$-functions of modular forms \cite{Broadhurst:MZVModularQFT}. Also in massless integrals, counterexamples to polylogarithmic results have been identified \cite{BrownDoryn:FramingsForGraphHypersurfaces,BrownSchnetz:K3Phi4}, even in supersymmetric theories \cite{NandanPaulosSpradlinVolovich:StarIntegrals,CaronHuotLarsen:UniquenessTwoLoopMasterContours}.

In the following we devote our attention exclusively to the very simple class of Feynman integrals that evaluate to multiple polylogarithms.

\subsection{Multiple zeta values}
Riemann zeta values $\mzv{n}=\sum_{k=1}^{\infty} k^{-n}$ had already occurred in very early calculations in quantum electrodynamics \cite{Sommerfield:MagneticMoment} and are featured in almost every recent computation. An outstanding example is the series $\ZZ{n}$ of $n$-loop zigzag graphs (figure~\ref{fig:zigzags}), which have recently been proven to evaluate to a rational multiple of $\mzv{2n-3}$ \cite{BrownSchnetz:ZigZag}. This result was conjectured by Broadhurst and Kreimer, who identified \emph{multiple} zeta values (MZV)
\begin{equation}
	\mzv{n_1,\ldots,n_r}
	\defas \sum_{0<k_1<\cdots<k_r}
	\frac{1}{k_1^{n_1}\!\cdots k_r^{n_r}}
	\quad\text{where}\quad
	n_1,\ldots,n_r \in \N
	\quad\text{and}\quad
	n_r>1
	\label{eq:def-MZV-intro}%
\end{equation}%
\begin{figure}
	\centering
	$ \ZZ{3} = \Graph[0.5]{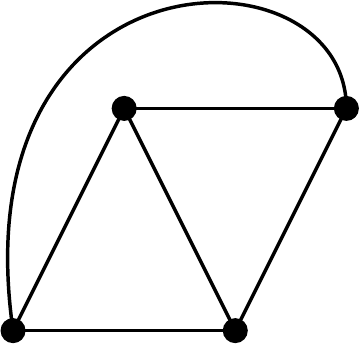} \qquad
	  \ZZ{4} = \Graph[0.5]{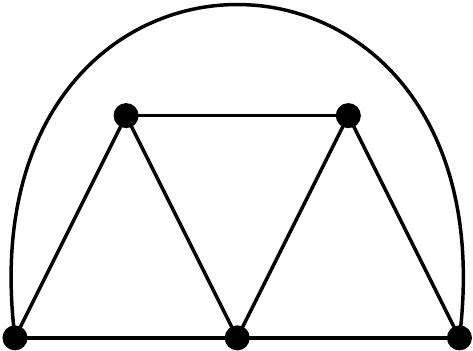} \qquad
	  \ZZ{5} = \Graph[0.5]{zz5} \qquad \cdots $
	  \caption{The zigzag-series of primitive $\phi^4$ graphs in $\dimension=4$ dimensions.}%
	  \label{fig:zigzags}%
\end{figure}%
in vacuum diagrams of $\phi^4$-theory \cite{BroadhurstKreimer:KnotsNumbers,BroadhurstKreimer:MZVPositiveKnots}. Cutting an edge maps such graphs to propagators like \eqref{eq:oneloop-master}, which have a trivial (power-like) dependence on a \emph{single scale} (the momentum). Their value at $q^2 \defas 1$ is often just called \emph{period}.
For many years it was unclear why MZV occur this way, until Francis Brown achieved a breakthrough \cite{Brown:PeriodsFeynmanIntegrals}. 
\begin{theorem}\label{theorem:vw3-MZV}
	If a graph $G$ has vertex-width $\vw(G) \leq 3$ at most three, then its periods are multiple zeta values.
\end{theorem}
Hitherto this statement has been the only supply of an infinite, non-trivial family of Feynman graphs proven to evaluate to MZV, without requiring special relations between $\dimension$ and the indices $\EP_e$.\footnote{For particular choices of these parameters, additional symmetries can become available and some infinite families of graphs have been computed this way, including \cite{DavydychevUssyukina:LadderDiagramsArbitraryRungs,Isaev:OperatorApproach}.}
A graph $G$ has vertex-width at most three if one can select three of its vertices, marked as \emph{external}, and repeatedly
\begin{itemize}
	\item remove edges between external vertices or
	\item delete an external vertex if it has only one neighbour (and select this neighbour as external)
\end{itemize}
such that finally all edges could be removed and only three (external) vertices remain. For example, the zigzags $\ZZ{n}$ belong to this family. So at least for these graphs, the appearance of MZV is understood. 

\subsection{Multiple polylogarithms}
To investigate scattering processes one must consider more complicated integrals that depend on more than one scale and thus develop a non-trivial dependence on these kinematic invariants.
It has been known for long that in four dimensions, all $1$-loop integrals \cite{tHooftVeltman:ScalarOneLoop} may be expressed in terms of logarithms and the dilogarithm of Euler \cite{Euler:MZV}. Since then, a plethora of exact results has been obtained in terms of more general multiple polylogarithms (MPL)
\begin{equation}
	\Li_{n_1,\ldots,n_r}(z_1,\ldots,z_r)
	\defas \sum_{0<k_1<\cdots<k_r}
	\frac{z_1^{k_1}\!\cdots z_r^{k_r}}{k_1^{n_1}\!\cdots k_r^{n_r}},\quad
	n_1,\ldots,n_r \in \N
	\label{eq:def-MPL-intro}%
\end{equation}%
of several variables \cite{Zagier:MZVApplications,Goncharov:MplCyclotomyModularComplexes}. Many different techniques have been developed to tackle such integrals and most of them are reviewed nicely in the book \cite{Smirnov:AnalyticToolsForFeynmanIntegrals}. Recently significant progress in the evaluation of Feynman integrals was achieved in particular using modern summation techniques \cite{AblingerBluemleinRaabSchneider:IteratedBinomialSums,AblingerBluemleinSchneider:GeneralizedHarmonicSumsAndPolylogarithms,MochUwerWeinzierl:NestedSums} and with considerable improvements of the method of differential equations \cite{HennSmirnov:PlanarThreeLoop,Henn:MultiloopSimple,ArgeriDiVitaMastroliaMirabellaSchlenkSchubertTancredi:MagnusDyson,GehrmannManteuffelTancrediWeihs:TwoLoopqqVV}.
In particular it became clear that even for intricate kinematics (depending on up to four scales), there are numerous examples of integrals that can be expressed as multiple polylogarithms.

Unfortunately, these powerful techniques are applied on a case-by-case basis and it is unclear a priori if they will be successful or not. All results obtained this way restrict to relatively low loop orders. For example, the method of differential equations requires a \emph{reduction to master integrals} which is a demanding problem that gets extremely hard to solve at growing loop orders. Only after this reduction one can build up the differential equations and study the system.

No result comparable to theorem~\ref{theorem:vw3-MZV} (applicable to an infinite number of graphs, characterized by a combinatorial criterion) is available for multi-scale Feynman integrals.

\subsection{Hyperlogarithms}
The multiple polylogarithms \eqref{eq:def-MPL-intro} can be represented as iterated integrals in terms of \emph{hyperlogarithms} (occasionally we thus treat hyperlogarithms and MPL as synonyms)
\begin{equation}
	\Hyper{\letter{\sigma_1}\!\cdots\letter{\sigma_r}}(z)
	\defas \int_{z_1}^z \frac{\dd z_1}{z_1-\sigma_1} \int_{0}^{z_1} \frac{\dd z_{2}}{z_{2}-\sigma_{2}} \int \cdots \int_0^{z_{r-1}} \frac{\dd z_r}{z_r-\sigma_r}
	\label{eq:intro-hlog}%
\end{equation}
which where introduced long ago \cite{LappoDanilevsky:CorpsRiemann}. Francis Brown devised an algorithm to compute some vacuum integrals with the help of these hyperlogarithms \cite{Brown:TwoPoint}. This novel method of integration requires a special property of the integrand called \emph{linear reducibility}, which was shown to hold for all graphs with vertex-width at most three. Hence theorem~\ref{theorem:vw3-MZV} is in principle effective in that all corresponding periods can in theory be computed with the algorithm.

An implementation of this program has unfortunately not been published, though it had been applied sporadically \cite{BroedelSchlottererStieberger:PolylogsMZVSuperstringAmplitudes,AnastasiouDuhrDulatHerzogMistlberger:RealVirtualInclusiveHiggs,AnastasiouDuhrDulatMistlberger:SoftTripleRealHiggs,ChavezDuhr:Triangles}. Some of these applications consider multi-scale integrals and in \cite{Lueders:LinearReduction} it was verified that linear reducibility applies to some $4$-point integrals.

\subsection{Goals}
Our aim is to gain a better theoretic understanding of the integrals that can be computed with hyperlogarithms, but at the same time we want to supply efficient tools to actually perform these calculations in practice. In particular we will
\begin{enumerate}
	\item provide an implementation of the hyperlogarithm integration method that is suitable for practical calculations (in particular of Feynman integrals),
	\item study linear reducibility for non-trivial kinematic dependence and
	\item extend the algorithm to divergent integrals that are regularized analytically.
\end{enumerate}

\section{Overview}
For the most part of this thesis we tried to separate our two main topics as far as possible: Properties of Feynman integrals in the parametric representation and algorithms for symbolic integration with hyperlogarithms. This is important because the latter have a much wider range of applications than just Feynman integrals, and conversely our results on the analytic continuation of parametric integrals and recursion formulas are very likely of relevance for other methods of integration than hyperlogarithms.

However, we originally developed both aspects in parallel and worked out in particular those details that are relevant for their combination. In the final chapter we give a series of examples obtained with this marriage.

Knowledge of quantum field theory is not necessary to understand this thesis and we try to keep physical input to a minimum. The reader may find accounts on Feynman integrals in perturbation theory in most introductory textbooks like \cite{ItzyksonZuber}. The thesis \cite{Bogner:PhD} nicely summarizes different steps and the arising complications during Feynman integral calculations.

\subsection{Schwinger parameters}
In chapter~\ref{chap:parametric} we recall the well-known \emph{parametric representation} of Feynman integrals, which for the example of the $1$-loop propagator becomes
\begin{equation*}
	\FR\left(\Graph[0.4]{1master}\right)
	= \frac{\Gamma(\EP_1 + \EP_2-\dimension/2)}{\Gamma(\EP_1)\Gamma(\EP_2)}
	\int_0^{\infty} \restrict{
			\frac{\SP_1^{\dimension/2-\EP_2-1} \SP_2^{\dimension/2-\EP_1 - 1}}{(\SP_1 + \SP_2)^{\dimension-\EP_1 - \EP_2}}
		}{\SP_2 = 1}
	\dd \SP_1.
\end{equation*}
It is ideally suited to study the analytic properties of Feynman integrals as meromorphic functions of the dimension $\dimension$ and indices $\EP_e$ \cite{Speer:GeneralizedAmplitudes}. Often one wants to evaluate a Feynman integral at a point $(\dimension,\EP_e)$ where it is divergent. It has become common practice to use the analytic continuation to regulate these divergences; most frequently one keeps the indices $\EP_e$ fixed and varies only $\dimension = 4-2\varepsilon$ \cite{tHooftVeltman:RegularizationGaugeFields}, which is called \emph{dimensional regularization}.
 
We recall the power-counting to reveal infrared- and ultraviolet divergences \cite{Speer:SingularityStructureGenericFeynmanAmplitudes}. As a result of an integration-by-parts procedure, we derive an algorithm to generate convergent integral representations for any chosen expansion point. In particular we prove
\begin{theorem}
	Let $G$ denote a Feynman graph with Euclidean kinematics (non-negative masses and positive definite metric). Then for any choice of $\dimension$ and $\EP_e$, $\FR(G) = \sum_{i=1}^{N} q_i \FR(G_i)$ can be written as a finite linear combination such that
	\begin{enumerate}
		\item $\FR(G_i)$ is convergent,
		\item $q_i \in \Q(\dimension,\EP_e\colon e \in \edges)$ is a rational prefactor,
		\item $G_i = \restrict{G}{\dimension=\dimension^{(i)}, \EP = \EP^{(i)}}$ differs from $G$ only by integer shifts $\EP_e^{(i)} - \EP_e \in \N$ of the indices and even shifts $\dimension^{(i)} - \dimension \in 2\N_0$ of the dimension.
	\end{enumerate}
\end{theorem}
In practice this means that one never has to compute a divergent integral and all coefficients in the $\varepsilon$-expansion of a divergent integral admit convergent integral representations. This proves that these coefficients are \emph{periods} for algebraic values of the kinematic invariants, which was known before but relied on the principle of \emph{sector decomposition} and a version of Hironaka's desingularization theorem \cite{BognerWeinzierl:Periods}. The advantage of our approach is that the obtained representation is naturally interpreted in terms of Feynman integrals in the original Schwinger parameters. This is very important for the application of hyperlogarithms for their integration.

Finally we show how combinatorics of graphs and graph polynomials may be exploited to obtain recursive integral formulas in the parametric representation. This idea is very clear in momentum space but difficult to carry out in practice. We find that it is very natural and simple in the parametric representation, \emph{before} kinematics are introduced.

For this purpose we introduce \emph{forest functions} of two different kinds, adapted to three- and four-point kinematics. It will become clear that these new objects (which will turn out to be the inverse Laplace transforms of Feynman integrals) are very useful for recursive computations of a graph in terms of its subgraphs.

Together with the material of chapter~\ref{chap:hyperlogs}, this framework provides a simplified proof of theorem~\ref{theorem:vw3-MZV} and the following two extensions:

\begin{theorem}
	All coefficients in the $\varepsilon$-expansion of a $3$-point function (with three arbitrary external momenta $p_i^2$) with vertex-width three are linearly reducible and evaluate to rational linear combinations of multiple polylogarithms and MZV.
	\label{theorem:vw3-momentum}%
\end{theorem}
\begin{theorem}
	\label{theorem:ladderbox-intro} %
	Let $G$ denote any of the box-ladder graphs $B_n$ (indicted in figure~\ref{fig:ladderboxes}) or a minor\footnote{A minor is any graph obtained by deletion or contraction of edges, sometimes also called \emph{subtopology}.} of such a graph. For vanishing internal masses and light-like external momenta $p_1^2 = p_2^2 = 0$, the Feynman integral $\FR(G)$ is linearly reducible and can be computed in terms of multiple polylogarithms.
	
	In the fully on-shell case (where $p_3^2 = p_4^2 = 0$ as well) it can be expressed with multiple zeta values and harmonic polylogarithms of the ratio $x = (p_1 + p_4)^2 / (p_1 + p_2)^2$.
\end{theorem}
The precise (and more refined) forms of these main results of this thesis are given as theorems~\ref{theorem:vw3-3pt} and \ref{theorem:ladderbox-kinematics}. 

\begin{figure}\centering
	$B_3 = \Graph[0.4]{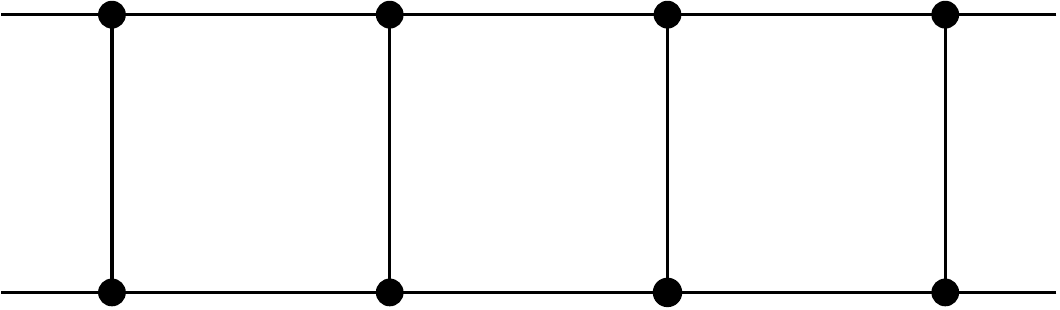}$
	\qquad
	$B_4 = \Graph[0.4]{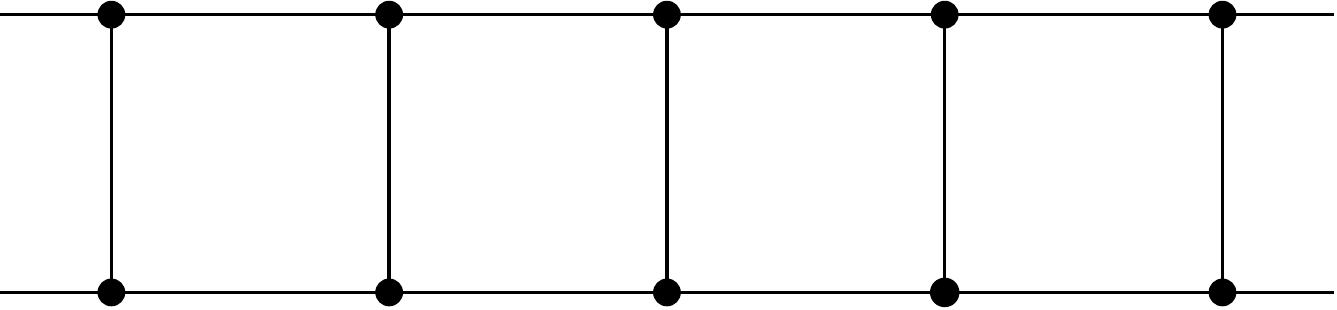}$
	\caption{Examples of the \emph{box ladder} graphs $B_n$ with $n=3$ and $n=4$ loops.} %
	\label{fig:ladderboxes} %
\end{figure}
\subsection{Hyperlogarithms}
In chapter~\ref{chap:hyperlogs} we give a self-contained account of the theory of hyperlogarithms. Our presentation focuses on a formulation of each result in an algorithmic form in order to make a possible implementation obvious.

Their application to compute multivariate integrals requires the \emph{linear reducibility} of the integrand. We review this property and introduce a refined algorithm for the approximation of Landau varieties, which improves the original method of \cite{Brown:PeriodsFeynmanIntegrals}. This polynomial reduction provides an upper bound on the \emph{symbol} \cite{GoncharovSpradlinVerguVolovich:ClassicalPolylogarithmsAmplitudesWilsonLoops} of the hyperlogarithms that can appear.
We will use it to prove our main theorems as a consequence of the recursion formulas we set up in chapter~\ref{chap:parametric}.

Our implementation {\HyperInt} \cite{Panzer:HyperIntAlgorithms} of the algorithms in the computer algebra system {\Maple} will be described briefly as well. It proved itself very useful for practical computations of Feynman integrals \cite{Panzer:MasslessPropagators,Panzer:DivergencesManyScales,Panzer:LL2014}.

\subsection{Applications and examples}
In the final chapter we present selected details, interesting results and observations from practical calculations of Feynman integrals using the methods we developed before.
The majority of explicit results obtained are new.

As an example of a non-MZV period, we comment on a counterterm in massless $\phi^4$ theory which evaluates to multiple polylogarithms at primitive sixth roots of unity. This evaluation used a parity result which we prove in section~\ref{sec:Periods}:
\begin{theorem}
	Consider the algebra $\Q[\Li_{n_1,\ldots,n_r}(1,\ldots,1,\xi_6)]$ of multiple polylogarithms at the primitive sixth root $\xi_6 = e^{\imag\pi/3}$ of unity \cite{Deligne:GroupeFondamentalMotiviqueN}. If $r$ and $n_1+\cdots+n_r$ have the same parity, then $\imag\Imaginaerteil \Li_{\vec{n}}(1,\ldots,1,\xi_6)$ can be written as a linear combination of products of such values and terms with lower depth (smaller $r$). The same decomposition is possible for the real part when the parity of $r$ and $n_1+\cdots + n_r$ is different.
\end{theorem}
Independently from our work, Oliver Schnetz informed us that he obtained a different proof of this result based on motivic periods.

\section{Outlook}

\subsection{Beyond Feynman integrals}

The integration algorithms of chapter $2$ are applicable not to Feynman integrals alone, but to any integral whose integrand is built from polylogarithms and rational functions such that the criterion of linear reducibility is fulfilled. One such example are hypergeometric functions, which are important in particle physics because several Feynman integrals have been rewritten in terms of hypergeometric functions.
Several programs for their expansion are available; some are based on the integral representation \cite{HuberMaitre:HypExp,HuberMaitre:HypExp2} while others use summation methods \cite{MochUwer:XSummer,Weinzierl:SymbolicExpansion,MochUwerWeinzierl:NestedSums,Weinzierl:HalfIntegerExpansion}. Parametric integration could simplify and extend the cases for which expansions can be computed. This is ongoing work together with Christian Bogner, who already applied the method \cite{BognerBrown:SymbolicIntegration} of multivariate iterated integrals to integration of hypergeometric functions \cite{BognerBrown:GenusZero}.

Also, the considerations of chapter one concerning the convergence of parametric integrals, their divergences and analytic regularization extend to suitable parametric integrals in general: We only require that the original integrand is a product of polynomials raised to some powers and such that each monomial appearing in any of these polynomials has coefficients with positive real part. Then we know:
\begin{itemize}
	\item The integral is an analytic function of the coefficients of the monomials in the domain where they all have positive real part.
	\item Analytic continuation in terms of exponents of the polynomials (like $\dimension$ and $\EPE_e$ in the Feynman integral case) are possible with partial integrations like in lemma~\ref{lemma:anapartial}.
\end{itemize}

\subsection{Linear reducibility}
Still only very little is understood about linear reducibility. Very interestingly, we observed many cases where this criterion fails in the Schwinger parameters, but is restored after a suitable change of variables (see section~\ref{sec:P711} and \cite{Panzer:DivergencesManyScales}).
A first attempt to a systematic study of at least one particular change of variables was given in \cite{Yeats:SomeInterpretations}, but much more work remains to be done. In particular the cases of integrals with many masses (which are in general not reducible) but still evaluate to polylogarithms, like \cite{HennSmirnov:Bhabha}, are not understood from the parametric integration viewpoint.

\subsection{Phenomenology}
Many interesting phenomenological applications of hyperlogarithms are already feasible. Recall that in principle the program {\HyperInt} suffices to compute
\begin{itemize}
	\item $5$-loop massless propagators (at least the $\phi^4$ graphs, very likely all),
	\item $3$-loop massless $3$-point functions in position- and momentum space,
	\item all minors of massless box ladder graphs with two light-like and two massive external legs (or simpler kinematics)
	\item plenty of further integrals, also involving masses \cite{Panzer:DivergencesManyScales}.
\end{itemize}
For example, very recently a result on $3$-loop ladder boxes with one leg off-shell has been published \cite{DiVitaMastroliaSchubertYundin:ThreeLoopLadderBox}. Our methods can extend these results to two off-shell legs and arbitrary loop number.

Also we advertise in section~\ref{sec:ex-renormalized-parametric} that renormalized observables can be calculated parametrically without introducing a regularization. A very promising and interesting project for a relevant application of this method could for example be the computation of the $\beta$-function of quantum electrodynamics (at least in the quenched case).

\subsection{Number theory}
Oliver Schnetz's theory of graphical functions and single-valued hyperlogarithms is incredibly powerful and apt to compute periods of $11$-loop vacuum diagrams. This is totally out of reach with direct parametric integration and the current program. But the computation of further small graphical functions will provide more and more periods at higher weight, using the recursive methods of adding edges and appending vertices to graphical functions \cite{Schnetz:GraphicalFunctions}. With a growing pool of data on periods at hand, one might hope to gain an intuition and further insight into the highly constrained structure of Feynman periods.

In a different direction, by now we know many graphs which evaluate to multiple zeta values even though their parametric integration entails singularities that suggest alternating sums to appear, see for example section~\ref{sec:P79}. This discrepancy poses an interesting problem for future research.

\subsection{Implementation}
The implementation of {\HyperInt} could be improved considerably. Most severely, the current program can not deal efficiently with complicated rational functions, such that the naive reduction to finite integrals (in the initial presence of divergences) as constructed in corollary~\ref{cor:convergent-anareg-integrand} is not viable for many divergences.

Instead, this reduction to finite master integrals should be implemented in the designated programs on integration by parts \cite{ManteuffelStuderus:Reduze2,Smirnov:Fire4LiteRed}.

\section*{Acknowledgments}

I thank Francis Brown for his beautiful articles, helpful discussions, hospitality at IHES and extremely careful reading of my manuscript, 
my supervisor Dirk Kreimer for continuous encouragement, confidence, full support on all projects and wise counsel, 
David Broadhurst for sharing his wisdom, endless enthusiasm and igniting the passion for physics and mathematics again and again,
Oliver Schnetz for strong support, compliments, collaboration and motivation through ever more challenging problems to tackle,
Christian Bogner for many fruitful discussions, detailed examples and independent checks with his program,
Marcel Golz for saving me from confusing Lagrange with Laplace and
all of the Kreimer Gang for a great atmosphere at work and uncountable instructive discussions which ever so often made me realize how little I really understood.

My teachers and supervisors from Cambridge, FU, BTU and the Max-Steenbeck-Gymnasium Cottbus were fantastic and I had the chance to learn a lot in many wonderful courses.
I am indebted to the Studienstiftung, not only for a scholarship to study abroad, but also for the chance to attend exciting summer schools. My interests into physics, mathematics and informatics were also greatly supported by the effort of many people involved in the preparation and training for the various undergraduate olympiads.

Furthermore I profited tremendously from many discussions at seminars, conferences and schools, and I only do not dare to list all those people for the danger of overlooking someone unintentionally.

Finally, my research would not have been possible without the love, encouragement, endless confidence and support wherever possible by my family.
I thank Maxie for her long-term friendship, and PopKon for the best hours after work.

Figures were generated with {\JaxoDraw} \cite{BinosiTheussl:JaxoDraw} and {\Axodraw} \cite{Vermaseren:Axodraw}.

Thanks to the help of attentive readers, in particular Francis Brown and Jianqiang Zhao, several misprints could be corrected in this updated version of my thesis. It also contains a few additional comments.

\chapter{Parametric Feynman integrals}
\label{chap:parametric}%

In perturbative quantum field theory, Feynman integrals are contributions to the Dyson series and are naturally expressed in position- or momentum space. But it was soon realized that these can be rewritten in what we call the \emph{(Schwinger-) parametric representation}, which is for example briefly described in \cite{ItzyksonZuber}.

While as of today most calculations are based on the momentum space representation, many early articles exploited the parametric representation to great effect in the study divergences and renormalization. A detailed study of many analytic properties of Feynman integrals was collected in the book \cite{Nakanishi:GraphTheoryFeynmanIntegrals}. It focuses on the combinatorial properties of Feynman graphs and their relations to their analysis.

In this thesis we argue that this representation is also very well adapted to the evaluation of Feynman integrals in terms of iterated integrals. The integration algorithms we will develop in chapter~\ref{chap:hyperlogs} happen to be extremely useful to compute Feynman integrals in the parametric representation.

We devote the first sections of this chapter to a self-contained derivation of the parametric representation, partly for convenience of the reader but also because the book \cite{Nakanishi:GraphTheoryFeynmanIntegrals} is very difficult to obtain nowadays and many details can not be found in modern references. Furthermore, we require a good understanding of certain generalizations of graph polynomials for the recursion formulas developed at the end of this chapter.
Therefore we include a proof of the well-known matrix-tree-theorem \ref{theorem:matrix-tree}.

Furthermore, we recall power-counting theorems to assess the convergence of those integrals and explain a general method for obtaining the analytic regularization in terms of convergent integrals. This is an elementary procedure, yet it is crucial for our approach of integration and may actually be of use on more general grounds, as we will briefly discuss.

Afterwards we briefly sketch how renormalization can be carried out in the parametric representation, making use of the Hopf algebra of Feynman graphs. We will only address logarithmic ultraviolet divergences here and discuss the angle- and scale-dependence of renormalized Feynman integrals, the renormalization group (or dependence on the renormalization scheme) and the \emph{period} which gives a contribution to $\beta$-functions.

The last part of this chapter is dedicated to two particular, infinite classes of Feynman integrals: recursively three-point graphs (\emph{vertex-width} $3$) and ladder boxes. We will define these and derive recursive integral representations that allow for their efficient computation with the algorithms of chapter~\ref{chap:hyperlogs}. Some explicit results and comments are given in sections~\ref{sec:ex-3pt} and \ref{sec:ex-ladderboxes}.

\section{Representations using the Schwinger trick}
\subsection{Feynman graphs}
Throughout this thesis we will consider connected multigraphs $G=(\vertices,\edges,\source,\target)$ (multiple edges connecting the same pair of vertices as well as self-loops are allowed) consisting of finite sets of vertices $\vertices(G)$ and edges $\edges(G)$. We assume\footnote{All results are independent of the chosen orientation as we only consider scalar integrals.} that each edge $e \in \edges(G)$ is directed from a source vertex $\source(e)$ to a target vertex $\target(e)$ and encode this data in the \emph{incidence matrix} $\IM$ through
\begin{equation}
	\forall e\in\edges, v\in\vertices:
	\quad
	\IM_{e,v} \defas
			\begin{cases}
				-1 & \text{if $e$ starts in $v = \source(e)$,} \\
				1 & \text{if $e$ ends in $v = \target(e)$ and} \\
				0 & \text{when $e$ is not incident to $v$.} \\
			\end{cases}
	\label{eq:incidence-matrix} %
\end{equation}%
\nomenclature[E]{$\IM$}{incidence matrix of a graph, equation~\eqref{eq:incidence-matrix}\nomrefpage}%
\nomenclature[G]{$G$}{a graph}%
\nomenclature[g]{$\gamma$}{a subgraph}%
\nomenclature[V]{$\vertices$}{the set of vertices of a graph}%
\nomenclature[E]{$\edges$}{the set of edges of a graph}%
\nomenclature[v]{$v$}{a vertex of a graph}%
\nomenclature[e]{$e$}{an edge of a graph}%
\nomenclature[e]{$\source(e), \target(e)$}{the source and target of a (directed) edge $e$}%
\nomenclature[alpha e]{$\SP_e$}{the Schwinger parameter associated to an edge $e$, see equation~\eqref{eq:Schwinger-trick}\nomrefpage}
Each edge $e \in \edges$ represents a scalar particle of non-negative mass $m_e \in \R_{\geq 0}$, whose propagation is described in momentum space by $(k_e^2 + m_e^2)^{-1}$ ($k_e$ is the momentum of the particle $e$). We allow this propagator to be raised to some power $\EP_e \in \C$, called \emph{index} of $e$. Furthermore, the vertices $\vertices = \vertices_{\Tint} \cupdot \vertices_{\Text}$ are partitioned into \emph{internal} and \emph{external} vertices.

We use elementary concepts and results from graph theory without further reference as they can be obtained from any book on the subject, including \cite{West:GraphTheory} which we recommend. However, we will focus on combinatorics of graph polynomials and will provide the corresponding proofs and definitions along the way. A superb reference for this combinatorial theory of Feynman graphs is \cite{Nakanishi:GraphTheoryFeynmanIntegrals}.

Most of the time, a \emph{subgraph} $\gamma \subseteq G$ is identified with its edges $\edges(\gamma)\subseteq \edges(G)$ and so we write $\gamma \subseteq \edges(G)$. In this case, we always assume that $\vertices(\gamma) = \vertices(G)$ contains all vertices of $G$ (so $\gamma$ is a \emph{spanning subgraph}). A \emph{forest} $F \subseteq \edges(G)$ is a subgraph without cycles and a \emph{tree} is a connected forest. We denote the set of connected components of $G$ by $\comps(G)$ and recall the formula
\begin{equation}
	\loops{G}
	= \abs{\edges(G)} - \abs{\vertices(G)} + \abs{\comps(G)}
	\label{eq:loop-number} %
\end{equation}%
\nomenclature[pi0G]{$\comps(G)$}{the set of connected components of a graph $G$}%
that counts the number of independent cycles in $G$.\footnote{This is the first Betti number $\loops{G} = \dim H_1(G)$ of $G$ as a simplicial complex, just as $\abs{\comps(G)} = \dim H_0(G)$.} The \emph{contraction} $G\contract\gamma$ is obtained from $G$ by replacing each connected component of $\gamma$ with a single vertex (the edges of $\gamma$ are not present anymore in $G\contract\gamma$). The graph $G$ is \emph{one-particle irreducible} (1PI) if it is connected and stays so ($\abs{\comps(G\setminus e)}= 1$) upon deletion of any edge ($e \in \edges$). Feynman graphs are usually assumed to be 1PI, but we will also consider non-1PI graphs in constructions of a graph in sections~\ref{sec:vw3} and \ref{sec:ladderboxes}.

\subsection{Scalar momentum space integrals}\label{sec:momentum-space}
In momentum space, each external vertex $v \in \vertices_{\Text}$ is assigned an incoming external momentum
$
	\ExMom(v)
	\in
	\R^{\dimension}
$
belonging to $\dimension$-dimensional Euclidean space-time $\R^{\dimension}$ (results for Minkowski space-time can be obtained through analytic continuation). We set $\ExMom(v) = 0 $ for internal $v \in \vertices_{\Tint}$.%
\nomenclature[p(v)]{$\ExMom(v)$}{external momentum entering at vertex $v$, section~\ref{sec:momentum-space}\nomrefpage}%
\nomenclature[D]{$\dimension$}{dimension of space-time}

Scalar Feynman rules $\FR$ assign to $G$ the integral\footnote{This choice of constant prefactor removes explicit powers of $\pi$ in \eqref{eq:feynman-integral-parametric}.}
\begin{equation}
	\FR(G)
	=
	\prod_{e\in\edges}
			\int_{\R^{\dimension}} \frac{\dd[\dimension] k_e}{\pi^{\dimension/2}} \left( k_e^2 + m_e^2 \right)^{-\EP_e}
	\prod_{v\in\vertices\setminus\set{v_0}}
			\pi^{\dimension/2}
			\delta^{(\dimension)}\left( k_v \right)
	\label{eq:feynman-integral-momentum} %
\end{equation}%
\nomenclature[Phi]{$\FR(G)$}{momentum space Feynman rules applied to $G$, equation~\eqref{eq:feynman-integral-momentum}\nomrefpage}%
over the momenta $k_e$ flowing through edge $e$, which are subject to momentum conservation constraints $\delta^{(D)}(k_v)$ where
\begin{equation*}
	k_v \defas \ExMom(v) + \sum_{e\in\edges} \IM_{e,v} k_e
\end{equation*}
collects the total momentum flowing into $v$.
We omit this factor for an arbitrary vertex $v_0 \in \vertices$ to strip off an overall $\delta^{(\dimension)}\left( \sum_{v\in\vertices} \ExMom(v) \right)$ from the result. 
The \emph{Schwinger trick}
\begin{equation}
	\frac{1}{P^{\EP}}
	= \frac{1}{\Gamma(\EP)} \int_0^{\infty} \SP^{\EP - 1} e^{-\SP P} \dd\SP
	\quad\text{valid for}\quad
	P, \Realteil(\EP) > 0
	\label{eq:Schwinger-trick} %
\end{equation}
and 
$
	(2\pi)^{\dimension} \delta^{(D)}(k)
	=
	\int e^{i x k} \dd[\dimension]x
$
transforms this representation into
\begin{multline}
	\FR(G)
	=
	\prod_{e\in\edges}
			\int_0^{\infty} \frac{\SP_e^{\EP_e -1} \dd \SP_e}{\Gamma(\EP_e)}
	\prod_{v\in\vertices\setminus\set{v_0}}
			\int_{\R^{\dimension}} \dd[\dimension] x_v (4\pi)^{\dimension/2}
	\prod_{e\in\edges}
			\int_{\R^{\dimension}} \frac{\dd[\dimension] k_e}{\pi^{\dimension/2}} \\
	\times
	\exp\left[ 
			-\sum_{e\in\edges} \SP_e (m_e^2 + k_e^2) 
			+ i \hspace{-2mm}\sum_{v\in\vertices\setminus\set{v_0}}\hspace{-2mm} x_v \cdot \left( \ExMom(v) + \sum_{e\in\edges} k_e \IM_{e,v} \right)
	\right].
	\label{eq:feynman-integral-parametric-derivation} %
\end{multline}

\begin{definition}
	\label{def:graph-matrix}
	The \emph{graph matrix} $\GM(G)$ is the square matrix of size $\abs{\edges} + \abs{\vertices} - 1$ built out of $\IM$ and the diagonal matrix $\DM$ as
	\begin{equation}
		\GM(G)
		\defas
		\begin{pmatrix}
			\DM & \tilde{\IM} \\
			-\tilde{\IM}^{\Transpose} & 0\\
		\end{pmatrix},
		\quad
		\DM
		\defas
		\diag\left( \SP_1,\ldots,\SP_{\abs{\edges}} \right)
		= \setlength{\arraycolsep}{0mm} 
		\begin{pmatrix}
			\SP_1 & & \\[-2mm]
			& \ddots & \\[-2mm]
			& & \SP_{\abs{\edges}} \\
		\end{pmatrix}
		\label{eq:graph-matrix} %
	\end{equation}%
\nomenclature[M(G)]{$\GM(G)$}{graph matrix of $G$, equation~\eqref{eq:graph-matrix}\nomrefpage}%
	where the reduced incidence matrix $\tilde{\IM}$ is obtained from $\IM$ upon deletion of the column $v_0$.
	The \emph{Laplace matrix} $\LM$ and its dual $\hat{\LM}$ are the square matrices of size $\abs{\vertices} - 1$ given by
	\begin{equation}
		\LM
		\defas
		\tilde{\IM}^{\Transpose} \DM \tilde{\IM}
		\quad\text{and}\quad
		\hat{\LM}
		\defas
		\tilde{\IM}^{\Transpose} \DM^{-1} \tilde{\IM}.
		\label{eq:laplace-matrix} %
	\end{equation}
\end{definition}%
\nomenclature[L]{$\LM$}{Laplace matrix of a graph, equation~\eqref{eq:laplace-matrix}\nomrefpage}%
Collecting all momenta and position variables into the vectors $k = (k_e)_{e\in\edges} \in \R^{\dimension \abs{\edges}}$ and $x = (x_v)_{v\in\vertices\setminus\set{v_0}} \in \R^{\dimension (\abs{\vertices}-1)}$, completion of the square lets us rewrite the argument of the exponential in \eqref{eq:feynman-integral-parametric-derivation} as
\begin{equation*}
	- \left( k - \frac{i}{2} \DM^{-1} \tilde{\IM} x \right)^{\Transpose}
		\DM
		\left( k - \frac{i}{2} \DM^{-1} \tilde{\IM} x \right)
	- \left( \frac{x}{2} - i \hat{\LM}^{-1} p \right)^{\Transpose}
		\hat{\LM}
		\left( \frac{x}{2} - i \hat{\LM}^{-1} p \right)
	- p^{\Transpose} \hat{\LM}^{-1} p
	- \sum_{e\in\edges} \SP_e m_e
\end{equation*}
where we interpret $p=\left( \ExMom(v) \right)_{v\in\vertices\setminus\set{v_0}} \in \R^{\dimension\left( \abs{\vertices} -1 \right)}$. Hence the Gau{\ss}ian integrals first over $k$ and then over $x$ in \eqref{eq:feynman-integral-parametric-derivation} yield the \emph{parametric representation}
\begin{equation}
	\FR(G)
	=
	\prod_{e\in\edges}
			\int_0^{\infty} \frac{\SP_e^{\EP_e -1} \dd \SP_e}{\Gamma(\EP_e)}
	\cdot
	\frac{e^{-\phipol / \psipol}}{\psipol^{\dimension/2}}.
	\label{eq:feynman-integral-parametric} %
\end{equation}
It depends on the first and second \emph{Symanzik polynomials} $\psipol$ (which we also just call graph polynomial) and $\phipol$ given by
\begin{equation}
	\psipol
	=
	\det \DM
	\cdot
	\det \left( \tilde{\IM}^{\Transpose} \DM^{-1} \tilde{\IM} \right)
	= \det \DM
		\cdot
		\det \hat{\LM}
	\quad\text{and}\quad
	\phipol
	=
	\psipol \left(
		\sum_{e\in\edges} \SP_e m_e
		+
		p^{\Transpose} \hat{\LM}^{-1} p
	\right)
	.
	\label{eq:graph-polynomials} %
\end{equation}
Going back to Kirchhoff, these enjoy a long history and we refer to \cite{BognerWeinzierl:GraphPolynomials} for a review. Often they are also denoted as $\mathcal{U} = \psipol$ and $\mathcal{F} = \phipol$.

To find a combinatorial description of these polynomials, one invokes the
\begin{theorem}[Matrix-Tree-Theorem]
	\label{theorem:matrix-tree} %
	For subsets $I\subseteq \edges$ and $W\subseteq\vertices$ let $\IM(I,W)$ denote the matrix $\IM$ after deleting rows $I$ and columns $W$. If it is square, that is $\abs{\edges\setminus I} = \abs{\vertices\setminus W}$ , then
	\begin{equation}
		\det \IM(I,W) =
		\begin{cases}
			\pm 1 & \parbox[t]{9cm}{if $F \defas \edges\setminus I$ is a forest with $\abs{W}$ connected components, each containing precisely one vertex of $W$,} \\
			0 & \text{otherwise.} \\
		\end{cases}
		\label{eq:matrix-tree} %
	\end{equation}
\end{theorem}

\begin{proof}
	If $F \defas \edges \setminus I$ contains a loop $C$, $\det \IM(I, W) = 0$ because the corresponding rows
	\begin{equation*}
		\sum_{e \in C} \pm \IM(I,W)_e = 0
	\end{equation*}
	add to zero when each row $\IM(I, W)_e$ is taken with the sign $+1$ when $C$ runs through $e$ along its orientation and $-1$ if $C$ contains $e$ in reversed direction.
	
	Similarly, if $F$ contains a path $v \rightarrow \ldots \rightarrow w$ for distinct $v,w \in W$, adding the rows $\IM(I)_e$ of these edges with the appropriate signs gives a vector with only two non-zero components, namely in the columns $v$ and $w$. But these do not appear in $\IM(I,W)$, so again $\det \IM(I,W) = 0$.

	Now let $F$ be free of cycles and such paths, it follows that it has $\abs{\comps(F)} = \abs{\vertices} - \abs{F} = \abs{W}$ components as claimed; each of which contains precisely one vertex in $W$.
	Choose any edge $e\notin I$ that connects some $w\in W$ to some other vertex $v\notin W$. Then the $e$'th row of $\IM(I,W)$ contains only one non-zero entry, namely $\IM_{e,v} = \pm 1$. Expanding along this row we find
	\begin{equation*}
		\det \IM(I,W)
		= \pm \det \IM(I \cup \set{e}, W \cup \set{v}).
	\end{equation*}
	As $\edges\setminus \left( I \cup \set{e} \right)$ is a forest with one vertex of $W \cup \set{v}$ in each component, we can apply the argument again and continue until we are left with a trivial one-by-one determinant. This proves $\det \IM(I,W) = \pm 1$.
\end{proof}

\begin{theorem}
	\label{theorem:graph-polynomials} %
	The graph polynomials for a connected graph $G$ are given by
	\begin{equation}
		\psipol
		=
		\sum_T \prod_{e\notin T} \SP_e
		\quad\text{and}\quad
		\phipol
		= \psipol \sum_{e\in\edges} \SP_e m_e^2
		+
		\sum_F \ExMom(F)^2 \prod_{e \notin F} \SP_e
		,
		\label{eq:graph-polynomials-combinatorial} %
	\end{equation}
	where the sums run over all spanning trees $T$ and spanning two-forests $F$ which are defined to be those subsets of $\edges$ that do not contain any cycles/loops and have $\abs{\comps(T)} = 1$ or $\abs{\comps(F)} = 2$ connected components.

	We write $\ExMom(F) \defas \sum_{v\in F_0} \ExMom(v)$ for the momentum flowing into the component $F_0 \in \comps(F)$ that contains $v_0$.
\end{theorem}

\begin{proof}
	First notice that by linearity of the determinant we can expand \eqref{eq:graph-polynomials} as
	\begin{equation}
		\psipol
		= \det \begin{pmatrix}
				\DM & \tilde{\IM} \\
				0 & \tilde{\IM}^{\Transpose} \DM^{-1} \tilde{\IM} \\
			\end{pmatrix}
		= \det \GM(G)
		= \sum_{S \subseteq \edges}
			\prod_{e\in S} \SP_e
			\det \begin{pmatrix}
				0 & \tilde{\IM}(S) \\
				-\tilde{\IM}(S)^{\Transpose} & 0 \\
			\end{pmatrix}
		,
		\label{eq:graph-polynomial-proof-expansion} %
	\end{equation}
	where $\tilde{\IM}(S)$ denotes $\tilde{\IM}$ after deletion of the rows $S$ (for the second equality, multiply the first $\abs{\edges}$ rows of $\GM(G)$ with $\IM^{\Transpose} \DM^{-1}$ and add this to the lower $\abs{\vertices}-1$ rows). Let $\GM_S$ denote the last matrix in this equation, then
	\begin{equation}
		\rank \GM_S
		= 2\rank \tilde{\IM}(S)
		\leq 2 \min\set{\abs{\vertices}-1,\abs{\edges\setminus S}}
		< \abs{\vertices}-1 + \abs{\edges\setminus S}
		\label{eq:graph-polynomial-proof-ranks} %
	\end{equation}
	whenever $\abs{\vertices}-1 \neq \abs{\edges\setminus S}$. Therefore $\GM_S$ can only be non-singular for square $\tilde{\IM}(S)$ with $\abs{\vertices}-1 = \abs{\edges\setminus S}$, but then theorem \ref{theorem:matrix-tree} immediately shows that
	\begin{equation*}
		\det \GM_S
		= \left[ \det \tilde{\IM}(S) \right]^2
		= \begin{cases}
				1 & \text{if $\edges\setminus S$ is a spanning tree and} \\
				0 & \text{otherwise.}\\
			\end{cases}
	\end{equation*}
	For $\phipol$ we compute the components $(\hat{\LM}^{-1})_{v,w} = (-1)^{v+w} \cdot \det \hat{\LM}(\set{w},\set{v}) \cdot \det \hat{\LM}^{-1}$ for any two $v,w\in\vertices\setminus\set{v_0}$ in terms of minors (deleting columns from $\tilde{\IM}$) as
	\begin{equation*}
		\psipol \cdot (\hat{\LM}^{-1})_{v,w}
		(-1)^{v+w}
		= \det \begin{pmatrix}
				\DM & \tilde{\IM}(\set{v}) \\
				0 & \tilde{\IM}(\set{w})^{\Transpose} \DM^{-1} \tilde{\IM}(\set{v}) \\
			\end{pmatrix}
		= \det M(\set{w},\set{v}),
	\end{equation*}
	where $M(\set{w},\set{v})$ denotes the graph matrix \eqref{eq:graph-matrix} after deleting row $w$ and column $v$.
	Expanding like \eqref{eq:graph-polynomial-proof-expansion} and analyzing the rank as in \eqref{eq:graph-polynomial-proof-ranks} shows that this equals
	\begin{equation*}
		= 
			\sum_{S \subseteq \edges}
			\prod_{e\in S} \SP_e
			\det \begin{pmatrix}
				0 & \tilde{\IM}(S, \set{v}) \\
				-\tilde{\IM}(S,\set{w})^{\Transpose} & 0 \\
			\end{pmatrix}
		= \hspace{-4mm}
				\sum_{\substack{S \subseteq \edges \\ \abs{\edges\setminus S} = \abs{\vertices}-2}}
			\hspace{-3mm}
			\prod_{e\in S} \SP_e
			\det \tilde{\IM}(S,\set{v})
			\cdot
			\det \tilde{\IM}(S,\set{w})
	\end{equation*}
	and we invoke theorem \ref{theorem:matrix-tree} again to deduce that we only get non-zero contributions when $F\defas\edges\setminus S$ is free of loops and therefore a spanning two-forest. Further, $F$ may not connect $v_0$ to neither $v$ nor $w$, so $v,w\notin F_0$ lie together in the other connected component.
	The signs conspire to $\det \tilde{\IM}(S,\set{v}) \det \tilde{\IM}(S,\set{w}) = (-1)^{v+w}$, because	\begin{align*}
		\det \tilde{\IM}(S, \set{v})
		&=
		\det \left(\cdots \not{c_v} \cdots c_w \cdots \right)
		= - \sum_{u \notin F_0 \cup \set{w}}
			\det \left( \cdots \not{c_v} \cdots c_u \cdots \right)
		\\
		&= -\det \left(\cdots \not{c_v}\cdots c_v \cdots \right)
		=
		(-1)^{v+w} \det \tilde{\IM}(S, \set{w})
	\end{align*}
	and $\sum_{u \notin F_0} c_u = 0$ if $c_u$ ($u \neq v_0$) denote the columns of $\tilde{\IM}(S, \set{v})$. Thus we conclude
	\begin{equation*}
		\psipol \cdot p^{\Transpose} \hat{\LM}^{-1} p
		= \sum_F \prod_{e\notin F} \SP_e \sum_{v,w \notin F_0} \ExMom(v) \ExMom(w)
		= \sum_F \prod_{e\notin F} \SP_e \left[ \sum_{v \notin F_0} \ExMom(v) \right]^2
		= \sum_F \ExMom(F)^2 \prod_{e\notin F} \SP_e
	\end{equation*}
	since by momentum conservation, 
	$
		\ExMom(F)
		= \sum_{v\in F_0} \ExMom(v)
		= -\sum_{v\notin F_0} \ExMom(v)
	$.
\end{proof}
\begin{remark}
	\label{remark:symanzik-properties} %
	Note the following elementary properties of Symanzik polynomials:
	\begin{enumerate}
		\item
$\psipol$ is independent of masses and momenta and linear in each individual $\SP_e$. The respective linear and constant coefficients are related to contractions and deletions:
\begin{equation}\label{eq:psi-contraction-deletion}
	\psipol_G = \SP_e \psipol_{G\setminus e} + \psipol_{G/e}.
\end{equation}
In the special case of a loop ($\source(e)=\target(e)$), this is modified to $\psipol_G = \SP_e \psipol_{G\setminus e}$.

		\item $\phipol$ is linear in $\SP_e$ only for zero mass $m_e = 0$ and otherwise quadratic. If $m_e=0$, the contraction-deletion formula \eqref{eq:psi-contraction-deletion} holds for $\phipol$ as well.

		\item
			Both $\psipol$ and $\phipol$ are homogeneous in the Schwinger parameters of degrees
			\begin{equation}
				\deg (\psipol)
				=
				\loops{G}
				\quad\text{and}\quad
				\deg(\phipol)
				=
				\loops{G}+1
				.
				\label{eq:phipsi-degrees}%
			\end{equation}
			For connected $G$, this \emph{loop number} is
			$
				\loops{G}
				\urel{\eqref{eq:loop-number}}
				\abs{\edges} - \abs{\vertices} + 1
			$.
	\end{enumerate}
\end{remark}
\begin{figure}
	\centering
	$ C_2 = \Graph[0.55]{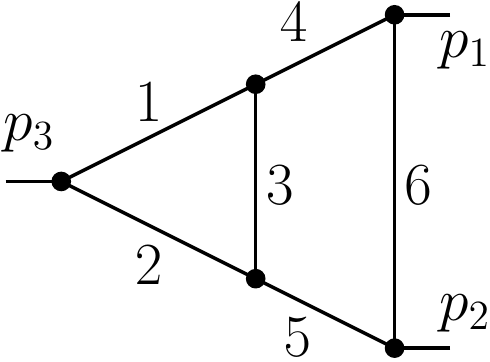} $ \qquad
	$ C_2 \contract 3 = \Graph[0.55]{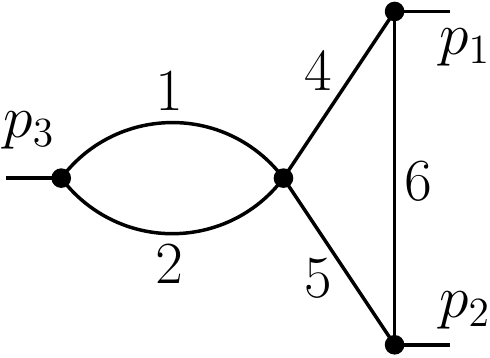} $ \qquad
	$ C_2 \setminus 3 = \Graph[0.55]{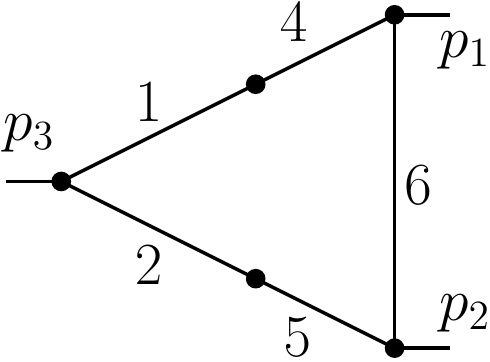} $
	\caption{Contraction and deletion of an edge.}%
	\label{fig:triladder-contraction-deletion}%
\end{figure}%
\begin{example}\label{ex:psiphi-C2}
	Consider the two-loop triangle ladder $C_2$ from figure~\ref{fig:triladder-contraction-deletion}. The first Symanzik polynomial of a cycle is just the sum of all Schwinger parameters, so the contraction and deletion of edge $3$ give
	\begin{equation}
		\psipol_{C_2} = \SP_3 \psipol_{C_2 \setminus 3} + \psipol_{C_2/3}
		= \SP_3 \left( \SP_1 + \SP_2 + \SP_4 + \SP_5 + \SP_6 \right)
		+ (\SP_1 + \SP_2)(\SP_4 + \SP_5 + \SP_6).
		\label{eq:psipol_C2}%
	\end{equation}
	If we let $p_3^2 = 1$, $p_1^2 = p_2^2 = 0$, and set all internal masses $m_e = 0$ to vanish, then we compute from $\phipol_{C_2} = \SP_3 \phipol_{C_2\setminus 3} + \phipol_{C_2\contract 3}$ the second Symanzik polynomial as
	\begin{equation}
		\phipol_{C_2} 
		= \SP_3 (\SP_1 + \SP_4)(\SP_2 + \SP_5)
		+\SP_1 \SP_2(\SP_4 + \SP_5 + \SP_6) + \SP_4 \SP_5 (\SP_1 + \SP_2).
		\label{eq:phipol_C2}%
	\end{equation}
\end{example}

\subsection{Projective integrals}
\label{sec:projective-integrals}%
The above mentioned homogeneity of the Symanzik polynomials \eqref{eq:graph-polynomials-combinatorial} allows us to carry out one integration in \eqref{eq:feynman-integral-parametric} as follows:
Choose any hyperplane $H(\SP) \defas \sum_e H_e \SP_e$ with $H_e \geq 0$ not all zero and insert $1 = \int_0^{\infty} \dd\Pscale\ \delta\left( \Pscale - H(\SP) \right)$ into \eqref{eq:feynman-integral-parametric}. After substituting $\SP_e$ for $\lambda \SP_e$,
\begin{equation*}
	\FR(G)
	= \left[
			\prod_e \int_{0}^{\infty} \frac{\SP_e^{\EP_e-1} \dd \SP_e}{\Gamma(\SP_e)}
		\right]
		\frac{\delta\left( 1- H(\SP)\right)}{\psipol^{\dimension/2}}
		\int_0^{\infty} \Pscale^{\sdd - 1} e^{-\Pscale \phipol / \psipol} \dd \Pscale
\end{equation*}
where the \emph{superficial degree of divergence} $\sdd$ of $G$ is given by
\begin{equation}
	\label{eq:def-sdd} %
	\sdd
	\defas
		\sum_{e \in \edges} \EP_e
		- \frac{\dimension}{2} \loops{G}
	.
\end{equation}
Hence the integral over $\Pscale$ gives Euler's $\Gamma$-function such that
\begin{align}
	\FR(G)
	=& \frac{\Gamma(\sdd)}{\prod_{e} \Gamma(\EP_e)}
		\int \Omega \cdot I_G
	\quad\text{where}
	\label{eq:feynman-integral-projective} %
	\\
	\int \Omega
	\defas&
		\left[
			\prod_{e} \int_{0}^{\infty} \dd \SP_e
		\right]
		\delta\left( 1- H(\SP) \right)
	\quad\text{and}\quad
	I_{G}
	\defas
		\frac{1}{\psipol^{\dimension/2}} 
		\left( \frac{\psipol}{\phipol} \right)^{\sdd}
		\prod_{e} \SP_e^{\EP_e - 1}
	.
	\label{eq:projective-delta-form-integrand} %
\end{align}
By construction, the integral \eqref{eq:feynman-integral-projective} does not depend on the choice of $H$. This fact is sometimes called \emph{Cheng-Wu theorem} and applies to the integral $\int \Omega \cdot I$ whenever the integrand $I(\Pscale \SP) = \Pscale^{-\abs{\edges}} \cdot I(\SP)$ is homogeneous. In fact, $H$ induces a bijection
\begin{equation*}
	\R_{+}^{\abs{\edges}}
	\longrightarrow
		\R_+ \times \RP_{+}^{\abs{\edges}-1}
	,\ 
	\SP
	\mapsto
	\big(
		H(\SP),
		[\SP]
	\big)
	\quad\text{with inverse}\quad
	\big(
		\Pscale,
		[\SP]
	\big)
	\mapsto
	\Pscale \cdot \frac{\SP}{H(\SP)}
	,
\end{equation*}
between the integration domain of \eqref{eq:feynman-integral-parametric} and $\R_{+} \defas \setexp{\Pscale \in \R}{\Pscale > 0}$ times the positive piece
$
	\RP_{+}^{\abs{\edges}-1}
	\defas
		\setexp{[\SP]}{\SP_1,\ldots,\SP_{\abs{\edges}}>0}
$
of projective space. Here 
$
	[\SP] 
	\defas
	\big[\SP_1\mathbin{:}\cdots\mathbin{:}\SP_{\abs{\edges}}\big]
$
denote homogeneous coordinates. In the coordinates $(\Pscale, [\SP])$, the volume form splits as
\begin{equation*}
	\bigwedge_{e} \dd \SP_e
	= \Pscale^{\abs{\edges}-1} \dd \Pscale
		\wedge
		\Omega_{H}
	,
	\quad\text{where}\quad
	\Omega_{H}
	\defas
	\sum_{e=1}^{\abs{\edges}}
		(-1)^{e-1} \frac{\SP_e}{H}
		\bigwedge_{e' \neq e} \dd \left( \frac{\SP_{e'}}{H} \right)
\end{equation*}
defines a smooth volume form 
$
	\Omega_{H}
	\in
	\Omega^{\abs{\edges}-1}\left( \RP_+^{\abs{\edges}-1} \right)
$.
Hence \eqref{eq:feynman-integral-parametric} becomes
\begin{align}
	\FR(G)
	&= \int_{\RP_+^{\abs{\edges}-1}} \Omega_H
		\int_{0}^{\infty} \frac{\dd \Pscale}{\Pscale} \cdot \Pscale^{\sdd}
		\left( \frac{H^{\loops{G}}}{\psipol} \right)^{\dimension/2}
			e^{-\Pscale/H \cdot \phipol/\psipol}
			\prod_{e} \left( \frac{\SP_e}{H} \right)^{\EP_e - 1} \frac{1}{\Gamma(\EP_e)}
	\nonumber\\
	&=
		\frac{\Gamma(\sdd)}{\prod_{e} \Gamma(\EP_e)}
		\int_{\RP_+^{\abs{\edges}-1}} \Omega_{H} \cdot \left[ H^{\abs{\edges}} \cdot I_{G} \right]
	\label{eq:feynman-integral-projective-hyperplane} %
\end{align}
and its independence of $H$ follows immediately from $\Omega_{H'} = (H/H')^{\abs{\edges}} \cdot \Omega_{H}$ for any other hyperplane $H'$. In this sense, \eqref{eq:feynman-integral-projective} is the projective integral of the smooth, $H$-independent volume form $\Omega\cdot I_G = \Omega_H \cdot H^{\abs{\edges}} \cdot I_G$ on $\RP_{+}^{\abs{\edges}-1}$.\footnote{%
Beware that the affine form $\Omega = \Omega_H \cdot H^{\abs{\edges}} = \sum_e (-1)^{e-1} \SP_e \bigwedge_{e'\neq e} \dd \SP_{e'} $ is not homogeneous of degree zero. Well-defined forms on the projective space $\RP_+^{\abs{\edges}-1}$ are instead given by $\Omega_{H}$ itself and the product $H^{\abs{\edges}}\cdot I_G$ as shown in \eqref{eq:feynman-integral-projective-hyperplane}.}

Though this interpretation is very appealing to algebraic geometry, we will not dwell on it further. In the sequel we shall always refer to \eqref{eq:feynman-integral-projective} and exploit the invariance by choosing $H$ in \eqref{eq:projective-delta-form-integrand} as suitable to assist our needs.

\subsection{Spanning forest polynomials}
It is very useful to have combinatorial, graph-theoretic descriptions for the Symanzik polynomials and generalizations thereof at hand. One such tool are the spanning forest polynomials, which were introduced and discussed in detail in \cite{BrownYeats:SpanningForestPolynomials}. We recall
\begin{definition}
	\label{def:forestpolynom} %
	Let $P = \set{P_1,\ldots,P_k}$ denote a partition $P_1 \cupdot \cdots \cupdot P_k \subseteq \vertices(G)$ of a subset of the vertices of the graph $G$. Then the associated \emph{spanning forest polynomial} $\forestpolynom[G]{P}$ and its dual $\forestpolynomDual[G]{P}$ are given by
	\begin{equation}
		\forestpolynom[G]{P}
		\defas
			\sum_{F} \prod_{e \notin F} \SP_e
		\quad\text{and}\quad
		\forestpolynomDual[G]{P}
		\defas
			\sum_{F} \prod_{e \in F} \SP_e
		,
		\label{eq:def-forestpolynom} %
	\end{equation}%
	where the sums run over all spanning forests $F$ of $G$ with precisely $k = \abs{P}$ connected components $\comps(F) = \set{T_1,\ldots,T_k}$ such that $P_i \subseteq \vertices(T_i)$ for all $1 \leq i \leq k$ (note that this implies $T_i \cap \vertices(T_j) = \emptyset$ for $i \neq j$).
	We also write $\forestpolynom{P_1,\ldots,P_k} = \forestpolynom{P}$.
\end{definition}
\begin{table}
	$	\Graph[0.5]{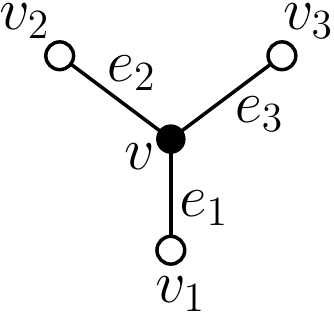} $ \hfill
	\begin{tabular}{rcccc}
		Partition:&
		$ \set{1,2}$,$\set{3} $ &
		$ \set{1,3}$,$\set{2} $ &
		$ \set{2,3}$,$\set{1} $ &
		$ \set{1}$,$\set{2}$,$\set{3}$ \\
		\cmidrule[1pt](lr){2-2}
		\cmidrule[1pt](lr){3-3}
		\cmidrule[1pt](lr){4-4}
		\cmidrule[1pt](lr){5-5}
		Forests:&
		$ \Graph[0.4]{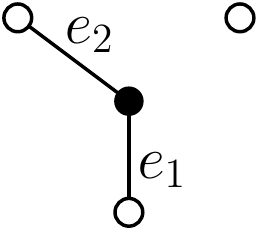} $ &
		$ \Graph[0.4]{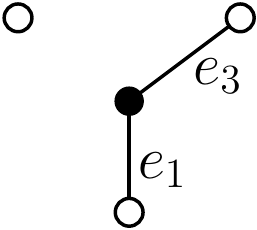} $ &
		$ \Graph[0.4]{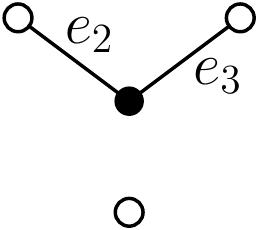}	$ &
		$ \Graph[0.4]{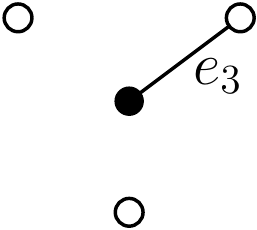} $ {\textcolor{lightgray}{\vrule width 1pt}}
		$ \Graph[0.4]{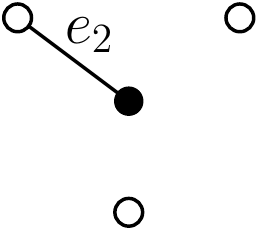} $ {\textcolor{lightgray}{\vrule width 1pt}}
		$ \Graph[0.4]{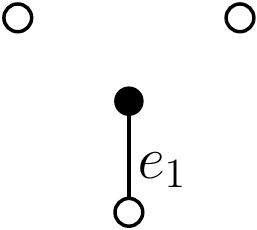} $ \\
	\end{tabular}
	\caption{The star and spanning forests for some partitions of its tips $\set{v_1,v_2,v_3}$.}%
	\label{tab:star-forests}%
\end{table}%
\begin{example}\label{ex:forestpolynom}
	The forests of the star graph contributing to $\forestpolynom[\StarSymbol]{P}$ for selected partitions are shown in table~\ref{tab:star-forests}. We read off the forest polynomials
	\begin{equation*}
		\forestpolynom[\StarSymbol]{\set{1,2},\set{3}}
		= \SP_3
		,\ 
		\forestpolynom[\StarSymbol]{\set{1,3},\set{2}}
		= \SP_2
		,\ 
		\forestpolynom[\StarSymbol]{\set{2,3},\set{1}}
		= \SP_1
		,\ 
		\forestpolynom[\StarSymbol]{\set{1},\set{2},\set{3}}
		= \SP_1 \SP_2 + \SP_1 \SP_3 + \SP_2 \SP_3.
	\end{equation*}
\end{example}
We also encountered these polynomials already: In the proof of theorem~\ref{theorem:graph-polynomials}, we expressed the inverse of the dual Laplace matrix as
\begin{equation}
	\widetilde{\LM}^{-1}_{v,w}
	= \psipol^{-1} \cdot \forestpolynom{\set{v_0}, \set{v,w}}.
	\label{eq:Laplace-inverse-forestpolynom}%
\end{equation}
We will use these polynomials in the following section and prominently for recursions in sections~\ref{sec:3pt-recursions} and \ref{sec:ladderbox-forestfunctions}. From \eqref{eq:loop-number} we find that the spanning forest polynomials are homogeneous of degrees
\begin{align}
	\deg\left( \forestpolynomDual[G]{P} \right)
	&= \abs{F}
	= \abs{\vertices(G)} - \abs{P}
	= 1 + \abs{\edges(G)} - \loops{G} - \abs{P}
	\quad\text{and}
	\label{eq:degree-forestpolynomDual}%
	\\
	\deg \left( \forestpolynom[G]{P} \right)
	&= \abs{\edges(G)} - \deg\left( \forestpolynomDual[G]{P} \right)
	= \abs{P} + \loops{G} - 1
	.
	\label{eq:degree-forestpolynom}%
\end{align}

\subsection{Position space and graphical functions}
\label{sec:position-space} %
Fourier transformation connects momentum space Feynman integrals \eqref{eq:feynman-integral-momentum} with a representation in position space. We usually prefer the former because the propagator $(k^2 + m^2)^{-1}$ is a rational function. Note that in position space, it translates not to a rational function but can be expressed in terms of a Bessel function instead. If we consider the massless case though, the position space propagator stays rational:
\begin{align}
	\posprop[\EP](x-y)
	&\defas
	\int \frac{\dd[\dimension] k}{(2\pi)^{\dimension}} \frac{e^{ik(x-y)}}{k^{2\EP}}
	= \frac{\Gamma(\dimension/2 - \EP)}{4^{\EP} \pi^{\dimension/2}\Gamma(\EP)} \cdot \norm{x-y}^{2\EP - \dimension}
	.
	\label{eq:position-space-propagator} %
\end{align}
In particular note that for $\EP = 1$, we get the propagators
\begin{equation}
	\posprop(x)
	\defas
		\posprop[1](x)
	= \begin{cases}
			\frac{1}{4\pi^2} \norm{x}^{-2}	& \text{when $\dimension=4$ and} \\
			\frac{1}{4 \pi^3} \norm{x}^{-4}	& \text{when $\dimension=6$.} \\
		\end{cases}
	\label{eq:position-space-propagators} %
\end{equation}
Therefore, up to the replacement $\EP_e \mapsto \dimension/2 - \EP_e$ and an overall prefactor, the Fourier transform of \eqref{eq:feynman-integral-momentum} is $\widehat{\Phi}(G)$ as defined in
\begin{proposition}
	\label{prop:parametric-position-space} %
	Let $G$ be a connected graph with a partition $\vertices = \vertices_{\Tint} \cupdot \vertices_{\Text}$ into internal and external vertices.\footnote{%
		In the momentum space representation, these are
		$\vertices_{\Text} = \setexp{v \in \vertices}{\ExMom(v) \neq 0}$.
	}
	Then
	\begin{equation}
		\widehat{\Phi}(G)
		\defas
		\prod_{\mathclap{v \in \vertices_{\Tint}}} \int_{\R^{\dimension}} \frac{\dd[\dimension] x_v}{\pi^{\dimension/2}}
		\cdot
		\prod_{e \in \edges}
		\norm{x_{\target(e)} - x_{\source(e)}}^{-2 \EP_e} \!\!
		=
		\prod_{e \in \edges}
			\int_0^{\infty} \frac{\SP_e^{\EP_e - 1} \dd\SP_e}{\Gamma(\EP_e)}
		\cdot
		\frac{e^{-\widehat{\phipol}/\widehat{\psipol}}}{\widehat{\psipol}^{\dimension/2}}
		\label{eq:parametric-position-space} %
	\end{equation}%
\nomenclature[PhiTilde]{$\widehat{\Phi}(G)$}{position space Feynman rules applied to $G$, equation~\eqref{eq:parametric-position-space}\nomrefpage}%
	where
	$	\widehat{\psipol}	= \forestpolynomDual{P}	$
	for
	$ P \defas \setexp{\set{v}}{v \in \vertices_{\Text}} $
	sums all $\abs{\vertices_{\Text}}$-forests $F$ with precisely one external vertex in each connected component. The polynomial $\widehat{\phipol}$ is given by
	\begin{equation}
		\widehat{\phipol}
		=
		\sum_{\substack{v,w\in\vertices_{\Text}\\v < w}}
			\norm{x_v - x_w}^2
			\cdot
			\forestpolynomDual{P_{v,w}}
		\quad\text{with}\quad
		P_{v,w}
		\defas
		\left( P \setminus \set{\set{v},\set{w}} \right) \cupdot \set{\set{v,w}}
		,
		\label{eq:phipol-position-space} %
	\end{equation}%
\nomenclature[phi]{$\widehat{\phipol}$}{position space polynomial, equation~\eqref{eq:phipol-position-space}\nomrefpage}%
	where $\forestpolynomDual{P_{v,w}}$ sums all forests of $\abs{\vertices_{\Text}}$ components, one of which contains both $v$ and $w$.
\end{proposition}
\begin{proof}
	The Schwinger trick \eqref{eq:Schwinger-trick} introduces the integrals $\int_0^{\infty} \frac{\SP_e^{\EP_e - 1} \dd \SP_e}{\Gamma(\EP_e)}$ and the factor
	\begin{equation*}
		\exp\Big[ -\sum_e \SP_e \left( x_{e_1} - x_{e_2} \right)^2 \Big]
		=
		\exp \Big( -x^{\Transpose} \LM x \Big),
	\end{equation*}
	throughout this proof we let $\LM \defas \IM^{\Transpose} \DM \IM$ instead of \eqref{eq:laplace-matrix} (without deleting a fixed vertex $v_0$).
	When we split the vector $x = (x_v)_{v \in \vertices} = (x_{\Tint}, x_{\Text})$ and this matrix
	$
		\LM = \Big(\begin{smallmatrix}
			\LM_{\Tint} & -B \\
			-B^{\Transpose}	&	\LM_{\Text} \\
		\end{smallmatrix} \Big)
	$
	into the internal and external vertex positions, completing the square
	\begin{equation*}
		x^{\Transpose} \LM x
		= \left( x_{\Tint} - \LM_{\Tint}^{-1} B x_{\Text} \right)^{\Transpose} \LM_{\Tint} \left( x_{\Tint} - \LM_{\Tint}^{-1} B x_{\Text} \right)
		+ x_{\Text}^{\Transpose} \left(
				\LM_{\Text} - B^{\Transpose} \LM_{\Tint}^{-1} B
			\right) x_{\Text}
	\end{equation*}
	in the Gau{\ss}ian integral
	$
		\prod_{v \in \vertices_{\Tint}} \int_{\R^{\dimension}} \frac{\dd[\dimension] x_v}{\pi^{\dimension/2}}
		\cdot
		\exp\left( -x^{\Transpose} \LM x \right)
	$
	proves \eqref{eq:parametric-position-space} with
	\begin{equation}
		\widehat{\psipol} = \det \LM_{\Tint}
		\quad\text{and}\quad
		\widehat{\phipol}
		=
			\widehat{\psipol} \cdot
			x_{\Text}^{\Transpose} \left( \LM_{\Text} - B^{\Transpose} \LM_{\Tint}^{-1} B \right) x_{\Text}
		.
		\label{eq:position-polynomials-determinant} %
	\end{equation}
	Since $\det \LM_{\Tint} = \det \LM(\vertices_{\Text},\vertices_{\Text})$, we consider minors of $\LM$ from deletion of rows $W$ and columns $W'$ of vertices. To apply the matrix-tree theorem~\ref{theorem:matrix-tree}, we rewrite \eqref{eq:laplace-matrix} as
	\begin{equation*}
			\det \LM(W,W')
		= \det \DM \cdot \det
			\begin{pmatrix}
				\DM^{-1} & \IM(W') \\
				-\IM(W)^{\Transpose} & 0 \\
			\end{pmatrix}
		= \sum_{S \subseteq E} 
			\det \begin{pmatrix}
				0 & \IM(S, W') \\
				-\IM(S, W)^{\Transpose} & 0 \\
			\end{pmatrix}
			\cdot
			\prod_{e \notin S} \SP_e
	\end{equation*}
	and conclude, just as in the proof of theorem~\ref{theorem:graph-polynomials}, that
	\begin{equation}\begin{split}
		\det \LM(W, W')
		&=
			\sum_F \sigma(F) \prod_{e \in F} \SP_e
		\quad\text{with signs}
		\\
		\sigma(F)
		&=
			\det \IM(\edges \setminus F, W)
			\cdot
			\det \IM(\edges \setminus F, W')
		\in
			\set{1,-1}
		\label{eq:all-minors-matrix-tree} %
	\end{split}
	\end{equation}
	where $F$ runs over all forests that contain precisely one vertex of $W$ and one vertex of $W'$ in each connected component (i.\,e. $\abs{\comps(F)} = \abs{W} = \abs{W'}$).
	This formula is also known as the \emph{all-minors matrix-tree theorem} \cite{BognerWeinzierl:GraphPolynomials}. As an immediate consequence, we read off our claimed formula $\widehat{\psipol} = \forestpolynomDual{P}$ upon setting $W=W' = \vertices_{\Text}$.

	To interpret $\widehat{\phipol}$, just as in the proof of theorem~\ref{theorem:graph-polynomials} we compute for any $a,b \in \vertices_{\Tint}$
	\begin{equation*}
		\widehat{\psipol} \cdot (\LM_{\Tint}^{-1})_{a,b}
		= (-1)^{a+b} \cdot \det \LM(\vertices_{\Text} \cup \set{a}, \vertices_{\Text} \cup \set{b})
		= (-1)^{a+b} \sum_F \sigma(F) \prod_{e \in F} \SP_e
	\end{equation*}
	using \eqref{eq:all-minors-matrix-tree} with $W= \vertices_{\Text} \cup \set{a}$, $W' = \vertices_{\Text} \cup \set{b}$ and find $\sigma(F) = (-1)^{a+b}$. In short, $\widehat{\psipol} \cdot (\LM_{\Tint}^{-1})_{a,b} = \forestpolynomDual{\set{a,b},P}$. For distinct $v,w \in \vertices_{\Text}$ as shown in figure~\ref{fig:B-LM-B-v,w},
	\begin{equation}
		\widehat{\psipol} \cdot
		\left( B^{\Transpose} \LM_{\Tint}^{-1} B \right)_{v,w}
		= \sum_{a, b \in \vertices_{\Tint}}
			\forestpolynomDual{\set{a,b},P}
			\cdot
			\sum_{\substack{
					e=\set{v,a} \in \edges \\
					f=\set{w,b} \in \edges
			}}	\SP_e \SP_f 
		= \forestpolynomDual{P_{v,w}} + \widehat{\psipol} \cdot \left(\LM_{\Text} \right)_{v,w} 
		\label{eq:B-LM-B-v,w} %
	\end{equation}
	where we sum over (possibly multiple) edges $e,f$ connecting $v,w$ to $a,b$.
	Note that for any forest $F$ contributing to $\forestpolynomDual{\set{a,b},P}$, $F' \defas F \cupdot \set{e,f}$ is also a forest and contributes to $\forestpolynomDual{P_{v,w}}$. The last equality in \eqref{eq:B-LM-B-v,w} follows since each such $F'$ occurs exactly once as $F$ can be reconstructed from $F'$ by removing the unique edges $e,f \in F'$ that are first and last in the path connecting $v$ and $w$ in $F'$.
	Only $F'$ which contain an edge $e$ that connects the external $v$ and $w$ directly can not occur this way and must be subtracted (in this case $F' \setminus \set{e}$ are precisely the forests of $\widehat{\psipol} = \forestpolynomDual{P}$).

	Similarly we obtain (see figure~\ref{fig:B-LM-B-v,v})
	\begin{equation}
		\widehat{\psipol} \cdot
		\left( B^{\Transpose} \LM_{\Tint}^{-1} B \right)_{v,v}
		= \sum_{a, b \in \vertices_{\Tint}}
			\forestpolynomDual{\set{a,b},P}
			\cdot
			\sum_{\substack{
					e=\set{v,a} \in \edges \\
					f=\set{v,b} \in \edges
			}}	\SP_e \SP_f 
			=
			\sum_{F'} \prod_{e' \in F'} \SP_{e'} 
			\cdot
			\sum_{v \in f \in C} \SP_f
		\label{eq:B-LM-B-v,v-1}
	\end{equation}
	where $F' \defas F \cupdot \set{e}$ runs over forests whose connected components partition $\vertices_{\Text}$ into the singletons $P$ (as before $F$ shall be a forest contributing to $\forestpolynomDual{\set{a,b},P}$). The edges $f$ must connect $v$ with a vertex $b$ in the same connected component $C$ of $F'$ that $v$ lies in.
	If $b$ lies in another component $C'$ let $\set{w} = C' \cap \vertices_{\Text}$ (so $w \neq v$), then the forest $F' \cupdot \set{f}$ contributes to the partition $P_{v,w}$ such that
	\begin{equation}
		\widehat{\psipol} \cdot
		\left( B^{\Transpose} \LM_{\Tint}^{-1} B \right)_{v,v}
		=
			\widehat{\psipol} \cdot
				\left( \LM_{\Text} \right)_{v,v}
			- \sum_{w \in \vertices_{\Text} \setminus \set{v}}
				\forestpolynomDual{P_{v,w}}
		.
		\label{eq:B-LM-B-v,v-2} %
	\end{equation}
	Plugging \eqref{eq:B-LM-B-v,w} and \eqref{eq:B-LM-B-v,v-2} into \eqref{eq:position-polynomials-determinant} we finally arrive at
	\begin{equation*}
		\widehat{\phipol}
		= \sum_{v \in \vertices_{\Text}} 
				x_v^2 \cdot \sum_{w \neq v \in \vertices_{\Text}} \forestpolynomDual{P_{v,w}}
		-2\sum_{\mathclap{v<w \in \vertices_{\Text}}}
				x_v^{\Transpose} x_w \cdot \forestpolynomDual{P_{v,w}}
		=	\sum_{\mathclap{v<w \in \vertices_{\Text}}}
				\norm{x_v - x_w}^2 \cdot \forestpolynomDual{P_{v,w}}
		. \qedhere
	\end{equation*}
\end{proof}
\begin{figure}
	\subfloat[%
		Adding $e$ and $f$ to the forest $F$ in \eqref{eq:B-LM-B-v,w} yields a forest contributing to $\forestpolynomDual{P_{v,w}}$. Grey areas indicate the connected components of $F$, each of which contains precisely one external vertex.%
	]{\parbox{0.43\linewidth}{\centering%
		$\Graph[0.57]{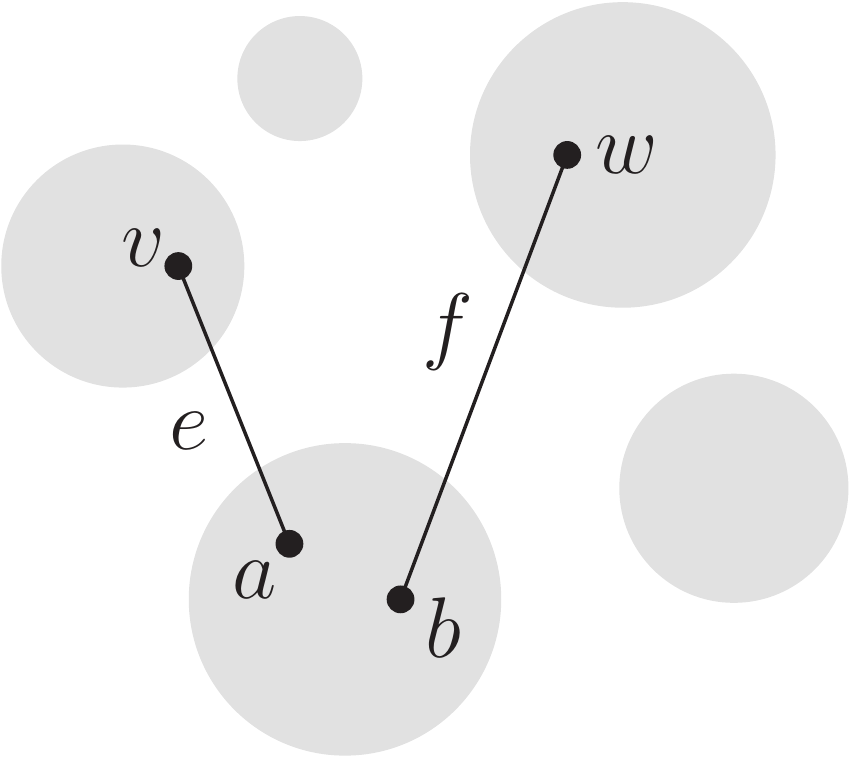}$%
		\label{fig:B-LM-B-v,w} %
	}}%
	\hfill%
	\subfloat[%
		We depict the connected components of $F'$ for \eqref{eq:B-LM-B-v,v-1}, $b$ must lie in $C$. When we extend the sum to all edges $f$ incident to $v$, additional contributions arise when $f$ connects to a different component $C'$ (indicated by the dashed line $f'$).%
	]{\parbox{0.55\linewidth}{\centering%
		$\Graph[0.5]{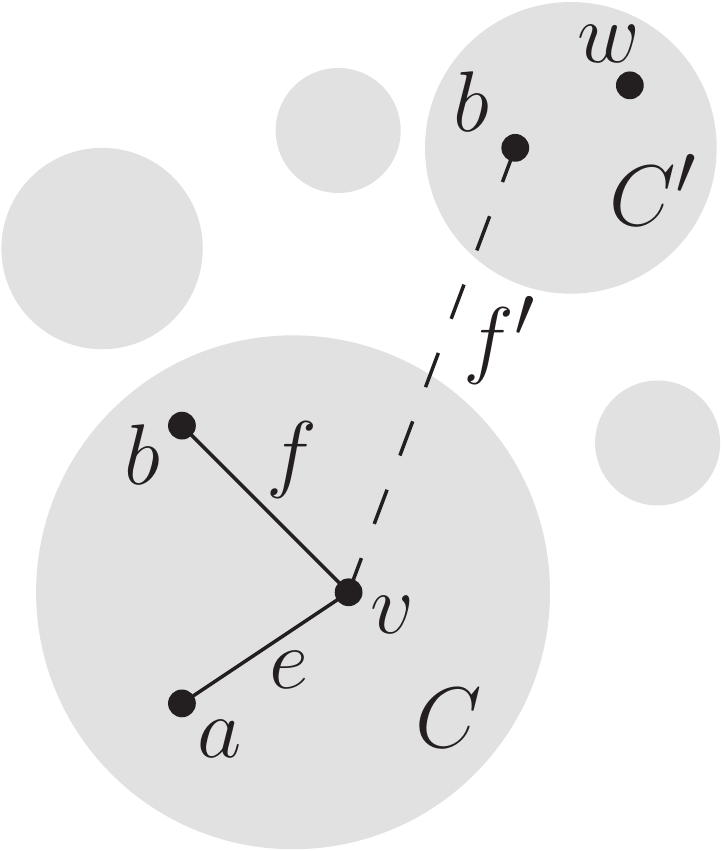}$%
		\label{fig:B-LM-B-v,v} %
	}}%
	\caption{Illustration of the proof of proposition~\ref{prop:parametric-position-space}.}%
	\label{fig:B-LM-B} %
\end{figure}
The formulas \eqref{eq:parametric-position-space} and \eqref{eq:phipol-position-space} can already be found in \cite{Nakanishi:GraphTheoryFeynmanIntegrals}. We will exploit this combinatorial description in section~\ref{sec:vw3-applications}, but note that in general the computation of $\widehat{\psipol}$ and $\widehat{\phipol}$ in terms of the determinants \eqref{eq:position-polynomials-determinant} could be more efficient than the explicit enumeration of spanning forests.
\begin{corollary}
\label{cor:position-space-projective}%
From \eqref{eq:degree-forestpolynomDual} we find 
$\deg\left( \widehat{\psipol} \right) = \abs{\vertices_{\Tint}}$ and 
$\deg\left( \widehat{\phipol} \right) = \abs{\vertices_{\Tint}}+1$, so the projective version of \eqref{eq:parametric-position-space}---the analogue of \eqref{eq:feynman-integral-projective}---reads
\begin{equation}
	\widehat{\Phi}(G)
	=	\frac{\Gamma(\widehat{\sdd})}{\prod_e \Gamma(\EP_e)}
		\int \frac{\Omega}{\widehat{\psipol}^{\dimension/2}}
			\left(\frac{\widehat{\psipol}}{\widehat{\phipol}} \right)^{\widehat{\sdd}}
			\prod_e \SP_e^{\EP_e - 1}
	\quad\text{where}\quad
	\widehat{\sdd}
	\defas
		\sum_e \EP_e - \dimension/2 \cdot \abs{\vertices_{\Tint}}
	.
	\label{eq:position-space-projective}%
\end{equation}
\end{corollary}

\begin{remark}[Dualization]
\label{rem:position-space-dual-parametric} %
	Inversion of the variables $\SP_e \mapsto \SP_e^{-1}$ transforms
	$\widehat{\psipol} \mapsto \forestpolynom{} \cdot \prod_{e} \SP_e^{-1}$
	and
	$\widehat{\phipol} \mapsto \phipol \cdot \prod_{e} \SP_e^{-1}$
	where
	$\forestpolynom{} \defas \forestpolynom{P}$
	and
	$\phipol \defas \sum_{v<w \in \vertices_{\Text}} \norm{x_a - x_b}^2 \cdot \forestpolynom{P_{a,b}}$.
	So with $\widehat{\EP}_e \defas \dimension/2 - \EP_e$, equations~\eqref{eq:parametric-position-space} and \eqref{eq:position-space-projective} take the form
	\begin{equation}
		\widehat{\Phi}(G)
		= \prod_{e} \int_0^{\infty} \frac{\SP_e^{\widehat{\EP}_e - 1} \dd \SP_e}{\Gamma(\EP_e)}
		\cdot \frac{e^{-\phipol/\forestpolynom{}}}{{\forestpolynom{}}^{\dimension/2}}
		= \frac{\Gamma(\widehat{\sdd})}{\prod_{e} \Gamma(\EP_e)}
			\int \frac{\Omega}{{\forestpolynom{}}^{\dimension/2}}
			\left( \frac{\forestpolynom{}}{\phipol} \right)^{\widehat{\sdd}}
			\prod_e \SP_e^{\widehat{\EP}_e - 1}
		.
		\label{eq:position-space-projective-dual} %
	\end{equation}
\end{remark}

\subsubsection{Graphical functions}
The translation invariance of \eqref{eq:phipol-position-space} means that we may restrict to $x_{v_0} = 0$ for a particular vertex $v_0 \in \vertices_{\Text}$. Furthermore, \eqref{eq:position-space-projective} is homogeneous like $\widehat{\Phi}(G, \Pscale x) = \Pscale^{-\widehat{\sdd}} \cdot \widehat{\Phi}(G,x)$ so it is enough to compute it for $\norm{x_{v_1}} = 1$, where $v_1 \in \vertices_{\Text} \setminus \set{v_0}$ is another external vertex.

In the case of $\abs{\vertices_{\Text}} = 3$ external vertices $\vertices_{\Text} = \set{v_0,v_1,v_z}$ this means
\begin{equation}
	\widehat{\Phi}(G) 
	=	\norm{x_{v_1} - x_{v_0}}^{-2\widehat{\sdd}} 
		\cdot 
		\gf[G](z,\bar{z})
	,
	\label{eq:position-space-graphical-function} %
\end{equation}
where $\gf[G](z,\bar{z})$ depends only on two ratios which we parametrize by two complex variables $z, \bar{z} \in \C$ subject to the conditions\footnote{For real Euclidean vectors $x_v$, they will be either conjugates $\bar{z} = \conjugate{z}$ or both real $z,\bar{z} \in \R$.}
\begin{equation}
	z \bar{z} = \frac{\norm{x_{v_z} - x_{v_0}}^2}{\norm{x_{v_1} - x_{v_0}}^2}
	\quad\text{and}\quad
	(1-z)(1-\bar{z}) = \frac{\norm{x_{v_z} - x_{v_1}}^2}{\norm{x_{v_1} - x_{v_0}}^2}
	.
	\label{eq:position-cross-ratios} %
\end{equation}
Note that $\gf[G](z,\bar{z})$ is given by formula~\eqref{eq:position-space-projective-dual} upon setting
\begin{equation}
	\phipol
	= 		\forestpolynom{\set{v_0,v_1},\set{v_z}}
			+ z\bar{z} \cdot \forestpolynom{\set{v_0,v_z},\set{v_1}}
			+ (1-z)(1-\bar{z}) \cdot \forestpolynom{\set{v_1,v_z},\set{v_0}}
	.
	\label{eq:phipol-graphical-function} %
\end{equation}
\begin{example}
	Edges between external vertices contribute to $f_G$ only through rational factors, so the triangle graph ($C_1$ in figure~\ref{fig:startriangle-added}) gives $f_{C_1} = \big[z\bar{z}(1-z)(1-\bar{z})\big]^{-1}$. 
	The first interesting function comes from the star (table~\ref{tab:star-forests}) whose forest polynomials we determined in example~\ref{ex:forestpolynom}. 
	If we set $v_0 \defas v_3$ and $v_z \defas v_2$, we find $\phipol = \SP_2 + z\bar{z} \SP_1 + (1-z)(1-\bar{z}) \SP_3$ and \eqref{eq:position-space-projective-dual} takes the form ($\widehat{\sdd} = 1$)
	\begin{equation*}\begin{split}
		\gf[\StarSymbol](z,\bar{z})
		&=
			\int_0^{\infty} \frac{\Omega}{(\SP_1 \SP_2 + \SP_1 \SP_3 + \SP_2 \SP_3)(\SP_2 + z\bar{z} \SP_1 + (1-z)(1-\bar{z}) \SP_3)}
		\\
		&=
			\int_0^{\infty} \frac{\Omega}{(\SP_1 z + \SP_2)(\SP_1 \bar{z} + \SP_2)} \log \frac{(\SP_1 + \SP_2)(z\bar{z}\SP_1 + \SP_2)}{(1-z)(1-\bar{z})\SP_1 \SP_2}
		\\
		&= \frac{1}{z-\bar{z}} \left[ 
				2 \Li_2(z) - 2\Li_2(\bar{z})
				+ \log (z\bar{z}) \log\frac{1-z}{1-\bar{z}}
			\right].
	\end{split}\end{equation*}
In the case of complex conjugated $\conjugate{z} = \bar{z}$, this is just $\gf[\StarSymbol](z,\conjugate{z}) = 2\BlochWigner(z)/\Imaginaerteil(z)$, in terms of the famous Bloch-Wigner dilogarithm function \cite{Zagier:Dilogarithm}
	\begin{equation}
		\BlochWigner
		\defas \Imaginaerteil\left( \Li_2(z) \right) + \arg(1-z) \log \abs{z}.
		\label{eq:BlochWigner}%
	\end{equation}%
	Note that its symmetries
	$\BlochWigner(z)
		= \BlochWigner(1-1/z)
		= \BlochWigner(1/(1-z))
		= - \BlochWigner(-z/(1-z))
		= - \BlochWigner(1-z)
		= - \BlochWigner(1/z)
	$ follow immediately from the integral representation of $\gf[\StarSymbol]$.
\end{example}
These \emph{graphical functions} $\gf[G]$ were recently introduced in \cite{Schnetz:GraphicalFunctions} and are very interesting for at least the following reasons:
\begin{enumerate}
	\item For complex conjugate $\conjugate{z} = \bar{z}$ they are single-valued real-analytic functions on $\C \setminus \set{0,1}$ and in many aspects behave similar to analytic functions.\footnote{For example, convergent integrals $\int_{\C} f(z,\conjugate{z})\ \dd z \wedge \dd \conjugate{z}$ can be computed by a residue theorem.}
	\item Often $\gf[G](z,\bar{z})$ can be computed explicitly in terms of multiple polylogarithms and a rich set of tools is available to perform such calculations.
	\item They are extremely powerful to evaluate vacuum periods in scalar field theory.
	\item Up to a rational prefactor, conformally invariant four-point integrals evaluate to graphical functions.
\end{enumerate}
For the first three points we refer to~\cite{Schnetz:GraphicalFunctions}; the application to (conformally invariant) supersymmetric Yang-Mills theory was demonstrated in~\cite{DDEHPS:LeadingSingularitiesOffShellConformal}. We will come back to the computation of periods in $\fieldphi^4$-theory in section~\ref{sec:ex-periods}.

\subsection{Tensor integrals}\label{sec:tensor-integrals}
Physical theories that contain not only scalar particles but also fields of higher spin (fermions, vector bosons, gravitons) lead to more general Feynman rules that introduce products $\prod_i k_{e_i}^{\LI_i}$ of momenta into the numerator of \eqref{eq:feynman-integral-momentum}, where $1\leq\LI_i\leq\dimension$ denote space-time indices.

Such \emph{tensor integrals} admit a Schwinger parametrization as well and explicit formulas are well-known \cite{Nakanishi:GraphTheoryFeynmanIntegrals,Tarasov:ConnectionBetweenFeynmanIntegrals,BergereZuber:RenormalizationFeynmanParametric,Smirnov:EvaluatingFeynmanIntegrals}. These parametric integrands have the form $P/(\psipol^n \phipol^m)$, where $P$ denotes some polynomial in the Schwinger variables $\SP_e$ and the exponents $n = \dimension/2 - \sdd + \delta_n$, $m=\sdd+\delta_m$ are shifted from their values for the scalar integral by integers $\delta_n,\delta_m \in \N_0$. Each monomial in $P$ thus gives the parametric integrand of the scalar integral but in dimension $\dimension + 2(\delta_n + \delta_m)$ and with shifted indices $\EP_e$.

Therefore, tensor integrals are just linear combinations of scalar integrals\footnote{In particular, there is no need to worry about \emph{irreducible scalar products} which are necessary if one does not allow for shifted dimensions.} and we do not need to discuss them any further. For completeness, let us still recall the idea behind their parametric representation, following \cite{KreimerSarsSuijlekom}.

\begin{remark}
	Tensor integrals carry highly non-trivial structures though. Most prominently, the sum of all contributions to the QED $\beta$-function features cancellations of transcendental numbers which appear in the individual diagrams \cite{BroadhurstDelbourgoKreimer:Unknotting}.
	Despite many efforts, this phenomenon still seems far from being understood.
\end{remark}

\subsubsection{Derivatives and auxiliary momenta}
We assign an auxiliary momentum $\AUX_e \in \R^{\dimension}$ to every edge and set $\ExMom(v) \defas - \sum_e \IM_{e,v} \AUX_e$. This already incorporates momentum conservation and we consider the scalar integral 
\begin{equation}
	\FR(G)
	=
	\prod_{e\in\edges}
			\int_{\R^{\dimension}} \frac{\dd[\dimension] k_e}{\pi^{\dimension/2}} \left( [k_e + \AUX_e]^2 + m_e^2 \right)^{-\EP_e}
	\prod_{v\in\vertices\setminus\set{v_0}}
			\pi^{\dimension/2}
			\delta^{(\dimension)}\left( \sum_{e\in\edges} \IM_{e,v} k_e \right)
	\label{eq:feynman-integral-momentum-tensor}%
\end{equation}
as a function of the unconstrained variables $\setexp{\AUX_e}{e\in\edges}$. A momentum in the numerator can be generated with the differential operator $\hat{\AUX}_{e,\LI} \defas -\frac{1}{2\SP_e} \frac{\partial}{\partial\AUX_{e}^{\LI}}$, because
\begin{equation*}
	\hat{\AUX}_{e,\LI}
	\frac{1}{\left[ (k_e + \AUX_e)^2 + m_e^2 \right]^{\EP_e}}
	=
	\frac{\EP_e}{\SP_e} \frac{(k_e + \AUX_e)_{\LI}}{\left[ (k_e + \AUX_e)^2 + m_e^2 \right]^{\EP_e+1}}.
\end{equation*}
In the parametric representation, the factor $\EP_e/\SP_e$ reduces $\SP_e^{\EP_e}/\Gamma(\EP_e+1)$ back to $\SP_e^{\EP_e-1} / \Gamma(\EP_e)$ and hence we can compute the tensor integral by replacing each numerator momentum $k_{e_i}^{\LI_i}$ with $\hat{\AUX}_{e_i}^{\LI_i}$ and let this operator act on the scalar integrand $I$ from \eqref{eq:projective-delta-form-integrand}.\footnote{Care is needed when some edge $e_i = e_j$ appears twice in the numerator, as additional \emph{Leibniz terms} need to be subtracted off again.} All one needs for this computation is the relation
\begin{equation}
	\hat{\AUX}_e^{\LI} 
	\phipol
	= 
	- \AUX_e^{\LI} \restrict{\psipol}{\SP_e = 0} + \sum_{f \neq e} (-1)^{e+f}\dodgson^{e,f} \SP_f \AUX_{f}^{\LI}
	\label{eq:tensor-differential-phipol}%
\end{equation}
in terms of the metric tensor $g^{\mu,\nu}$ ($=\delta_{\mu,\nu}$ in the Euclidean case) and the Dodgson polynomial $\dodgson^{e,f}$ introduced in definition~\ref{def:dodgson}. In terms of spanning forest polynomials,
\begin{equation}
	(-1)^{e+f+1} \dodgson^{e,f}
	= \forestpolynom{\set{\source(e),\source{f}},\set{\target(e),\target(f)}}
	- \forestpolynom{\set{\source(e),\target{f}},\set{\target(e),\source(f)}}
	\label{eq:dodgson-as-spanning-2-forest}%
\end{equation}
sums all forests $F$ such that both $F \cupdot \set{e}$ and $F \cupdot \set{f}$ are spanning trees, with a positive sign if $e$ and $f$ connect the two components of $F$ in the same direction and a negative sign otherwise. Worked examples can be found in \cite{KreimerSarsSuijlekom}.

\section{Divergences and analytic regularization}\label{sec:divergences}
The singularities of Feynman integrals $\FR(G,\Kinematics,\EP,\dimension)$ as functions of the \hypertarget{eqkinematics}{kinematics}
$\Kinematics = \set{m_e^2} \cup \set{\ExMom^2(F)}$
\label{eqkinematics}%
(internal masses and external momenta), the indices $\EP_e$ and the dimension $\dimension$ of space-time have been studied in great detail and are perfectly understood for Euclidean kinematics. It is well-known that plain power counting suffices to study the convergence of a Feynman integral, in the momentum \eqref{eq:feynman-integral-momentum} as well as the parametric representation \eqref{eq:feynman-integral-projective}.
This simplicity (combined with combinatorics of graphs) is in fact crucial to prove the renormalizability of a quantum field theory, but we will mostly be concerned with the computation of individual Feynman integrals in this thesis and comment on renormalization only in sections~\ref{sec:renormalization} and \ref{sec:ex-renormalized-parametric}.

As we shall recall below, absolute convergence of $\FR(G)$ is guaranteed in a non-empty domain $\Lambda_{G} \subset \C^{\abs{\edges}+1}$ (bounded by linear inequalities) of values $(\EP,\dimension) \in \Lambda_G$ of the indices and the dimension. Strikingly, the analytic continuation of $\FR(G)$ in these variables defines a meromorphic function on $\C^{\abs{\edges}+1}$ with singularities on hyperplanes. This \emph{analytic regularization} has been studied (dominantly in the parametric representation) in great detail both purely mathematically \cite{Speer:SingularityStructureGenericFeynmanAmplitudes,Speer:GeneralizedAmplitudes} and with a view towards physics, for example through the \emph{dimensional renormalization} scheme \cite{BreitenlohnerMaison:DimRenAction,BreitenlohnerMaison:DimRenMasslessI,BreitenlohnerMaison:DimRenMasslessII}. The special case of \emph{dimensional regularization} (keeping $\EP_e$ fixed and studying the dependence on $\dimension$ only) became particularly popular in the momentum space representation \cite{tHooftVeltman:RegularizationGaugeFields} and underlies the majority of all exact computations of Feynman integrals accomplished so far.

In section~\ref{sec:anareg} we show how this analytic continuation can be implemented directly on the level of the parametric integrand, which yields a representation of divergent Feynman integrals in terms of convergent ones. This relation extends the applicability of hyperlogarithms to singular, analytically regularized Feynman integrals. It is also interesting in itself and might be useful for other techniques as well. For example we will relate it to \emph{sector decomposition}.

In addition we will comment on restrictions and open problems in the case of Minkowski kinematics (in this metric, momentum squares $\ExMom(F)^2$ can be negative and may introduce additional and more complicated singularities).

\subsection{Euclidean power counting}
\subsubsection{Ultraviolet (UV) divergences}
We need conditions that guarantee absolute convergence of the Feynman integrals $\FR(G)$. These are easiest to obtain in the fully massive case, where the integrand of the momentum space representation \eqref{eq:feynman-integral-momentum} is smooth and divergences can arise only from the integration over large momenta. Indeed, in his excellent article \cite{Weinberg:HighEnergy} Weinberg proves
\begin{theorem}
	\label{theorem:Weinberg} %
	The scalar Feynman integral $\FR(G)$ from \eqref{eq:feynman-integral-momentum} is absolutely convergent provided that all propagators are massive $m_e > 0$ and that for all 1PI $\gamma \subseteq \edges$,
	\begin{equation}
		\sdd(\gamma)
		\urel{\eqref{eq:def-sdd}}
		\sum_{e \in \gamma} \EP_e - \frac{\dimension}{2} \loops{\gamma}
		> 0
		.
		\label{eq:UV-sdd} %
	\end{equation}
\end{theorem}
We call $\sdd(\gamma)$ the \emph{superficial degree of ultraviolet divergence} of the subgraph $\gamma$. It describes the contribution to the integral \eqref{eq:feynman-integral-momentum} from the domain where all $k_e =  k_e' \sqrt{\Pscale}$ for edges $e \in \gamma$ approach infinity jointly as $\Pscale \rightarrow \infty$, while $k_e$ ($e \notin \gamma$) and $k_e'$ ($e \in \gamma$) stay fixed. There the integrand falls off like $\Pscale^{-\sum_{e \in \gamma} \EP_e}$ and the rescaling of $\loops{\gamma}$ independent loop momenta contributes $\Pscale^{\loops{\gamma} \dimension/2}$, so \eqref{eq:UV-sdd} is clearly necessary for absolute convergence (the content of the theorem is the non-trivial sufficiency of this condition).

Note that $\eqref{eq:UV-sdd}$ only needs to hold for $\gamma$ that are 1PI, since any $k_e$ ($e \in \gamma$) not contained in a loop in $\gamma$ is fixed by momentum conservation in terms of the $k_e$ with $e \notin \gamma$ and external $\ExMom(v)$'s.

A proof of theorem~\ref{theorem:Weinberg} in the parametric representation \eqref{eq:feynman-integral-parametric} is sketched in \cite{BergereLam:BPAlpha}, where also the renormalization is addressed directly in the parametric representation. We like to point out the modern treatment of ultraviolet divergences and their renormalization from the viewpoint of algebraic geometry \cite{BrownKreimer:AnglesScales}.
In fact, for our purpose it is important to understand the result \eqref{eq:UV-sdd} in the Schwinger parameters, where ultraviolet divergences correspond to singularities in \eqref{eq:feynman-integral-parametric} when $\SP_e \rightarrow 0$ for $e \in \gamma$. We will come back to this after mentioning the situation with vanishing masses.

\begin{example}
	\label{ex:dunce-divergences} %
	Take $G$ the dunce's cap from figure~\ref{fig:dunce-cocommutative-bijection} with unit indices $\EP_1 = \EP_2 = \EP_3 = \EP_4 = 1$ in $\dimension = 4 - 2\varepsilon$ dimensions. From $\sdd = 4 - \dimension = 2\varepsilon$ we see that $\FR(G)$ is superficially \emph{logarithmically} divergent ($\restrict{\sdd}{\varepsilon=0} = 0$).
	There is a single logarithmic ultraviolet subdivergence formed by the edges $3$ and $4$:
	\begin{equation*}
		\sdd(\set{3,4})
		= \EP_3 + \EP_4 - \dimension/2
		= \varepsilon
		.
	\end{equation*}
	Convergence of $\FR(G)$ therefore requires $\varepsilon>0$.
\end{example}
\begin{example}
	\label{ex:2loop-master-d6-divergences} %
	Consider the two-loop master integral $F$ from figure~\ref{fig:low-loop-masters} with unit indices $\EP_1=\cdots=\EP_5 = 1$ in $\dimension = 6 - 2\varepsilon$ dimensions. The superficial degree of ultraviolet divergence of $G$ is $\sdd(G) = 5 - 2\cdot (3-\varepsilon) = -1 + 2\varepsilon$ and we call $G$ \emph{quadratically} divergent ($\restrict{\sdd}{\varepsilon=0} = -1$).

	Furthermore, we find two logarithmic ultraviolet subdivergences, namely
	\begin{equation*}
		\sdd\left( \set{1,4,5} \right)
		=
		\sdd\left( \set{2,3,5} \right)
		=
		3 - (3 - \varepsilon)
		= \varepsilon
	\end{equation*}
	which are called \emph{overlapping}, since $\set{1,4,5} \cap \set{2,3,5} = \set{5} \neq \emptyset$. 
	Convergence of $\FR(G)$ in the momentum space \eqref{eq:feynman-integral-momentum} or parametric representations \eqref{eq:feynman-integral-parametric} requires $\varepsilon > 1/2$. Note that the projective integral \eqref{eq:feynman-integral-projective} converges already for $\varepsilon>0$, because the overall divergence is captured in the factor $\Gamma(\sdd)$.
\end{example}
\subsubsection{Infrared (IR) divergences}
A different type of divergence can appear only when a graph $G$ contains massless propagators $m_e = 0$. These are very important and ubiquitous in the calculations of physical scattering amplitudes, because fully massive graphs ($m_e >0$ for all edges $e$) are actually very rare amidst the abundance of graphs containing one or more massless propagators. These arise for example from
\begin{itemize}
	\item
		massless gauge bosons (gluons and photons) in the Standard model and
		
	\item
		approximate computations where some masses (light quarks or leptons, in particular neutrinos) are considered negligibly small in comparison to other scales in the process (masses of heavy quarks, leptons or $W^{\pm}$ and $Z$ bosons).
\end{itemize}
While vanishing masses $m_e = 0$ simplify the kinematic dependence of a Feynman integral compared to the massive case, they also introduce a divergence of the propagator $(k_e^2 + m_e^2)^{-\EP_e} = k_e^{-2 \EP_e}$ at $k_e =0$ which can be non-integrable.
An extension of Weinberg's theorem~\ref{theorem:Weinberg} to this case was worked out for example in \cite{LowensteinZimmermann:PowerCountingMassless}, where we find
\begin{theorem}
	\label{theorem:LowensteinZimmermann} %
	The scalar Feynman integral $\FR(G)$ from \eqref{eq:feynman-integral-momentum} is absolutely convergent, given that $\sdd(\gamma)>0$ for any 1PI $\gamma \subseteq \edges$ and furthermore
	\begin{equation}
		-\sdd(G / \gamma)
		=
			\frac{\dimension}{2} \loops{G / \gamma}
			-\sum_{e \notin \gamma} \EP_e
		= \sdd(\gamma) - \sdd(G)
		> 0
		\label{eq:IR-sdd} %
	\end{equation}
	for all $\gamma \subseteq \edges$ that contain all massive edges ($m_e \neq 0 \Rightarrow e \in \gamma$) and connect all external vertices 
	$
		\vertices_{\Text}
		\defas
		\setexp{v \in \vertices}{\ExMom(v) \neq 0}
	$ with each other (all $\vertices_{\Text}$ lie in the same connected component of $\gamma$) such that $G/\gamma$ is 1PI.
\end{theorem}
We call $-\sdd(G/\gamma)$ the \emph{superficial degree of infrared divergence} associated to the subgraph $\gamma^c = \edges \setminus \gamma$.
It is the leading power of $\Pscale \rightarrow 0$ after rescaling $k_e = k_e' \sqrt{\Pscale}$ for $e \in \gamma^c$ (thus $m_e = 0$) and fixed values of $k_e'$ ($e \in \gamma^c$) and $k_e$ ($e \notin \gamma^c$).
Note that momentum conservation allows $k_e \rightarrow 0$ for all $e \in \gamma^c$ only when no momentum flows through $G/\gamma$ (that means all external momenta enter at the same vertex in $G/\gamma$ and therefore sum to zero).
\begin{figure}
	\centering
	$M = \Graph[0.36]{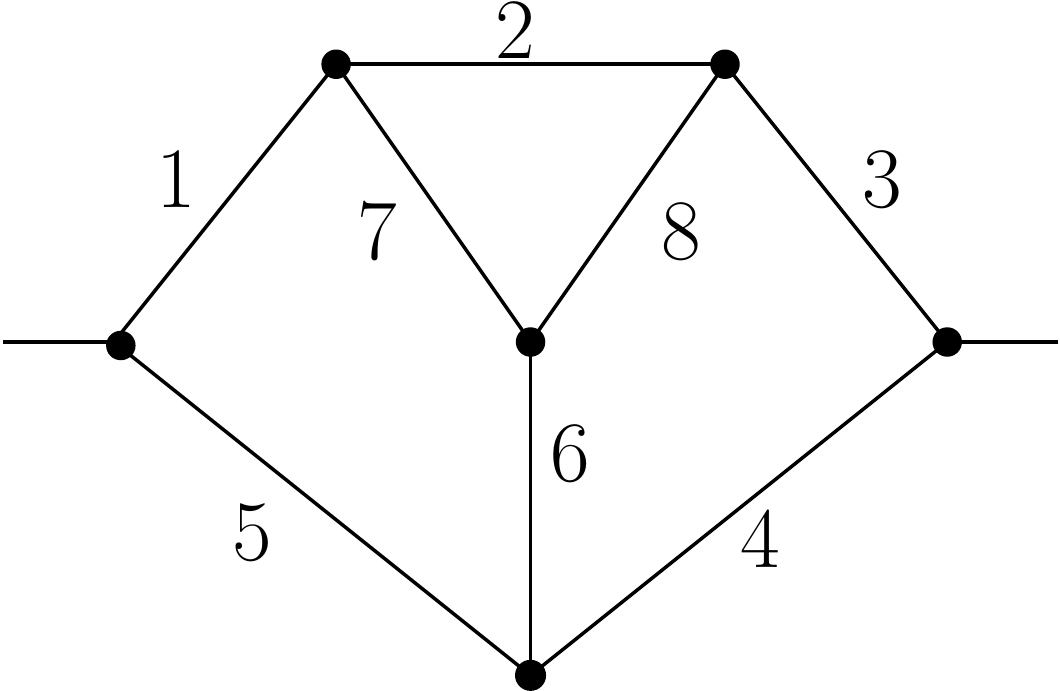}$ \qquad
	$\gamma = \Graph[0.46]{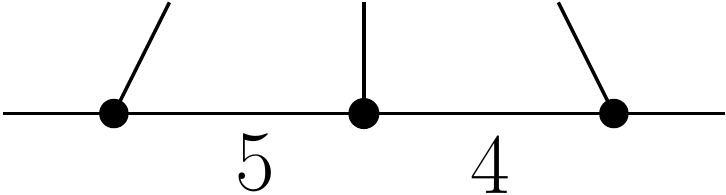}$ \qquad
	$M/\gamma = \Graph[0.4]{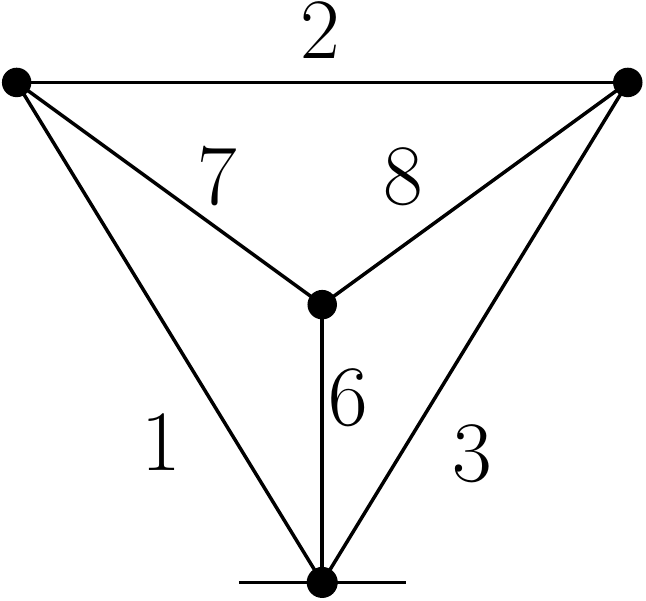}$
	\caption{The Mercedes graph and its infrared subdivergence.}%
	\label{fig:mercedes}%
\end{figure}
\begin{example}
	\label{ex:mercedes-divergences} %
	Consider the Mercedes (or Benz) graph $M$ of figure~\ref{fig:mercedes} in $\dimension = 4 - 2\varepsilon$ with massless propagators $m_e = 0$ and unit indices $\EP_e = 1$ for all edges $e$. It turns out that $\sdd(\gamma)>0$ (even when $\varepsilon = 0$) for all $\emptyset \neq \gamma \subseteq \edges$, so $M$ is ultraviolet-finite. But for the subgraph $\gamma = \set{4,5}$ we find (at $\varepsilon = 0$) a logarithmic infrared divergence
	\begin{equation*}
		\sdd(G / \gamma)
		= 6 - 3\cdot (2-\varepsilon)
		=	3 \varepsilon.
	\end{equation*}
	Graphically it corresponds to the co-graph $G/\gamma$ which is scaleless (no masses or external momenta, because all momenta enter at the same vertex and sum to zero by momentum conservation) as shown in figure~\ref{fig:mercedes}. So convergence of $\FR(G)$ requires $\varepsilon<0$.
\end{example}
In the parametric representation \eqref{eq:feynman-integral-parametric}, such an infrared divergence manifests itself at \mbox{$\SP_e \rightarrow \infty$} for $e \notin \gamma$ (see lemma~\ref{lemma:vanishing-degree-UV} and remark~\ref{rem:divergences-from-affine-to-projective} below).
A very detailed discussion and proof of theorem~\ref{theorem:LowensteinZimmermann} in the parametric representation can be found in \cite{Speer:SingularityStructureGenericFeynmanAmplitudes}, while we also recommend the instructive short exposition in \cite[section~4.4]{Smirnov:AnalyticToolsForFeynmanIntegrals}.

\subsection{Scaling degrees for Schwinger parameters}

To understand the above convergence criteria in the parametric representation \eqref{eq:feynman-integral-projective}, we investigate how its integrand \eqref{eq:projective-delta-form-integrand},
\begin{equation*}
	I_G
	= \psipol^{-\dimension/2}
		\left(\psipol/\phipol \right)^{\sdd(G)}
		\prod_{e} \SP_e^{\EP_e -1},
\end{equation*}
scales when a subset of variables \mbox{$\SP_e \rightarrow 0$} approaches zero jointly in
\begin{definition}
	\label{def:vanishing-degree} %
	For any $\gamma \subseteq \edges$ and a suitable\footnote{We think of some $I_G$, but our arguments hold for more general integrands (see remark~\ref{rem:more-general-integrands}).} parametric integrand $I$, let
	$
		I^{(\gamma)}
		\defas \restrict{I}{\SP_e = \Pscale \SP_e' \ \forall e\in \gamma}
	$
	denote $I$ after rescaling all $\SP_e$ with $e \in \gamma$ by a number $\Pscale$. The \emph{vanishing degree} $\deg_{\gamma}\left( I \right)$ is the unique number such that
	\begin{equation}
		I^{(\gamma)}
			\in
			\bigo{\lambda^{\deg_{\gamma}(I)}}
		,\quad\text{by which we mean that}\quad
		\lim_{\Pscale \rightarrow 0}
			\left[
				I^{(\gamma)} \cdot \Pscale^{-\deg_{\gamma}(I)}
			\right]
			\neq 0,\infty
		\label{eq:def-vanishing-degree} %
	\end{equation}
	is finite and non-zero.
\end{definition}
\begin{remark}
For functions $f_i$ with exponents $n_i$, the vanishing degree follows
\begin{equation}
		\deg_{\gamma} \left( \prod_i f_i^{n_i} \right)
		= \sum_i n_i \deg_{\gamma}\left( f_i \right)
		\quad\text{and}\quad
		\deg_{\gamma} \left( \sum_i f_i \right)
		\geq \min_i \left[ \deg_{\gamma} \left( f_i \right) \right]
		,
		\label{eq:vanishing-degree-rules} %
\end{equation}
where an inequality can occur when the individual leading contributions of different $f_i$ cancel each other in the sum. However, if $f = \sum_{n} a_{n} \SP^{n}$ is a polynomial and $\SP^n = \prod_e \SP_e^{n_e}$ denotes its distinct monomials ($n \in \N_0^{\edges}$), then there can be no such cancellation and
\begin{equation}
	\deg_{\gamma} (f)
	=
		\min_{n\colon a_n \neq 0} \deg_{\gamma} \left( \SP^n \right)
	= \min_{n\colon a_n \neq 0} \Bigg( \sum_{e \in \gamma} n_e \Bigg)
	.
	\label{eq:vanishing-degree-polynomial} %
\end{equation}
\end{remark}
\begin{lemma}
	\label{lemma:vanishing-degree-UV}%
	For $\gamma \subseteq \edges$, the integrand \eqref{eq:projective-delta-form-integrand} scales as
	\begin{equation}
		\deg_{\gamma} \left( I \right) + \abs{\gamma}
		= \begin{cases}
				- \sdd\left( G / \gamma \right) & \text{if $G/\gamma$ is $0$-scale and} \\
				\phantom{-}\sdd(\gamma) & \text{otherwise.} \\
			\end{cases}
		\label{eq:vanishing-degree-UV} %
	\end{equation}
	Here a graph $Q = G/\gamma$ is called \emph{$0$-scale} (or \emph{tadpole}) if it does not depend on any kinematic invariants (all internal masses and external momenta of $Q$ vanish), equivalently $\phipol_Q = 0$.
\end{lemma}
\begin{proof}
	First of all, we apply \eqref{eq:vanishing-degree-rules} to \eqref{eq:projective-delta-form-integrand} and find
	\begin{equation}
		\deg_{\gamma} \left( I \right) 
		= -\frac{\dimension}{2} \cdot \deg_{\gamma} (\psipol) 
			- \sdd(G) \cdot \left( \deg_{\gamma} \phipol - \deg_{\gamma} \psipol \right)
			+ \sum_{e \in \gamma} (\EP_e - 1)
		.
		\label{eq:vanishing-degree-integrand} %
	\end{equation}
	Spanning trees $T$ and spanning two-forests $F$ can share at most $\abs{\gamma} - \loops{\gamma}$ edges with $\gamma$(otherwise they would contain a loop), and furthermore this maximum is attained%
\footnote{%
	First remove $\loops{\gamma}$ suitable edges from $\gamma$ to obtain a forest $F \subseteq \gamma$, then add adequate edges $K \subset\edges \setminus \gamma$ to construct such a tree $T = F \cupdot K$ with $\abs{T \cap \gamma} = \abs{\gamma} - \loops{\gamma}$.
}
	in
	\begin{equation}
		\deg_{\gamma} \left( \psipol \right)
		\urel{\eqref{eq:graph-polynomials}}
			\min_T \abs{\gamma \setminus T}
		= \abs{\gamma} - \max_T \abs{\gamma \cap T}
		= \loops{\gamma}
		.
		\label{eq:vanishing-degree-psi} %
	\end{equation}
	For the second Symanzik polynomial $\phipol$, no cancellations are possible between monomials multiplying a mass and those stemming from a two-forest $F$, as all kinematic invariants $m_e, \ExMom^2(F) \geq 0$ are non-negative (Euclidean momenta). Therefore
	\begin{equation}
		\deg_{\gamma} \left( \phipol \right)
		\urel{\eqref{eq:graph-polynomials}}
		\min\set{
			\min_{F:\ \ExMom^2(F) \neq 0} \abs{\gamma \setminus F},
			\loops{\gamma} + \deg_{\gamma} \sum_{e \in \edges} m_e \SP_e
		}
		\geq
			\loops{\gamma}
		\label{eq:vanishing-degree-phi} %
	\end{equation}
	and equality holds only when $G/\gamma$ has a scale (is not $0$-scale), because this means that
	\begin{itemize}
		\item
			there exists a massive ($m_{e'}>0$) edge $e' \in G/\gamma$, that means $e' \notin \gamma$, wherefore
			$\deg_{\gamma} \sum_{e \in \edges} m_e \SP_e = 0$,
			
		\item
			or $G/\gamma$ admits a two-forest $F'$ such that $\ExMom^2(F') \neq 0$. But such $F'$ can be augmented to a two-forest $F = F' \cupdot F''$ of $G$ with $\ExMom^2(F) = \ExMom^2(F') \neq 0$ by selecting a spanning $\abs{\comps(\gamma)}$-forest $F'' \subseteq \gamma$ of $\gamma$. As $\abs{F''} = \abs{\gamma} - \loops{\gamma}$ and $F' \cap \gamma = \emptyset$, indeed $\abs{\gamma \setminus F} = \loops{\gamma}$.
	\end{itemize}
	Otherwise $G/\gamma$ is $0$-scale, so $\gamma \subseteq \setexp{e \in \edges}{m_e>0}$ must contain all massive edges and \mbox{$\deg_{\gamma}\left( \phipol \right) \geq \loops{\gamma} + 1$}. If there is at least one massive edge $e \in \edges$ ($m_e>0$) at all, then $\deg_{\gamma} \sum_{e\in\edges} m_e \SP_e = 1$ and $\deg_{\gamma}\left( \phipol \right) = \loops{\gamma} + 1$.

	But when all edges are massless ($m_e = 0$), the massive contribution in \eqref{eq:vanishing-degree-phi} is absent altogether and $\deg_{\gamma}\left( \phipol \right) = \min_{F:\ \ExMom^2(F) \neq 0} \abs{\gamma \setminus F}$. Now $\gamma$ must contain all external vertices $\vertices_{\Text}$ in the same connectivity component $C$ as $G / \gamma$ is assumed to be $0$-scale. After removing $\loops{\gamma}$ edges from $\gamma$ to obtain a spanning forest $F' \subseteq \gamma$ with the same components $\comps\left( F' \right) = \comps(\gamma)$, we can further remove a suitable edge $e \in F'\cap \edges(C)$ such that both components of $C \setminus e$ contain at least one external vertex. Then $F' \setminus e$ can be extended to a two-forest $F$ of $G$ (by adding edges from $\edges \setminus \gamma$) with $\ExMom^2(F) \neq 0$ and $\abs{\gamma \setminus F} = \loops{\gamma} + 1$.

	We conclude that for scaleful $G/\gamma$, \eqref{eq:vanishing-degree-integrand} reduces to $\sdd(\gamma) - \abs{\gamma}$ since $\deg_{\gamma} \psipol = \deg_{\gamma} \phipol$. For $0$-scale $G/\gamma$, we must replace $\sdd(\gamma)$ by $\sdd(\gamma) - \sdd(G) = - \sdd(G/\gamma)$.
\end{proof}
Note that in this setup of Euclidean momenta, $G/\gamma$ is $0$-scale precisely when $\gamma$ comprises all massive edges and furthermore contains all external vertices $\vertices_{\Text}$ in the same connected component.
\begin{remark}
Let $\emptyset \neq \gamma \subsetneq \edges$ and insert the factor $1 = \int_0^{\infty} \dd \Pscale\ \delta\left( \Pscale - \sum_{e \in \gamma} \SP_e \right)$ into the projective representation \eqref{eq:feynman-integral-projective}. The substitution of $\SP_e$ with $\Pscale\SP_e$ for all $e \in \gamma$ shows
\begin{equation}
	\int I \ \Omega
	= \int \Omega\ \delta\left( 1-\sum_{e \in \gamma} \SP_e \right)
		\int_0^{\infty} \frac{\dd \Pscale}{\Pscale} \Pscale^{\abs{\gamma} + \deg_{\gamma}(I)}
			\cdot \widetilde{I^{(\gamma)}},
	\label{eq:uvdiv-delta-inserted} %
\end{equation}
where $\widetilde{I^{(\gamma)}} \defas I^{(\gamma)} \cdot \Pscale^{-\deg_{\gamma}(I)}$ is finite at $\Pscale \rightarrow 0$. Therefore $\abs{\gamma} + \deg_{\gamma}(I) > 0$ is apparently necessary for the absolute convergence of $\FR(G)$. The content of theorem~\ref{theorem:LowensteinZimmermann} and lemma~\ref{lemma:vanishing-degree-UV} lies in the sufficiency of this simple criterion.
\end{remark}
\begin{corollary}[Euclidean convergence]
	\label{corollary:projective-convergence-euclidean}%
	With non-exceptional Euclidean external momenta, the projective integral \eqref{eq:feynman-integral-projective} is absolutely convergent precisely when $\abs{\gamma} + \deg_{\gamma}(I) > 0$ for all $\emptyset \neq \gamma \subsetneq \edges$.
\end{corollary}
Note that apart from 1PI graphs, we must also consider individual edges $\gamma = \set{e}$ to ensure convergence ($\EP_e>0$) of the Schwinger trick \eqref{eq:Schwinger-trick}. The condition $\sdd(G)>0$ is needed for the parametric representation \eqref{eq:feynman-integral-parametric} but not for finiteness of the projective integral \eqref{eq:feynman-integral-projective} where it is already integrated out and captured by the prefactor $\Gamma(\sdd(G))$.
The restriction to Euclidean momenta precisely requires that
\begin{itemize}
	\item all masses $m_e \geq 0$ are non-negative and
	\item $\big[\sum_{v \in W} \ExMom(v) \big]^2 > 0$ for any $\emptyset \neq W \subsetneq \vertices_{\Text}$.
\end{itemize}
This \emph{non-exceptional configuration of momenta} was used in the proof of lemma~\ref{lemma:vanishing-degree-UV}.
\begin{remark}\label{rem:divergences-from-affine-to-projective}
	A UV-divergence at $\SP_e \rightarrow 0$ for $e \in \gamma$ in the affine integral~\eqref{eq:feynman-integral-parametric} is localized, in the projective representation~\eqref{eq:feynman-integral-projective}, on the subset
	\begin{equation*}
		\setexp{\big[\SP_1:\cdots:\SP_{\abs{E}}\big]}
		{\text{
			$\SP_e =0$ for $e\in\gamma$ and $\SP_e \in \R_+$ otherwise
		}}
		\subset
		\RP_+^{\abs{E}-1}
		.
	\end{equation*}
	But also a potential IR-divergence ($\SP_e \rightarrow \infty$ for $e \notin \gamma$) in~\eqref{eq:feynman-integral-parametric} corresponds, projectively, to a singularity on this same set,	since the projective form $I\ \Omega$ is invariant under simultaneous rescaling of all Schwinger variables:
	Substituting $\SP_e = \SP_e' \cdot \Pscale^{-1}$ for $e \in \gamma$ gives the same vanishing degree for small $\Pscale$ as rescaling $\SP_e = \SP_e' \cdot \Pscale$ for those $e \notin \gamma$.

	So we note that while UV- and IR-divergences are separated ($\SP_e \rightarrow0$ versus $\SP_e \rightarrow \infty$) in~\eqref{eq:feynman-integral-parametric}, they are treated uniformly in the projective form (see corollary~\ref{corollary:projective-convergence-euclidean}).
\end{remark}

\subsection{Non-Euclidean momenta}\label{sec:non-Euclidean}
Quantum field theory is formulated in Minkowski space, so finally results computed in the convenient Euclidean region must be analytically continued back to the physical region.\footnote{This is sometimes referred to as \emph{Wick rotation}, though many physics textbooks unfortunately confuse this term with a mere deformation of the integration contour for the timelike momentum components. However, it is a true analytic continuation.} This continuation does not pose a problem, but many kinematical configurations impossible to realize in the Euclidean region.

A very typical example concerns lightlike external momenta $\ExMom(v)^2 = 0$ (on-shell massless particles), which in the Euclidean metric always implies $\ExMom(v) = 0$ and thus no dependence on $\ExMom(v)$ whatsoever. Not so for the Minkowski metric. For example, a three-point graph (like $C_2$ from figure~\ref{fig:triladder-contraction-deletion}) can have two lightlike external momenta $\ExMom(v_1)^2 = \ExMom(v_2)^2 = 0$ and still depend on the free variable $\ExMom(v_3)^2 = [\ExMom(v_1) + \ExMom(v_2)]^2  = 2 \ExMom(v_1) \ExMom(v_2)$, while in Euclidean metric they would impose $\ExMom(v_1) = \ExMom(v_2) = 0$ and thus $\ExMom(v_3) = -\ExMom(v_1)-\ExMom(v_2)=0$ as well.

In such a case one can try to compute in the Euclidean region with general kinematics, perform the analytic continuation to Minkowski space and then take the desired limit. But a problem occurs if this limit diverges, then the analytic regularization of the integral with restricted kinematics can not be obtained from the analytic regularization of the non-exceptional configuration in a straightforward way. This situation (kinematic constraints that introduce additional divergences) is rather common in practice.
\begin{figure}
	\centering
	$ B_1 = \Graph[0.4]{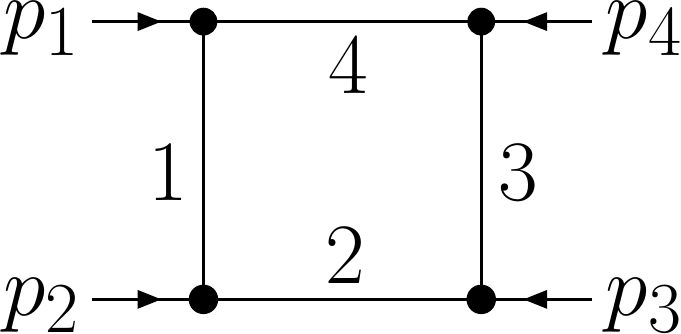} $ \quad
	$\gamma = \set{1,2} = \Graph[0.4]{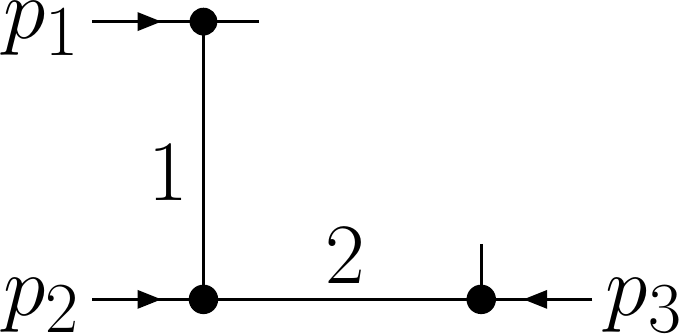} $ \quad
	$ B_1/\gamma = \Graph[0.4]{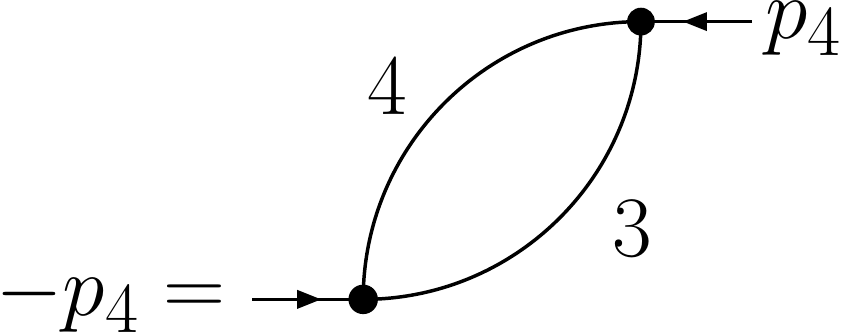} $
	\caption{The massless box graph $B_1$ and one of its four infrared subdivergences.}%
	\label{fig:box-divergences}%
\end{figure}%
\begin{example}
	\label{ex:onshell-box-divergences} %
	The one-loop on-shell massless box $B_1$ of figure~\ref{fig:box-divergences} with $\EP_i = 1$ and $\ExMom(v_i)^2 = m_i^2 = 0$ for $1 \leq i \leq 4$  is ultraviolet-finite in $\dimension = 4 - 2\varepsilon$ dimensions and still a non-trivial function of the two variables $s = [\ExMom(v_1) + \ExMom(v_2)]^2$ and $t=[\ExMom(v_1) + \ExMom(v_4)]^2$, through $\phipol = s \SP_2 \SP_4 + t \SP_1 \SP_3$. However, we locate four infrared divergences
	\begin{equation}
		\deg_{\gamma} (I) + \abs{\gamma}
		=
		-\sdd\left( B_1 / \gamma \right)
		=
		-\varepsilon
		\quad\text{at the corners}\quad
		\gamma
		= \set{1,2},\set{2,3},\set{3,4},\set{4,1}
		\label{eq:onshell-box-divergences} %
	\end{equation}
	in accordance with lemma~\ref{lemma:vanishing-degree-UV}. Observe that the quotients $B_1 / \gamma$ are $0$-scale even though $\gamma$ does not contain all external vertices in the same connected component: the only momentum running through $B_1 / \gamma$ is some vanishing $\ExMom(v_i)^2 = 0$.

	Note that by \eqref{eq:vanishing-degree-UV} we always encounter an infrared divergence in $\dimension = 4-2\varepsilon$ at every two-valent external vertex $v_i$ with $\ExMom(v_i)^2 = 0$ and $\EP_e = \EP_f = 1$ for the two edges $e,f$ incident to $v_i$.
\end{example}
\begin{example}
	\label{ex:triangle-divergences}%
	The triangle graph $G$ with one internal mass $m=m_3$ ($m_1 = m_2 = 0$) and lightlike $p_3^2 = 0$ has the integral representations (in $\dimension = 4-2\varepsilon$ with indices $\EP_e=1$)
	\begin{equation*}
		\FR\left(\Graph[0.3]{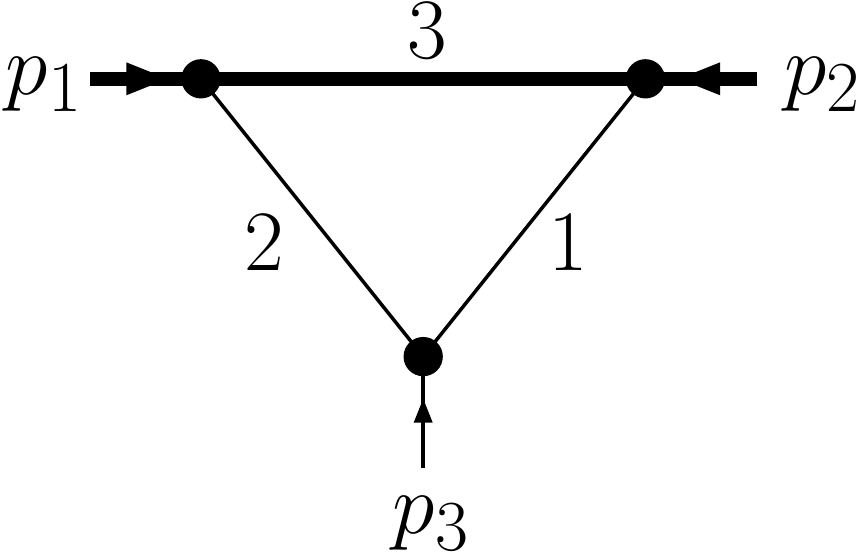}\right)
		= \int \frac{\dd[\dimension] k}{\pi^{\dimension/2}}
			 \frac{1}{(k^2 + m^2)(k+p_2)^{2}(k-p_1)^{2}}
		= \Gamma(1+\varepsilon)
			 \int \frac{\Omega}{\psipol^{1-2\varepsilon} \phipol^{1+\varepsilon}}.
	\end{equation*}
	The parametric integrand contains a factor $\SP_3^{-1-\varepsilon}$ because
	$
		\phipol
		= \SP_3 \left( m^2 \psipol + p_1^2 \SP_2 + p_2^2 \SP_1 \right)
		$ splits, which constitutes a logarithmic divergence at $\SP_3 \rightarrow 0$. It corresponds to the subgraph $\gamma=\set{3}$ with $-\sdd\left( G\contract\gamma\right) = \sdd\left( \gamma \right) - \sdd(G) = -\varepsilon$ since $G\contract\gamma$ is $0$-scale.
\end{example}
\begin{remark}
The scaling degrees $\deg_{\gamma} (I)$ of a parametric integrand depend on the kinematics only through the set of monomials in the second Symanzik polynomial that have a non-zero coefficient. Setting $\ExMom(v_i)^2 = 0$ in the examples above excludes certain monomials from $\phipol$ (which are present for generic $\ExMom(v_i)^2 \neq 0$) and therefore potentially changes $\deg_{\gamma}(\phipol)$ in \eqref{eq:vanishing-degree-polynomial}.

Similarly, on-shell external particles $\ExMom(v_i)^2 = -m_i^2$ with a mass $m_i = m_e$ that is also carried by an internal propagator $e$ can yield cancellations of monomials in $\phipol$ as well.
\end{remark}
These examples of non-Euclidean kinematic constraints are very mild in the following sense: All non-vanishing coefficients of the second Symanzik polynomial $\phipol$ are positive (which is automatic for a Euclidean metric). This implies the positivity $\phipol>0$ and therefore the smoothness of the integrand $I$ inside the integration domain $(0,\infty)^{\edges}$, such that divergences of $\FR(G)$ must stem from boundary contributions (integration over regions where $\SP_e \rightarrow 0,\infty$ for one or several edges $e$).

The convergence of $\FR(G)$ is therefore still likely to be assessable through power counting. Indeed, we saw that corollary~\ref{corollary:projective-convergence-euclidean} stays valid in our examples and correctly predicts the non-Euclidean divergences. Though \eqref{eq:uvdiv-delta-inserted} reveals $\abs{\gamma} + \deg_{\gamma}(I)>0$ as necessary for absolute convergence, its sufficiency (in the Euclidean case) is a non-trivial consequence of the special structure of the Symanzik polynomials. For arbitrary polynomials we should extend definition~\ref{def:vanishing-degree} to the more general rescalings of
\begin{definition}
	\label{def:vanishing-degree-general} %
	For a vector $0 \neq \ScaleVec \in \R^{\edges}$ (also called \emph{region} or \emph{sector}), the rescaled integrand
$
	I^{(\ScaleVec)}
	\defas I\left( \Pscale^{\ScaleVec_1} \SP_1, \cdots, \Pscale^{\ScaleVec_{\abs{\edges}}}\SP_{\abs{\edges}} \right)
$
	determines a vanishing degree $\deg_{\ScaleVec}(I)$ by
	\begin{equation}
		I^{(\ScaleVec)} 
			\in \bigo{\Pscale^{\deg_{\ScaleVec}(I)}}
		\quad\text{such that}\quad
		\widetilde{I^{(\ScaleVec)}} \defas I^{(\ScaleVec)} \cdot \Pscale^{-\deg_{\ScaleVec}(I)}
		\quad\text{has}\quad
		\lim_{\Pscale \rightarrow 0} \widetilde{I^{(\ScaleVec)}} 
			\neq 0,\infty
		\label{eq:vanishing-degree-general} %
	\end{equation}
	finite and non-zero. The associated degree of divergence is
	\begin{equation}
		\sdd_{\ScaleVec}(I)
		\defas 
			\sum_{e \in \edges} \ScaleVec_e
			+\deg_{\ScaleVec} (I)
		.
		\label{eq:sdd-general} %
	\end{equation}
\end{definition}
When $ \emptyset \neq \supp \ScaleVec \defas \setexp{e}{\ScaleVec_e \neq 0} \subsetneq \edges$ (not all variables can be rescaled simultaneously in the projective form $\int I\ \Omega$) we insert a factor 
$
	1 
	= \int_0^{\infty} 
			\dd\Pscale
			\ \delta\left( \Pscale - \sum_{e\colon \ScaleVec_e \neq 0} \SP_e^{1/{\ScaleVec_e}} \right)
$ into the integrand as we did in \eqref{eq:uvdiv-delta-inserted}.\footnote{Here we choose the hyperplane ($\delta$-function) in $\Omega$ of \eqref{eq:projective-delta-form-integrand} such that it does not constrain any of the rescaled variables $\SP_e$ with $\ScaleVec_e \neq 0$.} After substituting $\SP_e$ for $\Pscale^{\ScaleVec_e} \SP_e$ we arrive at
\begin{equation}
	\int I\ \Omega
	= \int \Omega
		\ \delta\Bigg( 1 - \sum_{e\colon \ScaleVec_e \neq 0} \SP_e^{1/{\ScaleVec_e}} \Bigg)
		\int_0^{\infty} \frac{\dd \Pscale}{\Pscale} 
			\Pscale^{\sdd_{\ScaleVec}}
			\cdot
			\widetilde{I^{(\ScaleVec)}},
	\label{eq:anareg-delta-factor} %
\end{equation} 
where 
$
	\widetilde{I^{(\ScaleVec)}} 
$
is finite at $\Pscale \rightarrow 0$.
It follows that the absolute convergence of $\int I\ \Omega$ requires $\sdd_{\ScaleVec} > 0$ for all possible regions $\ScaleVec$ with $\emptyset \neq \supp \ScaleVec \subsetneq \edges$. This highlights how extremely special the Feynman integrands are for Euclidean kinematics, considering that corollary~\ref{corollary:projective-convergence-euclidean} allows us to deduce $\sdd_{\ScaleVec}>0$ for any sector $\ScaleVec \in \R^{\edges}$ from only checking the special (and few) $\ScaleVec$ indexed by subgraphs $\emptyset \neq \gamma \subsetneq \edges$ via $\ScaleVec_e = 1$ for $e \in \gamma$ (and $\ScaleVec_e = 0$ otherwise).

In fact we can prove that $\sdd_{\ScaleVec} (I) > 0$ for all $\emptyset \neq \supp \ScaleVec \subsetneq \edges$ implies convergence of $\int I\ \Omega$ (when $\phipol$ has only non-negative coefficients), but we are still left with the task to identify the divergent sectors. In general, it does not suffice to only consider the simple forms where $\ScaleVec_e \in \set{0,1}$.
\begin{example}
	Consider the two-loop $3$-point graph $C_2$ from figure~\ref{fig:triladder-contraction-deletion} with lightlike $p_2^2 = p_3^2 = 0$ and $p_1^2 = 1$. From its graph polynomials \eqref{eq:psipol_C2} and \eqref{eq:phipol_C2} we find that with unit indices $\EP_e = 1$ in $\dimension = 4 - 2\varepsilon$, the vector $\ScaleVec=(3,2,2,1,2,0)$ yields a divergence
	\begin{equation*}
		\sdd_{\ScaleVec}(I)
		= 10 + 3 \varepsilon \deg_{\ScaleVec}(\psipol) - (2+2\varepsilon) \deg_{\ScaleVec}(\phipol)
		= 10 + 6\varepsilon - 5 (2+2\varepsilon)
		= -4\varepsilon.
	\end{equation*}
	This example is taken from \cite{Smirnov:ProblemsStrategyRegions} which discusses the technique of \emph{expansion by regions}. Within this approach, this particular region $\ScaleVec$ was identified as one contribution to the expansion of $C_2$ in a particular kinematic limit.
\end{example}
As there are uncountably many sectors to consider for any given integrand $I$, it is not immediately clear how the positivity of all these degrees can be checked and the determination of all divergent sectors is a non-trivial problem \cite{Smirnov:ProblemsStrategyRegions}. But note that
\begin{itemize}
	\item $\sdd_{\ScaleVec}(I)$ is a continuous function of $\ScaleVec$,
	\item $\sdd_{\ScaleVec} = \sdd_{\ScaleVec + \lambda(1,\ldots,1)}$ for any $\lambda$ (projectiveness of $I\ \Omega$),
	\item $\sdd_{\ScaleVec} = \sdd_{\lambda\ScaleVec}$ for any $\lambda>0$,
	\item $\sdd_{\ScaleVec}(I_G) = \sum_e \ScaleVec_e \EP_e + (\sdd(G)-\dimension/2) \deg_{\ScaleVec}(\psipol_G) - \sdd(G)  \deg_{\ScaleVec} (\phipol_G)$ is a linear function (in $\ScaleVec$) as long as the minimal monomials of $\phipol_G$ and $\psipol_G$ in \eqref{eq:vanishing-degree-polynomial} stay the same.
\end{itemize}
By the first two observations, we may restrict to sectors in the compact
\begin{equation*}
	\Delta
	\defas \setexp{\ScaleVec \in \R_{\geq 0}^{\edges}}{\ScaleVec_1 + \cdots+ \ScaleVec_{\edges}=1}.
\end{equation*}
For any pair $\SP^n$ and $\SP^m$ of monomials ($n,m\in \N_0^{\edges}$) that occur in $\psipol, \phipol$ (with non-vanishing coefficient), let $H_{n,m}$ denote the orthogonal complement of $n-m$, the hyperplane
\begin{equation*}
	H_{n,m} \defas
	\setexp{\ScaleVec \in \R^{\edges}}{\deg_{\ScaleVec} (\SP^{n_i}) = \deg_{\ScaleVec} (\SP^{n_j})}
	= \Vanishing\left( \sum_{e \in \edges} \ScaleVec_e (n-m)_e \right)
	= (n-m)^{\perp}.
\end{equation*}
All these hyperplanes divide $\Delta = \bigcup_i \Delta_i$ into a finite number of convex polytopes, such that $\deg_{\ScaleVec}$ is a linear function inside each $\Delta_i$ and thus attains its minimum on one of the finitely many vertices ($0$-simplices) of $\Delta_i$. This simple observation shows
\begin{corollary}\label{cor:finitely-many-sectors}
	The Symanzik polynomials of a graph $G$ determine a finite number $N$ of sectors $\ScaleVec_k \in \Delta$ such that $\sdd_{\ScaleVec_k}(I_G)>0$ (for all $1 \leq k \leq N$) already implies that $\sdd_{\ScaleVec}(I_G) > 0$ holds for any $\ScaleVec \in \R^{\edges}$.
\end{corollary}
\begin{remark}\label{rem:more-general-integrands}
	More generally, we can replace $I_G$ by any homogeneous integrand of the form $I = \prod_k f_k^{\EP_k}$, built from polynomials $f_k \in \C[\SP_1,\ldots,\SP_{\edges}]$ raised to constant ($\SP$-independent) powers $\EP_k$ such that $I\ \Omega$ is a projective form.
\end{remark}
This idea was used in \cite{PakSmirnov:GeometricApproachAsymptotic} to devise an efficient algorithm to find all divergent sectors (in the slightly different context of asymptotic expansions of Feynman integrals).

\subsubsection{Non-positive $\phipol$}
A much more severe complication of Minkowski space, which is impossible in the Euclidean case, is the occurrence of negative coefficients in $\phipol$.
\begin{example}
	On-shell massless $4$-point kinematics are defined by $p_1^2 = p_2^2 = p_3^3 = p_4^2 = 0 = s+t+u$ for $s = (p_1+p_2)^2$, $t = (p_1 + p_4)^2$ and $u=(p_1+p_3)^2$. This condition requires at least one of the Mandelstam invariants $\set{s,t,u}$ to be negative, which means that $\phipol$ acquires zeros inside the domain of integration.
	
	In this particular case, the problem can be circumvented by first relaxing the constraint $s+t+u = 0$ and taking all three variables to be positive. The analytic continuation $u \rightarrow -s-t$ can be computed afterwards, since it does not introduce further divergences (no monomials of $\phipol$ drop out).
\end{example}
Some much more complex examples have been discussed in \cite{JantzenSmirnov:PotentialAndGlauberRegions}, including a case where
\begin{equation*}
	\phipol = \SP_1 ( \SP_1 + \SP_2 + \SP_3 + \SP_4 + \SP_5)m^2 + (\SP_2 - \SP_3)(\SP_4 - \SP_5)q^2 - \imag\epsilon.
\end{equation*}
The infinitesimal imaginary part is needed to fix the ambiguity in the definition of $\log(\phipol)$, as $\phipol$ takes negative (and zero) values inside the domain $\SP \in \R_+^5$ of integration. This problem was resolved by a change of variables that transforms $\phipol$ into the positive polynomial $\SP_1(\SP_1+\SP_2 + \SP_3 + \SP_4 + \SP_5) m^2 + \SP_2 \SP_4 q^2$.

In general however, the desingularization of the second Symanzik polynomial is not well-understood. Instead, practical computations rely on the automatization of general \emph{sector decomposition} algorithms \cite{BinothHeinrich:SectorDecomposition,BinothHeinrich:SectorDecomposition2}.

But this interesting problem is not the subject of this thesis and we will only consider cases where all coefficients of $\phipol$ are positive.

\subsection{Analyticity and convergence}
We saw in corollary~\ref{cor:finitely-many-sectors} that in many cases, the absolute convergence of a Feynman integral $\FR(G)$ is equivalent to the simultaneous fulfilment 
\begin{equation}
	\Lambda_G
	\defas \bigcap_{k} \setexp{(\EP_1,\ldots,\EP_{\edges},\dimension)}{\sdd_{\ScaleVec_k}(I_G) > 0}
	\subset \R^{\edges + 1}
	\label{eq:def-convergence-domain}%
\end{equation}
of a finite number of homogeneous linear inequalities among the propagator indices $\EP_e$ and the dimension $\dimension$ of space-time. Crucially, this region is non-empty \cite{Speer:SingularityStructureGenericFeynmanAmplitudes}.
\begin{theorem}
	\label{theorem:nonempty-convergence-domain} %
	For Euclidean, non-exceptional external momenta, the domain $\Lambda_{G}$ of absolute convergence of $\FR(G)$ is non-empty provided that $G$ is free of tadpoles.\footnote{A tadpole is a subgraph $\gamma$ independent of any kinematics ($\phipol_{\gamma} = 0$), see also lemma~\ref{lemma:vanishing-degree-UV} (``$0$-scale'').}
\end{theorem}
Inside $\Lambda_G$ and the Euclidean region, differentiation under the integral sign\footnote{See the theorem on holomorphic parameter integrals in \cite{Sauvigny:PDE1}.} proves that $\FR(G)$ is an analytic function of the variables $\EP_e$, $\dimension$ and also of the kinematic invariants $\Kinematics = \set{m_e^2} \cup \set{\ExMom(F)^2}$ that appear in the second Symanzik polynomial $\phipol$. The analytic continuation of $\FR(G)$ behaves very differently with respect to these arguments:
\begin{itemize}
	\item $\FR(G)$ is a multivalued function of the kinematics $\Kinematics$ with non-trivial monodromies. The singularities are governed by \emph{Landau equations} and branch cut ambiguities can be related to \emph{Cutkosky rules}. We do not address these questions here.\footnote{These subjects are treated in most textbooks on quantum field theory and have been studied extensively in the literature, but are still far from being solved \cite{BlochKreimer:LandauOneLoop}.}
	\item The dependence on $\EP_e$ and $\dimension$ is single-valued. Strikingly, $\FR(G)$ extends to a meromorphic function of these variables. As we saw already, the corresponding singularities are linear hypersurfaces in $\R^{\edges + 1}$.
\end{itemize}
Below we derive convergent integral representations for the analytic continuation of $\FR(G)$ to points outside the convergence domain $\Lambda_G$ of the original integral $\int I\ \Omega$. 
Our argument relies only on $\Lambda_G \neq \emptyset$ and reproves (via a constructive argument based on integration by parts) Speer's result \cite{Speer:SingularityStructureGenericFeynmanAmplitudes} that $\FR(G)$ is meromorphic in $\EP_e$ and $\dimension$, with poles only on hyperplanes (see corollary~\ref{cor:convergent-anareg-integrand} and remark~\ref{rem:anareg-divisors}).

\subsection{Analytic regularization}
\label{sec:anareg} %
Consider a projective integrand $I$ after the rescaling \eqref{eq:anareg-delta-factor}. The partial integration
\begin{equation}
	\label{eq:anareg-partial} %
	\int_0^{\infty} 
		\frac{\dd \Pscale}{\Pscale} 
		\Pscale^{\sdd_{\ScaleVec}} \cdot \widetilde{I^{(\ScaleVec)}}
	=
		\restrict{
			\frac{\Pscale^{\sdd_{\ScaleVec}}}{\sdd_{\ScaleVec}}
			\widetilde{I^{(\ScaleVec)}}
		}{\Pscale=0}^{\infty} 
		- \frac{1}{\sdd_{\ScaleVec}}
			\int_0^{\infty} \dd \Pscale \cdot \Pscale^{\sdd_{\ScaleVec}}
				\frac{\partial}{\partial \Pscale} \widetilde{I^{(\ScaleVec)}}
\end{equation}
has vanishing boundary contribution inside the convergence domain $\Lambda_G$ (as $\sdd_{\ScaleVec}>0$).\footnote{The vanishing when $\Pscale \rightarrow \infty$ follows from $\Pscale^{\sdd_{\ScaleVec}} \widetilde{I^{(\ScaleVec)}} = I^{(\ScaleVec)} = I^{(-\ScaleVec)} (\Pscale^{-1}) = \Pscale^{-\sdd_{-\ScaleVec}} \widetilde{I^{(-\ScaleVec)}}(\Pscale^{-1})$.}
We substitute back $\Pscale^{\ScaleVec_e} \SP_e$ with $\SP_e$, identify 
$
	\int_0^{\infty} \frac{\dd\Pscale}{\Pscale}
	\delta\left( 1-\Pscale^{-1} \sum_e \SP_e^{1/{\ScaleVec_e}} \right)
	= 1
	$ in \eqref{eq:anareg-delta-factor} and conclude that $\int \Omega \ I = \int \Omega\ \anapartial{\ScaleVec} \left( I \right)$ for the differential operator ($\partial_e \defas \frac{\partial}{\partial \SP_e}$)
\begin{equation}
	\label{eq:def-anapartial} %
	\anapartial{\ScaleVec}
	\defas
		1-\frac{1}{\sdd_{\ScaleVec}} 
		\sum_{e\in \edges} \ScaleVec_e \partial_e \SP_e
	= \frac{1}{\sdd_{\ScaleVec}} \left[
			\deg_{\ScaleVec}
			-\sum_{e} \ScaleVec_e \SP_e \partial_e
		\right].
\end{equation}
Both $\int \Omega\ I$ and $\int \Omega\ \anapartial{\ScaleVec}(I)$ are absolutely convergent on $\Lambda_G$ and coincide there, so their analytic continuations are the same.
\begin{example}[Triangle graph from example~\ref{ex:triangle-divergences}]\label{ex:triangle-anapartial}
		With respect to $\gamma=\set{3}$ we have $\widetilde{I^{(\gamma)}} = \psipol^{2\varepsilon-1} \cdot \left[ m^2 \psipol + p_2^2\SP_1 + p_1^2\SP_2 \right]^{-1-\varepsilon}$ with $\sdd_{\gamma}= -\varepsilon$ and deduce
		\begin{equation}
			\label{eq:triangle-divergent-partial} %
			\int \frac{\Omega}{\psipol^{1-2\varepsilon} \phipol^{1+\varepsilon}}
			= 
			\frac{1}{\varepsilon} \cdot \int \frac{\Omega}{\SP_3^{\varepsilon}} \frac{\partial}{\partial \SP_3} \widetilde{I^{(\gamma)}}
			= \frac{1}{\varepsilon} \cdot 
				\int \frac{\Omega\ \SP_3 }{\psipol^{1-2\varepsilon} \phipol^{1+\varepsilon}}
				\left[ \frac{2\varepsilon-1}{\psipol}
				- \frac{(1+\varepsilon)\SP_3 m^2}{\phipol}
				\right]
		\end{equation}
		as an identity of absolutely convergent integrals on their joint domain $\Lambda_{G} = \set{\varepsilon<0}$ of convergence. Note that the integral on the right-hand side has an increased regime $\set{\varepsilon<1}$ of convergence and can thus be expanded near $\varepsilon\rightarrow 0$.
\end{example}
\begin{example}\label{ex:anapartial-single-edges}
	Consider a single edge $\gamma = \set{f}$ and assume that $0 \neq \restrict{\phipol}{\SP_f=0}$ such that $\sdd_{\gamma} = \EP_f$. When $\EP_f\leq 0$, this divergence originates from the Schwinger trick \eqref{eq:Schwinger-trick}. Including the $\Gamma^{-1}(\SP_e)$ prefactors \eqref{eq:feynman-integral-projective} into the integrand, the partial integration \eqref{eq:def-anapartial} associated to $\gamma$ ($\anapartial{\ScaleVec}$ where $\ScaleVec_f = 1$ and $\ScaleVec_e=0$ for $e\neq f$) replaces
	\begin{equation}
		\left[\prod_{e\in\edges}\frac{\SP_e^{\EP_e-1}}{\Gamma(\EP_e)}\right] 
		\frac{1}{\psipol^{\dimension/2-\sdd}\phipol^{\sdd}}
		\quad\text{by}\quad
		\frac{\SP_f^{\EP_f}}{\Gamma(\EP_f+1)}
		\partial_f \left[\prod_{e\in\edges\setminus\set{f}}\frac{\SP_e^{\EP_e-1}}{\Gamma(\EP_e)}\right] 
		\frac{1}{\psipol^{\dimension/2-\sdd}\phipol^{\sdd}}.
		\label{eq:anapartial-single-edges}%
	\end{equation}
	This corresponds to the analytic continuation $P^{-\EP} = \Gamma^{-1}(\EP+1) \int_0^{\infty} \SP^\EP (-\partial_\SP) e^{-\SP P}\ \dd \SP$ of \eqref{eq:Schwinger-trick}. In the limit when $\EP \in -\N_0$ becomes a negative integer, iteration of this formula results in the elementary $P^{-\EP} = \int_0^{\infty} (-\partial_{\SP})^{1-\EP} e^{-\SP P} = \lim_{\SP\rightarrow 0} (-\partial_{\SP})^{-\EP} e^{-\SP P}$.
\end{example}
\begin{remark}
	\label{remark:ibp-parametric-noboundary}%
	Because $\int \Omega\ I = \int \Omega\ \anapartial{\ScaleVec}(I)$ holds for any $\ScaleVec \in \R^{\edges}$, \eqref{eq:def-anapartial} shows that 
\begin{equation}
	\label{eq:ibp-parametric-noboundary} %
	\int \Omega\ \left( \partial_e \SP_e I \right) = 0
	\quad\text{for any edge $e \in \edges$}.
\end{equation}
We just proved this identity in its literal sense whenever the integral is convergent. But this also means that the analytic continuation of $\int \Omega\ (\partial_e \SP_e I)$ is identically zero everywhere, because it is defined (and zero) on its non-empty domain of convergence. So if we understand $\int \Omega\ I$ as a symbol defined to be the analytic continuation of the actual integral from its domain of convergence, \eqref{eq:ibp-parametric-noboundary} is valid everywhere.

This very nicely explains the analogous formula $\int \dd[\dimension] k \ \partial_{k_{\mu}} \cdots = 0$, which is typically introduced as an \emph{axiom}\footnote{It has been known before though that it is a consequence of analytic continuation \cite{Collins}.} for the definition of dimensional regularization in momentum space and the corner stone of \emph{integration by parts} in momentum space \cite{ChetyrkinTkachov:IBP}.
\end{remark}
We summarize our results as
\begin{lemma}
	\label{lemma:anapartial} %
	Let $I$ denote a parametric integrand such that $I\ \Omega$ is a projective form. For any sector $\ScaleVec \in \R^{\edges}$ with $\sdd_{\ScaleVec}(I) \neq 0$, the parametric integrand $I' \defas \anapartial{\ScaleVec}\left( I \right)$ fulfils
	\begin{enumerate}
		\item $\int I\  \Omega$ = $\int I'\ \Omega$ as analytically regularized integrals,
		\item $\sdd_{\ScaleVec'} \big( I' \big) \geq \sdd_{\ScaleVec'} \left( I \right)$ for any $\ScaleVec'$ and
		\item $\sdd_{\ScaleVec} ( I' ) > \sdd_{\ScaleVec}\left( I \right)$ increases.
	\end{enumerate}
\end{lemma}
\begin{proof}
	We proved property~1 above and 3 is immediate from \eqref{eq:anareg-partial}.	To see property~2, consider any polynomial $p = \sum_{n} c_{n} \SP^{n}$ in Schwinger parameters with monomials $\SP^n$. Then for any exponent $\EP \neq 0$,
\begin{equation*}
	\deg_{\ScaleVec'} \SP_e \partial_e \left( p^{\EP} \right)
	= \deg_{\ScaleVec'} \left( \EP p^{\EP} \frac{\SP_e \partial_e p}{p} \right)
	= \deg_{\ScaleVec'} (p^{\EP})
		+ \min_{c_n \neq 0, n_e \neq 0} \deg_{\ScaleVec'} \SP^{n}
		- \deg_{\ScaleVec'} (p)
	\geq
	\deg_{\ScaleVec'} (p^{\EP})
\end{equation*}
because the minimum runs over a subset of the monomials in \eqref{eq:vanishing-degree-polynomial}. The integrand $I = \prod_p p^{\EP_p}$ is a product of powers of polynomials, thus the statement follows from Leibniz's rule and the form \eqref{eq:def-anapartial} of the operator $\anapartial{\ScaleVec}$, together with \eqref{eq:vanishing-degree-rules}.
\end{proof}
\begin{remark}\label{rem:integer-sectors}
	In corollary~\ref{cor:finitely-many-sectors}, we can choose $\ScaleVec_k \in \Q^{\edges}$ ($\Delta$ and $H_{n,m}$ are defined over $\Q$, so are the polytopes $\Delta_i$) and then rescale them by the common denominator of their components. Thus we can assume $\ScaleVec_k \in \N_0^{\edges}$ such that $\widetilde{I}^{(\ScaleVec_k)}$ will be analytic and admit a Taylor series in $\Pscale$. Then $\sdd_{\ScaleVec}(I') \geq 1 + \sdd_{\ScaleVec}(I)$ increases by an integer.
\end{remark}
\begin{corollary}\label{cor:convergent-anareg-integrand}
	Given a projective form $I_G\ \Omega$ and a point $(\EP,\dimension) \notin \Lambda_G$ outside its domain of convergence, finitely many applications of operators $\anapartial{\ScaleVec}$ (for suitable sectors $\ScaleVec$) on $I_G$ suffice to generate an integrand $I'$ which converges at $(\EP, \dimension)$ and computes the analytic continuation $\int I'\ \Omega$ of $ \int I\  \Omega$.
\end{corollary}
\begin{proof}
	Choose the $\ScaleVec_k \in \N_0^{\edges}$ as in remark~\ref{rem:integer-sectors} and apply $I_{(j+1)} = \anapartial{\ScaleVec_k} (I_j)$ to $I_0 \defas I$ until $\sdd_{\ScaleVec_k}(I_j) > 0$. Repeat this process for each $k$ to reach a convergent integrand $I_N$.
\end{proof}
\begin{remark}\label{rem:anareg-divisors}
	The only singularities occur through the denominators $\sdd_{\ScaleVec_k}(I_j) \in \sdd_{\ScaleVec_k}(I) + \N_0$ in \eqref{eq:def-anapartial}, so $\int I\ \Omega$ is meromorphic with poles along hyperplanes
	\begin{equation}
		\bigcup_k \setexp{\sdd_{\ScaleVec_k}(I) = -n}{n \in \N_0}.
		\label{eq:anareg-divisors}%
	\end{equation}
\end{remark}

\subsection{Applications}
The convergent integral representation $\int I'\ \Omega$ for the analytic continuation of $\int I_G\ \Omega$ can itself be interpreted in terms of Feynman integrals.
From \eqref{eq:def-anapartial} it is clear that $I'$ is a linear combination of terms $(\psipol/\phipol)^{\sdd}/\psipol^{\dimension/2} \cdot P/ (\psipol^n \phipol^m)$ with integers $n,m \in \N_0$ and monomials $P = \prod_e \SP_e^{\EP_e'-1}$, so they correspond to the integrand $I_{G'}$ of the same graph, but with shifted indices $\EP_e' \in \EP_e + \N_0$ instead of $\EP_e$ and in dimension $\dimension' = \dimension + 2(n+m)$.
\begin{example}
	In the example~\ref{ex:triangle-anapartial} of the triangle $G$ with unit indices in $\dimension = 4 -2\varepsilon$, we find two terms in \eqref{eq:triangle-divergent-partial}. Paying attention to the $\Gamma$-prefactors in \eqref{eq:feynman-integral-projective}, they give
	\begin{equation}
		\FR(G,1,1,1,4-2\varepsilon)
		= \frac{2\varepsilon-1}{\varepsilon} \FR(G,1,1,2,6-2\varepsilon)
		- \frac{m^2}{2\varepsilon} \FR(G,1,1,3,6-2\varepsilon).
		\label{eq:triangle-anareg-Feynman}%
	\end{equation}
\end{example}
\begin{corollary}\label{cor:finite-anareg-as-Feynman}
	Every analytically regularized Feynman integral $\FR(G,\EP,\dimension)$ can be written as a linear combination $\FR(G,\EP,\dimension )= \sum_i r_i \FR(G, \EP_i, \dimension_i)$ of convergent Feynman integrals related to the same graph (but with integer-shifted indices and dimension) with prefactors $r_i$ that are rational functions in kinematics $\Kinematics$, dimension $\dimension$ and indices $\EP$.
\end{corollary}
Special cases of relations between Feynman integrals with different values of the space-time dimension have been known for long \cite{Tarasov:ConnectionBetweenFeynmanIntegrals}. They have been used to derive difference equations that allow for rapid numerical evaluation to very high precision \cite{Lee:DimensionalRecurrenceAnalyticalProperties,LeeSmirnov:EasyWay}, but we could not find a statement similar to our result in the literature.
\subsubsection{Finite master integrals}
A standard technique for calculations involving many individual Feynman diagrams is to exploit integration by parts (IBP) \cite{ChetyrkinTkachov:IBP} to express all of them in terms of a small number of \emph{master integrals}. Corollary~\ref{cor:finite-anareg-as-Feynman} proves that we can choose this basis to contain only finite integrals, which has several benefits in practice compared to divergent bases (see the following subsections for example).

Several computer programs that automatize IBP reductions are publicly available already \cite{Smirnov:Fire4LiteRed,ManteuffelStuderus:Reduze2}, but so far none of them allows to force the selection of a finite set of master integrals. 
This is expected to be rectified soon, because it suffices to just enumerate suitable (subdivergence free) integrals until one can select a basis.
A first study of these aspects in relation with IBP, including instructive examples, appeared after this thesis was finished \cite{ManteuffelPanzerSchabinger:QuasiFinite}.

\subsubsection{Integration with hyperlogarithms}
Having a convergent integral representation at hand allows for the computation of the individual terms of an $\varepsilon$-expansion $\FR(G) = \sum_{n} c_n \varepsilon^n$ of a Feynman integral analytically, because we may expand the integrand. This was our original motivation to construct these integrands \cite{Panzer:DivergencesManyScales}. It is very important that they are Feynman integrals themselves, in particular the integrands $I_n$ for each coefficient $c_n = \int I_n\ \Omega$ are polynomials
	\begin{equation}
		I_n
		\in
		\Q\left[
				\Kinematics,
				\psipol^{-1}, \phipol^{-1},
				\log( \psipol), \log (\phipol),
				\SP_e, \SP_e^{-1}, \log(\SP_e) \colon e \in \edges
		\right].
		\label{eq:epsilon-expansion-integrands}%
	\end{equation}
As we shall discuss in chapter~\ref{chap:hyperlogs}, such integrals can be computed with hyperlogarithms for \emph{linearly reducible} graphs $G$. In particular this means that in an analytic regularization scheme, the presence of subdivergences does not introduce essentially new difficulties or obstructions for exact integration techniques.

\subsubsection{Numeric evaluation}
Standard techniques like Monte Carlo integration can in principle be applied directly to a convergent integral. The reduction above therefore extends the applicability of these methods to divergent Feynman integrals as well.
However, the accuracy might be very poor due to the (integrable) singularities of the parametric integrand (on the boundary where some $\SP_e \rightarrow 0,\infty$). 
But note that further applications of $\anapartial{\ScaleVec_k}$ can be used to actually ensure finiteness (or even vanishing) of the integrand on these boundaries, which should increase the accuracy considerably. 
Studying the practical feasibility of this approach seems to be a very interesting project for future research.

\subsubsection{Comparison with sector decomposition}
Available programs \cite{BorowkaCarterHeinrich:SecDec2,BognerWeinzierl:ResolutionOfSingularities,Smirnov:FIESTA3} for numeric evaluation of Feynman integrals are based on sector decomposition \cite{BinothHeinrich:SectorDecomposition,BinothHeinrich:SectorDecomposition2}, which is a very general method for desingularization of polynomials \cite{BognerWeinzierl:Periods} and therefore applicable in most cases of practical interest (at least for Euclidean kinematics; extensions to Minkowski kinematics were discussed only rather recently \cite{BorowkaCarterHeinrich:SecDec2}).
This techniques divides the integral $\int I\ \Omega = \sum_i \int_{[0,1]^{\edges-1}} I_i$ into many convergent integrals constructed by iterations of changes of variables (poles in the analytic regulators are made explicit in prefactors). While the most common application is to use this representation for numeric integration, one can also think of exploiting sector decomposition solely for the purpose of desingularization and aim for an exact evaluation of the individual integrals $I_i$.

However, the changes of variables mean that the polynomials in $I_i$ are different from the original $\set{\psipol,\phipol}$ and the form \eqref{eq:epsilon-expansion-integrands} is not guaranteed anymore. 
So even for linearly reducible graphs, it is not clear if each $I_i$ as constructed by the sector decomposition algorithm is in fact linearly reducible as given. 
Furthermore, due to the subdivision of the integration domain each individual $I_i$ typically evaluates to more complicated functions and numbers than the actual Feynman integral; only in the sum of all $I_i$ these spurious structures cancel each other (see \cite{ManteuffelPanzerSchabinger:QuasiFinite} for an illustration).

Therefore, at least for the purpose of exact integration (with hyperlogarithms but also for other approaches), we consider our representation highly preferable over a desingularization via sector decomposition.

\section{Renormalization}
\label{sec:renormalization}%

It is the aim of perturbative quantum field theory to provide results on measurable quantities (like cross sections) that can be compared with the observations in an experiment. Therefore it is crucial to deal with the divergences occurring in Feynman diagrams and to find a way of absorbing these infinities in order to arrive at finite predictions.

This problem of \emph{renormalization} has been discussed and developed in the literature for more than sixty years. A rather recent addition to its underpinnings is the concept of Hopf algebra, introduced by Dirk Kreimer first in \cite{Kreimer:NonlinearDSE}. It stimulated a plethora of fruitful developments (in physics as well as pure mathematics) which we have no chance to recall here. Introductory texts into this subject are available by now, the reviews \cite{Manchon,Panzer:PM2012} are particularly suitable for our needs here.

We merely want to summarize very briefly the renormalization by kinematic subtraction in the case of logarithmic ultraviolet divergences. Our focus lies on its formulation in the Schwinger parametric representation, which has been studied in great detail long ago \cite{BergereZuber:RenormalizationFeynmanParametric,BergereLam:BPAlpha} and recently from a modern viewpoint of algebraic geometry \cite{BlochKreimer:MixedHodge,BrownKreimer:AnglesScales}.

In particular we recall the convergent integral representation for renormalized Feynman integrals, which is based on the \emph{forest formula} from the earliest days of renormalization theory. The parametric representation was used widely during those times, but the actual evaluation of the integrals in this form was too complicated. After the invention of dimensional regularization, huge progress in the evaluation of Feynman integrals was possible in momentum space. As of today, the standard machinery in perturbative quantum field theory is almost exclusively centered on dimensional regularization.

Our goal is to advertise the idea to directly compute renormalized integrals using the forest formula in the parametric representation, without ever introducing a regulator in the first place. In section~\ref{sec:ex-renormalized-parametric} we carry out this program in a few examples.

\subsection{Hopf algebra of ultraviolet divergences}
We consider the Hopf algebra $\FeynHopf$ of scalar, logarithmically divergent Feynman diagrams. As an algebra, $\FeynHopf = \Q[\FeynGraphs]$ is free, commutative and generated by connected, scalar, logarithmically divergent Feynman graphs
\begin{equation}
	\FeynGraphs
	\defas
		\setexp{G}{
				\comps(G) = \set{G},\ 
				\sdd(G)=0\ \text{and}\ 
				\sdd(\gamma) \leq 0\ \text{for all subgraphs}\ \gamma \subset G
		}
	\label{eq:def-FeynHopf}%
\end{equation}
that have at worst logarithmically divergent subgraphs.\footnote{Note that this implies that $G \in \FeynGraphs$ is \emph{one-particle irreducible} (1PI), that is, it can not be disconnected by deletion of a single edge.} We denote the empty graph by $\1$. The coproduct $\cop$ and the reduced coproduct $\copred$ are linear maps defined on every graph $G$ by
\begin{equation}
	\cop, \copred\colon \FeynHopf \longrightarrow \FeynHopf \tp \FeynHopf,\quad
	\cop(G)
	\defas
	\hspace{-2ex}\sum_{\gamma \subseteq G\colon \sdd(\gamma) = 0} \hspace{-2ex}
		\gamma \tp G/\gamma
	= \1 \tp G + G \tp \1 + \copred(G)
	\label{eq:def-FeynCop}%
\end{equation}
to extract all subdivergences $\gamma$ and the remaining quotients $G/\gamma$ (where each connected component of $\gamma$ has been shrunken to a single vertex). Since $\FeynHopf$ is graded by the number of loops, we can compute the \emph{antipode} $\antipode$ recursively by
\begin{equation}
	\antipode\colon \FeynHopf \longrightarrow \FeynHopf,\quad
	S(\1) = \1
	\quad\text{and}\quad
	\antipode(G)
	= - \hspace{-2ex}\sum_{\gamma \subsetneq G\colon \sdd(\gamma) = 0}\hspace{-2ex} \antipode(\gamma) G/\gamma
	\quad\text{for $G \neq \1$.}
	\label{eq:FeynAntipode}%
\end{equation}
An explicit solution to this relation is given by the forest formula. To state it we let $\forests(G)$ denote the \emph{forests} of $G$, which are those subsets $F\subset \setexp{\gamma}{\gamma \subsetneq G} \cap \FeynGraphs$ of proper subgraphs of $G$ such that any pair of subgraphs is either (edge-) disjoint or nested:
\begin{equation}
	F \in \forests(G)
	\gdw
\text{For any $\gamma_1,\gamma_2 \in F$, either $\gamma_1 \cap \gamma_2 = \emptyset$, $\gamma_1 \subseteq \gamma_2$ or $\gamma_2 \subseteq \gamma_1$.}
	\label{eq:def-forests}%
\end{equation}
Mind that the empty forest $\emptyset \in \forests(G)$ is always included.
If we set $\gamma/F \defas \gamma/\bigcup_{\delta \in F, \delta \subsetneq \gamma} \delta$ to the contraction of all proper subgraphs $\delta$ of $\gamma$ that are contained in the forest $F$, we can state the forest formula as
\begin{equation}
	\antipode(G)
	= - \sum_{F \in \forests(G)} (-1)^{\abs{F}} G/F \prod_{\gamma \in F} \gamma/F.
	\label{eq:forest-formula}%
\end{equation}
The Feynman rules are a character on $\FeynHopf$, that means $\FR(G_1 G_2) = \FR(G_1) \FR(G_2)$, but in general ill-defined. They depend on the kinematic invariants $\Kinematics=\set{m_i^2} \cup \set{p_i \cdot p_j}$ including the masses of particles in the theory and products of external momenta. We choose a renormalization point $\widetilde{\Kinematics}$ and write $\restrict{\FR}{\widetilde{\Kinematics}}$ for the Feynman rules with these reference kinematics. The associated counterterms $\FR_-$ are given by
\begin{equation}
	\FR_- (G)
		= \restrict{\FR}{\widetilde{\Kinematics}}^{\convolution - 1} (G)
		= \restrict{\FR}{\widetilde{\Kinematics}} \Big(\antipode (G) \Big)
		\urel{\eqref{eq:FeynAntipode}}
		- \sum_{\gamma \subsetneq G\colon \sdd(\gamma)=0}
			\FR_-(\gamma) \restrict{\FR}{\widetilde{\Kinematics}} (G/\gamma)
	\label{eq:counterterm}%
\end{equation}%
\nomenclature[Phi-,Phi+]{$\FR_-, \FR_+$}{counterterms \eqref{eq:counterterm} and renormalized Feynman rules \eqref{eq:birkhoff}\nomrefpage}%
and the renormalized Feynman rules $\FRren$ are determined via the Birkhoff decomposition%
\begin{equation}
	\FRren = \FR_- \convolution \FR,\quad\text{meaning} \quad
	\FRren(G)
		= \sum_{\mathclap{\gamma \subsetneq G\colon \sdd(\gamma) = 0}}
		\FR_-(\gamma) \left[ \FR(G/\gamma) - \restrict{\FR}{\widetilde{\Kinematics}}(G/\gamma) \right].
	\label{eq:birkhoff}%
\end{equation}%
\begin{example}
	If $\copred(G) = 0$ (so $G$ has no subdivergence), we call $G$ \emph{primitive} and find $\antipode(G) = -G$, $\FR_-(G) = - \restrict{\FR}{\widetilde{\Kinematics}}(G)$ and $\FRren(G) = \FR(G) - \restrict{\FR}{\widetilde{\Kinematics}}(G)$ is a simple subtraction.

	When $G$ has a single subdivergence $\copred(G) = \gamma \tp G/\gamma$, we find
	$\antipode(G) = -G + \gamma \cdot G/\gamma$, the counterterm 
	$\FR_-(G)
	= \restrict{\FR}{\widetilde{\Kinematics}}(G)
	+ \restrict{\FR}{\widetilde{\Kinematics}}(\gamma) \restrict{\FR}{\widetilde{\Kinematics}}(G/\gamma)$ and the renormalized
	\begin{equation*}
		\FRren(G)
	= \FR(G) - \restrict{\FR}{\widetilde{\Kinematics}}(G)
	- \restrict{\FR}{\widetilde{\Kinematics}}(\gamma) \Big[ 
			\FR(G/\gamma)
			- \restrict{\FR}{\widetilde{\Kinematics}}(G/\gamma)
		\Big].
	\end{equation*}
\end{example}
In particular, evaluation at the renormalization point always gives $\restrict{\FRren}{\widetilde{\Kinematics}}(G) = 0$, unless $G = \1$.

\subsubsection{Renormalization group}
\label{sec:renormalization-group}%
Suppose we choose another renormalization point $\widetilde{\Kinematics}'$, then we get different renormalized Feynman rules $\FRren'$. They are related to $\FRren$ through the renormalization group equation
\begin{equation}
	\FRren'
	= \restrict{\FR}{\widetilde{\Kinematics}'}^{\convolution - 1} \convolution \FR
	= \restrict{\FR}{\widetilde{\Kinematics}'}^{\convolution - 1} \convolution \restrict{\FR}{\widetilde{\Kinematics}}^{} \convolution \restrict{\FR}{\widetilde{\Kinematics}}^{\convolution - 1} \convolution \FR
	= \restrict{\FRren'}{\widetilde{\Kinematics}} \convolution \FRren
	= \restrict{\FRren}{\widetilde{\Kinematics}'}^{\convolution - 1} \convolution \FRren.
	\label{eq:rge-finite}%
\end{equation}
Equivalently, we can think of this as keeping the scheme (renormalization point) fixed, but varying the actual kinematics instead. The $\beta$-function of a theory is determined by a very special such variation: We rescale all kinematic invariants by a common factor.
\begin{definition}
	\label{def:period}%
	Suppose all kinematic invariants $\Kinematics_{\scalelog} \defas \big\{m^2_i\, e^{\scalelog}\big\} \cup \big\{(p_i \cdot p_j)\,e^{\scalelog}\big\}$ are simultaneously scaled by a factor $e^{\scalelog}$. Then the \emph{period map} 
	$	\period\colon \FeynHopf \longrightarrow \R$,
	\begin{equation}
		\period
		\defas -\Big[ 
			\partial_{\scalelog} \restrict{\FRren}{\widetilde{\Kinematics}_{\scalelog}}
		\Big]_{\scalelog=0}
		= - \left( \sum_{\theta \in \Kinematics} (\theta \partial_{\theta}) \FRren \right)_{\Kinematics = \widetilde{\Kinematics}}
		\label{eq:def-period}%
	\end{equation}
	measures the scaling dependence of $\FRren$ at the renormalization point.
\end{definition}
These numbers govern the full scaling dependence, because one can prove \cite{Panzer:Mellin}
\begin{equation}
	- \partial_{\scalelog} \restrict{\FRren}{\Kinematics_{\scalelog}}
	= \period \convolution \restrict{\FRren}{\Kinematics_{\scalelog}},
	\quad\text{such that}\quad
	\restrict{\FRren}{\Kinematics_{\scalelog}}
	= \exp_{\convolution}\left( -\period \scalelog \right) \convolution \FRren
	\label{eq:rge-infinitesimal}%
\end{equation}
reveals $\FRren(G)$ as a polynomial in $\scalelog$. If $\FRren(G)$ depends only on a single kinematic invariant $\theta$, we call $G$ to be \emph{one-scale} and conclude that it is a polynomial in $\log(\theta/\widetilde{\theta})$ and completely determined by the period map alone.

In general, periods depend on the chosen renormalization point $\widetilde{\Kinematics}$. From \eqref{eq:rge-finite} one infers that the periods $\period'$ for the point $\widetilde{\Kinematics}'$ are related by the conjugation
\begin{equation}
	\period'
	= \restrict{\FRren}{\widetilde{\Kinematics}'}^{\convolution - 1}
		\convolution
		\period
		\convolution
		\restrict{\FRren}{\widetilde{\Kinematics}'}^{}.
	\label{eq:period-scheme-dependence}%
\end{equation}
This implies that $\period(G) = \period'(G)$ is independent of the renormalization point when $G$ is primitive. In section~\ref{sec:ex-periods} we return to the computation of these interesting numbers.

We give a detailed account of the algebraic structures and proofs of the results presented above in \cite{Panzer:Mellin,Panzer:PM2012}.

\subsection{Parametric representation}
\label{sec:renormalization-parametric}%
This general formulation of renormalization is now applied to Feynman integrals in the representation \eqref{eq:feynman-integral-parametric}. Our subtractions for the renormalization are determined by a choice $\widetilde{\Kinematics}$ of reference values for the kinematic invariants, so we let
$ \widetilde{\phipol}_G \defas \restrict{\phipol_G}{\widetilde{\Kinematics}}$ denote the second Symanzik polynomial \eqref{eq:graph-polynomials-combinatorial} evaluated at these values of masses and momenta.
The following formula for the renormalized Feynman rules $\FRren$ follows from \eqref{eq:forest-formula}, \eqref{eq:birkhoff} and \eqref{eq:feynman-integral-parametric} and was discussed in \cite{BlochKreimer:MixedHodge}:
\begin{equation}
	\FRren (G)
	= \left[
			\prod_{e \in \edges}
			\int_{0}^{\infty} \frac{\SP_e^{\EP_e - 1}\ \dd\SP_e}{\Gamma(\EP_e)}
		\right]
		\sum_{F \in \forests(G)}
		\!\!\!(-1)^{\abs{F}}\,
		\frac{
			e^{ -\frac{\phipol_{G\!/\!F}}{\psipol_{G\!/\!F}} }
			-
			e^{ -\frac{\widetilde{\phipol}_{G\!/\!F}}{\psipol_{G\!/\!F}}}
		}{
			\psipol_F^{\dimension/2}
		}
		\prod_{\gamma \in F}
		e^{ -\frac{\widetilde{\phipol}_{\gamma\!/\!F}}{\psipol_{\gamma\!/\!F}} }
	.
	\label{eq:feynman-renormalized-parametric} %
\end{equation}
Here we abbreviate $\psipol_F \defas \psipol_{G\!/\!F} \prod_{\gamma \in F} \psipol_{\gamma/\!F}$. This integral is absolutely convergent. As in section~\ref{sec:projective-integrals} we rescale all Schwinger parameters by $\Pscale$ such that each forest contributes an integral of the form $\int_0^{\infty} \frac{\dd \Pscale}{\Pscale} \big[e^{-\Pscale A}-e^{-\Pscale B} \big] = -\ln \frac{A}{B}$, so
\begin{equation}
	\FRren (G)
	= \frac{1}{\prod_{e \in \edges} \Gamma(\EP_e)}
		\int \Omega
			\prod_{e \in \edges} \SP_e^{\EP_e -1}
		\!\!\!\!\sum_{F \in \forests(G)} \!\!
		\frac{(-1)^{1+\abs{F}}}{\psipol_F^{\dimension/2}}
		\ln \frac{
			\frac{\phipol_{G\!/\!F}}{\psipol_{G\!/\!F}}
			+
			\sum_{\gamma \in F}
				\frac{\widetilde{\phipol}_{\gamma\!/\!F}}{\psipol_{\gamma\!/\!F}}
		}{
			\frac{\widetilde{\phipol}_{G\!/\!F}}{\psipol_{G\!/\!F}}
			+
			\sum_{\gamma \in F}
				\frac{\widetilde{\phipol}_{\gamma\!/\!F}}{\psipol_{\gamma\!/\!F}}
		}
	.
	\label{eq:feynman-renormalized-projective} %
\end{equation}
This representation has been studied in great detail and extensions to incorporate quadratic divergences are available \cite{BrownKreimer:AnglesScales}. By definition~\ref{def:period}, the period becomes
\begin{equation}
	\period(G)
	= \frac{1}{\prod_{e \in \edges} \Gamma(\EP_e)}
		\int \Omega
		\sum_{F \in \forests(G)}
		\frac{(-1)^{\abs{F}}}{\psipol_F^{\dimension/2}}
			\frac{
				\frac{\widetilde{\phipol}_{G\!/\!F}}{ \psipol_{G\!/\!F} }
			}{
				\frac{\widetilde{\phipol}_{G\!/\!F} }{\psipol_{G\!/\!F} }
				+ \sum_{\gamma \in F}
				\frac{\widetilde{\phipol}_{\gamma\!/\!F}}{ \psipol_{\gamma\!/\!F} }
			},
	\label{eq:period-projective}%
\end{equation}
because $\restrict{\phipol}{\Kinematics_{\scalelog}} = \phipol e^{\scalelog}$ in contrast to $\widetilde{\phipol}$ which is independent of $\Kinematics$ and $\scalelog$.
\begin{example}[Primitive divergence]
	Consider a logarithmically divergent graph $G$ without subdivergences and all indices $\EP_e = 1$. The renormalized Feynman rule and the period are
	\begin{equation}
		\FRren(G) 
		= -
			\int \frac{\Omega}{\psipol^{\dimension/2}}
			\ln \frac{\phipol}{\widetilde{\phipol}}
		\quad\text{and}\quad
		\period(G)
		= \int \frac{\Omega}{\psipol^{\dimension/2}}.
		\label{eq:period-primitive-projective}%
	\end{equation}
	So indeed, $\period(G)$ is independent of the renormalization point (the integrand does not contain $\widetilde{\phipol}$) and we see that
	$ \period(G) =- \restrict{\partial_{\scalelog}}{\scalelog=0} \restrict{\FRren(G)}{\widetilde{\Kinematics}_{\scalelog}}$ holds indeed. If $G$ is one-scale, then $ \FRren(G) = - \scalelog \cdot \period(G)$ is just a logarithm $\scalelog = \ln(\phipol/\widetilde{\phipol}) = \ln (\theta/\widetilde{\theta})$ of the ratio of the scale $\Kinematics=\set{\theta}$ and its value at the renormalization point.

	In dimensional regularization, we set $\dimension = \dimension_0 - 2\varepsilon$ and find
	$\sdd 
	= \varepsilon\loops{G}$ if $G$ is logarithmically divergent in $\dimension_0$ dimensions. The unrenormalized Feynman rules converge for $\varepsilon>0$ and give the Laurent series
	\begin{align}
		\FRdim(G)
		&= \Gamma(\sdd) 
		\int\frac{\Omega}{\psipol^{\dimension/2}} \left( \frac{\psipol}{\phipol} \right)^{\sdd}
		\!\!\!= \Gamma(\varepsilon\loops{G}) 
			\sum_{n \geq 0} \frac{(-\varepsilon)^n}{n!}
				\int \frac{\Omega}{\psipol^{\dimension_0/2}} 
				\ln^n \frac{\phipol^{\loops{G}}}{\psipol^{1+\loops{G}}}
		\label{eq:primitive-epsilon-expansion}%
		\\
		&= \frac{\period(G)}{\varepsilon\loops{G}}
			 + \bigo{\varepsilon^0}.
		\label{eq:primitive-residue-period}%
	\end{align}
	So the period appears as the residue of the regularized Feynman rules at $\varepsilon \rightarrow 0$. Epsilon-expansions like \eqref{eq:primitive-epsilon-expansion} can often be computed with hyperlogarithms, see the examples in chapter~\ref{chap:examples}.
\end{example}

\section{Vertex-width three}
\label{sec:vw3} %
In this section we take a close look on a particular class of massless Feynman graphs which is infinite yet so special that they can all be computed explicitly in terms of multiple polylogarithms. This result was obtained originally by Francis Brown for $0$- and $1$-scale graphs. We will give a new proof which extends to the case of three massive external particles.

\begin{definition}
	\label{def:vw}%
	A \emph{construction} $\sigma=(e_1,\ldots,e_{\abs{\edges}})$ of a graph $G$ is a total order on its edges. It defines sequences $G_k$, $G^{k}$ of graphs induced\footnote{So $\edges(G_k) \defas \set{e_1,\ldots,e_k}$ and $\vertices(G_k) \defas e_1 \cup \ldots \cup e_k$ contains all vertices touched by any $e \in \edges(G_k)$.} by the edges $\set{e_1,\ldots,e_k}$ and $\big\{e_{k+1},\ldots,e_{\abs{\edges}}\big\}$, respectively. The \emph{vertex-width} of a construction is
	\begin{equation}
		\vw\left( e_1,\ldots,e_{\abs{\edges}} \right)
		\defas
		\max_{1\leq k < \abs{\edges}} \abs{\actives_k},\quad
		\actives_k
		\defas
		\vertices\left( G_k \right) \cap \vertices\left( G^{k+1} \right)
		\label{eq:def-vw-order}%
	\end{equation}
	and we define the \emph{vertex-width} of the graph $G$ as the minimum over all constructions:
	\begin{equation}
		\vw(G)
		\defas
		\min_{\sigma} \vw \left( e_{\sigma(1)},\ldots,e_{\sigma(\abs{\edges})} \right)
		.
		\label{eq:def-vw}%
	\end{equation}
\end{definition}
The idea is this: Suppose we start with the empty graph 
\mbox{$\emptyset = G_{0} \subsetneq G_1 \subsetneq \cdots \subsetneq G_{\abs{\edges}} = G$}
and construct $G$, adding one edge at a time in the given order $\sigma$. The remaining edges form the graphs 
\mbox{$G=G^0 \supsetneq G^1 \supsetneq \cdots \supsetneq G^{\abs{\edges}} = \emptyset$}
and at each stage $k$ share a set
\mbox{$\actives_k = \vertices ( G_k ) \cap \vertices\big( G^{k+1} \big)$}
of \emph{active} vertices with the so far constructed $G_k$. The vertex-width bounds the size of these $\actives_k$.

Figure~\ref{fig:zz5-vw3} shows a construction $\sigma$ of the zigzag graph $\ZZ{5}$ with
$\vw(\sigma) = 3$. Obviously there are infinitely many connected graphs $G$ with $\vw(G) \leq 3$, including all zigzag graphs $\ZZ{n}$ and the wheels $\WS{n}$ with $n$ spokes. The aforementioned result is

\begin{theorem}[theorem 118 and corollary 122 of \cite{Brown:PeriodsFeynmanIntegrals}]
	If $\vw(G)\leq 3$, then all periods of $G$ are in $\MZV$.
	\label{theorem:vw3-Brown} %
\end{theorem}
This statement means that all coefficients of the $\varepsilon$-expansion of $\int I_G \ \Omega$ (expanding indices $\EP_e = \EPZ_e + \varepsilon\EPE_e$ near integers $\EPZ_e$ and the dimension $\dimension\in 2\N-2\varepsilon$ near an even integer) are rational linear combinations of multiple zeta values. By \eqref{eq:feynman-integral-projective} this property carries over to the Feynman integral $\FR(G)$, up to the $\Gamma$-prefactors which introduce the Euler-Mascheroni constant $\gamma_E$ into the expansion.

\begin{figure}
	\begin{gather*}
		\Graph[0.5]{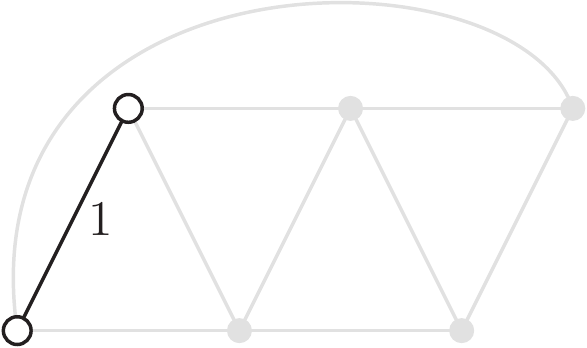}
		\Graph[0.5]{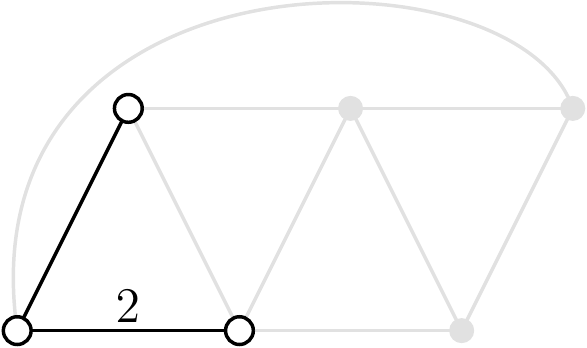}
		\Graph[0.5]{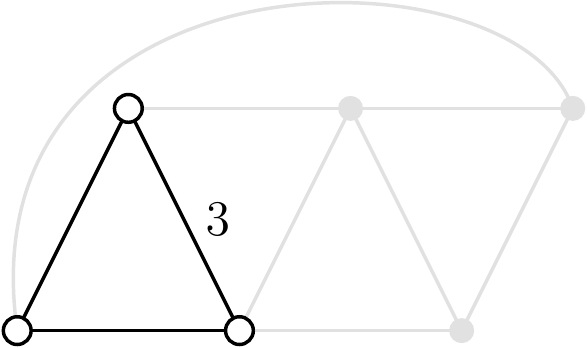}
		\Graph[0.5]{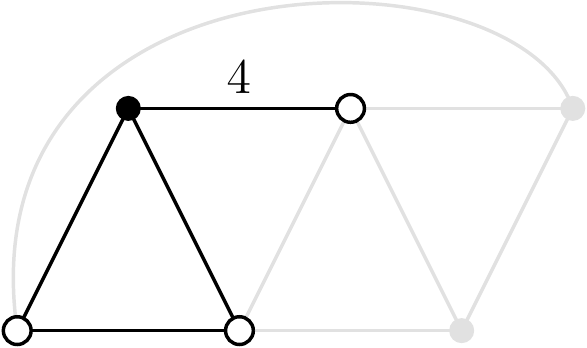}
		\Graph[0.5]{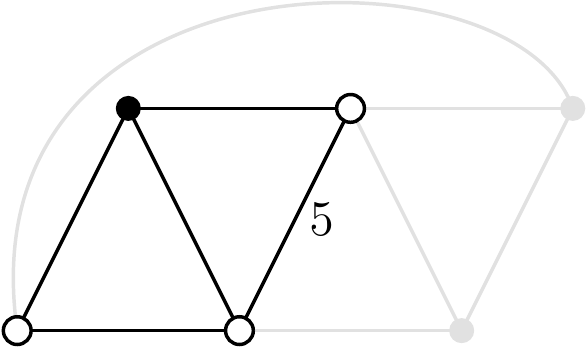}
		\\
		\Graph[0.5]{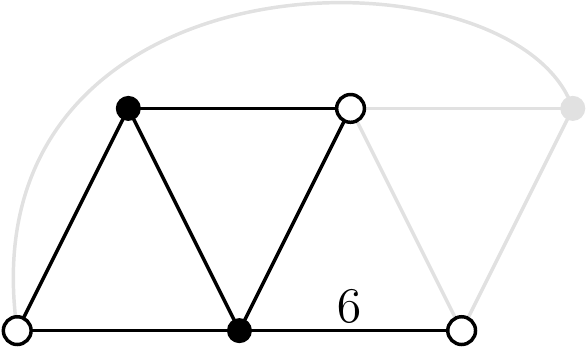}
		\Graph[0.5]{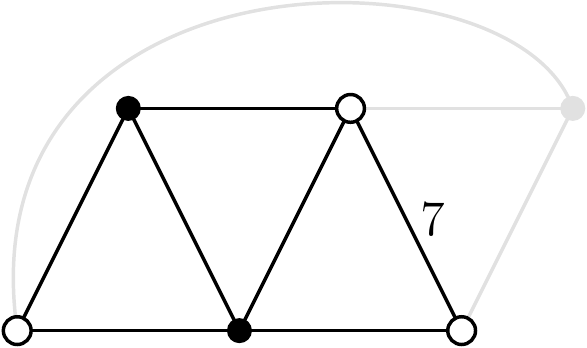}
		\Graph[0.5]{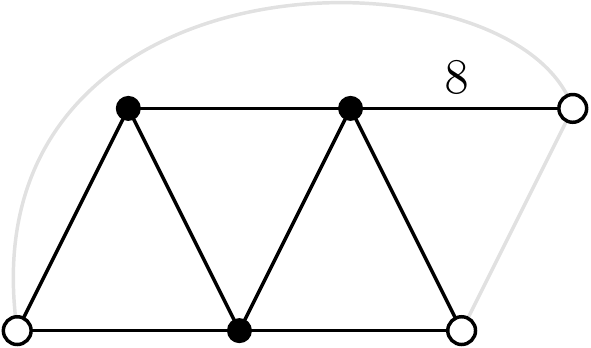}
		\Graph[0.5]{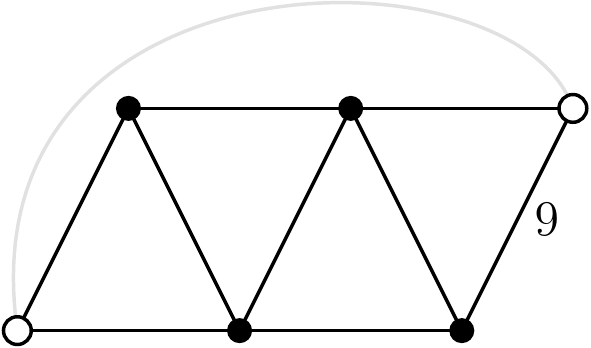}
		\Graph[0.5]{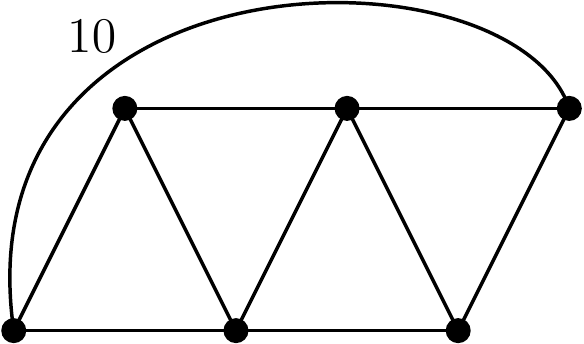}
	\end{gather*}%
	\caption[A construction of the zigzag graph $\ZZ{5}$]{A construction $e_1,\ldots,e_{10}$ of the zigzag graph $G=\ZZ{5}$: $G_k$ is drawn in black, $G^{k}$ in grey and they intersect in the vertices $\actives_k$ (white circles). These are never more than three, so $\vw(\ZZ{5}) \leq 3$ and in fact equality holds.}%
	\label{fig:zz5-vw3}%
\end{figure}

\subsection{Some general properties}
\begin{theorem}
	\label{theorem:vw3-planar}%
	Every graph $G$ with $\vw(G) \leq 3$ is planar.
\end{theorem}

\begin{proof}
	Since $G$ has $\vw(G) \leq r$ if and only if each of its connected components $H\in\comps(G)$ meets $\vw(H) \leq r$ as well, we may restrict to connected $G$. We can also exclude any parallel edges, self-loops or vertices of valency one (any of these can simply be added without destroying the planarity of an embedding).

	We take any construction which achieves $\vw\big(e_1,\ldots,e_{\abs{\edges}}\big) \leq 3$ and inductively assign polar coordinates $r\colon \vertices \rightarrow \N$ and $\phi\colon \vertices \rightarrow \set{0,\frac{2}{3}\pi,\frac{4}{3}\pi}$ such that drawing all edges as straight lines yields a planar embedding of $G$.
	
	Our algorithm iterates over $k$ from $1$ to $\abs{\edges}$. As illustrated in figure~\ref{fig:vw3-planarity-proof}, at each stage $k$ exactly one of the following cases occurs:
	\begin{enumerate}
		\item[(1)] $e_k$ connects $v,w\in\actives_{k-1}$

		\item[(2)] $e_k$ connects $v,w \notin \actives_{k-1}$: Since $v,w$ are incident to at least one further vertex each, we will have $\actives_{k-1} \cupdot \set{v,w} \subseteq \actives_k$ and therefore $\abs{\actives_{k-1}} \leq \abs{\actives_k} - 2 \leq 1$. Hence we can choose $\phi(v) \neq \phi(w)$ both distinct from $\phi(\actives_{k-1})$ and further set $r(v) = r(w) \defas k$.

		\item[(3)] $e_k$ connects one vertex $v\in \actives_{k-1}$ with one vertex $w\in \vertices(G^k)\setminus \actives_{k-1}$: If $v \notin \actives_k$, set $\phi(w) \defas \phi(v)$ and $r(w) \defas k$. Otherwise we must have $\actives_k = \actives_{k-1} \cupdot \set{w}$ ($w$ is incident to at least one further edge, so $w \in \actives_{k}$) and from $\abs{\actives_k} \leq 3$ we know $\abs{\actives_{k-1}} \leq 2$, so we can choose some $\phi(w) \notin \phi(\actives_{k-1})$ and set $r(w) \defas k$.
	\end{enumerate}
	This construction ensures that for any $k$, the embedding of $G_k$ with straight lines is contained in the triangle with corners $\Delta_k = \set{v_{\theta}}$, where $v_{\theta} \in \vertices_{k,\theta} \defas \vertices(G_k) \cap \phi^{-1}(\theta)$ denotes the farthest vertex $r(v_{\theta}) = \max r(\vertices_{k,\theta})$ of $G_k$ on the ray $\phi=\theta$. In particular $\actives_k \subseteq \Delta_k$ is a subset of these corners.
	
	By construction all edges lie on the sides of such triangles $\Delta_{k}$, except for the radial edges in case (3) when $v \notin \actives_{k}$. None of these can cross and planarity is obvious.
\end{proof}
\begin{figure}\centering
	\begin{tabular}{cccc}
		$\Graph[0.75]{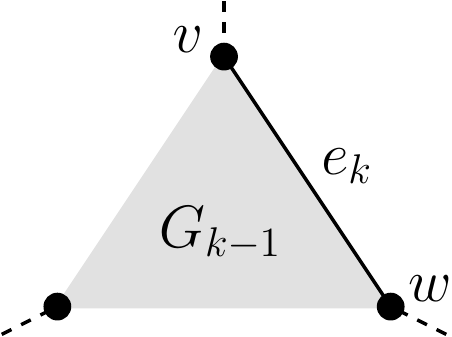}$ &
		$\Graph[0.75]{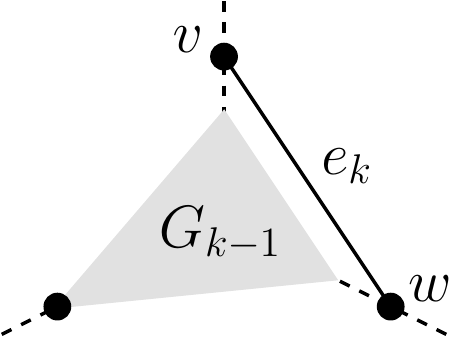}$ &
		$\Graph[0.6]{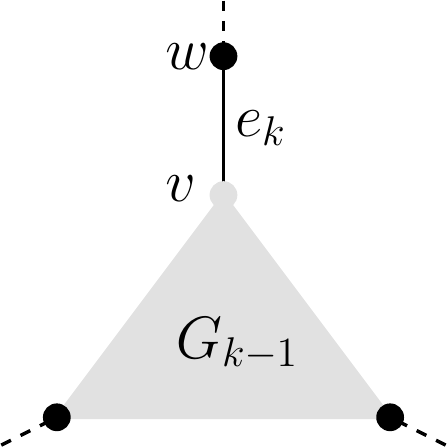}$ &
		$\Graph[0.75]{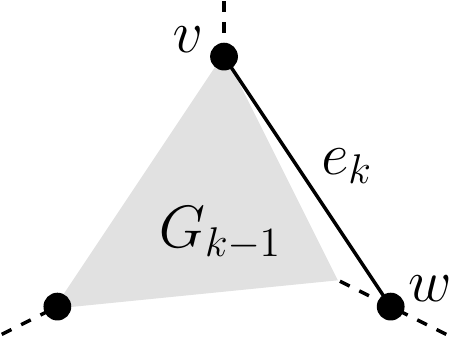}$ \\
		(1) & (2) & (3): $v \notin \actives_k$ & (3): $v \in \actives_k$ \\
	\end{tabular}%
	\caption[Different cases in the construction of a planar embedding of a graph with vertex-width $3$]{The proof of theorem~\ref{theorem:vw3-planar} distinguishes the displayed cases to extend the planar embedding of $G_{k-1}$ (grey) by the edge $e_k=\set{v,w}$. Any forthcoming edges can connect only at the extremal vertices (black dots) on each ray of constant $\phi$ (dashed).}%
	\label{fig:vw3-planarity-proof}%
\end{figure}
\begin{remark}\label{rem:planar-dual-construction}%
	From this construction it follows that the same sequence of edges gives rise to a construction of the planar dual $\widehat{G}$ of $G$ (relative to this planar embedding) with $\vw \leq 3$ as well. Note that for $3$-connected $G$, the planar embedding and $\widehat{G}$ are unique \cite{Whitney:UniqueDual}.
\end{remark}
The sets $\actives_k$ are cuts of $G$, so the vertex-width $\vw(G) \geq \vconn(G)$ bounds the \emph{connectivity}
\begin{equation}
	\vconn(G)
	\defas
	\max \setexp{n \in \N_0}{G \setminus C\ \text{is connected for all}\ C \subset \vertices(G)\ \text{with}\ \abs{C} = n}
	.
	\label{eq:def-connectivity}%
\end{equation}
As mentioned in section~\ref{sec:one-scale-insertions}, for the computation of Feynman integrals we only need to consider $3$-connected simple graphs $G$, $\vconn(G) \geq 3$. In this case each vertex is at least $3$-valent and $\abs{\actives_k} = 3$ for all $2 \leq \abs{\edges}-2$. Furthermore the first three edges $e_{\sigma(1)}$, $e_{\sigma(2)}$ and $e_{\sigma(3)}$ of any construction $\sigma$ of $G$ with $\vw(G) =3$ must either form a triangle or a star:
$e_{\sigma(1)} = \set{v_1,w}$ and $e_{\sigma(2)} = \set{v_2,w}$ share one vertex $w$ (otherwise $\actives_2 = e_{\sigma(1)} \cupdot e_{\sigma(2)}$ has four elements) and if $w \notin e_{\sigma(3)}$, the third edge can only connect $v_1$ with $v_2$.

One can therefore test for $\vw(G)\leq 3$ very efficiently with
\begin{lemma}
	\label{lemma:vw3-test}%
	Given any simple and $3$-connected graph $G$, an algorithm can decide $\vw(G) = 3$ (and if positive provide a construction $\sigma$ with $\vw(\sigma) = 3$) in time $\bigo{\abs{\vertices} \cdot \abs{\edges}}$.
\end{lemma}
\begin{proof}
	Suppose $G$ has a construction $\sigma$ with $\vw(\sigma) = 3$ that starts out with a triangle $\Delta$. Then $e_{\sigma(4)} = \set{v,w}$ must connect one of $v \in \Delta$ to a new vertex $w \notin \Delta$, so necessarily $v \notin \actives_4 = \set{w} \cupdot \Delta\setminus \set{v}$ and $v$ is $3$-valent. Swapping $\sigma(4)$ with the edge of $\Delta$ that does not touch $v$ yields a construction $\sigma'$ that also achieves $\vw(\sigma') = 3$.
	
	Thus we only need to look for constructions that begin with a star $e_{\sigma(i)} = \set{v_i, w}$ ($1 \leq i \leq 3$) defined by some three-valent vertex $w$. Starting from $\actives \defas \set{v_1,v_2,v_3}$ and $I \defas \edges \setminus \set{e_1,e_2,e_3}$, repeat the following steps as often as possible:
	\begin{itemize}
		\item Remove any edges from $I$ that connect active vertices: $I \defas I \setminus \setexp{e \in I}{e \subseteq \actives}$.
		\item If some $v \in \actives$ is incident to only one edge $e=\set{v,w} \in I$, remove $e$ from $I$ and replace $v$ by its neighbour $w$: $\actives \defas \actives \setminus \set{v} \cupdot \set{w}$.
	\end{itemize}
	If this process ends in $I = \emptyset$, the order $\sigma$ in which edges were removed from $I$ is a construction with $\vw(\sigma) = 3$. Otherwise, $I \neq \emptyset$ proves that any construction of $G$ starting with the star around $w$ must have $\vw(\sigma) \geq 4$.

	To implement this test it suffices to scan through the edges $e = \set{v,w}$ incident to $v$ every time a vertex $v$ is added to $A$: After deletion of those with $w \in A$, the algorithm is iterated as some one-valent vertex in $A$ is replaced by its neighbour. This procedure requires a time linear in $\abs{\edges}$ and there are at most $\abs{\vertices}$ initial stars ($3$-valent vertices $w$) to check, so the total runtime is in $\bigo{\abs{\vertices} \cdot \abs{\edges}}$.
\end{proof}
Apparently the vertex-width can not decrease when an edge is removed or contracted, so $\vw(H) \leq \vw(G)$ for all minors $H \minoreq G$.\footnote{%
	A \emph{minor} of a graph $G$ is any graph $H = G\setminus I / K$ obtained from deletion and contraction of disjoint sets $I \cupdot K \subseteq \edges(G)$ of edges.%
}
Hence the theorem of Robertson and Seymour \cite{RobertsonSeymour:XX} applies: The set of graphs $G$ with $\vw(G) \leq 3$ can be characterized by a finite set of forbidden minors. For example, the five graphs shown in figure~\ref{fig:vw3-forbidden-minors} each have a vertex-width of four and can thus not appear as a minor of $G$ when $\vw(G) \leq 3$. Recently the sufficiency of this condition was proven \cite{BlackCrumDeVosYeats:ForbiddenMinorsSplit5} and we quote
\begin{theorem}
	A simple, $3$-connected graph $G$ has vertex-width $\vw(G)=3$ if and only if it contains none of $\set{K_{3,3}, K_5, C, O, H}$ as a minor (see figure~\ref{fig:vw3-forbidden-minors}).
	\label{theorem:vw3-minors} %
\end{theorem}%
\begin{figure}\centering
	\begin{tabular}{ccccc}
		$\Graph[0.45]{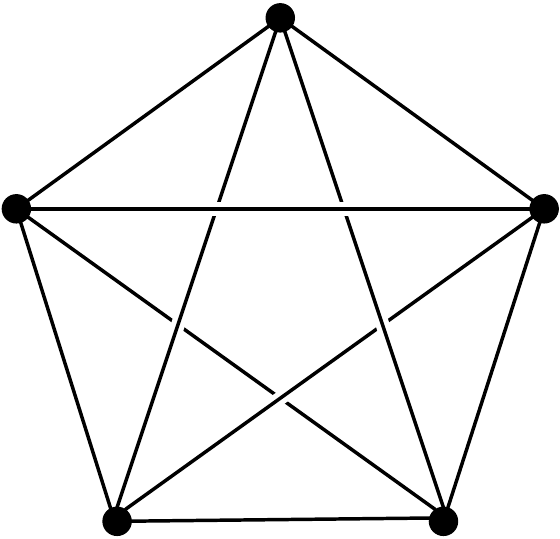}$ & 
		$\Graph[0.4]{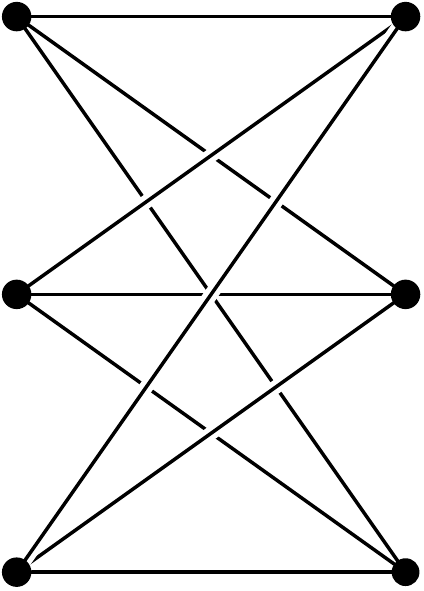}$ &  $\Graph[0.55]{cube}$ & $\Graph[0.5]{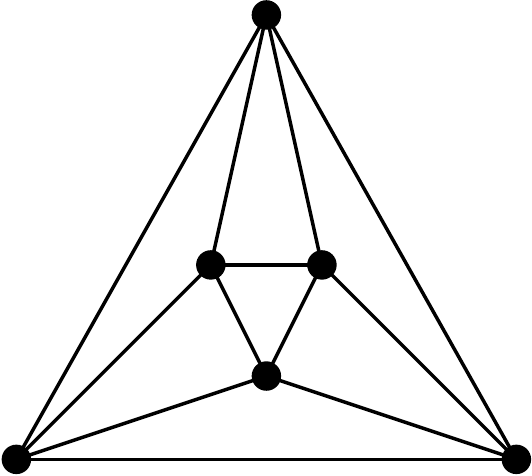}$ & $\Graph[0.5]{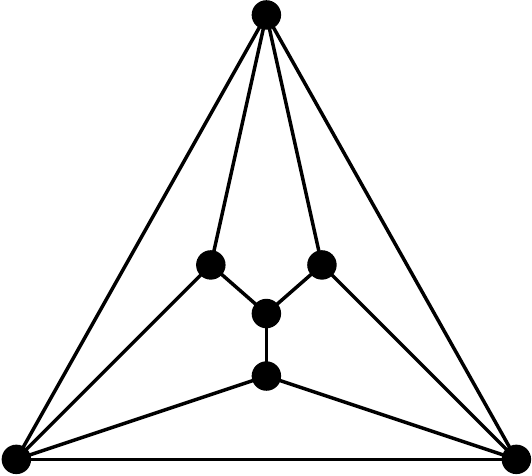}$ \\
		$K_5$ & $K_{3,3}$ & $C$ & $O$ & $H$ \\
	\end{tabular}%
	\caption[Forbidden minors for simple $3$-connected graphs with $\vw(G) =3$]{The forbidden minors for simple $3$-connected graphs $G$ with $\vw(G) = 3$ from theorem~\ref{theorem:vw3-minors} contain the non-planar complete graphs $K_5$ and $K_{3,3}$, as well as three polyhedra: The cube $C$ together with its dual (the octahedron $O$) and the self-dual heptahedron $H$.}%
	\label{fig:vw3-forbidden-minors}%
\end{figure}%
Note that this result entails an alternative (but non-constructive) proof of theorem~\ref{theorem:vw3-planar} via Wagner's theorem \cite{Wagner:EbeneKomplexe} ($K_{3,3}$ and $K_5$ are the forbidden minors for planarity).

Very interestingly, these forbidden minors were originally discovered by Iain Crump in his thesis \cite{Crump:ForbiddenMinors3Connected} to describe the seemingly unrelated \emph{splitting} property as we will mention in the following section.

\subsection{Denominator reduction}
\label{sec:denominator-reduction} %
In chapter~\ref{chap:hyperlogs} we explain the integration with hyperlogarithms in detail. But we anticipate that we will compute the integral of say $I_0 = \psipol^{-2}$ step by step, constructing the \emph{partial integrals}
\begin{equation}
	I_{k} \defas \int_0^{\infty} I_{k-1}\ \dd \SP_{e_k}
	\quad\text{for all}\quad
	1 \leq k < \abs{\edges}
	\label{eq:def-partial-integral} %
\end{equation}%
\nomenclature[Ik]{$I_k$}{partial integral, equation~\eqref{eq:def-partial-integral}\nomrefpage}%
along a suitable construction $\sigma=(e_1,\ldots,e_{\abs{\edges}})$ of $G$. We will find that $I_k = \sum_i f_i L_i$ is a linear combination of hyperlogarithms $L_i$ with rational prefactors $f_i$. This requires \emph{linear reducibility}, which in particular implies that the denominator of each $f_i$ must factor linearly in the next Schwinger variable $\SP_{e_{k+1}}$.

If indeed a summand of $I_k$ has a denominator $D_k = (a \SP_{e_{k+1}} + b)(c \SP_{e_{k+1}} + d)$ that factors, the integration of $\SP_{e_{k+1}}$ results in a contribution to $I_{k+1}$ with denominator $D_{k+1} \defas ad - bc$. In general these quadratic polynomials do not factorize.

But in their seminal work, Bloch, Esnault and Kreimer showed that during the first four integrations ($k \leq 4$), all denominators are products of linear polynomials \cite[section~8]{BlochEsnaultKreimer:MotivesGraphPolynomials}. Brown subsequently introduced these\footnote{Note that the polynomials in \cite{BlochEsnaultKreimer:MotivesGraphPolynomials} are defined via the \emph{cycle} (or \emph{circuit}) \emph{matrix} instead of the incidence matrix $\IM$.} as \emph{Dodgson polynomials} \cite{Brown:PeriodsFeynmanIntegrals} with
\begin{definition}
	\label{def:dodgson}%
	For any sets $I,J, K \subset \edges$ with $\abs{I} = \abs{J}$, the \emph{Dodgson polynomial} is
	\begin{equation}
		\dodgson^{I,J}_K
		\defas
		\restrict {\det \GM(I,J) }{\SP_e = 0 \ \forall e\in K},
		\label{eq:def-dodgson} %
	\end{equation}%
\nomenclature[Psi IJ K]{$\dodgson^{I,J}_K$}{Dodgson polynomial, equation~\eqref{eq:def-dodgson}\nomrefpage}%
	where $\GM(I,J)$ denotes the minor of $\GM$ obtained by deleting rows $I$ and columns $J$.
\end{definition}
\begin{remark}
	Dodgson polynomials depend on the choices made in the construction of a graph matrix $\GM$ for $G$ through an overall sign, so we understand that we stick to one particular matrix $\GM$ and thereby fix the order of its rows (edges) and columns (vertices) throughout.
	For the same reason, the orientation of the edges (signs in the incidence matrix $\IM$) and the deleted column $v_0 \in \vertices$ must stay the same.
\end{remark}
In particular all denominators factor into such linear polynomials $\dodgson^{I,J}_K$ with $I \cup J \cup K \subseteq \set{e_1,e_2,e_3,e_4}$, as long as $k \leq 4$. Therefore the first obstruction to linear reducibility can occur at $k=5$ and indeed a new polynomial enters the game at this stage: the \emph{five-invariant}
\begin{equation}
	\fiveinv{i,j,k,l,m}
	=
	\pm \det \begin{pmatrix}
		\dodgson^{ij,kl}_{m} & \dodgson^{ijm,klm} \\
		\dodgson^{ik,jl}_{m} & \dodgson^{ikm,jlm} \\
	\end{pmatrix}
	\label{eq:def-fiveinv}%
\end{equation}
which is associated to a set $\set{i,j,k,l,m} \subseteq \edges$ of five distinct edges. Any permutation of these changes the determinant in \eqref{eq:def-fiveinv} only by an overall sign \cite{Brown:PeriodsFeynmanIntegrals}.
In the special case when $I_0 = \psipol^{-2}$, it even turns out that $I_5 = L/D_5$ is a trilogarithm $L$ divided by the common denominator $D_5 = \fiveinv{e_1,\ldots,e_5}$.

A typical five-invariant of a complicated graph has irreducible quadratic components. But when one of the four Dodgson polynomials in \eqref{eq:def-fiveinv} vanishes, $\fiveinv{i,j,k,l,m}$ degenerates into the product of two Dodgson polynomials and we say that $\set{i,j,k,l,m}$ splits.  Francis Brown \cite{Brown:PeriodsFeynmanIntegrals} proved that when $e_1,\ldots,e_{\abs{\edges}}$ is a construction with vertex-width three, all denominators are Dodgson polynomials (thus linear). In particular, $\fiveinv{e_{\sigma(1)},\ldots,e_{\sigma(5)}}$ splits.

Actually, a stronger statement can be made.
\begin{theorem}[Iain Crump \cite{Crump:ForbiddenMinors3Connected}]
	A simple, $3$-connected graph $G$ splits if and only if it contains none of $\set{K_{3,3}, K_5, C, O, H}$ as a minor. Furthermore, this condition is equivalent to $\vw(G)=3$.
	\label{theorem:split-minors} %
\end{theorem}
This means that not only $\fiveinv{e_1,\ldots,e_5}$, but in fact \emph{every} five-invariant of $G$ splits when $\vw(G) \leq 3$. Conversely, the splitting of every five-invariant in a $3$-connected graph requires $\vw(G) = 3$. Hence graphs with $\vw(G) \leq 3$ are extremely special from this viewpoint and in \cite{Crump:ForbiddenMinors3Connected} we even find
\begin{conjecture}
	If $\vw(G) \leq 3$, then it is linearly reducible with respect to any ordering of its edges and the polynomials in the reduction are all Dodgson polynomials.
\end{conjecture}

We like to mention that the characterization (in terms of forbidden minors) of graphs with vertex-width less than or equal to three is much more complicated when we drop the requirement of $3$-connectedness. This very intricate problem was solved in \cite{BlackCrumDeVosYeats:ForbiddenMinorsSplit5}.

\subsubsection{Quadratic identities}

Known factorizations of denominators are consequences of local combinatorics (e.\,g.\ the presence of triangles or $3$-valent vertices) and the representation of the graph polynomial $\psipol = \det \GM$ as a determinant. Already Stembridge \cite{Stembridge:CountingPoints} observed that the \emph{Dodgson identity}
\begin{equation}
	\det \GM(ij,ij) \cdot \det \GM
	= \det \GM(i,i) \cdot \det \GM(j,j) - \det \GM(i,j) \cdot \det \GM(j,i)
	\label{eq:dodgson-identity} %
\end{equation}
proves the factorization of the third denominator 
\begin{equation*}
	D_2
	=
	\dodgson^{1}_{2} \dodgson^{2}_{1} - \dodgson^{12}_{} \dodgson_{12}
	=
	\left(\dodgson^{1}_{2} \right)^2
	.
\end{equation*}
In fact, Dodgson \cite{Dodgson:CondensationOfDeterminants} refers to a much more general result as ``well-known'': Jacobi's determinant formula \cite{Brown:PeriodsFeynmanIntegrals}. Its application to the graph matrix can be phrased as
\begin{lemma}[corollary 10 in \cite{VlasevYeats:FourVertexIdentity}]
	\label{lemma:dodgson-jacobi}%
	If the edge sets 
	$I = \set{I_1,\ldots,I_r}$,
	$J=\set{J_1,\ldots,J_r}$,
	$A$,
	$B$,
	$K \subseteq \edges$
	fulfil $A \cap I = B \cap J = \emptyset$ and $\abs{A} = \abs{B}$, then
	\begin{equation}
		\det\left( \dodgson^{A \cup \set{I_i}, B \cup \set{J_j}}_{K} \right)_{1 \leq i,j\leq r}
		= \dodgson^{A \cup I, B \cup J}_{K}
		\left( \dodgson^{A,B}_{K} \right)^{r-1}
		.
		\label{eq:dodgson-jacobi}%
	\end{equation}
\end{lemma}
These and further identities where studied in detail in \cite{Brown:PeriodsFeynmanIntegrals}. Applications to denominator reduction include impressive computations of point-counts of graph hypersurfaces \cite{BrownSchnetz:K3Phi4} and explicit graphical criteria for the \emph{weight-drop} phenomenon \cite{BrownYeats:SpanningForestPolynomials}.

In order to understand these identities combinatorially, we need an alternative description of Dodgson polynomials in the spirit of \eqref{eq:graph-polynomials-combinatorial} for the Symanzik polynomials. This was worked out in \cite{BrownYeats:SpanningForestPolynomials} and we present a slightly more general formulation and proof of
\begin{lemma}
	\label{lemma:dodgson-as-forests}%
	For equinumerous sets $I, J \subset \edges(G)$ of edges, the Dodgson polynomial
	\begin{equation}
		\dodgson_{G}^{I,J}
		= 
		\sum_{P} \epsilon_P^{I,J} \cdot \forestpolynom[G\setminus (I \cup J)]{P}
		\quad\text{with signs}\quad
		\epsilon_P^{I,J} \in \set{\pm 1}
		\label{eq:dodgson-as-forests}%
	\end{equation}
	is a linear combination of spanning forest polynomials.
	These are indexed by partitions $P$ of $\vertices(I \cup J)$ such that $(I\setminus J)/P$ and $(J \setminus I)/P$ are spanning trees, where $H/P$ denotes the graph obtained from the edges $H$ after identification of all vertices that belong to the same parts of $P$.
	The signs are $\epsilon_P^{I,J} = (-1)^{N_{I,J}} \cdot \dodgson_{(I \cup J)/P}^{I,J}$ where
	\begin{equation}
		N_{I,J}
		= r
			+ \abs{\setexp{(e,f)}{\text{$e \in I \cap J$, $f \in I \SymDiff J$ such that $e<f$}}}
			+	\sum_{e \in I \SymDiff J} e
		\label{eq:dodgson-sign-IJ}%
	\end{equation}
	and
	$
		I \SymDiff J 
		\defas 
			I \setminus J \cupdot J \setminus I 
	$ stands for the symmetric difference.
\end{lemma}
\begin{proof}
	Let $I\setminus J = \set{i_1,\ldots,i_r}$ and $J \setminus I = \set{j_1,\ldots,j_r}$ in ascending order ($i_1<\cdots<i_r$). If we move the columns $I \setminus J$ (and rows $J \setminus I$) of the graph matrix $\GM(I,J)$ such that they end up in the first $r$ columns (rows), without changing the relative order of any other columns (rows), then we pick up
	\begin{equation*}
		N_{I,J}
		\equiv
		\sum_{k=1}^{r} \left( i_{k}- \abs{\setexp{e \in J}{e < i_k}} - k \right)
		+ \sum_{k=1}^{r} \left( j_{k}- \abs{\setexp{e \in I}{e < j_k}} - k \right)
		\mod 2
	\end{equation*}
	minus signs: $i_k$ is the $\left[i_k - \abs{\setexp{e\in J}{e < i_k}} \right]$'th column of $\GM(I,J)$ and must be moved to column $k$.

	Now we expand $M(I,J)$ with respect to the Schwinger variables as in \eqref{eq:graph-polynomial-proof-expansion} and obtain, by the matrix-tree-theorem \eqref{eq:matrix-tree},
	\begin{equation}
		\dodgson^{I,J}
		= \sum_{S \supseteq I \cup J} 
				\det \tilde{\IM}(I \cupdot \edges \setminus S) 
				\cdot \det \tilde{\IM}(J \cupdot \edges \setminus S)
				\cdot \prod_{e \notin S} \SP_e
		\label{eq:dodgson-signs-proof-expansion}%
	\end{equation}
	where the non-vanishing contributions are precisely those for which both $S \setminus I$ and $S \setminus J$ are spanning trees. Hence $F \defas S \setminus (I \cup J)$ is a spanning forest that contributes to $\forestpolynom[G \setminus (I \cup J)]{P}$ for the partition $P=\comps(F)$.
	Note that in \eqref{eq:dodgson-signs-proof-expansion}, both minors of the incidence matrix share the identical last rows $\tilde{\IM}_e$ ($e \in F$).

	Contracting any edge $e=\set{v,w} \in F$ amounts to adding the column $w$ to the column $v$, followed by an expansion in the row $e$ (where the only non-zero entry then remains in column $w$).  This multiplies both of the determinants in \eqref{eq:dodgson-signs-proof-expansion} with the same factor and does not change the overall sign. After all $e \in F$ were contracted, any vertices belonging to the same part of $P$ have been identified.
\end{proof}
	To compute $\dodgson_{(I \cup J)/P}^{I,J}$, choose any part $P_0 \in P$ to fix a root of both trees $(I\setminus J)/P$ and $(J \setminus I)/P$. Every edge $i_k \in I\setminus J$ has two endpoints in $(I\setminus J)/P$, and with $\phi_I(k)$ we denote that one which is further away from the root as shown in figure~\ref{fig:dodgson-signs}. This gives us two bijections 
	$\phi_I, \phi_J\colon \set{1,\ldots,r} \longrightarrow P \setminus \set{P_0}$ and we have \cite{BrownYeats:SpanningForestPolynomials}
	\begin{equation}
		\dodgson_{(I \cup J)/P}^{I,J}
		=
		\sgn \left( \phi_I \circ \phi_J^{-1} \right)
		\cdot
		\prod_{k=1}^{r} \left[
			\IM_{i_k, \phi_I(k)} \IM_{j_k, \phi_{J}(k)}
		\right],
		\label{eq:dodgson-as-forests-signs}%
	\end{equation}
	where the product counts how many of the edges in the trees $(I\setminus J)/P$ and $(J\setminus I)/P$ are directed towards the root.
	Note that when $r = 1$, the permutations $\phi_I$ and $\phi_J$ do not play any role.
\begin{figure}\centering
	$ G=\Graph[0.8]{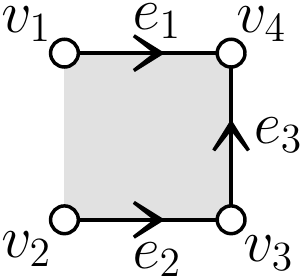}$
	\hfill
	\begin{tabular}{cc@{\hspace{1mm}}cc@{\hspace{1mm}}cc@{\hspace{1mm}}c}
		\toprule
		$I,J$ &
			\multicolumn{2}{c}{$\set{1,2}$, $\set{1,3}$} &
			\multicolumn{2}{c}{$\set{1,2}$, $\set{2,3}$} &
			\multicolumn{2}{c}{$\set{1,3}$, $\set{2,3}$}
		\\
		\midrule
		$H/P$ &
			$ \Graph[0.5]{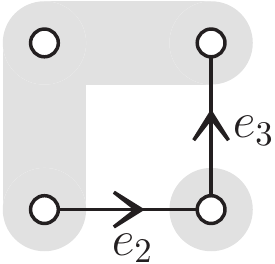} $ &
			$ \Graph[0.5]{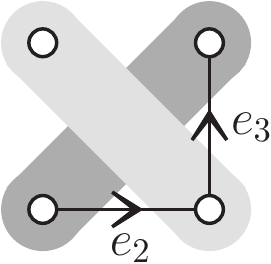} $ &
			$ \Graph[0.5]{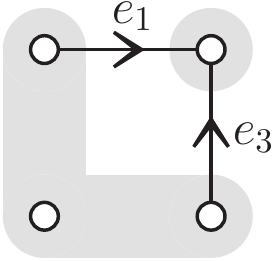} $ &
			$ \Graph[0.5]{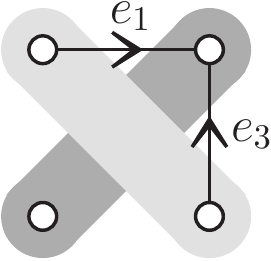} $ &
			$ \Graph[0.5]{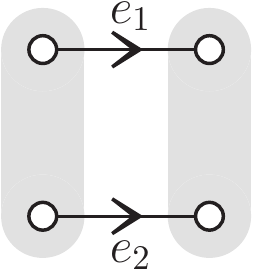} $ &
			$ \Graph[0.5]{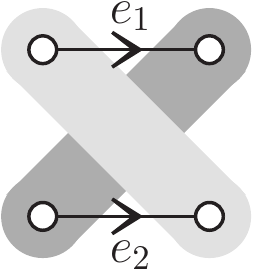} $
		\\\addlinespace
		$P$ &
			$124,3$ &
			$13,24$ &
			$123,4$ &
			$13,24$ &
			$12,34$ &
			$13,24$
		\\
		$\dodgson_{H/P}^{I,J}$ &
			$ -1 $ &
			$ -1 $ &
			$ +1 $ &
			$ +1 $ &
			$ +1 $ &
			$ -1 $
		\\
		\bottomrule
	\end{tabular}%
	\caption[Signs of spanning forest polynomials in Dodgson polynomials determined by associated trees]{The grey area of $G$ indicates that in the drawing we only show the part $H$ of $G$ that is formed by the three edges $e_i$ and the four vertices $v_i$ they touch. In example~\ref{ex:dodgson-signs-boxforest} we compute the corresponding Dodgson polynomials with \eqref{eq:dodgson-as-forests}. The signs are determined by whether the two edges $I \SymDiff J$ of $H/P$ connect both parts in the same direction.}%
	\label{fig:dodgson-signs}%
\end{figure}
\begin{example}
	\label{ex:dodgson-signs-boxforest}%
	We consider $I \cup J \subseteq \set{e_1,e_2,e_3}$ for three particular edges of $G$ arranged and oriented as shown in figure~\ref{fig:dodgson-signs}. Also assume that $e_1<e_2<e_3$ are indeed the first three edges in the order chosen to define the graph matrix.

	For $I=\set{1,2}$ and $J=\set{1,3}$ (so $r=1$) we find that $N_{I,J} = 1 +2 +3 + 2$ from \eqref{eq:dodgson-sign-IJ} is even such that $\epsilon_P^{I,J} = \dodgson_{H/P}^{I,J}$ where $H$ contains the four vertices $v_i$ and the three edges $e_i$.
	The only partitions $P$ to consider are $124,3$ and $13,24$. In both cases, $e_2$ and $e_3$ connect the two parts of $P$ in opposite directions and thus $\epsilon_P^{I,J} = -1$ from \eqref{eq:dodgson-as-forests-signs}:
	\begin{equation}
		\dodgson^{12,13}
		= - \forestpolynom{\set{1,2,4},\set{3}} - \forestpolynom{\set{1,3},\set{2,4}}
		.
		\label{eq:dodgson-as-forest-boxladder-1} %
	\end{equation}
	Similarly, we obtain expressions for the Dodgson polynomials
	\begin{equation}
		\dodgson^{12,23}
		= \forestpolynom{\set{1,2,3},\set{4}} + \forestpolynom{\set{1,3},\set{2,4}}
		\quad\text{and}\quad
		\dodgson^{13,23}
		= \forestpolynom{\set{1,2},\set{3,4}} - \forestpolynom{\set{1,3},\set{2,4}}
		\label{eq:dodgson-as-forest-boxladder-2}%
	\end{equation}
	from analyzing the partitions summarized in the table of figure~\ref{fig:dodgson-signs}. Note that $N_{I,J}$ is even in all these cases.
\end{example}

\subsection{Forest functions}
\label{sec:forest-functions}%
The main aim of this section to provide an independent and simplified proof of theorem~\ref{theorem:vw3-Brown} (linear reducibility of graphs with $\vw{G} \leq 3$) which was originally stated in \cite{Brown:PeriodsFeynmanIntegrals}. In particular we follow an inductive approach and instead of considering the full graph $G$ at once, we use functions of only three variables and build up $G$ edge by edge.

This recursion is very similar in spirit to the construction of graphical functions in \cite{Schnetz:GraphicalFunctions}, but formulated in Schwinger parameters.

\begin{definition}
	\label{def:GfunForest}%
	Let $G$ be a graph with three marked vertices $\set{v_1,v_2, v_3} = \vertices_{\Text}(G)$. We introduce the abbreviations $f = (f_1, f_2, f_3)$ for the spanning forest polynomials
	\begin{equation}
		f_1 \defas \forestpolynom[G]{\set{v_1},\set{v_2,v_3}}
		,\quad
		f_2 \defas \forestpolynom[G]{\set{v_2},\set{v_1,v_3}}
		\quad\text{and}\quad
		f_3 \defas \forestpolynom[G]{\set{v_3},\set{v_1,v_2}}
		\label{eq:3pt-forest-shorthand} %
	\end{equation}
	and define the \emph{forest function} $\GfunForest{G}\colon \R_+^3 \longrightarrow \R_{+}$ of $G$ by
	\begin{equation}
		\GfunForest{G}(z)
		=
		\GfunForest{G}(z_1,z_2,z_3)
		\defas
		\int_{0}^{\infty}
			\psipol^{-\dimension/2}
			\prod_{i=1}^{3}
				\delta\left( \frac{f_i}{\psipol} - z_i \right)
			\cdot
			\prod_{e \in \edges} \SP_e^{\EP_e -1} \dd \SP_e
		.
		\label{eq:def-GfunForest} %
	\end{equation}
\end{definition}%
\nomenclature[f G]{$\GfunForest{G}(z)$}{forest function of a $3$-point graph, equation~\eqref{eq:def-GfunForest}\nomrefpage}%
In the following we will always assume that the indices $\EP_e$ are such that the integral \eqref{eq:def-GfunForest} converges absolutely, so $\GfunForest{G}(z)$ is analytic in $z$. In fact it can be extended at least to the domain $\setexp{z \in \C^3}{\Realteil(z_i) > 0 \ \text{for $1 \leq i \leq 3$}}$. If we rescale the argument $z$ and all Schwinger parameters $\SP_e$ in \eqref{eq:def-GfunForest}, power counting shows the homogeneity
\begin{equation}
\GfunForest{G}(\lambda z)
	= \lambda^{\sdd - 3} \cdot \GfunForest{G}(z)
	.
	\label{eq:GfunForest-scaling} %
\end{equation}
\begin{example}
	\label{ex:GfunForest-3edge}%
	The star of figure~\ref{fig:GfunForest-few} has $\psipol = 1$ and $f_i = \SP_i$, for the triangle the polynomials are $\psipol = \SP_1 + \SP_2 + \SP_3$ and $f_i = \SP_1 \SP_2 \SP_3 / \SP_i$. The forest integrals evaluate to
	\begin{equation}
		\GfunForest{\StarSymbol}(z)
		= z_1^{\EP_1 -1} z_2^{\EP_2-1} z_3^{\EP_3 -1}
		\qquad\text{and}\qquad
		\GfunForest{\TriangleSymbol}(z)
		= \frac{(z_1 z_2 z_3)^{\dimension/2-1}}{\SunrisePsi^{\dimension}}
			\prod_{k=1}^3 \left( \frac{\SunrisePsi}{z_k} \right)^{\EP_k}
		,
		\label{eq:GfunForest-3edge}%
	\end{equation}
	where we write $\SunrisePsi$ for the polynomial
	\begin{equation}
		\SunrisePsi
		= \SunrisePsi(z)
		\defas
		z_1 z_2 + z_2 z_3 + z_3 z_1
		.
		\label{eq:SunrisePsi}%
	\end{equation}%
\end{example}%
\nomenclature[psi sunrise]{$\SunrisePsi$}{sunrise graph polynomial, equation~\eqref{eq:SunrisePsi}\nomrefpage}%

\subsection{Recursions}\label{sec:3pt-recursions}
In this section we derive formulas to compute forest integrals recursively, adding edges one by one. Examples will be given for the graphs in figure~\ref{fig:GfunForest-few}. Appending a vertex is simple:
\begin{figure}
	\centering
	$
		K_{1,3} = \hspace{-7mm}\Graph[0.65]{star_forest}
		\quad
		C_1= \hspace{-4mm} \Graph[0.67]{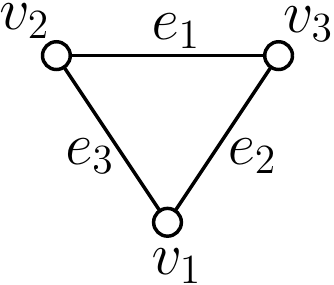} \hspace{-2mm}
		\rightarrow
		C_1'=\hspace{-5mm}\Graph[0.5]{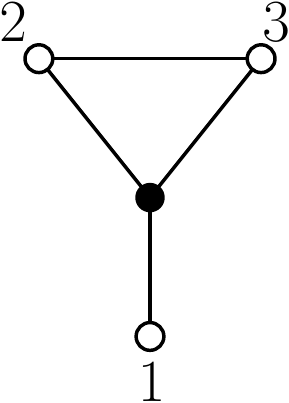}\hspace{-2mm}
		\rightarrow
		\WS{3}^{-}=\hspace{-7mm}\Graph[0.65]{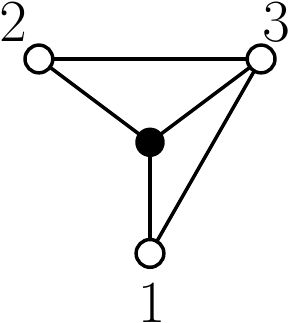}\hspace{-2mm}
		\rightarrow
		\WS{3}=\hspace{-5mm}\Graph[0.65]{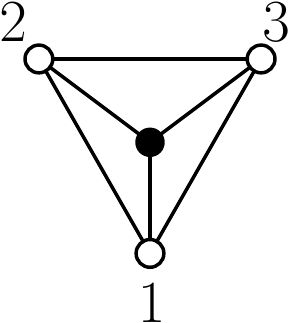}
	$ 
	\caption[Three-point graphs with few edges]{%
		Three-point graphs with few edges. Forest functions for the star $K_{1,3}$ and the triangle $C_1$ are computed in example~\ref{ex:GfunForest-3edge}. We can then add edges one by one to construct the wheel $\WS{3}$ with $3$ spokes.}%
	\label{fig:GfunForest-few}%
\end{figure}
\begin{lemma}
	\label{lemma:Gfun-Forest-vertex}%
	Given a graph $G$ with three external vertices $\set{v_1, v_2, v_3} = \vertices_{\Text}(G)$, append an edge $e = \set{v_1,v_1'}$ to a new vertex $v_1'$ and call the resulting graph $G'$.
	This means $\vertices(G') = \vertices(G) \cupdot \set{v_1'}$, $\edges(G') = \edges(G) \cupdot \set{e}$ and we set $\vertices_{\Text}(G') \defas \set{v_1',v_2,v_3}$. Furthermore an index $\EP_e$ is assigned to the new edge.
	
	Then we can compute the forest integral $\GfunForest{G'}$ from $\GfunForest{G}$ following
	\begin{equation}
		\GfunForest{G'}(z)
		=
			\int_0^{z_1}
			\GfunForest{G}\left( z_1 - \SP_e, z_2, z_3 \right)
			\, \SP_e^{\EP_e - 1} 
			\ \dd \SP_e
		.
		\label{eq:Gfun-Forest-vertex}%
	\end{equation}
\end{lemma}
\begin{proof}
	Any spanning forest $F$ of $G'$ that does not contain $e$ has $\set{v_1'} \in \comps(F)$ as one connected component. Therefore only
	$ f_1' = \SP_e \psipol + f_1$
	depends on $\SP_e$ while
	$ \psipol' = \psipol$,
	$ f_2' = f_2$ and
	$ f_3' = f_3$
	are constant. We set $y \defas f / \psipol$ so our claim is the consequence of
	\begin{equation*}
		\GfunForest{G'}(z)
		=
			\int_0^{\infty} \SP_e^{\EP_e - 1} \dd \SP_e
			\int_0^{\infty} 
			\GfunForest{G}(y) \cdot
			\delta(y_1 + \SP_e - z_1)
			\delta(y_2 - z_2)
			\delta(y_3 - z_3)
			\ \dd[3] y
		. \qedhere
	\end{equation*}
\end{proof}
\begin{example}
	\label{ex:GfunForest-append-vertex}%
	We start with the triangle in $\dimension = 4$ dimensions with unit indices, that is
	$ \GfunForest{C_1}(z)
		= 1/\SunrisePsi
	$
	from \eqref{eq:GfunForest-3edge}, and append a vertex at $v_1$ as in figure~\ref{fig:GfunForest-few}:
	\begin{equation}
		\GfunForest{C_1'}(z)
		= \int_0^{z_1} \frac{\dd \SP}{(z_1-\SP)(z_2+z_3) + z_2 z_3}
		= \frac{1}{z_2+ z_3} \ln \frac{\SunrisePsi}{z_2 z_3}
		.
		\label{eq:GfunForest-C1'}%
	\end{equation}
\end{example}
\begin{lemma}
	\label{lemma:3pt-forest-identity}%
	With notations as in definition~\ref{def:GfunForest} and $f_{123} \defas \forestpolynom[G]{\set{1},\set{2},\set{3}}$, we have the identity
	\begin{equation}
		\psipol \cdot f_{123} 
		= f_1 f_2 + f_1 f_3 + f_2 f_3
		= \SunrisePsi(f)
		.
		\label{eq:3pt-forest-identity} %
	\end{equation}
\end{lemma}
This result is a reformulation of the Dodgson identity \eqref{eq:dodgson-identity} and we omit the proof since it is a parallel (but simpler) argument to our lemma~\ref{lemma:forest-identity-ladderbox} below and refer to \cite{Brown:PeriodsFeynmanIntegrals,BrownYeats:SpanningForestPolynomials}.
\begin{lemma}
	\label{lemma:Gfun-Forest-edge} %
	Let $G$ be a graph with three external vertices $\vertices_{\Text}(G) = \set{v_1, v_2, v_3}$ and add an edge $e = \set{v_2,v_3}$ to construct $G' \defas G \cup \set{e}$. Then
	\begin{equation}
		\GfunForest{G'} (z)
		=
		\SunrisePsi^{\EP_e + \sdd-\dimension}
		\int_{0}^{z_1}
			\frac{
				\GfunForest{G}(z_1 - x, z_2, z_3)
				\,
				x^{\dimension/2 - \EP_e - 1}
			}{
				\SunrisePsi(z_1 - x, z_2, z_3)^{\sdd - \dimension/2}
			}
			\ 
			\dd x
		.
		\label{eq:Gfun-Forest-edge} %
	\end{equation}
\end{lemma}
\begin{proof}
	Looking at the spanning forests of $G'$ and whether they contain $e$ or not, we obtain that
	$\psipol' = \SP_e \psipol + f_2 + f_3$,
	$f_1' = \SP_e f_1 + f_{123}$,
	$f_2' = \SP_e f_2$ and
	$f_3' = \SP_e f_3$
	such that
	\begin{align*}
		\GfunForest{G'}(z)
		&=
			\int_0^{\infty} \SP_e^{\EP_e - 1} \dd \SP_e
			\int_0^{\infty} \dd[3] y\ 
			\GfunForest{G}(y)
			\cdot
			(\SP_e + y_2 + y_3)^{-\dimension/2}
		\\
		& \quad \times
			\delta\left( \frac{\SP_e y_1 + \SunrisePsi(y)}{\SP_e + y_2 + y_3} - z_1 \right)
			\delta\left( \frac{\SP_e y_2}{\SP_e + y_2 + y_3} - z_2 \right)
			\delta\left( \frac{\SP_e y_3}{\SP_e + y_2 + y_3} - z_3 \right)
		.
	\end{align*}
	Note that we used \eqref{eq:3pt-forest-identity} to rewrite $f_{123}$. The solution of the $\delta$-constraints is given by
	\begin{equation*}
		y_1 = z_1 - \frac{z_2 z_3}{\SP_e - z_2 - z_3}
		,\quad
		y_2 = \frac{\SP_e z_2}{\SP_e - z_2 - z_3}
		\quad\text{and}\quad
		y_3 = \frac{\SP_e z_3}{\SP_e - z_2 - z_3}
		.
	\end{equation*}
	So $\SP_e + y_2 + y_3 = \frac{\SP_e^2}{\SP_e - z_2 - z_3}$ and with the measure
	$\SP_e^3\ \dd[3] z = \left( \SP_e - z_2 - z_3 \right)^3 \dd[3] y$
	we find
	\begin{equation*}
		\GfunForest{G'}(z)
		=
		\int_{\SunrisePsi/z_1}^{\infty}
		\GfunForest{G}\left(
			\frac{\SP_e z_1 - \SunrisePsi(z)}{\SP_e - z_2 - z_3},
			\frac{\SP_e z_2}{\SP_e - z_2 - z_3},
			\frac{\SP_e z_3}{\SP_e - z_2 - z_3}
		\right)
		\frac{(\SP_e - z_2 - z_3)^{\dimension/2-3}}{\SP_e^{\dimension-\EP_e - 2}}
		\ \dd \SP_e
		.
	\end{equation*}
	Finally we exploit the homogeneity \eqref{eq:GfunForest-scaling} to pull out $\SP_e/(\SP_e - z_2 - z_3)$ from the arguments of $\GfunForest{G}$ and substitute $\SP_e = \SunrisePsi/x$.
\end{proof}
\begin{example}
	\label{ex:GfunForest-add-ege}%
	We add edges opposite to $v_2$ and $v_3$ to $C_1'$ as shown in figure~\ref{fig:GfunForest-few}. In $\dimension = 4$ dimensions and with unit indices $\EP_e = 1$ we find, starting from \eqref{eq:GfunForest-C1'}, that
	\begin{align}
		\GfunForest{\WS{3}^{-}}(z)
		&= \frac{1}{\SunrisePsi} \int_0^{z_2}
				\frac{1}{z_2-x+z_3}
				\ln \frac{(z_2-x)(z_1+z_3) + z_1 z_3}{(z_2-x) z_3}
			\ \dd x
		\nonumber\\
		&= -\frac{1}{\SunrisePsi} \left[ 
					\Li_2\left(1-\frac{\SunrisePsi}{z_1z_3} \right) 
					+\Li_2\left(1-\frac{\SunrisePsi}{z_2z_3} \right) 
					+\frac{1}{2} \ln^2\frac{z_1}{z_2}
					+\mzv{2}
				\right]
		\quad\text{and}
		\label{eq:GfunForest-WS3-}%
		\\
		\GfunForest{\WS{3}}(z)
		&= \frac{1}{\SunrisePsi^2} \int_0^{z_3} \left[ \SunrisePsi \cdot \GfunForest{\WS{3}^{-}} \right]_{z_3=z_3-x} \ \dd x
		\nonumber\\
		&= -\frac{1}{\SunrisePsi^2} 
				\sum_{i<j} \left[
					(z_i+z_j)	\Li_2\left( 1-\frac{\SunrisePsi}{z_i z_j} \right)
					+ \frac{z_k}{2} \ln^2 \frac{z_i}{z_j}
					+ z_k \mzv{2}
				\right]_{\mathrlap{\set{i,j,k} = \set{1,2,3}}}.
		\label{eq:GfunForest-WS3}%
	\end{align}
\end{example}

\subsection{Stars and triangles}
\label{sec:stars-triangles}%
We like to briefly comment on another approach to iteratively construct Feynman integrals parametrically. Even though it is very closely related to the forest integrals above, this different viewpoint is conceptually simpler and may be helpful in future, more general applications.

The idea is not to refer to the spanning forest polynomials $f$ explicitly, but rather consider partial integrals with three free (not integrated) Schwinger variables.
\begin{figure}
	\centering
	$G_{\StarSymbol} \defas \Graph[0.7]{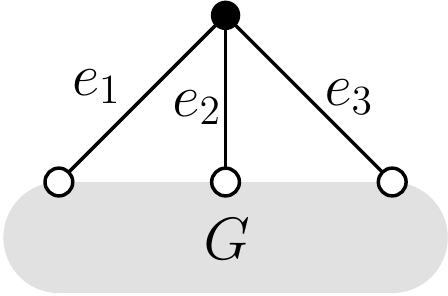} \qquad
	 G_{\TriangleSymbol} \defas \Graph[0.7]{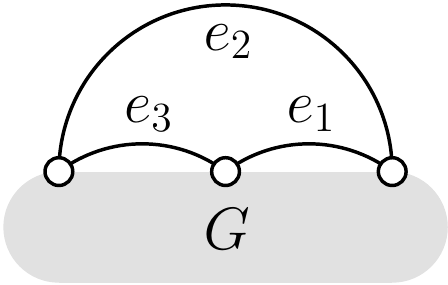}$%
	\hfill
	$\Graph[0.7]{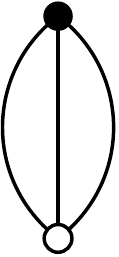} \qquad \Graph[0.6]{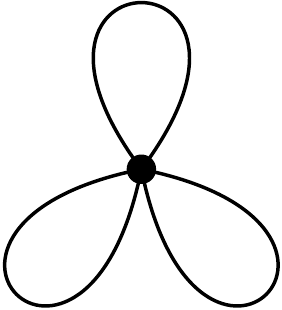}$%
	\caption[Definition of the star- and triangle graphs]{Definition of the star- and triangle graphs (the external vertices $v_1$, $v_2$ and $v_3$ of $G$ are the white circles, from left to right) and examples for the case when $G$ consists of just one vertex $v_1$ = $v_2$ = $v_3$.}%
	\label{fig:startriangle-added}%
\end{figure}%
\begin{definition}
	\label{def:star-triangle-functions}%
	Let $G$ be a graph with three marked vertices $\set{v_1,v_2, v_3} = \vertices_{\Text}(G)$ (not necessarily distinct). By $G_{\StarSymbol}$ and $G_{\TriangleSymbol}$ we denote the graphs obtained after adding a star or a triangle to $G$, as shown in figure~\ref{fig:startriangle-added}. In particular these have three additional edges $e_i$.
	The \emph{star-} and \emph{triangle functions} 
	$\GfunStar{G}, \GfunTriangle{G}\colon \R_+^3 \longrightarrow \R_+$
	are the partial integrals
	\begin{align}
		\GfunStar{G}(z)
		\defas
		\int_{0}^{\infty} \hspace{-1mm}
			\psipol_{G_{\StarSymbol}}^{-\dimension/2}
			\hspace{-1mm}\prod_{e \in \edges(G)}\hspace{-1mm} \SP_e^{\EP_e -1} \dd \SP_e
		\quad\text{and}\quad
		\GfunTriangle{G}(z)
		\defas
		\int_{0}^{\infty}\hspace{-1mm}
			\psipol_{G_{\TriangleSymbol}}^{-\dimension/2}
			\hspace{-1mm}\prod_{e \in \edges(G)}\hspace{-1mm} \SP_e^{\EP_e -1} \dd \SP_e
		\label{eq:def-star-triangle-function} %
		,
	\end{align}
	given that these converge.
	They depend on the three free Schwinger parameters $z_i = \SP_{e_i}$ associated to the additional edges, the dimension $\dimension$ and the indices $\EP_e$.
\end{definition}%
\nomenclature[fG,fG]{$\GfunStar{G},\GfunTriangle{G}$}{star and triangle function of $G$, equation~\eqref{eq:def-star-triangle-function}\nomrefpage}%
These functions are analytic and their homogeneity is determined by power counting
$\deg(\psipol_G) = \loops{G} = \loops{G_{\StarSymbol}} -2 = \loops{G_{\TriangleSymbol}} - 3$.
In terms of the superficial degree of divergence $\sdd$ of $G$ from \eqref{eq:def-sdd} we conclude that
\begin{equation}
	\GfunStar{G}(\lambda z)
	= \lambda^{\sdd - \dimension} \cdot  \GfunStar{G}(z)
	\quad\text{and}\quad
	\GfunTriangle{G}(\lambda z)
	= \lambda^{\sdd - 3\dimension/2} \cdot \GfunTriangle{G}(z)
	.
	\label{eq:Gfun-StarTriangle-scaling} %
\end{equation}
\begin{example}
	\label{ex:StarTriangle-3edge}%
	If $G$ is just the single vertex, no integration has to be performed. The graph polynomials are $\SunrisePsi$ for $G_{\StarSymbol}$ and $z_1 z_2 z_3$ for $G_{\TriangleSymbol}$ (see figure~\ref{fig:startriangle-added}, so
	\begin{equation}
		\GfunStar{G} = \SunrisePsi^{-\dimension/2}
		\quad\text{and}\quad
		\GfunTriangle{G} = (z_1 z_2 z_3)^{-\dimension/2}.
		\label{eq:Gfun-StarTriangle-K1}%
	\end{equation}
\end{example}
Stars and triangles are related by a well-known change of variables.
\begin{lemma}[Star-Triangle duality]
	\label{lemma:star-triangle-duality}%
	Let $G$ be as in definition~\ref{def:star-triangle-functions}.
	Then we have
	\begin{align}
		\psipol_{G_{\StarSymbol}}
		&=
		(z_1 z_2 + z_2 z_3 + z_3 z_1)\psipol_{G}
		+ \sum_{i \neq j} z_i f_j
		+ f_{123}
		\quad\text{and}
		\label{eq:psipol-star} %
		\\
		\psipol_{G_{\TriangleSymbol}}
		&=
		z_1 z_2 z_3\psipol_{G}
		+ \sum_{i \neq j} f_i z_i z_j
		+ (z_1 + z_2 + z_3) f_{123}
		,
		\label{eq:psipol-triangle} %
	\end{align}
	where the sums run over $i,j \in\set{1,2,3}$ and $f_{\cdot}$ denotes the forest polynomials of $G$.
	It follows that under the change of variables 
	\begin{equation}
		z_i
		= \frac{1}{y_i} \frac{y_1 y_2 y_3}{y_1 + y_2 + y_3}
		\qquad\text{with inverse}\qquad
		y_i 
		= \frac{z_1 z_2 + z_1 z_3 + z_2 z_3}{z_i},
		\label{eq:star-triangle-variable-change} %
	\end{equation}
\hide{
	which yields the transformations
	\begin{equation}
		x_1 x_2 + x_1 x_3 + x_2 x_3
		= \frac{y_1 y_2 y_3}{y_1 + y_2 + y_3}
		\quad
		x_1 x_2 x_3
		= \frac{(y_1 y_2 y_3)^2}{(y_1 + y_2 + y_3)^3}
		\quad
		y_1 y_2 y_3
		= \frac{(x_1 x_2 + x_1 x_3 + x_2 x_3)^3}{x_1 x_2 x_3}
		\quad
		y_1 + y_2 + y_3
		= \frac{(x_1 x_2 + x_1 x_3 + x_2 x_3)^2}{x_1 x_2 x_3}
		,
		\label{eq:star-triangle-polynomial-transformations} %
	\end{equation}
}%
	we obtain the identities
	\begin{equation}
		\psipol_{G_{\StarSymbol}}
		= \frac{1}{y_1 + y_2 + y_3}
			\restrict{\psipol_{G_{\TriangleSymbol}}}{z_i \mapsto y_i}
		\quad\text{and}\quad
		\GfunStar{G}(z)
		=
		(y_1 + y_2 + y_3)^{\dimension/2}
		\GfunTriangle{G}(y)
		.
		\label{eq:star-triangle-psipol-transformation} %
	\end{equation}
\end{lemma}
\begin{proof}
	This result amounts to a simple classification of spanning trees $T$ of $G_{\StarSymbol}$:
	Since $T$ must contain at least one of the edges $e_i$ in order to connect $v$ (the centre of the star) to $\vertices(G)$, the set $S \defas T \cap \set{e_1,e_2,e_3}$ is non-empty. Now distinguish
	\begin{itemize}
		\item[$\abs{S} = 3$:]
			$T \setminus S$ is a three-forest of $G$ with each $v_i$ in a separate component. All trees with $\abs{S} = 3$ therefore add up to $f_{123} = \forestpolynom[G]{\set{v_1},\set{v_2},\set{v_3}}$.

		\item[$\abs{S} = 2$:]
			Say $e_i \notin S = \set{e_j, e_k}$, then $T \setminus S$ is a two-forest with $v_j$ and $v_k$ in different components. All these add up to
			$
				z_i \forestpolynom[G]{\set{j}, \set{k}}
				= z_i (f_j + f_k)
			$.

		\item[$ \abs{S} = 1$:]
			Here $T \setminus S$ is a spanning tree of $G$ itself. If we fix $S = \set{e_i}$, then such trees contribute $\psipol_G \prod_{j \neq i} z_j$.
	\end{itemize}
	We omit the analogous argument for $G_{\TriangleSymbol}$, see also \cite[examples 32 and 33]{Brown:PeriodsFeynmanIntegrals}.
\hide{
	Adding up all these contributions gives \eqref{eq:psipol-star}. The same way of reasoning can be applied to the spanning forests $T$ of $G_{\TriangleSymbol}$:
	\begin{itemize}
		\item[$S = \emptyset$:]
			$T$ is a spanning tree of $G$ itself. These add up to $z_1 z_2 z_3 \psipol_G$.

		\item[$S = \set{e_i}$:]
			$T \setminus S$ is a two-forest of $G$ with the two vertices $e_i = \set{v_j, v_k}$ in different components. Their contributions sum to $z_j z_k \forestpolynom[G]{\set{j},\set{k}} = z_j z_k (f_j + f_k)$.

		\item[$\abs{S}=2$:]
			$T \setminus S$ is a three-forest of $G$ with each $v_i$ in a different component.
	\end{itemize}
}%
\end{proof}
Using definition~\ref{def:GfunForest} we can express the star- and triangle functions as integrals of the forest function: From \eqref{eq:psipol-star} and \eqref{eq:psipol-triangle} we read off
\begin{corollary}
	Let $G$ be as in definition~\ref{def:star-triangle-functions}, then we have the identities
	\begin{align}
		\GfunStar{G}(z)
		&=
			\int_0^{\infty}
				\GfunForest{G}(x)
				\cdot
				\Bigg[ 
					\sum_{i < j}
					(x_i+z_i)(x_j+z_j)
				\Bigg]^{-\dimension/2}
			\!\!\dd[3] x
		,
		\label{eq:star-function-from-forest} %
		\\
		\GfunTriangle{G}(z)
		&=
				\int_0^{\infty}
			\GfunForest{G}(x)
			\cdot
			\Bigg[ 
				z_1 z_2 z_3
				+ \sum_{i \neq j} z_i z_j x_j
				+ (z_1 + z_2 + z_3) \sum_{i<j} x_i x_j
			\Bigg]^{-\dimension/2}
				\!\!\dd[3] x
		.
		\label{eq:triangle-function-from-forest} %
	\end{align}
\end{corollary}
It is obvious that we can write down recursion formulas for the partial integrals in the same way as we did for the forest function. The calculations are very similar and straightforward, so we omit the details of the proof and only state the results of our
\begin{lemma}
	\label{lemma:Gfun-star-triangle-recursion} %
	Let $G'$ be obtained from $G$ by adding an edge $e = \set{v_2, v_3}$. Then
	\begin{align}
		\GfunStar{G'}(z)
		&= \SunrisePsi^{\EP_e + \sdd - \dimension}
			\int_0^{\infty} 
			\frac{
				\GfunStar{G}(x+z_1, z_2 ,z_3)
				\,
				x^{\dimension/2 - \EP_e -1}
			}{
				\left[ \SunrisePsi(x + z_1, z_2, z_3) \right]^{\sdd-\dimension/2}
			}
			\ \dd x
		\quad\text{and}
		\label{eq:GfunStar-addedge}%
		\\
		\GfunTriangle{G'}(z)
		&= \int_0^{\infty} 
			\GfunTriangle{G}\left( \frac{x z_1}{x + z_1}, z_2, z_3 \right)
			\frac{x^{\EP_e -1}\ \dd x}{	(z_1 + x)^{\dimension/2} }
		.
		\label{eq:GfunTriangle-addedge}%
	\end{align}
	When $G'$ denotes $G$ after appending a new vertex $v_1'$ with an edge $e = \set{v_1^{}, v_1'}$, then
	\begin{align}
		\GfunStar{G'}(z)
		&= \int_{0}^{\infty}
				\GfunStar{G}(x+z_1, z_2, z_3)
				\,
				x^{\EP_e - 1}
				\ \dd x
		\quad\text{and}
		\label{eq:GfunStar-appendvertex}%
		\\
		\GfunTriangle{G'}(z)
		&= \left( \frac{z_2 z_3}{z_1+z_2+z_3} \right)^{\EP_e}
			\!\!
			\int_0^{\infty}
			\!
			\GfunTriangle{G}\left(\frac{z_1}{1+x}, z_2, z_3  \right)
			\left[ 1+\frac{x(z_2+z_3)}{z_1+z_2+z_3} \right]^{\sdd - \dimension/2}
			\frac{x^{\EP_e - 1} \dd x}{(1+x)^{\dimension/2}}
		.
		\label{eq:GfunTriangle-appendvertex}%
	\end{align}
\end{lemma}
\subsection{Applications and kinematics}\label{sec:vw3-applications}
In the article \cite{DavydychevUssyukina:AnApproachLadderDiagrams}, Ussyukina and Davydychev developed a recursive approach to compute massless off-shell three-point functions with ladder topology as shown in figure~\ref{fig:three-point-ladders}. They evaluated the finite integrals
\begin{equation}
	\FR\left( C_L \right)
	=
		\left( \frac{1}{p_3^2} \right)^L
		\Phi^{(L)}\left( 
			\frac{p_1^2}{p_3^2},
			\frac{p_2^2}{p_3^2}
		\right)
	\quad\text{for $\dimension = 4$ and $\EP_e = 1$}
	\label{eq:ladder3pt-d4} %
\end{equation}
to arbitrary loop-order $L$ in terms of the polylogarithms \cite{DavydychevUssyukina:LadderDiagramsArbitraryRungs}
\begin{equation}
	\Phi^{(L)}(x,y)
	\defas
	\frac{1}{z-\bar{z}}
		\sum_{k=L}^{2L}
		\binom{k}{L}
		\frac{
			\ln^{2L - k} (x/y)
		}{
			(2L - k)!
		}
		\left[ 
			\Li_{k}\left( 1-\frac{1}{z} \right)
			-
			\Li_{k}\left( 1-\frac{1}{\bar{z}} \right)
		\right]
	\label{eq:Ussyukina-polylogarithm} %
\end{equation}
of the variables $z$, $\bar{z}$ that parametrize $ x = z \bar{z}$ and $y = (1-z)(1-\bar{z})$.
Knowing this simple result, it is natural to ask if it can be obtained by parametric integration. More importantly, the formula \eqref{eq:ladder3pt-d4} is only valid in $\dimension=4$ dimensions and for $p_1^2, p_2^2, p_3^2 > 0$. 

\begin{figure}
	\centering
	$
		C_1 = \Graph[0.4]{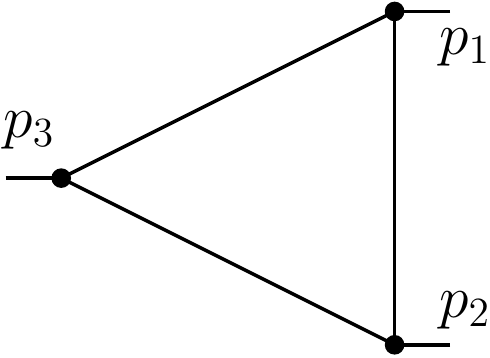} \quad
		C_2 = \Graph[0.4]{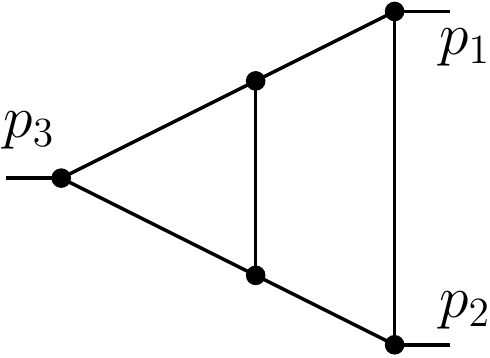} \quad
		C_3 = \Graph[0.4]{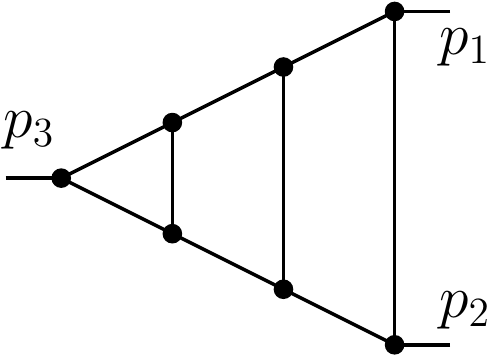} \quad
		C_4 = \Graph[0.4]{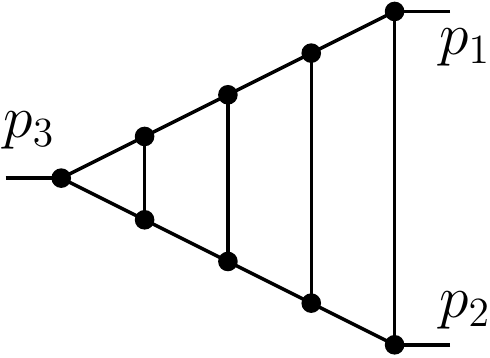} \quad
		\cdots
	$
	\caption{The triangle ladder series $C_n$ from Ussyukina and Davydychev.}%
	\label{fig:three-point-ladders}%
\end{figure}%
But since massless external particles occur in the standard model, one also needs these integrals with some momenta $p_i^2 = 0$ on the light cone. This usually introduces infrared divergences and makes dimensional regularization necessary. Furthermore, one also needs tensor integrals of these graphs and higher orders in their $\varepsilon$-expansion.

Even though the kinematics become simpler when some $p_i^2 = 0$ vanish, these divergences render the computation more difficult with standard approaches. No all-loop result like \eqref{eq:ladder3pt-d4} exists in the literature. Even the two-loop case was computed only recently \cite{ChavezDuhr:Triangles}. Interestingly, this result was obtained with parametric integration using hyperlogarithms. To our knowledge it was the first time that this method was systematically applied to compute higher order contributions to $\varepsilon$-expansions in practice.

Furthermore, we found that linear reducibility also holds for all three-loop massless three-point graphs \cite{Panzer:DivergencesManyScales}. This study was a result of applying the polynomial reduction algorithm to all these graphs individually.

We also found counterexamples to linear reducibility at four loops, but it had become clear that a surprisingly huge number of massless three-point functions is linearly reducible. This property can be reduced to the forest functions with
\begin{lemma}
	\label{lemma:vw3-momenta} %
	Let $G$ have three external vertices $v_1,v_2,v_3 \in \vertices(G)$, massless internal propagators $m_e = 0$ and degree of divergence $\sdd \neq 0$. If we parametrize the momenta $p_i$ entering $G$ at $v_i$ by $p_1^2 = z \bar{z} p_3^2$ and $p_2^2 = (1-z)(1-\bar{z}) p_3^2$, then
	\begin{equation}
		\FR(G)
		=	p_3^{-2\sdd}
			\frac{\Gamma(\sdd)}{\prod_e \Gamma(\EP_e)}
			\int_0^{\infty}
				\frac{
					\GfunForest{G}(x)
					\,\Omega
				}{
					\left[
						z\bar{z} x_1 + (1-z)(1-\bar{z}) x_2 + x_3
					\right]^{\sdd}
				}
		.
		\label{eq:vw3-3pt-projective}%
	\end{equation}
\end{lemma}
\begin{proof}
	The second Symanzik polynomial of $G$ in the given kinematics reads
	\begin{equation*}
		\phipol 
		= p_3^2 \cdot \left[ 
			z\bar{z} f_1
			+ (1-z)(1-\bar{z}) f_2
			+ f_3
		\right]
		.
	\end{equation*}
	If $\FR(G)$ is convergent, we can insert this into \eqref{eq:feynman-integral-parametric} to obtain
	\begin{align*}
		\FR(G)
		&=
			\int_0^{\infty}
				\prod_e \frac{\SP_e ^{\EP_e - 1} \dd \SP_e}{\Gamma(\EP_e)}
			\int_0^{\infty} 
			\frac{e^{-\phipol/\psipol}}{\psipol^{\dimension/2}}
			\delta^{(3)}\left( \frac{f}{\psipol} - x \right)
			\dd[3] x
		\\
		&=
			\frac{1}{\prod_e \Gamma(\EP_e)}
			\int_0^{\infty} 
			\GfunForest{G}(x)
			\exp \left\{
				- p_3^2 \big[ z\bar{z} x_1 + (1-z)(1-\bar{z}) x_2 + x_3 \big]
			\right\}
			\ \dd[3] x
	\end{align*}
	which is the Laplace transform of $\GfunForest{G}$. When we exploit the homogeneity \eqref{eq:GfunForest-scaling}, we can perform one integration and arrive at the projective integral representation \eqref{eq:vw3-3pt-projective}.
\end{proof}
\begin{example}
	In $\dimension=4$ with $\EP_e=1$, we have $\GfunForest{\TriangleSymbol} = 1/\SunrisePsi$ from \eqref{eq:GfunForest-3edge} and therefore
	\begin{equation*}
		\FR(C_1)
		= p_3^{-2} \int \frac{\Omega}{\left[ z\bar{z} x_1 + (1-z)(1-\bar{z}) x_2 + x_3 \right]\SunrisePsi(x)}
		= \frac{4\imag \BlochWigner(z)}{p_3^2 (z-\bar{z})}
	\end{equation*}
	with the Bloch-Wigner dilogarithm from \eqref{eq:BlochWigner}. So the first Ussyukina-Davydychev function of \eqref{eq:Ussyukina-polylogarithm} is $4\imag\BlochWigner(z)/(z-\bar{z}) = \Phi^{(1)}(z\bar{z},(1-z)(1-\bar{z}))$.
\end{example}
Together with the recursion formulas of section~\ref{sec:3pt-recursions}, lemma~\ref{lemma:vw3-momenta} essentially proves our theorem~\ref{theorem:vw3-momentum} as we will show in section~\ref{sec:recursion-reducibility}.

\subsubsection{Vacuum periods}
Suppose $G$ is logarithmically divergent ($\sdd = 0$) and primitive (free of subdivergences). From \eqref{eq:primitive-residue-period} we can compute its period as the residue
\begin{equation}
	\period(G)
	= \Res_{\sdd \rightarrow 0} \FR(G)
	\urel{\eqref{eq:vw3-3pt-projective}}
		\frac{1}{\prod_e \Gamma(\EP_e)}
		\int_0^{\infty} \GfunForest{G} \ \Omega
	.
	\label{eq:period-from-GfunForest}%
\end{equation}
\begin{example}
	\label{ex:WS3-period}%
	We compute the period of the wheel $\WS{3}$ with $3$ spokes (figure~\ref{fig:GfunForest-few}) in $\dimension = 4$ dimensions, with unit indices $\EP_e = 1$:
	\begin{align}
		\period\left( \WS{3} \right)
		&\urel{\eqref{eq:period-from-GfunForest}}
			\int_0^{\infty} \dd z_1 
			\int_0^{\infty}  \GfunForest{\WS{3}}(z_1,z_2,1) \ \dd z_2
		\nonumber\\
		&\urel{\eqref{eq:GfunForest-WS3}}
			\int_0^{\infty} \frac{
				3 \mzv{3} - 3\Li_3(-z_1) +	\big[ \Li_2(-z_1) - \mzv{2}\big] \ln z_1
			}{(1+z_1)^2} \ \dd z_1
		= 6 \mzv{3}
		.
		\label{eq:WS3-period}%
	\end{align}
\end{example}

\subsubsection{Graphical functions}
To compute a graphical function, we can combine \eqref{eq:phipol-graphical-function} with \eqref{eq:position-space-projective-dual} to find
\begin{equation}
	\gf[G](z,\bar{z})
	= \frac{\Gamma(\widehat{\sdd})}{\prod_e \Gamma(\EP_e)}
		\int_0^{\infty} 
		\frac{
				\Omega\ 
				\GfunForest{G'}(x)
		}{
				\SunrisePsi^{ \dimension/2- \widehat{\sdd}}
				\left[x_{v_z} + z\bar{z} x_{v_1} + (1-z)(1-\bar{z}) x_{v_0}\right]^{\widehat{\sdd}}
		}
	.
	\label{eq:graphical-from-GfunForest}%
\end{equation}
In this formula we replaced $\forestpolynom{P} = f_{123}$ with $\SunrisePsi(f)/\psipol$ using \eqref{eq:3pt-forest-identity}, and by $G'$ we indicate that the indices $\EP_e$ in $G$ need to be replaced with $\widehat{\EP}_e$. 
\begin{remark}
	It seems natural to express $\gf[G]$ in terms of a variation of the forest function \eqref{eq:def-GfunForest} given by
\begin{equation}
	\widehat{\GfunForest{G}}(z)
	\defas \int_0^{\infty}
		f_{123}^{-\dimension/2}
		\delta^{(3)} \left( 
			\frac{f}{f_{123}} - z
		\right)
		\prod_{e \in \edges} \SP_e^{\widehat{\EP}_e-1} \dd \SP_e,
	\label{eq:GfunForestDual}%
\end{equation}
because then \eqref{eq:position-space-projective-dual} directly translates to the Laplace transform
\begin{align}
	\gf[G](z,\bar{z})
	&= \frac{1}{\prod_e \Gamma(\EP_e)}
		\int_0^{\infty} 
		\widehat{\GfunForest{G}}(x)
		\exp\left[
			-x_{v_z} - z\bar{z} x_{v_1} - (1-z)(1-\bar{z}) x_{v_0}
		\right]
		\dd[3] x
	\label{eq:graphical-GfunForestDual-parametric}%
	\\
	&= \frac{\Gamma(\widehat{\sdd})}{\prod_e \Gamma(\EP_e)}
		\int_0^{\infty} 
		\frac{
			\Omega\ 
			\widehat{\GfunForest{G}}(x)
		}
		{
			\left[x_{v_z} + z\bar{z} x_{v_1} + (1-z)(1-\bar{z}) x_{v_0}\right]^{\widehat{\sdd}}
		}
	.
	\label{eq:graphical-GfunForestDual-projective}%
\end{align}
But referring to \eqref{eq:3pt-forest-identity} again, we realize that this is identical to \eqref{eq:graphical-from-GfunForest}, since $\widehat{\GfunForest{G}}(z) = \GfunForest{G'}(z) \cdot \big[ \SunrisePsi(z) \big]^{\widehat{\sdd}-\dimension/2}$.
\end{remark}

\section{Ladder boxes}\label{sec:ladderboxes}
Among the myriad of multi-loop Feynman integral calculations, the early results of Ussyukina and Davydychev are unique for the reason that they evaluated a series of three- and four-point functions, shown in figures~\ref{fig:three-point-ladders} and \ref{fig:box-ladders}, to arbitrary loop order \cite{DavydychevUssyukina:LadderDiagramsArbitraryRungs}. In particular they computed the box ladders $B_n$ (for unit indices $\EP_e = 1$),
\begin{equation}
	\FR\left( B_L \right)
	= \frac{1}{t s^L}
		\Phi^{(L)}\left(
			\frac{p_1^2 p_3^2}{s t},
			\frac{p_2^2 p_4^2}{s t}
		\right)
	\ \text{where}\ 
	s \defas (p_1 + p_2)^2,
	\ 
	t \defas (p_1 + p_4)^2
	\label{eq:ladderboxes-d4-offshell} %
\end{equation}
in terms of the polylogarithms $\Phi^{(L)}$ of \eqref{eq:Ussyukina-polylogarithm}. However, this very simple result only holds fully off-shell ($p_1^2,p_2^2,p_3^2,p_4^2 >0$) such that the integrals are finite, in exactly $\dimension = 4$ dimensions. As explained in \cite{Broadhurst:SummationLadderDiagrams}, a conformal symmetry is the reason why the complicated kinematics of an off-shell $4$-point function (which in general depends on $6$ independent scales: $p_1^2$, $p_2^2$, $p_3^2$, $p_4^2$, $s$ and $t$) in this special case is reduced essentially to only two dimensionless ratios. This argument that relates $\FR(B_L)$ to the three-point function $\FR(C_L)$ breaks down as soon as $\dimension \neq 4$ (or $\EP_e \neq 1$).\footnote{Generalizations of the symmetry exist only for very special relations among the indices $\EP_e$ and $\dimension$ \cite{Isaev:OperatorApproach}.}
\begin{figure}
	\centering
	\begin{tabular}{>{$}c<{$}>{$}c<{$}>{$}c<{$}>{$}c<{$}>{$}c<{$}}
	\Graph[0.35]{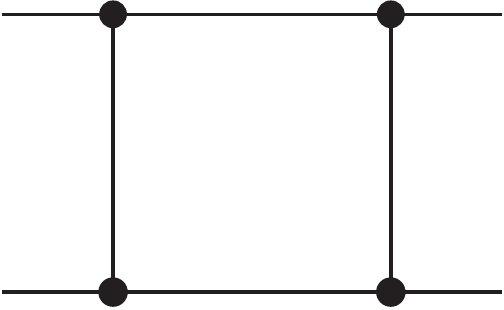} &
	\Graph[0.35]{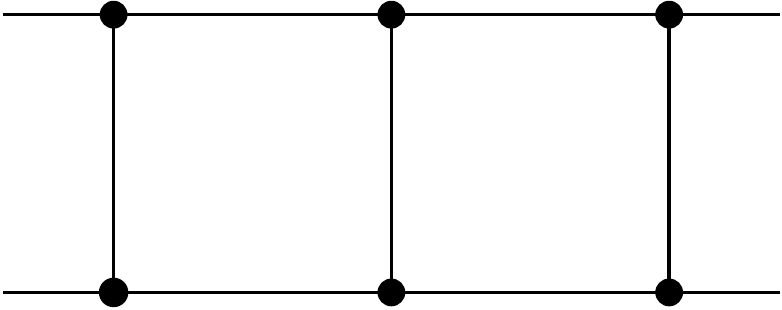} &
	\Graph[0.35]{boxes3loop} &
	\Graph[0.35]{boxes4loop} &
	\cdots\\[4mm]
	B_1 & B_2 & B_3 & B_4 & \\
\end{tabular}
	\caption{The series of box ladder graphs $B_n$.}%
	\label{fig:box-ladders}%
\end{figure}

Computations of scattering amplitudes involving massless external particles (photons, gluons or light quarks and leptons) however demand lightlike $p_i^2$, where the box integrals acquire infrared divergences. The evaluation of these on-shell ladder-boxes as a Laurent series in dimensional regularization turned out to be much more complicated.

For example, it took roughly ten years until the double-box was computed on-shell in \cite{Smirnov:DoubleBoxOnShell}, with the triple box following in \cite{Smirnov:OnShellTripleBox}. These computations, using Mellin-Barnes techniques, become excessively demanding with an increasing number of loops and at the time of writing, the author is not aware of an exact result in the four-loop case.

Interestingly, all known results for massless on-shell four-point functions, which (up to a prefactor) depend only on one dimensionless ratio $x= \frac{s}{t} = s/t$ of Mandelstam invariants, are harmonic polylogarithms (of $x$). It is therefore tempting to ask whether this holds for an infinite series of graphs and secondly, if these can be computed by parametric integration using hyperlogarithms.

Therefore we study box ladder graphs parametrically. Following our approach from section~\ref{sec:forest-functions} on three-point functions, we develop recursion formulas that allow for a simple inductive computation of the box integrals. Studying the polynomials that occur in this formulas will can show (see section~\ref{sec:recursion-reducibility} that all these integrals are linearly reducible and evaluate to a specific class of polylogarithms.

Some new results obtained with this approach are reviewed in section~\ref{sec:ex-ladderboxes}, where we also comment on some generalizations of our method.

\subsection{Forest functions}\label{sec:ladderbox-forestfunctions}
\begin{definition}
	\label{def:GfunForestBox} %
	Let $v_1, v_2, v_3, v_4 \in \vertices(G)$ denote four distinct vertices of a connected graph $G$. We introduce the vector $f = (f_{12},f_{14},f_3,f_4)$ to abbreviate the forest polynomials
	\begin{equation}
		f_{12}
			\defas \forestpolynom[G]{\set{1,2},\set{3,4}}
		,\quad
		f_{14}
			\defas \forestpolynom[G]{\set{1,4},\set{2,3}}
		,\quad
		f_{3}
			\defas \forestpolynom[G]{\set{3},\set{1,2,4}}
		\quad\text{and}\quad
		f_{4}
			\defas \forestpolynom[G]{\set{4},\set{1,2,3}}
		.
		\label{eq:def-forestpolynoms-ladderbox} %
	\end{equation}%
\nomenclature[f i]{$f_i$}{spanning forest polynomials of $3$- or $4$-point functions, see \eqref{eq:3pt-forest-shorthand} and \eqref{eq:def-forestpolynoms-ladderbox}}%
	Assuming that these are algebraically independent\footnote{This excludes trivial cases like constant forest polynomials for example.} from each other, they define a function
	$\GfunForestBox{G}\colon \R_{+}^4 \longrightarrow \R_+$ of a vector $z = (z_{12},z_{14}, z_{3}, z_{4})$ by
	\begin{equation}
		\GfunForestBox{G}\left( z \right)
		\defas
		\int_{\R_+^{\edges}}
		\psipol_G^{-\dimension/2}
		\cdot\delta^{(4)}\left( \frac{f}{\psipol} - z \right)
		\prod_{e \in \edges} \SP_e^{\EP_e - 1}\ \dd \SP_e
		.
		\label{eq:def-GfunForestBox} %
	\end{equation}
\end{definition}%
\nomenclature[f G]{$\GfunForestBox{G}(z)$}{forest function of a box-ladder graph, equation~\eqref{eq:def-GfunForestBox}\nomrefpage}%
Again we shall assume that the indices $\EP_e$ and $\dimension$ are such that \eqref{eq:def-GfunForestBox} converges absolutely, hence $\GfunForestBox{G}(z)$ is analytic in $z$. As four integrations are omitted, the homogeneity is given by
\begin{equation}
	\GfunForestBox{G}(\lambda z)
	=
	\lambda^{\sdd - 4} \GfunForestBox{G}(z)
	.
	\label{eq:GfunForestBox-scaling}%
\end{equation}
\begin{figure}
	\centering\setlength{\tabcolsep}{0mm}\begin{tabular}{crcccccc}
		$ \Graph[0.7]{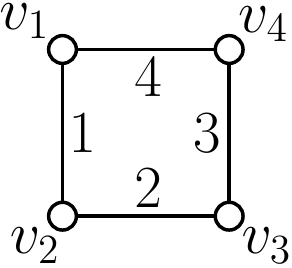}$ & $\mapsto$ & $ \Graph[0.7]{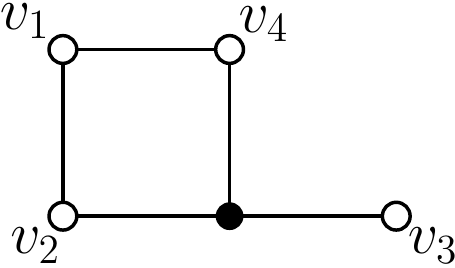}$ & $\mapsto$ & $ \Graph[0.7]{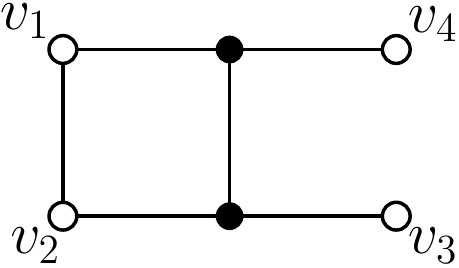}$ &  $\mapsto$ & $ \Graph[0.7]{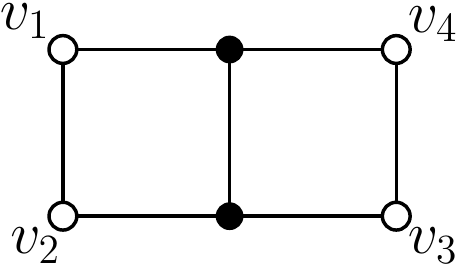}$ \\
	$B_1$ & & $B_1'$ & & $B_1''$ & & $B_2$\\
	\end{tabular}%
	\caption[Construction of the double box $B_2$ from the box $B_1$ with the moves of figure~\ref{fig:boxladder-addedge}]{Starting from the box $B_1$, the double box $B_2$ can be constructed by adding edges according to the moves of figure~\ref{fig:boxladder-addedge}. Internal vertices are shown in black.}%
	\label{fig:boxladders-forest}%
\end{figure}
\begin{example}
	\label{ex:GfunForestBox-1loop} %
	The polynomials of the box graph $B_1$ of figure~\ref{fig:boxladders-forest} read
	\begin{equation*}
		\psipol = \SP_1 + \SP_2 + \SP_3 + \SP_4,\quad
		f_{12} = \SP_2 \SP_4,\quad
		f_{14} = \SP_1 \SP_3,\quad
		f_3 = \SP_2 \SP_3 \quad
		\text{and}\quad
		f_4 = \SP_3 \SP_4
		.
	\end{equation*}
	The change of variables inverse to
	$
		z = f / \psipol
	$
	may be summarized as
	\begin{align}
		(\SP_1, \SP_2, \SP_3, \SP_4, \psipol, \dd[4] \SP)
		&= \left(
				\frac{z_{14} \BoxPoly}{z_3 z_4},
				\frac{\BoxPoly}{z_4},
				\frac{\BoxPoly}{z_{12}},
				\frac{\BoxPoly}{z_3},
				\frac{\BoxPoly^2}{z_{12} z_3 z_4},
				\frac{\BoxPoly^4\ \dd[4] z}{z_{12}^2 z_3^3 z_4^3}
		\right)
		\quad\text{in terms of}
		\nonumber\\
		\BoxPoly(z)
		&\defas
		z_{12}\left( z_{14} + z_3 + z_4 \right) + z_3 z_4
		.
		\label{eq:def-BoxPoly}%
	\end{align}%
\nomenclature[Q z]{$\BoxPoly(z)$}{a singularity of $4$-point forest functions, equation~\eqref{eq:def-BoxPoly}\nomrefpage}%
	Inserting this transformation into \eqref{eq:def-GfunForestBox} gives the general result
	\begin{equation}
		\GfunForestBox{B_1}(z)
		= 
			(z_3 z_4)^{\dimension/2-3}
			\left( \frac{z_{12}}{\BoxPoly^2} \right)^{\dimension/2-2}
			\left( \frac{z_{14} \BoxPoly}{z_3 z_4} \right)^{\EP_1 -1}
			\left( \frac{\BoxPoly}{z_4} \right)^{\EP_2 - 1}
			\left( \frac{\BoxPoly}{z_{12}} \right)^{\EP_3 - 1}
			\left( \frac{\BoxPoly}{z_3} \right)^{\EP_4 - 1}
		\label{eq:GfunForestBox-1loop} %
	\end{equation}
	and we note in particular that in the special case of unit indices $\EP_1 = \EP_2 = \EP_3 = \EP_4 = 1$,
	\begin{equation}
		\GfunForestBox{B_1}(z)
		=	\frac{1}{z_3 z_4} 
		\quad (\dimension=4)
		\qquad\text{and}\qquad
		\GfunForestBox{B_1}(z)
		= \frac{z_{12}}{\BoxPoly^2}
		\quad (\dimension=6).
		\label{eq:GfunForestBox-1loop-d4-d6} %
	\end{equation}
\end{example}
\begin{figure}
	\centering
	$	G = \!\!\!\Graph[0.6]{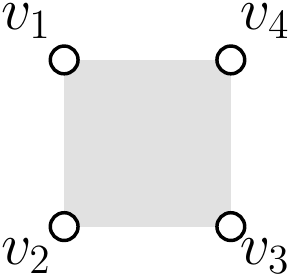} $
	\quad $\mapsto$ \quad
	$	\Graph[0.6]{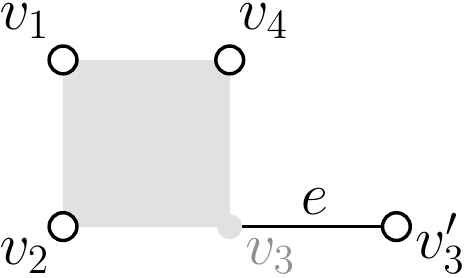} $
	\ or\quad
	$	\Graph[0.6]{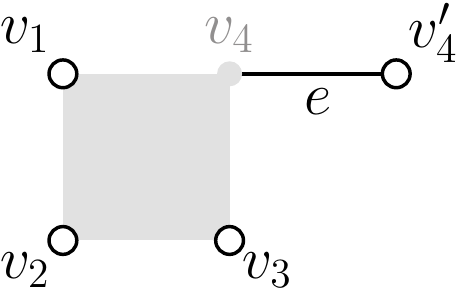} $
	\ or\quad
	$	\Graph[0.6]{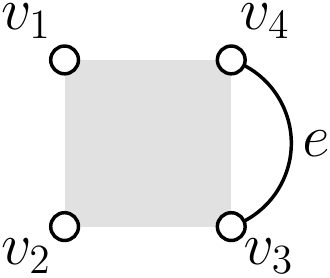} $
	\caption[Allowed ways to add edges to a four-point forest function]{We consider three different ways to add an edge $e$ to the graph $G$. Either we replace one of the external vertices $v_3$, $v_4$ with a new one that is attached via $e$, or we keep all external vertices when we add the edge $e = \set{v_3,v_4}$.} %
	\label{fig:boxladder-addedge} %
\end{figure}
We want to find recursion formulas for $\GfunForestBox{G \cup e}$ in terms of $\GfunForestBox{G}$ such that we can add edges $e$ to $G$ as shown in figure~\ref{fig:boxladder-addedge}. Replacing the external vertex $v_4$ (analogously for $v_3$) is simple:
\begin{lemma}
	\label{lemma:GfunForestBox-add-vertex} %
	Let $G'$ denote the graph obtained from $G$ by appending a new external vertex $v_4'$ through the edge $e = \set{v_4^{}, v_4'}$. Then (analogously for $e=\set{v_3^{},v_3'}$)
	\begin{equation}
		\GfunForestBox{G'}(z)
		= \int_0^{z_4} 
			\GfunForestBox{G}(z_{12}, z_{14}, z_3, z_4 - x)
			\cdot x^{\EP_e - 1}
			\ \dd x 
		.
		\label{eq:GfunForestBox-add-vertex} %
	\end{equation}
\end{lemma}
\begin{proof}
	The forest polynomials of $G'$ are identical to those of $G$ (also $\psipol_G = \psipol_{G'}$), except for
	$f_4' = x \psipol + f_4$ where $x = \SP_e$. Hence $f'/\psipol_{G'} = f/\psipol_{G} + (0,0,0,x)$ and therefore
	\begin{equation*}
		\GfunForestBox{G'}(z)
		= \int_0^{\infty} x^{\EP_e - 1} \dd x
			\int_{0}^{\infty}
				\GfunForestBox{G}(z')
				\delta(z_{12}' - z_{12})
				\delta(z_{14}' - z_{14})
				\delta(z_3' - z_3)
				\delta(x + z_4' - z_4)
				\ \dd[4] z'
		. \qedhere
	\end{equation*}
\end{proof}
\begin{example}
	\label{ex:ladderbox-forest-append}%
	In $\dimension=6$ dimensions, the forest functions for the graphs $B_{1}'$ and $B_1''$ shown in figure~\ref{fig:boxladders-forest} are simple to compute from \eqref{eq:GfunForestBox-1loop-d4-d6} using \eqref{eq:GfunForestBox-add-vertex}. We obtain
	\begin{align}
		\GfunForestBox{B_1'}(z)
		&= \int_{0}^{z_3}
				\GfunForestBox{B_1}(z_{12},z_{14},z_3',z_4) 
				\ \dd z_3'
		=	\frac{z_{3}}{(z_{14}+z_{4})\cdot\BoxPoly}
		\quad\text{and}
		\label{eq:ladderbox-forest-1'}%
		\\
		\GfunForestBox{B_1''}(z)
		&= \int_{0}^{z_4}
				\GfunForestBox{B_1'}(z_{12},z_{14},z_3,z_4') 
				\ \dd z_4'
		= \frac{1}{z_{12}-z_{14}} 
			\ln \frac{
				z_{12} (z_3 + z_{14})(z_4 + z_{14})
			}{
				z_{14} \cdot \BoxPoly
			}
		.
		\label{eq:ladderbox-forest-1''}%
	\end{align}
\end{example}
In order to add an edge $e = \set{v_3,v_4}$, we derive an identity of forest polynomials in
\begin{lemma}
	\label{lemma:forest-identity-ladderbox} %
	Let $v_1, v_2, v_3, v_4 \in \vertices(G)$ denote four distinct vertices of $G$ such that it is impossible to find two disjoint paths connecting $v_1$ with $v_3$ and $v_2$ with $v_4$.
	Then we have the following quadratic identity of forest polynomials of $G$:
	\begin{equation}
		\psipol_G^{} \cdot
		\forestpolynom[G]{\set{1,2},\set{3},\set{4}}
		=
		f_{12}
		\left(
			f_{14} + f_3 + f_4
		\right)
		+ f_3
			f_4
		=
		\BoxPoly(f).
		\label{eq:forest-identity-ladderbox} %
	\end{equation}
\end{lemma}
\begin{figure}%
	\begin{equation*}
		\Graph[0.4]{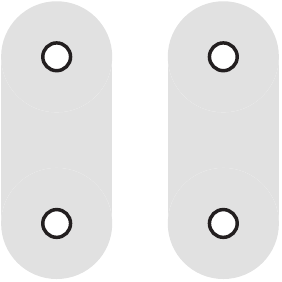}
		\ \cdot\ 
		\left(
			\Graph[0.4]{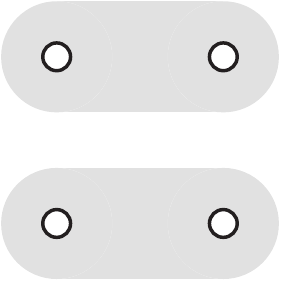}
			+\Graph[0.4]{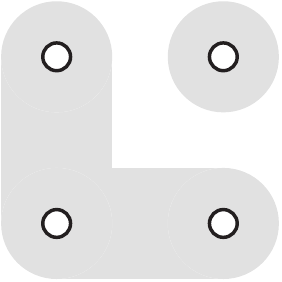}
			+\Graph[0.4]{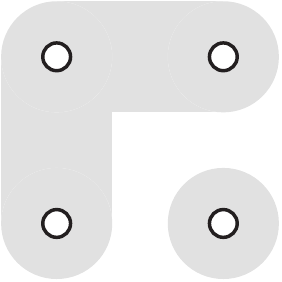}
		\right)
		+ \Graph[0.4]{spanning_forest_124_3_circled} 
		\ \cdot\ \Graph[0.4]{spanning_forest_123_4_circled}
		=
		\Graph[0.4]{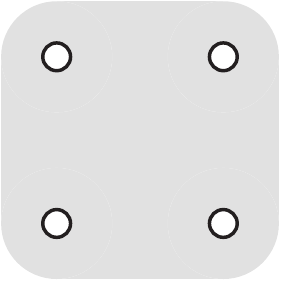}
		\ \cdot\ 
		\Graph[0.4]{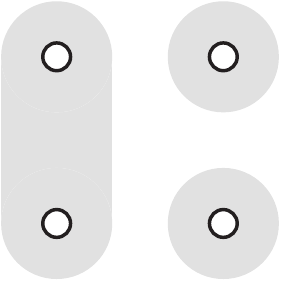}
	\end{equation*}%
	\caption[The forest polynomial identity \eqref{eq:forest-identity-ladderbox}]{Picture of the forest polynomial identity \eqref{eq:forest-identity-ladderbox}: The grey areas show how the four vertices are allocated to the connected components of the forests that contribute to the corresponding spanning forest polynomial.}%
	\label{fig:forest-identity-ladderbox} %
\end{figure}
\begin{proof}
	Construct $G'$ by adding three edges $e_1 = \set{v_1, v_4}$, $e_2 = \set{v_2,v_3}$ and $e_3 = \set{v_3, v_4}$ to $G$, as shown in figure~\ref{fig:dodgson-signs}. We apply lemma~\ref{lemma:dodgson-jacobi} with $I=\set{2,3}$, $J=\set{1,3}$, $A= \set{1}$ and $B=\set{2}$ to find the Dodgson identity 
	$
		\dodgson_{G'}^{12,12} \dodgson_{G'}^{13,23} - \dodgson_{G'}^{13,12} \dodgson_{G'}^{12,23}
		=
		\dodgson_{G'}^{123,123} \dodgson_{G'}^{1,2}
	$.
	Into this we insert \eqref{eq:dodgson-as-forest-boxladder-1}, \eqref{eq:dodgson-as-forest-boxladder-2} and the expansions
	\begin{align*}
		\dodgson_{G'}^{12,12}
		&= \psipol_{G'\setminus 12/3}
		= \SP_3 \psipol_G
			+ \forestpolynom[G]{\set{3},\set{4}}
		\\
		\dodgson_{G'}^{1,2}
		&= \forestpolynom[G \cup e_2]{\set{1,2},\set{3,4}}
			- \forestpolynom[G \cup e_2]{\set{1,3},\set{2,4}}
		= \SP_3 \left(
				\forestpolynom[G]{\set{1,2},\set{3,4}}
				- \forestpolynom[G]{\set{1,3},\set{2,4}}
			\right)
			+ \forestpolynom[G]{\set{1,2},\set{3},\set{4}}
		.
	\end{align*}
	After setting $\SP_3 = 0$ we arrive at the identity
	\begin{equation*}
		\psipol_G^{}
		\forestpolynom[G]{\set{1,2},\set{3},\set{4}}
		=
			(f_{3} + f_{13})(f_{4} + f_{13}) + (f_{12} - f_{13}) \forestpolynom[G]{\set{3},\set{4}}
		= \BoxPoly(f) + f_{13} \left( f_{12} - f_{14} \right)
	\end{equation*}
	between forest polynomials of $G$, where $f_{13} \defas \forestpolynom[G]{\set{1,3},\set{2,4}}$. Here we exploited that
	\begin{equation}
		\forestpolynom[G]{\set{3},\set{4}}
		= \forestpolynom[G]{\set{3},\set{1,2,4}}
			+\forestpolynom[G]{\set{4},\set{1,2,3}}
			+\forestpolynom[G]{\set{1,4},\set{2,3}}
			+\forestpolynom[G]{\set{1,3},\set{2,4}}
		= f_3 + f_4 + f_{14} + f_{13}
		\label{eq:refine-partition-3,4}%
	\end{equation}
	which sums all possible ways how $v_1$ and $v_2$ can be distributed among the parts of this partition.
	The condition on $G$ that any two paths connecting $v_1$ with $v_3$ and $v_2$ with $v_4$ must share a vertex is equivalent to $f_{13} = 0$.
\end{proof}
\begin{lemma}
	\label{lemma:GfunForestBox-add-edge} %
	Let $G'$ be the graph obtained from $G$ by adding a new edge $e = \set{v_3, v_4}$ to $G$. Let $\sdd \defas \sdd_G$ denotes the degree of divergence of the original graph, then
	\begin{equation}
		\GfunForestBox{G'}(z)
		= \BoxPoly^{\EP_e + \sdd - \dimension}
		\int_0^{z_{12}} x^{\dimension/2-\EP_e-1} \left[ \BoxPoly^{\dimension/2-\sdd} \cdot \GfunForestBox{G} \right]_{z_{12} = z_{12} - x}
			\dd x.
		\label{eq:GfunForestBox-add-edge} %
	\end{equation}
\end{lemma}
\begin{proof}
	The spanning forests $F$ of $G'$ that do not include $e$ are precisely the spanning forests of $G$. A spanning forest $F$ including $e$ puts $v_3$ and $v_4$ in the same component of a partition. Hence $\forestpolynom[G']{P} = x \forestpolynom[G]{P}$ (we write $x = \SP_e$) for every partition $P$ that contains $v_3$ and $v_4$ in different parts; in particular
	\begin{equation*}
		f_3' = x f_3, \quad
		f_4' = x f_4 \quad\text{and}\quad
		f_{14}' = x f_{14}
		.
	\end{equation*}
	In contrast,
	$ \psipol_{G'} = x \psipol_G + \forestpolynom[G]{\set{3},\set{4}}$
	can contain spanning trees $T$ with $e \in T$, such that $v_3$ and $v_4$ lie in different components of $T \setminus \set{e}$. We use \eqref{eq:refine-partition-3,4} and $f_{13} = 0$, because we only consider graphs $G$ that allow for a planar drawing with $v_1, v_2, v_3, v_4$ on the outer face in counter-clockwise order. Thus
	\begin{equation*}
		\psipol_{G'}
		= \psipol_{G} \cdot \left( 
					x + z_3' + z_4' + z_{14}'
			\right)
		\ \text{where}\ 
		z' 
		\defas 
		\frac{f}{\psipol_G}
		.
	\end{equation*}
	Similarly we find
	$
		f_{12}'
		= x f_{12} + \forestpolynom[G]{\set{1,2},\set{3},\set{4}}
	$
	and invoke \eqref{eq:forest-identity-ladderbox} to deduce
	\begin{equation*}
		f_{12}'
		= \psipol_G \cdot \left[ 
				z_{12}'(x + z_3' + z_4' + z_{14}') + z_3' z_4'
			\right]
		= \psipol_G \cdot \left[ x z_{12}' + \BoxPoly(z') \right]
		.
	\end{equation*}
	In this way we expressed all relevant forest polynomials of $G'$ in terms of $x$ and the forest polynomials of $G$. Putting this together we find
	\begin{equation*}
		\GfunForestBox{G'}(z)
		=
			\int_{0}^{\infty} \! \dd x\ 
				x^{\EP_e - 1}
			\int_0^{\infty} \dd[4] z'\ 
				\GfunForestBox{G}(z')
				(x + z_3' + z_4' + z_{14}')^{-\dimension/2}
				\delta^{(4)}\left( y - z  \right)
	\end{equation*}
	where using the above calculations, the ratios $y \defas f' / \psipol_{G'}$ are given explicitly by
	\begin{equation*}
		y
		= \left( y_{12}, y_{14}, y_3, y_4 \right)
		=
			\frac{1}{x + z_3' + z_4' + z_{14}'}
			\cdot
			\Big( 
				x z_{12}' + \BoxPoly(z'),
				x z_{14}',
				x z_3',
				x z_4'
			\Big)
		.
	\end{equation*}
	With $x + z_3' + z_4' + z_{14}' = x^2/ (x - y_3 - y_4 - y_{14})$, the inverse transformation reads
	\begin{equation*}
		z'
		= \frac{1}{x - y_3 - y_4 - y_{14}} \cdot \Big(
				x y_{12} - \BoxPoly(y),
				x y_{14},
				x y_{3},
				x y_{4}
			\Big)
	\end{equation*}
	with measure $\dd[4] z' = \left( \frac{x}{x-y_3-y_4-y_{14}} \right)^4 \dd[4] y$. Resolving the $\delta$-constraints this way, we find
	\begin{align*}
		\GfunForestBox{G'}(z)
		&= \int_{\BoxPoly/z_{12}}^{\infty}
			\! \dd x\ 
			x^{\EP_e + 3 - \dimension}
			(x-z_3-z_4-z_{14})^{\dimension/2 - 4}
		\\&\qquad\qquad
		\times
		\GfunForestBox{G}\left( 
			z_{12} - \tfrac{z_3 z_4}{x-z_3-z_4-z_{14}},
			\tfrac{x z_{14}}{x-z_3-z_4-z_{14}},
			\tfrac{x z_3}{x-z_3-z_4-z_{14}},
			\tfrac{x z_4}{x-z_3-z_4-z_{14}}
		\right).
	\end{align*}
	Finally we pull out the factor $x/(x-z_3 - z_4 - z_{14})$ from the arguments of $\GfunForestBox{G}$ using the homogeneity \eqref{eq:GfunForestBox-scaling} and change the integration variable from $x = \BoxPoly/x'$ to $x'$.
\end{proof}
\begin{example}
	\label{ex:ladderbox-forest-edge} %
	In $\dimension = 6$ dimensions, the forest integral of the double box $B_2$ from figure~\ref{fig:boxladders-forest} can be computed from \eqref{eq:ladderbox-forest-1''} for unit indices $\EP_e = 1$ using
	\begin{equation*}
		\GfunForestBox{B_2}(z)
		= \BoxPoly^{-2}
			\int_0^{z_{12}}
			x\cdot
			\GfunForestBox{B_1''}\left( z_{12} - x, z_{14}, z_3, z_4 \right)
			\ \dd x.
		\end{equation*}
	This integral can be expressed in terms of logarithms and dilogarithms, for instance
	\begin{align}
		\GfunForestBox{B_2}(z)
		&= 
			\frac{z_{12}-z_{14}}{\BoxPoly^2} \left[
					\Li_2\left( \frac{z_3 z_4}{\BoxPoly} \right)
					- \Li_2\left( \frac{z_3 z_4 (z_{14}-z_{12})}{z_{14} \BoxPoly} \right)
					+ \ln \frac{\BoxPoly}{z_3 z_4}
						\ln \frac{(z_{14}+z_3)(z_{14}+z_4)}{z_{14}(z_{14}+z_3 + z_4)}
			\right]
		\nonumber\\&\quad
		+ \frac{z_{12}}{\BoxPoly^2} \ln \frac{z_{14} z_3 z_4}{z_{12}(z_{14}+z_3)(z_{14}+z_4)}
		- \frac{\ln (z_3 z_4/\BoxPoly)}{\BoxPoly (z_{14} + z_3 + z_4)}
		.
		\label{eq:ladderbox-forest-2} %
	\end{align}
\end{example}
\subsection{Kinematics}
\label{sec:ladderbox-kinematics}%
The analogue of theorem~\ref{theorem:vw3-3pt} for the massless planar $4$-point topology is
\begin{theorem}
	\label{theorem:ladderbox} %
	Assume that $G$ has four external vertices $v_1$ to $v_4$ such that $f_{13} = 0$ (there are now disjoint paths in $G$ that connect $v_1$ with $v_3$ and $v_2$ with $v_4$).
	When all internal masses $m_e=0$ vanish and the external momenta $p_i$ entering $G$ at $v_i$ fulfil $p_1^2 = p_2^2 = 0$, the Feynman integral $\FR(G)$ is a projective integral of $\GfunForestBox{G}$:
	\begin{equation}
		\FR(G)
		=
			\frac{\Gamma(\sdd)}{\prod_e \Gamma(\EP_e)}
			\int_0^{\infty}
				\frac{
					\GfunForestBox{G}(z)
					\,\Omega
				}{
					\left[
						(p_1+p_2)^2 z_{12} + (p_1 + p_4)^2 z_{14} + p_3^2 z_3^{} + p_4^2 z_4^{}
					\right]^{\sdd}
				}
		.
		\label{eq:ladderbox-projective}%
	\end{equation}
	When $p_3^2 = p_4^2 = 0$ vanish as well, we set $x = (p_1 + p_4)^2 / (p_1 + p_2)^2$ and find
	\begin{equation}
		\FR(G)
		= (p_1+p_2)^{-2\sdd}
			\frac{\Gamma(\sdd)}{\prod_e \Gamma(\EP_e)}
			\int_0^{\infty}
				\frac{
					\GfunForestBox{G}(z)
					\,\Omega
				}{
					\left[
						 z_{12} + x z_{14}
					\right]^{\sdd}
				}
		.
		\label{eq:ladderbox-projective-onshell}%
	\end{equation}
\end{theorem}
\begin{proof}
	The second Symanzik polynomial of $G$ in the given kinematics reads
	\begin{equation*}
		\phipol 
		=	 (p_1+p_2)^2 f_{12} + (p_1 + p_4)^2 f_{14} + p_3^2 f_3^{} + p_4^2 f_4^{}
	\end{equation*}
	such that for convergent $\FR(G)$, we can insert this into \eqref{eq:feynman-integral-parametric} and find
	\begin{align*}
		\FR(G)
		&=
			\int_0^{\infty}
				\prod_e \frac{\SP_e ^{\EP_e - 1} \dd \SP_e}{\Gamma(\EP_e)}
			\int_0^{\infty} 
			\frac{e^{-\phipol/\psipol}}{\psipol^{\dimension/2}}
			\delta^{(4)}\left( \frac{f}{\psipol} - z \right)
			\dd[4] z
		\\
		&=
			\frac{1}{\prod_e \Gamma(\EP_e)}
			\int_{0}^{\infty} 
			\GfunForestBox{G}(z)
			\exp\left[-(p_1+p_2)^2 z_{12} - (p_1+p_4)^2 z_{14} - p_3^2 z_3^{} - p_4^2 z_4^{} \right]
			\dd[4] z
	\end{align*}
	which is the Laplace transform of $\GfunForestBox{G}$. When we exploit the homogeneity, we can perform one integration and arrive at the projective integral representation \eqref{eq:ladderbox-projective}.
\end{proof}
\begin{example}
	\label{ex:ladderbox-1-onshell}%
	In $\dimension = 6$, starting from \eqref{eq:GfunForestBox-1loop-d4-d6} we obtain the well-known box result
	\begin{align}
		s \FR(B_1)
		&= \int_{0}^{\infty} \frac{\dd z_{12}}{z_{12}+x}
			\int_{0}^{\infty} \dd z_3 \int_{0}^{\infty} \dd z_4
			\ 
			\frac{z_{12}}{[ z_{12}( 1+z_{3}+z_{4} ) +z_{4}z_{3}]^{2}}
		\nonumber\\
		&= \int_{0}^{\infty} \frac{\dd z_{12}}{z_{12}+x}
			\int_{0}^{\infty} \frac{\dd z_3}{(1+z_{3})(z_{12}+z_{3})}
		= \int_{0}^{\infty} \frac{\dd z_{12}\ \ln z_{12}}{(z_{12}+x)(z_{12} - 1)}
		= \frac{\pi^2 + \ln^2 x}{2(1+x)}
		.
		\label{eq:ladderbox-1-onshell} %
	\end{align}
\end{example}

\chapter{Hyperlogarithms}
\label{chap:hyperlogs}%
As we motivate below, the iterated integration of rational functions makes it necessary to introduce a class of special functions called hyperlogarithms. Our goal is to present a self-contained account of their fundamental properties, but with a view towards their application to the evaluation of definite integrals of multivariate rational functions.

The main results in section~\ref{sec:hlog-algorithms} are algorithms for integration, differentiation, analytic continuation and series expansion of hyperlogarithms. Many examples illustrate how these reduce explicit computations to formal manipulations of words and therefore lend themselves to straightforward implementation on a computer.
In the next chapter~\ref{chap:hyperint} we comment on our own realization in a computer algebra system.

Section~\ref{sec:Periods} recalls structural results on the special values of multiple polylogarithms. We prove a parity theorem for values at primitive sixth roots of unity which is needed in the computation of a particularly interesting example in $\phi^4$-theory, see section~\ref{sec:P711}.

In the multivariate setting, hyperlogarithms with rational prefactors are not closed anymore under integration. We recall the necessary criterion of linear reducibility and further issues related to the presence of multiple variables in section~\ref{sec:multiple-integrals}. We propose a variant of the polynomial reduction with compatibility graphs in order to track singularities along the recursive integral formulas from the previous chapter. As an application, we will prove the main results of this thesis in section~\ref{sec:recursion-reducibility}.

Apart from this extension and our general algorithm to compute regularized limits, this entire chapter is essentially based on the work of Francis Brown \cite{Brown:PeriodsFeynmanIntegrals,Brown:TwoPoint,Brown:MZVPeriodsModuliSpaces}. We like to point out and recommend the lecture notes \cite{Brown:IteratedIntegralsQFT} which contain an excellent introduction to iterated integrals adapted to our context, as well as the combined exposition \cite{Brown:ModuliSpacesFeynmanIntegrals} on multiple zeta values, moduli spaces and Feynman integrals.

A reader unfamiliar with polylogarithms should skim section~\ref{sec:polylogarithms} before section~\ref{sec:hlog-algorithms} to get a feeling for typical hyperlogarithms which we use in the examples.

Note that all tensor products in this thesis are understood over $\Q$.

\section{How does one integrate rational functions?}
The functions $\C(z)$ in one variable $z$ are closed under the differential operator $\partial_z=\partial/\partial_z$, but antiderivatives (primitives) $\int \dd z$ do not always exist. Partial fractioning yields the basis
\begin{equation*}
	\C(z)
	=
	\C\left[
		z,\frac{1}{z-\sigma} \colon
		\sigma \in \C
	\right]
	=
	\bigoplus_{n \in \N_0} \C \cdot z^n
	\oplus
	\bigoplus_{\sigma \in \C,n \in \N}
	\C\cdot\left(\frac{1}{z-\sigma}\right)^{n}
\end{equation*}
in which the elements $\frac{1}{z-\sigma}$ can not be integrated inside $\C(z)$. Therefore the logarithm $\log (z-\sigma) = \int_{1}^{z-\sigma} \frac{\dd x}{x}$ is introduced as a first transcendental function. It suffices to find primitives $\int P(z)\ \dd z$ (in short $\int P$) of any rational function $P \in \C(z)$.

But adjoining the logarithms alone does still not provide an algebra closed under taking primitives. Further transcendentals are now needed to integrate $\int P \log(z)$. It was Kummer \cite{Kummer:IntegrationenRationalerFormeln} who first studied such \emph{iterated integrals} $\int P \int Q$ of two arbitrary rational functions $P,Q \in \C(z)$ systematically and showed that they can all be expressed in terms of rationals, logarithms and the dilogarithm $\Li_2(z)$ of \eqref{eq:def-Mpl}.
He also considered triple integrals $\int P \int Q \int R$ and found that they can all be written in terms of the same functions and only one new transcendental: the trilogarithm $\Li_3(z)$.

Such an analysis becomes more and more difficult with an increasing number of integrations, but using partial fractioning and integration by parts we can reduce all integrands to the simple form $\frac{1}{z-\sigma}$. These integrals were mentioned by Poincar\'{e} in his study of linear differential equations with algebraic coefficients \cite{Poincare:GroupesEquationsLineaires}: He remarked that the dependence of their solutions on the coefficients can be expanded in the functions
\begin{equation*}
	\Lambda(z, \sigma_1)
	\defas
	\int_0^z \frac{\dd x}{x - \sigma_1}
	\quad\text{and}\quad
	\Lambda(z, \sigma_1,\ldots,\sigma_{n+1})
	\defas
	\int_0^z \frac{\Lambda(x,\sigma_1,\ldots,\sigma_n)}{x - \sigma_{n+1}} \dd x
	.
\end{equation*}
Lappo-Danilevsky carried this out in great detail \cite{LappoDanilevsky} and called these functions \emph{hyperlogarithms}. He first introduced them in \cite{LappoDanilevsky:CorpsRiemann} by
\begin{equation*}
	\HyperLappo{b}{\sigma}{z}
	\defas
	\int_b^z \frac{\dd x}{x-\sigma}
	\quad\text{and}\quad
	\HyperLappo{b}{\sigma_1,\ldots,\sigma_{n+1}}{z}
	\defas
	\int_b^z \frac{\HyperLappo{b}{\sigma_1,\ldots,\sigma_{n}}{x}}{x-\sigma_{n+1}} \dd x,
\end{equation*}
defined whenever $\set{\sigma_1,\ldots,\sigma_n} \cap \set{b,z} = \emptyset$. We will denote these functions by
\begin{equation}
	\int_b^z \letter{\sigma_n}\!\!\cdots\letter{\sigma_1}
	\defas
	\HyperLappo{b}{\sigma_1,\ldots,\sigma_n}{z}
	\quad\text{or even with}\quad
	\int_{\gamma} \letter{\sigma_n}\!\!\cdots\letter{\sigma_1}
	\label{eq:def-Hlog-from-to}%
\end{equation}
when we want to stress the dependence on the homotopy class of the path of integration
\mbox{$\gamma\colon [0,1] \longrightarrow \C\setminus\set{\sigma_1,\ldots,\sigma_n}$}
from $\gamma(0) = b$ to $\gamma(1)=z$. 

The $\C(z)$-span of all hyperlogarithms is by construction the smallest extension of $\C(z)$ which is closed under taking primitives (basically we just added everything without a primitive as a new transcendental function). What makes this approach sensible is that we understand all these new special functions and their relations perfectly well.

\section{Preliminaries on iterated integrals}
\label{sec:iteratedintegrals}%
The notion \eqref{eq:def-Hlog-from-to} of iterated integrals makes perfect sense for arbitrary one-forms $\omega_i \in \Omega^{1}(X)$ on a smooth manifold $X$. Given a smooth path $\gamma\colon [0,1]\longrightarrow X$ we set
\begin{equation}
	\int_{\gamma} \omega_1\!\cdots\omega_n
	\defas
	\int_0^1 \gamma^{\ast}(\omega_1) (t_1)
	\int_0^{t_1} \gamma^{\ast}(\omega_2)(t_2)
	\int_0^{t_2}
	\!\cdots
	\int_0^{t_{n-1}} \gamma^{\ast}(\omega_n)(t_n)
	,
	\label{eq:def-II}%
\end{equation}
which defines a functional of the path $\gamma$ (it does not depend on its parametrization). In terms of the initial piece $\gamma_t \defas \restrict{\gamma}{[0,t]}$ of $\gamma$ from $\gamma(0)$ to $\gamma(t)$, we can write equivalently
\begin{equation}
	\int_{\gamma} \omega_{1} \!\cdots \omega_{n}
	= \int_0^{1} \left[
		\gamma^{\ast}(\omega_1)(t) \int_{\gamma_t} \omega_2\scalebox{0.5}{\!}\cdots\omega_n
		\right]
	\label{eq:def-II-iterated}%
\end{equation}
to stress the underlying idea of iteration. By linear extension, $\int_{\gamma} w$ is defined for any element $w \in T( \Omega^1(X) )$ of the tensor algebra (we set $\int_{\gamma} \emptyWord \defas 1$ for the empty word $\emptyWord$).

Chen studied iterated integrals in great detail and for example constructed a complex out of them that computes the cohomology of the path space over $X$ \cite{Chen:IteratedPathIntegrals,Chen:II}. It is related to the \emph{bar construction} of differential graded algebras, though we will not need any of this machinery here except for his result on \emph{homotopy invariance}.
In general $\int_{\gamma}w$ depends on the shape of $\gamma$, but we want to construct functions of the endpoint $z = \gamma(1)$ only. So we require that $\int_{\gamma} w = \int_{\gamma'} w$ for all homotopic\footnote{We consider all homotopies relative to the endpoints, so that $\gamma(0)$ and $\gamma(1)$ stay fixed.} paths $\gamma \homotop \gamma'$. The dependence on $\gamma$ then remains only through its homotopy class which reflects that $\int_{\gamma} w$ is a multivalued function of $\gamma$'s endpoints.
\begin{lemma}[Chen's integrability condition]
	For any $w \in T(\Omega^1(X))$, the iterated integral $\int_{\gamma} w$ is homotopy invariant if $\delta(w) = 0$, where $\delta\colon T\left( \Omega(X) \right) \longrightarrow T\left( \Omega(X) \right)$ is defined by
\begin{equation}
	\delta \left( \omega_1\!\cdots\omega_n \right)
	\defas
	\sum_{k=1}^{n} \omega_1\!\cdots(d \omega_i)\cdots \omega_n
	-
	\sum_{k=1}^{n-1} \omega_1\!\cdots\left( \omega_k \wedge \omega_{k+1} \right)\cdots\omega_n
	.%
	\label{eq:def-bar-differential}%
\end{equation}
\end{lemma}
\begin{proof}
	By Poincar\'{e}'s lemma, homotopy invariance of $\int_{\gamma} w$ is equivalent to the closedness of its integrand $\sum_{i,j} \omega_i \int_{\gamma} \omega_j u_{i,j}$, where we write $w = \sum_{i,j} \omega_i \omega_j u_{i,j}$ for words $u_{i,j}$ to get a grip on the first two letters. The case $k=1$ in \eqref{eq:def-bar-differential} of $\delta(w)=0$ indeed implies
	\begin{equation*}
		\dd \sum_{i,j} \omega_i \int_{\gamma} \omega_{j} u_{i,j}
		= \sum_{i} 
				\dd \omega_i
				\int_{\gamma} \sum_j u_{i,j}
			-\sum_{i,j} \omega_i \wedge \omega_j \int_{\gamma} u_{i,j}
		= 0,
	\end{equation*}
	and the requirement that the integrand is homotopy invariant itself follows recursively from the conditions where $k>1$.
\end{proof}
The crucial result of Chen is that $\delta(w) = 0$ is also \emph{necessary} for the homotopy invariance if we only consider a suitable (small) subspace $V \subset \Omega^{1}(X)$. This restriction is necessary to obtain a result on linear independence of iterated integrals (see lemma~\ref{lemma:hyperlog-independence}). For example, any exact form $\omega_1 = \dd f$ will reduce the iterated integral
\begin{equation*}
	\int_{\gamma} \omega_1\!\cdots \omega_n
	\urel{\eqref{eq:def-II-iterated}}
	\int_0^1 \dd t\ (f\circ \gamma)'(t) \int_{\gamma_t} \omega_2\cdots\omega_n
	=
		f(\gamma(1)) \int_{\gamma} \omega_2 \cdots \omega_n
		- \int_{\gamma} (f \omega_2) \cdots \omega_n
\end{equation*}
to simpler ones, involving only $n-1$ integrations. Hence it suffices to take $V \cap \dd (\Omega^{0}(X)) = \set{0}$. When $X$ has dimension one, all forms are closed and the condition $\delta(w) = 0$ is vacuous as $\dd \omega_i,\omega_i \wedge \omega_j \in \Omega^{2}(X) = \set{0}$, so $V$ is generated by a finite basis $\setexp{\letter{\sigma}}{\sigma \in \Sigma}$ of $1$-forms and their equivalence classes yield a basis of the cohomology $H^{1}(X)$.

This is the case for hyperlogarithms with singularities in some fixed set $\Sigma \subset \C$, where $X = \C\setminus \Sigma$ and $\letter{\sigma} \defas \dd \log (z-\sigma)$. We will always assume that $0 \in \Sigma$.
\nomenclature[Sigma]{$\Sigma$}{a finite set $\Sigma\subset \C$}%
\nomenclature[omega sigma]{$\letter{\sigma}$}{the differential form $\dd z/(z-\sigma)$}

For several variables ($\dim X>1$), homotopy invariance is more complicated to characterize. We return to this question briefly in section~\ref{sec:multiple-variables}.

\subsection{The shuffle (Hopf) algebra}
\label{sec:shuffle-algebra}%
If $V$ has a $\Q$-basis $\setexp{\letter{\sigma}}{\sigma \in \Sigma}$, we write $T(\Sigma) = T(V)$ for its graded tensor algebra
\begin{equation}
	T(\Sigma)
	\defas
	\bigoplus_{n=0}^{\infty}
	T_n
	=
	\bigoplus_{w \in \Sigma^{\times}} \Q w
	\quad\text{with components}\quad
	T_n
	\defas
	V^{\tp n}
	=
	\bigoplus_{w\in \Sigma^n} \Q w
	\label{eq:def-tensor-algebra}%
\end{equation}%
\nomenclature[T(Sigma)]{$T(\Sigma)$}{tensor algebra with letters $\letter{\sigma}$ ($\sigma\in\Sigma$), equation~\eqref{eq:def-tensor-algebra}\nomrefpage}%
of \emph{weight} $n$ which are spanned by the words $w = \letter{\sigma_1}\!\!\cdots\letter{\sigma_n} \in \Sigma^n$ with $\abs{w} = n$ letters. The set
$
	\Sigma^{\times}
	\defas
	\set{\emptyWord}
	\cupdot
	\Sigma
	\cupdot
	\Sigma^2
	\cupdot
	\ldots
$
of all words contains the empty word $\set{\emptyWord} = \Sigma^0$ in weight zero. It acts as the unit for the non-commutative \emph{concatenation product} on $T(\Sigma)$ defined by
\begin{equation}
	\omega_{1}\!\cdots\omega_n
	\cdot\,
	\omega_{n+1}\!\cdots\omega_{n+m}
	\defas
	\omega_1\!\cdots\omega_{n+m}.
	\label{eq:def-concatenation-product}%
\end{equation}%
We also introduce the commutative \emph{shuffle product} defined recursively by
\begin{equation}
	(\omega_1 w_1)
	\shuffle
	(\omega_2 w_2)
	\defas
	\omega_1\left( w_1 \shuffle \omega_2 w_2 \right)
	+
	\omega_2\left( \omega_1 w_1 \shuffle w_2 \right)
	\label{eq:def-shuffle-product}%
\end{equation}%
\nomenclature[v sha w]{$v \shuffle w$}{shuffle product of words, equation~\eqref{eq:def-shuffle-product}, page~\pageref{eq:def-shuffle-product}}%
and $\emptyWord \shuffle w = w \shuffle \emptyWord = w$. It shuffles the letters of both factors in all possible ways that keep the order among the letters of each factor, so we can also write
\begin{gather*}
	\omega_1\!\cdots\omega_n
	\shuffle
	\omega_{n+1}\!\cdots\omega_{n+m}
	=
	\sum_{\pi \in \Shuffles{n}{m}}
	\omega_{\pi(1)}\!\cdots\omega_{\pi(n+m)}
	\qquad\text{where}
	\nonumber\\
	\Shuffles{n}{m}
	\defas
	\setexp{
		\pi \in S_{n+m}
	}{\pi^{-1}(1)<\cdots<\pi^{-1}(n)
		\ \text{and}\ 
		\pi^{-1}(n+1)<\cdots<\pi^{-1}(n+m)
	}%
\end{gather*}
denotes the shuffles of $1,\ldots,n$ and $n+1,\ldots,n+m$ (a special set of permutations). Note that \eqref{eq:def-shuffle-product} also holds with respect to the last letter of a word:
\begin{equation}
	(w_1 \omega_1) \shuffle (w_2 \omega_2)
	=
	(w_1 \shuffle w_2\omega_2)\omega_1
	+
	(w_1 \omega_1 \shuffle w_2)\omega_2
	.%
	\label{eq:shuffle-product-tail} %
\end{equation}
\begin{lemma}
	\label{lemma:iint-shuffle} %
	For any words $v,w \in T(\Sigma)$, we have the product identity
	\begin{equation}
		\int_{\gamma} \!v
		\cdot
		\int_{\gamma} \!w
		=
		\int_{\gamma} (v \shuffle w)
		.
		\label{eq:iint-shuffle} %
	\end{equation}
\end{lemma}
\begin{proof}
	By linearity it suffices to consider individual words $v=\omega_1 v'$, $w=\omega_2 w'$ and with an induction over the lengths of the words we may already assume
	$
				\int_{\gamma} v'
				\cdot
				\int_{\gamma} w
				= \int_{\gamma} (v' \shuffle w)
	$
	and
	$
				\int_{\gamma} v
				\cdot
				\int_{\gamma} w'
				= \int_{\gamma} (v \shuffle w')
	$. With $\gamma_t \defas \restrict{\gamma}{[0,t]}$,
	\begin{align*}
		\int_{\gamma}\! v \cdot \int_{\gamma} \!w
		&\urel{\eqref{eq:def-II}}
		\int_{0}^{1} \!\!\gamma^{\ast} (\omega_1)(t) \int_{\gamma_t} \!\!v'
		\cdot
		\int_{0}^{1} \!\!\gamma^{\ast} (\omega_2)(s) \int_{\gamma_s} \!\!w'
		\\
		&\urel{}
		\int_{0}^{1} \!\!\gamma^{\ast} (\omega_1)(t) \left[
				\int_{\gamma_t}\!\! v'
				\cdot
				\int_{\gamma_t}\!\! w
		\right]
		+
		\int_{0}^{1} \!\!\gamma^{\ast} (\omega_2)(s) \left[
				\int_{\gamma_s} \!\! v
				\cdot
				\int_{\gamma_s} \!\! w'
		\right]
		\\
		&\urel{} \int_{\gamma} \omega_1\left( v' \shuffle w \right)
			+\int_{\gamma} \omega_2\left( v \shuffle w' \right)
	\end{align*}
	where we split the outermost integrations according to whether $s<t$ or $s>t$.
\end{proof}
So in particular the span of iterated integrals is an algebra and the integration map $w \mapsto \int_{\gamma} w$ is a morphism of algebras, we also say that $\int_{\gamma}$ is a \emph{character} on $T(\Sigma)$. A second very important formula relates integrals along two different paths that can be concatenated. It is often attributed to Chen, but it was stated before by Lappo-Danilevsky for hyperlogarithms \cite[M\'{e}moire 2, \S 2 (13)]{LappoDanilevsky}.
\begin{lemma}
	\label{lemma:path-concatenation}%
	Let $\gamma,\eta\colon [0,1] \longrightarrow X$ denote paths that meet at $\gamma(1) = \eta(0)$ and $\gamma \concat \eta$ their concatenation running from $\gamma(0)$ to $\eta(1)$. For any word $\omega_1\!\cdots\omega_n \in T(\Sigma)$ we have
	\begin{equation}
		\int_{\gamma \concat \eta} \omega_1 \!\cdots \omega_n
		=
		\sum_{k=0}^{n} 
			\int_{\eta} \omega_1 \!\cdots\omega_{k}
			\cdot
			\int_{\gamma} \omega_{k+1} \!\cdots \omega_n
		.
		\label{eq:path-concatenation}%
	\end{equation}
\end{lemma}
\begin{proof}
	Say that $\gamma \concat \eta (t) = \gamma(2t)$ for $t \in [0,\frac{1}{2}]$ and $\eta(2t-1)$ otherwise. For $w'=\omega_2\!\cdots\omega_n$,
	\begin{align*}
		\int_{\gamma \concat \eta} w
		&\urel{\eqref{eq:def-II}}
			\int_0^{1/2} (\gamma \concat \eta)^{\ast} (\omega_1)(t) \int_{(\gamma \concat \eta)_{t}} w'
			+
			\int_{1/2}^{1} (\gamma \concat \eta)^{\ast} (\omega_1)(t) \int_{(\gamma \concat \eta)_{t}} w'
		\\&\urel{}
			\int_{0}^{1} \gamma^{\ast} (\omega_1)(s) \int_{\gamma_s} w'
			+
			\int_{0}^{1} \eta^{\ast} (\omega_1)(u) \int_{\gamma \concat \eta_u} w'
	\end{align*}
	with $s = 2t$ and $u=2t-1$ proves the statement inductively.
\end{proof}
This result will be used in the sequel to compute monodromies, analytic continuations and expansions of hyperlogarithms near singular points. The \emph{deconcatenation coproduct}
\begin{equation}
	\cop\colon T(\Sigma) \longrightarrow T(\Sigma) \tp T(\Sigma),\quad
	\cop \left( \omega_1\!\cdots\omega_n \right)
	\defas
	\sum_{k=0}^n \omega_1\!\cdots\omega_k
		\tp
		\omega_{k+1}\!\cdots\omega_n
	\label{eq:def-deconcatenation-coproduct}%
\end{equation}
is often abbreviated by $\cop(w) = \sum_{(w)} w_{(1)} \tp w_{(2)}$ and endows $T(\Sigma)$ with the structure of a \emph{Hopf algebra} \cite{Manchon,Sweedler,Reutenauer:FreeLieAlgebras}. This means that $\cop (v \shuffle w) = \sum_{(v),(w)} (v_{(1)} \shuffle w_{(1)}) \tp (v_{(2)} \shuffle w_{(2)})$ is multiplicative and furthermore we have the \emph{antipode}
\begin{equation}
	\antipode\colon T(\Sigma) \longrightarrow T(\Sigma),
	\quad
	\omega_1\!\cdots\omega_n
	\mapsto
	(-\omega_n)\cdots(-\omega_1)
	=
	(-1)^n \omega_n\!\cdots\omega_1.
	\label{eq:def-antipode}%
\end{equation}
It is easy to check that $\int_{\gamma} (\antipode w) = \int_{\gamma^{-1}} w$ is the iterated integral along the inverted path $\gamma^{-1}(t) = \gamma(1-t)$. The augmentation $ \Q\cdot \emptyWord \subset T(\Sigma)$ defines the \emph{co-unit}
\begin{equation}
	\counit\colon T(\Sigma) \longrightarrow \Q,
	\quad
	\sum_{w\in\Sigma^{\times}} \lambda_w \cdot w
	\mapsto
	\lambda_{\emptyWord}
	\label{eq:def-counit}%
\end{equation}
which just extracts the coefficient of the empty word. The antipode obeys its defining relation $\antipode \convolution \id = \id \convolution \antipode = \counit$ where the \emph{convolution product}
\begin{equation}
	(f \convolution g) (w)
	\defas
	\sum_{(w)} f\left(w_{(1)}\right) \cdot g\left( w_{(2)} \right)
	\label{eq:def-convolution-product}%
\end{equation}
is defined for any pair of linear maps $f,g\colon T(\Sigma) \longrightarrow \alg$ that take values in a commutative algebra $\alg$. 
Note that $f \convolution g$ is itself linear and if both $f$ and $g$ are characters (multiplicative), so stays
$
	(f\convolution g)(v \shuffle w) 
	=
		(f\convolution g)(v) 
		\cdot
		(f\convolution g)(w) 
$.

Within this Hopf algebra terminology (which is summarized nicely in \cite{Manchon}), the path concatenation formula \eqref{eq:path-concatenation} can be stated as
$ \int_{\gamma \concat \eta} = \int_{\eta} \convolution \int_{\gamma}$.

\subsection{Regularization}
\label{sec:shuffle-regularization}%
It is well known that any shuffle algebra
$
	T(\Sigma) \isomorph \Q\left[ \Lyndons(\Sigma) \right]
$
is free and and an explicit algebra basis is furnished by \emph{Lyndon words} \cite{Radford:BasisShuffleAlgebra}. These are defined with respect to a total order $<$ on $\Sigma$ as those words which are smaller than all their proper suffixes (with respect to the lexicographic order on $\Sigma^{\times}$ induced by $<$),
	\begin{equation}
		\Lyndons(\Sigma)
		\defas
		\setexp{w=\letter{\sigma_1}\!\cdots\letter{\sigma_n} \in \Sigma^{\times}}{
			n > 0
			\ \text{and}\ 
			w < \letter{\sigma_i}\!\cdots\letter{\sigma_n}
			\ \text{for all}\ 
			1 < i \leq n
		}.
		\label{eq:def-lyndon-words} %
	\end{equation}
This means that each $w \in T(\Sigma)$ has a unique representation
as a polynomial in Lyndon words. In the sequel it will prove extremely useful to exploit this structure to rewrite a general word in terms of shuffle products of special words which enjoy additional properties. Without explicitly referring to it every time, we will make frequent use of
\begin{lemma}
	\label{lemma:shuffle-decomposition} %
	For disjoint sets $A,B \subset \Sigma$, any $w \in T(\Sigma)$ admits a unique decomposition
	\begin{equation}
		w = \sum_{a \in A^{\times}} \sum_{b \in B^{\times}} a \shuffle b \shuffle w_{A,B}^{(a,b)}
		\label{eq:shuffle-decomposition} %
	\end{equation}
	into words $a$ ($b$) that consist of letters only in $A$ ($B$) and words $w_{A,B}^{(a,b)}$ which neither begin with a letter in $B$ nor end in a letter from $A$.
\end{lemma}
An explicit formula to compute \eqref{eq:shuffle-decomposition} is provided by
\begin{lemma}
	\label{lemma:shuffle-decomposition-formula}%
	Let $w$ be any word and $\letter{\sigma}$ one of its letters such that $w = u\letter{\sigma} a$ for some word $u$ and another word $a = \letter{a_1}\!\!\cdots\letter{a_n}$. Then we have the identity
	\begin{equation}
		w
		=
		\sum_{i=0}^n \left[ u \shuffle (-\letter{a_i}) \cdots (-\letter{a_1}) \right] \letter{\sigma}
			\shuffle \letter{a_{i+1}}\!\!\cdots\letter{a_n}
		=
		\sum_{(a)} \left[ u \shuffle S\big(a_{(1)}\big) \right] \letter{\sigma} \shuffle a_{(2)}
		.
		\label{eq:shuffle-decomposition-formula-tail} %
	\end{equation}
	The same holds in the reversed form $a \letter{\sigma} u = \sum_{(a)} a_{(1)} \shuffle \letter{\sigma} \left[ \big(S a_{(2)}\big) \shuffle u \right]$.
\end{lemma}
\begin{proof}
	The statement is trivial for $n=0$ and we apply induction over $n$: For $n>0$, the outer shuffle product in \eqref{eq:shuffle-decomposition-formula-tail} decomposes with respect to the last letter using \eqref{eq:shuffle-product-tail} into
	\begin{multline*}
		\left\{ 
				\sum_{i=0}^{n-1} 
						\left[ u \shuffle (-\letter{a_i})\cdots(-\letter{a_1}) \right] \letter{\sigma}
						\shuffle
						\letter{a_{i+1}}\!\!\cdots\letter{a_{n-1}}
		\right\}\letter{a_n}
	\\
	+ \left\{ 
				u
				\shuffle
				\sum_{i=0}^n (-\letter{a_i})\cdots(-\letter{a_1}) \shuffle \letter{a_{i+1}}\!\!\cdots\letter{a_n}
		\right\} \letter{\sigma}
		.
	\end{multline*}
	The first summand is the desired $u\letter{\sigma}\letter{a_1}\!\!\cdots\letter{a_{n-1}}\letter{a_n}$ by the induction hypothesis, whereas the second summand vanishes since it is 
	$
		\left\{
			u
			\shuffle
			(\antipode \convolution \id) (\letter{a_1}\!\!\cdots\letter{a_n})
		\right\}
		\letter{\sigma}
	$
	and $S \convolution \id = \counit$ vanishes on any non-empty word.
\end{proof}
\begin{proof}[Proof of lemma~\ref{lemma:shuffle-decomposition}]
	First use \eqref{eq:shuffle-decomposition-formula-tail} to write $w = \sum_{a \in A^{\times}} a \shuffle w^{(a)}$ such that with $w^{(a)}$ is free of words that end in $A$. Then express each $w^{(a)} = \sum_{b \in B^{\times}} b \shuffle w^{(a,b)}$ with the reversed form of lemma~\ref{lemma:shuffle-decomposition-formula} such that $w^{(a,b)}$ does not contain words that begin with a letter in $B$.
Note that the last letter of any word in $w^{(a,b)}$ is either the last letter of $w^{(a)}$ or some letter in $B$, and therefore not in $A$. So indeed we obtained an expansion of the form \eqref{eq:shuffle-decomposition}.

To finish the proof of lemma~\ref{lemma:shuffle-decomposition}, we must only realize the uniqueness of all $w^{(a,b)}$. So assume that for $w=0$, there would be some non-zero $w^{(a,b)}$. 
Pick $a$ and $b$ of maximal total length $\abs{a}+\abs{b}$ such that $w^{(a,b)}\neq0$. Then the only words on the right-hand side of \eqref{eq:shuffle-decomposition} which have $b$ as a prefix and $a$ as a suffix are precisely $b w^{(a,b)} a$, stemming from $b \shuffle w^{(a,b)} \shuffle a$ (and no other summand). But these words must have coefficient zero ($w=0$), thus $w^{(a,b)}=0$ contradicts our choice of $a$ and $b$.
\end{proof}
\begin{definition}
	\label{def:shuffle-regularization} %
	For disjoint sets $A,B \subset \Sigma$, the \emph{shuffle regularization} is the coefficient of the empty words $a=b=\emptyWord$ in the decomposition \eqref{eq:shuffle-decomposition}:
	\begin{equation}
		\WordReg{A}{B}\colon T(\Sigma) \longrightarrow T(\Sigma)
		,\quad
		w \mapsto
		w_{A,B}^{(\emptyWord,\emptyWord)}
		.
		\label{eq:def-shuffle-regularization} %
	\end{equation}
	We will mostly consider cases with $\abs{A}, \abs{B} \leq 1$ and write $\WordReg{\sigma}{\tau}$ instead of $\WordReg{\set{\sigma}}{\set{\tau}}$. Furthermore, empty sets are suppressed: $\WordReg{\sigma}{} \defas \WordReg{\sigma}{\emptyset}$ and $\WordReg{}{\tau} \defas \WordReg{\emptyset}{\tau}$.
\end{definition}%
\begin{remark}
	The shuffle regularization fulfils the following properties, which are immediate consequences of the definitions:
	\begin{enumerate}
		\item $\WordReg{A}{B} (w \shuffle w') = \WordReg{A}{B}(w) \shuffle \WordReg{A}{B}(w')$ for all $w,w' \in \Sigma^{\times}$ ($\WordReg{A}{B}$ is a character),

		\item $\WordReg{A}{B}(w) = w$ for all words $w = \letter{\sigma_1}\!\!\cdots\letter{\sigma_n}$ with $\sigma_1 \notin B$ and $\sigma_n \notin A$,

		\item $\WordReg{A}{B}(w) = 0$ when $1 \neq w \in A^{\times} \cup B^{\times}$ has only letters in $A$ or only letters in $B$,

		\item	$
			\WordReg{A}{\emptyset} \circ \WordReg{\emptyset}{B}
			=	\WordReg{A}{B}
			= \WordReg{\emptyset}{B} \circ \WordReg{A}{\emptyset}
			$ commute and

		\item $\WordReg{A}{B} = \WordReg{A}{B} \circ \WordReg{A}{B}$ is a projection.
	\end{enumerate}
	For the multiplicativity note that by \eqref{eq:def-shuffle-product} and \eqref{eq:shuffle-product-tail}, the first (last) letter of $ w_{A,B}^{(a,b)} \shuffle {w'}_{A,B}^{(a,b)}$ is the first (last) letter of either factor and thus not in $B$ ($A$).
\end{remark}
\begin{corollary}
	\label{corollary:shuffle-regularization}%
	Lemma~\ref{lemma:shuffle-decomposition-formula} says that for any words $u\in T(\Sigma)$, $a \in T(A)$ and $\sigma\notin A$, the regularizations are explicitly computed by the formulas
	\begin{equation}
		\WordReg{A}{} \left(u\letter{\sigma}a \right) 
		= \left(u \shuffle \antipode a \right) \letter{\sigma}
		\quad\text{and}\quad
		\WordReg{}{A} \left(a\letter{\sigma}u \right) 
		= \letter{\sigma} \left(u \shuffle \antipode a \right)
		.
		\label{eq:shuffle-regularization-formula}%
	\end{equation}
	Furthermore, the identity \eqref{eq:shuffle-decomposition-formula-tail} translates into
	\begin{equation}
		u\letter{\sigma}a
		= \sum_{(a)} \WordReg{A}{}\left(u \letter{\sigma} a_{(1)} \right) \shuffle a_{(2)}
		\quad\text{and}\quad
		a\letter{\sigma}u
		= \sum_{(a)} a_{(1)} \shuffle \WordReg{}{A}\left(a_{(2)}\letter{\sigma}u \right)
		.
		\label{eq:shuffle-regularization-identity}%
	\end{equation}
\end{corollary}
In Hopf algebra terms, this says that
$
	\id 
	= \WordReg{A}{} \convolution \ProjectOn{A}
	= \ProjectOn{A} \convolution \WordReg{}{A}
$
when we let $\ProjectOn{A}\colon T(\Sigma) \longrightarrow T(A)$ denote the natural projection on words with all letters in $A$. It is easy to prove a manifest form of the shuffle decomposition \eqref{eq:shuffle-decomposition} as the identity
\begin{equation}
	\id
	= \ProjectOn{B} \convolution \WordReg{A}{B} \convolution \ProjectOn{A}
	,\quad\text{equivalently}\quad
	\WordReg{A}{B}
	= \ProjectOn{B}^{\convolution -1} \convolution \id \convolution \ProjectOn{A}^{\convolution-1}.
	\label{eq:reglim-convolution-formula}%
\end{equation}
Write any word $w = \letter{b_1}\!\!\cdots\letter{b_r} u \letter{a_1}\!\!\cdots\letter{a_s}$ such that $b_1,\ldots,b_r \in B$ and $a_1,\ldots,a_s \in A$ but $u$ neither begins in $B$ nor ends in $A$. Then with \eqref{eq:def-antipode} we get
\begin{equation}
	\WordReg{A}{B}(w)
	=
	\sum_{k=0}^r \sum_{l=0}^s (-1)^{k+s-l} \letter{b_k}\!\!\cdots\letter{b_1} \shuffle (\letter{b_{k+1}}\!\!\cdots\letter{b_r} u \letter{a_1}\!\!\cdots\letter{a_l}) \shuffle \letter{a_s}\!\!\cdots\letter{a_{l+1}}
	.
	\label{eq:wordreg-general}%
\end{equation}

\subsection{Multiple variables}
\label{sec:multiple-variables}%
\begin{definition}
	\label{def:bar-objects}%
	Let $S \subset \Q[z_1,\ldots,z_n]$ denote a set of irreducible polynomials in $n$ variables and $X_S \defas \Affine^n \setminus \bigcup_{f \in S} \Vanishing(f)$ the variety given as the complement of the hypersurfaces defined by their vanishing loci $\Vanishing(f) \subset \Affine^n$. The set%
\nomenclature[S]{$S$}{a set $S\subset \Q[z_1,\ldots,z_n]$ of irreducible polynomials}%
\nomenclature[V(f)]{$\Vanishing(f)$}{the zero set $\Vanishing(f) = \setexp{z}{f(z)=0}$ of a polynomial $f$}%
\nomenclature[XS]{$X_S$}{the complement $X_S = \Affine^n \setminus \bigcup_{f\in S} \Vanishing(f)$ of $S$, page~\pageref{sec:multiple-variables}}
	\begin{equation}
		\Omega_S
		\defas
		\setexp{\letter{f}}{f \in S}
		\subset
		\Omega^1\left( X_S \right)
		,\quad
		\letter{f} \defas \dd \log(f)
		\label{eq:def-bar-object-forms}%
	\end{equation}%
\nomenclature[omega f]{$\letter{f}$}{differential form $\letter{f}=\dd \log(f)$, equation~\eqref{eq:def-bar-object-forms}\nomrefpage}%
	of smooth, closed but not exact one-forms generates a space of homotopy invariant iterated integrals corresponding to the \emph{integrable words}
	\begin{equation}
		\BarObjects(S)
		\defas
		\setexp{w \in T(\Omega_S)}{\delta(w) = 0}.
		\label{eq:def-bar-objects}%
	\end{equation}%
\end{definition}%
\nomenclature[B(S)]{$\BarObjects(S)$}{integrable words of $T(\Omega_S)$, equation~\eqref{eq:def-bar-objects}\nomrefpage}%
In general, $\BarObjects(S) \subset T(\Omega_S)$ forms a non-trivial Hopf subalgebra (one can check that it is closed under $\cop$ and $\shuffle$).
One key observation in \cite{Brown:MZVPeriodsModuliSpaces} is that if all $f = z_1 f^{1} + f_1 \in S$ are linear in $z_1$, then one can embed\footnote{%
Explicitly, this is the map $(\HlogProjection{z_1} \tp \ProjectOn{z_1=0}) \cop$ from section~\ref{sec:ii-several-variables}.}
\begin{equation}
	\BarObjects(S)
	\hookrightarrow
	T(\Sigma_1)
	\tp
	\BarObjects\left( \restrict{S}{z_1=0} \right)
	\quad\text{with}\quad
	\Sigma_1
	\defas
	\setexp{-\frac{f_{1}}{f^{1}}}{f \in S \quad\text{and}\quad f^{1} \neq 0}
	\label{eq:barobjects-decomposition} %
\end{equation}
into a product of a (simple) hyperlogarithm algebra and iterated integrals in one variable less. The set $\restrict{S}{z_1=0} \subset \Q[z_2,\ldots,z_n]$ holds the irreducible factors of all constant parts $f_1 = \restrict{f}{z_1=0}$ of the polynomials $f \in S$. If these are all linear in $z_2$, we can continue this process and may eventually arrive at a full factorization
\begin{equation}
	\BarObjects(S)
	\hookrightarrow
	T(\Sigma_1) \tp \cdots \tp T(\Sigma_n)
	\label{eq:barobjects-decomposition-full}%
\end{equation}
into hyperlogarithm algebras. Indeed this is the key concept of this chapter, because it allows us to use the simple hyperlogarithms to describe iterated integrals in many variables without much effort.

Francis Brown originally developed this for the \hypertarget{eqmodulispace}{moduli spaces}\label{eqmodulispace} $\Moduli{0}{n}$ of the Riemann sphere $\RSphere = \C \cup \set{\infty}$ with $n$ marked points \cite{Brown:MZVPeriodsModuliSpaces}. Explicitly,
$\Moduli{0}{n+3} = X_S$
is characterized by the hypersurfaces
$
	S 
	= 
	\setexp{z_i,1-z_i}{1 \leq i \leq n} \cup \setexp{z_i - z_j}{1 \leq i < j \leq n}
$
such that the factorization \eqref{eq:barobjects-decomposition-full} applies with $\Sigma_i = \set{0,1,z_1,\ldots,z_{i-1}}$ (using \eqref{eq:barobjects-decomposition} first for $z_n$, then for $z_{n-1}$ and so on). Moreover, since the projections $\pi_i$ forgetting $z_i$ in
\begin{equation*}
	\Moduli{0}{n+3} \xrightarrow{\pi_n} \Moduli{0}{n+2} 
	\xrightarrow{\pi_{n-1}} \cdots \xrightarrow{\pi_2}
	\Moduli{0}{4} \xrightarrow{\pi_1} \Moduli{0}{3}
\end{equation*}%
\nomenclature[M 0 n]{$\Moduli{0}{n}$}{moduli space of $\RSphere$ with $n$ marked points\nomrefpage}%
are actually fibrations, the maps in \eqref{eq:barobjects-decomposition} and \eqref{eq:barobjects-decomposition-full} are surjective and therefore isomorphisms. These lift to (non-canonical) isomorphisms of the associated algebras of iterated integrals (after choosing base points and tensoring with $\C$), which can be computed effectively with the \emph{symbol map} of \cite{BognerBrown:SymbolicIntegration,BognerBrown:GenusZero}.

Note that the factorization \eqref{eq:barobjects-decomposition-full} can always be enforced if we allow for algebraic zeros $\sigma \in \Sigma_i \subset \overline{\Q[z_{i+1},\ldots, z_n]}$ to factorize also non-linear polynomials $f \in \Q[z_i,\ldots,z_n]$ completely with respect to $z_i$. The point is that subsequent integrals of such functions in general leave the space of hyperlogarithms and introduce new special functions. A simpler case occurs if we only need to extend the constants from $\Q$ to $\overline{\Q}$ and indeed we give an example adjoining sixth roots of unity in section~\ref{sec:P711}.
\begin{figure}
	\centering
	$
		\Graph[0.4]{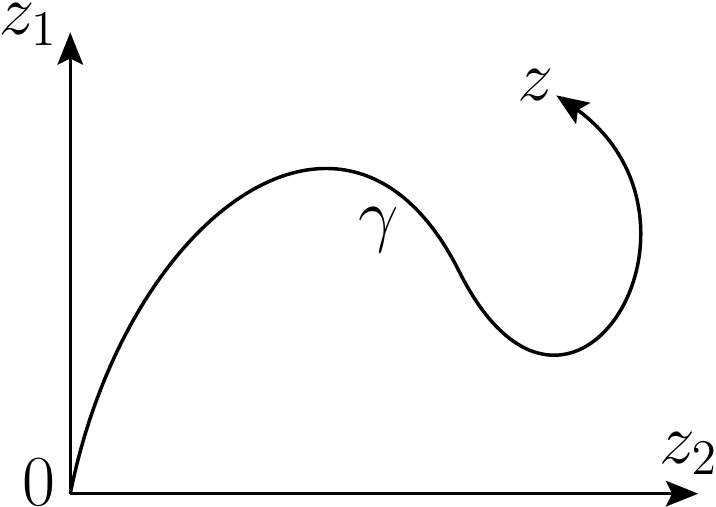}
		\rightarrow
		\Graph[0.4]{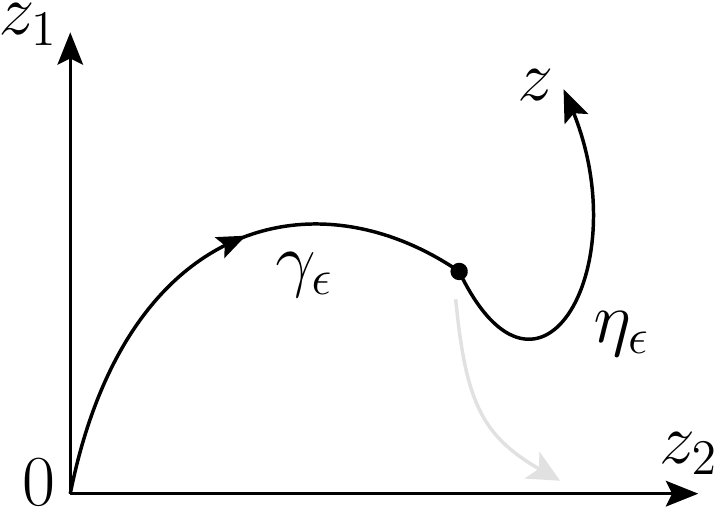}
		\rightarrow
		\Graph[0.4]{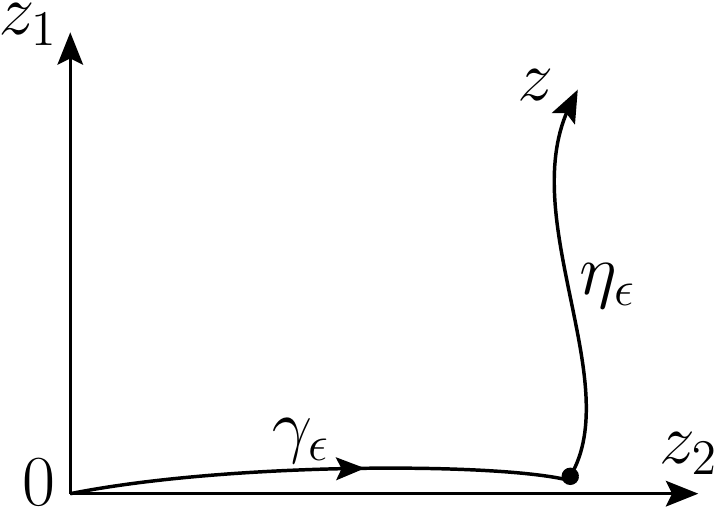}
		\rightarrow
		\Graph[0.4]{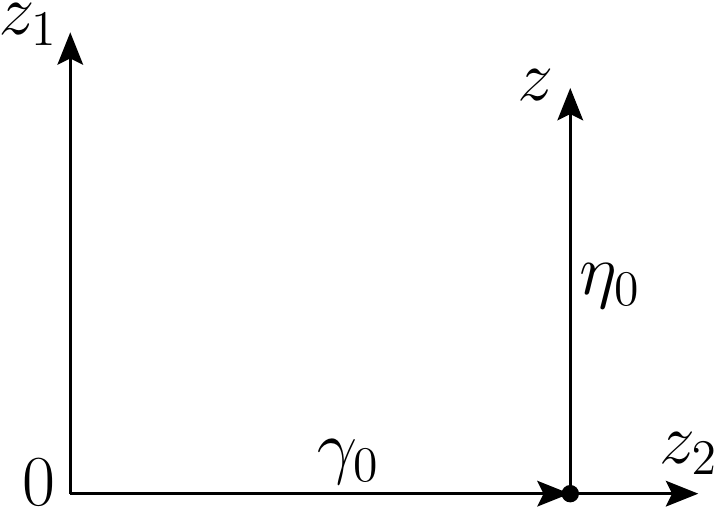}
	$
	\caption[Decomposition of an iterated integral into fibre and base]{A homotopy from the integration path $\gamma \homotop \gamma_{\epsilon} \concat \eta_{\epsilon}$ to the horizontal $\gamma_0$ $(z_1=0)$ followed by the vertical line $\eta_0$ with constant $z_2$, \ldots, $z_n$ yields the factorization \eqref{eq:barobjects-decomposition} of iterated integrals of $z=(z_1,\ldots,z_n)$ into hyperlogarithms along the \emph{fibre} $z_1$ and iterated integrals of $(z_2,\ldots, z_n)$ in the \emph{base} $\set{z_1 = 0}$.}%
	\label{fig:fibration-contour-deformation}%
\end{figure}
\begin{remark}
The idea behind \eqref{eq:barobjects-decomposition} is to exploit homotopy invariance of $\int_{\gamma} w$ for any integrable word $w \in \BarObjects(S)$. We may deform the path $\gamma \homotop \gamma_{\epsilon} \concat \eta_{\epsilon}$ continuously such that $\eta_0(t) = (tz_1, z_2, \ldots, z_n)$ moves only in the $z_1$ direction and $\gamma_0(t) = (0,\gamma^{(2)}(t), \ldots, \gamma^{(n)}(t))$ has constantly $\gamma_0^{(1)}(t) = 0$ as shown in figure~\ref{fig:fibration-contour-deformation}. According to lemma~\ref{lemma:path-concatenation},
\begin{equation}
	\int_{\gamma} \! w
	= \sum_{(w)} \int_{\eta_{\epsilon}}\!\!\! w_{(1)}\cdot \int_{\gamma_{\epsilon}}\!\!\! w_{(2)}
	\quad\text{for all $\epsilon > 0$, so}\quad
	\int_{\gamma} \!w
	= \lim_{\epsilon\rightarrow 0}\ 
	\sum_{(w)} \int_{\eta_{\epsilon}}\!\!\! w_{(1)}\cdot \int_{\gamma_{\epsilon}}\!\!\! w_{(2)}
	.
	\label{eq:iteratedintegrals-decomposition-limit}%
\end{equation}
If both factors stay finite individually when $\epsilon \rightarrow 0$ we are done: 
\begin{itemize}
	\item $\int_{\eta_0} w_{(1)}$ is a hyperlogarithm with differential forms $\dd z_1/(z_1 - f_1/f^1)$ on the one-dimensional fibre of the projection $z \mapsto (z_2,\ldots,z_n)$ and
	\item $\int_{\gamma_0} w_{(2)}$ does not depend on $z_1$ at all with its forms $ \dd \log (f_1)$ on the base $\set{z_1=0}$ of the projection.
	\end{itemize}%
In the next section we show how to deal with divergences in the limit $\epsilon \rightarrow 0$. Note that in general, the so-defined inclusion \eqref{eq:barobjects-decomposition} is not surjective; to guarantee an isomorphism we need
$\pi_1\colon X_S \longrightarrow X_{\restrict{S}{z_1=0}}$
to be a fibration.\footnote{This amounts to the requirement $S_1 \subset S$ for the \emph{reduction} $S_1$ from definition~\ref{def:lin-reducible-polys}. Geometrically, adjoining $S_1$ to $S$ precisely cuts out all pieces from the base above which the fibre of $\pi_1$ degenerates.}
For our application to integration problems in section~\ref{sec:multiple-integrals} this is not required and it suffices to have an ambient algebra of hyperlogarithms at hand.
\end{remark}

\subsection{Tangential base points}%
\label{sec:tangential-base-points}%
If the initial point $u =\gamma(0)$ tends to zero, a hyperlogarithm $\int_{\gamma} \! w = \sum_k \log^k(u) f^{(k)}(u)$ can diverge logarithmically but with coefficients $f^{(k)}(u)$ that are analytic at $u \rightarrow 0$.
We can therefore define a regularization by mapping any $\log(u)$ to zero, so $\int_{\gamma} w \defas f^{(0)}(0)$ for a path $\gamma\colon (0,1] \longrightarrow \C \setminus \Sigma$ with $\gamma(0) \defas \lim_{t \rightarrow 0} \gamma(0) = 0$.
To make this well-defined, we must fix the branch of $\log(u)$ that we annihilate. This is can be achieved by requiring $\dot{\gamma}(0) = 1$ as initial tangent to $\gamma$, because then $\gamma(t) \in \C \setminus (-\infty,0]$ for sufficiently small $t$ and we can use the principal branch of $\log(\gamma(t))$.

The extra effort necessary to work with such \emph{tangential basepoints} is worthwhile because it simplifies the geometry (as we do not introduce a new distinguished point $\gamma(0) \notin \Sigma$) and therefore the periods that appear. For example we shall recall in example~\ref{ex:moduli-space-mzv} how this technique suffices to prove that all periods of $\Moduli{0}{n}(\R)$ are multiple zeta values \cite{Brown:MZVPeriodsModuliSpaces}, which is not true for ordinary basepoints (that lie in the interior of the moduli space).
\begin{example}
	Say $w_{(1)} = \dd \log (z_1)$ in \eqref{eq:iteratedintegrals-decomposition-limit}, then $\int_{\eta_{\epsilon}} w_{(1)} = \log( z_1) - \log( \eta_\epsilon(0))$ diverges when $\epsilon\rightarrow 0$. The reason is that the limit $\eta_0(0) = \gamma_0(1) = (0, z_2, \ldots, z_n)$ of endpoints does not lie in $X_S$, so it can be singular even for a smooth form on $X_S$.
\end{example}
The generalization to higher dimensions may be stated in the form of
\begin{lemma}
	Suppose the endpoint $\gamma(1) \rightarrow b \in \partial(X_S)$ of $\gamma\colon [0,1] \longrightarrow X_S$ approaches a smooth point $b$ on the boundary of $X_S$, so $f(b)=0$ for a unique polynomial $f \in S$. Then any integrable $w \in \BarObjects(S)$ can develop at worst logarithmic singularities, because
	\begin{equation}
		\int_{\gamma} \!w
		= \sum_k \log^k f(\gamma(1)) \cdot \int_{\gamma} \! w_k
		\label{eq:divergences-barobjects}%
	\end{equation}
	for integrable $w_k \in \BarObjects(S)$ that give iterated integrals $\int_{\gamma} w_k$ which are finite at $\gamma(1) \rightarrow b$.
\end{lemma}
\begin{proof}
	The regularization of lemma~\ref{lemma:shuffle-decomposition} gives unique $w = \sum_k \letter{f}^k \shuffle w'_k$ such that $w'_k$ does not begin with $\letter{f}$. By \eqref{eq:reglim-convolution-formula}, each $w'_k \in \BarObjects(S)$ is integrable because $\BarObjects(S)$ is a Hopf subalgebra (and because each power $\letter{f}^k \in \BarObjects(S)$ is integrable itself). 
	Let $w'_k = \letter{g}u$ ($g \neq f$), then $F \defas \int_{\gamma} u$ diverges only logarithmically at $\gamma(1)\rightarrow b$ by induction, so $\int_{\gamma} w_k' = \int_0^1 (F g)(\gamma(t)) \gamma'(t)\ \dd t$ stays finite in this limit.

	We conclude with \eqref{eq:iint-shuffle} and note that $\int_{\gamma} \letter{f} = \log f(\gamma(1)) - \log f(\gamma(0))$, so we can explicitly set $w_k \defas \sum_{l \geq k} w'_l /[k!(l-k)!] \cdot [-\log f(\gamma(0))]^{l-k}$ to get \eqref{eq:divergences-barobjects}.
\end{proof}
As before we may thus extend the definition of iterated integrals to allow for an endpoint (or base point) $\gamma(1) = b \in \partial(X_S)$ outside of $X_S$ by setting $\int_{\gamma} w \defas \int_{\gamma} w_0$ in \eqref{eq:divergences-barobjects}.
To make this well/defined, we must fix the branch of $\log f(z)$ we annihilate. Again this is conveniently enforced by a tangent condition like $\partial_t f(\gamma(t)) \rightarrow 1$ as $t \rightarrow 1$. We will come back to this in section~\ref{sec:multiple-integrals}.

\section{Properties and algorithms for hyperlogarithms}
\label{sec:hlog-algorithms}%
This section is based on the algorithm for parametric integration presented by Francis Brown \cite{Brown:TwoPoint}, which in turn emerged from his thesis~\cite{Brown:MZVPeriodsModuliSpaces} on the moduli spaces $\Moduli{0}{n}$. The latter is a great reference also for hyperlogarithms themselves. Further comments on the relation between the moduli space setting and the apparently more general (but in fact equivalent) treatment below are made in \cite{Brown:PeriodsFeynmanIntegrals}.

We like to extend the definition~\eqref{eq:def-Hlog-from-to} of hyperlogarithms to allow for the singular base point $\gamma(0) = 0 \in \Sigma$. This amounts to choosing a branch of $\log (z)$, which we take to be the principal one (analytic on $\C \setminus (-\infty,0]$ with $\log 1 = 0$).
\begin{definition}
	\label{def:Hlog}%
	Given a finite set $\Sigma \subset \C$ of points containing $0 \in \Sigma$, any $w \in T(\Sigma)$ has a unique expression of the form $w = \sum_k \letter{0}^k \shuffle w_k$ such that each $w_k$ does not contain any words ending in $\letter{0}$ (by lemma~\ref{lemma:shuffle-decomposition}). The associated \emph{hyperlogarithm} is defined by
\begin{equation}
	\Hyper{w}(z)
	\defas
	\sum_k \frac{\log^k (z)}{k!} \cdot \int_0^z \! w_k
	\quad\text{for}\quad
	w = \sum_k \letter{0}^k \shuffle w_k
	\quad\text{with}\quad
	w_k = \WordReg{0}{}(w_k) 
	\ \forall k.
	\label{eq:def-Hlog}%
\end{equation}
\end{definition}
The iterated integral $\int_{0}^{z} w_k$ is absolutely convergent and analytic at $z\rightarrow 0$ (even if letters $\letter{0}$ appear in $w_k$) as we work out in \eqref{eq:hlog-zsum} below, so \eqref{eq:def-Hlog} makes sense.
In fact, if we take a path $\gamma\colon (0,1] \longrightarrow \C \setminus \Sigma$ from $\gamma(0) = \lim_{t\rightarrow 0} \gamma(t)=0$ to $\gamma(1)=z$, then
\begin{equation*}
	\int_{\restrict{\gamma}{[\epsilon,1]}} \! w
	\urel{\eqref{eq:iint-shuffle}}
	\sum_k \frac{[\log(z)-\log(\gamma(\epsilon))]^k}{k!} \int_{\gamma} \! w_k
	\mapsto
		\Hyper{w}(z)
	\quad\text{when}\quad
	\log(\gamma(\epsilon)), \gamma(\epsilon) \rightarrow 0
\end{equation*}
provided that $\gamma$ does not wind around zero.\footnote{By this we mean that in $\int_{\restrict{\gamma}{[\epsilon,1]}}\!\!\! \letter{0} = \log(z) - \log (\gamma(\epsilon))$, $\log (\gamma(\epsilon))$ denotes the principal branch again.} So in view of section~\ref{sec:tangential-base-points}, $\Hyper{w}(z) = \int_\gamma w$ is a homotopy invariant functional of paths with tangent $\dot{\gamma}(0)=1$. For the examples in figure~\ref{fig:tangent-basepoints} we get $\int_{\gamma_1} \letter{0} = \int_{\gamma_2} \letter{0} = \log(z)$, while $\int_{\gamma_3} \letter{0} = \log(z) - 2\pi\imag$.
\begin{figure}
	\centering
	\includegraphics[width=0.24\textwidth]{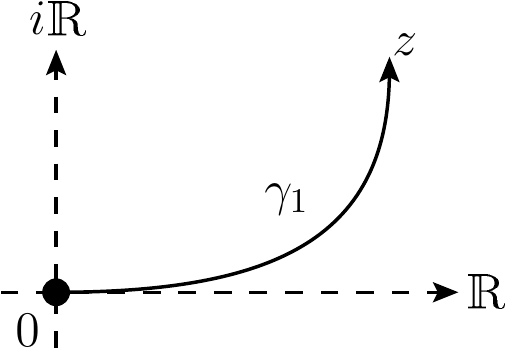}\qquad
	\includegraphics[width=0.3\textwidth]{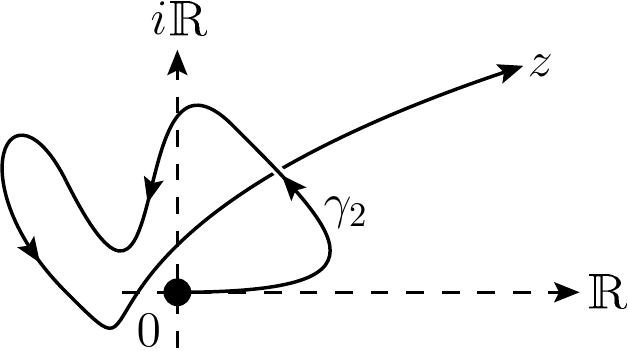}\qquad
	\includegraphics[width=0.27\textwidth]{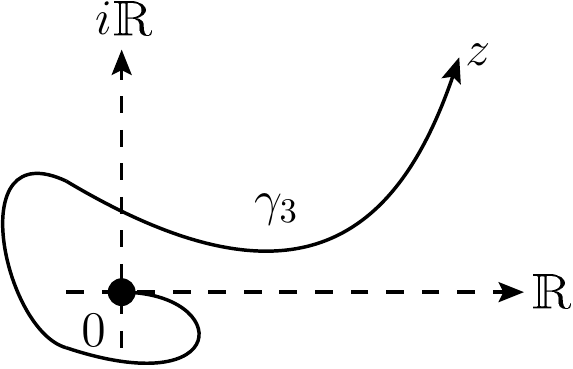}
	\caption[Homotopic and non-homotopic paths with tangential basepoint]{Relative to the endpoint $z=\gamma(1)$ and the tangential basepoint \mbox{$\dot{\gamma}(0) = 1$}, $\gamma_1 \homotop \gamma_2$ are homotopic to each other but not to $\gamma_3$.}%
	\label{fig:tangent-basepoints}%
\end{figure}%

Instead of this dependence on $\gamma$ we consider $\Hyper{w}(z)$ as a multivalued analytic function on $\C \setminus \Sigma$, unambiguously defined by its restriction to the open ball $B_\epsilon(\epsilon)$ around its radius $\epsilon=1/2\cdot\min\setexp{\abs{\sigma}}{0 \neq \sigma\in\Sigma}$. Inside $B_{\epsilon}(\epsilon)$, we take just straight line paths to define $\int_0^z w_k$ in \eqref{eq:def-Hlog} and the principal branch of $\log(z)$.
Typically we will have no positive singularities ($\Sigma \cap \R_+ = \emptyset$) and then each $\Hyper{w}(z)$ admits a unique analytic continuation to the full positive real axis $\R_+$.
\begin{example}
	Hyperlogarithms for the letters $\Sigma = \set{0, 1}$ are known as multiple polylogarithms $\Li_{n_1,\ldots,n_r}(z)$ of a single variable (see section~\ref{sec:polylogarithms}) and acquire path-dependent imaginary parts when $z>1$. We will mostly consider $\Sigma = \set{0,-1}$ instead, so multiple polylogarithms $\Li_{n_1,\ldots,n_r}(-z)$ of negative argument which are analytic on all $z \in \R_+$.
\end{example}

We will abbreviate the linear map $w \mapsto \Hyper{w}$ with $\Hyper{\cdot}$ and its image, the hyperlogarithms with singularities in $\Sigma$, by
\begin{equation}
	\HlogAlgebra(\Sigma) 
	\defas
	\im\left( \Hyper{\cdot} \right)
	= \bigoplus_{w \in \Sigma^{\times}} \Q \cdot \Hyper{w}
	.
	\label{eq:def-HlogAlgebra}%
\end{equation}%
\nomenclature[L(Sigma)]{$\HlogAlgebra(\Sigma)$}{algebra of hyperlogarithms with letters in $\Sigma$, equation~\eqref{eq:def-HlogAlgebra}\nomrefpage}%
The directness of this sum follows from lemma~\ref{lemma:hyperlog-independence} and $\HlogAlgebra(\Sigma)$ is an algebra by
\begin{lemma}
	\label{lemma:Hlog-character}%
	For any $w$, $v \in T(\Sigma)$, $\Hyper{w \shuffle v}(z) = \Hyper{w}(z) \cdot \Hyper{v}(z)$ is multiplicative.
\end{lemma}
\begin{proof}
	The regularizations $w = \sum_k \letter{0}^k \shuffle w_k$ and $v = \sum_{l} \letter{0}^l \shuffle v_l$ imply
\begin{equation*}
	\Hyper{w}(z) \cdot \Hyper{v}(z)
	\urel{\eqref{eq:def-Hlog}}
		\sum_{k,l} \frac{\log^{k+l} z}{k! l!} \int_{0}^k \!\! w_k \cdot \int_{0}^{k} \!\! v_l
	\urel{\eqref{eq:iint-shuffle}}
		\sum_{k,l} \frac{\log^{k+l} z}{k! l!} \int_{0}^k (w_k \shuffle v_l)
	= \Hyper{w \shuffle v}(z)
\end{equation*}
because $w \shuffle v = \sum_{k,l} \letter{0}^k \shuffle \letter{0}^l \shuffle (w_k \shuffle v_l)$ and $\letter{0}^k \shuffle \letter{0}^l = (k+l)!/(k! l!) \cdot \letter{0}^{k+l}$.
\end{proof}
\begin{lemma}
	For any non-empty word $w=\letter{\sigma}w'$ we have $\partial_z \Hyper{w}(z) = \frac{1}{z-\sigma} \Hyper{w'}(z)$.
	\label{lemma:Hlog-differential}%
\end{lemma}
\begin{proof}
	First assume that $w$ does not end with $\letter{0}$, then $\Hyper{w}(z) = \int_0^z w = \int_0^z \frac{\dd x}{x-\sigma} \int_0^x w'$ by definition \eqref{eq:def-Hlog-from-to} proves the statement. We further find for any $k$ that
	\begin{equation*}
		\partial_z \Hyper{\letter{0}^k \shuffle w}(z)
		\urel{\eqref{eq:def-Hlog}}
			\frac{\log^{k-1} z}{z (k-1)!} \Hyper{w}(z)
			+ \frac{\log^k z}{k!}\frac{\Hyper{w'}(z)}{z-\sigma} 
		\urel{\eqref{eq:def-Hlog}}
			\frac{\Hyper{\letter{0}^{k-1} \shuffle w}(z)}{z-0}
			+\frac{\Hyper{\letter{0}^k \shuffle w'}(z)}{z-\sigma}
		= R\left(\letter{0}^k \shuffle w \right)
	\end{equation*}
	coincides with the linear map $R$ defined by $R(\letter{\tau}v') \defas \Hyper{v'}(z) / (z-\tau)$ because
	$ \letter{0}^k \shuffle w
	= \letter{0}\big( \letter{0}^{k-1} \shuffle w \big)
	+ \letter{\sigma} \big( \letter{0}^k \shuffle w' \big)
	$ by \eqref{eq:def-shuffle-product}. Applying this to the regularization $w=\sum_k \letter{0}^k \shuffle w_k$ of an arbitrary word yields the desired result $\partial_z \Hyper{w}(z) = R\big( \sum_k \letter{0}^k \shuffle w_k \big) = R(w)$.
\end{proof}
This result means that we can rephrase the definition~\ref{def:Hlog} of hyperlogarithms as
\begin{equation}
	\Hyper{\letter{0}^n} (z)
	=	\frac{\log^n(z)}{n!}
	\quad\text{and}\quad
	\Hyper{\letter{\sigma}w} (z)
	=	\int_0^z \frac{\dd z'}{z' - \sigma} L_{w} (z')
	\quad\text{when}\quad \letter{\sigma}w \notin \set{\letter{0}}^{\times}
	.
	\label{eq:def-Hlog-noshuffle}%
\end{equation}
In the following our principal goals is to pull back all computations with hyperlogarithms onto the shuffle algebra, thereby rephrasing analytic operations and properties of $\Hyper{w}(z)$ in combinatorial terms of the words $w$. Hence it is important to note
\begin{lemma}%
	\label{lemma:hyperlog-independence}%
	Let $\overline{\C(z)}$ denote the algebraic closure of $\C(z)$, then hyperlogarithms are linearly independent over $\overline{\C(z)}$.
	In other words, the map $\sum_w f_w \tp w \mapsto \sum_w f_w \Hyper{w}$ from $\overline{\C(z)} \tp T(\Sigma)$ into the space of multivalued, locally analytic functions on $\C\setminus \Sigma$ is injective.
\end{lemma}
\begin{proof}
	We assume the opposite and consider the counterexamples
	\begin{equation}
		\Xi
		\defas
		\setexp{
			f: \Sigma^{\times} \longrightarrow \C(z)
		}{
			0<\abs{\supp(f)}<\infty
			\ \text{and}\ 
			\sum_w f_w(z) \cdot \Hyper{w}(z)
			= 0
		}
		\label{eq:independence-proof-counterexamples}%
	\end{equation}
	with rational coefficients $f_w \in \C(z)$ first (we write $\supp(f) \defas \setexp{w}{f_w \neq 0}$). For each $f \in \Xi$, let $n(f) \defas \max \setexp{\abs{w}}{f_w \neq 0} > 0$ denote the highest occurring weight and collect the corresponding words in $W(f) \defas \setexp{\abs{w}=n(f)}{f_w \neq 0} = \supp(f) \cap \C^{n(f)}$.

	Now choose a counterexample $f \in \Xi$ with minimal $n(f)$, and among these select one with minimal $\abs{W(f)}$. We can take any $v \in W(f)$ and divide each $f_w$ by $f_v$ such that we may assume that $f_{v} = 1$. But then the derivative
	\begin{equation}
		0
		=
			\sum_{v \neq w\in W(f)} \left(\partial_z f_{w} \right) \Hyper{w}
			+
			\sum_{w \in W(f)} f_w \partial_z \Hyper{w}(z)
			+
			\sum_{w \in \supp (f) \setminus W(f)} \partial_z \left( f_w \Hyper{w} \right)
		\label{eq:independence-proof-derivative}%
	\end{equation}
	contains at most $n(f)-1$ hyperlogarithms of weight $n(f)$ (since $\partial_z f_v = 0$) and can therefore not be in $\Xi$, considering the minimal choice of $f$. In particular the occurring hyperlogarithms must be linearly independent over $\C(z)$ and we can compare the coefficients in \eqref{eq:independence-proof-derivative} to conclude
	\begin{equation*}
		\partial_z f_w = 0
		\ \text{for all}\ 
		w \in W(f)
		\quad\text{and}\quad
		\partial_z f_w
		=
		- \sum_{\sigma} \frac{f_{\letter{\sigma}w}}{z-\sigma}
		\ \text{for all}\ 
		w \in \supp (f) \setminus W(f).
	\end{equation*}
	But this means that all $0 \neq f_w \in \C$ are constant for $w\in W(f)$ and then $f_w \in -\sum_{\sigma} f_{\letter{\sigma}w} \ln (z-\sigma) + \C$ contradicts the rationality of $f_w(z)$ whenever $w$ is such that $\letter{\sigma}w \in W(f)$ for some $\sigma \in \C$ (i.e.\ $f_{\letter{\sigma}w} \neq 0$).

	Hence the linear independence of hyperlogarithms is proven for rational coefficients $f_w(z) \in \C(z)$. 
	As a consequence, we find that hyperlogarithms $f = \sum_w f_w \Hyper{w}$ with $f_w \in \C(z)$ and $n(f)>0$ are transcendental over $\C(z)$: 
	Otherwise there would be $a_0,\ldots,a_N \in \C(z)$ with $\sum_{n=0}^N a_n f^n = 0$ ($a_N \neq 0$). Using lemma~\ref{lemma:Hlog-character} to multiply out $f^N$ via the shuffle product, linear independence implies that all words of maximum weight $N\cdot n(f)$ must yield a vanishing contribution, leading to the contradiction $a_N f_n^N = 0$ where $f_n \defas \sum_{w \in W(f)} f_w \Hyper{w} \neq 0$.

	Finally we consider \eqref{eq:independence-proof-counterexamples} with algebraic coefficients $f_w(z) \in \overline{\C(z)}$ and apply exactly the same argument as before. The contradiction is again that for some $w$, we would have $f_w(z) \in -\sum_{\sigma} f_{\letter{\sigma}w} \ln(z-\sigma) + \C$ for some constants $f_{\letter{\sigma}w} \in \C$ (not all of them zero). But as we just remarked, such a logarithm is not an algebraic function.
\end{proof}
This independence is lost when one allows for rational functions $f(z) \in \C(z)$ in the arguments of hyperlogarithms $\Hyper{w}\left( f(z) \right)$. Studying the plethora of functional equations that arise this way is a fundamental topic in the study of polylogarithms. The algorithms we will develop actually allow us to write every such polylogarithm in the above basis of hyperlogarithms.

The crucial path concatenation formula \eqref{eq:path-concatenation} does not apply directly to $\Hyper{w}(z)$ due to the special role of $\letter{0}$. But it turns out to be consistent with tangential base points in
\begin{lemma}
	\label{lemma:hlog-path-concatenation} %
		Let $u,z \in \C \setminus \Sigma$, then
		$
			\Hyper{}(z)
			=
			\int_u^z \convolution \Hyper{}(u)
		$ where $\Hyper{}(z)$ is defined via the concatenation of paths chosen for $\Hyper{}(u)$ and $\int_u^z$. For a particular word this reads
		\begin{equation}
			\Hyper{\letter{\sigma_1}\!\cdots\,\letter{\sigma_n}}(z)
			=
			\sum_{k=0}^n
			\int_u^z \letter{\sigma_1}\!\cdots\letter{\sigma_{k}}
				\cdot
				\Hyper{\letter{\sigma_{k+1}}\!\cdots\,\letter{\sigma_n}}(u)
			.
			\label{eq:hlog-path-concatenation} %
		\end{equation}
\end{lemma}
\begin{proof}
	The claim is identical to \eqref{eq:path-concatenation} when $\letter{\sigma_n} \neq 0$ since then $\Hyper{w}(z) = \int_0^z w$. But
	\begin{equation*}
			\Hyper{\letter{0}}(z)
			= \log (z)
			= \log \left( \frac{z}{u} \right) + \log (u)
			= \int_u^z \letter{0} + \Hyper{\letter{0}}(u)
	\end{equation*}
	holds as well and suffices to prove the identity \eqref{eq:hlog-path-concatenation} for all words $w \in T(\Sigma)$. Here we exploit that both $\int_{u}^z$ and $\Hyper{}(u)$ are characters on the Hopf algebra $T(\Sigma)$ by lemmas~\ref{lemma:iint-shuffle} and \ref{lemma:Hlog-character}, which implies that $\Phi \defas \int_u^z \convolution \Hyper{}(u)$ is a character as well.
	So we let it act on $w = \sum_k \letter{0}^k \shuffle w_k = \sum_k \letter{0}^{\shuffle k} / k! \shuffle w_k$ with $w_k$ not ending in $\letter{0}$ and conclude with
	\begin{equation*}
			\Phi(w)
			=
				\sum_k \frac{\Phi(\letter{0})^k}{k!}
				\cdot
				\Phi(w_k)
			=
			\sum_k \frac{[\Hyper{\letter{0}}(z)]^k}{k!} \cdot \Hyper{w_k}(z)
			=
				\Hyper{\sum_k \letter{0}^k \shuffle w_k}(z) 
			= \Hyper{w}(z). \qedhere
	\end{equation*}
\end{proof}
\begin{remark}
	We can use this result also to express any iterated integral $\int_u^z$ in terms of the hyperlogarithms based at zero: As any character on the Hopf algebra $T(\Sigma)$, the map $w\mapsto\Hyper{w}(u)$ has a unique inverse with respect to the convolution product. It is explicitly given by the antipode \eqref{eq:def-antipode} as
	$ [w \mapsto \Hyper{w}(u)]^{\convolution - 1} =  [w\mapsto \Hyper{S(w)}(u)] $.
	If we multiply both sides of $\Hyper{}(z) = \int_u^z \convolution \Hyper{}(u)$ with this we solve for
	$	\int_u^z = \Hyper{}(z) \convolution \Hyper{S}(u)$. For an individual word,
	\begin{equation*}
		\int_{u}^z \letter{\sigma_1}\!\!\cdots\letter{\sigma_n}
		= \sum_{k=0}^n 
				(-1)^{n-k}
				\Hyper{\omega_{\sigma_1}\!\cdots\,\omega_{\sigma_{k}}}(z) 
				\cdot \Hyper{\omega_{\sigma_{n}}\!\cdots\,\omega_{\sigma_{k+1}}}(u).
	\end{equation*}
	In particular, the algebras of iterated integrals generated by $\int_u^z$ and $\Hyper{\cdot}(z)$ coincide up to the constants given by special values of hyperlogarithms at $u$:
	\begin{equation}
		\int_u^z \!\! T(\Sigma)
		\isomorph
		\HlogAlgebra(\Sigma)(z) \tp \HlogAlgebra(\Sigma)(u)
		,\quad\text{for $u$ fixed:}\quad
		\int_u^z \!\! T(\Sigma) \tp \C
		=
		\HlogAlgebra(\Sigma)(z) \tp \C
		.
		\label{eq:hlogalgebras-equivalent-over-C}%
	\end{equation}
\end{remark}
\begin{definition}
	\label{def:punctured-regulars}%
	Given $\Sigma \subset \C$, we define the regular functions on $\C\setminus\Sigma$ as
	\begin{equation}
		\regulars( \Sigma)
		\defas
		\Q\Big[
			z,
			\frac{1}{z-\sigma}\colon \sigma \in \Sigma
		\Big].
		\label{eq:punctured-regulars}%
	\end{equation}
	In order to perform partial fraction decompositions, we extend this algebra to
	\begin{equation}
		\regulars^{+}(\Sigma)
		\defas
		\regulars(\Sigma)
			\tp
		\Q\Big[
			\sigma_i,
			\frac{1}{\sigma_i-\sigma_j}\colon \sigma_i, \sigma_j \in \Sigma 
				\ \text{and}\ 
				\sigma_i \neq \sigma_j
		\Big]
		.
		\label{eq:punctured-regulars+}%
	\end{equation}
\end{definition}%
\begin{lemma}\label{lemma:hlog-primitives}
	The algebra $\regulars^{+}(\Sigma) \tp \HlogAlgebra(\Sigma)$ is closed under taking primitives.
\end{lemma}
\begin{proof}
	Any $g \in \regulars^{+}(\Sigma)$ decomposes uniquely into partial fractions of the form
	\begin{equation}
		g(z)
		=
		\sum_{\sigma \in \Sigma}
		\sum_{n \in \N}
		\frac{A_{\sigma,n}}{(z-\sigma)^n}
		+
		\sum_{n \in \N_0} A_n z^n,
	\end{equation}
	so to integrate $f(z) = g(z) \Hyper{w}(z)$ we can apply partial integration formulae
	\begin{align}
		\int \frac{\Hyper{w}(z)}{(z-\sigma)^{n+1}} \dd z
		&=
		-\frac{\Hyper{w}(z)}{n(z-\sigma)^n} + \int \frac{\partial_z \Hyper{w}(z)}{n(z-\sigma)^n} \dd z
		\quad\text{for any $\sigma\in\Sigma$ and $n \in \N$},
		\label{eq:hlog-primitive-nenner}%
		\\
		\int z^n \Hyper{w}(z) \dd z
		&= \frac{z^{n+1}}{n+1} \Hyper{w}(z) - \int \frac{z^{n+1} \partial_z \Hyper{w}(z)}{n+1} \dd z
		\quad\text{for any $n \in \N_0$.}
		\label{eq:hlog-primitive-zaehler}%
	\end{align}
	This recursion terminates because in each step the weight of the hyperlogarithm is reduced. The only remaining case is lemma~\ref{lemma:Hlog-differential}:
	$ \Hyper{\letter{\sigma}w}(z) $ is a primitive of $\frac{\Hyper{w}(z)}{z-\sigma}$.
\end{proof}
	
\subsection{Divergences and analytic properties}
Any hyperlogarithm $\Hyper{w}(z)$ with $w \in T(\Sigma)$ is locally analytic on $\C \setminus \Sigma$. In this section we study the behaviour near the singular points $\Sigma$ as summarized in
\begin{proposition}
	\label{prop:hlog-divergences}%
	Given any $w \in T(\Sigma)$ and a point $\tau \in \Sigma \cup \set{\infty}$, the hyperlogarithm
	\begin{equation}
		\Hyper{w}(z) = \sum_{k=0}^N f_{w,\tau}^{(k)} (z) 
		\cdot
		\begin{cases}
			\log^k(z-\tau) & \text{when $\tau \neq \infty$ and} \\
			\log^k(z) & \text{when $\tau = \infty$} \\
		\end{cases}
		\label{eq:hlog-divergences} %
	\end{equation}
	admits a unique expansion into hyperlogarithms $f_{w,\tau}^{(k)}(z)$ that are analytic at $z \rightarrow \tau$.\footnote{For $\tau=\infty$ this means that $f_{w,\infty}^{(k)}(1/z)$ is analytic at $z \rightarrow 0$.} In particular, the divergences at $z \rightarrow \tau$ are logarithmic. The number $N$ is bounded by the length of the longest word occurring in $w$. More precisely, for $w = \letter{\sigma_1}\ldots\letter{\sigma_n}$ and $\tau \neq \infty$, we have $N$ at most the maximum number of consecutive letters $\letter{\tau}$ that appear in $w$.
\end{proposition}
We will prove this result and moreover show how the functions $f_{w,\tau}^{(k)}$ in \eqref{eq:hlog-divergences} can be computed explicitly. This logarithmic nature of the divergences motivates the
\begin{definition}
	\label{def:reglim} %
	Suppose that near $\tau \in \C \cup \set{\infty}$, the function $f(z)$ is a polynomial in $\log(z-\tau)$ with coefficients that are Laurent series in $z$. So for some $M \in \N$,
	\begin{equation}\begin{split}
		f(z)
		&= 
		\sum_{n=0}^{N} \log^n (z-\tau) \sum_{m=-M}^{\infty} (z-\tau)^m A_{n,m}
		\quad\text{in the case $\tau \in \C$ and}
		\\
		f(z)
		&= 
		\sum_{n=0}^{N} \log^n(z) \sum_{m=-M}^{\infty} \left( \frac{1}{z} \right)^m A_{n,m}
		\quad\text{when $\tau = \infty$.}
		\label{eq:log-laurent-expansion} %
	\end{split}\end{equation}
	Then the \emph{regularized limit} of $f(z)$ at $z \rightarrow \tau$ is
	\begin{equation}
		\AnaReg{z}{\tau} f(z)
		\defas A_{0,0}.
		\label{eq:def-reglim} %
	\end{equation}
\end{definition}
Note that the expansion \eqref{eq:log-laurent-expansion} is unique (after fixing the branch of the logarithms) and the regularized limit therefore well-defined. If we multiply two functions $f$ and $g$ that both have at worst logarithmic singularities at $z \rightarrow \tau$ (so we can set $M=0$), we can multiply their expansions and find
	\begin{equation}
		\AnaReg{z}{\tau} \left( f \cdot g \right)
		=
		\AnaReg{z}{\tau} (f)
		\cdot
		\AnaReg{z}{\tau} (g)
		\quad\text{when $M \geq 0$ in \eqref{eq:log-laurent-expansion} for both $f$ and $g$.}
		\label{eq:reglim-product} %
	\end{equation}
	In particular this applies to all hyperlogarithms $f,g \in \HlogAlgebra(\Sigma)(z)$ and therefore the map 
	$
		\AnaReg{z}{\tau}
		\colon
		\HlogAlgebra(\Sigma)(z)
		\longrightarrow
		\C
	$
	is a morphism of algebras. Below we will compute this map, strictly speaking its pull-back onto the tensor algebra $T(\Sigma)$, in terms of the combinatorial operator $\WordReg{}{\tau}$ from~\eqref{eq:def-shuffle-regularization}.

	Furthermore note that for any $f(z)$ in \eqref{eq:log-laurent-expansion}, the limit $z \rightarrow \tau$ is finite if and only if
	\begin{equation}
		A_{n,m} = 0
		\quad
		\text{whenever $n>0$ or $m<0$}.
		\label{eq:loglaurent-finiteness}%
	\end{equation}
	As a consequence,
	\begin{equation}
		\AnaReg{z}{\tau} f(z)
		=
		\lim_{z \rightarrow \tau} f(z)
		\quad
		\text{whenever the limit exists}.
		\label{eq:lim=reglim}%
	\end{equation}
	\begin{remark}[Tangential base points]\label{rem:AnaReg-tangential-basepoints}%
	For hyperlogarithms $f(z) = \Hyper{w}(z) = \int_0^z w$, our definition~\eqref{eq:def-reglim} is consistent with section~\ref{sec:tangential-base-points}:
		$\AnaReg{z}{\tau} \Hyper{w}(z) = \int_0^{\tau} w$
	if we define the right-hand side as in \eqref{eq:divergences-barobjects} via annihilation of $\log (z-\tau)$ when $\tau\neq\infty$.
	Note that here and in \eqref{eq:log-laurent-expansion} we fix the principal branch of the logarithm, so we think of letting $z\rightarrow \tau$ approach the singularity from the right such that $(z-\tau) \in \C \setminus (-\infty,0]$ gives a well-defined $\Imaginaerteil \log(z-\tau) \in (-\pi,\pi)$.
	In the case $\tau=\infty$, the limit $z\rightarrow \infty$ is understood along the positive real axis, so $\log(z)$ in \eqref{eq:log-laurent-expansion} is on the principal branch and real-valued.
	For example, $\AnaReg{z}{\infty} \log(2 z) = \log 2$ and $\AnaReg{z}{\infty} \log(-z) = \pm \imag \pi$ depends on whether $\log(-z)$ is defined by continuation from $z=-1$ into the lower or the upper half-plane. 
	\end{remark}

\subsubsection{Behaviour near zero}
\begin{definition}
	\label{def:zsum}%
	For any integers $n$, $r$, $n_1$,\ldots, $n_r \in \N$ and arbitrary numbers $z_1$,\ldots,$z_r \in \C$ we have an associated \emph{Z-sum} \cite{MochUwerWeinzierl:NestedSums} given as\footnote{Beware that our order of the arguments $n_i$ and $z_i$ in \eqref{eq:zsum} is reversed with respect to \cite{MochUwerWeinzierl:NestedSums}.}
	\begin{equation}
		\Zsum{n}{n_1,\ldots,n_r}{z_1,\ldots,z_r}
		\defas
		\sum_{0< k_1 < \cdots<k_r \leq n} \frac{z_1^{k_1}\!\cdots z_r^{k_r}}{k_1^{n_1}\!\cdots k_r^{n_r}}.
		\label{eq:zsum} %
	\end{equation}
\end{definition}
These sums form a Hopf algebra on their own. They have recently been studied in great detail and powerful tools for their manipulation are available \cite{AblingerBluemleinSchneider:GeneralizedHarmonicSumsAndPolylogarithms}. In a sense this focus on the sums is an alternative viewpoint to study hyperlogarithms. In this thesis however we will not take this route and treat hyperlogarithms as iterated integrals exclusively.
\begin{lemma}
		\label{lemma:hlog-zsum} %
	For any word $w$ that does not end in $\letter{0}$, the hyperlogarithm $\Hyper{w}(z)$ is analytic at $z \rightarrow 0$. If
$
	w=\letter{0}^{n_r-1}\letter{\sigma_r}\!\cdots\letter{0}^{n_1-1}\letter{\sigma_1}
$
with non-zero $\sigma_i \in \C$ and $n_i\in \N$, then
	\begin{equation}
		\Hyper{w}(z)
		= (-1)^r
			\sum_{k=r}^{\infty} \frac{\left(z/\sigma_r\right)^k}{k^{n_r}}
			\Zsum{k-1}{n_{1},\ldots,n_{r-1}}{\frac{\sigma_2}{\sigma_{1}},\ldots,\frac{\sigma_r}{\sigma_{r-1}}}
		.
 		\label{eq:hlog-zsum} %
	\end{equation}
\end{lemma}
\begin{proof}
	As $\Hyper{w}(z) = \int_0^z w$ is an iterated integral \eqref{eq:def-Hlog-from-to}, it suffices to apply
\begin{equation*}
	\int_0^z \frac{\dd x}{x} \sum_n a_n x^n
	= \sum_{n=1}^{\infty} \frac{a_n}{n} z^n
	\quad\text{and}\quad
	\int_0^z \frac{\dd x}{x-\sigma_k} \sum_n a_n x^n
	= -\frac{1}{\sigma} \sum_{n,m_k=0}^{\infty} \frac{a_n z^{n+m_k+1}}{(n+m_k+1) \sigma_k^{m_k}}
\end{equation*}
	repeatedly as given by the form of $w$. We start with the constant $\Hyper{\emptyWord}(z) = 1$ and find
	\begin{equation*}
		\Hyper{w} (z)
		=
		(-1)^r \hspace{-3ex}
		\sum_{m_1,\ldots,m_r \geq 0}
		\frac{z^{m_1+1 + m_2+1 + \ldots + m_r+1} \sigma_1^{-(m_1+1)} \cdots \sigma_r^{-(m_r+1)}}
				{ (m_1+1)^{n_1} (m_1+1+m_2+1)^{n_2} \cdots (m_1+1+\ldots+m_r+1)^{n_r} }
		.
	\end{equation*}
	This is the Z-sum \eqref{eq:zsum} with $k_i = (m_1+1)+\ldots+(m_i+1)$ and $z_i = \sigma_{i+1}/\sigma_i$.
\end{proof}
In particular note that $\Hyper{w}(z)$ vanishes at zero with order $r$, the number of letters in $w$ that are different from $\letter{0}$.
The regularization of an arbitrary word $w = \sum_k \letter{0}^k \shuffle w_k$ ($w_k$ not ending on $\letter{0}$) in definition~\ref{def:Hlog} makes the claim of proposition~\ref{prop:hlog-divergences} explicit:
\begin{equation}
	\Hyper{w}(z)
	=
	\sum_k \log^k(z) \cdot f^{(k)}_{w,0}(z)
	\quad\text{where}\quad
	f_{w,0}^{(k)}(z)
	= \frac{\Hyper{w_k}(z)}{k!}
	\quad\text{is analytic at $z\rightarrow 0$.}
	\label{eq:hlog-zero-logseries} %
\end{equation}
\begin{corollary}
	\label{corollary:hlog-reg0}%
	$L_w(z)$ is convergent at $z\rightarrow 0$ if and only if $w$ does not contain any of the words $\setexp{\letter{0}^n}{n\in\N}$. The limit is then $\lim_{z \rightarrow 0} \Hyper{w}(z) = \counit(w)$ is the coefficient of the empty word in $w$. More generally we have $\AnaReg{z}{0} \Hyper{w}(z) = \counit(w)$ for all $w \in T(\Sigma)$.
\end{corollary}

\subsubsection{Singularities at non-zero letters}
Even though a single hyperlogarithm $\Hyper{w}(z)$ (of a word $w \in \Sigma^{\times}$) is in general locally analytic only on $\C \setminus \Sigma$, it has at most a single point of divergence in $\Sigma$ by
\begin{lemma}
	\label{lemma:hlog-single-divergence}%
	Any hyperlogarithm $\Hyper{w}(z)$ with $w = \letter{\sigma_1}\!\!\cdots\letter{\sigma_n}$ has a finite limit when $z \rightarrow \sigma \neq \sigma_1$ (which may depend on the homotopy class of the path), provided $z$ approaches $\sigma$ in a way such that $\arg(z-\sigma)$ is bounded ($z$ shall not wind infinitely often around $\sigma$).
\end{lemma}
\begin{proof}
	For all words that involve only one letter (e.g.\ $w = \letter{0}^n$), there is nothing to prove. We use an induction over the length $n = \abs{w}$ of $w$ in the other cases: Suppose that $w = \letter{\sigma_1}w'$, then we decompose $w' = \sum_k \letter{\sigma}^k \shuffle w_k$ such that each $w_k$ does not begin with $\letter{\sigma}$. By induction we know that $\Hyper{w_k}(z)$ is finite at $z \rightarrow \sigma$, hence the integrand in
	\begin{equation*}
		\Hyper{w}(z)
		\urel{\eqref{eq:def-Hlog-noshuffle}}
		\int_0^z
		\frac{\dd z'}{z'-\sigma_1}
		\Hyper{w'}\left( z' \right)
		=
		\int_0^z \frac{\dd z'}{z'-\sigma_1} \sum_{k=0}^{\abs{w'}} \frac{\log^k(z'-\sigma)}{k!} \Hyper{w_k}(z')
	\end{equation*}
	is dominated near $z' \rightarrow \sigma \neq \sigma_1$ by some constant times $\log^{\abs{w'}}(z'-\sigma)$. But this integrates to $(z-\sigma)\sum_{k=0}^{\abs{w'}} \frac{(-1)^k \abs{w'}!}{(\abs{w'}-k)!} \log^{\abs{w'}-k}(z-\sigma)$ which has a finite limit (zero) at $z\rightarrow \sigma$.
\end{proof}
\begin{example}
	The multiple polylogarithm
	$\Li_{2,1} (z)
	= \Hyper{w} (z)
	$ with	$w= \letter{1}\letter{0}\letter{1}$ has a logarithmic divergence at $z \rightarrow 1$. For $w = \letter{1} \shuffle \letter{0}\letter{1}
	- 2 \letter{0}\letter{1}\letter{1}
	$ shows $\Li_{2,1}(z) = -\log(1-z) \Li_2(z) - 2 \Li_{1,2}(z)$ where $\Li_2(1) = \mzv{2}$ and
	$\Li_{1,2}(1) = \mzv{1,2} = \mzv{3}$ stay finite.
\end{example}
We can now extend the path concatenation lemma~\ref{lemma:hlog-path-concatenation} to singular base points.
\begin{lemma}
	\label{lemma:hlog-path-concatenation-singular} %
	For $z \in \C \setminus \Sigma$ and $\sigma \in \Sigma$ we have 
	$
		\Hyper{}(z)
		=
		\int_{\sigma}^z \convolution \Hyper{\WordReg{}{\sigma}} (\sigma)
	$, explicitly
	\begin{equation}
		\Hyper{\letter{\sigma_1}\!\cdots\,\letter{\sigma_n}}(z)
		=
		\sum_{k=0}^n
			\int_{\sigma}^z \letter{\sigma_1}\!\!\cdots\letter{\sigma_k}
			\cdot
			\Hyper{\WordReg{}{\sigma}(\letter{\sigma_{k+1}}\!\cdots\,\letter{\sigma_n})}(\sigma)
		.
		\label{eq:hlog-path-concatenation-singular} %
	\end{equation}
	Through the regularization $\WordReg{}{\sigma}$ and lemma~\ref{lemma:hlog-single-divergence}, the second factor is finite. But the iterated integral $\int_{\sigma}^z \letter{\sigma_1}\!\!\cdots\letter{\sigma_k}$ is divergent when $\sigma_k = \sigma$. Instead we define
		\begin{equation}
			\int_{\sigma}^z \letter{\sigma}
			\defas
			\Hyper{\letter{\sigma}}(z)
			\label{eq:hlog-singular-concatenation-branch}%
		\end{equation}
		as the branch of $\log (z-\sigma)$ determined by the original integration path and extend this to a character. This means $\int_{\sigma}^z w = \sum_k \Hyper{\letter{\sigma}^k}(z) \cdot \int_{\sigma}^z (w_k)$ if $w=\sum_k \letter{\sigma}^k \shuffle w_k$ is regularized such that $w_k$ does not end in $\letter{\sigma}$.
\end{lemma}
\begin{proof}
	Let $\tau \in \C \setminus \Sigma$ denote a point in the vicinity of $\sigma$ where we split the integration $\Hyper{}(z) = \int_{\tau}^z \convolution \Hyper{}(\tau)$ according to \eqref{eq:hlog-path-concatenation}, see figure~\ref{fig:path-concatenation-singular}. 
	In a word $w$ we now focus on a sequence of the letter $\letter{\sigma}$, so say $w = u \big(\letter{\sigma}^n\big) v$ where $u \in \im\left( \WordReg{\sigma}{} \right)$ does not end in $\letter{\sigma}$ and $v \in \im (\WordReg{}{\sigma})$ does not begin with $\letter{\sigma}$.
	From corollary~\ref{corollary:shuffle-regularization} we know that
	\begin{equation*}
		u \big( \letter{\sigma}^k \big)
		=
		\sum_{\mu=0}^{k} \letter{\sigma}^{k-\mu} \shuffle \WordReg{\sigma}{} \left( u \letter{\sigma}^\mu \right)
		\quad\text{and}\quad
		\big( \letter{\sigma}^k \big) v
		=
		\sum_{\nu=0}^{k} \letter{\sigma}^{k - \nu} \shuffle \WordReg{}{\sigma} \left( \letter{\sigma}^{\nu} v \right),
	\end{equation*}
	which we use to rewrite those summands of $\Hyper{w}(z) = \sum_{(w)} \int_{\tau}^z w_{(1)} \cdot \Hyper{w_{(2)}}(\tau)$ that split $w=w_{(1)}w_{(2)}=u\big(\letter{\sigma}^n\big)v$ between $u$ and $v$ as follows:
	\begin{equation*}
		\sum_{k=0}^{n}
			\int_{\tau}^z u \letter{\sigma}^k
			\cdot
			\Hyper{\letter{\sigma}^{n-k}v}(\tau)
		=
		\sum_{\mu + \nu + a + b = n}
			\int_{\tau}^z \WordReg{\sigma}{}\left( u \letter{\sigma}^{\mu} \right)
			\cdot
			\Hyper{\WordReg{}{\sigma} \left( \letter{\sigma}^{\nu} v \right)}(\tau)
			\cdot
			\int_{\tau}^z \letter{\sigma}^a
			\cdot
			\Hyper{\letter{\sigma}^b}(\tau)
		.
	\end{equation*}
	The sum over $a+b = n - \mu - \nu$ of the last two terms reduces to $\Hyper{\letter{\sigma}^{n-\mu-\nu}}(z)$ and the other two factors are both finite in the limit $\tau \rightarrow \sigma$. So with \eqref{eq:hlog-singular-concatenation-branch}, the above becomes
	\begin{equation*}
		\sum_{\mu + \nu \leq n}
			\int_{\sigma}^{z} \WordReg{\sigma}{}\left( u \letter{\sigma}^{\mu} \right)
			\cdot
			\int_{\sigma}^{z} \letter{\sigma}^{n-\mu-\nu}
			\cdot
			\Hyper{\WordReg{}{\sigma} \left( \letter{\sigma}^{\nu} v \right)}(\sigma)
		=
		\sum_{k=0}^{n}
			\int_{\sigma}^{z} \left( u \letter{\sigma}^{k} \right)
			\cdot
			\Hyper{\WordReg{}{\sigma}\left( \letter{\sigma}^{n-k} v \right)}(\sigma)
		.
	\end{equation*}
	Add these identities for all the sequences $\letter{\sigma}^{n_k}$ in
	$w = \letter{\sigma}^{n_0} \letter{\sigma_{i_1}} \letter{\sigma}^{n_1} \letter{\sigma_{i_2}} \!\!\!\cdots \letter{\sigma_{i_r}} \letter{\sigma}^{n_r}$
	(each $\sigma_{i_k} \neq \sigma$), this covers the full coproduct of $w$ to conclude
	\begin{equation*}
		\sum_{(w)}
		\int_{\sigma}^z w_{(1)} \cdot \Hyper{\WordReg{}{\sigma}\left( w_{(2)} \right)}(\sigma)
		= \lim_{\tau\rightarrow \sigma} \ 
			\sum_{(w)} \int_{\tau}^z w_{(1)} \cdot \Hyper{w_{(2)}}(\tau)
		\urel{\eqref{eq:hlog-path-concatenation}}
			\lim_{\tau\rightarrow \sigma} 
			 \Hyper{w}(z)
		= \Hyper{w}(z). \qedhere
	\end{equation*}
\end{proof}
\begin{figure}
	\centering
	\includegraphics[width=0.4\textwidth]{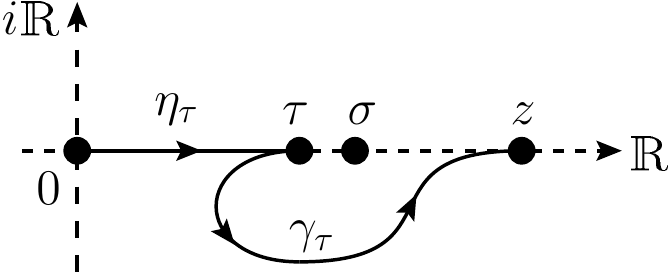}
	\caption[Path concatenation at a singular intermediate point]{The path concatenation \eqref{eq:hlog-path-concatenation-singular} for a singular intermediate point $\sigma \in \Sigma$ is obtained as the limit $\tau \rightarrow \sigma$. The integration path from zero to $z$ is assumed to be homotopic to $\eta_{\tau} \concat \gamma_{\tau}$ and determines the branch of $\log(z-\sigma)$.}%
	\label{fig:path-concatenation-singular} %
\end{figure}
This settles \eqref{eq:hlog-divergences}, because $\int_{\sigma}^z w$ is analytic at $z \rightarrow \sigma$ (with vanishing limit) when $w$ does not end on $\letter{\sigma}$ and we can only have logarithmic divergences from $\int_{\sigma}^{z} \letter{\sigma} = \Hyper{\letter{\sigma}}(z) = \log \frac{z-\sigma}{-\sigma}$. So $\lim_{z \rightarrow \sigma} \int_{\sigma}^z w = 0$ whenever $w$ is not of the form $\letter{\sigma}^k$ and thus
	\begin{equation}
		\AnaReg{z}{\sigma} \Hyper{\letter{\sigma}^n w}(z)
		= \sum_{k=0}^{n} \frac{\log^k(-1/\sigma)}{k!}L_{\WordReg{}{\sigma}(\letter{\sigma}^{n-k}w)}(\sigma)
		\label{eq:hlog-reglim-singularity}%
	\end{equation}
	if $w$ does not begin with $\letter{\sigma}$. In order to compute the full expansion \eqref{eq:log-laurent-expansion} at $z \rightarrow \sigma$ it is handy to transform the hyperlogarithm $\int_{\sigma}^z$ into the basis $\Hyper{\cdot}(z-\sigma)$ using
\begin{lemma}[Change of variables]
	\label{lemma:Moebius-transformation} %
	Let $f(z) = \frac{az + b}{cz + d}\in\Aut(\RSphere)$ denote a M\"{o}bius transformation and $\gamma$ a path such that $\int_{\gamma} w$ converges for a given word $w\in T(\Sigma)$.
	
	Then
	$
		\int_{\gamma} w
		= \int_{f \circ \gamma} \WordTransformation{f}(w)
	$
	where $\WordTransformation{f}\colon T(\Sigma) \longrightarrow T(f(\Sigma))$ substitutes letters according to
	\begin{equation}
		\WordTransformation{f}(\letter{\sigma}) 
		\defas
		\letter{f(\sigma)} - \letter{f(\infty)},
		\quad\text{setting any}\quad
		\letter{\infty}
		\defas
		0.
		\label{eq:Moebius-transformation}%
	\end{equation}
\end{lemma}
\begin{proof}
	This is just the transformation of \eqref{eq:def-II} under $f$ and the identity
\begin{equation*}
	\frac{\dd f^{-1}(z)}{f^{-1}(z) - \sigma}
	=
	\frac{\dd z}{z-f(\sigma)} - \frac{\dd z}{z-f(\infty)}.\qedhere
\end{equation*}
\end{proof}
Note that $\Aut(\RSphere) \ni f \mapsto \WordTransformation{f} \in \Aut\left( T(\C) \right)$ furnishes a representation of the automorphisms of the Riemann sphere $\RSphere$: $\WordTransformation{f}^{-1} = \WordTransformation{f^{-1}}$ and $\WordTransformation{f \circ g} = \WordTransformation{f} \circ \WordTransformation{g}$.
\begin{example}
	\label{ex:dilog-past-one}%
	For the dilogarithm $\Hyper{\letter{0}\letter{1}}(z) = - \Li_2(z)$, the decomposition \eqref{eq:hlog-path-concatenation-singular} reduces with $\WordReg{}{1}(\letter{1}) = 0$, $\letter{0}\letter{1} = \letter{0} \shuffle \letter{1} - \letter{1}\letter{0}$ and $\Li_2(1) = \mzv{2}$ to
	\begin{equation*}
		\Hyper{\letter{0}\letter{1}}(z)
		= \int_1^z \!\letter{0}\letter{1} 
		+ \int_1^z \!\letter{0} \cdot \Hyper{\WordReg{}{1}(\letter{1})}(1)
		+ \Hyper{\WordReg{}{1}(\letter{0}\letter{1})}(1)
		= -\mzv{2}
		+ \int_1^z \!\!\letter{0}\, \cdot \int_1^z \!\!\letter{1}
		- \int_1^z \!\!\letter{1} \letter{0}
		.
	\end{equation*}
	We can apply the transformation $f(z)=1-z$ to the iterated integrals $\int_1^z \letter{0}$ and $\int_1^z \letter{1}\letter{0}$, while $\int_1^z \letter{1}$ is defined as $\Hyper{\letter{1}}(z) = \log(1-z)$. Finally this shows
	\begin{equation*}
		\Li_2(z)
		= \mzv{2} + \int_0^{1-z} \letter{0}\letter{1} - \log(1-z) \int_0^{1-z} \letter{1}
		= \mzv{2} - \Li_2(1-z) - \log(z) \cdot \log(1-z).
	\end{equation*}
	For real $z>1$, the branch of $\log(1-z) = \Hyper{\letter{1}}(z)$ is determined by the path of integration. If it passes below $1=\sigma$ as in figure~\ref{fig:path-concatenation-singular}, we have 
	\begin{equation*}
		\log(1-z)
		= \imag \pi + \log (z-1)
		\quad\text{and}\quad
		\Imaginaerteil \Li_2(z)
		= - \pi \log(z)
		\quad\text{for}\quad
		z>1.
	\end{equation*}
\end{example}
In general the analytic continuation past a positive singularity $\sigma>0$ introduces imaginary parts from $\eqref{eq:hlog-singular-concatenation-branch}$, but it may happen that the hyperlogarithm at hand is analytic at $\sigma$ (for example this is granted for Feynman integrals in the Euclidean region). In this case, no explicit imaginary parts occur:
\begin{lemma}\label{lemma:singular-expansion-analytic}
	For $w\in T(\Sigma)$, the hyperlogarithm $\Hyper{w}(z)$ is analytic at $z \rightarrow \sigma \in \Sigma$ if and only if \eqref{eq:hlog-path-concatenation-singular} is equal to $\Hyper{w}(z) = \int_{\sigma}^z \WordReg{\sigma}{} \convolution \Hyper{\WordReg{}{\sigma}(\cdot)}(\sigma)$.
\end{lemma}
\begin{proof}
	As an iterated integral, $\int_{\sigma}^z \WordReg{\sigma}{}$ is analytic. Conversely, analyticity of $\Hyper{w}(z)$ at $z\rightarrow \sigma$ implies the vanishing (identically, not only at $z = \sigma$) of all coefficients of powers of $\log (z-\sigma)$, which equals $\int_{\sigma}^z \letter{\sigma}$ up to a constant. So if
	$	\letter{\sigma_1}\!\!\cdots\letter{\sigma_k}
		= \sum_{i = 0}^{k} \letter{\sigma}^i \shuffle w_{k,i}
	$
	is the regularization with all $w_{k,i}$ not ending in $\letter{\sigma}$, we see from
	\begin{equation*}
		\Hyper{\letter{\sigma_1}\!\cdots\letter{\sigma_n}}(z)
		= \sum_{i=0}^n \Hyper{\letter{\sigma}^i}(z) \sum_{k=0}^{n-i} \lambda_k \int_{\sigma}^z w_{k,i}
		\quad\text{with}\quad
		\lambda_k \defas \Hyper{\WordReg{}{\sigma}(\letter{\sigma_{k+1}\!\!\cdots\sigma_{n}})}(\sigma),
	\end{equation*}
	that $\sum_k \lambda_k \int_{\sigma}^z w_{k,i} = 0$ must vanish for all $i>0$. The linear independence of lemma~\ref{lemma:hyperlog-independence} carries over to $\int_{\sigma}^z$ via the shift $f(z) = z-\sigma$, hence we must have $\sum_k \lambda_k w_{k,i} = 0$. Thus $\Hyper{w}(z)$ collapses to $\sum_k \lambda_k \int_{\sigma}^z w_{k,0}$ where $w_{k,0} = \WordReg{\sigma}{}(\letter{\sigma_1}\!\cdots\letter{\sigma_k})$.
\end{proof}

\subsubsection{Expansion at infinity}
We could use a M\"{o}bius transformation to translate our result for singular points $\sigma \in \Sigma$ to infinity, but it is sensible to study this case separately in detail because it will appear at every step in the integration algorithm we have in mind.

The divergence at $z \rightarrow \infty$ is of course only logarithmic:
Consider a straight path $\gamma$ from $\gamma(0) = u \in \R_+$ to $\gamma(1) = \infty$ on the positive real axis such that $\im \gamma \cap \Sigma = [u,\infty) \cap \Sigma = \emptyset$. Then on all of $\gamma$, we can bound $1/\abs{z-\sigma} \leq C/z$ for any $\sigma \in \Sigma$ by
	\begin{equation*}
		C
		\defas 1 + \frac{\max \setexp{\abs{\sigma}}{\sigma \in \Sigma}}{\dist(\gamma, \Sigma)},
		\quad\text{therefore}\quad
		\abs{\int_u^z \letter{\sigma_1}\!\!\cdots\letter{\sigma_n}}
		\leq 
			C^n
			\int_u^z \letter{0}^n
		= \frac{C^n}{n!} \log^n\left( \frac{z}{u} \right)
		.
	\end{equation*}
	Via path concatenation \eqref{eq:hlog-path-concatenation} this logarithmic bound extends to all hyperlogarithms: For large $z$,
$ \abs{\Hyper{w}(z)} 
	\leq \sum_{(w)} \big|\int_u^z w_{(1)}\big| \cdot \,\big| \Hyper{w_{(2)}}(u) \big|
	\leq \log^{\abs{w}}(z/u) \cdot C^{\abs{w}}\max_{(w)} \big| \Hyper{w_{(2)}}(u) \big|
$.
\begin{lemma}
	\label{lemma:hlog-infinity-finite}%
	For any $w \in T(\Sigma)$ and $\sigma \in \Sigma$, the hyperlogarithm $\Hyper{(\letter{\sigma}-\letter{-1})w}(z)$ has a finite limit at $z \rightarrow \infty$.
\end{lemma}
\begin{proof}
	The representation \eqref{eq:def-Hlog-noshuffle} yields an absolutely convergent integral\footnote{%
In the special case when $\sigma=0$ and $w = \letter{0}^n$, equation~\eqref{eq:hlog-finite-infinity} does not apply. Instead one calculates explicitly that
$
\Hyper{(\letter{0}-\letter{-1})\letter{0}^n}(z)
= \frac{\ln^{n+1}(z)}{(n+1)!(z+1)}
+ \int_0^z \frac{\ln^{n+1}(z')}{(n+1)!(z'+1)^2} \dd z'
$
which is clearly finite at $z \rightarrow \infty$.}
\begin{equation}
	\Hyper{(\letter{\sigma} - \letter{-1})w}(z)
	=
	\int_0^z \frac{(1+\sigma) \Hyper{w}(z')}{(z'-\sigma)(z' + 1)} \dd z'
	,
	\label{eq:hlog-finite-infinity}%
\end{equation}
because the integrand decays quadratically at large $z'$ (up to at most logarithmic growth of $\Hyper{w}(z')$ in the numerator, which does not spoil convergence).
\end{proof}
\begin{definition}
	\label{def:reginf-word} %
	The map $\WordReg{}{\infty}\colon T(\Sigma) \longrightarrow T(\Sigma \cup \set{-1})$ is given by linear extension of
	\begin{equation}
		\WordReg{}{\infty} \left(\letter{\sigma_1}\!\!\cdots\letter{\sigma_n} \right)
		\defas
		\sum_{k=1}^{n} 
			\left( \letter{\sigma_k} - \letter{-1} \right)
			\left[ (-\letter{-1})^{k-1} \shuffle \letter{\sigma_{k+1}}\!\!\cdots\letter{\sigma_n}\right]
		\label{eq:def-reginf-word} %
	\end{equation}
	and $\WordReg{}{\infty} (\emptyWord) \defas \emptyWord$ for the empty word.
	Furthermore we set $\WordReg{0}{\infty} \defas \WordReg{}{\infty} \circ \WordReg{0}{}$, using the shuffle-regularization $\WordReg{0}{}$ onto words not ending in $\letter{0}$ (from definition~\ref{def:shuffle-regularization}).
\end{definition}%
Note that $\WordReg{}{\infty}( \letter{-1}^n) = 0$ for any $n\in\N$. By lemma~\ref{lemma:hlog-infinity-finite}, all words $w$ in the image of $\WordReg{}{\infty}$ have a finite limit of $\Hyper{w}(z)$ at $z \rightarrow \infty$.
\begin{lemma}
	\label{lemma:reginf-properties} %
	The maps $\WordReg{}{\infty}$ and $\WordReg{0}{\infty}$ are projections and multiplicative, so in particular 
	$
		\WordReg{}{\infty} (w \shuffle w')
		=
		\WordReg{}{\infty} (w) \shuffle \WordReg{}{\infty} (w')
	$.
	Furthermore, any word fulfils the identity
	\begin{equation}
		\letter{\sigma_1}\!\!\cdots\letter{\sigma_n}
		=	\sum_{k=0}^{n} 
			\letter{-1}^k 
			\shuffle \WordReg{}{\infty} \left( \letter{\sigma_{k+1}}\!\!\cdots\letter{\sigma_n}\right)
		.
		\label{eq:reginf-shuffle-identity} %
	\end{equation}
\end{lemma}
\begin{proof}
	Formula \eqref{eq:reginf-shuffle-identity} is evident for $n=1$. When $n>1$, we insert \eqref{eq:def-reginf-word} into its right-hand side and apply the shuffle-product recursion \eqref{eq:def-shuffle-product} to obtain
	\begin{equation*}
		\sum_{\mathclap{0 \leq k < i \leq n}}
		\left( \letter{\sigma_i} - \letter{-1}\right) \left[
			\letter{-1}^k
			\shuffle
			(-\letter{-1})^{i-k-1}
			\shuffle
			\letter{\sigma_{i+1}} \!\!\cdots \letter{\sigma_n}
		\right]
		+
		\letter{-1} \sum_{k=1}^{n} \left[
			\letter{-1}^{k-1}
			\shuffle
			\WordReg{}{\infty} \left( \letter{\sigma_{k+1}}\!\!\cdots \letter{\sigma_n} \right)
		\right]
		.
	\end{equation*}
	By induction the second sum evaluates to $\letter{-1}\letter{\sigma_2}\!\cdots\letter{\sigma_n}$, whereas 
	\begin{equation*}
		\sum_{k=0}^{i-1} \letter{-1}^k \shuffle (-\letter{-1})^{i-k-1}
		=
		\letter{-1}^{i-1}
		\sum_{k=0}^{i-1} \binom{i-1}{k} (-1)^{i-1-k}
		= \delta_{i,1}
	\end{equation*}
	reduces the first sum to $(\letter{\sigma_1}-\letter{-1})\letter{\sigma_2}\!\cdots \letter{\sigma_n}$. So we proved \eqref{eq:reginf-shuffle-identity}, which can be written as $\WordReg{}{\infty} = \Psi^{\convolution -1} \convolution \id$ where $\Psi(\letter{\sigma_1}\!\!\cdots\letter{\sigma_n}) \defas \letter{-1}^n$ denotes the character on $T(\Sigma)$ that replaces all letters by $\letter{-1}$. As a product of characters, $\WordReg{}{\infty}$ is multiplicative as well. Since it annihilates $\letter{-1}^k$, applying $\WordReg{}{\infty}$ to equation~\eqref{eq:reginf-shuffle-identity} proves its idempotence $\WordReg{}{\infty} = \WordReg{}{\infty} \circ \WordReg{}{\infty}$. If instead we first apply $\WordReg{0}{}$ and then $\WordReg{}{\infty}$, we find that $\WordReg{0}{} = \Psi \convolution (\WordReg{0}{}\circ\WordReg{}{\infty})$ and finally $\WordReg{0}{\infty} = \WordReg{0}{\infty} \circ \WordReg{}{\infty}$.
\end{proof}
Now recall that $\Hyper{\letter{-1}}(z) = \log (z+1) = \log (z) + \log \left( 1+z^{-1} \right)$, so by \eqref{eq:def-reglim} it gives $\AnaReg{z}{\infty} \Hyper{\letter{-1}}(z) = \log \left(1+0\right) = 0$ as does $\AnaReg{z}{\infty}\Hyper{\letter{0}}(z) = 0$ as well. Equation~\eqref{eq:reginf-shuffle-identity} and the multiplicativity \eqref{eq:reglim-product} of $\AnaReg{z}{\infty}$ and of $\Hyper{\cdot}(z)$ (lemma~\ref{lemma:Hlog-character}) thus prove
\begin{corollary}
	\label{cor:reginf-from-word} %
	For any word $w\in T(\Sigma)$, the regularized limit of $\Hyper{w}(z)$ at infinity can be expressed in terms of the (absolutely) convergent iterated integral
	\begin{equation}
		\AnaReg{z}{\infty} \Hyper{w}(z)
		=
		\Hyper{\WordReg{0}{\infty}(w)} (\infty)
		.
		\label{eq:reginf-from-word} %
	\end{equation}
\end{corollary}
\begin{example}
	\label{ex:dilog-reginf}%
	From $\WordReg{0}{\infty} (\letter{0}\letter{-1}) = (\letter{0}-\letter{-1})\letter{-1}$ we deduce
	\begin{equation*}
		\AnaReg{z}{\infty}
			\Li_2(-z)
		=-\AnaReg{z}{\infty}
			\Hyper{\letter{0}\letter{-1}}(z)
		= -\Hyper{(\letter{0}-\letter{-1})\letter{-1}}(\infty)
		\urel{\eqref{eq:Moebius-transformation}}
			\int_0^1 \letter{0} \letter{1}
		= - \Li_2 (1)
		= -\mzv{2},
	\end{equation*}
	exploiting a change of variables $z\mapsto \frac{z}{1+z}$ to transform the iterated integral from the interval $[0,\infty]$ to $[0,1]$. It substitutes $\letter{0} \mapsto \letter{0} - \letter{1}$ and $\letter{-1} \mapsto -\letter{1}$ and implies that the regularized limit of $\Hyper{w}(z)$ at infinity is a multiple zeta value for any $w \in T(\set{0,-1})$:
	\begin{align}
		\AnaReg{z}{\infty} \HlogAlgebra(\set{0,-1})(z)
		&\urel{\eqref{eq:reginf-from-word}}
			\int_0^{\infty}\!\! \WordReg{0}{\infty}\left( T(\set{0,-1}) \right)
		\urel{\eqref{eq:def-reginf-word}}
		\Q +
		\int_0^{\infty} \!\! (\letter{0}-\letter{-1}) T(\set{0,-1}) \letter{-1}
		\nonumber\\&
		\urel{\eqref{eq:Moebius-transformation}}
			\Q +
			\int_0^1 \!\! \letter{0} T(\set{0,1}) \letter{1}
			= \Q + \sum_{\mathclap{w \in \set{0,1}^{\times}}} \Q\cdot \int_0^1 \!\! \letter{0}w\letter{1}
		\safed \MZV.
		\label{eq:reginf-gives-mzv}%
	\end{align}
\end{example}
\begin{remark}\label{rem:letter-choice-reginf}
	Corollary~\ref{cor:reginf-from-word} shows how to write the regularized limit $z\rightarrow\infty$ (as defined in definition~\ref{def:reglim}) of a hyperlogarithm as a convergent integral, which we exploit in some proofs below. Our approach to treat $z\rightarrow \infty$ differs from the previous cases in that \eqref{eq:def-reginf-word} adds words (which act as counterterms) with some letters replaced by $\letter{-1}$ in such a way that the hyperlogarithm of the sum stays finite at $z\rightarrow \infty$.

	If $-1 \notin \Sigma$, this approach has the deficit that it introduces an additional letter. In this case we could choose any point of $0\neq\tau\in\Sigma$, replace $-1$ with $\tau$ in \eqref{eq:def-reginf-word}, \eqref{eq:reginf-shuffle-identity} and the only modification to \eqref{eq:reginf-from-word} would be additional terms\footnote{%
		We can even avoid these corrections altogether by changing our definition of $\AnaReg{z}{\infty}$ to annihilate $\Hyper{\letter{\tau}}(z)$. This corresponds to a rescaling of the tangent at the base point at infinity.%
	}
	with explicit powers of $\log(-1/\tau)$, corresponding to contributions to \eqref{eq:reginf-shuffle-identity} where $k>0$.
	
	However, for the integration algorithm itself one does not need to compute $\WordReg{0}{\infty}(w)$ of a word $w$ explicitly in intermediate steps: lemmata~\ref{lem:reginf-differential}, \ref{lem:reginf-reglim-finite}, \ref{lem:reglim-reginf-rescaling} and proposition \ref{prop:reglim-reginf-from-word} stay inside $\Sigma$ or its leading coefficients and do not introduce unnecessary letters.
	Only at the end of the computation we want to express the final result given as iterated regularized limits like \eqref{eq:final-period-algebra} in terms of convergent \emph{period} integrals (see section~\ref{sec:Periods}). 
	Our choice $\tau=-1$ is motivated by the fact that this final alphabet always contains $-1$ in the case of Feynman integrals.\footnote{The reason is that the non-zero coefficients of the graph polynomial $\psi$ are $1$.}
\end{remark}
\begin{lemma}
	\label{lemma:reginf-expansion-convolution} %
	The expansion of $\Hyper{w}(z)$ at $z\rightarrow \infty$ is given by
	\begin{equation}
		\Hyper{\letter{\sigma_1}\!\cdots\,\letter{\sigma_n}}(z)
		= \sum_{k=0}^n
			\Hyper{\WordTransformation{1/z}(\letter{\sigma_1}\!\cdots\,\letter{\sigma_k})} \left(z^{-1}\right)
			\cdot
			\Hyper{\WordReg{0}{\infty}(\letter{\sigma_{k+1}}\!\cdots\,\letter{\sigma_n})}(\infty)
		.
		\label{eq:reginf-expansion-convolution} %
	\end{equation}
\end{lemma}
\begin{proof}
We split the integration path at a point $u$ such that $\Hyper{\cdot}(z) = \int_{u}^{z} \cdot \convolution \Hyper{\cdot}(u)$. Individually, both factors develop (at worst) logarithmic divergences at $u\rightarrow \infty$ which cancel in the total, $u$-independent sum. Therefore we can take the regularized limits
\begin{equation*}
	\Hyper{w}(z)
	\urel{\eqref{eq:reglim-product}}
	\sum_{(w)}
		\AnaReg{u}{\infty} \int_u^z w_{(1)}
		\cdot
		\AnaReg{u}{\infty} \Hyper{w_{(2)}} (u)
	\urel{\eqref{eq:reginf-from-word}}
	\sum_{(w)}
		\AnaReg{u}{\infty} \int_{1/u}^{1/z} \WordTransformation{1/z} \big( w_{(1)} \big)
		\cdot
		\Hyper{\WordReg{0}{\infty}\!\big(w_{(2)}\big)} (\infty)
	.%
\end{equation*}
Here we use lemma~\ref{lemma:Moebius-transformation} to change variables $z \mapsto 1/z$. If $w = \sum_k \letter{0}^k \shuffle w_k$ is written such that $w_k \in \im(\WordReg{0}{})$ does not end on $\letter{0}$, the iterated integrals $\int_0^{1/z} w_k$ are finite and
\begin{equation*}
	\AnaReg{u}{\infty} \int_{1/u}^{1/z}\!\! w
	\urel{\eqref{eq:reglim-product}}
		\sum_k
		\AnaReg{u}{\infty} \int_{1/u}^{1/z} \!\!\letter{0}^k\,
		\cdot
		\int_{0}^{1/z} \!\!w_k
	= \sum_k
		\AnaReg{u}{\infty} \frac{\log^k(u/z)}{k!}
		\cdot
		\Hyper{w_k}\!\left( \frac{1}{z}\right)
	= \Hyper{w}\!\left( \frac{1}{z} \right)
\end{equation*}
follows from $\log(u/z) = \log(u) + \Hyper{\letter{0}}(1/z)$ being mapped to $\Hyper{\letter{0}}(1/z)$ under $\AnaReg{u}{\infty}$.
\end{proof}
This decomposition suffices to compute the expansion \eqref{eq:log-laurent-expansion} at $z\rightarrow \infty$ explicitly for an arbitrary word, because $\Hyper{\WordTransformation{1/z}(w)}(z^{-1})$ is easily expanded using \eqref{eq:hlog-zsum} and \eqref{eq:hlog-zero-logseries}.
\begin{example}
	\label{ex:dilog-inf-expansion} %
	Continuing example~\ref{ex:dilog-reginf}, $\WordReg{0}{\infty}(\letter{-1}) = 0$ gives \eqref{eq:reginf-expansion-convolution} the form
	\begin{equation*}
		\Li_2(-z)
		= - \Hyper{\WordReg{0}{\infty}(\letter{0}\letter{-1})}(\infty)
		  - \Hyper{\WordTransformation{1/z}(\letter{0}\letter{-1})}(1/z)
		= - \mzv{2} -\tfrac{1}{2} \log^2(z) - \Li_2\left( -1/z \right)
	\end{equation*}
	using $\WordTransformation{1/z}(\letter{0}\letter{-1}) = (-\letter{0})(\letter{-1}-\letter{0})$. With \eqref{eq:hlog-zsum}, $\Li_2(-1/z) = \sum_{n=1}^{\infty} (-1/z)^n/n^2$ is analytic at $z \rightarrow \infty$.
\end{example}

\subsection{Dependence of regularized limits}
Considering the letters $\sigma\in\Sigma$ as variables themselves, a hyperlogarithm $\Hyper{w}(z)$ or $\int_{\gamma} w$ is a multivariate function which we so far only studied for fixed $\Sigma$. From the integral representation \eqref{eq:def-Hlog-from-to} we see that $\int_{\gamma} w$ is analytic in all variables $\Sigma \cup \set{\gamma(0),\gamma(1)}$, at least as long as $\dist(\gamma,\Sigma)>0$ (which guarantees absolute convergence).
\begin{lemma}
	\label{lem:hlog-total-differential} %
	The total differential of any hyperlogarithm is
	\begin{equation}\begin{split}
		\dd \Hyper{\letter{\sigma_1}\!\cdots\,\letter{\sigma_n}}(z)
		&= \Hyper{\not\letter{\sigma_1}\!\cdots}(z) \cdot \dd \log(z-\sigma_1)
		- \Hyper{\cdots\not\letter{\sigma_n}}(z) \cdot \dd \log (\sigma_n)
		\\&\quad
		+ \sum_{k=1}^{n-1} \Hyper{\cdots\not\letter{\sigma_{k+1}}\!\cdots \,-\, \cdots\not\letter{\sigma_k}\!\cdots}(z) \cdot \dd\log(\sigma_k-\sigma_{k+1})
		,
		\label{eq:hlog-total-differential} %
	\end{split}\end{equation}
	where $\cdots\mathord{\not\!\letter{\sigma_k}}\!\cdots$ denotes the word after deleting the $k$-th letter and summands with $\sigma_k = \sigma_{k+1}$ or $\sigma_n = 0$ do not contribute ($\dd \log 0 \defas 0$). 
	This can also be written in the form
	\begin{equation}
		\dd \Hyper{w}(z)
		=
		\sum_{k=1}^{n} \Hyper{\cdots\not\letter{\sigma_k}\!\cdots}(z)
		\cdot
		\dd\!\log\frac{\sigma_k - \sigma_{k-1}}{\sigma_{k+1}-\sigma_k}
		\quad\text{with}\quad
		\sigma_{n+1} \defas 0
		\quad\text{and}\quad
		\sigma_{0} \defas z
		.
		\label{eq:hlog-total-differential-compact} %
	\end{equation}
\end{lemma}
\begin{proof}
	The statement clearly holds for any $w=\letter{0}^n$, then only the first term $\Hyper{\letter{0}^{n-1}}(z)\cdot\dd\log (z)$ contributes on the right-hand side. For $w=\letter{\sigma}^{}\letter{0}^n$ ($\sigma \neq 0$) we check
	\begin{align*}
		\dd \Hyper{w}(z)
		&= \dd \int_0^{z} \frac{\dd z'}{z'-\sigma} \frac{\log^n (z')}{n!}
		= \frac{\log^n( z)}{n!}\frac{\dd z}{z-\sigma} 
		+ \dd\sigma \int_0^z \frac{\dd z'}{(z'-\sigma)^2} \frac{\log^n (z')}{n!}
		\\
		&= 
		 \Hyper{\letter{0}^n}(z)\frac{\dd z}{z-\sigma}
		+ \dd\sigma \int_0^z \dd z' \left[ 
		\frac{\sigma^{-1}}{z'-\sigma} \frac{\log^{n-1} (z')}{(n-1)!}
			-\partial_{z'} \frac{\log^n (z')}{n!} \left( \frac{1}{z'-\sigma} + \frac{1}{\sigma} \right)
		\right]
		\\
		&= 
			\Hyper{\letter{0}^n}(z)\frac{\dd z}{z-\sigma}
		- \left. \frac{\log^n (z')}{n!} \left( \frac{1}{z'-\sigma} + \frac{1}{\sigma}\right) \right|_{z' \rightarrow 0}^z \!\!\!\!\!\!\cdot\dd\sigma 
			+ \Hyper{\letter{\sigma}^{}\letter{0}^{n-1}}(z) \cdot \dd \log( \sigma)
		\\
		&= \Hyper{\letter{0}^n}(z) \cdot \dd \log\frac{z-\sigma}{\sigma}
			+\Hyper{\letter{\sigma}^{}\letter{0}^{n-1}}(z) \cdot \dd \log (\sigma)
		,
	\end{align*}
	paying special attention to the vanishing boundary term $\frac{z' \log^n (z')}{(z'-\sigma)\sigma} \rightarrow 0$ at $z' \rightarrow 0$. For any other word $w=\letter{\sigma_1}\!\!\cdots\letter{\sigma_n}$ ($n \geq 2$) we can use $\Hyper{\not\letter{\sigma_1}\!\cdots}(z) \rightarrow 0$ at $z \rightarrow 0$ to compute
	\begin{align*}
		\dd \Hyper{w}(z)
		&= \Hyper{\not\letter{\sigma_1}\!\cdots}(z) \frac{\dd z}{z - \sigma_1}
			+ \int_0^z \left[
						\frac{\Hyper{\not\letter{\sigma_1}\!\cdots}(z_1)}{(z_1-\sigma_1)^2} \dd \sigma_1
						+ \frac{1}{z_1 - \sigma_1} \dd \Hyper{\not\letter{\sigma_1}\!\cdots}(z_1)
			\right]
			\wedge \dd z_1 
		\\
		&= \Hyper{\not\letter{\sigma_1}\!\cdots}(z) \cdot \dd \log (z-\sigma_1)
			+ \int_0^z \frac{1}{z_1-\sigma_1} \left[ (\dd + \dd \sigma_1 \cdot \partial z_1) \Hyper{\letter{\sigma_2}\!\cdots}(z_1) \right] \wedge \dd z_1
		\tag{$\ast$}
	\end{align*}
	and apply induction to obtain the derivative of the weight $n-1$ hyperlogarithm $\Hyper{\not\letter{\sigma_1}\!\cdots}$:
	\begin{align*}
		\left( \dd + \dd \sigma_1 \cdot \partial_z \right) \Hyper{\not\letter{\sigma_1}\!\cdots}(z)
		&= \sum_{k=3}^n \dd \log \frac{\sigma_k-\sigma_{k-1}}{\sigma_{k+1}-\sigma_k} \cdot \Hyper{\not\letter{\sigma_1}\!\cdots\,\not\letter{\sigma_k}\!\cdots}(z)
		\\&\quad
		+ \left[
				\frac{\dd z + \dd \sigma_1 - \dd \sigma_2}{z-\sigma_2} 
				-\dd \log (\sigma_2-\sigma_3)
			\right]
			\Hyper{\not\letter{\sigma_1}\not\letter{\sigma_2}\!\cdots}(z)
		.
	\end{align*}
	Inserting this expression into $(\ast)$ readily shows
	\begin{multline*}
		\dd \Hyper{w}(z)
		= \Hyper{\not\letter{\sigma_1}\!\cdots}(z) \cdot \dd \log (z-\sigma_1)
			+ \sum_{k=3}^n \dd \log \frac{\sigma_k - \sigma_{k-1}}{\sigma_{k+1}-\sigma_k} \cdot \int_0^z \frac{\dd z'}{z' - \sigma_1} \Hyper{\not\letter{\sigma_1}\!\cdots\,\not\letter{\sigma_k}\!\cdots}(z')
		\\
		+ \dd \log(\sigma_1 - \sigma_2) \int_0^z \left[\frac{\dd z'}{z'-\sigma_1} - \frac{\dd z'}{z' - \sigma_2} \right] \Hyper{\not\letter{\sigma_1}\not\letter{\sigma_2}\!\cdots}(z')
			- \dd \log(\sigma_2 - \sigma_3) \int_0^z \frac{\dd z'}{z' - \sigma_1} \Hyper{\not\letter{\sigma_1}\not\letter{\sigma_2}\!\cdots}(z_1)
	\end{multline*}
	and we are done by integrating
	$
		\int_0^z \frac{\dd z'}{z'-\sigma_1} \Hyper{\not\letter{\sigma_1}\!\cdots\,\not\letter{\sigma_k}\!\cdots}(z') 
		=
		\Hyper{\cdots\,\not\letter{\sigma_k}\!\cdots}(z)
		$ with \eqref{eq:def-Hlog-noshuffle}.\footnote{If $w$ has only one non-zero letter $\sigma_k\neq 0$ ($k=1$ was treated separately before, so $k>1$), this formula does not apply to the term where $\letter{\sigma_k}$ is deleted. But then $\dd \log\frac{\sigma_k - \sigma_{k-1}}{\sigma_{k+1}-\sigma_k} = \dd \log \frac{\sigma_k}{-\sigma_k} = 0$ anyway.}
\end{proof}
In particular, this lemma shows that when $\sigma(t) \in \Sigma \subset \C(t)$ are considered as rational functions of a parameter $t$ say, then $\Hyper{w}(z)$ with $w \in \Sigma^{\times}$ is a hyperlogarithm in $t$ over an alphabet consisting of the zeros of $\sigma_i(t)-\sigma_j(t)$ and $\sigma_i(t) - z$.
For our application we need this result in the form of 
\begin{lemma}
	\label{lem:reginf-differential} %
	The total differential of any regularized limit at infinity is given by
	\begin{multline}
		\dd \AnaReg{z}{\infty}\Hyper{\letter{\sigma_1}\!\cdots\,\letter{\sigma_n}}(z)
		= 
			\sum_{k=2}^{n-1} \dd \log \frac{\sigma_{k}-\sigma_{k-1}}{\sigma_{k+1}-\sigma_k} \cdot \AnaReg{z}{\infty} \Hyper{\cdots\,\not\letter{\sigma_k}\!\cdots}(z)
		\\
			+\dd \log(\sigma_1 - \sigma_2) \cdot \AnaReg{z}{\infty} \Hyper{\not\letter{\sigma_1}\!\cdots}(z)
			-\dd \log (\sigma_n )\cdot \AnaReg{z}{\infty} \Hyper{\cdots\,\not\letter{\sigma_n}}(z)
		.
		\label{eq:reginf-total-differential} %
	\end{multline}
\end{lemma}
\begin{proof}
	This formula is just the regularized limit of \eqref{eq:hlog-total-differential-compact} at $z \rightarrow \infty$, so we must prove that a partial derivative commutes with this limit:
	$\partial_t \AnaReg{z}{\infty} L_w(z)
	= \AnaReg{z}{\infty} \partial_t L_w(z) $.
	The regularization \eqref{eq:hlog-divergences} means that we can write
	\begin{equation*}
		\AnaReg{z}{\infty} \partial_t \Hyper{w}(z)
		= \sum_k \AnaReg{z}{\infty} \left[ \log^k(z) \cdot \partial_t f_{w,\infty}^{(k)} (z) \right]
		\urel{\eqref{eq:reglim-product}}
			\AnaReg{z}{\infty} \partial_t f_{w,\infty}^{(0)}(z)
		\urel{\eqref{eq:lim=reglim}}
			\lim_{z\rightarrow\infty} \partial_t f_{w,\infty}^{(0)}(z),
	\end{equation*}
	because each $f_{w,\infty}^{(k)}(z)$ is analytic jointly in $t$ and $z$ since it is a linear combination of iterated integrals $\int_0^{1/z}$ by \eqref{eq:reginf-expansion-convolution}. The key implication is that $\partial_t f_{w,\infty}^{(k)}(z)$ remains analytic around $z\rightarrow \infty$ and we may interchange
	\begin{equation*}
		\AnaReg{z}{\infty} \partial_t \Hyper{w}(z)
		= \lim_{z \rightarrow \infty} \partial_t f_{w,\infty}^{(0)}(z)
		= \partial_t \lim_{z \rightarrow \infty} f_{w,\infty}^{(0)}(z)
		= \partial_t f_{w,\infty}^{(0)}(\infty)
		= \partial_t \AnaReg{z}{\infty} \Hyper{w}(z). \qedhere
	\end{equation*}
\end{proof}
\begin{example}
	\label{ex:reginf-as-hlog}%
	A single non-zero letter integrates to the logarithm
	\begin{equation*}
		\Hyper{\letter{-\sigma}}(z)
		= \log \left( \frac{z+\sigma}{\sigma} \right)
		= \log(z) + \log \left(\frac{1}{z}+\frac{1}{\sigma}\right)
		,\quad\text{so}\quad
		\AnaReg{z}{\infty} \Hyper{\letter{-\sigma}}(z)
		= \log\left( \frac{1}{\sigma}\right)
		= -\Hyper{\letter{0}}(\sigma).
	\end{equation*}
	For two letters, equation~\eqref{eq:reginf-total-differential} dictates that
	\begin{equation*}
		\dd \AnaReg{z}{\infty} \Hyper{\letter{-\sigma}\letter{-\tau}}(z)
		= \dd \log (\sigma-\tau) \cdot \AnaReg{z}{\infty} \Hyper{\letter{-\sigma}-\letter{-\tau}}(z)
		- \dd \log(-\tau) \cdot \AnaReg{z}{\infty} \Hyper{\letter{-\sigma}}(z).
	\end{equation*}
	In the case $\tau=1$ this simplifies with the above to
	$
		-\dd \log(\sigma-1) \cdot \Hyper{\letter{0}}(\sigma)
	$ and we conclude that 
	$
		\AnaReg{z}{\infty} \Hyper{\letter{-\sigma}\letter{-1}}(z)
		= - \Hyper{\letter{1}\letter{0}}(\sigma) + C
	$ for some constant $C$.
\end{example}
An iteration of lemma~\ref{lem:reginf-differential} allows us to write any $\AnaReg{z}{\infty} \Hyper{w}(z)$, which implicitly depends on a parameter $t$, in terms of explicit hyperlogarithms in $t$. This is all we need to integrate such a function with respect to $t$, so the following conclusion of this section is very useful in practice.
\begin{proposition}
	\label{prop:reginf-as-hlog} %
	Suppose the letters $\Sigma = \set{\sigma_1(t),\ldots\sigma_N(t)} \subset \C(t)$ depend rationally on a parameter $t$ and let $0 \in \Sigma$. Define the alphabet $\Sigma_t$ by
	\begin{equation}
		\Sigma_t
		\defas
		\bigcup_{1\leq i < j \leq N}
		\set{\text{zeros and poles of $\sigma_i(t) - \sigma_j(t)$}}
		\subset
		\C,
		\label{eq:def-alphabet-parameter-reduced} %
	\end{equation}%
\nomenclature[Sigma t]{$\Sigma_t$}{points where letters in $\Sigma$ coincide or get infinite, equation~\eqref{eq:def-alphabet-parameter-reduced}\nomrefpage}%
	then for any $w \in T(\Sigma)$ there exist $w_u \in T(\Sigma)$, indexed by words $u \in \Sigma_t^{\times}$, such that $\AnaReg{z}{\infty} \Hyper{w}(z) = \sum_u \Hyper{u}(t) \cdot \AnaReg{t}{0}\AnaReg{z}{\infty} \Hyper{w_u}(z)$ is a finite linear combination of hyperlogarithms $\Hyper{u}(t)$ in $t$ (with $t$-independent coefficients). In other words,
	\begin{equation}
		\AnaReg{z}{\infty}
			\HlogAlgebra(\Sigma)(z)
		\subseteq
		\HlogAlgebra(\Sigma_t)(t) \tp \AnaReg{t}{0} \AnaReg{z}{\infty} \HlogAlgebra(\Sigma)(z)
		.
		\label{eq:reginf-as-hlog} %
	\end{equation}
\end{proposition}
\begin{proof}
	A straightforward induction of \eqref{eq:reginf-total-differential} over the length $\abs{w}$ suffices: Suppose we have recursively written the right-hand side of \eqref{eq:reginf-total-differential} in the form of \eqref{eq:reginf-as-hlog}, namely
	\begin{equation*}
		\partial_t \AnaReg{z}{\infty} \Hyper{w}(z)
		=\!\!\!\!\!
		\sum_{u \in \Sigma_t^{\times}, \tau \in \Sigma_t} \frac{\lambda_{\tau,u}}{t-\tau} \Hyper{u}(t) \cdot c_u
		\quad\text{such that}\quad
		\AnaReg{z}{\infty} \Hyper{w}(z)
		=
		C + \sum_{\mathclap{u \in \Sigma_t^{\times}, \tau \in \Sigma_t}} \lambda_{\tau,u} \Hyper{\letter{\tau}u}(t) \cdot c_u.
	\end{equation*}
	Here we factored $\partial_t \log (\sigma_{k+1}-\sigma_k) = \prod_{\tau}(t-\tau)^{\lambda_{\tau}} \in \regulars(\Sigma_t)(t)$ completely (so $\lambda_{\tau,u} \in \Z$). By corollary~\ref{corollary:hlog-reg0}, the constant of integration is of the advertised form
	\begin{equation*}
		C = \AnaReg{t}{0} \AnaReg{z}{\infty} \Hyper{w}(z). \qedhere
	\end{equation*}
\end{proof}
\begin{example}\label{ex:reginf-oneletter-log}
	If $\Sigma = \set{0,t}$ contains just one non-zero element $t$, only the letter $\Sigma_{t} = \set{0}$ remains and every regularized limit is a polynomial in $\log(t)$:
	\begin{equation*}
		\AnaReg{z}{\infty} T(\set{0,t})(z)
		= \HlogAlgebra(\set{0})(t) \tp \AnaReg{t}{0} \AnaReg{z}{\infty} T(\Sigma)(z)
		= \Q[\log(t)] \tp \AnaReg{t}{0} \AnaReg{z}{\infty} T(\set{0,t})(z)
		.
	\end{equation*}
\end{example}

\subsection{Regularized limits of regularized limits}
\label{sec:reglim-reginf} %
In this section we present a simple algorithm to compute limits of the form
\begin{equation}
	\AnaReg{t}{0} \AnaReg{z}{\infty} \Hyper{w} (z)
	\quad\text{for words}\quad
	w \in \Sigma^{\times}, \quad \Sigma \subset \C(t)
	\label{eq:reglim-reginf-setup} %
\end{equation}
with letters $\sigma(t) \in \Sigma$ that are rational functions of the parameter $t$. These appear as the integration constants in proposition~\ref{prop:reginf-as-hlog} and need to be understood in particular for the computation of multiple integrals in section~\ref{sec:multiple-integrals}. In the simplest case, we can apply
\begin{lemma}
	\label{lem:reginf-reglim-finite} %
	Suppose that all letters in the alphabet $\Sigma$ have finite, non-positive limits
	$
		\lim_{t\rightarrow 0} \Sigma
		\defas \setexp{\lim_{t\rightarrow 0} \sigma(t)}{\sigma \in \Sigma} 
		\subset \C \setminus (0,\infty)
	$.
	Then any word $w=\letter{\sigma_1}\!\!\cdots\letter{\sigma_n} \in T(\Sigma)$ that ends in a letter with $\lim_{t\rightarrow 0} \sigma_n(t) \neq 0$ fulfils
	\begin{equation}
		\AnaReg{t}{0} \AnaReg{z}{\infty} \Hyper{w}(z)
		=
		\AnaReg{z}{\infty} \Hyper{\lim\limits_{t \rightarrow 0} w}(z)
		.
		\label{eq:reglim-reginf-finite} %
	\end{equation}%
\nomenclature[limtw]{$\displaystyle\lim_{t\rightarrow 0} w$}{letter-wise limit $\letter{\sigma(t)} \mapsto \letter{\sigma(0)}$, lemma~\ref{lem:reginf-reglim-finite}\nomrefpage}%
	Here we use the notation $\lim_{t \rightarrow 0} w \defas \letter{\lim_{t\rightarrow 0} \sigma_1(t)}\!\cdots \letter{\lim_{t\rightarrow 0} \sigma_n(t)}$.
\end{lemma}%
\begin{proof}
	We first consider combinations of the form $w=(\letter{\sigma_1} - \letter{-1})\letter{\sigma_2}\!\!\cdots\letter{\sigma_n}$ such that $\Hyper{w}(\infty)$ is the absolutely convergent integral from \eqref{eq:hlog-finite-infinity}, given explicitly by
	\begin{equation}\label{eq:hlog-finite-infinity-extended}
		\int_{0<z_1<\cdots < z_n < \infty}
		\ 
		\frac{1+\sigma_1(t)}{(z_1+1)(z_1 - \sigma_1(t))}
		\frac{1}{z_2 - \sigma_2(t)}
		\cdots
		\frac{1}{z_n - \sigma_n(t)}
		\ 
		\dd z_1 \cdots\, \dd z_n
		.%
	\end{equation}
	The limit at $t \rightarrow 0$ of the integrand is still absolutely integrable and the theorem of dominated convergence applies\footnote{%
		Explicitly, the integrand in \eqref{eq:hlog-finite-infinity-extended} is absolutely bounded by the (integrable) integrand of $L_{(\letter{0}-\letter{-1})\letter{0}^{n-2}\letter{-1}}$ times a constant $C=\abs{1+\sigma_1}\cdot C_1\cdots C_n$ such that for all $z \in \R_+$, $C_n \geq \frac{z+1}{\abs{z-\sigma_n}}$ and $C_i \geq \frac{z}{\abs{z-\sigma_i}}$ ($i<n$).
		For $i<n$ we can set $C_i=1$ whenever $\Realteil \sigma_i \leq 0$ and otherwise $C_i=\frac{\abs{\sigma_i}}{\abs{\Imaginaerteil \sigma_i}}$, which remains finite even if $\sigma_i(t) \rightarrow 0$ (the only divergent case is when $\arg\sigma_i \rightarrow 0$ and can be circumvented by a small perturbation of the direction from which we take $t\rightarrow 0$). A similar consideration for $C_n$ finally shows that the constant $C$ can be chosen uniformly in $t$.
}%
	to show $\lim_{t\rightarrow 0} \Hyper{w}(\infty) = \Hyper{\lim_{t\rightarrow 0} w}(\infty)$. Since this is a finite limit, it coincides with $\AnaReg{t}{0} \Hyper{w}(\infty)$. Now for the given word $w = \letter{\sigma_1}\!\!\cdots\letter{\sigma_n}$, by $\sigma_n \neq 0$ all contributions to $\WordReg{0}{\infty}(w) = \WordReg{}{\infty}(w)$ in \eqref{eq:def-reginf-word} are of the form just discussed. Therefore we can conclude with
	\begin{equation*}
		\AnaReg{t}{0} \AnaReg{z}{\infty} \Hyper{w}(z)
		\urel{\eqref{eq:reginf-from-word}}
		\AnaReg{t}{0} \Hyper{\WordReg{0}{\infty} (w)} (\infty)
		=
		\Hyper{\lim\limits_{t\rightarrow 0} \WordReg{0}{\infty} (w)}(\infty)
		=
		\Hyper{\WordReg{0}{\infty} \lim\limits_{t \rightarrow 0} w}(\infty)
		\urel{\eqref{eq:reginf-from-word}}
		\AnaReg{z}{\infty} \Hyper{\lim\limits_{t \rightarrow 0} w}(z),
	\end{equation*}
	where $\WordReg{0}{\infty}$ and $\lim_{t \rightarrow 0}$ commute because the latter just substitutes individual letters and keeps the last letter (which is either $\sigma_n$ or $-1$) non-zero.
\end{proof}
\begin{example}\label{ex:reglim-reginf-simple}%
	In example~\ref{ex:reginf-as-hlog}, the constant is 
	$C = \AnaReg{\sigma}{0} \AnaReg{z}{\infty} \Hyper{\letter{-\sigma}\letter{-1}}(z) 
		 = \AnaReg{z}{\infty} \Hyper{\letter{0}\letter{-1}}(z)
		 = \mzv{2}$
	as determined in example~\ref{ex:dilog-reginf}.
\end{example}
But in cases when $\lim_{t\rightarrow 0} \sigma_n(t) = 0$, the limiting integrand contains $\frac{\dd z_n}{z_n}$ and is not integrable, so the lemma does not apply. This naive method also fails when some letter $\lim_{t\rightarrow 0} \sigma_k(t) = \infty$ diverges. We may avoid both of these situations with the help of
\begin{lemma}
	\label{lem:reglim-reginf-rescaling} %
	For $\alpha \in \Z$ and any word $w = \letter{\sigma_1}\!\!\cdots\letter{\sigma_n}$ with rational letters $\sigma_k \in \C(t)$, let $\sigma_k' (t) \defas t^{\alpha} \cdot \sigma_k(t)$ denote rescaled letters and write $w' \defas \letter{\sigma_1'}\!\cdots\letter{\sigma_n'}$. Then
	\begin{equation}
		\AnaReg{t}{0} \AnaReg{z}{\infty} \Hyper{w}(z)
		=
		\AnaReg{t}{0} \AnaReg{z}{\infty} \Hyper{w'}(z)
		.
		\label{eq:reglim-reginf-rescaling} %
	\end{equation}
\end{lemma}
\begin{proof}
	Since $\AnaReg{z}{\infty} \Hyper{w}(z) = \AnaReg{z}{\infty} \Hyper{\WordReg{0}{}(w)}(z)$ we can restrict to words with $\letter{\sigma_n} \neq 0$. Then $\Hyper{w}(z)$ denotes an iterated integral \eqref{eq:def-II}, and the change $f(z) = z \, t^{\alpha}$ of variables proves $\Hyper{w}(z) = \Hyper{w'}\left(z \, t^{\alpha} \right)$.
	With regularizations \eqref{eq:hlog-divergences} at $\tau = \infty$ this reads
	\begin{equation*}
		\sum_k \log^k (z) \cdot f_{w,\infty}^{(k)}(z)
		=
		\sum_k \log^k \left( z\, t^{\alpha} \right) \cdot f_{w',\infty}^{(k)}\left( z\, t^{\alpha} \right).
	\end{equation*}
	As $f_{w',\infty}^{(k)}(z\, t^{\alpha})$ is analytic at $z \rightarrow \infty$, we can compare the coefficients of $\log(z)$ to deduce
	\begin{equation*}
		f_{w,\infty}^{(k)}(z)
		=
		\sum_{j \geq k} \binom{j}{k} \log^{j-k} \left( t^{\alpha} \right) \cdot f_{w',\infty}^{(j)} \left( z \, t^{\alpha} \right)
	\end{equation*}
	and in particular $\AnaReg{z}{\infty} \Hyper{w}(z) = f_{w,\infty}^{(0)}(\infty) = \sum_k (\alpha \log t)^k \cdot f_{w',\infty}^{(k)}(\infty)$. Acting with $\AnaReg{t}{0}$ annihilates every $\log(t)$ and leaves behind only 
	$\AnaReg{t}{0} f_{w',\infty}^{(0)}(\infty)
	$.
\end{proof}
\begin{definition}
	\label{def:letter-leading-coefficient} %
	For any rational $0 \neq \sigma(t) \in \C(t)$, the Laurent series
	\begin{equation}
		\sigma(t)
		=
		\sum_{n=N}^{\infty} a_n t^n
		\quad\text{at}\quad
		t \rightarrow 0
		\quad\text{with}\quad
		a_N \neq 0
		\label{eq:letter-laurent-series} %
	\end{equation}%
	defines the \emph{vanishing degree} $\deg_t(\sigma) \defas N \in \Z$ and a \emph{leading coefficient} $\leadCoeff_t(\sigma) \defas a_N$. 
	We set $\deg_t(0) \defas \infty$ and
	$\deg_t(\letter{\sigma_1}\!\!\cdots\letter{\sigma_n})
		\defas \min \setexp{\deg_t(\sigma_i)}{1 \leq i \leq n}$
	for any word.
\end{definition}%
\nomenclature[deg]{$\deg_t$}{vanishing degree of Laurent series $\sigma(t)$, minimum vanishing degree over all letters for words $w$, definition~\ref{def:letter-leading-coefficient}\nomrefpage}%
\nomenclature[lead t s]{$\leadCoeff_t(\sigma)$}{leading coefficient of $\sigma(t)$ as $t\rightarrow 0$, definition~\ref{def:letter-leading-coefficient}\nomrefpage}%
Whenever the final letter $\sigma_n \neq 0$ of a word $w = \letter{\sigma_1}\!\!\cdots\letter{\sigma_n}$ is of smallest vanishing degree $\deg_t(\sigma_n) = \deg_t(w)$, rescaling $\sigma_i'(t) \defas \sigma_i(t) \cdot t^{-\deg_t(w)}$ ensures finiteness of all $\sigma_k'(t)$ at $t \rightarrow 0$ and furthermore that $\sigma_n'(t) \rightarrow \leadCoeff_t(\sigma_n) \neq 0$. Hence in this case, lemmata \ref{lem:reglim-reginf-rescaling} and \ref{lem:reginf-reglim-finite} together prove that
\begin{equation}
	\AnaReg{t}{0} \AnaReg{z}{\infty} \Hyper{w}(z)
	=
	\AnaReg{z}{\infty} \Hyper{\ReglimWord{t}{0} (w) }(z)
	\quad\text{for}\quad
	\ReglimWord{t}{0} (w)
	\defas
	\lim_{t\rightarrow 0} \letter{\sigma_1'}\!\cdots\letter{\sigma_n'}.
	\label{eq:reglim-word-end-minimal}%
\end{equation}%
\nomenclature[reg t 0 w]{$\displaystyle\ReglimWord{t}{0}(w)$}{regularized limit of words, \eqref{eq:reglim-word-end-minimal} and \eqref{eq:def-reglim-word}\nomrefpage}%
Beware that we use three different notions of regularized limits:
\begin{itemize}
	\item The upper case operators \eqref{eq:def-reglim} like $\AnaReg{z}{\tau}$ act on functions, while
	\item lower case is reserved for combinatorial actions on words. We distinguish shuffle regularization \eqref{eq:def-shuffle-regularization} ($\WordReg{\sigma}{\tau}$) from limits of letters \eqref{eq:def-reglim-word} ($\ReglimWord{t}{0}$).
\end{itemize}
\begin{example}
	Rescaling $w = \letter{-1}\letter{-1/t}$ by $t$ gives $w' = \letter{-t}\letter{-1}$ and $\ReglimWord{t}{0} (w) = \letter{0}\letter{-1}$. So for the regularized limit, we get the same result as in example~\ref{ex:reglim-reginf-simple}:
	\begin{equation*}
		\AnaReg{t}{0} \AnaReg{z}{\infty} \Hyper{\letter{-1}\letter{-1/t}}(z)
		=
		\AnaReg{z}{\infty} \Hyper{\letter{0}\letter{-1}}(z)
		=
		\mzv{2}.
	\end{equation*}
\end{example}
The case when $\deg_t(\sigma_n) > \deg_t(w)$ can be dealt with using shuffle algebra. Let
\begin{equation*}
	k
	\defas
	\max \setexp{i}{\deg_t(\sigma_i) = \deg_t(w)}
	< n
\end{equation*}
denote the position of the last letter in $w$ with minimal vanishing degree. Using \eqref{eq:shuffle-decomposition-formula-tail} we can rewrite $w = \sum_i w_i \shuffle a_i$ such that each $w_i$ ends in $\letter{\sigma_k}$ and $a_i$ is a suffix of $\letter{\sigma_{k+1}}\!\!\cdots\letter{\sigma_n}$. Note that $\deg_t(w_i) = \deg_t(\sigma_k)$ and since $\abs{a_i} \leq n - k < n$ is shorter than $w$, a recursion of this process allows to write $w$ as a (not necessarily unique) finite sum
\begin{equation}
	w 
	= \sum_i \left( w_{i,1} \shuffle \ldots \shuffle w_{i,r_i} \right)
	\quad\text{such that}\quad
	\deg_t(\sigma_{i,j}) = \deg_t(w_{i,j}),
	\label{eq:word-rescaling-shuffle-decomposition} %
\end{equation}
by which we mean that all $w_{i,j} \in T(\Sigma)$ that occur have a final letter $\sigma_{i,j}$ of minimal vanishing degree.\footnote{The existence of such a decomposition also follows from $T(\Sigma) \isomorph \Q[\Lyndons{\Sigma}]$ if we extend the vanishing degree to a total order on $\Sigma$ such that the corresponding Lyndon words all end in a letter of minimal vanishing degree.} This extends \eqref{eq:reglim-word-end-minimal} to a multiplicative map
	\begin{equation}
		\ReglimWord{t}{0}\colon T(\Sigma) \longrightarrow T(\leadCoeff_t(\Sigma)),
		\quad
		w \mapsto \sum_i \ReglimWord{t}{0}(w_{i,1}) \shuffle \ldots \shuffle \ReglimWord{t}{0}(w_{i,r_i})
		\label{eq:def-reglim-word}%
	\end{equation}
	on all of $T(\Sigma)$, where we set $\ReglimWord{t}{0}(\letter{0}^n) \defas 0$ for any $n \in \N$ (we could also choose $\ReglimWord{t}{0}(\letter{0}^n) = \letter{0}^n$ here because we always apply $\WordReg{0}{\infty}$ afterwards). It is well-defined (see remark~\ref{rem:reglim-word-uniqueness}) and its image only contains words with letters in the alphabet
	\begin{equation}
		\leadCoeff_t(\Sigma)
		\defas
		\set{0} \cupdot \setexp{\leadCoeff_t(\sigma)}{0 \neq \sigma \in \Sigma}
		\label{eq:def-leading-coefficient-alphabet} %
	\end{equation}%
\nomenclature[lead t Sigma]{$\leadCoeff_t(\Sigma)$}{leading coefficients of the letters $\sigma(t) \in \Sigma$, equation~\eqref{eq:def-leading-coefficient-alphabet}\nomrefpage}%
of leading coefficients. Because all three operators $\Hyper{\cdot}(z)$, $\AnaReg{z}{\infty}(\cdot)$ and $\AnaReg{t}{0}(\cdot)$ are multiplicative, their application to \eqref{eq:word-rescaling-shuffle-decomposition} proves
\begin{proposition}
	\label{prop:reglim-reginf-from-word} %
	If $\leadCoeff_t(\Sigma) \subset \C \setminus (0,\infty)$, then any word $w \in T(\Sigma)$ fulfils
	\begin{equation}
		\AnaReg{t}{0} \AnaReg{z}{\infty} \Hyper{w}(z)
		=
		\AnaReg{z}{\infty} \Hyper{\ReglimWord{t}{0}(w)}(z).
		\label{eq:reglim-reginf-from-word} %
	\end{equation}
\end{proposition}
\begin{example}\label{ex:vanishing-degree-decomposition}
	For the word $w = \letter{a}\letter{bt}\letter{ct^2}$, one decomposition \eqref{eq:word-rescaling-shuffle-decomposition} is
	\begin{align*}
		w
		&= \letter{a} \shuffle \letter{bt} \letter{ct^2} - \letter{bt}\letter{a}\letter{ct^2} - \letter{bt} \letter{ct^2} \letter{a}
		\\
		&= \letter{a} \shuffle \left[ \letter{bt} \shuffle \letter{ct^2} - \letter{ct^2} \letter{bt} \right] - \letter{bt}\letter{a} \shuffle \letter{ct^2} + \letter{bt}\letter{ct^2}\letter{a} + \letter{ct^2}\letter{bt} - \letter{a} \letter{bt}\letter{ct^2}\letter{a}
		\\
		&= \letter{a} \shuffle \letter{bt} \shuffle \letter{ct^2}
			- \letter{a} \shuffle \letter{ct^2} \letter{bt}
			- \letter{bt}\letter{a} \shuffle \letter{ct^2}
			+ \letter{ct^2}\letter{bt}\letter{a}
	\end{align*}
	such that $\ReglimWord{t}{0} (w) = \letter{a}\shuffle \letter{b} \shuffle \letter{c} - \letter{a} \shuffle \letter{0} \letter{b} - \letter{c} \shuffle \letter{0}\letter{a} + \letter{0}\letter{0}\letter{a}$.
\end{example}
\begin{example}\label{ex:reg0-reginf-MZV}
	For one non-zero letter $-t$ in $\Sigma=\set{0,-t}$, we find $\leadCoeff_t(\Sigma) = \set{0,-1}$ and obtain in the limit $t \rightarrow 0$ the multiple zeta values
	\begin{equation*}
		\AnaReg{t}{0} \AnaReg{z}{\infty} \HlogAlgebra(\set{0,-t})
		= \AnaReg{z}{\infty} \HlogAlgebra(\set{0,-1})
		\urel{\eqref{eq:reginf-gives-mzv}}
			\MZV.
	\end{equation*}
\end{example}
\begin{remark}\label{rem:reglim-word-uniqueness}
	It is not immediately clear that $\ReglimWord{t}{0}(w)$ is fixed uniquely, because there can be different decompositions \eqref{eq:word-rescaling-shuffle-decomposition} (if several letters have the same vanishing degree) and the substitution of letters $\letter{\sigma}$ performed by $\ReglimWord{t}{0}$ depends on the word $w$ in which they appear:
$ \letter{\sigma} \mapsto \letter{\leadCoeff_t(\sigma)}$ if $\deg_t(\sigma) = \deg_t(w)$ and $\letter{\sigma} \mapsto \letter{0}$ otherwise. So we distinguish different degrees with
\begin{align}
	\Sigma^{(d)}
	&\defas \setexp{\sigma \in \Sigma}{\deg_t(\sigma) = d}
	\quad\text{and let}
	\label{eq:letters-by-vanishing-degree}\\
	\alg^{(d)}
	&\defas \Q \oplus T\left( \bigcupdot_{k \geq d} \Sigma^{(k)}\right) \Sigma^{(d)}
	\subset T(\Sigma)
	\nonumber%
\end{align}
denote the sub algebra generated (and spanned) by all words $w$ with final letter of minimal vanishing degree $d = \deg_{t}(w)$. Since $w_{i,j} \in \alg^{(\deg_t(w_{i,j}))}$ for every factor, any decomposition \eqref{eq:word-rescaling-shuffle-decomposition} can be regrouped according to
\begin{equation*}
	T(\Sigma)
	\isomorph
	\Q[\letter{0}] \tp
	\alg^{(D)} \tp \alg^{(D-1)} \tp \cdots \tp \alg^{(d)}
	\quad\text{for}\quad
	D = \max_{0\neq \sigma \in \Sigma} \deg_t(\sigma),\ 
	d = \min_{\sigma \in \Sigma} \deg_t(\sigma)
\end{equation*}
which follows from repeated application of lemma~\ref{lemma:shuffle-decomposition}. On each $\alg^{(k)}$, the regularization $\ReglimWord{t}{0}$ is a simple substitution:
$\letter{\sigma} \mapsto \letter{\leadCoeff_t(\sigma)}$ if $\sigma \in \Sigma^{(k)}$ and $\letter{\sigma} \mapsto \letter{0}$ otherwise, independent of the word in which this letter appears. This proves that $\ReglimWord{t}{0}$ is indeed well-defined, multiplicative and independent of any choices.
\end{remark}
\begin{corollary}
	\label{corollary:reglim-degree-decouple}%
	Suppose $\Sigma\setminus \set{0} = \Sigma^{(D)} \cupdot \ldots \cupdot \Sigma^{(d+1)} \cupdot \Sigma^{(d)} \subset \C(t)$ denotes an alphabet of rational letters, partitioned according to the vanishing degree \eqref{eq:letters-by-vanishing-degree}. Then the regularized limit $t\rightarrow 0$ takes values in the subalgebra
	\begin{equation}
	\ReglimWord{t}{0} T(\Sigma)
	= \prod_{k=d}^{D} T\left( \leadCoeff_t\big( \Sigma^{(k)} \big) \right)
	\subseteq T\big( \leadCoeff_t(\Sigma) \big)
	.%
	\label{eq:reglim-degree-decouple}%
\end{equation}
	Consequently, if $\leadCoeff_t(\Sigma) \cap (0,\infty) = \emptyset$, the corresponding hyperlogarithms decompose as
	\begin{equation}
		\AnaReg{t}{0} \AnaReg{z}{\infty} \HlogAlgebra(\Sigma)(z)
		= \prod_{k=d}^{D} \AnaReg{z}{\infty} \HlogAlgebra\left( \leadCoeff_t\big( \Sigma^{(k)} \big) \right)(z).
		\label{eq:reglim-degree-decouple-real}%
	\end{equation}
\end{corollary}
\begin{example}
	In example~\ref{ex:vanishing-degree-decomposition}, all letters in $\Sigma = \set{0,a,bt,ct^2}$ have different vanishing degrees. We saw explicitly that $\ReglimWord{t}{0}(w)$ is a linear combination of shuffle products of words over the alphabets $\leadCoeff_t(\Sigma^{(0)}) = \set{0,a}$, $\leadCoeff_t(\Sigma^{(1)}) = \set{0,b}$ and $\leadCoeff_t(\Sigma^{(2)}) = \set{0,c}$. So no matter which $w \in T(\Sigma)$ we choose, the regularized limit is a polynomial in logarithms and multiple zeta values (see examples~\ref{ex:reginf-oneletter-log} and \ref{ex:reg0-reginf-MZV}):
	\begin{equation*}
		\AnaReg{t}{0} \AnaReg{z}{\infty} \Hyper{w}(z)
		\in \prod_{-\sigma \in \set{a,b,c}}  \AnaReg{z}{\infty} \HlogAlgebra\left( \set{0,-\sigma} \right)
		= {\MZV}[\log(-a),\log(-b),\log(-c)].
	\end{equation*}
	When $a$, $b$ and $c$ are negative, these are real numbers. But note that for positive values, imaginary parts appear as we discuss in the next section.
\end{example}

\subsection{Analytic continuation and singularities on the path}
As homotopy invariant functionals, hyperlogarithms $\Hyper{w}(z)$ with $w \in T(\Sigma)$ are multivalued on $\C \setminus \Sigma$ and depend on the homotopy class of the path $\gamma\colon (0,1) \longrightarrow \C \setminus \Sigma$ of integration. It must be specified to give the limits
\begin{equation*}
	\AnaReg{z}{\infty} \Hyper{w}(z)
	= \lim_{z \rightarrow \infty} \Hyper{\WordReg{0}{\infty}(w)}(z)
\end{equation*}
a well-defined meaning. We adopt the simplest choice of a straight path on the positive real axis $\im(\gamma) = \R_+$, represented for example by $\gamma(t)=t/(1-t)$. This definition means that the quantity $\AnaReg{z}{\infty} \Hyper{w}(z)$ is completely described by $w$ alone, which greatly simplifies its implementation on a computer.

But it also requires the absence $\Sigma \cap (0,\infty) = \emptyset$ of any positive letters, because otherwise the analytic continuation of $\Hyper{w}(z)$ past such a point on $\R_+$ is ambiguous. For precisely this reason we had to require $\leadCoeff_t(\Sigma) \cap (0,\infty) = \emptyset$ in proposition~\ref{prop:reglim-reginf-from-word}.
\begin{example}
	\label{ex:reginf-log-cut}%
	For arbitrary $0 \neq \sigma \in \C$, the logarithm $\Hyper{\letter{\sigma}}(z) = \log \left( 1-z/\sigma \right)$ is analytic on $\abs{z}< \abs{\sigma}$ and can be continued (in $z$) along all of the positive real axis $\R_+$ if and only if $\sigma \notin \R_+$. In this case,
	\begin{equation}
		\AnaReg{z}{\infty} \Hyper{\letter{\sigma}}(z) 
		=
		\AnaReg{z}{\infty} \left[ \log (z) + \log\left( \frac{1}{z} - \frac{1}{\sigma} \right) \right]
		= \log\left( -\frac{1}{\sigma} \right)
		= -\ln \abs{\sigma} - \imag \arg(-\sigma)
		\label{eq:reginf-log-cut}%
	\end{equation}
	where $\arg(z) \in (-\pi, \pi)$ denotes the branch of the argument that is analytic on $z \in \C \setminus (-\infty,0]$. In particular, $f(\sigma) \defas \AnaReg{z}{\infty} \Hyper{\letter{\sigma}}(z)$ has a branch cut along $[0,\infty)$.

	Technically, $f(\sigma)$ is a homotopy invariant functional of paths $\gamma\colon (0,1) \longrightarrow \RSphere\setminus\set{0,\sigma,\infty}$ with tangential base points $\gamma(0) = 0$, $\gamma(1) = \infty$ and $\dot{\gamma}(0) = \dot{\gamma}(1) = 1$. The straight line $\R_+$ is not of this type when $\sigma \in \R_+$ and the homotopy classes it represents in the two cases $\sigma \in \Halfplane^{\pm} \defas \setexp{z \in \C}{\pm \Imaginaerteil z > 0}$ are different, wherefore $f(\sigma)$ jumps on $\R_+$.
\end{example}%
We can still give such limits a well-defined meaning when we specify a particular path $\mu\colon [0,1) \longrightarrow \C \setminus (0,\infty)$ along which the parameter $t = \mu(s)$ approaches zero. We require
\begin{equation}
	\lim_{s\rightarrow 1} \mu(s) = 0,
	\quad\text{negative real}\quad
	\lim_{s\rightarrow 1} \dot{\mu}(s) < 0
	\quad\text{and}\quad
	\im(\mu) \subset \Halfplane^{+}
	\label{eq:regzero-parameter-path-conditions} %
\end{equation}
such that $t$ approaches $0$ from the right, tangent to the real axis but from the upper half-plane like in figure~\ref{fig:zero-limit-path}. So $t$ has a small positive imaginary part ($t \rightarrow 0 + \imag \varepsilon$) which implies that any rational letter $\sigma(t) \in \C(t)$ with positive limit $\sigma(0) \in (0,\infty)$ will acquire a non-vanishing imaginary part as well, for small enough $t$ on the path $t = \mu(s)$.

Starting at a suitably small $t = \mu(0)$, we may therefore assume that all such letters $\sigma(t)=\sigma(\mu(s)) \in \Halfplane^{\pm}$ stay tied to some half-plane, for all $s<1$, wherefore $\Sigma \cap (0,\infty) = \emptyset$ and $\AnaReg{z}{\infty} \Hyper{w}(z)$ is well-defined at each point $t=\mu(s)$.
\begin{definition}
	\label{def:alphabet-halfplane-partition} %
	Any alphabet (set of rational functions) $\Sigma \subset \C(t) \setminus (0,\infty)$ depending on a parameter $t$ and not including positive constants is partitioned uniquely into
	\begin{equation}
		\Sigma
		=
		\widetilde{\Sigma} \cupdot \Sigma^{+} \cupdot \Sigma^{-}
		\subset \C(t) \setminus (0,\infty)
		\label{eq:alphabet-halfplane-partition} %
	\end{equation}%
\nomenclature[Sigma tilde pm]{$\widetilde{\Sigma}, \Sigma^{\pm}$}{partition of letters $\sigma(t)$ by their limits at $t\rightarrow 0$, equation~\eqref{eq:alphabet-halfplane-partition}\nomrefpage}%
	such that $\leadCoeff_t(\widetilde{\Sigma}) \cap (0,\infty) = \emptyset$ and $\leadCoeff_t(\Sigma^{\pm}) \subset (0,\infty)$. Here the letters with positive leading coefficient are separated by $\Sigma^{\pm} \subset \Halfplane^{\pm}$ for sufficiently small $t = \mu(s)$ ($s$ close to $1$).
	In particular we note that whenever $\leadCoeff_t(\sigma) \in (0,\infty)$ is positive for $\sigma \in \Sigma$,
	\begin{equation}
		\deg_t(\sigma) < 0
		\Rightarrow
		\sigma \in \Sigma^{-}
		\quad\text{and}\quad
		\deg_t(\sigma) > 0
		\Rightarrow
		\sigma \in \Sigma^{+}
		\label{eq:alphabet-halfplane-nonzero-degree} %
	\end{equation}
	because $\sigma(t) = t^{\deg_t(\sigma)}\big( \leadCoeff_t(\sigma) + \bigo{t} \big)$.
	We denote the finite positive limits by
	\begin{equation}
		\Sigma^{\pm}_0
		\defas
		\setexp{\lim_{t \rightarrow 0} \sigma(t)}{\sigma \in \Sigma^{\pm} \ \text{and}\ \deg_t(\sigma) = 0}
		=
		\leadCoeff_t \left( \Sigma^{(0)} \cap \Sigma^{\pm} \right)
		\subset (0,\infty)
		.
		\label{eq:def-alphabet-finite-limits} %
	\end{equation}
\end{definition}%
\nomenclature[Sigma pm 0]{$\Sigma^{\pm}_0$}{finite positive limits ($t\rightarrow 0$) of letters $\sigma(t)\in\Sigma$, equation~\eqref{eq:def-alphabet-finite-limits}\nomrefpage}%
\begin{figure}\setcapindent{1em}%
	\begin{minipage}[b]{0.4\textwidth}\centering%
		\includegraphics[width=0.8\textwidth]{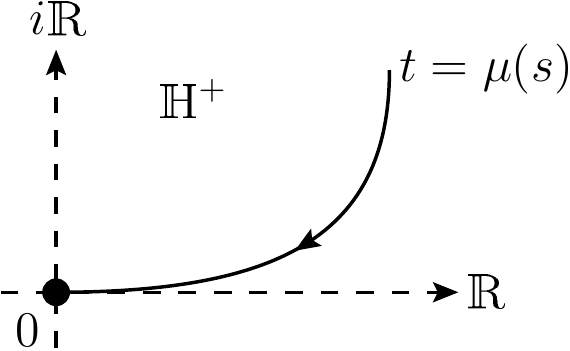}%
		\caption[Path for the limit $t\rightarrow 0$]{The limit $t \rightarrow 0$ is taken from the positive half-plane, corresponding to \eqref{eq:regzero-parameter-path-conditions}.}%
		\label{fig:zero-limit-path}%
	\end{minipage}\hfill
	\begin{minipage}[b]{0.58\textwidth}\centering%
		\includegraphics[width=\textwidth]{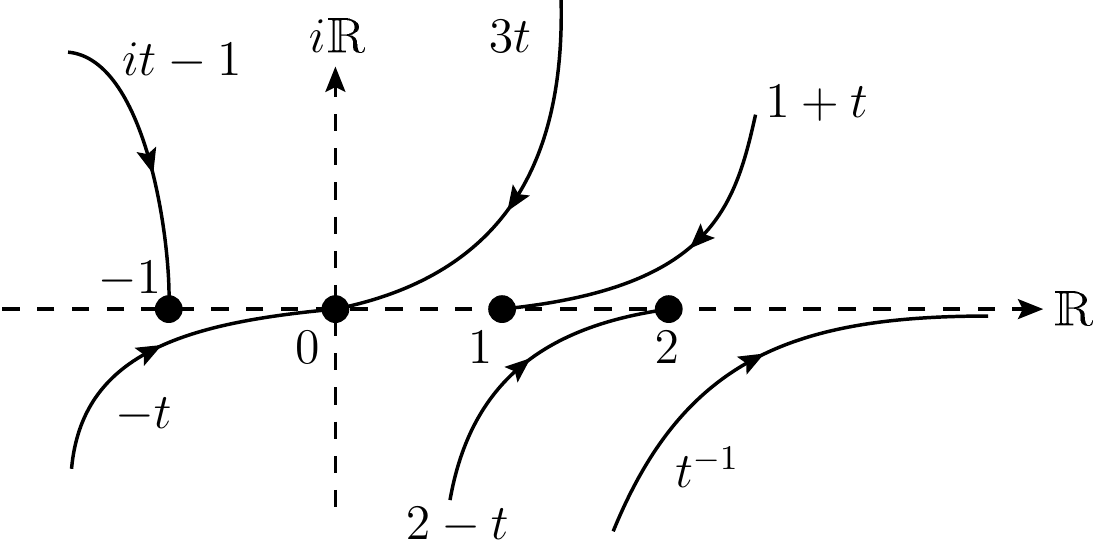}%
		\caption{Paths of the letters $\sigma(t) \in \Sigma$ of example~\ref{ex:alphabet-halfplane-partition} in the limit $t \rightarrow 0 + \imag \varepsilon$.}%
		\label{fig:alphabet-halfplane-partition} %
	\end{minipage}%
\end{figure}
\begin{example}
	\label{ex:alphabet-halfplane-partition}%
	The paths of the letters $\Sigma = \set{-1+\imag t, -t, 3t, 1+t, 2-t, t^{-1}}$ are shown in figure~\ref{fig:alphabet-halfplane-partition} for the limit $t = \mu(s) \rightarrow 0 + \imag\varepsilon$. The decomposition \eqref{eq:alphabet-halfplane-partition}  reads
	\begin{equation*}
		\widetilde{\Sigma}
		=
		\set{-1+\imag t, 0, -t}
		,\quad
		\Sigma^{+}
		=
		\set{3t,1+t}
		\quad\text{and}\quad
		\Sigma^{-}
		=
		\set{2-t, t^{-1}}
		.%
	\end{equation*}
\end{example}
Consider a word $w = \letter{\sigma_1}\!\!\cdots\letter{\sigma_n}$ with a final letter of minimal vanishing degree $d \defas \deg_t(\sigma_n) = \deg_t(w)$.
The rescaled letters $\sigma_k'(t) \defas \sigma_k(t) \cdot t^{-\deg_t(w)}$ all vanish at $t \rightarrow 0$ when $\deg_t(\sigma_k)>\deg_t(w)$, but the surviving contributions
\begin{equation*}
	\lim_{t \rightarrow 0} \sigma_k'(t)
	=
	\leadCoeff_t(\sigma_k)
	\neq 0
	\quad\text{of the letters}\quad
	\sigma_k
	\in
	\Sigma^{(d)}
	\defas
	\setexp{\sigma \in \Sigma}{\deg_t(\sigma)  = d}
\end{equation*}
can introduce positive letters (when $\Sigma^{\pm} \cap \Sigma^{(d)} \neq \emptyset$) into the word $w' = \ReglimWord{t}{0}(w)$ which prohibit a direct application of lemma~\ref{prop:reglim-reginf-from-word}.
Instead we can exploit homotopy invariance to deform the straight line $\R_+$ continuously (without crossing any letters) into a path $\gamma$ that avoids all letters $\leadCoeff_t(\Sigma)$ such that
\begin{figure}
	\centering
	\includegraphics[width=0.6\columnwidth]{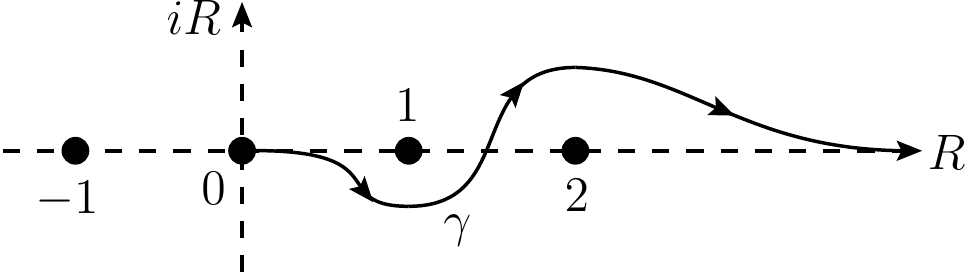}%
	\caption[Deformation of the straight integration path due to positive limits of letters]{The letters $\set{1+t,2-t} \subset \Sigma$ in example~\ref{ex:alphabet-halfplane-partition} induce a deformation of the real integration path $[0,\infty)$ towards $\gamma$, which avoids the positive limits in passing below $\Sigma_0^{+} = \set{1}$ and above $\Sigma_0^{-} = \set{2}$.}%
	\label{fig:contour-deformation}%
\end{figure}
\begin{equation}
	\AnaReg{t}{0+ \imag\varepsilon} \AnaReg{z}{\infty} \Hyper{w}(z)
	=
	\AnaReg{t}{0+ \imag\varepsilon} 
		\int_{\R_+} \WordReg{0}{\infty} (w)
	=
	\int_{\gamma} \WordReg{0}{\infty} (w')
	\label{eq:reglim-reginf-contour-deformation} %
\end{equation}
as a convergent iterated integral. This contour $\gamma$ is determined by the distribution of the letters $\Sigma^{(d)} \cap \Sigma^{\pm}$ among the half-planes, as illustrated for example~\ref{ex:alphabet-halfplane-partition} in figure~\ref{fig:contour-deformation}. In general we distinguish three cases:
\begin{itemize}
	\item[$d < 0$:]
		From \eqref{eq:alphabet-halfplane-nonzero-degree} we see $\Sigma^{(d)} \cap \Sigma^{\pm} \subset \Sigma^{-}$\!\!\!,\, all positive letters in $w'$ stem from $\sigma \in \Halfplane^{-}$ below $\R_+$. So after rescaling by $t^{-d}$, $\gamma$ must pass above all these letters $\leadCoeff_t\big(\Sigma^{(d)}\big) \cap (0,\infty) = \leadCoeff_t\big(\Sigma^{(d)} \cap \Sigma^{\pm}\big)$.

	\item[$d > 0$:]
		This case is the reflection of the earlier, so $\Sigma^{(d)} \cap \Sigma^{\pm} \subset \Sigma^{+}$ and $\gamma$ must pass below all positive letters $\leadCoeff_t(\Sigma^{(d)}) \cap (0,\infty)$ of $\ReglimWord{t}{0}(w)$.

	\item[$d=0$:]
		No rescaling is involved and $\lim_{t \rightarrow 0} \big( \Sigma^{(0)} \cap \Sigma^{\pm} \big) \subset \Sigma_0^- \cup \Sigma_0^+$ can approach the positive axis from both half-planes. So $\gamma$ must pass above $\Sigma_0^{-}$ and below $\Sigma_0^{+}$ as illustrated in figure~\ref{fig:contour-deformation}.
\end{itemize}
In the last case, we must require $\Sigma_0^{-} \cap \Sigma_0^{+} = \emptyset$ as otherwise $\gamma$ is pinched between letters from $\Sigma^{-}$ and $\Sigma^{+}$ that approach the same positive limit as $t \rightarrow 0$ from both half-planes. This situation is discussed in detail in the next section.

For now we rephrase the non-trivial contour integral \eqref{eq:reglim-reginf-contour-deformation} in terms of iterated integrals along straight lines. To this end denote the positive letters of $w'$ by
\begin{equation}
	\leadCoeff_t(\Sigma) \cap (0, \infty)
	= \set{\tau_1,\ldots,\tau_N}
	\quad\text{with}\quad
	0<\tau_1<\cdots<\tau_N<\infty.
	\label{eq:positive-split-letters}%
\end{equation}
So $\gamma$ is determined by a choice of sign $\delta_k = \pm 1$ for each $1\leq k \leq N$ encoding whether $\gamma$ passes below ($\delta_k = 1$) or above ($\delta_k = -1$) the letter $\tau_k$; we can encode this by adding an infinitesimal imaginary $\delta_k \imag \varepsilon$ to $\tau_k$ (then taking the straight path $\R_+$). Using lemma~\ref{lemma:hlog-path-concatenation-singular} we can write 
$\Hyper{w'}(z) = \sum_{(w')} \int_{\tau_1}^z w'_{(1)} \cdot \Hyper{\WordReg{}{\tau_1}(w'_{(2)})}(\tau_1)$ and compute the branch of $\log(z-\tau_1)$ in the definition of the first factor according to $\gamma$ for $z>\tau_1$ as%
\nomenclature[delta tau]{$\delta_{\tau}$}{the half-plane $\Halfplane^{\delta_{\tau}}$ from which the positive letter $\tau$ is approached, equation~\eqref{eq:reglim-reginf-branch}, page~\pageref{eq:reglim-reginf-branch}}%
\begin{equation}
	\int_{\tau_1}^z \letter{\tau_1}
	\urel{\eqref{eq:hlog-singular-concatenation-branch}}
		\Hyper{\tau_1}(z)
	= \log\left(1-\frac{z}{\tau_1} \right)
	= \log (z-\tau_1) - \log(\tau_1) + \imag \pi \delta_{\tau_1}
	.
	\label{eq:reglim-reginf-branch}%
\end{equation}
Taking the regularized limit $z\rightarrow \infty$, this maps to just $\imag\pi\delta_{\tau_1} - \Hyper{\letter{0}}(\tau_1)$. The remaining iterated integrals $\int_{\tau_1}^z u$ ($u$ not ending in $\tau_1$) can be transformed by $f(z)=z-\tau_1$ to hyperlogarithms $\Hyper{\WordTransformation{f}(u)}(z-\tau_1)$ to iterate this procedure with the next positive letter $\tau_2$ (now shifted to $\tau_2 - \tau_1$) and so on.
\begin{corollary}
	\label{corollary:reglim-reginf-splitted}%
	Suppose $w \in T(\Sigma)$ has rational letters $\Sigma \subset \C(t) \setminus (0,\infty)$ and ends with a letter of minimal vanishing degree $d\defas \deg_t(w)$. Furthermore assume that no pinching occurs at $t \rightarrow 0$, that means $d \neq 0$ or $\Sigma^{+}_{0} \cap \Sigma^{-}_{0} = \emptyset$ shall be fulfilled.

	Then we can explicitly compute a finite sum representation
	\begin{equation}
		\AnaReg{t}{0+\imag \varepsilon} \AnaReg{z}{\infty} \Hyper{w}(z)
		= \sum_k
			(\imag \pi)^{\lambda_k}
			\Hyper{w_{k}^{N}}(\infty)
			\Hyper{w_{k}^{N-1}}(\tau_N-\tau_{N-1})
			\cdots
			\Hyper{w_{k}^{1}}(\tau_2 - \tau_1)
			\Hyper{w_{k}^{0}}(\tau_1)
		\label{eq:reglim-reginf-splitted}%
	\end{equation}
	where $\lambda_k \in \Z$ and the words have letters $w_k^{r} \in T(\setexp{\sigma-\tau_r}{\sigma \in \leadCoeff_t(\Sigma)})$. Furthermore all hyperlogarithms that appear are finite because all words are regularized: $w_k^{r}$ does not begin with $\letter{\tau_{r+1}-\tau_r}$ for $r<N$ and $w_k^{N} \in \im\big(\WordReg{0}{\infty}\big)$.
\end{corollary}
\begin{example}
	In example~\ref{ex:dilog-past-one} we computed the splitting for $z>\tau_1 = 1$ as
	\begin{equation*}
		\Li_2(z)
		= \mzv{2} + \int_1^z \letter{1}\letter{0}
		-\Big[ \imag \pi + \log(z-1) \Big] \int_1^z \letter{0} 
	\end{equation*}
	for $\gamma$ passing below one ($\delta_1 = 1$). The regularized limits of $\log(z-1)$ and $\int_1^z \letter{0} = \Hyper{\letter{-1}}(z-1) = \log(z)$ at $z \rightarrow \infty$ vanish, so we do not get imaginary parts in the limit
	\begin{equation*}
		\AnaReg{z}{\infty} \Li_2(z)
		= \mzv{2} + \AnaReg{z}{\infty} \Hyper{\letter{0}\letter{-1}}(z-1) - \imag\pi \AnaReg{z}{\infty} \Hyper{\letter{-1}}(z-1)
		= \mzv{2} + \Hyper{(\letter{0}-\letter{-1})\letter{-1}}(\infty) = 2 \mzv{2}.
	\end{equation*}
\end{example}
\begin{remark}
	\label{remark:reginf-shift-invariance}%
	In the last step of the path decomposition we must compute the regularized limit $\AnaReg{z}{\infty}\Hyper{w}(z-\tau_N)$ of some word $w \in T(\setexp{\sigma-\tau_N}{\sigma \in \leadCoeff_t(\Sigma)})$. Since the argument is not $z$ but $z-\tau_N$, \eqref{eq:reginf-from-word} does not apply directly. But the regularized limit of $\log(z-\tau_N) = \log(z) + \log(1-\tau_N/z)$ at $z\rightarrow \infty$ is zero and therefore
\begin{equation*}
	\AnaReg{z}{\infty} \Hyper{w}(z-\tau_N)
	\urel{\eqref{eq:hlog-divergences}}
		\AnaReg{z}{\infty}
		\sum_k \log^k(z-\tau_N) \cdot f_{w,\infty}^{(k)}(z-\tau_N)
	\urel{\eqref{eq:reglim-product}}
		f_{w,\infty}^{(0)}(\infty)
	= \AnaReg{z}{\infty} \Hyper{w}(z).
\end{equation*}
\end{remark}

\subsection{Pinching the path of integration}
\begin{figure}
	\centering
	\includegraphics[width=0.4\textwidth]{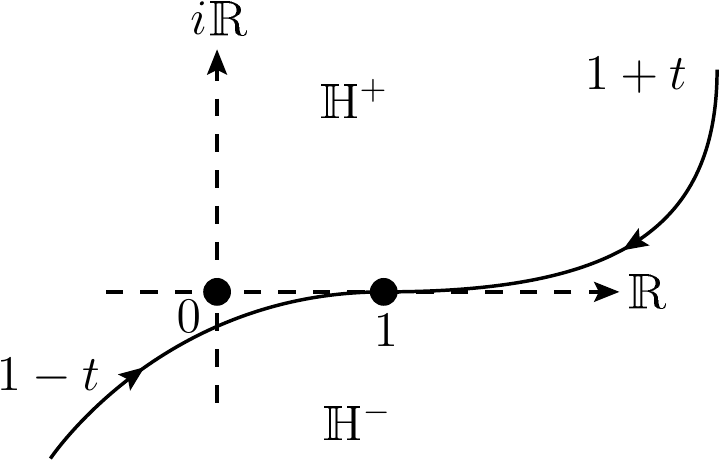}
	\caption[Pinching of the integration path between two letters]{In the limit when $t\rightarrow 0$, the integration path along the positive real axis $\R$ is pinched at $1$ by the two letters $1\pm t$. Example~\ref{ex:pinch-reglim-simple} computes such a limit for the word $\letter{1+t}\letter{1-t}$.}%
	\label{fig:reglim-pinch}%
\end{figure}
We now investigate the particular situation where the contour $\gamma$ of integration gets pinched.\footnote{In all our explicit computations of Feynman integrals so far, this actually never happened. But our solution of this technicality of pinching renders our algorithms applicable in general and might be needed for other applications or maybe even for calculations of (more complicated) Feynman integrals.} So we consider words with $\deg_t(w) = 0$, wherefore all involved letters $\sigma(t) \in \Sigma \subset \C(t)$ have a finite limit $\sigma(0)$ at $t\rightarrow 0$.

As in \eqref{eq:positive-split-letters} we split the integration at the positive limits
$	\big[
		\lim_{t \rightarrow 0} \Sigma\,
	\big]
	\cap
	(0,\infty)
	=
	\set{\tau_1,\ldots,\tau_N}$
and consider the corresponding decomposition
\begin{equation}
	\Hyper{\cdot}(z)
	=
	\int_{\tau_N}^z \convolution \int_{\tau_{N-1}}^{\tau_N} 
	\convolution \cdots \convolution
	\int_{\tau_1}^{\tau_2} \convolution \int_{0}^{\tau_1},
	\label{eq:positive-letter-split}%
\end{equation}
but \emph{before} taking the limit $t \rightarrow 0$. Thus no letter $\sigma(t)$ has yet become positive and all factors in equation~\eqref{eq:positive-letter-split} denote absolutely convergent iterated integrals, along straight paths from $\tau_k$ to $\tau_{k+1}$. It is convenient to transform them into $[0,\infty)$ using
\begin{equation}
	\int_{\tau_k}^{\tau_{k+1}} \!\!\!\!\letter{\sigma_1}\!\!\cdots\letter{\sigma_n}
	= \AnaReg{z}{\infty} \Hyper{\WordTransformation{f_k}(w)}(z)
	\quad\text{with}\quad
	f_k(z)
	\defas \begin{cases}
		\frac{z-\tau_k}{\tau_{k+1}-z} & \text{for $0 \leq k <N$ and} \\
		z-\tau_N & \text{when $k=N$.}\\
	\end{cases}
	\label{eq:positive-split-transformation}%
\end{equation}
Here we set $\tau_0 \defas 0$, $\tau_{N+1} \defas \infty$ and have already applied $\AnaReg{z}{\infty}$ to $\int_{\tau_N}^z$, following remark~\ref{remark:reginf-shift-invariance}, while $\Hyper{\WordTransformation{f}(w)}(\infty)$ is already finite for $k<N$.
Because $\AnaReg{t}{0}$ is linear and multiplicative, we may take this limit on each factor $\int_{\tau_k}^{\tau_{k+1}}$ individually.

The crucial point is that by construction, no letter $\sigma(t)$ approaches the interior of any of the intervals $[\tau_k,\tau_{k+1}]$ as $t \rightarrow 0$. So after the above transformation, $f_k(\sigma(t))$ does not land on $(0,\infty)$ at $t\rightarrow 0$ and we cannot have any pinches on the right-hand side of
\begin{equation}
	\AnaReg{t}{0} \AnaReg{z}{\infty} \Hyper{\cdot}(z)
	=
	\left[ \AnaReg{t}{0} \AnaReg{z}{\infty} \Hyper{\WordTransformation{f_0}(\cdot)}(z) \right]
	\convolution \cdots \convolution
	\left[ \AnaReg{t}{0} \AnaReg{z}{\infty} \Hyper{\WordTransformation{f_k}(\cdot)}(z) \right],
	\label{eq:reglim-reginf-pinch-decomposition}%
\end{equation}
which can therefore be evaluated with corollary~\ref{corollary:reglim-reginf-splitted}.
\begin{example}
	\label{ex:pinch-reglim-simple} %
	Consider $w=\letter{1+t}\letter{1-t}$ and the function $\AnaReg{z}{\infty} \Hyper{w}(z)$, which is well-defined only for $t \in \C \setminus \R$. In the limit $\Halfplane^{+} \ni t \rightarrow 0 + \imag\varepsilon$, the letters $1 \pm t \in \Halfplane^{\pm}$ approach $\tau_1 =1$ from both half-planes and thus pinch the integration contour $(0,\infty)$ as shown in figure~\ref{fig:reglim-pinch}. The subdivision \eqref{eq:positive-letter-split} reads 
	\begin{equation*}
		\AnaReg{z}{\infty} \Hyper{w}(z)
		=
		\int_0^1 \letter{1+t}\letter{1-t}
		+ \int_0^1 \letter{1-t} \cdot \AnaReg{z}{\infty} \int_1^{z} \letter{1+t}
		+ \AnaReg{z}{\infty} \int_1^{z} \letter{1+t}\letter{1-t}.
	\end{equation*}
	The last term is $\AnaReg{z}{\infty} \Hyper{\letter{t}\letter{-t}}(z)$ and becomes $\AnaReg{z}{\infty} \Hyper{\letter{1 + \imag\varepsilon}\letter{-1}}(z) = \frac{3}{2} \mzv{2} + \imag\pi\ln 2$ under $t \rightarrow 0 + \imag\varepsilon$ after a straightforward application of corollary~\ref{corollary:reglim-reginf-splitted}. The middle term contributes $\AnaReg{z}{\infty} \Hyper{\letter{t}}(z) \mapsto \imag\pi$ at $t \rightarrow 0+\imag\varepsilon$ from example~\ref{ex:reginf-log-cut}, multiplied with
	\begin{equation*}
		\AnaReg{t}{0}
		\int_0^1 \letter{1-t}
		= \AnaReg{t}{0}
			\int_0^{\infty} \left( \letter{-1+1/t} - \letter{-1} \right)
		=
		\AnaReg{z}{\infty} \Hyper{\letter{1-\imag\varepsilon}}(z)
		= -\imag\pi.
	\end{equation*}
	Finally, the first summand is
	$
		\int_0^{\infty} w
	$
	for $w = \WordTransformation{f}(\letter{1+t}\letter{1-t})$ with $f(z) = \frac{z}{1-z}$, one finds
	\begin{equation*}
		w=
		\letter{-1-1/t}\letter{-1+1/t} - \letter{-1}\letter{-1+1/t} - \letter{-1-1/t} \shuffle \letter{-1} + \letter{-1}\letter{-1-1/t}.
	\end{equation*}
	The term $\AnaReg{z}{\infty} \Hyper{\letter{-1-1/t} \shuffle \letter{-1}}(z)$ vanishes and all other words in $w$ end on a letter with minimal vanishing degree $\deg_t(-1 \pm 1/t) = -1$, so there is no pinch and we get
	\begin{equation*}
		\AnaReg{t}{0} \int_0^1 \letter{1+t}\letter{1-t}
		=
		\AnaReg{z}{\infty} \Hyper{\letter{-1}\letter{1-\imag\varepsilon} - \letter{0}\letter{1-\imag\varepsilon} + \letter{0}\letter{-1}}(z)
		= \frac{3}{2} \mzv{2} + \imag\pi\ln 2
	\end{equation*}
	using corollary~\ref{corollary:reglim-reginf-splitted} again.
	When we add up all three contributions we obtain the constant $\AnaReg{t}{0} \AnaReg{z}{\infty} \Hyper{\letter{1+t}\letter{1-t}}(z) = \tfrac{3}{2}\pi^2 + 2\pi\imag\ln 2$, whereafter
	\begin{equation*}
		F(t)
		\defas
		\AnaReg{z}{\infty} \Hyper{\letter{1+t}\letter{1-t}}(z)
		=
		\Hyper{\letter{0}\letter{1} + \letter{1}\letter{-1} - \letter{0}\letter{-1}}(t)
		+ \tfrac{3}{2} \pi^2
		+ \imag\pi\left[ 2 \Hyper{\letter{0}}(t) - \Hyper{\letter{1}}(t) + 2\ln 2 \right]
	\end{equation*}
	for $t \in \Halfplane^{+}$ is easily computed with proposition~\ref{prop:reginf-as-hlog}. Crossing the branch cut on the positive real axis swaps $\imag\pi$ with $-\imag\pi$ in this equation. The pinch manifests itself through the logarithmic divergence of $F(t)$ when $t \rightarrow 0$ coming from $2\pi\imag \Hyper{\letter{0}}(t) = 2\pi\imag \log(t)$.

	Note how the constant $\ln 2$ appeared, which is a period of the Riemann sphere $\RSphere\setminus \set{-1,0,1,\infty}$ with four punctures, even though the limit $\lim_{t\rightarrow 0} w = \letter{1}\letter{1}$ suggests a period of $\RSphere\setminus=\set{0,1,\infty}$ only. This phenomenon can only occur through a pinch.
\end{example}
\begin{lemma}
	After the transformation \eqref{eq:positive-split-transformation}, the letters $\sigma\in\Sigma$ behave as
\begin{equation}
	\deg_t(f_k(\sigma))
	=
	\deg_t\left( \frac{\sigma-\tau_k}{\tau_{k+1} - \sigma} \right)
	=
	\begin{cases}
		\phantom{-}\deg_t(\sigma)>0 & \text{if $\sigma(0) = \tau_k = 0$ with $k=0$,} \\
		\phantom{-}\sdeg_t(\sigma)>0 & \text{if $\sigma(0) = \tau_k$ and $k > 0$,} \\
		-\sdeg_t(\sigma)<0 & \text{if $\sigma(0) = \tau_{k+1}$ and} \\
		\phantom{-}0 & \text{when $\sigma(0) \notin \set{\tau_k,\tau_{k+1}}$.} \\
	\end{cases}
	\label{eq:vanishing-degree-transformed}%
\end{equation}
The \emph{subleading vanishing degree} $\sdeg_t(\sigma) \defas N+s$ of the power series $\sigma(t) = \sum_{n=N}^{\infty} a_n t^n$ with leading coefficient $0 \neq a_N = \leadCoeff_t(\sigma)$ and $N = \deg_t(\sigma)$ is defined such that $a_{N+1}=\cdots=a_{N+s-1}=0$ and the \emph{subleading coefficient} is $\subleadCoeff_t(\sigma) \defas a_{N+s} \neq 0$. Furthermore,
\begin{equation}
	\leadCoeff_t \left( f_k(\sigma) \right)
	=
	\leadCoeff_t\left( \frac{\sigma-\tau_k}{\tau_{k+1} - \sigma} \right)
	=
	\begin{cases}
		\phantom{-}\frac{\leadCoeff_t (\sigma)}{\tau_1} & \text{if $\sigma(0) = \tau_k = 0$ with $k=0$,} \\
		\phantom{-}\frac{\subleadCoeff_t (\sigma)}{\tau_{k+1} - \tau_k} & \text{if $\sigma(0) = \tau_k$ and $k> 0$,} \\
		-\frac{\tau_{k+1} - \tau_k}{\subleadCoeff_t(\sigma)} & \text{if $\sigma(0) = \tau_{k+1}$ and} \\
		\frac{\sigma(0) - \tau_k}{\tau_{k+1} - \sigma(0)} & \text{when $\sigma(0) \notin \set{\tau_k, \tau_{k+1}}$.} \\
	\end{cases}
	\label{eq:leadCoeff-transformed}%
\end{equation}
\end{lemma}
The important point of this simple calculation is that a pinch in $f(\Sigma)$ can only occur involving letters with $\deg_t(f(\sigma)) = 0$, but for those the limit $\frac{\sigma(0)-\tau_k}{\tau_{k+1}-\sigma(0)}$ is by construction non-positive. So indeed a pinch is impossible after the decomposition \eqref{eq:positive-letter-split}.
\begin{example}
	\label{ex:pinch-reglim-general} %
	We generalize example~\ref{ex:pinch-reglim-simple} to $w = \letter{1+at}\letter{1+bt}$ for arbitrary (non-zero) $a$ and $b$. The same calculation of $\AnaReg{t}{0} \AnaReg{z}{\infty} \Hyper{w}(z)$ results, after transforming all integrations to $[0,\infty)$, in
	\begin{multline*}
		\AnaReg{t}{0+\imag\varepsilon} \AnaReg{z}{\infty}
		\int_0^z \left(
			\letter{at} \shuffle \letter{-1-1/(bt)}
			+ \letter{at}\letter{bt}
			+ \letter{-1-1/(at)}\letter{-1-1/(bt)}
			- \letter{-1}\letter{-1-1/(bt)}
			+ \letter{-1}\letter{-1-1/(at)}
		\right)
		\\
		=\AnaReg{z}{\infty} \int_0^{\infty} \left( 
			\letter{a+\imag\varepsilon} \shuffle \letter{-1/b - \imag\varepsilon}
			+ \letter{a + \imag\varepsilon} \letter{b + \imag\varepsilon}
			+ \letter{-1/a - \imag\varepsilon} \letter{-1/b - \imag\varepsilon}
			- \letter{0} \letter{-1/b - \imag\varepsilon}
			+ \letter{0} \letter{-1/a - \imag\varepsilon}
		\right)
		.
	\end{multline*}
	This form is indeed completely general: When $\sigma \notin (0,\infty)$, the $\pm \imag\varepsilon$ in the letter $\letter{\sigma \pm \imag\varepsilon}$ can be ignored, but when $\sigma \in (0,\infty)$ becomes positive it denotes from which half-plane $\Halfplane^{\pm}$ this letter approaches its limit.

	The whole point of the algorithm in this section is that the resulting decomposition never mixes letters with $+\imag\varepsilon$ and $-\imag\varepsilon$ in the same word, such that they can all be computed along a non-pinched deformed contour \eqref{eq:reglim-reginf-contour-deformation} using corollary \ref{corollary:reglim-reginf-splitted}.

	In this particular example we can further simplify the calculation by the help of an inversion $f(z)=z^{-1}$ applied to all words containing $-\imag\varepsilon$ letters. The result simplifies to
	\begin{equation*}
		C(a,b)
		\defas
		\AnaReg{t}{0+\imag\varepsilon} \AnaReg{z}{\infty} \Hyper{\letter{1+at}\letter{1+bt}}(z)
		= \AnaReg{z}{\infty} \Hyper{\letter{a+\imag\varepsilon}\letter{b+\imag\varepsilon} - \letter{a+\imag\varepsilon}\shuffle \letter{-b+\imag\varepsilon} + \letter{-b+\imag\varepsilon}\letter{-a+\imag\varepsilon}}(z).
	\end{equation*}
	We know that this must be a constant in the non-pinched case $a,b \in \Halfplane^+$. First we compute $\Hyper{\letter{a}-\letter{-a}}(\infty) = \Hyper{\letter{1+\imag\varepsilon} - \letter{-1}}(\infty) = \imag\pi$ and the same for $b$, so
$\dd C(a,b) = \dd \log(a-b) \cdot \Hyper{\letter{a}-\letter{b} + \letter{-b}-\letter{-a}}(\infty) = 0$
indeed. But when $a>0$ and $b<0$ say, then the sign flips for $b$ and we get $\dd C(a,b) = 2\pi\imag \cdot \dd \log(a-b)$ such that $C(a,b) = 2\pi\imag \log(a-b) + \frac{3}{2} \pi^2$ is not a constant but really a function of $a$ and $b$.

This formalism can apparently be employed to deal with positive letters, associated branch choices and pinches in an automated way.
\end{example}
\begin{remark}
While this way of computation might seem very cumbersome (and in simple cases like the above shortcuts are certainly possible), it is guaranteed to work in all cases and can be automatized (we include it in our program {\HyperInt} of chapter~\ref{chap:hyperint}).
\end{remark}
\begin{corollary}\label{corollary:reglim-reginf-pinch}
	Suppose the rational letters $\Sigma \subset \C(t) \setminus (0,\infty)$ have $N$ pinch points $\set{\tau_1,\ldots,\tau_N} = \Sigma^{+}_0 \cap \Sigma^{-}_0$, ordered $0 < \tau_1 < \cdots < \tau_N$, and write
	\begin{equation}
		\subleadCoeff^{(\tau)}(\Sigma)
		\defas
		\set{0} \cup
		\setexp{\subleadCoeff_t(\sigma)}{\sigma \in \Sigma^{\pm}\ \text{such that}\ \lim_{t \rightarrow 0} \sigma(t) = \tau}
	\end{equation}
	for the subleading coefficients of the letters pinching at $\tau$. Then any $w \in T(\Sigma)$ admits a decomposition of $\AnaReg{t}{0+\imag\varepsilon} \AnaReg{z}{\infty} \Hyper{w}(z)$ like in \eqref{eq:reglim-reginf-splitted}, but with additional factors $\AnaReg{z}{\infty} \Hyper{v_k^{j}}(z)$ on the right where $v_k^j \in T(\subleadCoeff^{(\tau_j)})$. Explicitly, \eqref{eq:reglim-reginf-pinch-decomposition} and corollary~\ref{corollary:reglim-reginf-splitted} give an algorithm to express any regularized limit according to
	\begin{equation}
		\AnaReg{t}{0+\imag\varepsilon} \AnaReg{z}{\infty} \HlogAlgebra(\Sigma)(z)
		\subseteq
			\prod_{k=0}^{N}
			\AnaReg{z}{\infty} \HlogAlgebra\left( f_k(\Sigma) \right)
			\times \prod_{k=1}^N \Q[\log(\tau_k-\tau_{k-1})]\AnaReg{z}{\infty} \HlogAlgebra \left(\subleadCoeff^{(\tau_k)}(\Sigma) \right),
		\label{eq:reginf-pinch-product-algebra} %
	\end{equation}
	where possibly occurring positive letters in $\subleadCoeff^{(\tau_k)}(\Sigma)$ do not pinch. Hence these can be rewritten themselves using corollary~\ref{corollary:reglim-reginf-splitted}.
\end{corollary}
\begin{proof}
	We only need to comment on the contributions of words with $\deg_t(\WordTransformation{f_k}(w)) \neq 0$: By the different vanishing degrees \eqref{eq:vanishing-degree-transformed}, letters with $\sigma(0) =\tau_k$ and $\sigma(0) = \tau_{k+1}$ do not mix. So say we take a word with $\deg_t(\WordTransformation{f_k}(w)) > 0$, then from \eqref{eq:leadCoeff-transformed} all its letters are of the form $\subleadCoeff_t(\sigma)/(\tau_{k+1}-\tau_k)$ and we can rescale them simultaneously: Following the proof of lemma~\ref{lem:reglim-reginf-rescaling}, this only introduces explicit terms $\log(\tau_{k+1}-\tau_k)$.
\end{proof}
\begin{remark}
	The pinching letters with different subleading degrees decouple, which gives a more refined characterization of this representation just as in corollary~\ref{corollary:reglim-degree-decouple}.
\end{remark}
\begin{example}
	For $\Sigma=\set{0,1+t,1-t}$, the pinch at $\tau=1$ has $\subleadCoeff^{(1)}(\Sigma) = \set{0,\pm 1}$ and shows that alternating sums can contribute to the limit, as we explicitly saw in example~\ref{ex:pinch-reglim-simple}. More generally, a word over the alphabet $\Sigma=\set{0,1+a_1 t,\ldots,1+a_N t}$ can contribute a complicated period of $\C\setminus\set{0,a_1,\ldots,a_N}$ to the regularized limit $t \rightarrow 0$, even though the letters of the word itself have just two different limits $\set{0,1}$ and might suggest, naively, that multiple zeta values (and $\imag \pi$) should suffice.
\end{example}

\section{Polylogarithms}
\label{sec:polylogarithms}%
Unfortunately, many different names are currently used for hyperlogarithms and various special classes of them. To avoid confusion and to aid comparison with other sources, we like to briefly collect these terms and relate them with each other.
\subsection{Multiple polylogarithms}
The \emph{classical polylogarithms} $\Li_n$ of \emph{weight} $n\in\N$ are defined by the power series%
\footnote{%
While \eqref{eq:def-Li} extends to arbitrary $n\in\Z$, for $n \leq 0$ it only defines rational functions $\Li_n(z)\in\Q[\frac{1}{1-z}]$.%
}%
	\begin{equation}
		\Li_n (z)
		\defas
		\sum_{1\leq k}
		\frac{z^k}{k^n}
		,\quad\text{convergent when}\quad
		\abs{z} < 1
		\label{eq:def-Li}%
	\end{equation}
and go back to Euler \cite{Euler:MZV}. Lewin's wonderful book \cite{Lewin:PolylogarithmsAssociatedFunctions} became a standard reference on these functions. The (iterated) integral representations
\begin{equation}
	\Li_1(z) 
	= \int_0^z \frac{\dd z'}{1-z'}
	= -\log (1-z)
	\quad\text{and}\quad
	\Li_{n+1}(z)
	=\int_0^z \frac{\Li_{n}(z')}{z'}\  \dd z'
	\label{eq:Li-iterated-integral}%
\end{equation}
reveal them as the hyperlogarithms
$
	\Li_n(z)
	=
	-\Hyper{\letter{0}^{n-1}\letter{1}}(z)
$.
These are special cases of
\begin{definition}
	\label{def:MPL}%
	The multiple polylogarithm \cite{Zagier:MZVApplications,Goncharov:MplCyclotomyModularComplexes} (MPL) of \emph{weight} $n_1+\cdots+n_r$ and \emph{depth} $r$ associated to a sequence $n_1,\ldots,n_r \in \N$ is defined by the power series
	\begin{equation}
		\Li_{\vec{n}}(\vec{z})
		=
		\Li_{n_1,\ldots,n_r}(z_1,\ldots,z_r)
		\defas
		\sum_{1\leq k_1<\cdots<k_r}
		\frac{z_1^{k_1} \!\cdots z_r^{k_r}}{k_1^{n_1}\!\cdots k_r^{n_r}}
		\label{eq:def-Mpl}%
	\end{equation}
	in $r$ complex variables $z_i$.	Absolute convergence of \eqref{eq:def-Mpl} is assured when
	$
		\abs{z_k \!\cdots z_r} < 1
	$
	for all $1\leq k \leq r$.
	We adopt the convention $\Li_{\vec{n}}(z) \defas \Li_{n_1,\ldots,n_r}(1,\ldots,1,z)$ to identify MPLs of a single variable.
\end{definition}
\begin{lemma}%
	\label{lemma:Hlog-as-Mpl}%
	Any hyperlogarithm can be expressed in terms of multiple polylogarithms according to $\Hyper{\letter{0}}(z) = \log( z) = -\Li_1(1-z)$ and the formula
	\begin{equation}
		\Hyper{w} (z)
		= (-1)^r
			\Li_{n_1,\ldots,n_r}\left(
				\frac{\sigma_2}{\sigma_1}, \cdots\!, \frac{\sigma_r}{\sigma_{r-1}}, \frac{z}{\sigma_r} 
			\right)
			\ \text{where}\ 
			w = \letter{0}^{n_r-1}\letter{\sigma_r}^{}\!\!\cdots\letter{0}^{n_1-1}\letter{\sigma_1}^{},
		\label{eq:Hlog-as-Mpl}%
	\end{equation}
	which holds for any $r,n_1,\ldots,n_r \in \N$ and $\sigma_1,\ldots,\sigma_r \neq 0$ when $\abs{z} < \min \set{\abs{\sigma_1},\ldots,\abs{\sigma_r}}$.
	Conversely, in the domain $\bigcap_{k=1}^r \set{\abs{z_k\!\cdots z_r}<1}$ of convergence, any MPL \eqref{eq:def-Mpl} is equal to the hyperlogarithm
$		\Li_{\vec{n}} (\vec{z})
		= \Hyper{w}(1)
$
	associated to the word
	\begin{equation}
		\label{eq:Mpl-as-Hlog}%
		w=(-1)^r \letter{0}^{n_r-1}\letter{1/z_r}\letter{0}^{n_{r-1}-1}\letter{1/(z_r z_{r-1})}\!\cdots\letter{0}^{n_1-1}\letter{1/(z_r\cdots z_1)}.
	\end{equation}
\end{lemma}
This formula follows from \eqref{eq:hlog-zsum}. The special case of just one argument captures the hyperlogarithms with letters $\letter{0}$ and $\letter{1}$: $\HlogAlgebra(\set{0,1})(z) = \Q[\log(z),\Li_{\vec{n}}(z)\colon \vec{n} \in \N^{\times}]$,
\begin{equation}
	\dd \Li_{n_1,\ldots,n_r}(z)
	= \begin{cases}
			\frac{\dd z}{z} \Li_{n_1, \ldots, n_r - 1}(z)	&	\text{if $n_r>1$ and} \\
			\frac{\dd z}{1-z} \Li_{n_1,\ldots, n_{r-1}}(z) & \text{if $n_r=1$.} \\
		\end{cases}
	\label{eq:Mpl-single-variable-differential}%
\end{equation}
For several variables, the total differential \eqref{eq:hlog-total-differential} can be read off directly from \eqref{eq:def-Mpl} and takes the form 
\begin{equation}
	\dd \Li_{n_1,\ldots,n_r}\left( z_1,\ldots,z_r \right)
	=
	\sum_{j=1}^r
		\Li_{n_1,\ldots,n_j-1,\ldots,n_r}\left( z_1,\ldots,z_r \right)
		\ 
		\frac{\dd z_j}{z_j}.
	\label{eq:Mpl-total-differential}%
\end{equation}
Any initial $n_j = 1$ will contribute an index $n_j-1=0$ on the right-hand side, but
\begin{equation*}
	\sum_{k_j = k_{j-1} + 1}^{k_{j+1} - 1}
	\frac{z_j^{k_j}}{k_j^0}
	= \frac{z_j^{k_{j-1} + 1}- z_j^{k_{j+1}}}{1-z^j}
\end{equation*}
shows that such MPL with a zero index $n_j=0$ can be rewritten as
\begin{equation}\begin{split}
	\Li_{n_1,\ldots,0,\ldots,n_r}(z_1,\ldots,z_r)
	=&
	\frac{z_j}{1-z_j} \Li_{n_1,\ldots,\not 0,\ldots,n_r}\left( z_1,\ldots, z_{j-1} z_j, z_{j+1},\ldots, z_r \right)
	\\
	-& \frac{1}{1-z_j} \Li_{n_1,\ldots,\not 0,\ldots,n_r}\left( z_1,\ldots, z_{j-1}, z_j z_{j+1}, \ldots, z_r \right).
	\label{eq:Mpl-zero-index}%
\end{split}\end{equation}
Here the first term reads $\frac{z_1}{1-z_1}\Li_{n_2,\ldots,n_r}(z_2,\ldots,z_r)$ when $j=1$ and the second summand is absent of $j=r$.
\subsubsection{Alternative names}
In the physics literature, the notation $G(w;z) \defas \Hyper{w}(z)$ for hyperlogarithms was introduced in \cite{GehrmannRemiddi:Numerical2dHpl} and is widely used. They are referred to as \emph{Goncharov polylogarithms} (GPL) and also as \emph{generalized harmonic polylogarithms} \cite{AblingerBluemleinSchneider:GeneralizedHarmonicSumsAndPolylogarithms}. But note that in Goncharov's articles \cite{Goncharov:MplCyclotomyModularComplexes,Goncharov:MplMixedTateMotives}, hyperlogarithms are written as
\begin{equation}
	I_{n_1,\ldots,n_r}\left( a_1 : \cdots : a_{r+1} \right)
	\defas
	\Hyper{\letter{0}^{a_r-1} \letter{a_r} \cdots\, \letter{0}^{a_1-1}\letter{a_1}}\left( a_{r+1} \right)
	\quad\text{for}\quad
	a_1,\ldots,a_r \neq 0
	\label{eq:Gonchariv-iterated-integral}%
\end{equation}
or just $I\left( a_1: \cdots :a_{r+1} \right) \defas \Hyper{\letter{a_1}\!\cdots\,\letter{a_r}}\left( a_{r+1} \right)$.

\subsection{Special classes}
\subsubsection{Nielsen's generalized polylogarithms}
Already in \cite{Nielsen:DilogarithmusVerallgemeinerungen}, Nielsen studied generalizations of the classical polylogarithms. Among others, he introduced the functions
\begin{equation}
	S_{n,p}(z)
	\defas
	\frac{(-1)^{n+p-1}}{(n-1)! p!}
	\int_0^1 \log^{n-1} (t)
	\ 
	\log^p(1-zt)
	\frac{\dd t}{t}
	\label{eq:def-Nielsen}%
\end{equation}
which later appeared in perturbative quantum field theory and were rediscovered by particle physicists \cite{KoelbigMignacoRemiddi:NielsenNumerical,Koelbig:Nielsen}. These functions are however nothing but the hyperlogarithms
\begin{equation}
	S_{n,p}(z)
	=
	(-1)^p
	\Hyper{\letter{0}^n \letter{1}^p}(z)
	=
	\Li_{1^{(p-1)},n+1}(z).
	\label{eq:Nielsen-as-Hlog}%
\end{equation}

\subsubsection{Harmonic polylogarithms}
In \cite{RemiddiVermaseren:HarmonicPolylogarithms}, the \emph{harmonic polylogarithms} (HPL) $H(\vec{m}; x)$ of a single variable $x$ where defined for index strings $\vec{m} \in \set{-1,0,1}^{\times}$. Namely, with $0^{(w)}$ abbreviating a sequence of $w$ zeros,
\begin{equation}
	H\left(0^{(w)}; x\right)
	\defas
	\frac{1}{w!} \log^w (x)
	\quad\text{and}\quad
	H\left( a,\vec{m}; x \right)
	\defas
	\int_0^{x}
	f\left(a;x'\right) \,H\!\left( \vec{m}; x' \right) \ \dd x'
	\label{eq:def-Hpl}%
\end{equation}
where $f(0;x) \defas \frac{1}{x}$, $f(1;x) \defas \frac{1}{1-x}$ and $f(-1;x) \defas \frac{1}{1+x}$. Apparently these are hyperlogarithms over the alphabet $\Sigma=\set{-1,0,1}$, explicitly
\begin{equation}
	H(\vec{m}; x)
	=
	(-1)^{\abs{\setexp{k}{m_k=1}}}
		\cdot
		\Hyper{\vec{m}}(x).
	\label{eq:Hpl-as-Hlog}%
\end{equation}
Often a short-hand notation $H_{\vec{m}}(x)$ is used, where the indices $\vec{m} \in \Z^{\times}$ may be arbitrary integers. Then $\pm n \defas 0^{(n-1)},\pm 1$ encodes a sequence of $n-1$ zeros and a single letter $\pm 1$, e.g.\ $H_{3,-2}(x) \defas H(0,0,1,0,-1;x)$. Then
\begin{equation}
	H_{m_1,\ldots,m_r}(x)
	=
	\Hyper{\underline{m_1}\ldots\underline{m_r}}(x)
	\quad\text{where}\quad
	\underline{0} \defas \letter{0}
	\quad\text{and}\quad
	\underline{\pm m} \defas \mp \letter{0}^{m-1} \letter{\pm 1}
	\quad\text{for}\quad
	m\in\N.
	\label{eq:compressed-Hpl-as-Hlog}%
\end{equation}

\subsubsection{Cyclotomic harmonic polylogarithms}
Hyperlogarithms with letters $\sigma \in \Sigma = \set{0} \cupdot \setexp{e^{2\pi k/n}}{0 \leq k < n}$ that are roots of unity have attracted special interest (see also section~\ref{sec:Periods}). They have been called \emph{cyclotomic harmonic polylogarithms} in \cite{AblingerBluemleinSchneider:Cyclotomic}.

\subsubsection{Two-dimensional harmonic polylogarithms}
The two-loop four-point functions with one leg off-shell were calculated in \cite{GehrmannRemiddi:TwoLoopMasterIntegralsPlanar,GehrmannRemiddi:TwoLoopMasterIntegralsNonPlanar}. These depend not on one, but two dimensionless variables called $y$ and $z$. To express their results, the authors introduced the special family of hyperlogarithms
\begin{equation}
	G(w; y)
	\defas
	\Hyper{w}(y)
	\quad\text{with}\quad
	w \in \set{0,1,1-z,-z}^{\times}
	\label{eq:def-2dHpl}%
\end{equation}
called \emph{two-dimensional harmonic polylogarithms} (2dHPLs). This notation was fixed in \cite{GehrmannRemiddi:Numerical2dHpl}, were one also finds explicit formulas that express these functions for weight less than four in terms of Nielsen's generalized polylogarithms \eqref{eq:def-Nielsen}.

\subsection{Single-valued polylogarithms}
The monodromies of hyperlogarithms $\Hyper{w}(z)$ when $z$ encircles a singularity $\sigma \in \Sigma$ can be removed through suitable combinations with hyperlogarithms of the complex conjugate $\conjugate{z}$. The Bloch-Wigner dilogarithm $\BlochWigner$ from \eqref{eq:BlochWigner} is a very important example, and Francis Brown studied the full subalgebra of $\HlogAlgebra(\set{0,1})(z) \tp \HlogAlgebra(\set{0,1})(\bar{z}) \tp \MZV$ characterized by this property of single-valuedness on $\C \setminus \set{0,1}$ when $\bar{z}=\conjugate{z}$ are conjugate \cite{Brown:Uniformes}. They have been called \emph{single-valued multiple polylogarithms} (SVMP(L)) and also \emph{single-valued harmonic polylogarithms} (SVHPL).

Such functions occur naturally in certain Feynman integrals and provide a very efficient tool for practical computations \cite{Drummond:Ladders,Schnetz:GraphicalFunctions} (for example, they are a key ingredient to the proof of the zigzag conjecture \cite{BrownSchnetz:ZigZag}). Their special values at one are well understood \cite{Brown:SingleValuedMZV}.

But Feynman integrals demand more general functions, as was realized for the first time in \cite{ChavezDuhr:Triangles} and later for example also in \cite{DDEHPS:LeadingSingularitiesOffShellConformal}. In both cases, the differential form $\dd \log(z-\bar{z})$ had to be adjoined to form more general (but still single-valued) integrals. Our calculation of graphical functions (together with Oliver Schnetz) revealed even more general (single-valued) iterated integrals, involving also the forms $\dd \log(z \bar{z} - z - \bar{z})$, $\dd \log(1-z\bar{z})$ and $\dd \log(1-z-\bar{z})$. This stimulated a wide extension of the concept of single-valued hyperlogarithms, which is currently actively developed by Oliver Schnetz. Some examples of our results can be found in \cite{Panzer:DivergencesManyScales}, and we briefly comment on them also in section~\ref{sec:ex-3pt}.

\section{Periods}
\label{sec:Periods}%
In the previous section, we introduced the absolutely convergent integrals $\Hyper{\WordReg{0}{\infty}(w)}(\infty)$ as basic building blocks for our algorithms. These depend on the position of the letters $\sigma \in \Sigma$ as described by proposition~\ref{prop:reginf-as-hlog}.

But for fixed points $\Sigma \subset \C$, they just define constants which belong to the class of \emph{periods} \cite{Periods} when $\Sigma \subset \overline{\Q}$ is algebraic. This is obvious from \eqref{eq:hlog-finite-infinity-extended}, because periods may be defined as numbers that admit an integral representation with rational integrands over integration domains determined by rational inequalities, like $0<z_1<\cdots<z_n$.

Since such numbers appear en masse when we employ hyperlogarithms for integration, it is crucial to understand them very well. The most important aspects for practical applications are:
\begin{enumerate}
	\item For a fixed alphabet $\Sigma \subset \C$, the number of words $w \in \Sigma^n$ grows exponentially with the weight $n=\abs{w}$. But the constants $\Hyper{\WordReg{0}{\infty}(w)}(\infty)$ typically obey a huge number of relations, such that we can express all of them as $\Q$-linear combinations of only a few suitably chosen generators, which considerably reduces the size of the output.

	\item One wants to be able to detect if an expression is zero, in an automated way, preferably without resorting to numeric evaluations. This is only possible when the set of generators is linearly independent over $\Q$.
\end{enumerate}
Unfortunately, a basis over $\Q$ is not known in any case of interest; for example it is still not ruled out that $\MZV = \Q[\pi^2]$: all multiple zeta values could be polynomials in $\pi^2$ \cite{Waldschmidt:LecturesMZV}. On the other hand, the development of \emph{motivic periods} recently blossomed into several complete characterizations of the algebra of these analogues of our actual periods, for particular choices of $\Sigma$. Bases of motivic periods map to generating sets of the associated actual periods, and the main conjectures of the theory of periods imply that indeed they should stay linearly independent and form a basis over $\Q$.

We cannot go into detail on this fascinating and very active subject, but only collect results and references which are particularly important for our applications.
\subsection{Multiple zeta values and alternating sums}
\begin{definition}\label{def:MZV-N}
	Let $\mu_N \defas \big\{\xi\colon \xi^N = 1\big\}$ denote the $N$'th roots of unity and set
	\begin{equation}
		\MZV[N]
		\defas
		\lin_{\Q} \setexp{
			\Li_{n_1,\ldots,n_r}(\xi_1,\ldots,\xi_r)
		}{
			\text{every $n_i \in \N$, $\xi_i \in \mu_N$ and $(n_r,\xi_r) \neq (1,1)$}
		}
		\label{eq:def-MZV-N}%
	\end{equation}
	to be all rational linear combinations of multiple polylogarithms evaluated at such roots.\footnote{The condition $(n_r,\xi_r) \neq (1,1)$ is equivalent to the convergence of the series \eqref{eq:def-Mpl}.} It is filtered by the weight $n_1+\cdots+ n_r$ and the depth $r$.
\end{definition}
Equivalently we can characterize this space as the algebra of special values that hyperlogarithms over the alphabet $\set{0} \cup \mu_N$ take at one: From \eqref{eq:Mpl-as-Hlog},
\begin{equation}
	\AnaReg{z}{1} \HlogAlgebra(\set{0} \cup \mu_N)(z)
	\urel{\eqref{eq:hlog-path-concatenation-singular}}
	\int_0^1 \WordReg{0}{1} T(\set{0} \cup \mu_N)
	= \MZV[N]
	\label{eq:MZV-N-as-hlog}%
\end{equation}
where the integration path is the straight line from zero to one.
\begin{remark}
	More generally, we can consider all convergent integrals along smooth paths $\gamma\colon(0,1) \longrightarrow \C\setminus\Sigma$ with endpoints $\gamma(0),\gamma(1) \in \Sigma$, the \emph{convergent periods} of $\C \setminus \Sigma$:
\begin{equation}
	\period(\C \setminus \Sigma)
	\defas
	\sum_{\gamma}
	\setexp{\int_{\gamma} w}{
		\text{$w \in T(\Sigma)$ neither starts with $\letter{\gamma(1)}$ nor ends in $\letter{\gamma(0)}$}
	}.
	\label{eq:def-relative-periods}%
\end{equation}
From \eqref{eq:Moebius-transformation} and \eqref{eq:hlog-path-concatenation-singular} it follows that for roots of unity $\Sigma = \set{0} \cup \mu_N$, the only new period we can get is $\imag\pi$, so $\period(\C\setminus\Sigma) = \MZV[N][\imag\pi]$.
\end{remark}
The case of multiple zeta values $\MZV \defas \MZV[1]$ was studied and finally understood motivically by Francis Brown \cite{Brown:MixedTateMotivesOverZ}. With his results, the conjectures on periods would imply the existence of an isomorphism $\MZV \isomorph \Q[\pi^2,\Lyndons(\set{3,5,7,\ldots})]$ of weight-graded algebras. A concrete result settles a conjecture of Hoffman and provides a small set of generators.
\begin{theorem}
	\label{theorem:HoffmanBasis}%
	Multiple zeta values are spanned by the Hoffman elements
	\begin{equation}
		\MZV = \Q + \lin_{\Q} \setexp{\mzv{n_1,\ldots,n_r}}{\text{all}\ n_1,\ldots,n_r \in \set{2,3}}.
		\label{eq:HoffmanBasis}%
	\end{equation}
\end{theorem}
For our automated computations it is important to make such a statement effective, which means that we need an efficient method to determine explicitly a reduction of any multiple zeta value to this (or any other) conjectural basis. This can be achieved by the coproduct-based algorithm \cite{Brown:DecompositionMotivicMZV} which is available as a program \cite{Schnetz:ZetaProcedures}.

But the generators of Hoffman are not optimal in that they have very high depth for a given weight, and we like to express results with smallest depth possible. The conjecture due to Broadhurst and Kreimer \cite{BroadhurstKreimer:MZVPositiveKnots} on the depth filtration of MZV is still open even motivically \cite{Brown:DepthGraded}. For our practical purposes though, we only consider small weights (so far we did not exceed weight $11$ in any of our computations) and can therefore harvest the \emph{data mine} \cite{BluemleinBroadhurstVermaseren:Datamine}, which provides proven reductions to a (conjecturally) depth-minimal set of generators.

The data mine also covers alternating sums $\MZV[2]$, for which we know \cite{Deligne:GroupeFondamentalMotiviqueN}
\begin{theorem}
	\label{theorem:Deligne-N=2}%
	Every alternating sum is a rational linear combination of products of $\pi^{2p}$ and $\Li_{n_1,\ldots,n_r}(1,\ldots,1,-1)$ for Lyndon words with odd indices $n_i \in \OddN \defas \set{1,3,5,\ldots}$, with the same weight and at most the same depth.\footnote{Note that any power $\pi^{2p}$ ($p>0$) has depth 1 here (because it is a rational multiple of $\Li_{2p}(1)$).} In particular we can write
	\begin{equation}
		\MZV[2] = \Q\left[
			\pi^2,
			\Li_{n}(1,\ldots,1,-1)\colon n=(n_1,\ldots,n_r) \in \Lyndons(\OddN)
		\right].
		\label{eq:Alternating-Generators}%
	\end{equation}
\end{theorem}
Conjecturally, there are no further relations and $\MZV[2] \isomorph \Q[\pi^2] \tp T(\OddN)$ is an isomorphism of weight-graded and depth-filtered algebras. Note that \eqref{eq:Alternating-Generators} and \eqref{eq:Mpl-as-Hlog} show that each alternating sum is a hyperlogarithm over $\set{0,1}$ at $z=-1$:
\begin{equation*}
	\MZV[2] 
	= \AnaReg{z}{-1} \HlogAlgebra(\set{0,1})(z)
	= \int_0^{-1} T(\set{0,1}) \letter{1}.
\end{equation*}
We encounter alternating sums in many Feynman integral computations, but in the final answers for massless integrals these always combined to just multiple zeta values $\MZV$. We give an example of this phenomenon in section~\ref{sec:P79}.

\subsection{Primitive sixth roots of unity}
For one period we computed (see section~\ref{sec:P711}), the space $\MZV[2]$ was not enough and we needed sixth roots of unity. Far less data is available on these sums than in the previous cases and no table of reductions to a conjectural basis exists for high weights.\footnote{Very recently, a data mine for the Deligne subalgebra $\Deligne$ became available \cite{Broadhurst:Aufbau}.}
A detailed analysis up to weight four was performed by David Broadhurst in \cite{Broadhurst:SixthRoots}, who observed that Feynman integrals tend to lie in very special subspaces of $\MZV[6]$. One of them is by now well understood due to Deligne \cite{Deligne:GroupeFondamentalMotiviqueN}.
Let $\xi_6 \defas e^{\imag\pi/3}$ denote a primitive sixth root of unity and $\conjugate{\xi_6} = \xi_6^{-1} = 1-\xi_6$ its conjugate. We quote
\begin{theorem}
	\label{theorem:Deligne-N=6}%
	Define $\Deligne$ as the subalgebra of $\MZV[6]$ generated by $\Li_{n_1,\ldots,n_r}(z_1,\ldots,z_r)$ with $z_1,\ldots,z_r \in \set{1,\xi_6,\conjugate{\xi_6}}$ such that all products $\setexp{z_k\cdots z_r}{1\leq k \leq r}$ are contained either in $\set{1,\xi_6}$ or in $\set{1,\conjugate{\xi_6}}$, and $(n_r,z_r)\neq(1,1)$.

	Then each element of $\Deligne$ is a rational linear combination of products of $\imag\pi$ and $\Li_{n_1,\ldots,n_r}(1,\ldots,1,\xi_6)$ for Lyndon words (with $n_1,\ldots,n_r>1$), with at most the same total weight and depth (any power $(\imag\pi)^p$ with $p>0$ has depth $1$, see below):
	\begin{equation}
		\Deligne
	= \Q\left[ (\imag\pi), \Li_n(1,\ldots,1,\xi_6)\colon n=(n_1,\ldots,n_r) \in \Lyndons(\N \setminus \set{1}) \right].
		\label{eq:Deligne-Generators}%
	\end{equation}
\end{theorem}%
\nomenclature[Z 6 D]{$\Deligne$}{Deligne's subalgebra of $\MZV[6]$, theorem~\ref{theorem:Deligne-N=6}\nomrefpage}%
Again the main conjectures imply an isomorphism $\Deligne \isomorph \Q[\imag\pi] \tp T(\N \setminus \set{1})$, respecting weight and depth. As for alternating sums, this is an algebra of special values of hyperlogarithms (integrating along the straight line from $0$ to $\xi_6$):
\begin{equation}
	\Deligne = \HlogAlgebra(\set{0,1})(\xi_6)
	= \Q[\imag\pi] \int_0^{\xi_6} T(\set{0,1}) \letter{1}.
	\label{eq:Deligne-as-Mpl}%
\end{equation}
Through its definition, $\Deligne \supseteq \MZV$ contains the multiple zeta values and is closed under complex conjugation. For our application we needed to determine the real and imaginary parts of $\Deligne$ separately. In depth one we know \cite[chapter VII, section 5.3]{Lewin:PolylogarithmsAssociatedFunctions}
\begin{equation}
	\Li_n \left( e^{2\pi \imag x} \right)  + (-1)^n \Li_n\left( e^{-2\pi\imag x} \right)
	= - \frac{(2\pi\imag)^n}{n!} B_n(x)
	\label{eq:Li-bernoulli}%
\end{equation}
in terms of the rational Bernoulli polynomials $B_n(x)$, so any power of $\imag\pi$ has depth one and $\Realteil(\Li_{2n}(\xi_6))$ and $\imag\Imaginaerteil(\Li_{2n+1}(\xi_6))$ lie in $\Q[\imag\pi]$. On the other hand, the complementary $\Imaginaerteil(\Li_{2n}(\xi_6))$ and $\Realteil(\Li_{2n+1}(\xi_6))$ are expected to be transcendental constants independent of $\pi$. Indeed, already Lewin noticed that \cite[chapter VII, section 3.3]{Lewin:PolylogarithmsAssociatedFunctions}
\begin{equation}
	\Realteil \Big( \Li_{2n+1}(\xi_6) \Big)
	= \frac{1}{2} \left(1-2^{-2n}\right)\left( 1-3^{-2n} \right) \mzv{2n+1}.
\end{equation}
We generalize this parity result to all multiple polylogarithms in
\begin{proposition}\label{prop:Deligne-oddeven}%
	Abbreviate $\Li_{\vec{n}}(\xi_6) \defas \Li_{n_1,\ldots,n_r}(1,\ldots,1,\xi_6)$ and write $\abs{\vec{n}} \defas n_1+\cdots+ n_r$ for its weight. Then Deligne's subalgebra coincides with
	\begin{equation}
		\Deligne
		= \Q\left[(\imag\pi),
			\imag^{r+\abs{\vec{n}}}\Realteil\Big(\imag^{r+\abs{\vec{n}}}\Li_{\vec{n}}(\xi_6)\Big)\colon
			\vec{n}=(n_1,\ldots,n_r) \in \Lyndons\left( \N\setminus\set{1} \right)
		\right]
		\label{eq:Deligne-oddeven}%
	\end{equation}
	and every $\Li_{\vec{n}}(\xi_6)$ has a representation as a polynomial in these generators with less or equal weight and depth (where any power $(\imag\pi)^p$, $p>0$, is understood to have depth $1$).\footnote{So in particular, $\Realteil(\Li_{\vec{n}}(\xi_6))$ is expressible in terms of words with lower depth than $\vec{n}$ and products of lower weight, if $\vec{n}$ has weight and depth of different parity. The analogue holds for the imaginary parts, see also \eqref{eq:mpl-xi6-ReIm-reducible}.}
\end{proposition}
These generators $\Realteil\left( \Li_{\vec{n}}(\xi_6) \right)$ ($r+\abs{\vec{n}}$ even) and $\imag \Imaginaerteil\left( \Li_{\vec{n}}(\xi_6) \right)$ ($r+\abs{\vec{n}}$ odd) have the benefit that their products split into generators (conjecturally bases) of the subspaces $\Deligne = \Realteil \Deligne \oplus \imag \Imaginaerteil \Deligne$.

For the proof we need the well-known parity theorem on multiple zeta values \cite{Tsumura:CombinatorialEulerZagier,IharaKanekoZagier:DerivationDoubleShuffle}.
\begin{theorem}
	\label{theorem:MZV-parity}%
	Any multiple zeta value $\mzv{n_1,\ldots,n_r}$ with $r+\abs{\vec{n}}$ odd (except $\mzv{2}$) is a rational linear combination of MZV of depth $<r$ and products of MZV of weight $<\abs{\vec{n}}$.
\end{theorem}
If we write $\WeightDepth{N}{d} \MZV \defas \lin_{\Q}\setexp{\mzv{n_1,\ldots,n_r}}{\abs{\vec{n}} \leq N\ \text{and}\ r \leq d} \subset \MZV$ for the subspace of MZV with weight $\leq N$ and depth $\leq d$, this theorem says $\WeightDepth{N}{d}\MZV \subseteq \WeightDepth{N}{d-1} \MZV + (\WeightDepth{N-1}{d} \MZV)^2$ for $d+N$ odd.
In fact, the statement is more precise and every time we write just $(\WeightDepth{N-1}{d} \MZV)^2$ we actually mean the more refined combined weight-depth filtration
\begin{equation*}
	\sum_{\substack{N'+N'' \leq N \\ N', N'' < N}} \sum_{d' + d'' \leq d} \Big(\WeightDepth{N'}{d'}\MZV \Big) \cdot \Big( \WeightDepth{N''}{d''} \MZV \Big).
\end{equation*}
The product terms that occur in our proofs manifestly obey this strong form.
\begin{corollary}
	Let $ w = \letter{0}^{n_r-1}\letter{1}\!\cdots\letter{0}^{n_1-1}\letter{1}$ with weight $N=n_1+\cdots+n_r$. Then
	\begin{equation}
		\AnaReg{z}{\infty} \int_0^z \!\!w
		\in \WeightDepth{N}{r-1} \MZV + \Big( \WeightDepth{N-1}{r} {\MZV}[\imag \pi] \Big)^2.
		\label{eq:MZV-parity-0-inf}%
	\end{equation}
\end{corollary}
\begin{proof}
	From lemma~\ref{lemma:hlog-path-concatenation-singular} we know $\int_0^z = \int_{1}^{z} \convolution \int_0^1 \WordReg{}{1}$ where $\int_{1}^{z} \letter{1} = \log(1-z)$ maps into $\Q[\imag \pi]$ under $\AnaReg{z}{\infty}$. So up to products, we can replace $\int_1^z$ with $\int_1^z \WordReg{1}{}$ and find
	\begin{equation*}
		\AnaReg{z}{\infty} \int_0^z \!\! w
	\equiv \left[ \AnaReg{z}{\infty} \int_1^z \!\!\!\WordReg{1}{} \convolution \int_0^1 \!\!\!\WordReg{}{1}\right] (w)
		\equiv \AnaReg{z}{\infty}\int_1^z \!\!\!\WordReg{1}{} (w) + \int_0^1 \!\!\!\WordReg{}{1}(w)
		\mod \Big(\WeightDepth{N-1}{r} {\MZV}[\imag\pi]\Big)^2.
	\end{equation*}
	We apply the inversion $f(z) = z^{-1}$ to the first summand, which maps $w$ to $\WordTransformation{f}(w) = (-\letter{0})^{n_r-1}(\letter{1}-\letter{0})\cdots(-\letter{0})^{n_1-1}(\letter{1}-\letter{0}) = (-1)^{r+N}w + R$, where all words in $R$ have depth (number of letters $\letter{1}$) less than $r$. Therefore
	\begin{equation*}
		\AnaReg{z}{\infty} \int_0^z \!\! w
		\equiv (-1)^{r+N} \int_1^{0} \WordReg{1}{0} (w) + \int_0^1 \WordReg{}{1} (w)
		\mod \WeightDepth{}{} 
		\defas \WeightDepth{N}{r-1}\MZV + \Big(\WeightDepth{N-1}{r} {\MZV}[\imag\pi]\Big)^2.
	\end{equation*}
	If $r+N$ is odd, theorem~\ref{theorem:MZV-parity} applies to both summands and we are done. Otherwise, we use $\AnaReg{z}{0} \int_0^z = \int_1^0 \WordReg{1}{0} \convolution \int_0^1 \WordReg{}{1} \equiv \int_1^0 \WordReg{1}{0} + \int_0^1 \WordReg{}{1} \mod \WeightDepth{}{}$ to conclude
	\begin{equation*}
		\AnaReg{z}{\infty} \int_0^z w
		\equiv \AnaReg{z}{0} \left[ \int_1^z\!\! \WordReg{1}{} \convolution \int_0^1 \!\!\WordReg{}{1} \right] (w)
		\equiv \AnaReg{z}{0} \int_0^z \! w
		\urel[\equiv]{\ref{corollary:hlog-reg0}} 0
		\mod \WeightDepth{}{}. \qedhere
	\end{equation*}
\end{proof}
This says that a regularized limit at infinity $\AnaReg{z}{\infty} \Hyper{w}(z)$ of a multiple polylogarithm, $w \in T(\set{0,1})$, is always reducible into products of MZV and MZV of lower depth. Note that our proof applies also to the alphabet $\set{-1,0}$, with the only change that we split the integration at $-1$ instead of $1$. In this case we can dispose of the $\imag\pi$ in \eqref{eq:MZV-parity-0-inf}, as we understand the limit $z \rightarrow \infty$ to the right:\footnote{Note that $\AnaReg{z}{\infty} \Hyper{\letter{0}\letter{-1}}(z) = \mzv{2}$ is not considered a counterexample here, because $\mzv{2} = \pi^2/6$ is a product (even though $\pi \notin \MZV$, conjecturally).}
	\begin{equation}
		\AnaReg{z}{\infty} \int_0^z \!\! \letter{0}^{n_r-1}\letter{-1}\!\cdots\letter{0}^{n_1-1}\letter{-1}
		\in \WeightDepth{N}{r-1} \MZV + \left( \WeightDepth{N-1}{r} {\MZV} \right)^2
		\ \text{where}\ 
		N=n_1+\cdots+n_r.
		\label{eq:MZV-parity-0-inf-1}%
	\end{equation}
\begin{proof}[Proof of proposition~\ref{prop:Deligne-oddeven}]
	Consider any $\Li_{\vec{n}}(\xi_6) =(-1)^{r} \int_0^{\xi_6} w$ where $w = \letter{0}^{n_r-1}\letter{1}\!\cdots\letter{0}^{n_1-1}\letter{1}$ of weight $\abs{\vec{n}}$ and let again $\WeightDepth{}{} \defas \WeightDepth{\abs{\vec{n}}}{r} \Deligne + \left( \WeightDepth{\abs{\vec{n}}-1}{r} \Deligne \right)^2$ denote the subspace of elements with smaller depth and products of elements with lower weight. We apply the inversion $f(z)= z^{-1}$ to the complex conjugate and then split with lemma~\ref{lemma:hlog-path-concatenation-singular} at zero:
	\begin{equation*}
		(-1)^r\Li_{\vec{n}}(\conjugate{\xi_6})
		= \int_0^{1/\xi_6} \!\!\!w
		= \int_{\infty}^{\xi_6} \WordTransformation{f}(w)
		= \sum_{(w)} \int_0^{\xi_6} \WordTransformation{f}\left(w_{(1)}\right) \AnaReg{z}{0} \int_{\infty}^{z} \WordTransformation{f}\left(w_{(2)} \right).
	\end{equation*}
	Lemma~\ref{lemma:hlog-path-concatenation-singular} applies to the second factor,%
\footnote{An auxiliary split 
		$\int_{\infty}^z = \int_u^z \convolution \int_{\infty}^u$
		shows $\AnaReg{z}{0} \int_{\infty}^z = \int_u^0 \WordReg{}{0} \convolution \int_{\infty}^u$,
		then consider $z\defas u \longrightarrow \infty$ to conclude $\AnaReg{z}{0} \int_{\infty}^z = \AnaReg{z}{\infty} \int_z^0 \WordReg{}{0} = \AnaReg{z}{\infty} \int_0^z \WordReg{0}{} \antipode$ for the antipode $S$.}
		so 	\begin{equation*}
		\Li_{\vec{n}}(\conjugate{\xi_6})
		\equiv (-1)^r \int_0^{\xi_6} \WordTransformation{f}(w)
		\equiv (-1)^r (-1)^{r+\abs{\vec{n}}} \int_0^{\xi_6} w
		\equiv (-1)^{r+\abs{\vec{n}}} \Li_{\vec{n}}(\xi_6)
		\mod \WeightDepth{}{}
	\end{equation*}
	since $\WordTransformation{f}(w) \equiv (-1)^{r+\abs{\vec{n}}} w$ up to words of depth $< r$. Therefore, depending on the parity of $r+\abs{\vec{n}}$ either the real- or imaginary parts are reducible:
	\begin{equation}
		\WeightDepth{}{}
		\ni
		\Li_{\vec{n}}(\xi_6) - (-1)^{r+\abs{\vec{n}}} \Li_{\vec{n}}(\conjugate{\xi_6})
		= \begin{cases}
			2 \imag \Imaginaerteil \Li_{\vec{n}}(\xi_6) & \text{when $r+\abs{\vec{n}}$ is even and} \\
			2 \Realteil \Li_{\vec{n}}(\xi_6) & \text{when $r+\abs{\vec{n}}$ is odd.} \\
		\end{cases}
		\label{eq:mpl-xi6-ReIm-reducible}%
	\end{equation}
	The claim follows by induction over the weight and depth. 
\end{proof}

\section{Multiple integrals of hyperlogarithms}
\label{sec:multiple-integrals} %
The techniques presented in section~\ref{sec:hlog-algorithms} suffice to compute multivariate integrals of polylogarithms in special cases. We review this idea, which was developed by Francis Brown \cite{Brown:MZVPeriodsModuliSpaces,Brown:TwoPoint,Brown:PeriodsFeynmanIntegrals}, and comment on some differences to the univariate case. In particular, with several variables the class of polylogarithms with rational arguments and prefactors is not closed any more under indeterminate integration.

We are therefore forced to study the condition of \emph{linear reducibility} and recall the method of compatibility graphs. For our purposes, we introduce an adapted algorithm for polynomial reduction that is particularly apt to handle recursive integral representations like the ones we derived in sections~\ref{sec:3pt-recursions} and \ref{sec:ladderboxes}. As an application we prove the main theorems of this thesis in section~\ref{sec:recursion-reducibility}.

Introductions to the topics of this section are available in \cite{Brown:ModuliSpacesFeynmanIntegrals,Brown:TwoPoint} and the thesis \cite{Golz:EvaluationTechniques}. An experimental study of linear reducibility for $4$-point Feynman integrals in the on-shell case was reported in \cite{BognerLueders:MasslessOnShell,Lueders:LinearReduction} and further insights into iterated integrals in several variables (with their application to integration) may be found in \cite{BognerBrown:GenusZero,BognerBrown:SymbolicIntegration}.

\subsection{Partial integrals}
Suppose we want to compute an absolutely convergent, $N$-dimensional integral
\begin{equation}
	I_N \defas \int_{(0,\infty)^N} I_0\ \dd z_1 \wedge \ldots \wedge \dd z_N < \infty
	\label{eq:total-integral}%
\end{equation}
of an analytic integrand $I_0(z_1,\ldots,z_N)$. By Fubini's theorem, we may iterate the \emph{partial integrals}
\begin{equation}
	I_k (z_{k+1},\ldots,z_N)
	\defas
	\int_{(0,\infty)^{k}} I_0\ \dd z_1 \wedge \ldots \wedge \dd z_k
	= \int_0^{\infty} I_{k-1}\ \dd z_k,
	\label{eq:partial-integral}%
\end{equation}
which are (locally) analytic\footnote{%
In discussions we realized that this fact is not taught everywhere. We learned it from \cite[theorem~11]{Sauvigny:PDE1}.}
in the remaining variables.
For our methods to apply, these must be hyperlogarithms in the next integration variable. More precisely, we need
\begin{definition}\label{def:lin-reducible-integrand}
	An integrand $I_0(z_1,\ldots,z_N)$ is called \emph{linearly reducible} if \eqref{eq:total-integral} is finite and (after rearranging the variables if necessary) there exists a family $0 \in \Sigma_k \subset \Q(z_{k+1},\ldots,z_N)$ of rational alphabets such that for each $\ 0 \leq k < N$, $I_k$ is a product of hyperlogarithms of the form
\begin{equation}
	I_{k}
	\in
	\alg_{>k}
	\defas
		\C \tp
		\bigotimes_{j=k+1}^N
		\alg_j
	\ \text{where}\ 
	\alg_j
	\defas
		\regulars(\Sigma_j)(z_{j}) 
		\tp \HlogAlgebra(\Sigma_j)(z_{j}).
	\label{eq:partial-integral-hlog-product}%
\end{equation}
\end{definition}
These conditions precisely ensure that we can carry out each of the integrals $I_k = \int_{0}^{\infty} I_{k-1} (z_k)\ \dd z_k$ with the procedures of section~\ref{sec:hlog-algorithms}, along the following steps:
\begin{enumerate}
	\item
		Compute an antiderivative $F \in \regulars^{+}(\Sigma_{k-1})(z_{k}) \tp \HlogAlgebra(\Sigma_{k-1})(z_{k}) \tp \alg_{>k}$ of $I_{k-1}$ with lemma~\ref{lemma:hlog-primitives}, such that $\partial_{z_k} F = I_{k-1}$.
		
	\item 
		Expand $F$ at $z_k\rightarrow 0$ ($\infty$) as the series \eqref{eq:log-laurent-expansion} in $z_k$ ($z_k^{-1}$) and $\log (z_k)$, using \eqref{eq:hlog-zero-logseries} and \eqref{eq:reginf-expansion-convolution}. Since convergence is granted, all divergent terms must cancel and we obtain an explicit representation of the integral of the form
\begin{equation*}
	I_k
	\in \Q\left[
	\sigma_i,\frac{1}{\sigma_i-\sigma_j}\colon \sigma_i,\sigma_j \in \Sigma_{k-1}\ \text{and}\ \sigma_i \neq \sigma_j
			\right]
			\tp \AnaReg{z_k}{\infty} \HlogAlgebra(\Sigma_k)(z_k)
			\tp \alg_{>k}.
\end{equation*}

	\item
		Use proposition~\ref{prop:reginf-as-hlog} to write this limit as a hyperlogarithm in $z_{k+1}$,
\begin{equation*}
	\AnaReg{z_k}{\infty} \HlogAlgebra(\Sigma_k)(z_k)
	\subseteq
	\HlogAlgebra \Big( (\Sigma_k)_{z_{k+1}} \Big)(z_{k+1})
	\tp \AnaReg{z_k}{\infty} \HlogAlgebra\Big(\leadCoeff_{z_{k+1}} (\Sigma_k) \Big)(z_k)
\end{equation*}
with the rational leading coefficients\footnote{If a pinch occurs, we have to include further rational letters from subleading coefficients according to corollary~\ref{corollary:reglim-reginf-pinch}.}
$\Sigma_{k,k+1} \defas \leadCoeff_{z_{k+1}}(\Sigma_k) 
	\subset \Q(z_{k+2},\ldots,z_N)$.
Iterate
$\Sigma_{k,j+1} 
	\defas \leadCoeff_{z_{j+1}}(\Sigma_{k,j}) 
	\subset \Q(z_{j+2},\ldots,z_N)$
until this limit is decomposed into a product of hyperlogarithms in all variables, giving an element of
	\begin{equation}
		\HlogAlgebra \Big( (\Sigma_k)_{z_{k+1}} \Big)(z_{k+1})
		\tp \cdots \tp
		\HlogAlgebra \Big( (\Sigma_{k,N-1})_{z_{N}} \Big)(z_{N})
		\tp \AnaReg{z_k}{\infty} \HlogAlgebra\Big(\Sigma_{k,N} \Big)(z_k).
		\label{eq:integration-step-hlog-basis}%
	\end{equation}

\item[3.'] Project each of these hyperlogarithm algebras $\HlogAlgebra( (\Sigma_{k,j})_{z_{j+1}})$ onto $\HlogAlgebra( \Sigma_{j+1} \cap (\Sigma_{k,j})_{z_{j+1}})$ by mapping each word which contains a letter not in $\Sigma_{j+1}$ to zero.

\item Use the shuffle product to multiply these two elements $\alg_{>k} \tp \alg_{>k} \longrightarrow \alg_{>k}$, where the first factor is the one constructed in steps 2 and 3' from the integration of $z_k$ and the second factor is the part of $I_{k-1} \in \alg_k \tp \alg_{>k}$ we ignored so far.
\end{enumerate}
After these steps, we have computed an explicit representation of the partial integral $I_k$ as an element of the algebra $\alg_{>k}$. It is step 3. where we needed the rationality of $\Sigma_k$.

Note that in general the letters $\Sigma_{k,j} \subset \overline{\Q(z_{j+1},\ldots,z_N)}$ from \eqref{eq:def-alphabet-parameter-reduced} will be algebraic and not rational, which is where the constraint \eqref{eq:partial-integral-hlog-product} of linear reducibility comes into play. It guarantees that all such non-rational letters must drop out of the result. Given the linear independence of lemma~\ref{lemma:hyperlog-independence}, the projection in step 3' is therefore the identity map and seems superfluous at first. But in practice a lot of time and effort is saved by projecting out all such contributions immediately when they arise (see chapter~\ref{chap:hyperint}).

In the same way, the prefactors $\Q[\sigma_i,1/(\sigma_i - \sigma_j)] \subseteq \Q(z_{k+1},\ldots,z_N)$ are rational functions that do not in general factorize with respect to each variable, so
\begin{equation}
	\Q(z_{k+1},\ldots,z_N)
	\subset
	\regulars (\overline{\Q(z_{k+2},\ldots,z_N)})(z_{k+1})
	\tp \regulars(\overline{\Q(z_{k+3},\ldots,z_N)})(z_{k+2})
	\tp\cdots \tp \regulars(\overline{\Q})(z_N)
	\label{eq:rationals-product-basis}%
\end{equation}
can have irrational (but algebraic) poles with respect to $z_{k+1}$ and so on. Again, the criterion \eqref{eq:partial-integral-hlog-product} asserts that the integrand $I_k$ will have its rational prefactors inside the subalgebra $\regulars(\Sigma_{k+1})(z_{k+1}) \tp \cdots \tp \regulars(\Sigma_N)(z_N)$ of \eqref{eq:rationals-product-basis}.

For the period of the wheel $\WS{3}$, the above algorithm was worked out in detail in \cite{Brown:TwoPoint} and furthermore in the thesis \cite{Golz:EvaluationTechniques}. Their method differs though in that the limits $\AnaReg{z_{k+1}}{0} \AnaReg{z_k}{\infty} \Hyper{w}(z_k)$ are not replaced according to proposition~\ref{prop:reglim-reginf-from-word} but rather treated symbolically and kept till the very end.

\subsubsection{The final period}
For a linearly reducible integrand, the above procedure provides an upper bound on the space of periods that the integral $I_N$ must evaluate to. In the $k$'th integration step, the decomposition \eqref{eq:integration-step-hlog-basis} into products of hyperlogarithms leaves a period in the algebra
\begin{equation*}
	\AnaReg{z_N}{0} \cdots \AnaReg{z_{k+1}}{0} \AnaReg{z_k}{\infty} \HlogAlgebra(\Sigma_k)(z_k)
	\subseteq \AnaReg{z}{\infty} \HlogAlgebra(\Sigma_{k,N}).
\end{equation*}
If the initial integrand is defined over $\Q$ (remove the product with $\C$ in \eqref{eq:partial-integral-hlog-product} for $k=0$), we collect these integration constants to represent the total integral as the period
\begin{equation}
	I_N
	\in \prod_{k=1}^N \AnaReg{z_k}{\infty} \HlogAlgebra(\Sigma_{k,N})(z_k).
	\label{eq:final-period-algebra}%
\end{equation}
This characterization can be refined if we iterate corollary~\ref{corollary:reglim-degree-decouple} at each integration step to disentangle letters with different vanishing degrees.

\subsubsection{Testing the criterion}
To check if an integrand $I_0$ is linearly reducible, we can run the above algorithm without the projection step 3': Linear reducibility is disproven as soon as a non-rational letter remains in a hyperlogarithm in the representation \eqref{eq:integration-step-hlog-basis} or if a non-rational pole occurs in the form \eqref{eq:rationals-product-basis} of the rational factor of $I_k$.

Clearly this method is not practical for general results and we strive for criteria that provide linear reducibility for a wide class of integrands which are simple to describe and identify. This is the main aim of the remainder of this section.

\subsection{Iterated integrals of several variables}
\label{sec:ii-several-variables}%
Since we are using hyperlogarithms to represent multivariate functions, we pick up our discussion from section~\ref{sec:multiple-variables} and work out this relation in more detail.
Recall our setup from definition~\ref{def:bar-objects}, where we assigned differential forms $\letter{f} \defas \log f(z) \in \Omega_S \subset \Omega^1(X_S)$ on $X_S \defas \Affine^n \setminus \bigcup_{f \in S} \Vanishing(f)$ to irreducible polynomials $f \in S \subset \Q[z_1,\ldots,z_N]$. We will always assume $\set{z_1,\ldots,z_N} \subseteq S$ from now on and write
\begin{equation}
	\regulars(S)
	\defas \Q\left[ 
		z_1,\ldots,z_N,f^{-1}\colon f \in S
	\right]
	\label{eq:def-regulars-multi}%
\end{equation}%
\nomenclature[O S]{$\regulars(S)$}{regular functions on $X_S$, equation~\eqref{eq:def-regulars-multi}\nomrefpage}%
for the regular functions on $X_S$.
To integrable words $w \in \BarObjects(S) \subseteq T(\Omega_S)$ we can assign homotopy invariant iterated integrals $\int_b^z w$ which span the algebra
\begin{equation}
	\BarIntegrals[b](S)(z) \defas \int_b^z \BarObjects(S),
	\quad\text{and we know that}\quad
	\BarIntegrals(S)(z) \defas \C \tp \BarIntegrals[b](S) (z)
	\label{eq:def-barintegrals}%
\end{equation}%
\nomenclature[Bb(S)]{$\BarIntegrals[b](S)$}{homotopy invariant iterated integrals on $X_S$ with base point $b$, equation~\eqref{eq:def-barintegrals}\nomrefpage}%
is independent\footnote{An isomorphism $\C \tp \BarIntegrals[b](S)(z) \isomorph \C \tp \BarIntegrals[b'](S)(z)$ is determined by a choice of homotopy class (relative to the endpoints) for paths (in $X_S$) from $b$ to $b'$.} of the base point $b \in X_S$ by \eqref{eq:path-concatenation}. In section~\ref{sec:multiple-variables} we anticipated that these are hyperlogarithms if we vary only one variable at a time. If
\begin{equation}
	\ZerosWrt{S}{i}
	\defas \bigcup_{f \in S} \setexp{\sigma}{ 0 = \restrict{f}{z_i=\sigma}}
	\subset
	\overline{\Q(\setexp{z_j}{j \neq i})}
	\label{eq:def-zeros-wrt}%
\end{equation}%
\nomenclature[Sigma]{$\ZerosWrt{S}{i}$}{zeros of the polynomials $S$ with respect to $z_i$, equation~\eqref{eq:def-zeros-wrt}\nomrefpage}%
denotes the union of all zeros of any $f \in S$ with respect to $z_i$ and we write
\begin{equation}
	\HlogProjection{z_i}\colon T(\Omega_S) \longrightarrow T(\ZerosWrt{S}{i})
	\quad\text{for the map}\quad
	\letter{f} \mapsto
	\HlogProjection{z_i}(\letter{f})
	\defas
	\sum_{0=\restrict{f}{z_i=\sigma}} \letter{\sigma}
	\label{eq:def-HlogProjection}%
\end{equation}
such that $[\partial_{z_i} \log(f) ] \dd z_i = \HlogProjection{z_i}(\letter{f})$, we already know that 
$
	\int_b^z
	= \Hyper{\HlogProjection{z_i}(\cdot)}(z_i) \convolution \AnaReg{z_i}{0} \int_b^z
$
and expressed the regularized limit through convergent integrals in \eqref{eq:divergences-barobjects}.
	Recall that $\restrict{S}{z_i=0} \subset \Q[\setexp{z_j}{j \neq i}]$ denotes the irreducible factors of $\setexp{\restrict{f}{z_i=0}}{z_i \neq f \in S}$.
\begin{lemma}
	Let $b' \defas \restrict{b}{b_i = 0}$ and $z' \defas \restrict{z}{z_i = 0}$, then in terms of the map
	\begin{equation}
		\ProjectOn{z_i=0}\colon \BarObjects(S) \longrightarrow \BarObjects(\restrict{S}{z_i=0})
		,\quad
		\letter{f} \mapsto
		\begin{cases}
			0 & \text{if $f=z_i$ and otherwise}\\
			\sum_j \lambda_j \letter{g_j} & \text{where $\restrict{f}{z_i=0} = \prod_j g_j^{\lambda_j}$}\\
		\end{cases}
		\label{eq:def-BarObject-ProjectOn}%
	\end{equation}
	we can write $
		\AnaReg{z_i}{0} \BarIntegrals(S) (z)
		\subseteq \BarIntegrals(\restrict{S}{z_i=0}) (z')
	$ explicitly as
	\begin{equation}
		\AnaReg{z_i}{0} \int_b^z w
		= \sum_{(w)} \int_{b'}^{z'} \ProjectOn{z_i=0} \left(w_{(1)} \right)
			\cdot \AnaReg{y}{b'} \int_b^{y} w_{(2)}.
		\label{eq:BarObject-reglim}%
	\end{equation}
\end{lemma}
\begin{proof}
	Split the integral $\int_b^z = \int_{\eta} \convolution \int_b^y$ at $y \defas \restrict{b}{b_i = z_i}$ and choose the path $\eta$ from $y$ to $z$ such that $\eta_i(t) = z_i$ stays constant. Then $\int_{\eta} w$ vanishes for all words $w$ which contain a letter $z_i$, because $\eta^{*} (\letter{z_i}) = 0$. All other letters are analytic at $z_i\rightarrow 0$, so we can take the limit inside the integral: $\lim_{z_i \rightarrow 0} \int_{\eta} w = \int_{b'}^{z'} \ProjectOn{z_i=0} (w)$.
\end{proof}
This proves the analytic counterpart to the decomposition \eqref{eq:barobjects-decomposition-full}, namely that
\begin{equation}
	\BarIntegrals(S)
	\subseteq \HlogAlgebra(\Sigma_1)(z_1) \tp \BarIntegrals(\restrict{S}{z_1=0})
	\subseteq \HlogAlgebra(\Sigma_1)(z_1) \tp \cdots \tp \HlogAlgebra(\Sigma_N)(z_N) \tp \C
	\label{eq:barintegrals-hlog-products}%
\end{equation}
for the zeros $\Sigma_k$ of $S_k \defas \restrict{S_{k-1}}{z_k=0}$ in $z_k$. In this way we can define the iterated integrals $\int_b^z w$ also for the singular base point $b = 0$ through
\begin{equation}
	\int_0^z w
	\defas \AnaReg{b_N}{0} \cdots \AnaReg{b_1}{0} \int_b^z w
	\quad\text{and we set}\quad
	\BarIntegrals[0](S)
	\defas \AnaReg{b_N}{0} \cdots \AnaReg{b_1}{0} \BarIntegrals[b](S)
	\label{eq:barobjects-0basepoint}%
\end{equation}
which carries a $\Q$-structure again.
\begin{remark}
	This definition depends on the chosen order of the variables, because different regularized limits do not commute. For a simple example consider
	\begin{equation*}
		0
		= \AnaReg{y}{0} \log(y)
		= \AnaReg{y}{0} \AnaReg{x}{0} \log(2x + y)
		\neq \AnaReg{x}{0} \AnaReg{y}{0} \log(2x+y)
		= \AnaReg{x}{0} \log(2x)
		= \ln 2.
	\end{equation*}
	Geometrically this phenomenon is a consequence of the fact that the boundary $\partial X_S$ (the complement of $X_S$), which here is the divisor given by $S=\set{x,y,2x+y}$, is not normal crossing.\footnote{Here we include $x,y \in S$ to cover for our choice of tangential base point, even though the logarithm is singular only on $\Vanishing(2x+y)$. Alternatively we can also take $S = \set{2x+y}$ smooth, then the two orders of limits simply correspond to the (different) tangents $(0,-1)$ and $(-1,0)$ at $(x,y)=(0,0)$.}
	This may be resolved in a suitable model (blowup of the origin) of $X_S$, such that the different orders of limits will indeed define two distinct points in the model. In general such a desingularization is very hard to obtain (a deep result going back to Hironaka), but in the case of linear fibrations this is feasible. For a detailed discussion of the case of the moduli space $\Moduli{0}{n}$ we refer to \cite{Brown:MZVPeriodsModuliSpaces,BognerBrown:GenusZero}.
\end{remark}

\subsection{Linear reducibility}\label{sec:linear-reducibility}
Linear reducibility of an integrand is determined by its singularities, so we can abstract from its concrete form and consider at once a huge class of iterated integrals instead. 
\begin{definition}\label{def:lin-reducible-polys}
	A set $S \subset \C[z_1,\ldots,z_N]$ of irreducible polynomials is called \emph{linearly reducible} if, after rearranging the variables if necessary, $S = S^0$ can be extended to a family ($0 \leq k < N$) of irreducible polynomials $S^{k} \subset \C[z_{k+1},\ldots,z_N]$ linear in $z_{k+1}$ such that every integrable integrand $I_0 \in \regulars(S) \tp \BarIntegrals(S)$ is linearly reducible and furthermore $I_{k} \in \regulars(S^{k}) \tp \BarIntegrals(S^{k})$ for all $k < N$.
\end{definition}
Such a family $(S^k)_{k<N}$ is called a \emph{linear reduction} of $S$. If it exists, we can set
	\begin{equation}
		\Sigma_k 
		= \ZerosWrt{S}{k} 
		= \setexp{-f_k/f^k}{ f = f^k z_k + f_k \in S^{k-1}\ \text{with}\ f^k \neq 0}
		\subset \C(z_{k+1},\ldots,z_N).
		\label{eq:linear-reduction-zeros}%
	\end{equation}
	in definition~\ref{def:lin-reducible-integrand} by \eqref{eq:barintegrals-hlog-products}.

The goal of a \emph{polynomial reduction algorithm} is to construct linear reductions for as many sets $S$ as possible. The simplest such method was introduced in \cite{Brown:TwoPoint}.
\begin{definition}
	Let $S \subset \Q[z_1,\ldots,z_N]$ denote a set of irreducible polynomials. If all $ f = f^i z_i + f_i \in S$ are linear in $z_i$, the \emph{simple reduction} $S_i \subset \Q[\setexp{z_j}{j \neq i}]$ of $S$ with respect to $z_i$ is defined as the set of irreducible factors of the polynomials
	\begin{equation}
		\setexp{f^i, f_i}{f \in S}
		\cup
		\setexp{f^i g_i - g^i f_i}{f,g \in S}.
		\label{eq:all-resultants}%
	\end{equation}
\end{definition}
We do not consider the constants $\Q$ as irreducible polynomials. Note that monomials $z_j \in S$ do not influence the outcome $S_i$ of a reduction.
\begin{lemma}\label{lemma:simple-reduction}%
	If every polynomial in $S$ is linear in $z_i$ (so $S_i$ exists), then $\int_0^{\infty} I \ \dd z_i \in \regulars(S_i) \tp \BarIntegrals(S_i)$ for any integrand $I \in \regulars(S) \tp \BarIntegrals(S)$ such that the integral converges.
\end{lemma}
\begin{proof}
	Note that $\regulars(S) \subseteq \regulars(\ZerosWrt{S}{i})(z_i) \tp \regulars(S')$ if we let $S'$ denote the irreducible factors of the leading coefficients
$
	\setexp{f^i}{f \in S\ \text{and}\ f^i \neq 0} 
	\cup \setexp{f_i}{f \in S \ \text{and}\ f^i=0}
$. Also recall $\BarIntegrals(S) \subseteq \HlogAlgebra(\ZerosWrt{S}{i})(z_i) \tp \BarIntegrals(\restrict{S}{z_i=0})$ from \eqref{eq:barintegrals-hlog-products}, so
	\begin{equation*}
		\int_0^{\infty} I\ \dd z_i
		\in \Q\left[ 
			\sigma,\frac{1}{\sigma - \tau}\colon \sigma, \tau \in \ZerosWrt{S}{i}
		\right]
		\tp \regulars(S')
		\tp \AnaReg{z_i}{\infty} \HlogAlgebra(\ZerosWrt{S}{i})(z_i)
		\tp \BarIntegrals(\restrict{S}{z_i=0}).
		\tag{$\ast$}
	\end{equation*}
	Two zeros $\sigma=-f_i/f^i$, $\tau = -g_i/g^i \in \ZerosWrt{S}{i}$ have the difference $\sigma-\tau = (f^i g_i - g^i f_i)/(f^i g^i)$, hence the rational factors in $(\ast)$ lie in $\regulars(S_i)$ (note $S' \subseteq S_i$). Furthermore we see that
	$ 
		\dd \log(\sigma-\tau) = \dd \log (f^i g_i - g^i f_i) - \dd \log (f^i) - \dd \log (g^i)
	$, which implies
	\begin{equation}
		\AnaReg{z_i}{\infty} \HlogAlgebra(\ZerosWrt{S}{i})(z_i)
		\subseteq \BarIntegrals(S_i)
		\label{eq:reginf-hlog-as-bar}%
	\end{equation}
	via induction over the weight, appealing to lemma~\ref{lem:reginf-differential}.
\end{proof}
\begin{remark}\label{remark:integration-constants-basepoint-reglims}
	If $I \in \regulars(S) \tp \BarIntegrals[0](S)$, we can work over $\Q$ (rather than $\C$) and have
	\begin{equation*}
		\int_0^{\infty} I \ \dd z_1 
		\in \regulars(S_1) \tp \BarIntegrals[0](S_1) \tp \AnaReg{z_N}{0} \cdots \AnaReg{z_2}{0} \AnaReg{z_1}{\infty} \HlogAlgebra(\ZerosWrt{S}{1})(z_1)
		,
	\end{equation*}
	which is a stronger statement in that it discloses the periods which can occur as integration constants: They are of the form $\AnaReg{z_1}{\infty} \HlogAlgebra\big[\leadCoeff_{z_N}\!\!\cdots \leadCoeff_{z_2}\,\ZerosWrt{S}{1}\big](z_1)$.
\end{remark}
\begin{corollary}
	If all iterated reductions $S^{k} \defas (S^{k-1})_{k}$ of $S^0 \defas S$ up to $S^{N-1}$ exist, then $S$ is linearly reducible and the sets $S^k$ form a linear reduction.
\end{corollary}
\begin{example}[moduli space $\Moduli{0}{N+3}$]\label{ex:moduli-space-mzv}
Consider $N$ variables $z_1,\ldots,z_N$ and set
	\begin{equation*}
		S^{k} \defas \setexp{z_{i}+\cdots+z_N+1}{k< i \leq N} 
		\cup \setexp{z_i + \cdots + z_j}{k<i \leq j \leq N},
	\end{equation*}
	then $(S^{k})_{k+1} = S^{k+1}$ and conclude linear reducibility of $S^0$ with alphabets 
	\begin{equation*}
		\ZerosWrt{S^{k-1}}{k} 
		= \set{0,-1-z_{k+1}-\cdots-z_N}
		\cup \setexp{-z_i-\cdots-z_j}{k<i \leq j \leq N}.
	\end{equation*}
	These have only coefficients $0$ and $-1$, so from \eqref{eq:final-period-algebra} we obtain
	\begin{equation*}
		\int_0^{\infty} \dd z_1 \cdots \int_0^{\infty} \dd z_N \ F(z)
		\in \MZV
		\quad\text{for all}\quad
		F \in \regulars(S^0) \tp \BarIntegrals[0](S^0)
	\end{equation*}
	such that the integral converges. The variables $z_i$ parametrize the moduli space $\Moduli{0}{N+3}$, which we can also view as the configuration space $\tsetexp{x \in X^N}{x_i \neq x_j\ \text{for all}\ i \neq j}$ of $N$ distinct points on $X = \Affine^1 \setminus \set{0,1}$, upon setting $x_i = z_i + \cdots + z_N + 1$. In these coordinates, the positive hypercube $z \in (0,\infty)^N$ corresponds to the connected component (cell) of the real points $\Moduli{0}{N+3}(\R)$ where $x_N> \cdots > x_1 > 1$. 
	
	We just proved that all such integrals, in particular for rational integrands $F \in \regulars(S^0)$, evaluate to MZV. This was the original application of hyperlogarithmic integration by Francis Brown \cite{Brown:MZVPeriodsModuliSpaces}.
\end{example}

\subsubsection{Fubini reduction}
The linear independence of hyperlogarithms (lemma~\ref{lemma:hyperlog-independence}) in the representation \eqref{eq:barintegrals-hlog-products} and the basis \eqref{eq:rationals-product-basis} of rational functions mean that
\begin{equation}
	\regulars(S) \tp \BarIntegrals(S)
	\cap
	\regulars(S') \tp \BarIntegrals(S')
	= \regulars(S \cap S') \tp \BarIntegrals(S \cap S').
	\label{eq:barintegrals-intersection}%
\end{equation}
If the simple reductions $(S_i)_j$ and $(S_j)_i$ are defined, we can therefore obtain a possibly smaller bound $S_{\set{i,j}} \defas (S_i)_j \cap (S_j)_i$ on the singularities of convergent integrals
$\int_{\R_+^2} I\ \dd z_i \wedge \dd z_j \in \regulars(S_{\set{i,j}}) \tp \BarIntegrals(S_{\set{i,j}})$ with integrands $I \in \regulars(S) \tp \BarIntegrals(S)$. More generally, we can permute the order of integration at will due to Fubini's theorem in any dimension. The resulting reduction algorithm was introduced in \cite{Brown:TwoPoint}.
\begin{definition}\label{def:Fubini-reduction}
	Starting with irreducible polynomials $S_{\emptyset} \subset \Q[z_1,\ldots,z_N]$, the sets
	\begin{equation}
		S_{I} \defas \bigcap_{
			\substack{i \in I \\ S_{I \setminus \set{i}} \ \text{is defined and linear in $z_i$}}
		}
		\Big(S_{I \setminus \set{i}} \Big)_{z_i}
		\label{eq:def-Fubini-reduction}%
	\end{equation}
	are defined recursively for $I \subseteq [N] \defas \set{1,\ldots,N}$ if at least one $i \in I$ is admitted to the intersection. We call $S_{\emptyset}$ \emph{Fubini reducible} if $S_{I}$ is defined for some set of $\abs{I}=N-1$ elements.\footnote{This suffices since the univariate polynomials $S_I$ factorize and thus $S_{[N]}$ will be defined as well, at least over $\overline{\Q}$.}
\end{definition}
A Fubini reducible set is clearly linearly reducible, because there must be some permutation $\sigma$ of $[N]$ such that all sets $S_{\set{\sigma(1),\ldots,\sigma(k)}}$ are defined, which then provide an upper bound on the singularities of the corresponding partial integral.

In practice, this Fubini algorithm is very effective for low-dimensional problems. For example, it is known to suffice to prove linear reducibility of many vacuum graphs up to six loops \cite{Brown:TwoPoint} and also for plenty of massless on-shell four-point functions up to three loops \cite{Lueders:LinearReduction}.

But theoretically it seems very hard to describe and keep track of all polynomials that appear in the reduction explicitly. To obtain results for infinite families of graphs, perfect control of the reduction will be necessary though. 

\subsection{Compatibility graphs}
One observes that many of the polynomials in a Fubini reduction $S_I$ do not actually occur when a particular integral is calculated. Typically, the sets $S_I$ are gross overestimates of the singularities of a high-dimensional partial integral. An extremely powerful tool to address this problem is the concept of \emph{compatibility graphs}.
\begin{definition}
	Let $S \subset \C[z_1,\ldots,z_N]$ denote a set of irreducible polynomials and $C \subseteq \binom{S}{2}$ a set of undirected edges (\emph{compatibilities}) between them, then we call $(S,C)$ a \emph{compatibility graph} and two polynomials $f,g \in S$ are said to be \emph{compatible} if they are adjacent ($\set{f,g} \in C$).
\end{definition}
The idea is that instead of computing all resultants in \eqref{eq:all-resultants}, we only need to take compatible pairs into account. We introduce the abbreviations
\begin{equation}
	\resultant{f}{0}{i} \defas f_i, \qquad
	\resultant{f}{\infty}{i} \defas \begin{cases}
		f^i & \text{if $f^i \neq 0$,}\\
		f_i & \text{otherwise} \\
	\end{cases}
	\qquad\text{and}\qquad
	\resultant{f}{g}{i} 
	\defas	f^i g_i - f_i g^i
	\label{eq:def-linear-resultants}%
\end{equation}%
\nomenclature[f g i]{$\resultant{f}{g}{i}$}{resultant of $f$ and $g$ with respect to $z_i$, equation~\eqref{eq:def-linear-resultants}\nomrefpage}%
for the constant term, leading term and the resultant $\resultant{f}{g}{i} = - \resultant{g}{f}{i}$ of two linear polynomials $f=f^i z_i + f_i$ and $g=g^i z_i + g_i$ with respect to $z_i$.
\begin{definition}\label{def:cg-reduction}
	If $(S,C)$ is a compatibility graph where all $f\in S \subset \Q[z_1,\ldots,z_N]$ are linear in $z_i$, we define the \emph{reduced} compatibility graph $(S,C)_i \defas (S',C')$ as follows. Its vertices $S' \subset \Q[z_j\colon j \neq i]$ are the irreducible factors of the polynomials
\begin{equation}
	\setexp{\resultant{f}{0}{i},\resultant{f}{\infty}{i}}{f \in S}
	\cup
	\setexp{\resultant{f}{g}{i}}{\text{compatible pairs}\ f,g \in S}.
	\label{eq:cg-reduction-vertices}%
\end{equation}
The compatibilities $C' \subseteq \binom{S'}{2}$ are defined between all pairs of (distinct) irreducible factors of $\resultant{f}{g}{i}\cdot\resultant{g}{h}{i}\cdot \resultant{h}{f}{i}$ for all those triples $f,g,h \in \set{0,\infty} \cupdot S$ which are mutually compatible. Here we consider $0$ and $\infty$ as compatible with each other and every $f \in S$. 

In other words, $p,q \in S'$ are compatible ($\set{p,q} \in C'$) precisely if there exist $f,g,h \in S \cupdot \set{0,\infty}$ such that $p \divides \resultant{f}{g}{i}$, $q \divides \resultant{g}{h}{i}$ and each of the pairs $\set{f,g}, \set{g,h}, \set{h,f} \in C$ is compatible.
\end{definition}
\begin{remark}
	Since $\resultant{f}{z_i}{i} = -\resultant{f}{0}{i}$ and $\resultant{f}{z_j}{i} = z_j\resultant{f}{\infty}{i}$ when $j \neq i$, the resultants with monomials do not introduce any additional vertices (irreducible factors) nor edges (compatibilities). For example note that the compatibility between $\resultant{f}{0}{i}$ and $\resultant{f}{\infty}{i}$ (which comes from the mutually compatible triple $\set{f,z_i,z_j} \subseteq S$) is already taken into account through the triple $\set{f,0,\infty}$. In the same way, all compatible triples that involve a monomial generate only compatibilities that arise already from triples without monomials.

	For this reason, all monomials $\set{z_1,\ldots,z_N} \subseteq S$ can be dropped from $S$ without changing the reductions. We will therefore not show them when we draw a compatibility graph, but keep in mind that they always belong to $S$ and are compatible with each other and with all other polynomials in $S$.
\end{remark}
\begin{figure}
	\centering
	$ (S,C)\colon \quad \Graph[1.1]{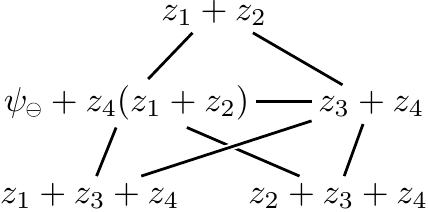}
		\quad
		\mapsto
		\quad
		(S^{\StarSymbol},C^{\StarSymbol})
		\defas
		(S,C)_4\colon  \Graph[1.1]{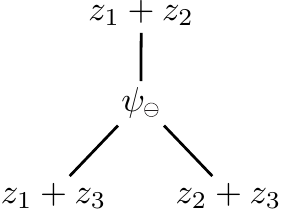}
	$
	\caption[Reduction of a simple compatibility graph]{Reduction of a compatibility graph $(S,C)$ with respect to $z_4$. The polynomial $\SunrisePsi = z_1 z_2 + z_1 z_3 + z_2 z_3$ was introduced in \eqref{eq:SunrisePsi}.}%
	\label{fig:cg-reduction-3pt}%
\end{figure}%
\begin{example}\label{ex:cg-reduction-3pt-step}
	Consider the compatibility graph $(S,C)$ in figure~\ref{fig:cg-reduction-3pt}. The constant coefficients $\resultant{f}{0}{4}$ already deliver all polynomials in $S'$, because all other resultants factorize into monomials except for
	$\resultant{\SunrisePsi + z_4(z_1 + z_2)}{z_1+z_2}{4} = (z_1+z_2)^2$ and 
	\begin{equation*}
		\resultant{\SunrisePsi + z_4(z_1 + z_2)}{\infty}{4}
		= \resultant{z_1+z_2}{\infty}{4}
		= \resultant{z_1+z_2}{z_3+z_4}{4}
		= z_1 + z_2.
	\end{equation*}
	Apart from the compatibilities between $\resultant{f}{0}{4}$ and $\resultant{g}{0}{4}$ (for compatible $f$ and $g$), any further compatibilities would have to be between $z_1+z_2 = \resultant{f}{g}{4}$ and a resultant of the form $\resultant{f}{0}{4}$ where $\set{f,g}$ is one of $\set{\SunrisePsi+z_4(z_1+z_2),z_1+z_2}$, $\set{z_1+z_2,\infty}$, $\set{\SunrisePsi+z_4(z_1+z_2),\infty}$ or $\set{z_1+z_2,z_3+z_4}$.
	But these only contribute the compatibility between $z_1 + z_2$ and $\SunrisePsi$, so the reduced compatibility graph $(S,C)_4$ is obtained by simply replacing each polynomial $f \in S$ by its constant term $\resultant{f}{0}{4}$ and keeping the original compatibilities.
\end{example}
\begin{proposition}\label{prop:cg-reduction}
	Given a set $S \subset \Q[z_1,\ldots,z_N]$ of irreducible polynomials (containing the monomials), let $C \defas \binom{S}{2}$ such that $(S,C)$ is a complete graph.

	If (after a permutation of the variables $\set{z_1,\ldots,z_N}$ where required) all iterated reductions $(S^k,C^k) \defas (S^{k-1},C^{k-1})_{k}$ exist from $(S^0,C^0) \defas (S,C)$ up to $k=N$, then $S$ is linearly reducible and the sets $S^k$ form a linear reduction of $S$.
\end{proposition}
This result is strong enough to allow for simple proofs of linear reducibility of some infinite families of Feynman integrals, as we exemplify in the next section.
The proof of proposition~\ref{prop:cg-reduction} will be given in the separate section~\ref{sec:landau-varieties}, because it requires a lot of special terminology and setup.
\begin{remark}
	Our definition~\ref{def:cg-reduction} is different from its original formulation in \cite{Brown:PeriodsFeynmanIntegrals}, where the compatibilities $C'$ where defined between factors of resultants $\resultant{f}{g}{i}$ and $\resultant{g}{h}{i}$ for \emph{all triples} $f,g,h \in S$ (not restricting to mutually compatible triples only). This results in more compatibilities and thus more polynomials after subsequent reductions (in example~\ref{ex:cg-reduction-3pt-step}, all polynomials would become compatible through the triples $\resultant{f}{0}{4}$ and $\resultant{g}{0}{4}$, as $f$ and $g$ would not be required anymore to be compatible themselves). 
	
	To remedy this surplus, \cite{Brown:PeriodsFeynmanIntegrals} applied the strategy to intersect compatibility graphs for different orders of reductions, just like in the Fubini algorithm \eqref{eq:def-Fubini-reduction}. We find it hard to justify this method in the general case and discussed several technical issues with Francis Brown. A thorough investigation of these problems needs to be addressed separately and is not contained in this thesis, we only give a comment at the end of section~\ref{sec:cg-intersection}.
\end{remark}

\subsection{Linear reducibility from recursion formulas}\label{sec:recursion-reducibility}
We apply proposition~\ref{prop:cg-reduction} to the recursion formulas from section~\ref{sec:vw3}. The idea is very simple: It suffices to keep only the last compatibility graph $(S^k,C^k)$ of a forest function of a graph $G_{k}$. We can add the edge $k+1$ directly on the compatibility graph and compute the effect of the integration by a reduction following definition~\ref{def:cg-reduction}.
Let us abbreviate \hypertarget{inline:barregulars}{rational linear combinations of iterated integrals}\label{inline:barregulars} on $X_S$ with 
$\BarIntegralsRegulars(S) \defas \regulars(S) \tp \BarIntegrals(S)$.
\begin{proposition}\label{prop:GfunStarTriangle-reduction}%
	Suppose the graph $G$ is $3$-constructible (with three external vertices). Then the forest-, star- and triangle functions of $G$ are linearly reducible and of type
	$
		\GfunTriangle{G}
		\in
		\BarIntegralsRegulars[0](S^{\TriangleSymbol}) \tp \MZV
	$ and
	$
		\GfunStar{G}, \GfunForest{G}
		\in
		\BarIntegralsRegulars[0](S^{\StarSymbol}) \tp \MZV
	$ for the polynomials
	\begin{equation}\begin{split}\label{eq:GfunStarTriangle-reduction}%
		S^{\TriangleSymbol} 
		&\defas \set{z_1,z_2,z_3,z_1+z_2, z_1+z_3, z_2 + z_3, z_1 + z_2 + z_3}
		\quad\text{and}
		\\
		S^{\StarSymbol} &\defas \set{z_1,z_2,z_3,z_1+z_2, z_1+z_3, z_2 + z_3, z_1z_2 + z_1 z_3 + z_2 z_3}.
	\end{split}\end{equation}
\end{proposition}
\begin{proof}
	Order the edges of $G$ according to a $3$-construction and write $G_k$ for the subgraph formed by the first $k$ edges. We prove inductively that the final compatibility graphs $(S^k,C^k)$ of $\GfunStar{G_{k}}$ and $\GfunForest{G_{k}}$ are contained in the star-shaped compatibility graph $(S^{\StarSymbol},C^{\StarSymbol})$ with 
	$C^{\StarSymbol} 
		\defas \setexp{\set{\SunrisePsi,f}}{\SunrisePsi \neq f \in S^{\StarSymbol}}
		$ of figure~\ref{fig:cg-reduction-3pt}.\footnote{The statement for $\GfunTriangle{G}$ can be derived analogously, but it follows immediately from the result on $\GfunStar{G}$ through the star-triangle duality \eqref{eq:star-triangle-psipol-transformation}.}
		The initial cases are given by \eqref{eq:GfunForest-3edge} and \eqref{eq:Gfun-StarTriangle-K1}. Now suppose $(S^{k-1},C^{k-1}) \subseteq (S^{\StarSymbol},C^{\StarSymbol})$, $\GfunStar{G_{k-1}} \in \BarIntegralsRegulars[0](S^{\StarSymbol}) \tp \MZV$ and let $G_k$ be constructed by appending a new vertex $v_1'$ to the external vertex $v_1 \in \vertices_{\Text}(G_{k-1})$. Using \eqref{eq:GfunStar-appendvertex}, we can write $\GfunStar{G_k} = \int_0^{\infty} I\ \dd \SP_k$ with the integrand
		\begin{equation*}
			I
			= \GfunStar{G}(z_1+\SP_k,z_2,z_3) \SP_k^{\EP_{k} - 1} \in \BarIntegralsRegulars[0](S) \tp \MZV
			\quad\text{where}\quad
			S \defas \set{\SP_k} \cupdot \setexp{\restrict{f}{z_1 = z_1 + \SP_k}}{f \in S^{k-1}}
		\end{equation*}
		is simply obtained from $S^{k-1}$ by replacing $z_1$ with $z_1 + \SP_k $. Indeed, the penultimate compatibility graph $(S,C)$ of $G_k$ just arises through this replacement (and adding the monomial $\SP_k$) from $(S^{k-1},C^{k-1})$, because the first $k-1$ reductions of Schwinger parameters do not involve $\SP_k$ at all. So the final reduction $(S^k, C^k) = (S^{k-1},C^{k-1})_k$ is example~\ref{ex:cg-reduction-3pt-step} (with $z_3$ and $z_1$ interchanged) and yields $(S^{\StarSymbol},C^{\StarSymbol})$ again (or a subgraph). The same reasoning applies also if the edge $k$ connects two external vertices of $G_{k-1}$, since \eqref{eq:GfunStar-addedge} only requires us to introduce $\restrict{\SunrisePsi}{z_1 = z_1+\SP_k}$ into the integrand, which is already compatible in $(S,C)$ with all other components.

		In the same way we can proceed for $\GfunForest{G_k}$. The only difference is that the integral representations \eqref{eq:Gfun-Forest-vertex} and \eqref{eq:Gfun-Forest-edge} are of the form $\int_0^{z_1} \restrict{I}{z_1 = z_1 - \SP_k} \dd \SP_k$. So we first change variables $\SP_k = z_1/(1+x)$ to integrate $x$ over $(0,\infty)$. Then we can compute the reduction of $x$ with the same outcome $(S^{\StarSymbol},C^{\StarSymbol})$ as before. Note that under a change of variables, each vertex $f \in S$ is replaced by (the clique on) its irreducible factors as shown in figure~\ref{fig:cg-variable-inversion}.
		
		Finally we need to check that the constants of integration are multiple zeta values. Since the polynomials in $S$ only have monomials with the coefficient $1$,\footnote{Recall figures~\ref{fig:cg-reduction-3pt} and \ref{fig:cg-variable-inversion} corresponding to the two different cases when we substituted $z_1 \mapsto z_1+\SP_k$ or $z_1 \mapsto z_1 x/(1+x)$.} the leading coefficients of the letters of the integrand $I$ give
		$\leadCoeff_{z_i} \leadCoeff_{z_j} \leadCoeff_{z_k} \ZerosWrt{S}{\SP_k} = \set{0,-1}$. So we can apply proposition~\ref{prop:reglim-reginf-from-word} to deduce that only MZV remain (see also \eqref{eq:reginf-gives-mzv} and remark~\ref{remark:integration-constants-basepoint-reglims}). Note that this holds for any order ($\set{i,j,k}=\set{1,2,3}$) of the three limits  $\AnaReg{z_i}{0}$ chosen in \eqref{eq:barobjects-0basepoint} to fix the singular base point $0$ in $\BarIntegrals[0]$.
\end{proof}
\begin{figure}
	\centering
	$
		\Graph[1.2]{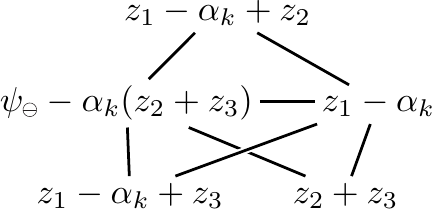}
		\quad\xrightarrow{\SP_k = z_1/(1+x)}\quad
		\Graph[1.2]{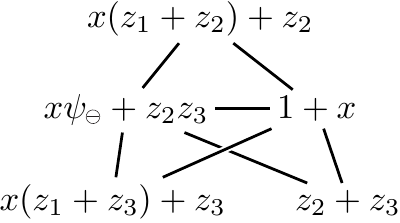}
	$
	\caption{Transformation of a compatibility graph under a change of variables.}%
	\label{fig:cg-variable-inversion}%
\end{figure}%
\begin{remark}\label{remark:cg-reduction-recursion}
	We can formulate the reduction of a shifted compatibility graph with $S' \defas \set{x} \cupdot \restrict{S}{z=z+x}$ (each edge $e \in C$ induces edges in $C'$ between all irreducible factors of its endpoints after the transformation) with respect to $x$ directly in terms of the original variable.
	Let $f' \defas \restrict{f}{z=z+x} \in S'$ with $f \in S$, then
	\begin{equation*}
		\resultant{f'}{0}{x} = f
		,\quad
		\resultant{f'}{\infty}{x} = \resultant{f}{\infty}{z}
		,\quad
		\resultant{f'}{z+x}{x} = -\resultant{f}{0}{z}
		\quad \text{and}\quad
		\resultant{f'}{g'}{x} = \resultant{f}{g}{z}
	\end{equation*}
	means that $(S',C')_x = (S'',C'')$ has the vertices $S''$ from the plain reduction \eqref{eq:cg-reduction-vertices} but also the original $ S \subseteq S''$. The only additional compatibilities with respect to the reduction of definition~\ref{def:cg-reduction} are between original polynomials that are compatible $\set{f,g} \in C\subseteq C''$, their resultants $\set{f,\resultant{f}{g}{z}} \in C''$ and $\set{f,\resultant{f}{0}{z}}, \set{f,\resultant{f}{\infty}{z}} \in C''$.

	Such a reduction fulfils $S \subseteq S''$, $C\subseteq C''$ and provides $\int_0^{\infty} \restrict{I}{z=z+x}\ \dd x \in \BarIntegralsRegulars(S'')$ for any integrand $I \in \BarIntegralsRegulars(S)$.
	The analogous consideration for the transformation $z\mapsto z-z/(1+x)=zx/(1+x)$ yields precisely the same reduced compatibility graph $(S'',C'')$ and shows $\int_0^z \restrict{I}{z=z-x}\ \dd x \in \BarIntegralsRegulars(S'')$.
\end{remark}
Proposition~\ref{prop:GfunStarTriangle-reduction} thus boils down to the fact that the compatibility graph $(S^{\StarSymbol}, C^{\StarSymbol})$ from figure~\ref{fig:cg-reduction-3pt} is stable (does not change) with respect to these reductions (of remark~\ref{remark:cg-reduction-recursion}). On the analytic side this means that the forest- and star functions remain iterated integrals in the family $\BarIntegralsRegulars(S^{\StarSymbol})$, no matter how many recursions of the integral formulae of section~\ref{sec:3pt-recursions} are applied.
\begin{theorem}
	\label{theorem:vw3-3pt} %
	Assume a graph $G$ is $3$-constructible with final vertices $\set{v_1,v_2,v_3} = \vertices_{\Text}(G)$ and massless propagators.
	Then $\int I_G\ \Omega$ from \eqref{eq:feynman-integral-projective} is linearly reducible via the recursions of its forest functions and the final integrations \eqref{eq:vw3-3pt-projective}.

	If the external momenta $p_i$ (entering $G$ at $v_i$) are parametrized by $p_2^2 = z\bar{z} p_1^2$ and $p_3^2 = (1-z)(1-\bar{z}) p_1^2$, then each coefficient in the $\varepsilon$-expansion of $p_3^{-2\sdd}\int I_G\ \Omega$ (with respect to the dimension $\dimension \in 2\N - 2\varepsilon$ and indices $\EP_e \in \Z + \EPE_e \varepsilon$) is of the form
	\begin{equation}
			\MZV
			\tp \BarIntegralsRegulars[0]\left( \set{z,\bar{z},1-z, 1-\bar{z}, z-\bar{z}, 1-z-\bar{z},1-z\bar{z},z\bar{z}-z-\bar{z}} \right)
		.
		\label{eq:vw3-3pt-expansion} %
	\end{equation}
\end{theorem}
\begin{proof}
	Using proposition~\ref{prop:GfunStarTriangle-reduction} to compute the forest function of $G$, we need to add the polynomial $\phipol = z\bar z x_1 + (1-z)(1-\bar{z}) x_2 + x_3$ to $S^{\StarSymbol}$ such that it is compatible with every other polynomial. Reducing two of the variables $x_i$ (recall that \eqref{eq:vw3-3pt-projective} is a projective integral, so one variable is fixed to one) we check that we arrive precisely at the polynomials \eqref{eq:vw3-3pt-expansion}.

	We assumed that $\FR(G)$ is convergent (such that we may expand the integrand in $\varepsilon$ and integrate each coefficient).
	But if it diverges (at $\varepsilon=0$), we can apply corollaries~\ref{cor:finite-anareg-as-Feynman} or \ref{cor:convergent-anareg-integrand} to express it in terms of convergent integrals (with shifted $\dimension$ and $\EP_e$) and process each of those as above.
\end{proof}

\begin{remark}
	Our result \eqref{eq:vw3-3pt-expansion} is not optimal, because we know (from direct polynomial reduction in Schwinger parameters) that the letters $\set{z,\bar{z},1-z,1-\bar{z},z-\bar{z}}$ suffice to express these functions. This is also clear from the position space viewpoint of graphical functions \cite{Schnetz:GraphicalFunctions}, at least in exactly $\dimension=4$ dimensions.

	This must be a consequence of a special property of the forest functions $\GfunForest{G}$, because there are integrands $f \in \BarIntegralsRegulars(S^{\StarSymbol})$ (even with compatibilities constrained to $C^{\StarSymbol}$) that integrate under \eqref{eq:vw3-3pt-projective} to polylogarithms that do involve the additional singularities $\set{1-z-\bar{z},1-z\bar{z},z\bar{z}-z-\bar{z}}$. Hence our setup apparently does not capture all relevant information of the forest functions.
\end{remark}

\subsubsection{Ladder box integrals}
\begin{proposition}\label{prop:GfunForestBox-reduction}
	Let $G$ be any minor of a box ladder with four external vertices. Then its forest function 
	$\GfunForestBox{G} 
		\in \BarIntegralsRegulars[0](S^{\ForestDeltaBoxSymbol})
		\tp \MZV
	$ is linearly reducible with compatibility graph $(S^{\ForestDeltaBoxSymbol},C^{\ForestDeltaBoxSymbol})$ shown in figure~\ref{fig:cg-box-ladder} (or a subgraph of it).
\end{proposition}
\begin{figure}
	\centering
	$\Graph[1.0]{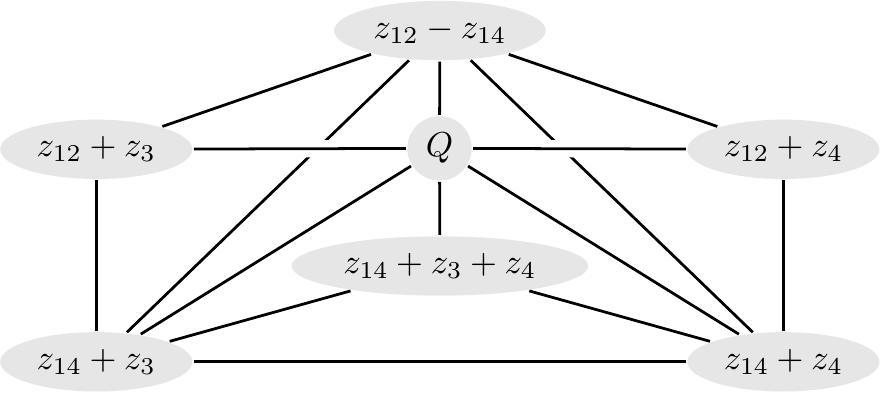}$
	\caption[Compatibility graph of box-ladder forest functions]{Compatibility graph $(S^{\ForestDeltaBoxSymbol}, C^{\ForestDeltaBoxSymbol})$ of box-ladder forest functions. The polynomial $\BoxPoly = z_{12}(z_{14}+z_3+z_4) + z_3 z_4$ was defined in \eqref{eq:def-BoxPoly}.}%
	\label{fig:cg-box-ladder}%
\end{figure}%
\begin{proof}
	Order the edges along a construction of $G$ according to the moves of figure~\ref{fig:boxladder-addedge}. We show the statement inductively, starting from the one-loop box $B_1$ (figure~\ref{fig:boxladders-forest}). Its forest integral \eqref{eq:GfunForestBox-1loop} lies in $\GfunForestBox{B_1} \in \BarIntegralsRegulars[0](\set{\BoxPoly,z_{12},z_{14},z_3,z_4})$, so the initial compatibility graph consists only of $\BoxPoly$ and the monomials.
	For each edge we add, we apply the formula~\eqref{eq:GfunForestBox-add-vertex} or \eqref{eq:GfunForestBox-add-edge} as appropriate and can proceed as in remark~\ref{remark:cg-reduction-recursion} to compute the compatibility graph of the next graph (forest function). Only \eqref{eq:GfunForestBox-add-edge} introduces an explicit polynomial into the integrand, but this is just $\BoxPoly$ which is anyway already compatible with every other polynomial in $S^{\ForestDeltaBoxSymbol}$.

	Therefore we only need to check that the reduction of $(S^{\ForestDeltaBoxSymbol}, C^{\ForestDeltaBoxSymbol})$ (after replacing $z_i$ by $x z_i/(1+x)$) with respect to $z_i$ (giving the compatibility graph of $\int_0^{z_i} \restrict{I}{z_i = z_i - x}\ \dd x$) reproduces the original graph, for each $i \in \set{12,3,4}$. This is easily checked; if for instance $i=3$, then the only non-trivial resultants are
	\begin{gather*}
		\resultant{\BoxPoly}{\infty}{z_3} = z_{12} + z_4
		\qquad
		\resultant{\BoxPoly}{0}{z_3} = z_{12}(z_{14} + z_4)
		\qquad
		\resultant{z_{14}+z_3+z_4}{0}{z_3} = z_{14} + z_4
		\\
		\resultant{z_{12}+z_3}{\BoxPoly}{z_3} = -z_{12}(z_{12}-z_{14})
		\qquad
		\resultant{z_{14}+z_3}{z_{12}+z_3}{z_3} = z_{12} - z_{14}
		\\
		\resultant{z_{14} + z_3}{\BoxPoly}{z_3} = z_4(z_{12}-z_{14})
		\qquad
		\resultant{\BoxPoly}{z_{14}+z_3+z_4}{z_3} = z_4 (z_{14}+z_4)
	\end{gather*}
	and can thus only lead to additional compatibilities incident to $z_{12}-z_{14}$, $z_{12}+z_4$ or $z_{14}+z_4$. For example, the mutual compatibility of $0$, $\infty$ and $\BoxPoly$ means that we must include $\set{z_{12}+z_4,z_{14}+z_4}$ from the first two resultants above, but this compatibility is already in $C^{\ForestDeltaBoxSymbol}$.
	One checks that indeed no new compatibilities arise this way. The case $i=4$ is covered by symmetry, and for $i=12$ the only non-trivial resultants to check are
	\begin{gather*}
		\resultant{z_{12}-z_{14}}{z_{12}+z_3}{z_{12}} = z_{14}+z_3
		\qquad
		\resultant{z_{12}-z_{14}}{z_{12}+z_4}{z_{12}} = z_{14}+z_4
		\\
		\resultant{\BoxPoly}{z_{12}+z_3}{z_{12}} = z_3 (z_{14} + z_3)
		\qquad
		\resultant{\BoxPoly}{z_{12}+z_4}{z_{12}} = z_4 (z_{14} + z_4)
	\end{gather*}
	and $	\resultant{z_{12}-z_{14}}{\BoxPoly}{z_{12}} = (z_{14}+z_3)(z_{14}+z_4)$. Again a simple check of all mutually compatible triples verifies that no new compatibilities are introduced in the reduction step.
	
	Finally we must examine the integration constants that appear. When a vertex is appended (say $i = 3$), the letters of the integrand are
	\begin{equation*}
		\ZerosWrt{\restrict{S^{\ForestDeltaBoxSymbol}}{z_3 = x z_3/(x+1)}}{x}
		= \set{
			0, -1, -\frac{z_{12}}{z_3+z_{12}}, - \frac{z_{14}}{z_3+z_{14}},-\frac{z_{14}+z_4}{z_{14}+z_3+z_4},
			-\frac{z_{12}(z_{14}+z_4)}{\BoxPoly}
		}
	\end{equation*}
	and we obtain only the letters $\set{0,-1}$ after taking the four limits $\AnaReg{z_k}{0}$ in \eqref{eq:barobjects-0basepoint} through $\leadCoeff_{z_k}$ (no matter in which order). In the case of $i=12$, the polynomial $z_{12}-z_{14}$ introduces a letter which lies on the integration path (of $x$) when $z_{12}>z_{14}$:
	\begin{equation*}
		\ZerosWrt{\restrict{S^{\ForestDeltaBoxSymbol}}{z_{12} = x z_{12}/(x+1)}}{x}
		= \set{
			0, -1, - \frac{z_3}{z_{12}+z_3}, -\frac{z_{4}}{z_{12}+z_4}, -\frac{z_3 z_4}{\BoxPoly}, \frac{z_{14}}{z_{12}-z_{14}}
		}.
	\end{equation*}
	If we let $z_{12} \rightarrow 0 $ \emph{before} $z_{14} \rightarrow 0$, we only obtain MZV by the same argument. But if the limit $z_{14} \rightarrow 0$ is applied first, $\leadCoeff_{z_{14}}(\Sigma_x)$ contains $1/z_{12}$ and thus $+1$ after application of $\leadCoeff_{z_{12}}$.
	We thus might encounter alternating sums $\MZV[2]$, since $-1$ and $0$ are present as well. However, corollary~\ref{corollary:reglim-degree-decouple} tells us that the letter $1/z_{12}$ decouples from all others, because $\deg_{z_{14}}\left( \frac{z_{14}}{z_{12}-z_{14}} \right)= 1$ is unique among the vanishing degrees in $\Sigma_x$ (which are otherwise either $0$ or negative in case one or both of the limits $z_3,z_4 \rightarrow 0$ had been applied before). Hence, apart from $\MZV$, we can only have periods $\AnaReg{z_{12}}{0}\AnaReg{x}{\infty} \HlogAlgebra(\set{0,1/z_{12}}) \subseteq {\MZV}[\imag\pi]$ by corollary~\ref{corollary:reglim-reginf-splitted} as no pinch occurs. Finally, we know from the definition of the forest function that the integrand must be analytic at $ x=\frac{z_{14}}{z_{12}-z_{14}}$ and no imaginary parts can appear (see also lemma~\ref{lemma:singular-expansion-analytic}). This concludes the proof that the integration constants lie in $\MZV$, no matter which order of the four limits $z_i \rightarrow 0$ is chosen to approach the base point $0$.
\end{proof}
\begin{remark}
	Not all minors of ladder boxes are extensions (via the steps of figure~\ref{fig:boxladder-addedge}) of the one-loop box $B_1$. A proper minor of $B_1$ (deletion or contraction of an edge) has less then four edges and does not define a forest function, because at least one of the four spanning forest polynomials \eqref{eq:def-forestpolynoms-ladderbox} will vanish identically and introduce the ill-defined $\delta(0)$ into \eqref{eq:def-GfunForestBox}.
	In section~\ref{sec:ex-ladderboxes} we demonstrate an extension of proposition~\ref{prop:GfunForestBox-reduction} which allows us to circumvent this problem completely.
	
	Note however that the above proof applies as-is to all ladder boxes themselves. This suffices in principle to conclude the full claim already, because linear reducibility is a minor-closed property of graphs \cite{BognerLueders:MasslessOnShell,Lueders:LinearReduction}.\footnote{Strictly speaking, the quoted result is formulated only for the Fubini reducibility from definition~\ref{def:Fubini-reduction}. We do not expect any difficulties though to apply this proof to our compatibility graph reduction.}
\end{remark}
\begin{example}
	Consider the graphs shown in figure~\ref{fig:boxladders-forest} and their forest functions \eqref{eq:ladderbox-forest-1'}, \eqref{eq:ladderbox-forest-1''} and \eqref{eq:ladderbox-forest-2}: Apart from the monomials, $\GfunForestBox{B_1}$ only has a singularity at $\BoxPoly = 0$. Next, $\GfunForestBox{B_1'}$ acquires a singularity at $z_{14}+z_4 = 0$ and $\GfunForestBox{B_1''}$ also introduces $z_{14}+z_3$ as well as $z_{12} - z_{14}$. The double box $\GfunForestBox{B_2}$ features the polynomial $z_{14}+z_3+z_4$ for the first time and we checked that for even more edges, eventually also the remaining $z_{12}+z_3$ and $z_{12} + z_4$ occur. Hence the set $S^{\ForestDeltaBoxSymbol}$ is minimal.
	
	The essential message of proposition~\ref{prop:GfunForestBox-reduction} is that the singularities do not continue to get more and more complicated beyond this point (when we add further edges), but are confined to $\bigcup_{f\in S^{\ForestDeltaBoxSymbol}} \Vanishing(f)$ in perpetuity.
\end{example}
Now let us apply this result to the Feynman integrals $\FR(G)$. We consider the projective integral $\int I_G\ \Omega = \Gamma^{-1}(\sdd)\big[ \prod_e \Gamma(\EP_e) \big] \FR(G)$ from \eqref{eq:feynman-integral-projective} to remove the Euler-Mascheroni constant $\gamma_E$ from the $\varepsilon$-expansion.
\begin{theorem}
	\label{theorem:ladderbox-kinematics} %
	Let $G$ be a minor of a ladder box with massless internal propagators and four external momenta $p_i$, entering $G$ at $v_i$ subject to $p_1^2 = p_2^2 = 0$.

	Then $\int I_G\ \Omega$ is linearly reducible. If the kinematics are parametrized by $p_3^2 = s z \bar{z}$, $p_4^2 = s(1-z)(1-\bar{z})$, $s = (p_1 + p_2)^2$ and $t = (p_1 + p_4)^2 = sx$, then each coefficient of the Laurent expansion of $s^{-\sdd(G)}\int I_G\ \Omega$ (in the dimension $\dimension \in 2\N - 2\varepsilon$ and/or the indices $\EP_e \in \Z + \EPE_e \varepsilon$) is a polylogarithm from the algebra
$\MZV \tp \BarIntegralsRegulars[0](S)$ for
	\begin{align}
		S &=
		\set{
			z,\bar{z},
			1-z, 1-\bar{z},
			z-\bar{z},
			1- z \bar{z},
			1-z-\bar{z},
			z\bar{z}-z-\bar{z},
			x+z-z\bar{z},
			x+\bar{z}-z\bar{z}
		}
		\nonumber\\
		&\quad\ \cup
		\setexp{x+\alpha - \beta z\bar{z} -\gamma (1-z)(1-\bar{z})}{\alpha,\beta,\gamma \in \set{0,1}}
		.
		\label{eq:ladderbox-full-kinematics} %
	\end{align}
	If only one leg $p_3^2 = s+t+u$ is off-shell (and $p_1^2 = p_2^2 = p_4^2 = 0$), then $\int I_G\ \Omega$ has coefficients that lie in the class
	\begin{equation}
		\MZV \tp \BarIntegralsRegulars[0](\set{s,t,u,s+t,s+u,t+u,s+t+u}).
		\label{eq:ladderbox-oneoff-kinematics}%
	\end{equation}
	When all external momenta are lightlike ($p_1^2 = \cdots = p_4^2 = 0$), the coefficients are rational linear combinations of harmonic polylogarithms and multiple zeta values:
	\begin{equation}
		\MZV \tp \BarIntegralsRegulars[0](\set{x,x+1})
		= \MZV \tp \HlogAlgebra(\set{0,-1})(x) \tp \Q\left[ x,\frac{1}{x}, \frac{1}{x+1} \right].
		\label{eq:ladderbox-onshell-kinematics}%
	\end{equation}
\end{theorem}
\begin{figure}
	\centering
	$
		\Graph[1.1]{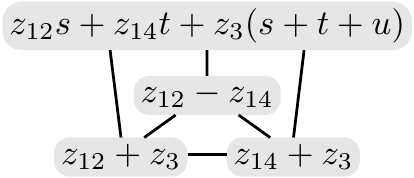}
			\hspace{-5mm}\xrightarrow{z_{14} \rightarrow 0}\hspace{-5mm}
		 \Graph[1.05]{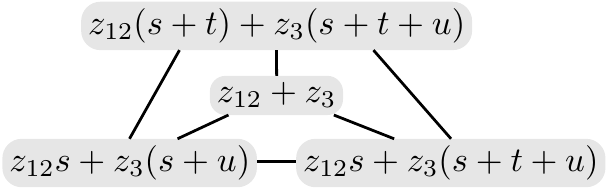}
			\hspace{-5mm}{\xrightarrow{z_{12} \rightarrow 0}} \hspace{-2mm}
		 \Graph[1.1]{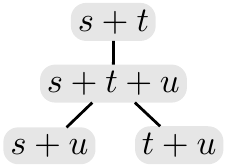}
	$
	\caption{Compatibility graph reductions for ladder boxes with one off-shell momentum $p_3^2 = s + t + u$.}%
	\label{fig:cg-box-oneoff}%
\end{figure}%
\begin{proof}
	If $\int I_G\ \Omega$ is divergent at the expansion point, we apply corollary~\ref{cor:convergent-anareg-integrand} to express it in terms of convergent integrals (with shifted $\dimension$ and $\EP_e$) and process each of those as follows.
	Using \eqref{eq:ladderbox-projective}, we can express $\int I_G\ \Omega$ as a projective integral with an integrand that has a compatibility graph 
	\begin{equation*}
		(S,C)
		=\left(
			S^{\ForestDeltaBoxSymbol} \cupdot \set{\phipol/\psipol},
			C^{\ForestDeltaBoxSymbol} \cupdot \setexp{\set{\phipol/\psipol, f}}{f \in S^{\ForestDeltaBoxSymbol}}
		\right)
	\end{equation*}
	obtained by adjoining the polynomial $\phipol/\psipol = z_{12}s + z_{14}t + z_3 p_3^2 + z_4 p_4^2$ and endowing it with compatibilities to all other polynomials.
	
	In the case where $p_4^2 = 0$, the reduction with respect to $z_4$ is simple because it does not interact with the $z_4$-independent $\phipol/\psipol = z_{12} s + z_{14} t + z_3(s+t+u)$. We obtain the leftmost graph of figure~\ref{fig:cg-box-oneoff} and compute the subsequent reductions of $z_{14}$ and $z_{12}$ as shown. Note that the final compatibility graph is isomorphic to $S^{\TriangleSymbol}$ of a triangle function $\GfunTriangle{G}$. The same argument as used in the proof before shows that the integration constants are in $\MZV$ (only $z_{12}-z_{14}$ could potentially introduce alternating sums, but the zero $z_{12}$ in $z_{14}$ decouples through its positive vanishing degree).

	When all momenta $p_1^2=\cdots=p_4^2=0$ hit the light cone, $\phipol/\psipol = s(z_{12} + x z_{14})$ is independent of both $z_3$ and $z_4$ and the situation becomes even simpler. The reductions result in complete graphs on the vertices (we do not list the monomials here)
	\begin{gather*}
		(S,C)_{z_4}
		= (\set{ z_{12}+z_{3}, z_{12}-z_{14}, z_{12}+z_{14} x, z_{14}+z_{3}}, \ldots)
		,\\
		(S,C)_{z_4,z_3}
		= (\set{z_{12}-z_{14}, z_{12}+z_{14} x}, \ldots)
		\qquad\text{and}\qquad
		(S,C)_{z_4,z_3,z_{14}}
		= (\set{1+x}, \emptyset).
	\end{gather*}
	The most general case covered by the theorem ($p_3^2, p_4^2 \neq 0$) goes through completely analogously, only with the complication that the compatibility graphs are bigger. We checked that we obtain linear reductions for the sequence $z_{3}$, $z_{14}$, $z_{12}$ of integration (setting $z_4 = 1$) and end up with $15$ irreducible polynomials (apart from the monomials) that lead to \eqref{eq:ladderbox-full-kinematics}. Finally we also verified that the integration constants are multiple zeta values, for every order of the limits $\AnaReg{z}{0}$, $\AnaReg{\bar{z}}{0}$ and $\AnaReg{x}{0}$. For these computations we used our program {\HyperInt}.
\hide{	\begin{multline*}
		z_{12}+z_{4}, 
		z_{12}-z_{14}, 
		z_{14}+z_{4}, 
		(1-z)(1-\bar{z})z_{4}+z_{12}+xz_{14}, 
		\\
		(1-z-\bar{z})z_{4}+(x-z\bar{z})z_{14}+z_{12},
		(1-z)(1-\bar{z})z_{4}+(1-z\bar{z})z_{12}+xz_{14},
		\\
		(1-z)(1-\bar{z})z_{4}+(x-z\bar{z})z_{14}+z_{12},
		(1-z)(1-\bar{z})z_{4}(z_4+z_{12})-z\bar{z}z_{12}(z_4+z_{14})+(z_{12}+z_4)(z_{12}+xz_{14})
	\end{multline*}}%
\end{proof}
This approach was tested in practice (using our program {\HyperInt}) and produced new results as we shall recall in section~\ref{sec:ex-ladderboxes}. There we also present some generalizations of theorem~\ref{theorem:ladderbox-kinematics} to wider classes of Feynman graphs.

\subsection{Landau varieties}\label{sec:landau-varieties}
To understand the origin of compatibility constraints to linear reductions, we follow the approach based on algebraic geometry that was developed in \cite[sections $5$--$6$]{Brown:PeriodsFeynmanIntegrals}. Our exposition here is extremely short and we must refer to the excellent original presentation, which includes many more details, proofs and illuminating examples. We realized that several properties of compatibility graphs are not yet completely understood and demand a thorough analysis in the future. 

Here we can only give a glimpse of this interesting and important subject, but still like to explain how proposition~\ref{prop:cg-reduction} comes about.

\subsubsection{Singularities of integrals}
Our aim is to study the singularities of partial integrals $I_k$, so the main reference is the book \cite{Pham:Singularities}. Following \cite{Brown:PeriodsFeynmanIntegrals}, we pass from the affine ambient space $\C^N$ to the compact manifold $P \defas \Projective^1_1 \times \cdots \times \Projective^1_N$ by projectivization of each variable $z_i$, which we view as coordinate $\Projective^1_i \setminus \set{\infty} \longrightarrow \C$, $[z_i:1] \mapsto z_i$ on its individual copy $\Projective^1_i \defas \CP^1$ of projective space.
Any subset $K \subset [N] \defas \set{1,\ldots,N}$ of variables induces the natural projection $\pi_{K^c}\colon P \longrightarrow \Projective^1_K \defas \prod_{i \in K} \Projective^1_i$ with fibre $\prod_{i \notin K} \Projective^1_i$.

\begin{definition}\label{def:landau-variety}
	Fix a projection $\pi \defas \pi_{K^c}$ and let $X^1 \subset P$ denote a closed analytic subset with an associated Whitney stratification $P=X^0 \supsetneq X^1 \supsetneq \cdots \supsetneq X^N$ of closed analytic subsets $X^k$ of codimension $k$ such that $X^k \setminus X^{k+1}$ are smooth manifolds. The \emph{critical set} $c A$ of an irreducible component $A \subset X^k \setminus X^{k+1}$ (open stratum) consists of the points where $\restrict{\pi}{A}$ does not submerse on $\Projective^1_K$:
\begin{equation}
	c A \defas \setexp{x \in A}{ 
			\rank D_x (\restrict{\pi}{A})
			< \abs{K}
		}.
	\label{eq:def-critical-set}%
\end{equation}
The \emph{Landau variety} $L(X,\pi)$ is the codimension $1$ part of $\pi\left( \bigcup_A c A \right)$, where the union runs over all strata $A$ of $X$.
\end{definition}
\begin{remark}
	Every irreducible polynomial $f \in \Q[z_1,\ldots,z_N]$ admits a unique irreducible lift $f' \in \Q[z_1^{},\prdu{z}_1,\ldots,z_N^{},\prdu{z}_N]$ with $f = \restrict{f'}{\prdu{z}_1=\cdots=\prdu{z}_N=1}$ such that $f'$ is homogeneous in each pair $\set{z_i^{},\prdu{z}_i}$ of coordinates (multiply each power $z_i^k$ in $f$ with ${\prdu{z}_i}^{d-k}$ where $d=\deg_{i} f$) and therefore defines a unique hypersurface $\Vanishing(f) \defas \Vanishing(f') \subset P$. Thus we can describe Landau varieties with irreducible polynomials in $z$, except for the hyperplanes $B_i^{\infty} \defas \Vanishing(\prdu{z}_i)$ at infinity. We also write $B_i^0 \defas \Vanishing(z_i)$, $B_i \defas B_i^0 \cupdot B_i^{\infty}$ and $B \defas \bigcup_{i=1}^N B_i$.
\end{remark}
Suppose $S \subset \Q[z_1,\ldots,z_N]$ is a set of irreducible polynomials and we consider an integrand $I_0 \in \BarIntegralsRegulars(S)$. It is (locally) analytic on $P\setminus X^1$, the complement of the codimension $1$ subset $X^1 \defas \bigcup_{f \in S} \Vanishing(f) \cup \bigcup_{i=1}^N B_i$. The Landau variety precisely describes the singularities of integrals of $I_0$. We cite\footnote{The proof \cite[theorem~58]{Brown:PeriodsFeynmanIntegrals} is formulated for logarithms $I_0 \in \regulars(S) \tp \Q[\log(z_i),\log(f)\colon 1 \leq i \leq N,f \in S]$ only, but it is clear that it generalizes to all iterated integrals $\BarIntegrals(S)$.}
\begin{theorem}\label{eq:theorem-analytic-outside-landau}%
	If the integrand $I_0 \in \BarIntegralsRegulars(S)$ is analytic on $(0,\infty)^N$ and \eqref{eq:total-integral} finite ($I_0$ is integrable), then any partial integral $I_K \defas \prod_{i \in K} [\int_0^{\infty} \dd z_i]\ I_0$ defines a multivalued analytic function on $\Projective^1_{K^c} \setminus L(X, \pi_K)$. It is free of singularities on $\prod_{i \notin K} (0,\infty)$.
\end{theorem}
It follows that $S$ is linearly reducible if all Landau varieties $S^k \defas L(X,\pi_{\set{1,\ldots,k}})$ are linear in $z_{k+1}$, as they constitute a linear reduction of $S$ by themselves. 
However, the converse is not necessarily true: A Landau variety $L(X,\pi_{\set{1,\ldots,k}})$ could contain components which are spurious in the sense that those never actually occur as singularities of the partial integral $I_k$ of any convergent integrand $I_0$. If such a spurious component is non-linear, we would not be able to detect linear reducibility from the computation of the Landau varieties.\footnote{We are not aware of an explicit example where this happens, because it is in general very difficult to compute Landau varieties exactly. This remains an interesting project for future research.}

Now we will focus on the computation of the Landau varieties of the initial set $S$.

\subsubsection{Approximations}
It is difficult to compute Landau varieties $L(X,\pi)$ exactly, because in our applications they are very degenerate and typically have many components. In particular, explicit formulas for resultants of several polynomials cannot be applied in practice (the coefficients of the polynomials are not general enough but satisfy algebraic relations) and in any case involve exceedingly complicated expressions. The exact knowledge of Landau varieties in our setting is basically limited to the first steps in the reduction of a graph hypersurface $X = B \cup \Vanishing(\psipol_G)$ as computed in \cite{Brown:PeriodsFeynmanIntegrals}.
\begin{proposition}\label{prop:reduction-dodgsons}
	If $S=\set{\psipol}$ is the first Symanzik polynomial $\psipol_G$ of a graph $G$, then the Landau varieties
	$
		L(X, \pi_{\set{1,\ldots,k}})
	$
	of $X = X_S \cup B$ are linear in all Schwinger variables for $k \leq 4$ and given by Dodgson polynomials \eqref{eq:def-dodgson}. Explicitly, $L(X, \pi_{\set{1,\ldots,k}})$ is
	\begin{equation*}
		\bigcup_{i>k} B_i \ \cup\ \text{irreducible factors of}\ 
		\setexp{\dodgson^{I,J}_K}{\abs{I} = \abs{J}\ \text{and}\ (I \cup J) \cupdot K = \set{1,\ldots,k}}.
	\end{equation*}
	This still holds when $k = 5$, except that apart from Dodgson polynomials, $L(X,\pi_{\set{1,\ldots,5}})$ contains also the five-invariant $\fiveinv{1,2,3,4,5}$ which can be non-linear.
\end{proposition}
The practical approach computes upper bounds on $L(X,\pi_{\set{1,2}}) \subseteq L(L(X,\pi_{\set{1}}), \pi_{\set{2}})$ by iteration of one-dimensional projections. If we just set $S^k \defas L(S^{k-1} \cup B, \pi_k)$, this is precisely the \emph{simple reduction} of lemma~\ref{lemma:simple-reduction}. Taking intersections over different representations of $\pi_K = \pi_{k_{\sigma(1)}} \circ \cdots \circ \pi_{k_{\sigma(r)}}$ as iterated one-dimensional projections into account ($K = \set{k_1,\ldots,k_{r}}$) yields the \emph{Fubini reduction} of definition~\ref{def:Fubini-reduction}.

These algorithms already fail to reproduce proposition~\ref{prop:reduction-dodgsons} and contain more and more \emph{spurious} components (polynomials which are not present in the actual Landau variety) as the dimension of the projection increases. When these become non-linear, linear reducibility can not be detected (nor disproven) with this method. We also need to avoid spurious polynomials to put tighter constraints on the periods that an integral can evaluate to (see section~\ref{sec:propagator-questions} for an example).

Francis Brown analyzed the Landau variety $L(X, \pi_{\set{i,j}})$ of a two-dimensional projection $\pi_{\set{i,j}} = \pi_j \circ \pi_i$ of type $(1,1)$-hypersurfaces $X \ni f = f^{ij} z_i z_j + f^i_j z_i + f_j^i z_j + f_{ij}$ in \cite{Brown:PeriodsFeynmanIntegrals}. He observed that not all pairs of iterated resultants $\resultant{\resultant{f_1}{f_2}{i}}{\resultant{f_3}{f_4}{i}}{j}$ need to be taken into account, but it suffices to only consider those where only three different polynomials (\emph{grandparents}) appear (say $f_2 = f_3$). In our definition~\ref{def:cg-reduction} we further restrict those compatibilities between $\resultant{f_1}{f_2}{i}$ and $\resultant{f_2}{f_3}{i}$ to the case when $\set{f_1,f_2,f_3}$ are mutually compatible with each other.\footnote{In the original formulation \cite{Brown:PeriodsFeynmanIntegrals}, the compatibilities $C$ do not influence the compatibilities $C_i$ at all.}

\subsubsection{Proof of proposition~\ref{prop:cg-reduction}}
We prove a stronger statement by induction: Suppose all iterated reductions $(S^k,C^k)$ exist and let $K \subset [N]$ be any set of variables disjoint from $[k]$. Then
\begin{equation}
	L(S \cup B, \pi_{[k] \cup K})
	\subseteq
	\bigcup_{\text{cliques}\ H \subseteq S^k \cup B_{>k}} L\left( H, \pi_{K} \right),
	\quad\text{where}\quad
	B_{>k} \defas\bigcup_{i>k} B_i,
	\label{eq:landau-clique-bound}%
\end{equation}
is bounded by the union of Landau varieties (with respect to $\pi_K$) of those subsets $H$ of polynomials that form a clique\footnote{A clique is a complete graph, so we require mutual compatibilities $\set{f,g} \in C^k$ for all $f,g\in H$.} in $\left( S^k,C^k \right)$. This restriction to cliques is essential, as otherwise \eqref{eq:landau-clique-bound} is just the simple bound $L(S\cup B, \pi_{[k]\cup K}) \subseteq L(L(S\cup B,\pi_{[k]}), \pi_K)$.

Since we start with the complete graph $C^0 = \binom{S}{2}$, \eqref{eq:landau-clique-bound} is trivial for $k=0$. So let us assume \eqref{eq:landau-clique-bound} holds for some $k$ and let $K \cap [k+1] = \emptyset$. Then
\begin{equation*}
	L\left(S\cup B, \pi_{[k+1]\cup K}\right)
	\subseteq L\left( L(S\cup B, \pi_{[k]}), \pi_{K \cup \set{k+1}} \right)
	\subseteq \bigcup_{\text{cliques}\ H \subseteq S^k \cup B_{>k}} L\left(H,\pi_{K \cup \set{k+1}}\right)
\end{equation*}
is granted and we approximate this further with $L(H,\pi_{K \cup \set{k+1}}) \subseteq L( L(H,\pi_{k+1}), \pi_K)$. The one-dimensional projection is easy to compute and corresponds to the simple reduction \eqref{eq:all-resultants}. Following \cite[lemma~76]{Brown:PeriodsFeynmanIntegrals}, we obtain
\begin{equation}
	L(H,\pi_{k+1})
	= \setexp{\resultant{f}{g}{k+1}}{f,g \in H \cup B_{k+1}}
	\subseteq S^{k+1}
	\label{}
\end{equation}
as all polynomials in $H$ are compatible with each other. This shows
$L(H,\pi_{K \cup \set{k+1}}) \subseteq L(S^{k+1}\cup B_{>k+1}, \pi_K)$, but we need to improve this bound and replace it by the union of $L(H',\pi_K)$ for cliques $H' \subseteq S^{k+1}\cup B_{>k+1}$.
	
So let $A$ denote a stratum of the stratification $X$ generated by $X^1 = H \subset \Projective^1_{[k]^c}$. After passing to the smallest subset of $H$ that still generates $A$, we may assume that $A \subseteq \bigcap_{i=1}^r \Vanishing(f_i)$ for $H=\set{f_1,\ldots,f_r}$ (beware that $A$ may have any codimension above or equal to $r$).
Let $f_i = a_i z_{k+1} + b_i$ denote the coefficients of $f_i$ and consider an arbitrary point $x \in cA$ in the critical set of $A$ (relative to $\pi_{K\cup\set{k+1}}$). We distinguish two cases:
\begin{enumerate}
	\item For some $1 \leq i \leq r$, we have $a_i(x) \neq 0$. In this situation, 
	\begin{equation}
		\left[ \bigcap_{j=1}^r \Vanishing(f_j) \right] \setminus \Vanishing(a_i)
		= \left[ \Vanishing(f_i) \cap \bigcap_{j \neq i} \Vanishing(\resultant{f_i}{f_j}{k+1}) \right] \setminus \Vanishing(a_i)
		\label{}
	\end{equation}
	has precisely one point in each fibre $\Projective^1_{k+1}$. Hence over the complement of $\Vanishing(a_i)$, the image $A' \defas \pi_{k+1}(A)$ is a smooth stratum of $H' \defas \bigcup_{j\neq i} \Vanishing(\resultant{f_i}{f_j}{k+1}) \subset \Projective^1_{[k+1]^c}$ with the same dimension. 
	The trivial fibration of the tangent spaces
	\begin{equation*}
		\setexp{v=\sum_{j=k+1}^N v^j \partial_j}{0 = v(f_1) =\cdots = v(f_r)}
		= \setexp{-v\left(\frac{b_i}{a_i}\right)\partial_{k+1} + v}{v \in T_{x'} H'}
	\end{equation*}
	over $x' \defas \pi_{k+1} (x)$ passes on to $A$, in particular $\pi_{k+1} (T_x A) = T_{x'} A'$. This proves $\pi_{K \cup\set{k+1}}(T_x A) = \pi_{K} (T_{x'} A')$ and therefore $\pi_{K \cup \set{k+1}}(cA) \subseteq \pi_K (cA')$ over the complement of $\Vanishing(a_i)$. Thus $\pi_{K\cup\set{k+1}}(x) \in L(H',\pi_K)$ for the clique $H'$.\footnote{Each pair $\resultant{f_i}{f_j}{k+1}$, $\resultant{f_i}{f_{j'}}{k+1}$ is compatible in definition~\ref{def:cg-reduction} because $\set{f_i,f_j,f_{j'}} \subseteq H$ form a triangle in the complete graph $H$.}

	\item For all $1\leq i \leq r$, the coefficients $a_i(x)=0$ vanish (at $x$). At these points,
\begin{equation*}
	\bigcap_{i=1}^r \Vanishing(f_i) \cap \bigcap_{i=1}^{r} \Vanishing(a_i)
	= \Projective^1_{k+1} \times \left( \bigcap_{i=1}^r \Vanishing(a_i) \cap \bigcap_{i=1}^{r} \Vanishing(b_i) \right)
\end{equation*}
is parallel to the fibre and the stratification reduces to the coefficients. Hence we find $\pi_{K\cup\set{k+1}}(x)\in L(H', \pi_K)$ for $H' \defas \setexp{a_i,b_i}{1 \leq i \leq r} \cup B_{>k+1} \subseteq S^{k+1}$. But generically, this set $H'$ is not a clique because $a_i = \resultant{f_i}{\infty}{k+1}$ and $b_j = \resultant{f_j}{0}{k+1}$ are not compatible for $i \neq j$. So unless all $a_1 = \cdots = a_r = 0$ vanish identically (then only $b_i$ appear and $H'$ is indeed a clique), we must find another set.

Then choose some $1 \leq i \leq r$ such that $a_i \neq 0$ is not identically zero and consider
\begin{equation}
	H'
	\defas
	\set{a_i,b_i}
	\cup
	\setexp{\resultant{f_i}{f_j}{k+1}}{j \neq i}
	\cup B_{>k+1},
	\label{eq:landau-clique-proof-parallel}%
\end{equation}
which indeed defines a clique (recall that $a_i = \resultant{f_i}{\infty}{k+1}$ and $b_i = \resultant{f_i}{0}{k+1}$). Since $\Vanishing(a_i) \cap \Vanishing(\resultant{f_i}{f_j}{k+1}) = \Vanishing(a_i) \cap \Vanishing(a_i b_j - a_j b_i) = \Vanishing(a_i) \cap \left( \Vanishing(a_j) \cup \Vanishing(b_i) \right)$and analogously
$\Vanishing(b_j) \cap \Vanishing(b_i) \subseteq \Vanishing(b_i) \cap \Vanishing(\resultant{f_i}{f_j}{k+1})$, we see that
\begin{equation*}
	\left( \bigcup_{f \in H'} \Vanishing(H') \right) \cap \Vanishing(a_i) \cap \Vanishing(b_i)
	\supseteq \left( \bigcup_{j=1}^{r} \big[\Vanishing(a_j) \cup \Vanishing(b_j)\big] \right) \cap \Vanishing(a_i) \cap \Vanishing(b_i).
\end{equation*}
Therefore also the stratification of the clique \eqref{eq:landau-clique-proof-parallel} generates the stratum $A$ and $\pi_{K\cup\set{k+1}}(x) \in L(H',\pi_K)$.
\end{enumerate}
This discussion completes the proof: For each point $x\in cA$ in the critical set of a stratum $A$ that is generated by a clique $H=\set{f_1,\ldots,f_r} \subseteq S^k\cup B_{>k}$, we showed that $S^{k+1}\cup B_{>k+1}$ contains a clique $H'$ such that $\pi_{K\cup\set{k+1}}(x) \in L(H',\pi_K)$. Hence we covered
\begin{equation*}
	L(H,\pi_{K\cup\set{k+1}})
	= \bigcup_A
	\pi_{K\cup\set{k+1}}(cA)
	\subseteq \bigcup_{\text{cliques}\ H' \subseteq S^{k+1} \cup B_{>k+1}} 
	L(H',\pi_K).
\end{equation*}
\begin{remark}[Interpretation via partial fractioning]
	If we only consider rational functions (and ignore the hyperlogarithms), the clique-rule in the definition~\ref{def:cg-reduction} is very easy to understand: Suppose the rational function $f \in \regulars(S)$ combines only denominators that are mutually compatible:
	\begin{equation*}
		F \in \sum_{\text{clique}\ H\subseteq S} \Q\left[ \setexp{z_i}{1\leq i \leq N} \cup \setexp{f^{-1}}{f \in H} \right].
	\end{equation*}
	The decomposition of 
$
	\frac{1}{f_i f_j} 
	= \left( \frac{\resultant{f_i}{\infty}{k}}{f_i}
		- \frac{\resultant{f_j}{\infty}{k}}{f_j} \right)
	\frac{1}{\resultant{f_i}{f_j}{k}}
$ into partial fractions shows that
\begin{equation*}
	\frac{1}{f_{1}\cdots f_{r}}
	= 
	\sum_{\mu=1}^r
	\frac{1}{f_{\mu}} \prod_{\nu\neq \mu} \frac{\resultant{f_{\mu}}{\infty}{k}}{\resultant{f_{\mu}}{f_{\nu}}{k}}
\end{equation*}
involves only coefficients (constant with respect to the integration variable $z_k$) whose denominators are cliques of resultants. This generalizes to higher powers of the polynomials $f_i$ in the denominator. Clearly, if $F$ does not involve a denominator that mixes the polynomials $f_i$ and $f_j$ say, then the partial fraction decomposition of $F$ (with respect to some $z_k$) can not introduce a resultant $\resultant{f_i}{f_j}{k}$ into the denominator.

Morally, there should exist a direct (combinatoric) way to prove proposition~\ref{prop:cg-reduction} by carrying over this very simple argument, avoiding the rather heavy machinery we resorted to above.
\end{remark}
\subsubsection{Intersecting compatibility graphs}\label{sec:cg-intersection}
The construction of compatibility graphs $(S^{K,k_{r+1}},C^{K,k_{r+1}}) \defas(S^K,C^K)_{k_{r+1}}$ is defined along some order $K=(k_1,\ldots,k_r)$ of the variables. It is tempting to construct a finer bound on the Landau varieties by intersecting 
\begin{equation}
	\left(S^{K}, C^{K}\right) 
	\defas \bigcap_{e \in K} \left(S^{K\setminus e}, C^{K\setminus e}\right)_e,
	\label{eq:cg-intersection}%
\end{equation}
just as in the Fubini algorithm \eqref{eq:def-Fubini-reduction}. But in fact it is not obvious that this is permissible. For example, consider the following situations:
\begin{enumerate}
	\item Suppose that the orders of reduction give the same polynomials $S^{1,2} = S^{2,1}$ but different compatibilities. Say $\set{f_1,f_2} \in C^{1,2}\setminus C^{2,1}$ and $\set{f_3,f_4} \in C^{2,1}\setminus C^{1,2}$ are compatible in the graph for one of the reduction orders, but not in the other. These compatibilities will drop out in the intersection $C^{1,2} \cap C^{2,1}$.

		But it may still be that the resultants $\resultant{f_1}{f_2}{3} \in S^{1,2,3}$ and $\resultant{f_3}{f_4}{3} \in S^{2,1,3}$ coincide and give a contribution to the Landau variety $L(X,\pi_{\set{1,2,3}})$, which might be missed if the compatibilities are removed through the intersection.

	\item More drastically, it might be that $S^{1,2} \neq S^{2,1}$ contain different polynomials. Then all resultants $\resultant{f_1}{f_2}{3}$ with $f_1 \in S^{1,2}\setminus S^{2,1}$ will be lost after the intersection, even if they contribute polynomials to $S^{1,2,3}$ that could also occur in $S^{2,1,3}$.
\end{enumerate}
Of course we know that $S^{1,2} \cap S^{2,1}$ is indeed an upper bound for the Landau variety $L(X,\pi_{1,2})$ by \eqref{eq:barintegrals-intersection}, but the problem with \eqref{eq:cg-intersection} is that in general the identity
\begin{equation*}
	(S,C)_k \cap (S',C')_k
	= (S \cap S', C\cap C')_k
\end{equation*}
does not hold for arbitrary graphs (we can construct counterexamples following the observations above). Interestingly, in our applications this intersection still always computed an upper bound on the Landau varieties which suggests that it might indeed hold for compatibility graphs (then they must obey further structure that we missed so far).

However, since a complete proof is not available yet we avoided such intersections in our applications of section~\ref{sec:recursion-reducibility}. While these intersections of compatibilities lie at heart of the original formulation of \cite{Brown:PeriodsFeynmanIntegrals}, they are not necessary for our improved reduction algorithm of definition~\ref{def:cg-reduction}.

\chapter{The {\Maple} program \HyperInt}
\label{chap:hyperint}%

\section{Introduction}
We implemented the algorithms of section~\ref{sec:hlog-algorithms} in the computer algebra system {\Maple} \cite{Maple}. 
Our aim was to compute Feynman integrals, but we designed the program such that it suits much more general calculations with polylogarithms. It is open source and may be obtained from \cite{Panzer:HyperIntAlgorithms}.

Facilities for numeric evaluations of hyper- and polylogarithms are not included, because such are not necessary for the integration algorithms and secondly there are already established programs available for this task \cite{Ginac:Introduction,VollingaWeinzierl:NumericalMpl}.

Not all features of {\HyperInt} are discussed here in full detail. Some additional functions are listed in appendix~\ref{chap:hyperint-reference} and demonstrated in \Filename{Manual.mw}.
\begin{remark}
	The program uses the \code{remember} option of {\Maple}, which creates lookup tables to avoid recomputations of functions. But some of these functions depend on global parameters as explained for instance in section~\ref{sec:factorization}.
Therefore, whenever such a parameter is changed, the function \mbox{\code{forgetAll}$()$} must be called to invalidate those lookup tables. Otherwise the program might behave inconsistently.
\end{remark}

\section{Installation and files}
\label{sec:installation}%
The program requires no installation. It is enough to load it during a {\Maple}-session by
\begin{MapleInput}
read "HyperInt.mpl";
\end{MapleInput}
if the file \code{HyperInt.mpl} is located in the current directory or another place in the search paths of {\Maple}. 
All together, we supply the following main files:
\begin{description}
	\item[\Filename{HyperInt.mpl}]
		\hfill\\
		Our implementation of the algorithms in section~\ref{sec:hlog-algorithms} together with supplementary procedures to handle Feynman graphs and Feynman integrals.

	\item[\Filename{periodLookups.m}]
		\hfill\\
		This table stores a reduction of multiple zeta values up to weight $12$ to a (conjectured) basis and similarly for alternating Euler sums up to weight $8$. It is not required to run the program, but necessary for efficient calculations involving high weights. A detailed explanation follows in section~\ref{sec:periods}.

	\item[\Filename{Manual.mw}]
		\hfill\\
		This {\Maple} worksheet explains the practical usage of {\HyperInt}. In particular it demonstrates plenty of explicit Feynman integral computations with details, explanations and further comments.

	\item[\Filename{HyperTests.mpl}]
		\hfill\\
		A series of various test cases for the program, see section~\ref{sec:hyperint-tests}. Calling {\Maple} with \code{maple HyperTests.mpl} must run without any error messages, otherwise the author will appreciate a notification of errors that occurred. These tests require that \code{periodLookups.m} can be found and loaded by {\HyperInt}.

		Since this test file entails various applications of {\HyperInt}, it supplements the manual and might support learning how to use the program.
\end{description}

\section{Representation of polylogarithms and conversions}
\label{sec:representations-conversions}
Words
$w_i=\letter{\sigma_1}\!\!\cdots\letter{\sigma_r} \safed [\sigma_1,\ldots,\sigma_r] \in \Sigma^r$
are written as lists and combined with rational prefactors $g_i$ to encode hyperlogarithms $f$ of some variable $z$
\begin{equation}
	\label{eq:maple-internal-hyperlog} %
	f 
	= [ 
				[g_1, w_1],
				[g_2, w_2],
				\ldots
		]
	\defas
	\sum_{i} 
		g_i(z) \Hyper{w_i}(z)
	.%
\end{equation}
But the building blocks for the integration algorithm are the regularized limits
\begin{equation}
	\label{eq:maple-internal-reginf} %
	f
	=	
	[ 
		[g_1,[ w_{1,1}, \ldots, w_{1,r_1}] ],
		[g_2,[ w_{2,1}, \ldots, w_{2,r_1}] ],
		\ldots
	]
	\defas
	\sum_i g_i
	\prod_{j=1}^{r_i} \Hyper{\WordReg{}{\infty}(w_{i,j})}(\infty),
\end{equation}
encoded by a list of pairs of rational prefactors $g_i$ and lists of words $w_{i,j} = \WordReg{0}{}\left(w_{i,j}\right)$ not ending on $\letter{0}$.
The products in this apparently wasteful representation are \emph{not} auto\-matically shuffled out and replaced by $\Hyper{\WordReg{}{\infty}(\shuffle_j w_{i,j})}(\infty)$, as suggested by lemma~\ref{lemma:iint-shuffle}, but instead kept separated for two reasons:
\begin{enumerate}
	\item Empirically, this expansion tends to increase the number of terms considerably.

	\item Our algorithm from section~\ref{sec:reglim-reginf} to compute $\AnaReg{t}{0}$ generates products of words with \emph{different sets of letters} distinguished by their vanishing degrees (corollary~\ref{corollary:reglim-degree-decouple}). Mixing such letters via shuffles introduces spurious letters in subsequent integration steps which we want to avoid.
\end{enumerate}
These representations make the implementation of the algorithms of section~\ref{sec:hlog-algorithms} straightforward, but to improve readability, {\HyperInt} understands the notations
\begin{align*}
	\code{Hlog}\left(z, [\sigma_1,\ldots,\sigma_r]\right)
	& \defas
	\Hyper{\letter{\sigma_1}\!\!\cdots\letter{\sigma_r}}(z),
	\\
	\code{Hpl}\left([n_1,\ldots,n_r],z\right)
	&\defas
	H_{n_1,\ldots,n_r}(z)
	=
	\Hyper{\underline{n_1},\ldots,\underline{n_r}}(z)
	\quad\text{and}
	\\
	\code{Mpl}\left( [n_1,\ldots,n_r], [z_1,\ldots,z_r] \right)
	& \defas
	\Li_{n_1,\ldots,n_r}\left(z_1,\ldots,z_r\right)
\end{align*}
for hyperlogarithms \eqref{eq:def-Hlog}, \eqref{eq:def-Hlog-noshuffle}, harmonic polylogarithms \eqref{eq:compressed-Hpl-as-Hlog} and multiple polylogarithms \eqref{eq:def-Mpl}. {\HyperInt} extends the native function \code{convert}$(f,\code{form})$ such that it can transform expressions $f$  which may contain any of the functions
\begin{equation*}
\set{
	\code{log},
	\code{ln},
	\code{polylog},
	\code{dilog},
	\code{Hlog},
	\code{Mpl},
	\code{Hpl}
}
\end{equation*}
into one of the possible target formats:
\begin{description}
	\item[$\code{form} \in \set{\code{Hlog}, \code{Hpl}, \code{Mpl}}$:]
	\hfill\\
	Expresses $f$ in terms of \code{form}-functions, using \eqref{eq:Hlog-as-Mpl}, \eqref{eq:Mpl-as-Hlog} and \eqref{eq:compressed-Hpl-as-Hlog}.

	\item[$\code{form}=\code{HlogRegInf}$:]
	\hfill\\
	Writes $f$ in the list representation \eqref{eq:maple-internal-reginf} through a M\"{o}bius transformation \eqref{eq:Moebius-transformation}.

	\item[$\code{form} = \code{i}$]
	\hfill\\
	Equal to $\code{form} = \code{Hlog}$, but produces the notation
	\begin{equation*}
		i[0,\sigma_n/z,\ldots,\sigma_1/z,1]
		=
		i[0,\sigma_n,\ldots,\sigma_1,z]
		\defas
		\code{Hlog}\left( z, [\sigma_1,\ldots,\sigma_n] \right)
	\end{equation*}
	that is used in {\zetaprocedures} \cite{Schnetz:ZetaProcedures}. In this program, the result can for example be evaluated numerically with $\code{evalz}\left( \cdot \right)$.
\end{description}
\begin{example}The dilogarithm $\Li_2(z)$ has representations
\begin{MapleInput}
convert(polylog(2,z), Hlog);
\end{MapleInput}
\begin{MapleMath}
	-\Hlog{1}{0, 1/z}
\end{MapleMath}
\begin{MapleInput}
convert(polylog(2,z), HlogRegInf);
\end{MapleInput}
\begin{MapleMath}
	[[1, [[-1+z, -1]]], [-1, [[-1, -1]]]]
\end{MapleMath} 
\begin{MapleInput}
convert(polylog(2,z), Mpl);
\end{MapleInput}
\begin{MapleMath}
	\Mpl{2}{z}
\end{MapleMath} 
\begin{MapleInput}
convert(polylog(2,z), i);
\end{MapleInput}
\begin{MapleMath}
	-i[0, 1/z, 0, 1]
\end{MapleMath}
\end{example}
Due to the many functional relations, a general polylogarithm $f(\vec{z})$ has many different representations. In particular, the representation \eqref{eq:maple-internal-reginf} is far from being unique.

It is therefore crucial to be able to express polylogarithms in a basis in order to simplify results and to detect relations. Such a representation is furnished by \eqref{eq:barintegrals-hlog-products} (the iteration of proposition~\ref{prop:reginf-as-hlog}) and implemented as the function
\begin{equation*}
	\code{fibrationBasis}\left(f, [z_1,\ldots,z_r], F\right)
	=
	\sum_i
		\Hyper{w_{i,1}}(z_1)
		\cdot
		\ldots
		\cdot
		\Hyper{w_{i,n}}(z_n)
		\cdot
		c_i,
	\label{eq:fibration-basis} %
\end{equation*}
which writes a polylogarithm $f(\vec{z})$ as the unique linear combination of products of hyperlogarithms such that each $w_{i,j} \in T(\Sigma_j)$ has algebraic letters 
$
	\Sigma_j
	\subset
	\overline{\C(z_{i+1},\ldots,z_n)}
$.
The result depends on the order $\vec{z} = [z_1,\ldots,z_r]$ of variables and entails constants
\begin{equation}
	\label{eq:fibration-basis-constants} %
	c_i
	\in
	\AnaReg{z_n}{0} \ldots \AnaReg{z_1}{0}
	\AnaReg{z}{\infty} \HlogAlgebra(\Sigma)(z).
\end{equation}
If the optional table $F$ is supplied, the result will be stored as $F_{[w_{i,1},\ldots,w_{i,n}]} = c_i$.
\begin{example}
	This function can be used to obtain functional relations between polylogarithms. For example,
	\begin{MapleInput}
fibrationBasis(polylog(2,1-z), [z]);
convert(%, Mpl);
	\end{MapleInput}
	\begin{MapleMath}
		-\Hlog{z}{1, 0} + \mzv{2}
		\\
		-\Mpl{2}{z} + \ln(z) \Mpl{1}{z} + \mzv{2}
	\end{MapleMath}%
	reproduces the classic identity $\Li_2(1-z) = \mzv{2} - \Li_2(z) - \log z \log(1-z)$. Similarly, we obtain the inversion relation for $\Li_5\left(-\frac{1}{x} \right) = \frac{1}{120} \ln^5 x + \frac{\mzv{2}}{6} \ln^3 x + \frac{7}{10}\mzv[2]{2} \ln x + \Li_5(-x)$:
	\begin{MapleInput}
fibrationBasis(polylog(5, -1/x), [x]):
convert(%, Mpl);
	\end{MapleInput}
	\begin{MapleMath}
		\frac{1}{6}\mzv{2}\ln(x)^3+\frac{1}{120} \ln(x)^5 + \Mpl{5}{-x} + \frac{7}{10} \mzv[2]{2} \ln(x)
	\end{MapleMath}
	As an example involving multiple variables, the five-term relation of the dilogarithm is recovered as
\begin{MapleInput}
polylog(2,x*y/(1-x)/(1-y))-polylog(2,x/(1-y))-polylog(2,y/(1-x)):
fibrationBasis(%, [x, y]);
\end{MapleInput}
	\begin{MapleMath}
		\Hlog{y}{0, 1}+\Hlog{x}{0, 1}-\Hlog{x}{1}\Hlog{y}{1}
	\end{MapleMath}
\end{example}
Note that for more than one variable, each choice $\vec{z}$ of order defines a different basis and a function may take a much simpler form in one basis than in another. 
For example, $\Li_{1,2}(y,x) + \Li_{1,2}(\frac{1}{y}, xy)$ is just
\begin{MapleInput}
f:=Mpl([1,2], [y,x])+Mpl([1,2], [1/y,y*x]):
fibrationBasis(f, [x,y]);
\end{MapleInput}
\begin{MapleMath}
	\Hlog{x}{0, 1/y, 1}+\Hlog{x}{0, 1, 1/y}
\end{MapleMath}
but in another basis takes the form
\begin{MapleInput}
fibrationBasis(f, [y,x]);
\end{MapleInput}
\begin{MapleMath}
	\Hlog{y}{0, 1, 1/x}+\Hlog{y}{0, 1/x}\Hlog{x}{1}
	\\
	-\Hlog{y}{0, 0, 1/x}-\Hlog{y}{0, 1}\Hlog{x}{1}
\end{MapleMath}
We like to emphasize that every order $\vec{z}$ defines a true basis without relations. In particular this means that $f=0$ if and only if $\code{fibrationBasis}(f, \vec{z})$ returns $0$, no matter which order $\vec{z}$ was chosen.

Analytic continuation in a variable $z$ is performed along a straight path, therefore the result can be ambiguous when this line contains a point where the function is not analytic. In this case, branches above and below the real axis will be distinguished by an auxiliary variable
\begin{equation}
	\delta_z
	=
	\begin{cases}
		+1 & \text{when $z \in \Halfplane^{+}$,} \\
		-1 & \text{when $z \in \Halfplane^{-}$.} \\
	\end{cases}
	\label{eq:which-half-plane}
\end{equation}
\begin{example}
	The path in figure~\ref{fig:path-concatenation-singular} is homotopic to the straight line for $z\in\Halfplane^{-}$. Hence our calculation in example~\ref{ex:dilog-past-one} agrees with {\HyperInt}'s result
\begin{MapleInput}
fibrationBasis(polylog(2, 1+z), [z]);
\end{MapleInput}
\begin{MapleMath}
	I\pi\delta_{z} \Hlog{z}{-1}-\Hlog{z}{-1, 0} + \mzv{2}
\end{MapleMath}
\end{example}
\section{Periods}
\label{sec:periods}%
Our algorithms express constants like \eqref{eq:fibration-basis-constants} through iterated integrals $\AnaReg{0}{\infty} \Hyper{w}(z)$ of words $w \in \overline{\Q}^{\times}$ with algebraic letters. These are transformed into special values $\Hyper{u}(1)$ of hyperlogarithms by $u=\code{zeroInfPeriod}(w)$.
{\HyperInt} can read lookup tables to write such periods in terms of a basis over $\Q$.

The file \Filename{periodLookups.m} provides such reductions (taken from the data mine project \cite{BluemleinBroadhurstVermaseren:Datamine}) for multiple zeta values up to weight 12 and alternating Euler sums ($u\in\set{-1,0,1}^{\times}$) up to weight $8$ in the notation
\begin{equation}
	\mzv{n_1,\ldots,n_r}
	\defas
	\Li_{\abs{n_1},\ldots,\abs{n_r}}\left( \frac{n_1}{\abs{n_1}},\ldots \frac{n_r}{\abs{n_r}} \right),
	\label{eq:def:euler-sum} %
\end{equation}
with indices $n_1,\ldots,n_r \in \Z \setminus\set{0}$, $n_r\neq 1$.
When $u \in\set{0,a,2a}^{\times} \cup \set{-a,0,a}^{\times}$, M\"{o}bius transformations are used to express $\Hyper{u}(1)$ in terms of alternating Euler sums and $\log (a)$.

\begin{example}
	{\HyperInt} automatically attempts to load \Filename{periodLookups.m}, but can run without it. With its help,
	\begin{MapleInput}
fibrationBasis(Mpl([3], [1/2]));
	\end{MapleInput}
	\begin{MapleMath}
		\frac{1}{6}\ln(2)^3-\frac{1}{2}\ln(2)\mzv{2}+\frac{7}{8}\mzv{3}
	\end{MapleMath}
	is reduced to MZV and $\ln 2$. But if \Filename{periodLookups.m} is not available, we obtain merely
	\begin{MapleInput}
fibrationBasis(Mpl([3], [1/2]));
	\end{MapleInput}
	\begin{MapleMath}
		-\mzv{-3}-\mzv{2,-1}-\mzv{1,-2} + \frac{1}{6}\ln(2)^3
	\end{MapleMath}
\end{example}

The user can define a different basis reduction or provide bases for periods involving higher weights\footnote{For MZV and alternating sums, \cite{BluemleinBroadhurstVermaseren:Datamine} provides reductions up to weights 22 and 12, respectively.}, or additional letters. These must be defined as a table,
\begin{equation}
	\code{zeroOnePeriods}[u]
	\defas
	\Hyper{u}(1),
\end{equation}
and saved to a file $f$. To read it one must call $\code{loadPeriods}(f)$.
\begin{example}
	Polylogarithms $\Li_{\vec{n}}(\vec{z})$ at fourth roots of unity $\vec{z} \in \set{\pm 1, \pm \imag}^{\abs{n}}$, like
	\begin{MapleInput}
f := Mpl([1,1],[I,-1])+Mpl([1,1],[-1,I]):
fibrationBasis(f);
	\end{MapleInput}
	\begin{MapleMath}
		\Hlog{1}{-I, I}+\Hlog{1}{-1, I}
	\end{MapleMath}
	are initially not known to {\HyperInt}. Up to weight $\abs{n}\leq 2$ they are expressible with $\ln 2$, $\imag$, $\pi$ and Catalan's constant $G \defas \Imaginaerteil\Li_2(\imag)$ as supplied in \Filename{periodLookups4thRoots.mpl}:
	\begin{MapleInput}
loadPeriods("periodLookups4thRoots.mpl"):
fibrationBasis(f);
	\end{MapleInput}
	\begin{MapleMath}
		\frac{1}{8}\mzv{2}+\frac{1}{2}\ln(2)^2-\frac{1}{4}I\pi\ln(2)+I\Catalan
	\end{MapleMath}
\end{example}

\section{Integration of hyperlogarithms}
\label{sec:integration}%
The most important function provided by {\HyperInt} is
\begin{equation}
	\label{eq:integrationStep}%
	\code{integrationStep}(f, z)
	\defas
	\int_0^{\infty} f(z)\ \dd z
\end{equation}
and computes the integral of a polylogarithm $f$, which must be supplied in the form \eqref{eq:maple-internal-reginf}.
First it explicitly rewrites $f(z) \in L(\Sigma)(z)$ following proposition~\ref{prop:reginf-as-hlog} as a hyperlogarithm in $z$. Then a primitive $F=\code{integrate}(f,z)$ is constructed with lemma~\ref{lemma:hlog-primitives} finally expanded at the boundaries $z\rightarrow 0, \infty$.
\begin{example}
	To compute $\int_0^{\infty} \frac{\Li_{1,1}\left(-x/y, -y \right)}{y(1+y)}\ \dd y$, one can enter
	\begin{MapleInput}
convert(Mpl([1,1],[-x/y,-y])/y/(y+1), HlogRegInf):
integrationStep(%, y):
fibrationBasis(%, [x]);
	\end{MapleInput}
	\begin{MapleMath}
		\mzv{2}\Hlog{x}{1}
		 + \Hlog{x}{1,0,1}
		 - \Hlog{x}{0,0,1}
	\end{MapleMath}
\end{example}
A more convenient and flexible interface is provided through the function
\begin{equation}
	\code{hyperInt}\left( f, [z_1=a_1..b_1,\ldots,z_r=a_r..b_r] \right)
	\defas
	\int_{a_r}^{b_r}
	\cdots
	\left[\int_{a_1}^{b_1} f\ \dd z_1 \right]
	\cdots
	\dd z_r
	\label{eq:hyperInt}%
\end{equation}
which computes multi-dimensional integrals by repeated application of \eqref{eq:integrationStep} in the order $z_1,\ldots,z_r$ as specified. It automatically transforms the domains $(a_k,b_k)$ of integration to $(0,\infty)$ and furthermore, $f$ can be given in any form that is understood by $\code{convert}\left( \cdot, \code{HlogRegInf} \right)$. If a variable $z_i$ is specified without a range, $a_i=0$ and $b_i=\infty$ is assumed.
\begin{example}
	\label{ex:moduli-space}%
	A typical integral studied in the origin \cite{Brown:MZVPeriodsModuliSpaces} of the algorithm is the period $I_2$ of $\mathfrak{M}_{0,6}(\R)$ computed in equation~(8.6) therein:
	\begin{MapleInput}
I2 := 1/(1-t1)/(t3-t1)/t2:
hyperInt(I2, [t1=0..t2, t2=0..t3, t3=0..1]):
fibrationBasis(%);
	\end{MapleInput}
	\begin{MapleMath}
		2\mzv{3}
	\end{MapleMath}
\end{example}
\begin{example}
The integrals $E_n$ of the ``Ising-class'' were defined and studied in \cite{BaileyBorweinCrandall:IsingClass}: Let $u_1 \defas 1$, $u_k \defas \prod_{i=2}^k t_i$ for any $k\geq 2$ and set
\begin{equation}
	E_n
	\defas
	2 \int_0^{1} \dd t_2 \ldots \int_0^1 \dd t_n
	\left( 
		\prod_{1 \leq j< k \leq n}
		\frac{u_j - u_k}{u_j+u_k}
	\right)^2.
	\label{eq:def:IsingE}%
\end{equation}
Because the denominators $u_j+u_k = (1+\prod_{i={j}}^{k-1} t_i) \prod_{i=k}^{n} t_i$ have very simple factors, it is easy to prove linear reducibility along the sequence $t_2,\ldots,t_n$ of variables and to show that all $E_n$ are rational linear combinations of alternating Euler sums.

We included a simple procedure \mbox{\code{IsingE}$\left( n \right)$} to evaluate them in the attached manual. In particular we can confirm the conjecture on $E_5$ made in \cite{BaileyBorweinCrandall:IsingClass}:
\begin{MapleInput}
IsingE(5);
\end{MapleInput}
\begin{MapleMath}
2\mzv{3}
\left(
 - 37
 + 232\ln (2) 
\right)
 - 4\mzv{2}
\left(
31
 - 20\ln  ( 2 ) 
 + 64 \ln^2  ( 2 )
\right)
\\
 - \frac{318}{5}\mzv[2]{2}
 + 42
 - 992\mzv{1,-3}
 - 40 \ln ( 2 )
 + 464 \ln^2 ( 2 )
 + \frac{512}{3} \ln^4 ( 2 )
\end{MapleMath}
Further exact results for $E_n$ up to $n=8$ are tabulated in appendix~\ref{sec:IsingE}. The time- and memory-requirements of their computations are summarized in table~\ref{tab:IsingE-resources}.
\end{example}
\begin{table*}
	\centering
	\begin{tabular}{rcccccccc}
		\toprule
		$n$ & 1 & 2 & 3 & 4 & 5 & 6 & 7 & 8 \\
		\midrule
		time & \SI{10}{\milli\second} & \SI{41}{\milli\second} & \SI{52}{\milli\second} & \SI{235}{\milli\second} & \SI{2.0}{\second} & \SI{40.6}{\second} & \SI{29.3}{\minute} & \SI{28}{\hour}\\
		RAM & \SI{35}{\mebi\byte} & \SI{51}{\mebi\byte} & \SI{51}{\mebi\byte} & \SI{76}{\mebi\byte} & \SI{359}{\mebi\byte} & \SI{1.6}{\gibi\byte} & \SI{1.9}{\gibi\byte} & \SI{30}{\gibi\byte}	\\
		\bottomrule
	\end{tabular}
	\caption[Time and memory requirements for the computation of Ising integrals $E_n$]{Resources consumed during computation of the Ising-type integrals $E_n$ of \eqref{eq:def:IsingE} running on Intel\textsuperscript{\textregistered} Core{\texttrademark} i7-3770 CPU @ \SI{3.40}{\giga\hertz}. The column with $n=1$ (when $E_n \defas 1$) requires no actual computation and shows the time and memory needed to load \Filename{periodLookups.m}.}%
	\label{tab:IsingE-resources}%
\end{table*}

\subsection{Singularities in the domain of integration}
\label{sec:integration-contour-deformation}%
The integration \eqref{eq:integrationStep} requires that $f(z) \in \HlogAlgebra(\Sigma)(z)$ is a hyperlogarithm without any letters $\Sigma_+ \defas \Sigma \cap (0,\infty) = \emptyset$ inside the domain of integration, which ensures that $f(z)$ is analytic on $(0,\infty)$.

Otherwise $f(z)$ can have poles or branch points on $\Sigma_+$ and the integration is then performed along a deformed contour $\gamma$ (like figure~\ref{fig:contour-deformation}) as in \eqref{eq:reglim-reginf-contour-deformation} by a splitting analogous to \eqref{eq:reglim-reginf-splitted}.\footnote{The only difference is that the limit $\AnaReg{t}{0}$ does not need to be taken.} The dependence on $\gamma$ is encoded in the variables
\begin{equation}
	\label{eq:contour-deformation-deltas}%
	\delta_{z,\sigma}
	=
	\begin{cases}
		+1 & \text{when $\gamma$ passes below $\sigma$},\\
		-1 & \text{when $\gamma$ passes above $\sigma$}.\\
	\end{cases}
\end{equation}
\begin{example}
	\label{ex:contour-deformation-integral}%
	The integrand $f(z)=\frac{1}{1-z^2}$ has a pole at $z\rightarrow 1$ and is not integrable over $(0,\infty)$. Instead, {\HyperInt} computes the finite contour integrals that avoid $1$:
	\begin{MapleInput}
hyperInt(1/(1-z^2), z): fibrationBasis(%);
	\end{MapleInput}
	\begin{MapleOutput}
Warning, Contour was deformed to avoid potential singularities at {1}.
	\end{MapleOutput}
	\begin{MapleMath}
		-\frac{1}{2} \cdot I \pi \delta_{z, 1}
	\end{MapleMath}
\end{example}
\begin{remark}
	Even if positive letters $\Sigma_{+}$ occur, $f(z)$ can be analytic on $(0,\infty)$ nonetheless. In this case the dependence on any $\delta_{z,\sigma}$ drops out in the result by lemma~\ref{lemma:singular-expansion-analytic}.
\end{remark}
\begin{example}
	The integrand $f(z)=\frac{\ln(z)}{1-z^2}$ is analytic at $z\rightarrow 1$ and absolutely integrable over $(0,\infty)$. Its integral is correctly computed by {\HyperInt}:
	\begin{MapleInput}
hyperInt(ln(z)/(1-z^2), z):
fibrationBasis(%);
	\end{MapleInput}
	\begin{MapleOutput}
Warning, Contour was deformed to avoid potential singularities at {1}.
	\end{MapleOutput}
	\begin{MapleMath}
		-\frac{3}{2} \mzv{2}
	\end{MapleMath}
\end{example}

\subsection{Detection of divergences}
\label{sec:divergence-detection}%
By default, the option $\code{\_hyper\_check\_divergences}=\code{true}$ is activated and triggers, after each integration $\int_0^{\infty} f(z)\ \dd z$, a test of convergence. The expansion \eqref{eq:log-laurent-expansion} of the primitive $F(z)$ of the integrand $f(z)=\partial_z F(z)$ is explicitly computed as
\begin{equation}
	\label{eq:check-divergences}%
	F(z)
	= \sum_{i=0}^{N} \log^i(z) \sum_{j=M}^{-\infty} z^{-j} F_{i,j}
	\quad\text{at}\quad
	z \rightarrow 0
\end{equation}
and all polylogarithms $F_{i,j}$ with $i>0$ or $j>0$ are explicitly checked to vanish $F_{i,j}=0$ using \code{fibrationBasis}; the limit $z \rightarrow \infty$ is treated analogously. This method is time-consuming and we recommend to deactivate it for any involved calculations, presuming that convergence is granted by the problem at hand (for instance by corollary~\ref{corollary:projective-convergence-euclidean}).
\begin{example}
	\label{ex:divergence-check}%
	An endpoint divergence at $z\rightarrow \infty$ is detected for $\int_0^{\infty} \frac{\ln z}{1+z}\dd z = \lim\limits_{z\rightarrow\infty} \Hyper{\letter{-1}\letter{0}}(z)$:
	\begin{MapleInput}
hyperInt(ln(z)/(1+z), z);
	\end{MapleInput}
	\begin{MapleError}
Error, (in integrationStep) Divergence at z = infinity of type ln(z)^2
	\end{MapleError}
\end{example}
The expansions \eqref{eq:check-divergences} are only performed up to $i,j\leq\code{\_hyper\_max\_pole\_order}$ (default value is $10$). If higher order expansions are needed, an error is reported and this variable must be increased.

Note that the expansion \eqref{eq:check-divergences} is only computed at the endpoints $z\rightarrow 0,\infty$. Polar singularities inside $(0,\infty)$ are not detected, e.g.\ $\code{hyperInt}\left(\frac{1}{(1-z)^2}, z\right) = \restrict{\frac{1}{1-z}}{0}^{\infty} = 1$ calculates the integral along a contour evading $z=1$ just as discussed in section~\ref{sec:integration-contour-deformation}.
One can split the integration
\begin{equation}
	\label{eq:integral-splitted}%
	\int_0^{\infty} f(z) \ \dd z
	=
	\sum_{i=0}^{k} \int_{\tau_i}^{\tau_{i+1}} f(z)\ \dd z
\end{equation}
at such critical points $\Sigma_+ = \set{\tau_1<\ldots<\tau_k}$ with $\tau_0 \defas 0$, $\tau_{k+1} \defas \infty$ with the effect that all singularities now lie at endpoints and will be properly analyzed by the program.

A problem arises if calculations involve periods for which no basis reduction is known to {\HyperInt}, because the vanishing $F_{i,j}=0$ of a potential divergence might not be detected. One can then set $\code{\_hyper\_abort\_on\_divergence} \defas \code{false}$ to continue with the integration. All $F_{i,j}\neq 0$ of \eqref{eq:check-divergences} are stored in the table \code{\_hyper\_divergences}.
\begin{example}
	\label{eq:divergence-period-relations}%
	When \Filename{periodLookups.m} is not loaded, the convergent integral
	\begin{MapleInput}
hyperInt(polylog(2,-1/z)*polylog(2,-z)/z,z);
	\end{MapleInput}
	\begin{MapleError}
Error, (in integrationStep) Divergence at z = infinity of type ln(z)
	\end{MapleError}
	is inadvertently classified as divergent. Namely, $F_{1,0}$ of \eqref{eq:check-divergences} is
	\begin{MapleInput}
entries(_hyper_divergences, pairs);
	\end{MapleInput}
	\begin{MapleMath}
		\left(z=\infty, \ln \left( z \right) \right)
		=
		4\mzv{1,3} + 2 \mzv{2,2} - \frac{1}{36} \pi^{4}
	\end{MapleMath}
	and its vanishing corresponds to an identity of MZV.
\end{example}
\begin{remark}
	In this way, the mere quality of convergence of an integral implies non-trivial rational relations between periods.
\end{remark}

\section{Factorization of polynomials}
\label{sec:factorization}%
Since we are working with hyperlogarithms throughout, it is crucial that all polynomials occurring in the calculation factor linearly with respect to the integration variable $z$. For example,
\begin{MapleInput}
integrationStep([[1/(1+z^2), []]], z);
\end{MapleInput}
\begin{MapleError}
Error, (in partialFractions) 1+z^2 is not linear in z
\end{MapleError}
fails because factorization is initially only attempted over the rationals $\K=\Q$. Instead we can allow for an algebraic extension $\K=\Q(R)$ by specification of a set $R=\code{\_hyper\_splitting\_field}$ of radicals:
\begin{MapleInput}
_hyper_splitting_field := {I}:
integrationStep([[1/(1+z^2), []]], z);
fibrationBasis(%);
\end{MapleInput}
\begin{MapleMath}
	\left[
		\left[
			\frac{1}{2}I, [[-I]]
		\right],
		\left[
			-\frac{1}{2}I, [[I]]
		\right]
	\right]
	\\
	\frac{1}{2} \pi
\end{MapleMath}
We can also go further and factorize over the full algebraic closure $\K=\overline{\Q(\vec{z})}$ by setting $\code{\_hyper\_algebraic\_roots}\defas\code{true}$. Over $\K$, all rational functions $\Q(\vec{z})$ factor linearly such that we can integrate any $f \in \AnaReg{t}{\infty} \HlogAlgebra(\Sigma)(t)$ as long as we start with rational letters $\Sigma\subset \Q(\vec{z})$.

This feature is to be considered experimental and only applied in \code{transformWord} which implements proposition~\ref{prop:reginf-as-hlog}: Given an irreducible polynomial $P \in \Q[\vec{z}]$ and a distinguished variable $z$, the symbolic notation
\begin{equation}
	\label{eq:root-letter}%
	\letter{\code{Root}(P,z)}
	\defas
	\sum
	\setexp{\letter{z_0}}{\restrict{P}{z=z_0} = 0}
\end{equation}
sums the letters corresponding to all the roots of $P$.
\begin{example}
	\label{ex:algebraic-letters}%
	A typical situation looks like this:
	\begin{MapleInput}
f,g:=Hlog(x,[-z,x+x^2]),Hlog(x,[x+x^2,-z]):
fibrationBasis(f+g, [x, z]);
	\end{MapleInput}
	\begin{MapleError}
Error, (in linearFactors) z+x+x^2 does not factor linearly in x
	\end{MapleError}
	To express $f+g$ as a hyperlogarithm in $x$, the roots $R=\code{Root}(P,x)=\set{-\frac{1\pm\sqrt{1-4z}}{2}}$ of $P=z+x+x^2$ seem necessary. After allowing for such algebraic letters, we obtain:
	\begin{MapleInput}
_hyper_algebraic_roots := true:
fibrationBasis(f+g, [x, z]);
	\end{MapleInput}
	\begin{MapleMath}
		- \Hlog{x}{-1,-z}
		 - \Hlog{x}{-z,-1}
		 \\
		 + \Hlog{x}{-z,0}
		 + \Hlog{x}{0,-z}
	\end{MapleMath}
	Since this result actually does not involve $\letter{R}$ at all one might wonder why it was necessary in the first place. The reason is that the individual contributions $f$ and $g$ indeed need $\letter{R}$. Only in their sum this letter drops out:\footnote{%
	In this extremely simple example this is clear since by lemma~\ref{lemma:Hlog-character},
	$
		f+g
		=
		\Hyper{\letter{-z}}(x)
		\cdot
		\Hyper{\letter{x(x+1)}}(x)
	$
	factorizes into $\log \frac{x+z}{z} \cdot \log \frac{x}{1+x}$. We thus see why our representation \eqref{eq:maple-internal-reginf} is preferable to one where all products of words are multiplied out (as shuffles).%
}%
	\begin{MapleInput}
alias(R = Root(z+x+x^2, x)):
fibrationBasis(f, [x, z]);
	\end{MapleInput}
	\begin{MapleMath}
		\Hlog{x}{R,-z}
		 + \Hlog{x}{R,-1}
		 - \Hlog{x}{R,0}
		\\
		 + \Hlog{x}{-z,0}
		 - \Hlog{x}{-z,-1}
		\\
		 - \Hlog{x}{-1,-z}
		 - \Hlog{x}{0,-z}
	\end{MapleMath}
\end{example}
Further processing of functions with such non-rational letters\footnote{%
These are sometimes referred to as generalized harmonic polylogarithms with \emph{nonlinear weights}.%
} is not supported by {\HyperInt} as their integrals are in general not hyperlogarithms anymore. The cancellation of example~\ref{ex:algebraic-letters} occurs often, so an option \code{\_hyper\_ignore\_nonlinear\_polynomials} (default value is \code{false}) is available to ignore all algebraic letters in the first place. That is, all words containing such a letter are immediately dropped when it is set to \code{true}.

In the example above this gives the correct result for $f+g$, but will provoke false answers when \code{fibrationBasis} is applied to $f$ or $g$ alone. Hence this option should only be used when linear reducibility is granted.

\section{Performance}
\label{sec:performance}%
During programming we focussed on correctness and we are aware of considerable room for improvement of the efficiency of {\HyperInt}.
But we hope that our code and the details provided in section~\ref{sec:hlog-algorithms} will inspire further, streamlined implementations, even outside the regime of computer algebra systems.
This is possible since apart from the factorization of polynomials (which is anyway computed during the linear reduction, prior to the actual integration), all operations boil down to elementary manipulations of words (lists) and computations with rational functions.

Ironically, often just decomposing into partial fractions becomes a severe bottleneck in practice, as was also noted in \cite{AblingerBluemleinRaabSchneiderWissbrock:Hyperlogarithms}. This happens when an integrand contains denominator factors to high powers or very large polynomials in the numerator.

We observed that {\Maple} consumes a lot of main memory, in very challenging calculations the demand grew beyond $\SI{100}{\gibi\byte}$. Often this turns out to be the main limitation in practice.

Our program uses some functions that are not thread-safe and can therefore not be parallelized automatically. But by linearity, a manual parallelization is straightforward: Multiple instances of {\Maple} can run $\code{hyperInt}\left( f_i,z \right)$ on disjoint portions of the integrand $f = \sum_i f_i$ and the results added afterwards (see also the manual).

\section{Application to Feynman integrals}
\label{sec:feynman-integrals}%
In chapter~\ref{chap:hyperlogs} we investigated hyperlogarithms on their own, but the algorithms were originally developed for the computation of Feynman integrals \cite{Brown:TwoPoint}. Essential results on their linear reducibility (including counterexamples) and the geometry of Feynman graph hypersurfaces were obtained in \cite{Brown:PeriodsFeynmanIntegrals} and extended in section~\ref{sec:multiple-integrals}.
Some further discussions on multi-scale and subdivergent integrals in the parametric representation are also given in \cite{BognerLueders:MasslessOnShell,BrownKreimer:AnglesScales,Kreimer:WheelsInWheels}.

In \cite{Panzer:MasslessPropagators,Panzer:DivergencesManyScales} we successfully applied our implementation to compute many non-trivial examples, including some massless propagators with up to six loops and also divergent integrals depending on up to seven kinematic invariants. All results\footnote{These can be downloaded from \url{http://www.math.hu-berlin.de/~panzer/}.} presented in these papers were computed using this program {\HyperInt}.

\subsection{\texorpdfstring{$\varepsilon$}{epsilon}-expansion}
Consider analytic regularization in the dimension $\dimension=4-2\varepsilon$ and the indices $\EP_e = \EPZ_e + \varepsilon \EPE_e$ near integers $\EPZ_e \in \Z$. By corollary~\ref{cor:finite-anareg-as-Feynman} we may assume that the projective integral \eqref{eq:feynman-integral-projective} converges absolutely for $\varepsilon=0$. Thus we can expand the integrand in $\varepsilon$ and obtain each coefficient $c_n$ of the Laurent series $\Phi(G) = \sum_{n} c_n \varepsilon^n$ as the period integral
\begin{equation}
	c_n
	= \Gamma(\sdd)
	\left[
	\prod_{e\in E}
	\int_0^{\infty}
	\frac{\dd \SP_e}{\Gamma(\EP_e)}
	\right]
	\frac{f^{(n)}}{Q^{(n)}}
	\delta(1-\SP_{e_N})
	\label{eq:expansion-coefficients}%
\end{equation}
where $Q^{(n)} \in \Q[\psipol^{-1},\phipol^{-1}]$ denotes a polynomial and $f^{(n)} \in \Q[\vec{\SP}, \log\vec{\SP},\log\phipol,\log\psipol]$. Whenever the Symanzik polynomials $S=\set{\psipol,\phipol}$ are linearly reducible, this integral can be computed with {\HyperInt}.

\subsection{Additional functions in {\HyperInt}}
Appendix~\ref{sec:function-list-feynman} lists auxiliary functions that support the calculation of Feynman integrals. These entail simple routines to construct the graph polynomials $\psipol$ and $\phipol$.

For divergent integrals, the parametric integrands \eqref{eq:feynman-integral-projective} cannot be integrated directly. The manual shows how {\HyperInt} can implement the analytic regularization through partial integrations \eqref{eq:def-anapartial} to construct a convergent integral representation.
\begin{figure}
	\centering
	\includegraphics[width=0.5\columnwidth]{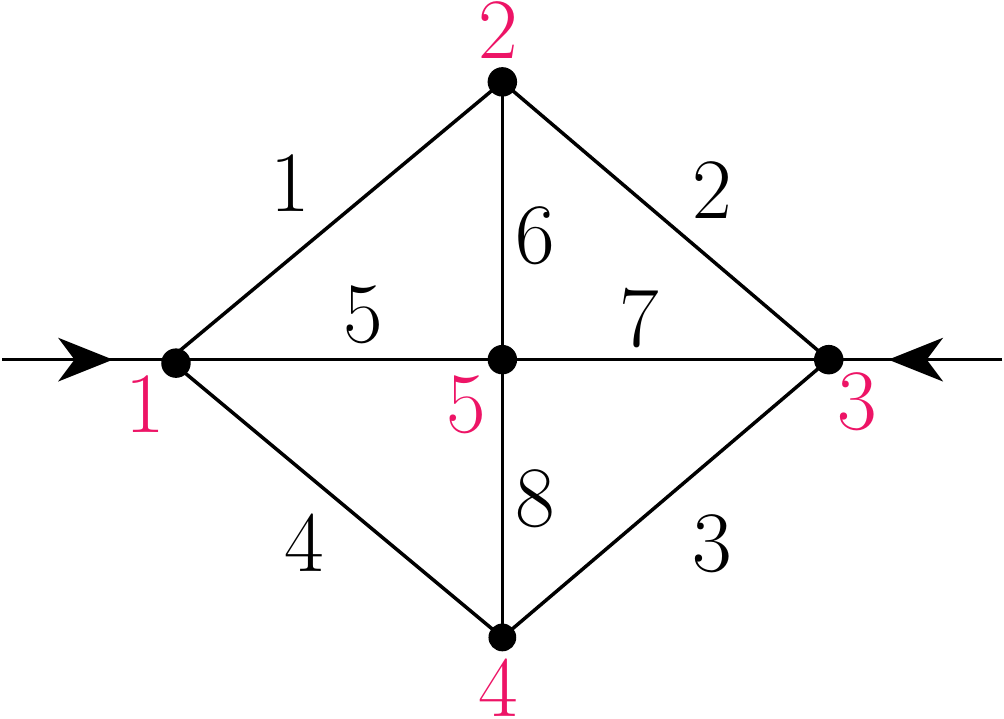}
	\caption[Four-loop massless propagator computed in section~\ref{sec:example-propagator}]{Four-loop massless propagator computed in section~\ref{sec:example-propagator}. This one is called $M_{3,6}$ in \cite{Panzer:MasslessPropagators}. Edges are labelled in black, vertices in red.}%
	\label{fig:propagator-graph}%
\end{figure}
\subsection{Examples}
\label{sec:example-propagator}%
Plenty of examples are provided in the {\Maple} worksheet \Filename{Manual.mw}, including graphs with subdivergences, massive propagators and more than two external momenta. Here we content ourselves with a simple four-loop massless propagator.

First we define the graph of figure~\ref{fig:propagator-graph} by its edges $E$ and specify two external momenta of magnitude one entering the graph at the vertices $1$ and $3$. The Symanzik polynomials $\psipol$ and $\phipol$ can be computed with
\begin{MapleInput}
E:=[[1,2],[2,3],[3,4],[4,1],[5,1],[5,2], [5,3],[5,4]]:
psi:=graphPolynomial(E):
phi:=secondPolynomial(E, [[1,1], [3,1]]):
\end{MapleInput}
This graph has vertex-width three and is therefore linearly reducible \cite{Brown:PeriodsFeynmanIntegrals}. Still let us calculate a polynomial reduction to verify this claim:
\begin{MapleInput}
L:=table(): S:=irreducibles({phi,psi}):
L[{}]:=[S, {S}]: cgReduction(L):
\end{MapleInput}
The function $\code{irreducibles}\left( \cdot \right)$ breaks up all polynomials into their irreducible factors and could be omitted here. The third instruction initializes $L_{\emptyset}$ with the complete graph on the two vertices $\set{\psi,\phi}$. We can check the linearity along some order $\vec{z}$:
\begin{MapleInput}
z:=[x[1],x[2],x[6],x[5],x[3],x[4],x[7],x[8]]:
checkIntegrationOrder(L, z[1..7]):
\end{MapleInput}
\begin{MapleOutput}
1. (x[1]): 2 polynomials, 2 dependent
2. (x[2]): 5 polynomials, 4 dependent
3. (x[6]): 8 polynomials, 4 dependent
4. (x[5]): 7 polynomials, 4 dependent
5. (x[3]): 6 polynomials, 6 dependent
6. (x[4]): 4 polynomials, 3 dependent
7. (x[7]): 1 polynomials, 1 dependent
Final polynomials:
\end{MapleOutput}
\begin{MapleMath}
	\set{}
\end{MapleMath}
The integrand is assembled according to \eqref{eq:projective-delta-form-integrand} which in this case is already convergent as-is. We expand to second order in $\varepsilon$ with
\begin{MapleInput}
sdd := nops(E)-(1/2)*4*(4-2*epsilon):
f := series(psi^(-2+epsilon+sdd)*phi^(-sdd), epsilon=0):
f := add(coeff(f,epsilon,n)*epsilon^n,n=0..2):
\end{MapleInput}
Now we integrate out all but the last Schwinger parameter
\begin{MapleInput}
hyperInt(f, z[1..-2]):
\end{MapleInput}
and reduce the result into a basis of MZV:
\begin{MapleInput}
fibrationBasis(eval(f, z[-1]=1)):
collect(%, epsilon);
\end{MapleInput}
\begin{MapleMath}
	\left(
		254\mzv{7}
		+780 \mzv{5}
		-200 \mzv{2}\mzv{5}
		-196 \mzv[2]{3}
		+80 \mzv[3]{2}
		-\frac{168}{5}\mzv[2]{2}\mzv{3}
	\right)
	\varepsilon^2
	\\
+ \left(
		-28 \mzv[2]{3}
		+140 \mzv{5}
		+\frac{80}{7} \mzv[3]{2}
	\right)
	\varepsilon
+20 \mzv{5}.
\end{MapleMath}

\section{Tests of the implementation}
\label{sec:hyperint-tests}%
As with any computer program, exhaustive testing is of supreme importance and we spent a considerable amount of time on this duty.
The file \Filename{HyperTests.mpl} contains a small subset of such test runs. These feature diverse applications of {\HyperInt} and therefore also supplement the manual.
The majority of our tests either verify
\begin{enumerate}
	\item functional equations between polylogarithms: If two expressions $A$ and $B$ represent the same function of $\vec{z}=[z_1,\ldots,z_r]$, then $\code{fibrationBasis}\left(A-B,\vec{z}\right)$ must evaluate to zero;\footnote{Care is required because of multi-valuedness. The order $\vec{z}$ must be chosen such that both expressions $A$ and $B$ are well-defined power series in the limit $ 0 \ll z_1 \ll \cdots \ll z_r \ll 1$. Otherwise one might end up comparing different analytic continuations to this region.}
	\item sequences of parametric integrals with known analytic results,
	\item analytic results for integrals obtained with other methods or
	\item numeric approximations of integrals.
\end{enumerate}
Probably the strongest tests of our implementation are the computations of $\varepsilon$-expansions of various single-scale \cite{Panzer:MasslessPropagators} and multi-scale \cite{Panzer:DivergencesManyScales} Feynman integrals, including all $3$- and some $4$-loop massless propagators (see also section~\ref{sec:propagators}).
We cross-checked those results that were known before with the available references, verified that symmetries of the Feynman graphs are reflected in the obtained $\varepsilon$-expansions and in some cases used established programs \cite{BognerWeinzierl:ResolutionOfSingularities,SmirnovTentyukov:FIESTA} to obtain numeric evaluations that confirmed our analytic formulas.

Further checks of {\HyperInt} include:
\begin{itemize}
	\item Plenty of functional and integral equations of polylogarithms given in the excellent books \cite{Lewin:PolylogarithmsAssociatedFunctions,Lewin:StructuralPropertiesPolylogarithms}.\footnote{In doing so we revealed a very few misprints, listed in appendix~\ref{chap:LewinTypos}.}

	\item Examples of transformations of polylogarithms into hyperlogarithms given in \cite{BroedelSchlottererStieberger:PolylogsMZVSuperstringAmplitudes}.

	\item The expansion (at $x,y\rightarrow 0$) of Euler's beta function $B(1-x,1-y)$ in the form\footnote{For this test we expand the integrand in $x$ and $y$, integrate each coefficient (of monomials up to total degree $12$) with {\HyperInt} and compare the result with the straightforward expansion of the left-hand side.}
\begin{equation*}
	\frac{\Gamma(1-x)\Gamma(1-y)}{\Gamma(2-x-y)}
	=
	\frac{\exp\left[ \sum\limits_{n=2}^{\infty} \frac{\zeta(n)}{n}( x^n + y^n-(x+y)^n)  \right]}{1-x-y}
	=
	\int_0^{\infty} \frac{z^{-x}\ \dd z}{(1+z)^{2-x-y}}.
\end{equation*}

	\item The identity of lemma~\ref{lemma:bubble-integral-helper-function} ($n\in\N$):
\begin{equation*}
	\int_{0}^{\infty} \left[ 
		\left( \frac{1}{x} - \frac{1}{x+z} \right) \Li_n(-x-z)
		-\frac{1}{x} \Li_n\left( -\frac{z}{x+1} \right)
	\right] \dd x
	=
	n \Li_{n+1}(-z).
\end{equation*}

\item
The periods of the bubble chain graphs in figure~\ref{fig:bubble-graphs-def} are 
$\period\left( \BBr{n}{m} \right) = (n+m)!$
and computed explicitly in lemma~\ref{lemma:bubble-graphs-transcendental-period} for $\period\left( \BBt{n}{m} \right)$. These results were reproduced by applying {\HyperInt} to the integral representations \eqref{eq:bubble-graph-rational-parametric} and \eqref{eq:BBt-projective}, for $n+m \leq 6$. 

\item Examples of period integrals on the moduli space $\mathfrak{M}_{0,n}(\R)$ given in \cite{Brown:MZVPeriodsModuliSpaces}.

\item
	Our results for the Ising integrals $E_n$ of \eqref{eq:def:IsingE} match the analytic results up to $n=4$ given in \cite{BaileyBorweinCrandall:IsingClass} and the numeric values obtained therein for $E_5$ and $E_6$ agree with our exact results collected in appendix~\ref{sec:IsingE}.

\item
We confirmed the matrix elements $\hat{I}_{1a}$, $\hat{I}_{1b}$, $\hat{I}_{2a}$, $\hat{I}_{2b}$, $\hat{I}_{4}$ of massive ladder graphs with operator insertions computed in \cite{Wissbrock:Massive3loopLadder} and checked the Benz graphs $I_1$, $I_2$ and $I_3$ of \cite{AblingerBluemleinRaabSchneiderWissbrock:Hyperlogarithms}. The integrals $\hat{I}_4$ and $I_1$ are part of our manual, where we correct misprints in the corresponding equations (3.18) and (3.1) of these articles.

\item
The massless hexagon integral \cite{DelDucaDuhrSmirnov:OneLoopMasslessHexagon,DixonDrummondHenn:OneLoopMasslessHexagon} as included in the manual.

\item We computed the exact on-shell results for ladder boxes $B_n$ up to $n=6$ in $\dimension=6$ dimensions. Their limits \eqref{eq:lbox2-zero} through \eqref{eq:lbox6-zero} when $t\rightarrow 0$ agree with the analytic ($n\in\set{2,3}$) and numeric ($n\in\set{4,5,6}$) results published in \cite{Kazakov:MultiBoxSix}.

\item All vacuum periods of $\phi^4$-theory that we computed with {\HyperInt} match the numeric values given in \cite{BroadhurstKreimer:KnotsNumbers,Schnetz:Census}, see also section~\ref{sec:ex-periods}.
\end{itemize}

\chapter{Applications and examples}
\label{chap:examples}%
The method of hyperlogarithms was first formulated for the computation of massless vacuum graphs and propagators. We give a short summary of the tremendous progress in this field that was recently possible due to our program and work of Oliver Schnetz. This includes the exact computation of all primitive $\phi^4$-periods up to seven loops and in particular we present the first period in massless $\phi^4$ theory which is not a multiple zeta value, but a polylogarithm at a primitive sixth root of unity.
We will also summarize our work \cite{Panzer:MasslessPropagators}, in which we computed $4$-loop massless propagators with arbitrarily $\epsilon$-dependent indices $\EP_e$ (before, this was only possible for at most two loops \cite{BierenbaumWeinzierl:TwoPoint}).

In section~\ref{sec:ex-renormalized-parametric} we argue that (at least in the absence of infrared divergences) hyperlogarithms are suitable for the exact calculation of renormalized Feynman integrals.
While ``renormalization under the integral sign'' (which does not require any regularization) has always been an important theoretical tool, its practical use in computations of individual Feynman integrals so far mostly restricted to numerical integration.
We advertise the approach of exact integration with several examples, including our new results \eqref{eq:w3-in-w3} and \eqref{eq:ws3-in-ws4}.

Phenomenology of particle physics entails much more than pure numbers and Feynman integrals with complicated dependence on kinematics must be computed. In particular, there have been many computations of massless $3$- and $4$-point integrals in the past. We shall recall our main results on massless $3$-constructible (section \ref{sec:ex-3pt}) and ladder-box graphs (section~\ref{sec:ex-ladderboxes}) and extend them to even broader families of graphs. Note that:
\begin{enumerate}
	\item The indices $\EP_e$ are not restricted to integers, but can be expanded in $\epsilon$ around integers. Propagator insertions therefore factorize into simpler graphs.

	\item Our method (based on the forest functions of sections~\ref{sec:vw3} and \ref{sec:ladderboxes}) applies to infinitely many Feynman graphs and not only in fully massless kinematics, but it allows for up to 2 massive external legs.
\end{enumerate}
So for this particular class of integrals, the parametric integration with hyperlogarithms exceeds the possibilities of other methods that are currently applied (complete results have so far only been published up to 3 loops). As an explicit example, we compute the 6-loop ladder box in six dimensions in \eqref{eq:lbox6-zero}.

Note that we found many linearly reducible Feynman integrals with internal masses. However, we have not yet a good understanding of this class and will not attempt a premature discussion here. The interested reader may refer to \cite{Panzer:DivergencesManyScales} and appendix~\ref{sec:massive-integrals}, where we show analytic results for a six-scale massive two-loop diagram.

\begin{figure}\centering
	$
	P_{6,2} = \Graph[0.3]{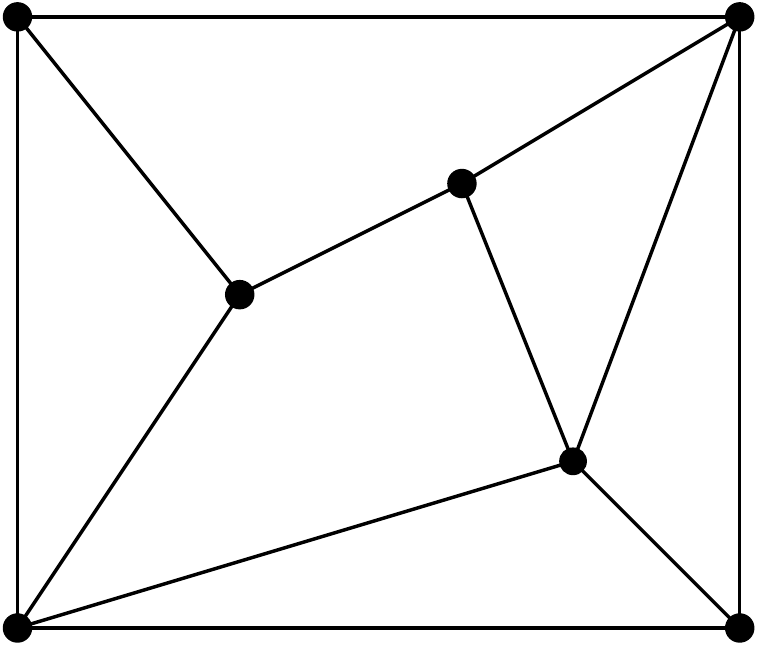}
		\quad
		P_{6,3} = \Graph[0.35]{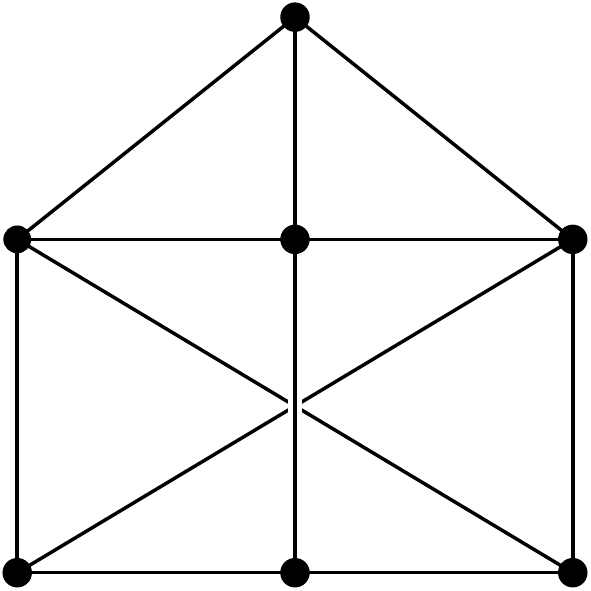}
		\quad
		P_{7,2} = \Graph[0.3]{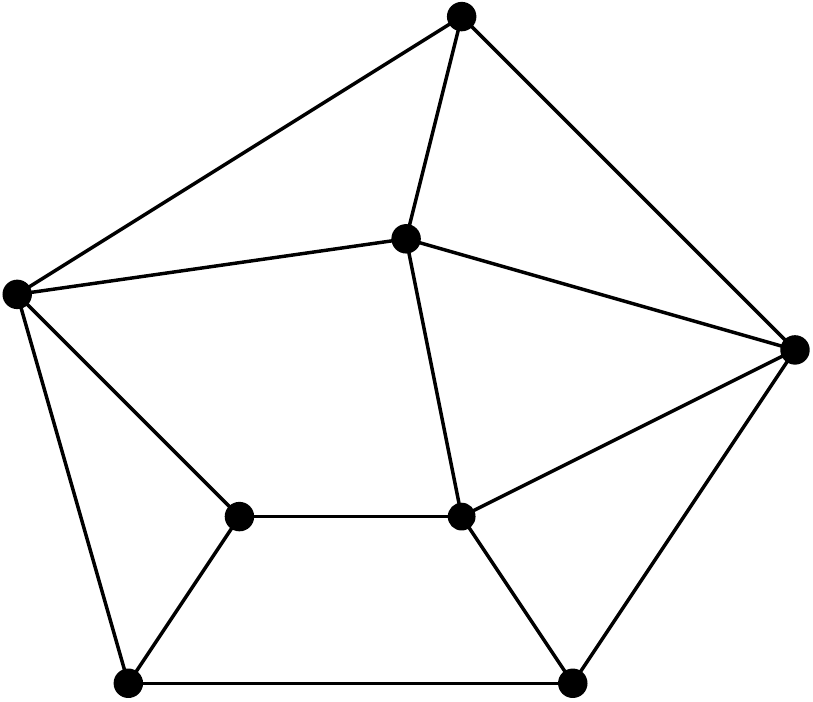}
		\quad
		P_{7,5} = \Graph[0.25]{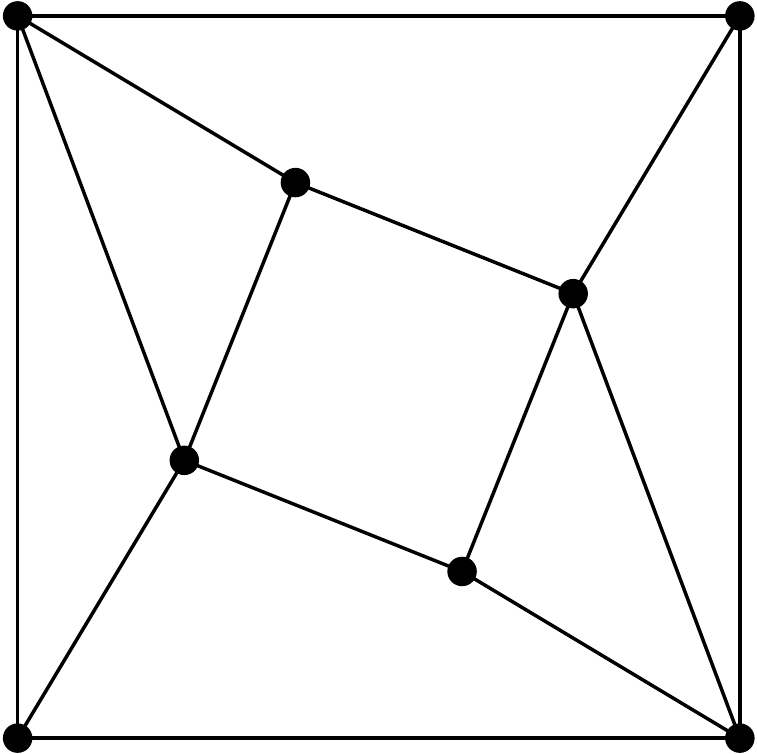}
	$
	\caption{%
		Some six and seven loop primitive $\phi^4$ vertex graphs from the census \cite{Schnetz:Census}.}%
	\label{fig:periods-examples}%
\end{figure}
\section{Periods of massless \texorpdfstring{$\fieldphi^4$}{phi\textasciicircum4} theory}
\label{sec:ex-periods}%
The structurally simplest Feynman integrals are associated to logarithmically divergent graphs $G$ without subdivergences. Their associated period $\period(G) = \int \psipol^{-\dimension/2}\ \Omega$ from \eqref{eq:period-primitive-projective} is a pure number, independent of any kinematics and completely describes the scaling behaviour of the actual Feynman integral (which \emph{does} depend on the kinematics).

A physically interesting case is $\phi^4$ theory ($4$-regular graphs) in $\dimension=4$ dimensions, where each vertex graph (that means $4$ external momenta) is logarithmically divergent.\footnote{Since the external momenta are irrelevant for the period, we do not draw them and find four three-valent vertices in the vertex graphs.} Because these periods determine the contribution of a graph to the $\beta$-function (renormalization of the coupling constant), their evaluation is of high interest.

The contributions of primitive graphs are distinguished in that they are
\begin{itemize}
	\item independent of the renormalization scheme and
	\item conjectured to contribute, at each order in the perturbative expansion, the periods with highest weight.
\end{itemize}
Various symmetries relate the periods of different graphs such that, despite the huge number of graphs at high loop orders, the periods of primitive graphs can be reduced to a relatively small set. There remain nine periods at seven loops \cite{BroadhurstKreimer:KnotsNumbers,Schnetz:Census} (not counting products), six of them were identified as multiple zeta values through high-precision numeric evaluations. The examples in figure~\ref{fig:periods-examples} (computed analytically in \cite{Panzer:MasslessPropagators}) are:
\begin{equation}
	\begin{alignedat}{2}
	\period(P_{6,2})
	&= 8\mzv[3]{3} + \tfrac{1063}{9} \mzv{9}
	&
	\period(P_{6,3})
	&= 252 \mzv{3} \mzv{5}
		+ \tfrac{432}{5} \mzv{3,5}
		- \tfrac{\numprint{25056}}{875} \mzv[4]{2} 
	\\
	\period(P_{7,5})
	&= 450 \mzv[2]{5} -189 \mzv{3} \mzv{7}
	\quad&
	\period(P_{7,2})
	&= \tfrac{\numprint{62957}}{192} \mzv{11}
		-9 \left( 
				\mzv{3,5,3}
				- \mzv{3}\mzv{3,5}
			\right)
		+ 35 \mzv[2]{3} \mzv{5}
	\end{alignedat}
	\label{eq:periods-examples}%
\end{equation}
The census \cite{Schnetz:Census} extended the enumeration of primitives to eight loops and supplied many more evaluations, all in terms of multiple zeta values (further complemented by even higher weight results \cite{BroadhurstKreimer:MZVPositiveKnots}). These works raised many questions on this apparent interaction of quantum field theory and number theory. The two we want to answer here are:
\begin{enumerate}
	\item Are these findings correct? Being based on numeric evaluations and integer relation detectors like PSLQ \cite{FergusonBaileyArno:PSLQ}, though extremely likely, they are not proven. Though instead of questioning their correctness, assuming their validity leaves us with the task to compute them rigorously. How far can we get with hyperlogarithms?

	\item Three periods at seven loops could not be identified originally. What are they? In particular, do multiple zeta values suffice for massless $\phi^4$ theory at seven loops?
\end{enumerate}

\subsection{Linear reducibility to seven loops}
A detailed analysis of linear reducibility for vacuum graphs was carried out in \cite{Brown:PeriodsFeynmanIntegrals}. In particular it was shown that the first counterexample occurs at six loops as the complete bipartite graph $K_{3,4}$. However, its non-reducibility can be circumvented by a subdivision of the integrand as explained in \cite{Brown:PeriodsFeynmanIntegrals}. We did not carry this out though, because this graph is particularly simple to evaluate with graphical functions by Oliver Schnetz.\footnote{An earlier computation of the period of $K_{3,4}$ was suggested in \cite{Schnetz:K34}.}

We computed the polynomial reduction of each graph at seven loops with our program and confirmed the expectation that all of them are linearly reducible, with the single very notable exception of $P_{7,11}$ which we discuss in detail in section~\ref{sec:P711}.
In the other cases we applied {\HyperInt} to actually compute the corresponding periods and confirmed all of the results conjectured in \cite{BroadhurstKreimer:KnotsNumbers}. For example, we reported \eqref{eq:periods-examples} in \cite{Panzer:MasslessPropagators}.

Together with the computation of $K_{3,4}$ and $P_{7,8}$ by Oliver Schnetz, our results below combine to prove
\begin{theorem}
	The periods $\int \Omega/\psipol^{2}$ of all primitive logarithmically divergent $\phi^4$ graphs with at most seven loops, except $P_{7,11}$, are rational linear combinations of multiple zeta values. Furthermore, the period of $P_{7,11}$ is a linear combination of imaginary parts of multiple polylogarithms evaluated at a primitive sixth root of unity, divided by $\sqrt{3}$.
\end{theorem}
All these periods are now known explicitly and proven and this result is optimal in the sense that at $8$ loops, non-polylogarithmic constants emerge \cite{BrownSchnetz:K3Phi4,BrownDoryn:FramingsForGraphHypersurfaces}. Apart from these exceptions, the computation of the remaining periods is currently in progress by Oliver Schnetz, building in several cases on our computation of graphical functions.
\begin{figure}\centering%
$
		P_{7,8}
		=
		\Graph[0.42]{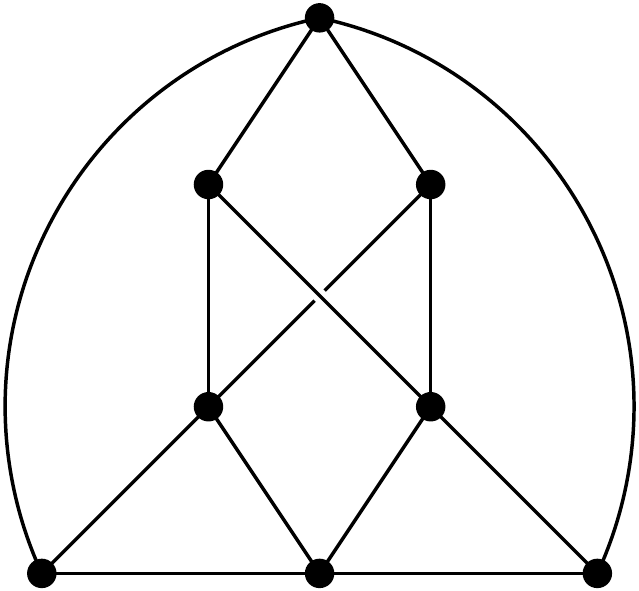}
		\qquad
		P_{7,9}
		=
		\Graph[0.5]{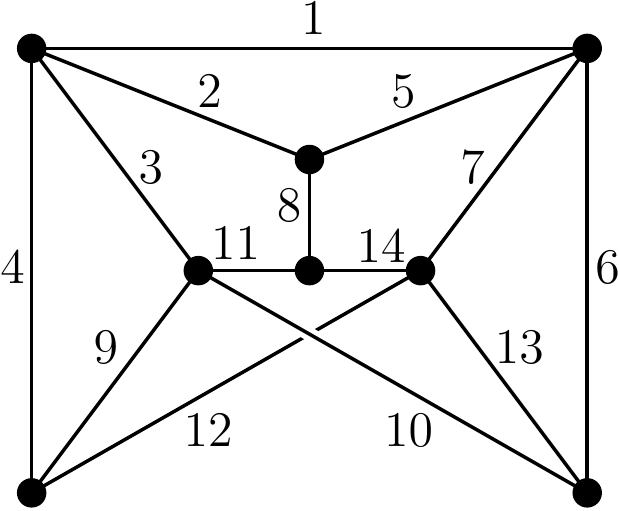}
		\qquad
		P_{7,11}
		=
		\Graph[0.5]{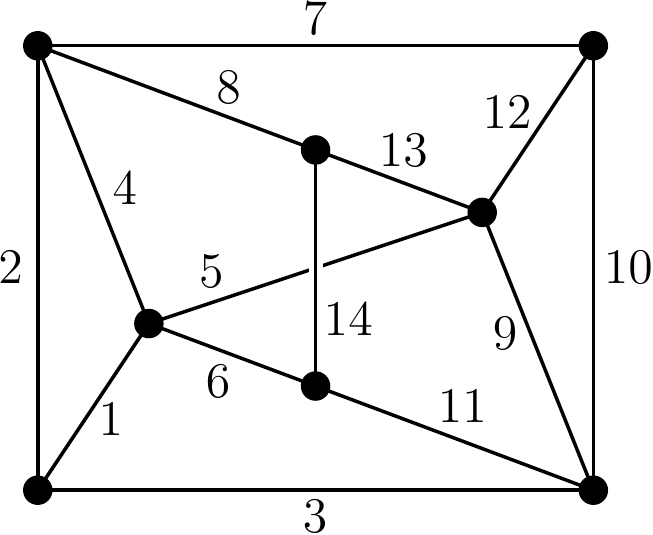}
$
	\caption[The $\phi^4$ graphs $P_{7,8}$, $P_{7,9}$ and $P_{7,11}$]{The most complicated primitive periods in $\phi^4$-theory at seven loops are given by the three graphs $P_{7,8}$, $P_{7,9}$ and $P_{7,11}$ in the notation of the census \cite{Schnetz:Census}.}%
	\label{fig:P78911}%
\end{figure}%

\subsection{No alternating sums: \texorpdfstring{$P_{7,9}$}{P 7,9}}
\label{sec:P79}%
The three graphs that could not be determined in \cite{BroadhurstKreimer:KnotsNumbers} are the \emph{complicated graphs} shown in figure~\ref{fig:P78911}. They behave unlike the other graphs in several respects.

One powerful tool to study the complexity of a Feynman integral is the $c_2$ invariant, or more generally the point counting function \cite{BrownSchnetzYeats:PropertiesC2,BrownSchnetz:K3Phi4,Doryn:InvariantInvariant}
\begin{equation}
	\set{\text{prime powers $q = p^n$}} \longrightarrow \N,
	\quad
	q \mapsto
	\PointCount{G}{q}
	\defas \abs{\setexp{\SP \in \F_q^{\edges(G)}}{\psipol_G(\SP) = 0}}.
	\label{eq:def-point-count}%
\end{equation}
It enumerates the number of zeros of the graph polynomial over a finite field $\F_q$ with $q$ elements. For all but the complicated graphs, this function is a polynomial in $q$. It was noticed in \cite{Schnetz:Fq} that $P_{7,8}$ (also mentioned by \cite{Doryn:CounterKontsevich}) and $P_{7,9}$ have exceptional prime $2$, so one needs two different polynomials to describe $\PointCount{G}{q}$ depending on whether $q$ is a power of $2$ or not.

Recently David Broadhurst \cite{Broadhurst:Radcor2013} obtained $22$ significant digits for both periods and was able to produce a convincing fit to the multiple zeta values
\begin{equation}\begin{split}
	P_{7,8}
	&= \tfrac{\numprint{22383}}{20} \mzv{11}
		+ \tfrac{4572}{5} \left( 
				\mzv{3,5,3}
				- \mzv{3}\mzv{3,5}
			\right)
			-700 \mzv[2]{3} \mzv{5}
			+ 1792 \mzv{3}\left( 
				\tfrac{27}{80}\mzv{3,5}
  	   + \tfrac{45}{64} \mzv{3} \mzv{5}
	     - \tfrac{261}{320} \mzv{8}
			\right),
	\\
	P_{7,9}
	&= \tfrac{\numprint{92943}}{160} \mzv{11}
     + \tfrac{3381}{20} \left(
		 		\mzv{3,5,3}
  	   - \mzv{3}\mzv{3,5} 
			 \right)
     - \tfrac{1155}{4} \mzv[2]{3} \mzv{5}
     + 896 \mzv{3}
    \left(
	    \tfrac{27}{80} \mzv{3,5}
     + \tfrac{45}{64} \mzv{3} \mzv{5}
     - \tfrac{261}{320} \mzv{8}
    \right).
\label{eq:P789}%
\end{split}\end{equation}
These are indeed more involved because of the last summand: All previous periods of weight $11$ are rational linear combinations of only the first three MZV; the product of $\mzv{3}$ and the weight $8$ combination is needed only for $P_{7,8}$ and $P_{7,11}$ (up to seven loops).

We confirmed (and thus proved) the result for $P_{7,9}$ with our program {\HyperInt}, while at the same time Oliver Schnetz verified $P_{7,8}$ independently.
\begin{remark}
	Strikingly, both the hyperlogarithm calculation of $P_{7,9}$ and the computation of $P_{7,8}$ with graphical functions need alternating sums. No matter which order of integration we choose, negative signs occur in the polynomials of the linear reduction, which result in regularized limits of the form $\int_0^{\infty} T(\set{0,\pm 1})$.

	Only after we reduced those sums to the data mine basis \cite{BluemleinBroadhurstVermaseren:Datamine}, we see that their combination lies in the subalgebra $\MZV \subseteq \MZV[2]$ of multiple zeta values. This phenomenon is ubiquitous in massless vacuum graphs (see section~\ref{sec:propagators} below) and raises at least the following interesting questions:
	\begin{enumerate}
		\item Do alternating sums appear at all as periods of massless vacuum graphs?
		\item Even if alternating sums (which are not MZV) occur at some stage, how can one understand the many cases when a period in $\MZV[2]$ is expressible as MZV?
		\item More generally, if a period of a graph lies in $\MZV$, is it possible to choose adapted variables for the integration such that the computation never introduces alternating sums in the first place?
	\end{enumerate}
	Oliver Schnetz very recently answered the first question in the affirmative\footnote{Personal communication on current research, to be published in the near future.}, but the others remain open.
\end{remark}

\subsubsection{Computational details}
To get a feel of the practical side of our calculations, we supply some details on our computation of $P_{7,9}$. First note that all of the complicated graphs are non-planar, so theorem~\ref{theorem:vw3-MZV} does not apply to them. We computed the linear reduction of section~\ref{sec:linear-reducibility} and chose the sequence
$e = (1, 2, 5, 8, 11, 14, 7, 6, 13, 10, 3, 12, 9, 4)$ of edges, with the labels of figure~\ref{fig:P78911}, for integration. So we set $I_0 \defas \restrict{\psipol}{\SP_4 = 1}^{-2}$ and compute $13$ integrals $I_k \defas \int_0^{\infty} I_{k-1} \dd \SP_{e_k}$ one after the other.

According to the linear reduction, each integrand is a hyperlogarithm in the next integration variable. We write them in the product basis \eqref{eq:barintegrals-hlog-products},
\begin{equation*}
	I_k \in
	\HlogAlgebra(\Sigma_{k,1})(\SP_{e_{k+1}}) \tp \cdots \tp \HlogAlgebra(\Sigma_{k,13-k})(\SP_{e_{13}})
	\quad\text{and set}\quad
	\Sigma_k \defas \Sigma_{k,1}\cupdot \cdots \cupdot \Sigma_{k,13-k}
\end{equation*}
to the full alphabet of letters. Each $I_k$ is homogeneous in weight as indicated in table~\ref{tab:P79-stats}. The complexity of each integration is dominated by three aspects:
\begin{itemize}
	\item the size of the alphabet $\Sigma_k$,
	\item the weight of the integrand and
	\item the size of the polynomials in the denominator and the letters $\Sigma_k$.
\end{itemize}
The first two determine the dimension of the ambient vector space of hyperlogarithms for an integration and strongly affect the runtime and memory requirements.
One might wonder why the first integrations take several seconds, even though they are elementary and involve only logarithms. The reason is that the graph polynomial to begin with has $1219$ monomials (so $P_{7,9}$ has as many spanning trees) and the program must compute partial fraction decompositions and factorizations of resultants of several polynomials with comparable size during the first integration steps. Note that {\HyperInt} does not exploit the advance knowledge of the factorizations \eqref{eq:dodgson-identity} and \eqref{eq:dodgson-jacobi} for Dodgson polynomials.

After a few integrations, the polynomials are very simple and the runtime is dominated by the recursion in the algorithm for \eqref{eq:reginf-as-hlog}.
Given that {\HyperInt} was not particularly designed towards efficiency, it was not expected by the author that this first implementation in a computer algebra system would allow for these interesting calculations at all.
Luckily we could resort to ample computing power and run up to fifty processes in parallel. Note that a single core would have needed two years for the eleventh integration only.

The denominator $d_{12}=(1+\SP_9)(1-\SP_9)$ of the last integrand $I_{12} = L_{12}(\SP_9)/d_{12}$, which is a hyperlogarithm $L_{12}(\SP_9) \in \HlogAlgebra(\set{-1,0})(\SP_9) \tp \MZV$, means that the primitive $\int I_{12}\ \dd \SP_9 \in \HlogAlgebra(\set{\pm 1,0})(\SP_9) \tp \MZV$ is a harmonic polylogarithm. The limit at $\SP_9 \rightarrow \infty$ is therefore computed as an alternating sum $\MZV[2]$.
\begin{remark}[Efficiency]
	This example shows a particularly bad choice of integration order. A careful investigation and preparation can drastically reduce the runtime required for such a calculation. For instance, a suitable transformation of the last four Schwinger parameters completely trivializes the $11$th integration (see remark~\ref{remark:last-integrations-3pt} below), which is excessively complex in the original variables (see table~\ref{tab:P79-stats}).

	However, the optimal approach to compute such periods was introduced by Oliver Schnetz. With graphical functions \cite{Schnetz:GraphicalFunctions}, he completed the calculation of $P_{7,8}$ in $15$ \emph{minutes} on a single core.
\end{remark}
\begin{table}
	\centering
	\begin{tabular}{rcccccccccccccc}
		Integrated $k$ & 1 & 2 & 3 & 4 & 5 & 6 & 7 & 8 & 9 & 10 & 11 & 12  \\\toprule
		Weight of $I_k$ & 0 & 1 & 1 & 2 & 3 & 4 & 5 & 6 & 7 & 8 & 9 & 10 \\
		$\abs{\Sigma_{k,1}}$ & $0$ & $4$ & $3$ & $8$ & $17$ & $16$ & $17$ & $15$ & $10$ & $5$ & $4$ & $2$\\
		$\abs{\Sigma_k}$ & 0 & 12 & 18 & 21 & 31 & 28 & 39 & 22 & 18 & 8 & 6 & 2 \\
Terms in $I_k$
& $1 $
& $12 $
& $18 $
& $125$
& $799$
& $8\Kilo$
& $48\Kilo$ 
& $100\Kilo$
& $280\Kilo$
& $220\Kilo$
& $27\Kilo$
& $589 $
\\
\midrule
Time $[\si{\second}]$
& $24$ 
& $12$ 
& $10$ 
& $28$ 
& $124$ 
& $5\Kilo$
& $400\Kilo$
& $4\Mega$
& $29\Mega$
& $62\Mega$
& $97\Kilo$
& $136$ 
\\
Mem $[\si{\gibi\byte}]$
& $1$
& $1$
& $1$
& $2$
& $2$
& $4$
& $20$
& $45$
& $17$
& $23$
& $8$
& $7$
\\
\bottomrule
	\end{tabular}%
	\caption{Details on the integration of $P_{7,9}$. We abbreviate $\Kilo=10^3$ and $\Mega=10^6$.}%
	\label{tab:P79-stats}%
\end{table}

\subsection{Primitive sixth roots of unity: \texorpdfstring{$P_{7,11}$}{P 7,11}}
\label{sec:P711}%
The by far most interesting and complicated period at seven loops is $P_{7,11}$. Its point counting function is unique (below eight loops) because it has the exceptional prime $3$, saying that $\PointCount{P_{7,11}}{q} = P_{q \mod 3}$ is given by three different polynomials depending on the remainder of $q$ by $3$ \cite{Schnetz:Fq}. Its period was computed to $11$ significant digits in \cite{BroadhurstKreimer:KnotsNumbers} and could not be assessed due to this lack of accuracy.

This graph is not linearly reducible because a quadratic denominator appears at some point for any possible order of integration \cite{Brown:PeriodsFeynmanIntegrals}. We choose $e=(8,13, 14, 6, 11, 5, 9, 12, 7, 10, 3, 4, 1, 2)$ with respect to the edge labels of figure~\ref{fig:P78911} (so we set $\SP_2 = 1$) which allows us to perform the first ten integrations. At this stage, the partial integral $I_{10} = L_{10} / d_{10}$ consists of a hyperlogarithm $L_{10}$ of weight $8$ and the irreducible, totally quadratic denominator
\begin{equation}
d_{10}
= \SP_{3}^2(\SP_{1}+\SP_{2}+\SP_{4})^2 + \SP_{2} \SP_{3} (\SP_{1}+\SP_{2}+\SP_{4}) (2 \SP_{1}-\SP_{4}) +\SP_{1} \SP_{2} (\SP_{1} \SP_{2}-\SP_{2} \SP_{4}-\SP_{4}^2).
	\label{eq:P711-d10}%
\end{equation}
Further integration of a Schwinger parameter would force us to take square roots of the discriminant of this polynomial and therefore escape the space of hyperlogarithms with rational letters (prohibiting any subsequent integration with our algorithms).
\subsubsection{Changing variables to extend linear reducibility}
But \eqref{eq:P711-d10} can be linearized: The way we presented it suggests to replace the variable $\SP_3$ by $\SP_3'$ such that $\SP_3 (\SP_1 + \SP_2 + \SP_4)=\SP_3' \SP_1$. Then one power of $\SP_1$ factors out of $d_{10}$ and the remaining factor is linear in $\SP_1$. However, it is not enough to only consider the denominator; the full polynomial reduction must be considered.

By computation we find $11$ polynomials in the reduction at this stage, which are all linear in each of the remaining variables (except $d_{10}$). The change of variables just described turns some of these polynomials into quadratic ones. We finally found that we should adhere two further replacements 
$\SP_4 = \SP_4'(\SP_2 + \SP_3')$ and then $\SP_1 = \SP_1' \SP_4'$ to obtain overall the change of variables
\begin{equation}
	\begin{aligned}
		\SP_1 &= \SP_1' \SP_4' \\
		\SP_4 &= \SP_4' (\SP_2 + \SP_3')
	\end{aligned}
	\quad
	\SP_3 = \frac{\SP_1' \SP_3' \SP_4'}{\SP_2 + \SP_4'(\SP_1' + \SP_2 + \SP_3')}
	\quad
	J = 
\frac{
\SP'_1 {\SP'_4}^2 (\SP_{2}+\SP'_{3})
}{
	\SP_{2}+\SP'_{4} (\SP'_1 + \SP_{2}+\SP'_{3})
}
	\label{eq:P711-variable-change}%
\end{equation}
with Jacobian $\dd[3] \SP = J\, \dd[3] \SP'$ ($\SP_2$ is left untouched). In these variables, not only the new denominator
\begin{equation}
	d_{10}'
	= d_{10}/J
	= (\SP_2 + \SP_3') (\SP_2 + \SP_2 \SP_4' - \SP_1') (\SP_1' \SP_4' + \SP_2 + \SP_2 \SP_4' + \SP_3' \SP_4')
	\label{eq:P711-d10'}%
\end{equation}
linearizes completely, but surprisingly also all other polynomials in the reduction. Their irreducible factors are given by the monomials and the linear (in each variable)
\begin{equation*}
	\Big\{
		\SP'_{1} \pm 1, 
		\SP'_{3} + 1, 
		\SP'_{4} \pm 1, 
		\SP'_{1} \SP'_{4} \pm 1, 
		1+\SP'_{1} \SP'_{4}+\SP'_{1},
		1+\SP'_{1} \SP'_{4}+\SP'_{4},
		\SP'_{1}-\SP'_{4}-1,
		1 + \SP'_{4}(1+\SP'_{1}+ \SP'_{3})
\Big\}.
\end{equation*}
Here we have already set $\SP_2 = 1$, which we are now forced to do: The non-linear transformation \eqref{eq:P711-variable-change} means that the $\delta$-constraint in the projective form \eqref{eq:projective-delta-form-integrand} becomes non-linear in our new variables. To avoid this and keep the integration domain simple, we just choose to fix the untouched $\SP_2$. Note that only $\SP_3'+1$ and $1+\SP_4'(1+\SP_1'+\SP_3')$ depend on $\SP_3'$, so after integrating $\SP_4'$ there can only be linear polynomials in $\SP_3'$.
\begin{corollary}\label{corollary:P711-reducible}
	After the change of variables \eqref{eq:P711-variable-change}, $P_{7,11}$ is linearly reducible.
\end{corollary}
\begin{table}
	\centering
	\begin{tabular}{rccccccccccccc}
	Integrated $k$ &	1 & 2 & 3 & 4 & 5 & 6 & 7 & 8 & 9 & 10 & 11 & 12 \\
\toprule
	Weight of $I_k$ & 0 & 1 & 1 & 2 & 3 & 4 & 5 & 6 & 7 & 8 & 9 & 10 \\
	$\abs{\Sigma_{k,1}}$ & $0$ & $3$ & $5$ & $17$ & $22$ & $19$ & $21$ & $18$ & $10$ & $0$ & $4$ & $3$ \\
	$\abs{\Sigma_k}$ & $0$ & $13$ & $18$ & $43$ & $31$ & $34$ & $49$ & $28$ & $18$ & $6$ & $6$ & $3$ \\
	Terms in $I_k$ &
1 & 13 & 18 & 182 &  733 &
$11\Kilo$ &
$92\Kilo$ &
$354\Kilo$ &
$101\Kilo$ &
8723 &
$58\Kilo$ &
3894 \\
\bottomrule
	\end{tabular}
	\caption{Details on the partial integrals of $P_{7,11}$.}%
	\label{tab:P711-statistics}%
\end{table}%
\subsubsection{Integration}
When we actually do the computation, we find that $I_{10}' = L_{10}'/d_{10}'$ is a hyperlogarithm
\begin{equation*}
	L_{10}' \in \HlogAlgebra\left( \set{0,-1,-\frac{1}{\SP_1'}, -\frac{1}{1+\SP_1'}} \right)(\SP_4')
	\tp \HlogAlgebra\left( \set{\pm 1, 0} \right)(\SP_1')
	\tp \Q[\mzv{2},\mzv{3},\mzv{5}]
\end{equation*}
involving only the small subset $\SP_4'+1$, $\SP_1'\SP_4'+1$, $\SP_1' \SP_4' + \SP_4' + 1$ and $\SP_1' \pm 1$ of polynomials. This behaviour is expected because we are computing the particular period $\int \Omega/\psipol^2$; further letters will occur for more general periods (see also the remarks \cite[\S4.5]{Brown:TwoPoint}). Interestingly, $L_{10}'$ does not depend on $\SP_3'$ at all (as we explain in remark~\ref{remark:last-integrations-3pt} below), so the next integration is elementary:
\begin{equation}
	I'_{11}
	= L_{10}' \int_0^{\infty} \frac{\dd \SP_3'}{d_{10}'}
	= \frac{L_{10}'}{(1+\SP_1' \SP_4')(1+\SP_4' - \SP_1')}
	\ln \frac{1+\SP_1' \SP_4' + \SP_4'}{\SP_4'}.
	\label{eq:P711-I11}%
\end{equation}
When we integrate out $\SP_4'$, we are finally left with an integrand
\begin{equation}
	I'_{12} = \frac{L_{12}'}{1 - \SP_1' + {\SP_1'}^2}
	\quad\text{and find that}\quad
	L_{12}' \in \HlogAlgebra\left( \set{0,\pm 1} \right)(\SP_1') \tp \Q[\mzv{2},\mzv{3},\mzv{5}]
	\label{eq:P711-I12}%
\end{equation}
is a harmonic polylogarithm of uniform weight $10$, given by $3894$ individual words. Again we see a simplification compared to a general period, because already at this stage we could in principle have the roots $\xi_6^{\pm 1} = e^{\pm \imag \pi/3}$ of $1-\SP_1'+{\SP_1'}^2$ as letters in $L_{12}'$. We know that the partial integral $I_{12}'(\SP_1')$ is analytic on $\SP_1' \in (0,\infty)$, which implies the non-trivial constraint that $L_{12}'(\SP_1')$ must be analytic at $\SP_1' \rightarrow 1$. We verified this explicitly. Since
\begin{equation}
	\frac{1}{1-\SP_1'+{\SP_1'}^2}
	=\frac{1}{\imag\sqrt{3}}\left( \frac{1}{\SP_1'-\xi_6} - \frac{1}{\SP_1'-\conjugate{\xi_6}} \right)
	= \frac{2}{\sqrt{3}} \Imaginaerteil \frac{1}{\SP_1'-\xi_6},
	\label{eq:P711-imag-denominator}%
\end{equation}
the final integration amounts just to prepending the letter $\letter{\xi_6}$ to each word in $L_{12}'$ and taking the imaginary part. After splitting the domain into $\SP_1' \in (0,1)$ and $\SP_1' \in (1,\infty)$ we can thus compute the period of $P_{7,11}$ as a $\MZV$-linear combination of multiple polylogarithms of the form $\Imaginaerteil(\Li_{n_1,\ldots,n_r}(\pm 1, \cdots, \pm 1, \xi_6))$. There are $\numprint{39366} = 2\cdot 3^9$ of these objects at weight $11$ and our explicit result uses only $4589$ out of these.

This algebra of periods (harmonic polylogarithms evaluated at $\xi_6$) was already investigated in \cite{Broadhurst:SixthRoots}, but no reduction to a conjectured basis (like the data mine \cite{BluemleinBroadhurstVermaseren:Datamine}) exists to the weight $11$ we require. But at least our exact result in this form can be evaluated numerically to very high precision with standard methods. We used Oliver Schnetz's {\zetaprocedures}\cite{Schnetz:ZetaProcedures} to compute $5000$ significant digits, which begin with
\begin{equation*}
	\period(P_{7,11})
	\approx \numprint{%
200.357566429275446967634590990100073795036337663163840606
}\ldots
	\label{eq:P711-numeric}%
\end{equation*}
This was enough (using PSLQ in the implementation \cite{Bailey:ARPREC}) to disprove that $\period(P_{7,11})$ could be in $\MZV$ with any reasonable size of integer coefficients.
\begin{remark}[Independence of $L_{10}'$ of $\SP_3'$]\label{remark:last-integrations-3pt}
	The graph $H \defas G_{10}$ built from the first $10$ edges ($5$ through $14$) in the integration order connects to the last four edges ($H' \defas G^{10}$) only through three vertices $v_i$. Let $v_i$ denote the vertex incident to edge $i$ ($1 \leq i \leq 3$), so edge $4$ connects $v_1$ and $v_2$. Then
	$\psipol_G = f_{123} \psipol_{H'} + f_{123}' \psipol_H + \sum_{i<j} f_i^{} f'_j$ in terms of
	$
		\psipol_{H'} = \SP_1 + \SP_2 + \SP_4
	$
	and the spanning forest polynomials \eqref{eq:3pt-forest-shorthand} of $H$ and $H'$, which are
	\begin{equation*}
		f_1' = \SP_1 \SP_4 
		,\ 
		f_2' =  \SP_2 \SP_4 
		,\ 
		f_3' = \SP_1 \SP_2 + \SP_1 \SP_3 + \SP_2 \SP_3 + \SP_3 \SP_4
		,\ 
		f_{123}' = \SP_4 (\SP_1\SP_2 + \SP_1\SP_3 + \SP_2 \SP_3).
	\end{equation*}
	Thus the dependence of $I_{10}$ is only through the three ratios $f_i'/\psipol_{H'}$, explicitly
	\begin{equation*}
		I_{10}
		= \frac{1}{\psipol_{H'}^2}
			\int_{\R_+^3} \GfunForest{H}(z)
		\left[ \sum_{1 \leq i< j \leq 3} \left( z_i + \frac{f_i'}{\psipol_{H'}}\right) \left( z_j +  \frac{f_j'}{\psipol_{H'}} \right)\right]^{-2}
		\ \dd[3] z.
	\end{equation*}
	We can exploit the homogeneity \eqref{eq:GfunForest-scaling} to reduce it further down to only two variables. After such a rescaling and the change of variables \eqref{eq:P711-variable-change}, we find that
	\begin{equation*}
		I_{10}
		= \frac{1}{{\SP_4'}^2 (1+\SP_3')^2} 
		\int_{\R_+^3}\frac{\GfunForest{H}(z)\ \dd[3] z}{
			\left[(z_1 + \SP_1' \SP_4') (1+\SP_1'+ z_2 + z_3) + (\SP_1' + z_3) (1+ z_2)\right]^{2}
		}
	\end{equation*}
	has a simple (rational) dependence on $\SP_3'$.
\end{remark}

\subsubsection{Identification}
Further progress was possible through Oliver Schnetz's intuition from his independent, amazing achievements. Using graphical functions he recently evaluated a graph with $8$ loops ($P_{8,33}$) to a multiple polylogarithm at sixth roots of unity. More precisely, the product of this period with $\imag\sqrt{3}$ lies in the algebra $\Deligne$ of rational linear combinations of $\Li_{\vec{n}}(\xi_6) = \Li_{n_1,\ldots,n_r}(1,\ldots,1,\xi_6)$, which we discussed in section \ref{sec:Periods}. He conjectured that $\period(P_{7,11})$ should also be of this form.

That would imply that it can be expressed as a linear combination of the conjectural basis elements \eqref{eq:Deligne-Generators} given by products of $\Li_{\vec{n}}(\xi_6)$ for Lyndon words $\vec{n} \in \Lyndons(\N \setminus \set{1})$ and powers of $\imag\pi$. This conjectural basis has $144$ elements at weight $11$, which was still too big to find a fit with PSLQ before our $5000$ digits of precision were exhausted.

The final ingredient was our decomposition \eqref{eq:Deligne-oddeven} of Deligne's basis into real- and imaginary parts: We take $\Realteil \Li_{\vec{n}}(\xi_6)$ for Lyndon words $\vec{n}=(n_1,\ldots,n_r)$ with $\abs{\vec{n}}+r$ even and $\imag\Imaginaerteil \Li_{\vec{n}}(\xi_6)$ otherwise. For the following extreme precision integer relation detections we used the PSLQ implementation \cite{Bailey:ARPREC}. Note that it works with real numbers and we do not keep track of powers of $\imag$ for now; so actually we consider the algebra $\Realteil \Deligne + \Imaginaerteil \Deligne \subset \R$.

To gain confidence with our setup, we first computed all $144$ basis elements at weight $11$ to $\numprint{10000}$ digits and checked:
\begin{itemize}
	\item There are no integer relations between these basis elements with coefficients of less than $67$ digits (further verification requires higher precision evaluations for the basis).
	\item The complementary pieces $\Realteil \Li_{\vec{n}}(\xi_6)$ for $\abs{\vec{n}}+r$ odd (and the imaginary part otherwise) fulfil an integer relation with the basis elements.
\end{itemize}
We proved that the latter relations \eqref{eq:mpl-xi6-ReIm-reducible} exist, but the question was whether these are reliably detected by the program. As a general rule, one needs at least $n\cdot d$ digits of precision to detect a relation among $n$ elements with integer coefficients of $d$ digits or less. In our case, these coefficients are very large and for some Lyndon words, the relation entailed integer coefficients with up to $50$ digits and required $7000$ digits of accuracy to be detected. This explains why our attempt to find a relation with $P_{7,11}$ with only $5000$ digits failed.

Considering \eqref{eq:P711-imag-denominator} we expect that this period should in fact belong to the imaginary subspace $\sqrt{3}\period(P_{7,11}) \in \Imaginaerteil \Deligne$. Proposition~\ref{prop:Deligne-oddeven} provides an explicit $\Q$-basis of this space, which has only $72$ elements such that $5000$ digits should be ample to detect our expected relation. Only $3000$ were exhausted when it was found (and confirmed to the full available accuracy of $5000$ digits), with integer coefficients of up to $40$ digits. It only uses $30$ out of the $72$ basis elements (the other coefficients in the relation are zero).

In our final result, we replaced some periods with more familiar multiple zeta values and shortened $\Li_{\vec{n}}(\xi_6)$ to just $\Li_{\vec{n}}$:
\begin{align}
	\sqrt{3} \, \period(P_{7,11}) &= 
	\Imaginaerteil \Big(
		\tfrac{\numprint{19285}}{6} \mzv{9} \Li_{2}
		-\tfrac{1029}{2} \mzv{7} \Li_{4}
		+240 \mzv[2]{3}(9 \Li_{2,3}-7 \mzv{3} \Li_{2})
	\Big)
		-\tfrac{\numprint{93824}}{9675}\pi^3 \mzv{3,5} 
	\nonumber\\ &
		+ \tfrac{2592}{215} \Imaginaerteil\Big(
				36  \Li_{2,2,2,5}
				+27 \Li_{2,2,3,4}
				+9  \Li_{2,2,4,3}
				+9  \Li_{2,3,2,4}
				+3  \Li_{2,3,3,3}
	\nonumber\\&\qquad\qquad\quad
				-43 \mzv{3}(\Li_{2,3,3}+3 \Li_{2,2,4})
		\Big)
		-\tfrac{\numprint{96393596519864341538701979}}{\numprint{790371465315684594157620000}} \pi^{11}
	\nonumber\\&
	+\tfrac{216}{\numprint{14755731798995}} \Imaginaerteil\Big(
			\numprint{2539186130125890} \Li_{8} \mzv{3}
			-\numprint{1269593065062945} \Li_{2,9}
	\nonumber\\&\qquad\qquad\qquad\qquad\quad\!
			-\numprint{413965317054502} \Li_{6} \mzv{5}\,
			-\ \,\numprint{996412983391539} \Li_{3,8}
	\nonumber\\&\qquad\qquad\qquad\qquad\quad\!
			-\numprint{546306741059841} \Li_{4,7}\ \ 
			-\ \ \numprint{156228639992955} \Li_{5,6}
	\Big)
	\nonumber\\&
	+\tfrac{2592}{\numprint{10945435}} \pi^2\Imaginaerteil\Big(
				\numprint{287205} \Li_{2,7}
				-\numprint{574410} \Li_{6} \mzv{3}
				+\numprint{55687} \Li_{4,5}
				+\numprint{168941} \Li_{3,6}
		\Big)
	\nonumber\\&
	+\pi \left(
		\tfrac{\numprint{11613751}}{9030} \mzv[2]{5}
		+\tfrac{\numprint{267067}}{602} \mzv{3,7}
		-\tfrac{\numprint{31104}}{215} \Realteil(	3 \Li_{4,6}+10 \Li_{3,7} )
	\right).
	\label{eq:P711}%
\end{align}
\begin{remark}
	After our computation, David Broadhurst investigated the source of the huge prime factors in the denominators of the depth two contributions to \eqref{eq:P711}. He discovered that Deligne's basis in terms of Lyndon words is far from optimal and hand-crafted an alternative basis which behaves much better ($2$ and $3$ are the only prime factors of the denominators).
	His data mine \cite{Broadhurst:Aufbau} extends \cite{BluemleinBroadhurstVermaseren:Datamine} to sixth roots of unity and provides interesting conjectures on $\Deligne$.
\end{remark}
\begin{remark}
	The strongest form of a conjecture due to Oliver Schnetz states that periods of massless $\phi^4$ theory are closed under the motivic coaction. This is a highly non-trivial condition and constrains the type of periods that may appear drastically. Oliver Schnetz confirmed, based on our numeric result and his extension of the coproduct-based decomposition algorithm presented in \cite{Brown:DecompositionMotivicMZV}, that $P_{7,11}$ indeed supports this conjecture.
	A much less stringent form of it was very recently proved by Francis Brown \cite{Brown:Houches14Coaction} and we are looking forward to further development of this fascinating idea.
\end{remark}
\begin{remark}
	Our integration of $P_{7,11}$ that produced the exact expression in terms of $\Li_{\vec{n}}(\pm 1,\ldots,\pm 1, \xi_6)$ was completed in one week. We point this out to show that the runtime for such a computation can be reduced drastically if a good integration order is chosen and variables are transformed suitably. Compare the tables~\ref{tab:P711-statistics} and \ref{tab:P79-stats}: The most demanding integration for $P_{7,9}$ was the eleventh (integrating $I_{10}$ to get $I_{11}$), while this step was trivial for $P_{7,11}$ due to our choice of variables (see remark~\ref{remark:last-integrations-3pt}).
\end{remark}

\section{Massless propagators}\label{sec:propagators}%
The renormalization group functions ($\beta$-functions and anomalous dimensions) are essential to understand a quantum field theory, as they describe their asymptotic behaviour. In the minimal subtraction scheme for dimensional regularization, these functions are determined by the counterterms which in turn can be computed from massless propagators (Feynman integrals with two external legs, also known as \emph{p-integrals}). Their calculation up to three loops in the seminal article \cite{ChetyrkinTkachov:IBP}, which promoted the technique of integration by parts (IBP), was a milestone in perturbative quantum field theory.

Given the strong demand for higher order corrections, it seems very surprising that it took almost thirty years before the required massless propagators at four loops could be computed \cite{BaikovChetyrkin:FourLoopPropagatorsAlgebraic,SmirnovTentyukov:FourLoopPropagatorsNumeric,LeeSmirnov:FourLoopPropagatorsWeightTwelve}. We shall recall how hyperlogarithms solve this problem of expanding any massless propagator up to four loops to arbitrary order \cite{Panzer:MasslessPropagators}, which extends to the infinite family of propagators obtained by iterated insertions of $\leq 4$-loop propagators into each other. This approach is due to Francis Brown \cite{Brown:TwoPoint}.

\subsection{One-scale insertions}\label{sec:one-scale-insertions}
Our goal is to compute the \emph{$\varepsilon$-expansion} (the Laurent series at $\varepsilon \rightarrow 0$) of dimensionally regulated Feynman rules $\FR(G) \in \R[\varepsilon^{-1},\varepsilon]]$ in $\dimension = 4-2\varepsilon$ dimensions.\footnote{We could as well expand around any other even dimension.}  The second Symanzik polynomial $\phipol = q^2 \restrict{\phipol}{q^2=1}$ is proportional to the only available scale, the square of the external momentum. Hence by \eqref{eq:feynman-integral-projective}, this dependence is just a power law $\FR(G) = q^{-2\sdd} \restrict{\FR(G)}{q^2=1}$ and we will set $q^2 = 1$ henceforth.

This shows that whenever a graph $G$ contains a subgraph $\gamma$ which has only two external momenta, it can be replaced by a single propagator raised to the power $\sdd$ (and the constant $\FR(\gamma)$). As a special case, any pair of edges meeting at a two-valent vertex is equivalent to a single propagator with the sum of the corresponding indices. 
\begin{figure}
	$ \Graph[0.7]{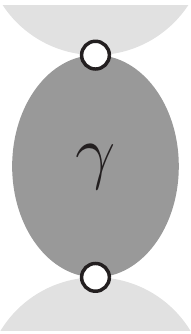} \mapsto \Graph[0.7]{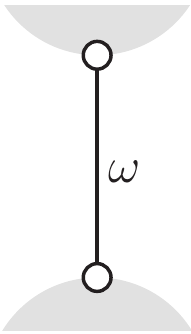} $
	\hfill
	$ \Graph[0.7]{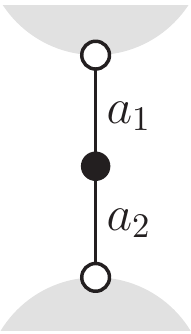} \mapsto \Graph[0.7]{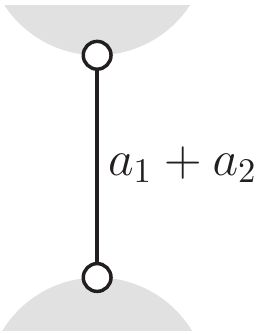} $
	\hfill
	$ F' = \Graph[0.5]{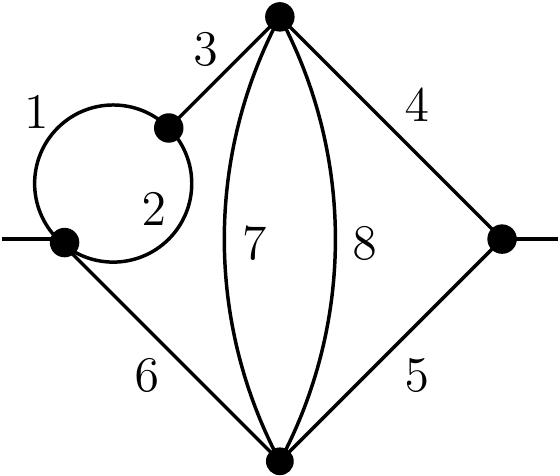} $
	\caption[Reduction rules for one-scale subgraphs]{Reduction rules for one-scale subgraphs. The white circles show the only two vertices where the subgraph $\gamma\subset G$ connects to the remainder of $G$. $F'$ is computed in example~\ref{ex:onescale-sub}.}%
	\label{fig:1scale-reductions}%
\end{figure}%
\begin{example}\label{ex:onescale-sub}
	The propagator $F'$ of figure~\ref{fig:1scale-reductions} can be computed in terms of the one- and two-loop master integrals from equation \eqref{eq:oneloop-master} and figure~\ref{fig:low-loop-masters}:
	\begin{equation*}
		\FR(F',\EP_1,\ldots,\EP_9)
		= \left[\onemaster{1}{1} \right]^2 \FR(F,\EP_1+\EP_2+\EP_3 - 2+\varepsilon, \EP_4,\EP_5,\EP_6,\EP_7+\EP_8 - 2+\varepsilon).
	\end{equation*}
\end{example}
These reduction rules are sketched in figure~\ref{fig:1scale-reductions} and drastically decrease the number of graphs that must be considered (we say $G$ is $3$-connected if no reduction rule can be applied to it). However, this requires the expansion of Feynman integrals not only in $\dimension$, but also in the indices $\EP_e$ as they can become $\varepsilon$-dependent if the corresponding edge involved a reduction of a graph with loops.

Standard methods cannot compute such general expansions analytically. Apart from the elementary one-loop master integral \eqref{eq:oneloop-master}, only the two-loop master integral had heretofore been computable in this general sense to arbitrary order \cite{BierenbaumWeinzierl:TwoPoint}. Just this single graph occupied physicists for decades \cite{Grozin:TwoLoop}.

Below we will summarize our findings in \cite{Panzer:MasslessPropagators}, namely that hyperlogarithms suffice to compute the simultaneous expansions (in $\varepsilon$ and $\EP_e$) of all massless propagators up to four loops, to arbitrary order. To clarify the significance of this, to our mind, striking progress let us briefly compare to the earlier results.
\begin{figure}\centering
	\begin{tabular}{ccccc}
		$\Graph[0.5]{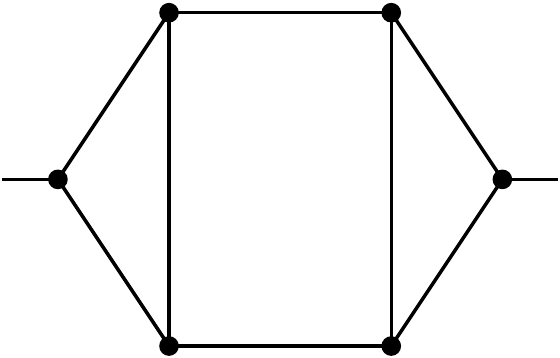}$&
		$\Graph[0.5]{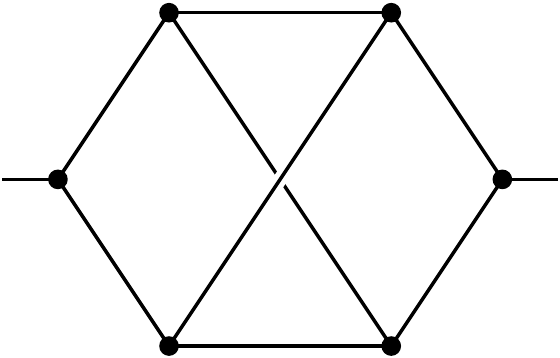}$&
		$\Graph[0.5]{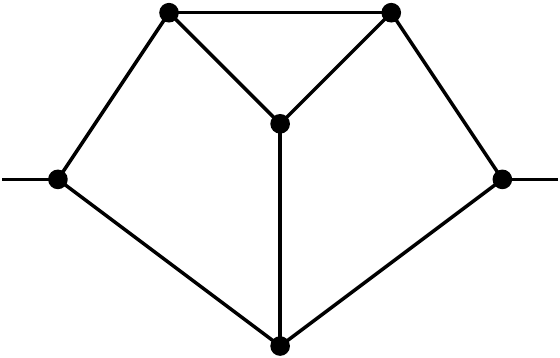}$&
		$\Graph[0.5]{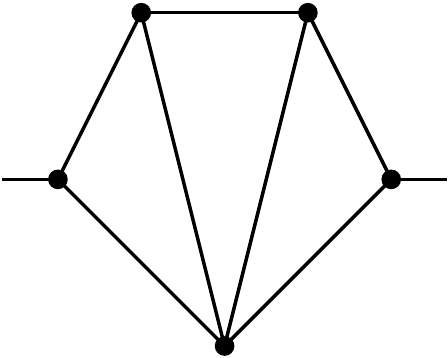}$&
		$\Graph[0.5]{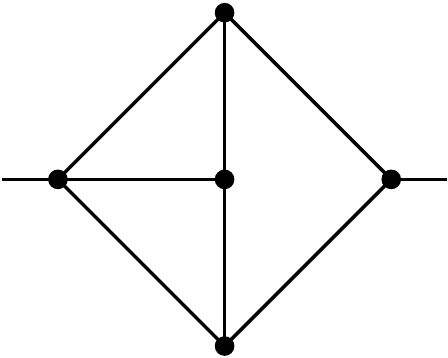}$ \\
		$L$ & $N$ & $M$ & $V$ & $Q$ \\
	\end{tabular}
	\caption{The five $3$-connected $3$-loop propagators as named in \cite{ChetyrkinTkachov:IBP}.}%
	\label{fig:3loop-propagators}%
\end{figure}

All $3$-connected $3$-loop propagators are shown in figure~\ref{fig:3loop-propagators} as determined in \cite{ChetyrkinTkachov:IBP}. The method of integration by parts provides relations between Feynman integrals $\FR(G;\vec{\EP}+\vec{z})$ with integer shifts $\vec{z} \in \Z^{\edges}$ of indices. It was applied to prove that the $\varepsilon$-expansion of any $3$-loop propagator \emph{with integer indices $\vec{\EP} \in \Z^{\edges}$} can be expressed through $\Gamma$-functions and the $\varepsilon$-expansions of $\FR(N;1,\ldots,1)$ and $\FR(L;1,\ldots,1)$. While the latter is rather simple to compute, the calculation of $N$ proceeded very slowly. Of the expansion
\begin{gather}
	\frac{\FR(N;1,\ldots,1)}{\Gscheme^3 (1-2\varepsilon)^2}
	= 20 \mzv{5}
+\left(\tfrac{80}{7} \mzv[3]{2} + 68 \mzv[2]{3}\right)\varepsilon
+\left(\tfrac{408}{5} \mzv{3} \mzv[2]{2} + 450 \mzv{7}\right)\varepsilon^{2} 
+\left(\tfrac{
	\numprint{102228}}{125} \mzv[4]{2}
	-2448 \mzv{3} \mzv{5}
	\right. \nonumber\\ \left.
	-\tfrac{9072}{5} \mzv{3,5}
\right)\varepsilon^{3}
+\left(
	\tfrac{\numprint{88036}}{9} \mzv{9}
	-\tfrac{4640}{3} \mzv[3]{3} 
	-\tfrac{\numprint{10336}}{7} \mzv[3]{2} \mzv{3} 
	+\tfrac{\numprint{19872}}{5} \mzv[2]{2} \mzv{5}
\right)\varepsilon^{4}
+\bigo{\varepsilon^{5}}
,
	\label{eq:N-expansion}%
\end{gather}
where $\Gscheme \defas \varepsilon\onemaster{1}{1}$ is a common prefactor, only the first three coefficients had been proven before. Similarly, the calculation \cite{BaikovChetyrkin:FourLoopPropagatorsAlgebraic} of all master integrals for four-loop massless propagators determines a finite number of coefficients for each graph (with all $\EP_e = 1$), just sufficient for four-loop computations. A reducible $5$-loop graph is in general not computable with this data, since the resulting $\leq 4$-loop propagator would need to be expanded in the index of an edge.

An alternative approach to exact computation is the numeric evaluation of the integrals to very high precision, which has recently become feasible and very effective through dimensional recurrence relations \cite{LeeSmirnov:EasyWay,LeeSmirnov:FourLoopPropagatorsWeightTwelve}. This method allows to obtain convincing fits of the numeric coefficients to multiple zeta values and suggests, tested to very high weight, that this might hold to all orders. A slightly weaker result can be proved with hyperlogarithms, even for expansion in the indices.
\begin{figure}
	$	Y_3 = \CloseProp{L} = \CloseProp{M} = \Graph[0.45]{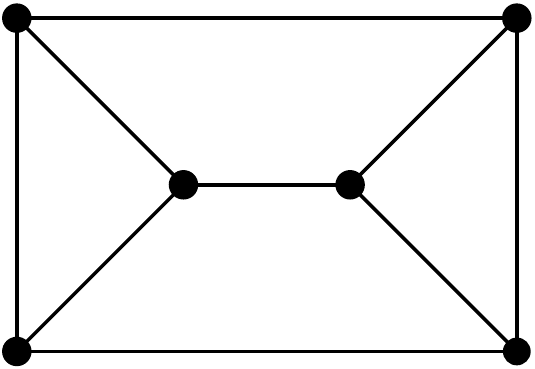} $ \hfill
	$	\WS{4} = \CloseProp{Q} = \CloseProp{V} = \Graph[0.4]{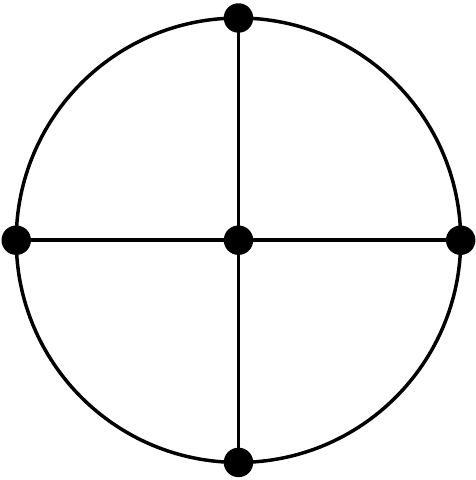} $ \hfill
	$	K_{3,3} = \CloseProp{N} = \Graph[0.3]{k3_3}	$
	\caption[Glueings of three-loop propagators]{Glueing the external edges of $L$ and $M$ gives the triangular prism $Y_3$, while $Q$ and $V$ yield the wheel with four spokes $\WS{4}$ and $N$ results in the complete bipartite graph $K_{3,3}$.}%
	\label{fig:4loop-vacuum}%
\end{figure}%
\subsection{Parametric integration}
Propagators $G$ can be transformed into vacuum integrals $\CloseProp{G} \defas G \cupdot \set{e_0}$ through addition of an edge $e_0$ that connects the two vertices where the external momenta enter (see figure~\ref{fig:4loop-vacuum}). This operation is called \emph{glueing} and conversely, \emph{cutting} any edge of a vacuum graph gives back a propagator. It is well-known that this symmetry relates the associated Feynman integrals, which is easily seen in Schwinger parameters \cite{Brown:TwoPoint}. Namely, the first Symanzik $\psipol_{\CloseProp{G}} = \psipol_G^{} \SP_{e_0} + \phipol_G^{}$ combines both polynomials of the propagator such that\footnote{%
	The proof is just $
	\int_0^{\infty} \SP_{e_0}^{\dimension/2 - \sdd -1}/ (\psipol \SP_{e_0} + \phipol)^{\dimension/2} \ \dd \SP_{e_0} 
		= \phipol^{-\sdd} \psipol^{\sdd - \dimension/2}\Gamma(\dimension/2-\sdd)\Gamma(\sdd)/\Gamma(\dimension/2)
	$.}
\begin{equation}
	\FR(G) 
	\urel{\eqref{eq:feynman-integral-projective}}
	\frac{
		\Gamma\left(\dimension/2 \right)
	}{
		\prod_{e\in \edges(\CloseProp{G})} \Gamma(\EP_{e})
	}
	\int \Omega\ I_{\CloseProp{G}}
	\quad\text{when we set}\quad
	\EP_{e_0}
	\defas \frac{\dimension}{2} - \sdd_{G}.
	\label{eq:projective-glueing}%
\end{equation}
We can therefore transform the expansions of two propagators with the same glueing into each other, and relations among vacuum graphs correlate different propagators that had so far been computed separately (we discuss explicit examples in \cite{Panzer:MasslessPropagators}). The expansion of \eqref{eq:projective-delta-form-integrand} in $\dimension = 4-2\varepsilon$ and $\EP_e = \EP_e' + \EPE_e$ around integers $\EP_e' \in \Z$ reads
\begin{equation}
	\label{eq:projective-index-expansion}%
	\begin{split}
	I_G
	&= I_G'
	\sum_{n,n_1,\ldots,n_{\edges(G)} \geq 0} \frac{\varepsilon^n \EPE_1^{n_1}\!\cdots\EPE_{\edges}^{n_{\edges}}}{n! n_1!\cdots n_{\edges}!} \ln^n \frac{\psipol_G^{1+\loops{G}}}{\phipol_G^{\loops{G}}}
	\prod_{e \in \edges(G)} \ln^{n_e} \frac{\psipol_G \SP_e}{\phipol_G}
	,\quad\text{for the glued}\hspace{-5mm} \\
	I_{\CloseProp{G}}
	&= I_{\CloseProp{G}}'
	\sum_{n,n_1,\ldots,n_{\edges(G)} \geq 0} \frac{\varepsilon^n \EPE_1^{n_1}\!\cdots\EPE_{\edges}^{n_{\edges}}}{n! n_1!\cdots n_{\edges}!} \ln^n \frac{\psipol_{\CloseProp{G}}}{\SP_{e_0}^{1+\loops{G}}}
		\prod_{e \in \edges(G)} \ln^{n_e} \frac{\SP_e}{\SP_{e_0}}
	\end{split}
\end{equation}
where $I_G'$ and $I_{\CloseProp{G}}'$ denote the integrands at the limits $\varepsilon=0$ and $\EP_e = \EP_e'$. Following section~\ref{sec:anareg} we may assume that $\int \Omega\ I_G' < \infty$ converges and integrate each coefficient in this expansion separately. They can be computed with hyperlogarithms if $G$ (equivalently $\CloseProp{G}$) is linearly reducible: Note that the logarithm powers do not influence the linear reducibility at all, since the graph polynomials are already considered in the reduction and the monomials $\log(\SP_e) = \Hyper{\letter{0}}(\SP_e)$ correspond to the letter $\letter{0}$, which is always assumed to be in the alphabet $\Sigma$ anyway.

Hence it suffices to consider the $3$-connected vacuum graphs with $4$ and $5$ loops (one more than the propagators have) shown in figures~\ref{fig:4loop-vacuum} and \ref{fig:5loop-vacuum}. The graphs $Y_3$, $\WS{4}$ and ${_5 P_1}$ to ${_5 P_6}$ have vertex-width three and fall under theorem~\ref{theorem:vw3-MZV}. In the remaining cases (recall that $\vw(G) \leq 3$ requires planarity; the cube $C = {_5P_7}$ is one of the forbidden minors in theorem~\ref{theorem:vw3-minors})  we checked linear reducibility through explicit computation of the polynomial reduction with {\HyperInt}. These contained some polynomials with different signs, which means that in the end we may get alternating sums. 
\begin{figure}\centering%
	\begin{tabular}{cccccc}
		$\Graph[0.25]{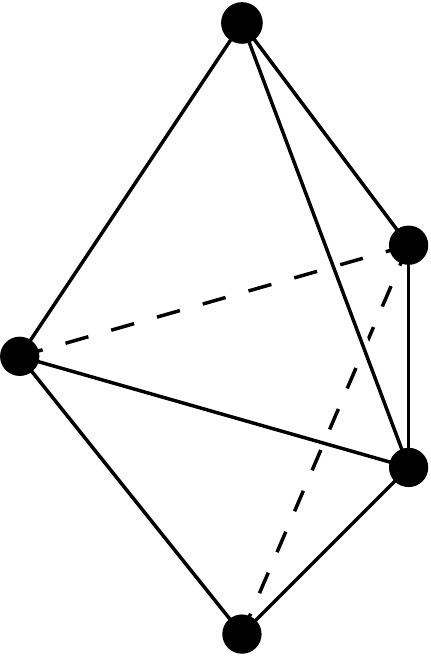}$ &
		$\Graph[0.3]{w5}$ &
		$\Graph[0.4]{zz5}$ &
		$\Graph[0.4]{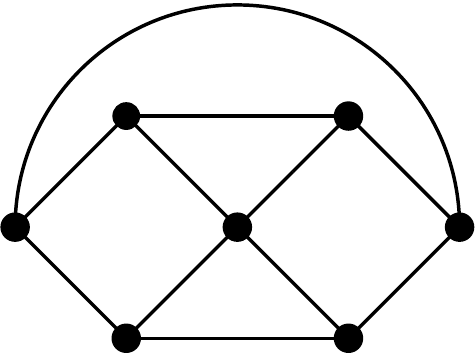}$ &
		$\Graph[0.4]{5P4}$ &
		$\Graph[0.4]{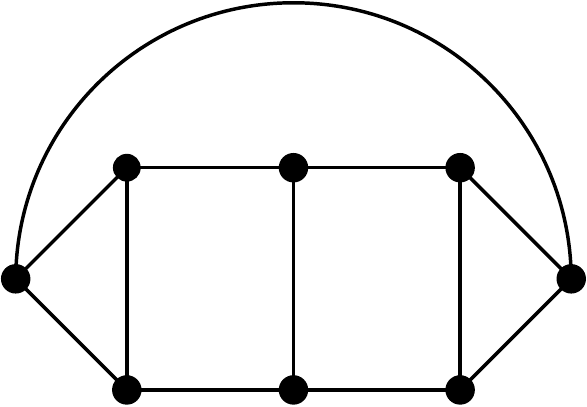}$ \\
		${_5 P_1}$ &  ${_5 P_2}$ & ${_5 P_3}$ & ${_5 P_4}$ & $ {_5 P_5}$ & ${_5 P_6}$\\
	\end{tabular}\\%
	\begin{tabular}{ccccc}
		$\Graph[0.37]{cube}$ &
		$\Graph[0.4]{5N}$ &
		$\Graph[0.4]{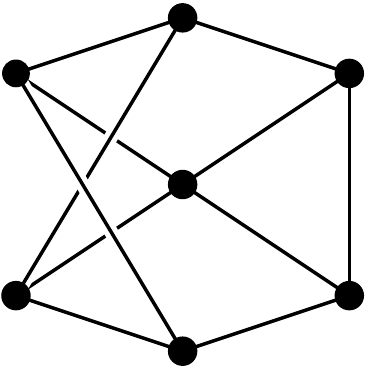}$ &
		$\Graph[0.35]{U}$ &
		$\Graph[0.38]{c8_4}$\\
		${_5 P_7}$ & ${_5 N_1}$ &$ {_5 N_2} $ & ${_5 N_3}$ & ${_5 N_4}$ \\
	\end{tabular}%
	\caption[All three-connected five-loop vacuum graphs]{%
	All three-connected five-loop vacuum graphs \cite{Panzer:MasslessPropagators}, divided into planar ($P$) and non-planar ($N$) ones. 
	The zigzag ${_5 P_3} = \ZZ{5}$ and ${_5 N_1}$ were considered in \cite{Brown:TwoPoint}. 
	Cutting any edge produces a propagator with four loops, deleting a three-valent vertex creates a three-loop three-point graph.}%
	\label{fig:5loop-vacuum} %
\end{figure}%
\begin{theorem}%
	\label{theorem:propagators}%
	All three- and four-loop massless propagators are linearly reducible. Each coefficient in their expansion of $I_G$ in $\varepsilon$ and the indices $\EP_e$ is a rational linear combination of multiple zeta values, except for the possibility that alternating sums might occur in the expansion of $N$ and the cuts of ${_5 N_1}, {_5 N_2}, {_5 N_3}, {_5 N_4}$ or ${_5 P_7}$.
\end{theorem}
Complete results for all $3$-loop graphs to order $\varepsilon^4$ are available digitally \cite{Panzer:MasslessPropagatorsData} and discussed in great detail in \cite{Panzer:MasslessPropagators}, including several $4$-loop propagators as well. Further data can be computed anytime with {\HyperInt}.

For illustration we still give two examples. If we expand near unit indices $\EP_e = 1 + \varepsilon \EPE_e$, the first terms of $N$ read
\begin{equation}\begin{split}
	\frac{\FR(N)}{\Gscheme^3(1-2\varepsilon)^2}
	&= 20 \mzv{5}
+\varepsilon \left\{
	\tfrac{80}{7} \mzv[3]{2}
	+ \mzv[2]{3} \left( 68+6 p_1 \right)
\right\}
\\&\quad
+\varepsilon^2 \left\{
	\tfrac{6}{5} \mzv[2]{2} \mzv{3} \left( 68+6 p_1 \right)
	+ \mzv{7} \left( 450 - 14 p_2 \right)
\right\}
+\bigo{\varepsilon^3}

	\label{eq:N-full}%
\end{split}\end{equation}
where the dependence on the indices is encoded into the polynomials
\begin{equation}\begin{split}
	p_1 &= 2\,\EPE_{1346}+3\,\EPE_{2578}
,\\
	p_2 &= 2 \left(
	\EPE_{2578}\EPE_{1346}
	+ \EPE_{3}\EPE_{457}
	+ \EPE_{4}\EPE_{238}
	+ \EPE_{6}\EPE_{127}
	+ \EPE_{1}\EPE_{568}
\right)
+3 (\EPE_{25}\EPE_{78}+\EPE_{16}\EPE_{34})
\\&\quad
+4 (\EPE_{7}\EPE_{8} + \EPE_{2}\EPE_{5})
+ ( 1-p_1 ) \EPE_{12345678}
-p_1 ( p_1 + 90 )/16 

	\label{eq:N-polynomials}%
\end{split}\end{equation}
with the abbreviation $\EPE_{1346} = \EPE_1 + \EPE_3 + \EPE_4 + \EPE_6$ and so on. The edge labellings are as shown in figure~\ref{fig:propagators-labelled} and setting $\EPE_e = 0$ for all edges reproduce \eqref{eq:N-expansion}, for brevity we do not give the (long) expressions for the coefficients of $\varepsilon^3$ and $\varepsilon^4$ here. A four-loop example with infrared subdivergences is (a cut of $\CloseProp{M}_{5,1} = {_5 N_2}$)
\begin{align}
	&\frac{\FR(M_{5,1})\cdot(1+\varepsilon[3+\EPE_{345678910}])(4+\EPE_{345678910})}{\Gscheme^4 (1-2\varepsilon)^3}
	=	-20 \mzv{5}\varepsilon^{-1}
-\tfrac{80}{7}\mzv[3]{2}
-\mzv[2]{3} \left( 68+6\,p_1 \right)
\nonumber\\&\quad
-\varepsilon \left\{
		\tfrac{1}{5} \mzv[2]{2}\mzv{3} \left( 408+36\,p_1 \right)
		+\mzv{7} \left( 170 - 7\,p_2 \right) 
\right\} 
+\bigo{\varepsilon^2}
,
	\label{eq:M51-full}%
\end{align}
where the polynomials $p_1,p_2 \in \Q[\EPE_1,\ldots,\EPE_{10}]$ are given by
\begin{equation}\begin{split}
	p_1 &=	2\,\EPE_{36810}+3\,\EPE_{4579}
 \qquad\text{and} \\
	p_2 &= 2 \left( \EPE_{8}-\EPE_{10} \right)  \left( \EPE_{4510}-\EPE_{789} \right)
+2 \left( \EPE_{3}-\EPE_{6} \right)  \left( \EPE_{567}-\EPE_{349} \right)
-\tfrac{5}{2} \EPE_{36810}
-\tfrac{55}{4} \EPE_{4579}
\\&\quad
+8 \EPE_{12}
-\tfrac{1}{8} p_1^2
+2 (\EPE_{12}\EPE_{345678910}
		+\EPE_{36}\EPE_{810}
		-\EPE_{47}\EPE_{59} )
-4 \left( \EPE_{4}^{2}
		+ \EPE_{5}^{2}
		+ \EPE_{7}^{2}
		+ \EPE_{9}^{2} \right)
.\hspace{-4mm}
	\label{eq:M51-polynomials}%
\end{split}\end{equation}
\begin{remark}
	In \cite{Panzer:MasslessPropagators} we were not yet aware of the general method of section~\ref{sec:anareg} for analytic regularization. Instead, we computed the periods of subdivergent graphs like $M_{5,1}$ with the help of carefully constructed counterterms.
\end{remark}
\begin{figure}
	\centering
	$N = \Graph[0.4]{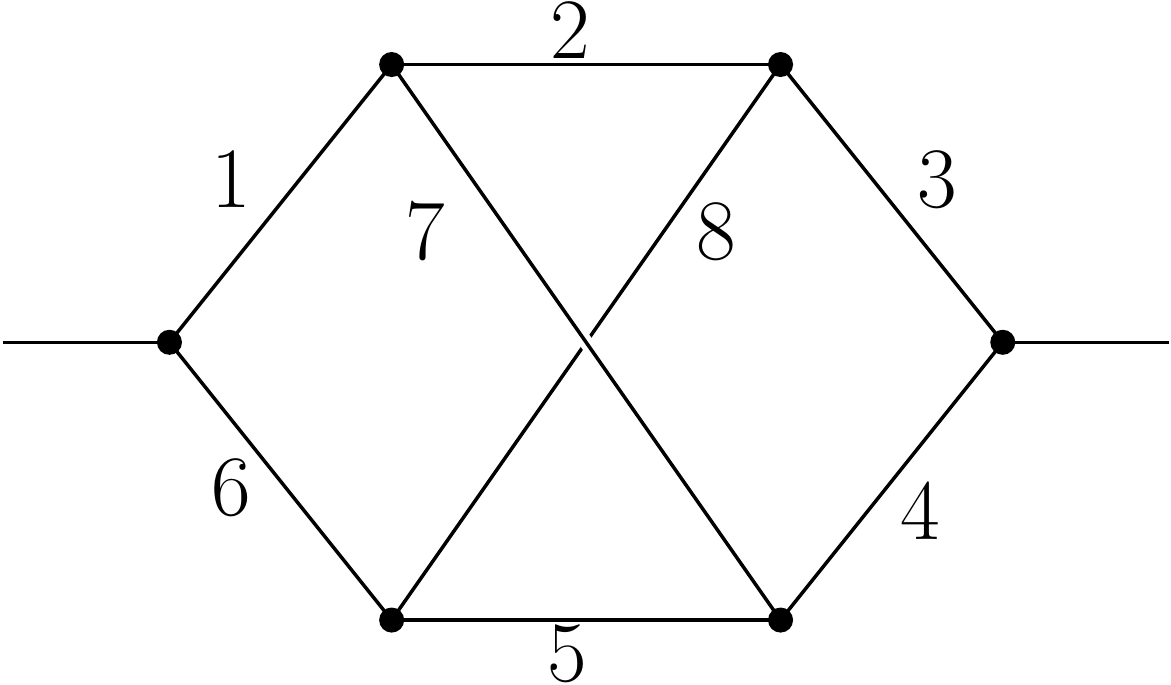}
		\qquad
	M_{5,1} = \Graph[0.4]{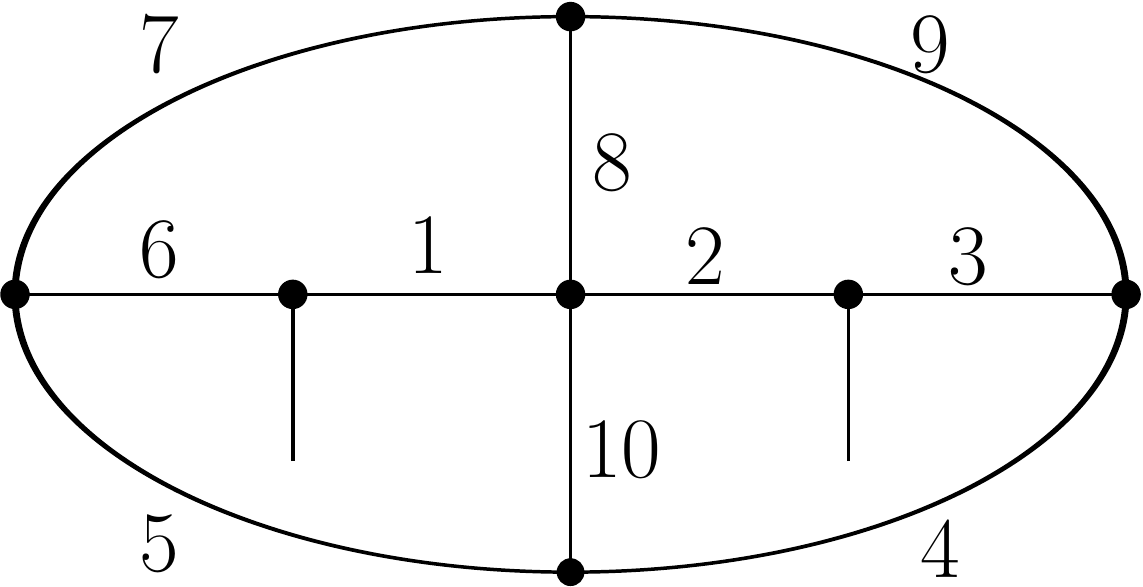}
	$
	\caption{The non-planar $3$-loop and a subdivergent $4$-loop propagator.}%
	\label{fig:propagators-labelled}%
\end{figure}

\subsection{Open questions}
\label{sec:propagator-questions}%
\subsubsection{Ramification of $4$-loop propagators}
The graphs $\ZZ{5} = {_5 P_3}$ and ${_5 N_1}$ are primitive $\phi^4$ vertices and were considered in \cite{Brown:TwoPoint} already, where it was proven that the expansion of $\ZZ{5}$ involves only multiple zeta values. The result for ${_5 N_1}$ assured linear reducibility, but the appearance of multiple polylogarithms at sixth roots of unity in its $\varepsilon$-expansion could not be excluded.
Such numbers had hitherto not appeared in massless $\phi^4$ periods and a numerical search for non-MZV in $\varepsilon$-expansions was undertaken in \cite{LeeSmirnov:FourLoopNonPlanarPropagators}. But no such number could be found, all evidence suggested that indeed multiple zeta values should suffice.

With compatibility graphs for polynomial reduction, we saw for the graphs above that indeed sixth roots of unity are spurious. But still we are left with the puzzle of the absence of alternating sums in the non-planar graphs and the cube $_5 P_7$: From the polynomial reduction we expect alternating sums, but all known results of their expansions are given by multiple zeta values.

The conceptual difference lies in the fact that we know that the polynomials involving negative signs (responsible for alternating sums in the result) are not spurious as in the previous case (sixth roots of unity in $_5 N_1$), because in the last integration step we do obtain harmonic polylogarithms using all three letters $\set{-1,0,1}$ (see also section~\ref{sec:P79}). Only when we reduce the total result as an alternating sum in a basis we observe the subtle cancellations needed to obtain just multiple zeta values.

Given that even the sixth roots of unity finally entered the scene (section~\ref{sec:P711}) of massless $\phi^4$ theory, we expect the same for alternating sums.\footnote{Very recently, Oliver Schnetz indeed identified alternating sums (that are not in $\MZV$) in periods of $8$-loop massless propagators.}
But for a particular graph, the most striking example being $N$ (with $3$ loops), we would like to understand if alternating sums can ever appear in its $\varepsilon$-expansion or if they always lie in the subspace of multiple zeta values.

A similar question concerns not the $\varepsilon$-expansion, but the periods for different integer powers $\EP_e\in \Z$. These correspond to convergent integrals $\int F\ \Omega$ for more general integrands $F \in \Q[\psipol_{\CloseProp{G}}^{-1},\SP_e\colon e \in \edges]$. Computations for $\CloseProp{G} = \WS{3}$ and $\CloseProp{G} = \WS{4}$ have shown that all such periods lie in $\Q \oplus \Q\mzv{3}$ and $\Q +\Q\mzv{3} + \Q \mzv{5}$, respectively. Remarkably, no even zeta values appeared so far.

One step to address this question is to generate a finite number of integrands whose expansions generate all periods under consideration, such that only these need to be computed. This work is in progress.

\subsubsection{Reducibility of $5$- and $6$-loop propagators}
Also already in \cite{Brown:TwoPoint}, the six-loop vertex graphs (carrying the periods of $5$-loop massless propagators) of $\phi^4$ theory had been analyzed with a similar separation into graphs with expansions provably in $\MZV$ and others known to be contained at least in $\MZV[6]$. It seems likely that these could be further constrained to alternating sums $\MZV[2]$ with explicit computation of the polynomial reduction.

However, more is needed for statements about a quantum field theory. Not only $\phi^4$ graphs occur, but also $3$-regular graphs which are ubiquitous in QED and part of QCD. Because linear reducibility is a minor closed property of graphs \cite{Brown:PeriodsFeynmanIntegrals,Lueders:LinearReduction}, it even suffices to only investigate the $3$-regular graphs at the loop order under consideration. It would be very interesting to carry out this analysis for $6$ loops, which is possible with the program we developed. Given that $3$-regular graphs with $6$ loops have $15$ edges and we successfully calculated $\phi^4$ periods at seven loops (where the graphs have $14$ edges), it even seems possible to practically compute such $5$-loop massless propagators with {\HyperInt}.

Even though the primitive periods are now known up to $7$ loops in $\phi^4$ theory, so far no $3$-regular graph at this loop order has been studied with respect to linear reducibility. It is therefore too early to speculate on all massless propagators with $6$ loops, but it is clear that a huge number of them can be computed (in the $\varepsilon$-expansion) with hyperlogarithms.

\section{Renormalized subdivergences}
\label{sec:ex-renormalized-parametric}%
Instead of calculating regularized integrals in the $\varepsilon$-expansion, physically meaningful (finite) quantities defined by renormalization can be computed directly. The forest formula \eqref{eq:feynman-renormalized-projective} provides a convergent integral representation (in the absence of infrared divergences) which can be computed without the need of any regularization.

We want to take advantage of this feature and evaluate renormalized integrals with subdivergences directly, using our tools for hyperlogarithms. We give two examples to show the feasibility of this approach:
\begin{itemize}
	\item One-scale graphs with many disjoint subdivergences.
	\item Cocommutative graphs.
\end{itemize}
In the first case we reproduce a known result, but without the use of dimensional regularization. Instead we give a complete and explicit calculation of the renormalized integrals in the parametric integration, using classical polylogarithms only.

The second subsection on cocommutative graphs shows some new results and is of particular interest because such graphs contribute renormalization point independent periods.

\subsection{Bubble chains}\label{sec:bubble-chains}
\begin{figure}\centering
	$
		\BBr{n}{m}
		\defas
		\Graph[0.4]{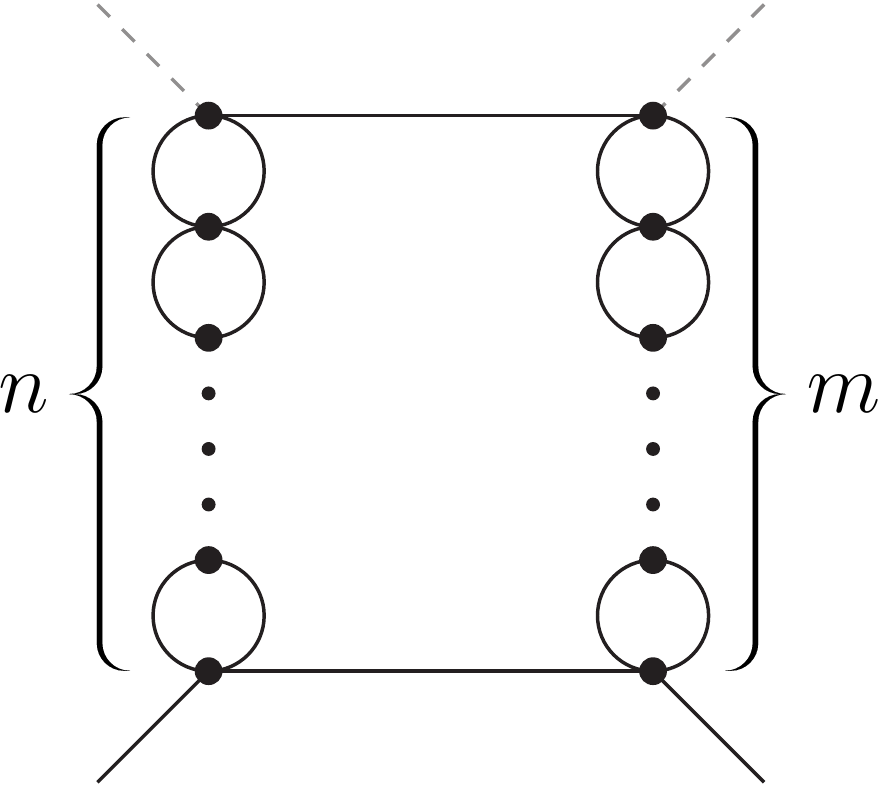}
		\qquad
		\BBt{n}{m}
		\defas
		\Graph[0.4]{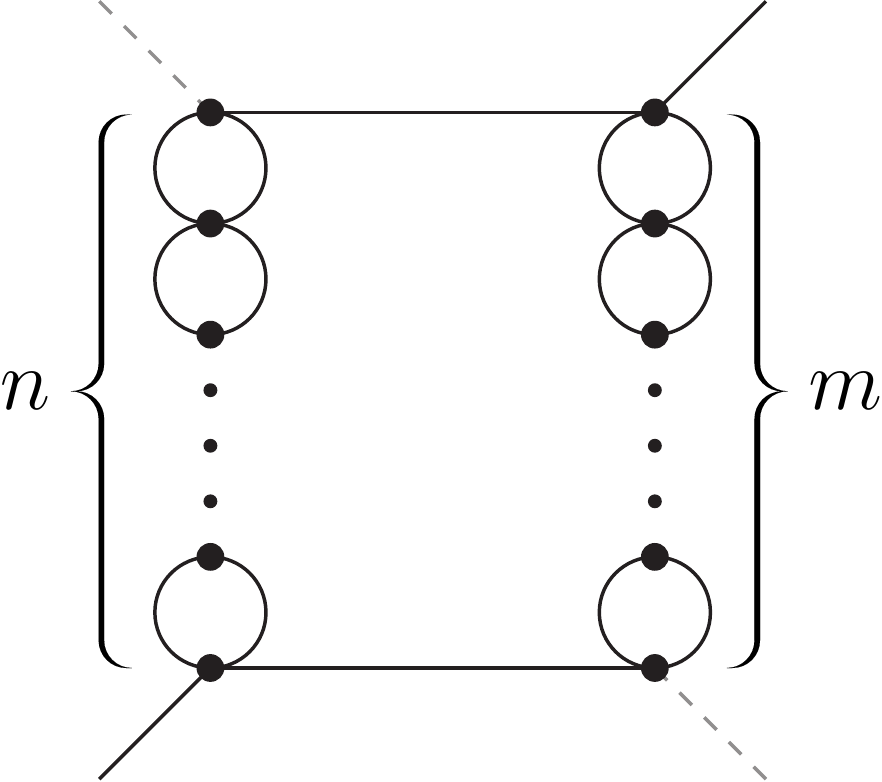}
		$\vspace{-2mm}
	\caption[Two series $\BBr{n}{m}$ and $\BBt{n}{m}$ of one-scale graphs with subdivergences]{Two series of one-scale graphs with subdivergences in four dimensions. These graphs arise as vertex graphs in $\phi^4$-theory upon nullification of two external momenta (dashed), incident to the two three-valent vertices.}%
	\label{fig:bubble-graphs-def}%
\end{figure}%
Figure \ref{fig:bubble-graphs-def} defines two families $\BBr{n}{m}$ and $\BBt{n}{m}$ ($n,m \in \N_0$) of massless, one-scale $\fieldphi^4$ vertex graphs in $\phi^4$-theory. These are logarithmically divergent in $\dimension=4$ dimensions, but they contain a series of bubbles $\gamma_i \isomorph \smallbubble$ as subdivergences. We denote $\gamma_I \defas \prod_{i \in I} \gamma_i$ for the (edge-disjoint) union of subdivergences indexed by a set $I$. The coproduct is
\begin{equation}
	\cop \BBr{n}{m}
	= \BBr{n}{m} \tp \1
	+ \sum_{I \subseteq [n]} \sum_{J \subseteq [m]} \gamma_I \gamma_J \tp \BBr{n}{m}/\left( \gamma_I \gamma_J \right)
	\label{eq:bubble-graphs-coprod}%
\end{equation}
where $[n] \defas \set{1,\ldots,n}$, we index bubbles in the left row with $I$ and on the right with $J$. Note that $\BBr{n}{m}/\left( \gamma_I \gamma_J \right) \isomorph \BBr{n-\abs{I}}{m-\abs{J}}$. The same formulas hold for $\BBt{n}{m}$ as well since these two families of graphs differ only by the choice of which of the four external momenta are nullified. One checks that $\BBr{n}{m}$ and $\BBr{n+m}{0} = \BBt{n+m}{0} = \BBt{0}{n+m}$ define identical Feynman integrals, so it suffices to compute $\BBt{n}{m}$.

Since $\BBr{n}{m}$ is not cocommutative for $n+m>1$, the associated renormalized Feynman rules depend on the renormalization scheme. To point this out we will rather think of $\BBr{n}{m}$ and $\BBt{n}{m}$ as the same graph, but with different renormalization schemes applied to them. The computation of their periods is elementary in dimensional regularization.
\begin{lemma}
	\label{lemma:bubble-graphs-transcendental-period}%
	The periods of $\BBt{n}{m}$ are given by the exponential generating function
	\begin{equation}
		\sum_{n,m\geq 0} \frac{x^n y^m}{n! m!} \period\left( \BBt{n}{m} \right)
		=	\frac{\exp \left\{
					-2 \sum_{r\geq 1} \frac{\mzv{2r+1}}{2r+1}
					\Big[ (x+y)^{2r+1} - x^{2r+1} - y^{2r+1} \Big]
				\right\}
				}{
					1-x-y
				}
		.%
		\label{eq:bubble-graphs-transcendental-period}%
	\end{equation}%
\end{lemma}%
\begin{proof}
	In $\dimension=4-2\varepsilon$ dimensions, repeated application of the one-loop master formula \eqref{eq:oneloop-master} evaluates the unrenormalized Feynman rules to
	\begin{equation*}
		\FR \left( \BBt{n}{m} \right)
		= \left[ \onemaster{1}{1}\right]^{n+m} \FR \left( \Graph[0.2]{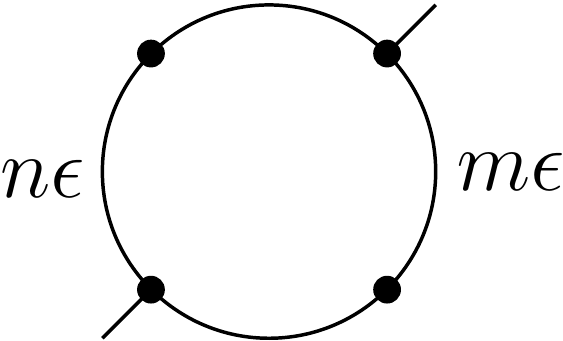} \right)
		= q^{-2(n+m+1)\varepsilon} \left[ \onemaster{1}{1}\right]^{n+m} \onemaster{1+n\varepsilon}{1+m\varepsilon},
	\end{equation*}
	in terms of the external momentum $q$. Now we renormalize by subtraction at $s\defas q^2 \mapsto 1$, so the counterterm of any bubble is just
	$
		\FR_- ( \gamma_i )
		= -\restrict{\FR}{s=1} ( \gamma_i)
		= -\onemaster{1}{1}
	$ and the multiplicativity of $\FR_-$ gives
	$
		\FR_-(\gamma_I \gamma_J)
		= [-\onemaster{1}{1} ]^{\abs{I} + \abs{J}}
	$. With the coproduct \eqref{eq:bubble-graphs-coprod}, the period \eqref{eq:def-period} becomes (in $\dimension = 4$)
	\begin{align*}
		\period\left( \BBt{n}{m} \right)
		&= \lim_{\epsilon\rightarrow 0}
		\big[ \onemaster{1}{1} \big]^{n+m}
		\sum_{I \subseteq [n]} \sum_{J \subseteq [m]}\!\!
		(-1)^{\abs{I} + \abs{J}}
		\cdot \varepsilon(1+\abs{I^c}+\abs{J^c})
		\onemaster{1+\abs{I^c}\varepsilon}{1+\abs{J^c}\varepsilon} \\
		&= \lim_{\varepsilon\rightarrow 0} \varepsilon^{-n-m}
		\sum_{i=0}^n \binom{n}{i} (-1)^{n+i} \sum_{j=0}^m  \binom{m}{j} (-1)^{m+j}
		\cdot
		f(\varepsilon i, \varepsilon j, \varepsilon),\tag{$\ast$}
	\end{align*}
	where we exploited $\lim_{\varepsilon\rightarrow 0} \left[ \varepsilon \onemaster{1}{1} \right] = 1$ and introduced the power series
	\begin{equation*}
		f(x,y,\varepsilon)
		\defas
		\frac{\Gamma(1-x-\varepsilon)\Gamma(1-y-\varepsilon)\Gamma(1+x+y+\varepsilon)}{\Gamma(1+x)\Gamma(1+y)\Gamma(2-x-y-2\varepsilon)}
		= \sum_{\mathclap{\nu,\mu,k\geq 0}} a_{\nu,\mu,k}\, x^{\nu} y^{\mu} \varepsilon^k
		\in \R [[x,y,\varepsilon]].
	\end{equation*}
	The sums over $i$ and $j$ in $(\ast)$ annihilate any term with $\nu<n$ or $\mu<m$ because
	\begin{equation*}
		\sum_{i=0}^{n} \binom{n}{i} (-1)^{n+i} \cdot i^k 
		= \begin{cases}
			0 &\text{whenever $k< n$ and} \\
			n! &\text{for $k=n$.}\\
		\end{cases}
	\end{equation*}
	But when $\mu+\nu+k > m+n$, then
	$
		\lim_{\varepsilon\rightarrow 0} \varepsilon^{-n-m} \cdot (i\varepsilon)^{\nu} (j\varepsilon)^{\mu} \varepsilon^k
		= 0
	$ vanishes as well, so the only contribution to $(\ast)$ left over is
	\begin{equation*}
		\period \left( \BBt{n}{m} \right)
		=
		a_{n,m,0} 
		\sum_{i=0}^{n} \binom{n}{i} (-1)^{n+i} \cdot i^{n}
		\sum_{j=0}^{m} \binom{m}{j} (-1)^{m+j} \cdot j^{m}
		= n! m! a_{n,m,0}.
	\end{equation*}
	To finish, expand $\Gamma(1-x) = \exp\big[ \gamma x + \sum_{n\geq 2} \mzv{n} x^n / n \big]$ in $f(x,y,0)$.
\end{proof}
The two different renormalization schemes give very different periods indeed: All
\begin{equation}
	\period\left( \BBr{n}{m} \right)
	= \period\left( \BBr{n+m}{0} \right)
	= \period\left( \BBt{n+m}{0} \right)
	= \restrict{\partial_x^{n+m} (1-x)^{-1}}{x=0}
	=	(n+m)!
	\label{eq:bubble-graphs-rational-period}%
\end{equation}
are integers, while the periods of $\BBt{n}{m}$ involve Riemann zeta values.
\begin{example}
	$\period( \BBt{1}{1} ) = 2$ is still rational, but for all other $n,m\geq 1$ we find zeta values like in
	$		\period( \BBt{1}{2} ) = 6 - 4\mzv{3}
	$. The values for $n+m \leq 6$ are:
	\begin{align*}
		\period( \BBt{1}{3}) &= 24 - 12\mzv{3}
		&
		\period( \BBt{1}{4}) &= 120 - 48\left( \mzv{3} +\mzv{5}\right)
		&
		\period( \BBt{1}{5})& = 720 -240\left(\mzv{3} + \mzv{5} \right)
		\\
		\period( \BBt{2}{2}) &= 24 - 16\mzv{3}
		&
		\period( \BBt{2}{3}) &= 120 - 72\mzv{3} - 48\mzv{5}
		&
		\period( \BBt{2}{4})& = 720 - 384\mzv{3} - 288\mzv{5} + 96\mzv[2]{3}
		\\
		& & & &
		\period( \BBt{3}{3} )& = 720 - 432\mzv{3} - 288\mzv{5} + 144\mzv[2]{3}
	\end{align*}
\end{example}
We do not want to discuss these particular numbers any further, but only remark that $\period(\BBt{n}{m})$ only contains products of at most $\min\set{n,m}$ zeta values and has integer coefficients.
\begin{lemma}
	The periods
	$
		\period\big( \BBt{n}{m} \big)
		\in \Z\left[\setexp{2(2r)!\mzv{2r+1}}{r\in \N}\right]
	$ are integer combinations of odd zeta values of weight at most $\leq n+m$.
\end{lemma}
\begin{proof}
	Expand the binomial $(x+y)^{2r+1}$ to rewrite the exponent of \eqref{eq:bubble-graphs-transcendental-period} as
	\begin{equation*}
		F(x,y) \defas
				-2 \sum_{r\geq 1} (2r)!\mzv{2r+1}
				\sum_{i=1}^{2r}
				\frac{x^i y^{2r+1-i}}{i!(2r+1-i)!}.
	\end{equation*}
	Its derivatives $\partial_x^n \partial_y^m \restrict{F(x,y)}{x=y=0} = -2\mzv{n+m}(n+m-1)!$ are integer combinations of odd zeta values. This property is passed on to the exponential
	\begin{equation*}
		\restrict{
			\partial_x^n \partial_y^m \exp(F)
		}{x=y=0}
		= \restrict{
			\Big[ 
				( \partial_x F ) + \partial_x
			\Big]^n
			\Big[
				( \partial_y F ) + \partial_y
			\Big]^m
		}{x=y=0}
		\in \Z\left[\setexp{(2r)!\,2\mzv{2r+1}}{r \in \N} \right]
	\end{equation*}
	via the identity $\partial_x \exp(F) = \exp(F) \big[(\partial_x F) + \partial_x \big]$ of differential operators. Finally it also extends to the product with $(1-x-y)^{-1}$ by Leibniz' rule and $\partial_x^n \partial_y^m\restrict{ (1-x-y)^{-1}}{x=y=0} = (n+m)!$.
\end{proof}

\subsubsection{Parametric integration}
\label{sec:bubbles-parametric}%
Here we demonstrate how the periods $\period(\BBr{n}{m}) = (n+m)!$ may be computed with hyperlogarithms\footnote{In fact, our choice of variables allows us to employ only classical polylogarithms of a single variable.} in the parametric representation. 
Of course we already know the result and the above calculation in dimensional regularization might seem a lot simpler (in particular to a physicist familiar with dimensional regularization), but the point we want to make is that such a calculation is indeed possible without any regulator, even when many subdivergences are present. For a new result obtained this way, see section~\ref{sec:ex-cocommutative}.
\begin{lemma}
	\label{lemma:bubble-graph-rational-parametric}%
	In the parametric representation, the period of $\BBr{n}{0}$ can be reduced to a projective integral over $n$ variables $x_1,\ldots,x_n \in \R_+$ of the form
	\begin{equation}
		\period\left( \BBr{n}{0} \right)
		= \int \frac{\Omega}{x_1 \!\cdots x_n}
			\sum_{\emptyset\neq I \subseteq [n]} \!\!\! (-1)^{I} \frac{\Li_1( - z_I)}{z_I},
		\quad\text{where}\quad
		z_I \defas \frac{x_{I^c}}{x_I}.
		\label{eq:bubble-graph-rational-parametric}%
	\end{equation}
	For any subset $I\subseteq [n] \defas \set{1,\ldots,n}$ we abbreviate $x_I \defas \sum_{i\in I} x_i$ and $x_{I^c} = \sum_{i\notin I} x_i$. Note that the summand with $I=[n]$ gives $z_I=0$, its contribution is understood as $(-1)^n \lim_{z \rightarrow 0} \Li_1(-z)/z = (-1)^{n+1}$.
	Recall that $\Omega = \delta(1-\sum_{i=1}^n \lambda_i x_i) \bigwedge_{i=1}^n \dd x_i$ for arbitrary $\lambda_1,\ldots,\lambda_n \geq 0$ that do not all vanish.
\end{lemma}
\begin{figure}\centering%
$		\gamma_i = \Graph[0.5]{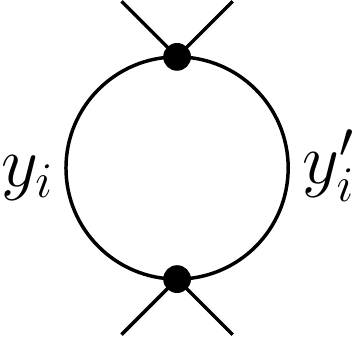}
		\qquad
		\Gamma_{\!I^c} = \Graph[0.45]{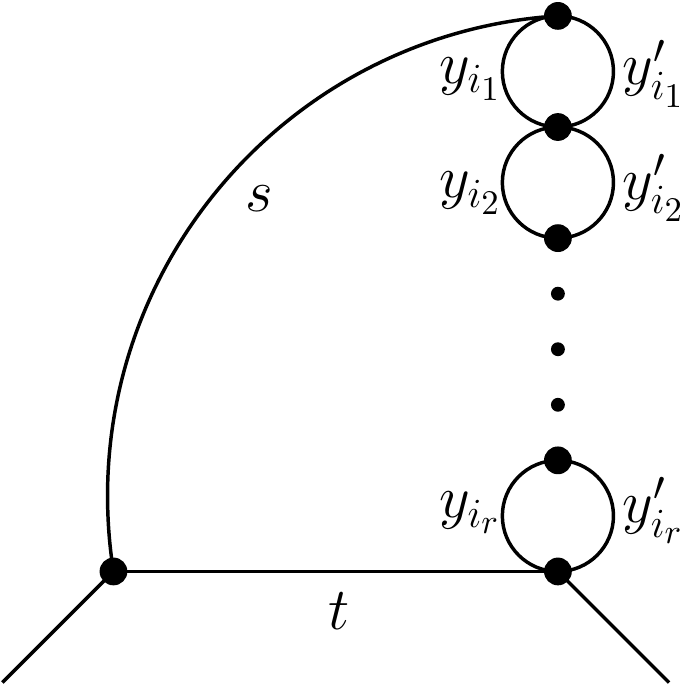}
		\qquad
		G_{\!I^c} = \Graph[0.45]{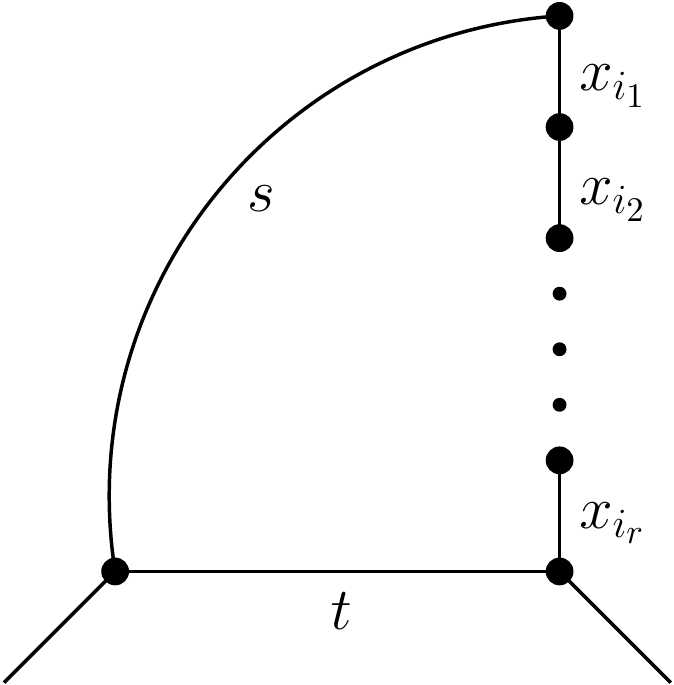}
$%
\caption[Sub- and quotient graphs of $\BBr{n}{0}$]{Subgraphs $\gamma_i$ and the quotient $\Gamma_{\!I^c} =\BBr{n}{0}/\prod_{i\in I} \gamma_i$ for $I^c = \set{i_1,\ldots,i_r}$, which becomes $G_{\!I^c}$ after reducing the parallel edge pairs $\set{y_i,y_i'}$ to a single edge.}%
	\label{fig:bubble-subcographs}%
\end{figure}%
\begin{proof}
	Let $\Gamma \defas \BBr{n}{0}$ and $\gamma_i$ denote the bubble-subgraph consisting of edges $y_i$ and $z_i$ as labelled in figure \ref{fig:bubble-subcographs}.
	The forest formula \eqref{eq:period-projective} for the period delivers
	\begin{equation*}
		\period\left( \Gamma \right)
		= \int \Omega_{\Gamma} \left\{ 
			\frac{1}{\psipol_{\Gamma}^2} + \sum_{\emptyset\neq I \subseteq [n]} (-1)^{I}
			\frac{
					\phipol_{\Gamma_{\!I^c}}/\psipol_{\Gamma_{\!I^c}}
			}{
					\psipol_{\gamma_I}^2 \psipol_{\Gamma_{\!I^c}}^2 
					\left[ 
						\phipol_{\Gamma_{\!I^c}}/\psipol_{\Gamma_{\!I^c}} 
						+ \sum_{i \in I}
							\phipol_{\gamma_i}/\psipol_{\gamma_i} 
					\right]
			}
		\right\},
	\end{equation*}
	where $\gamma_I \defas \prod_{i\in I} \gamma_i$ runs over the subdivergences and $\Gamma_{\!I^c} \defas \Gamma / \gamma_I \isomorph \BBr{n-\abs{I}}{0}$ is a shorthand for the corresponding cograph. Since a pair of parallel edges $y_i$ and $y'_i$ can not be contained in any spanning tree or forest, the graph polynomial
	\begin{equation*}
		\psipol_{\Gamma_{\!I^c}}(s,t,y,y')
		=	\psipol_{G_{\!I^c}}(s,t,x)
			\prod_{i \in I^c} (y_i+y'_i)
		\quad\text{and equally}\quad
		\phipol_{\Gamma_{\!I^c}}
		=	\phipol_{G_{\!I^c}}
			\prod_{i \in I^c} (y_i+y'_i)
	\end{equation*}
	can be expressed in terms of the graph $G_{I^c}$ of figure~\ref{fig:bubble-subcographs} where each pair $\set{y_i,y'_i}$ is replaced by a single edge, when we set $x_i = \frac{y_i y'_i}{y_i + y'_i}$. In particular,
	$	\phipol_{\Gamma_{\!I^c}}/\psipol_{\Gamma_{\!I^c}}
	= \phipol_{G_{\!I^c}}/\psipol_{G_{\!I^c}}
	$ depends only on $x$ (not individually on $y$ and $y'$) as does $\phipol_{\gamma_i} / \psipol_{\gamma_i} = x_i$ from
	\begin{equation*}
		\psipol_{\gamma_i} 
		= y_i + y'_i
		\quad\text{and}\quad
		\phipol_{\gamma_i}
		= y_i  y'_i.
	\end{equation*}
	The dependence of the integrand for $\period(\Gamma)$ above on $y$ and $y'$ is thus only through the prefactor $\prod_{i=1}^n (y_i + y_i')^{-2}$ and we can integrate them out using\footnote{In this step we choose the constraint $\delta(1-s)$ in $\Omega$ such that it does not depend of $y$ and $y'$.}
	\begin{equation*}
		\prod_{i=1}^n
		\int
		\frac{\dd y_i \,\dd y_i'}{(y_i+y_i')^2} \delta\left(x_i - \frac{y_i y_i'}{y_i+y'_i} \right)
		= \prod_{i=1}^n
			\frac{\dd x_i}{x_i}.
	\end{equation*}
	Together with $\psipol_{G_{\!I^c}} = s+t+x_{I^c}$ and $\phipol_{G_{\!I^c}} = t (s + x_{I^c})$, we have expressed $\period(\Gamma)$ as
	\begin{equation*}
		\int \!\widehat{\Omega}
			\,\Bigg\{ 
				\frac{1}{\psipol_{G_{\!\emptyset}}^2}
				+ \!\!\!\!\!\sum_{\emptyset \neq I \subseteq [n]}\!
					\frac{(-1)^I
							\phipol_{G_{\!I^c}}
					}{
						\psipol_{G_{\!I^c}}^2 \big[ 
								x_I \psipol_{G_{\!I^c}} + \phipol_{G_{\!I^c}}
						\big]
					}
			\Bigg\} \!
		= \!\!\int \!\!
			\sum_{I \subseteq [n]}
			\frac{\widehat{\Omega}\,
					(-1)^I t (s+x_{I^c})
			}{
				(s+t+x_{I^c})^2 \big[ 
					t (s+C) + x_I (s+x_{I^c})
				\big]
			}%
	\end{equation*}
	where $\widehat{\Omega} \defas \dd t \wedge \dd s \wedge \Omega/ \prod_i x_i$ and $C \defas x_{[n]}$. The integral over $t$ is elementary and gives
	\begin{equation*}
		\period(\Gamma)
		=
		\int \frac{\dd s \wedge \Omega}{x_1\!\cdots x_n}
		\Bigg\{ 
			\frac{1}{s+C}
			+ \sum_{\emptyset \neq I \subseteq [n]}
				\frac{(-1)^I}{s+x_{I^c}} \left[
					1 + \frac{x_I}{s+x_{I^c}}\log \left( \frac{x_I}{s + C} \right)
				\right]
		\Bigg\},
	\end{equation*}
	such that the integral over $s$ becomes elementary as well (with integration by parts) and proves the claim. Note that $\Li_1(-z_I)/z_I = x_I/x_{I^c} \cdot \log (x_I/C)$.
\end{proof}
\begin{lemma}
	For any $p\in\N$ and $z\in\C$ with $\Realteil(z)>-1$, the integral%
\footnote{%
The $\Li_p(z)$'s in the integrand are well-defined as the analytic continuation of \eqref{eq:def-Mpl} along the straight path from $0$ to $z$, since this never hits the singularity at $z=1$.
}%
	\begin{equation}
		f_p(z)
		\defas
		\int_{0}^{\infty}
		\left[ 
		\left( \frac{1}{x}-\frac{1}{x+z} \right) \Li_p(-x-z)
		-\frac{1}{x} \Li_p\left( -\frac{z}{x+1} \right)
		\right]
		\dd x
		\label{eq:bubble-integral-helper-function-def}
	\end{equation}
	converges absolutely and evaluates to $f_p(z) = p\Li_{p+1}(-z)$.%
	\label{lemma:bubble-integral-helper-function}%
\end{lemma}
\begin{proof}
	Taylor expanding $\Li_p(-x-z) = \Li_p(-z) + \frac{x}{z}\Li_{p-1}(-z) + \bigo{x^2}$ and $\Li_p\left( -\frac{z}{x+1} \right) = \Li_{p}(-z) - x\Li_{p-1}(-z) + \bigo{x^2}$ with \eqref{eq:Mpl-single-variable-differential} reveals the analyticity of the integrand at $x\rightarrow 0$. When $x\rightarrow\infty$, $\Li_p\left( -\frac{z}{x+1} \right) = \frac{z}{x^2} + \bigo{x^{-3}}$ is holomorphic and integrable. 
	Convergence of \eqref{eq:bubble-integral-helper-function-def} then follows from $\frac{1}{x}-\frac{1}{x+z} = \frac{z}{x^2} + \bigo{x^{-3}}$ since $\Li_p(-x-z)$ diverges at $x\rightarrow\infty$ only logarithmically.

	Due to absolute convergence we may interchange integration and differentiation to obtain, inductively, that
	\begin{align*}
		\partial_z f_p(z)
		&= \int_0^{\infty}  \left[ 
		\frac{\Li_{p}(-x-z)}{(x+z)^2} + \left( \frac{1}{x} - \frac{1}{x+z} \right) \frac{\Li_{p-1}(-x-z)}{x+z} - \frac{1}{xz} \Li_{p-1}\left( -\frac{z}{x+1} \right)
		\right]\dd x
		\\
		&= \frac{1}{z}f_{p-1}(z)
		+\int_0^{\infty} \frac{\Li_{p}(-x-z)-\Li_{p-1}(-x-z)}{(x+z)^2} \dd x
	 	\\
		&= \frac{1}{z}(p-1) \Li_{p}(-z) - \restrict{\frac{\Li_{p}(-x-z)}{x+z}}{x=0}^{\infty}
		= p \frac{\Li_{p}(-z)}{z}.
	\end{align*}
	Therefore $f_p(z) = p \Li_{p+1}(-z)$ with constant of integration $\lim_{z\rightarrow 0} f_p(z) = 0$, because we can take this limit on the integrand (which becomes zero).
\end{proof}
\begin{lemma}
	For any $2 \leq n \in \N$, $p \in \N$ and $x_1,\ldots,x_{n-1} > 0$, the integral
		\begin{equation}
			\int_0^{\infty}
				\frac{\dd x_n}{x_n} 
				\sum_{\emptyset \neq I \subseteq [n]} 
				\!\!\!\frac{(-1)^I\Li_p(-z_I)}{z_I}
				= \sum_{ \emptyset \neq I \subseteq [n-1]} 
					\!\!\!\frac{(-1)^I}{z_I}
					\left[ 
						p \Li_{p+1}(-z_I)
						+
						\sum_{k=1}^p \Li_k(-z_I)
					\right]
			\label{eq:bubble-lip-recursion}%
		\end{equation}
		is absolutely convergent. On the right-hand side, we have set $z_I = x_{I^c}/x_I$ for $I^c \defas [n-1] \setminus I$ (while on the left, $I^c = [n] \setminus I$).
\end{lemma}
\begin{proof}
	We write $x=x_n$, $C = x_1 + \cdots + x_{n-1}$ and collect the different summands according to whether $n \in I$ or not. After adding the zero $ -1/(x+C) \sum_{k=0}^{n-1} \binom{n-1}{k} (-1)^k$, the left-hand side of \eqref{eq:bubble-lip-recursion} becomes
	\begin{multline*}
		\sum_{\emptyset \neq I \subsetneq [n-1]} (-1)^I \int_0^{\infty} \frac{\dd x}{x}
		\left[ 
			\frac{x_I}{x + x_{I^c}} \Li_p\left( -\frac{x+x_{I^c}}{x_I} \right)
			- \frac{x+x_I}{x_{I^c}} \Li_p \left( -\frac{x_{I^c}}{x+x_I} \right)
			- \frac{x}{x+C}
		\right]
		\\
		+ \int_0^{\infty} \dd x \left\{ 
				\frac{(-1)^{n+1}}{x}
				- \frac{1}{C} \Li_p\left( -\frac{C}{x} \right) 
				+ (-1)^{n-1} \frac{C}{x^2} \Li_p \left(-\frac{x}{C} \right)
				-\frac{1+(-1)^{n-1}}{x+C}
			\right\},
	\end{multline*}
	where now $I^c \defas [n-1] \setminus I$. We will now see that the individual integrals are convergent and compute them separately. First we check with \eqref{eq:Li-iterated-integral} that the last one integrates to
	\begin{equation*}
		\left[
			(-1)^n \frac{C}{x} \sum_{k=1}^p \Li_k\left( -\frac{x}{C} \right)
			-\frac{x}{C} \sum_{k=1}^p \Li_k\left( - \frac{C}{x} \right)
		\right]_{\mathrlap{x\rightarrow 0}}^{\mathrlap{\infty}}
		= \big[- 1 - (-1)^n \big] \lim_{z \rightarrow 0} \sum_{k=1}^{p} \frac{\Li_k(-z)}{z}
		= p\big[ 1 + (-1)^n \big]
	\end{equation*}
	and substitute $x \mapsto x \cdot x_I$ in the remaining integrals to rewrite them as
	\begin{align*}
		&
		\int_0^{\infty} \frac{\dd x}{x} \left[ 
			\frac{\Li_p(-x-z_I)}{x+z_I}
			-\frac{x+1}{z_I} \Li_p\left( -\frac{z_I}{x+1} \right)
			-\frac{x}{x+1+z_I}
		\right]
		\\
		&= \int_0^{\infty} \!\!\!\dd x\left[ 
				\frac{\Li_p(-x-z_I)}{x (x+z_I)}
					-\frac{1}{x z_I} \Li_p\left( -\frac{z_I}{x+1} \right)
			\right]
			-\int_0^{\infty} \frac{\dd x}{z_I} \left[ 
					\Li_p\left( -\frac{z_I}{x+1} \right) 
					+\frac{z_I}{x+1+z_I}
			\right].
	\end{align*}
	The first term is just $f_p(z_I)/z_I$ from \eqref{eq:bubble-integral-helper-function-def}, the second evaluates to
	\begin{equation*}
		- \int_{1/z_I}^{\infty} \!\!\dd x
			\left[ \Li_p\left(-\frac{1}{x}\right) - \Li_0\left( -\frac{1}{x} \right)\right]
		= - \left[
					x\sum_{k=1}^{p} \Li_k\left( -\frac{1}{x} \right)
				\right]_{x \rightarrow 1/z_I}^{\infty}
		= p	+ \sum_{k=1}^p \frac{\Li_k(-z_I)}{z_I}.
	\end{equation*}
	To finish the proof we only need to add up all contributions and note that
	\begin{equation*}
		p \big[1 + (-1)^n \big] + \sum_{\mathclap{\emptyset \neq I \subsetneq [n-1]}} (-1)^I p
		= 2p (-1)^n
		= \lim_{z \rightarrow 0} \frac{(-1)^n}{z} \left[ 
				p \Li_{p+1}(-z) + \sum_{k=1}^p \Li_{k}(-z)
			\right]
	\end{equation*}
	corresponds to the term with $I = [n-1]$ on the right-hand side of \eqref{eq:bubble-lip-recursion} in our short-hand convention.
\end{proof}
\begin{corollary}
	For any $n, p \in \N$ we compute the following projective integrals (over positive variables $x_1$, \ldots, $x_n$), which generalize \eqref{eq:bubble-graph-rational-parametric}:
	\begin{equation}
		\int \frac{\Omega}{x_1\!\cdots x_n}
		\sum_{\emptyset \neq I \subseteq [n]} \!\!\!
					(-1)^I \frac{\Li_p(-z_I)}{z_I}
		= n! \binom{p+n-2}{p-1}.
		\label{eq:generalized-bubble-integral}%
	\end{equation}
	In particular, the case $p=1$ implies $\period(\BBr{n}{0}) = n!$ using lemma~\ref{lemma:bubble-graph-rational-parametric}.
\end{corollary}
\begin{proof}
	We perform an induction over $n$: For $n=1$, the integrand is just $1/x_1$ by our convention and the projective integral tells us to evaluate at $x_1 = 1$. So indeed the left-hand side gives $1 = 1! \binom{p-1}{p-1}$ for all $p$. When $n>1$, we use \eqref{eq:bubble-lip-recursion} to integrate out $x_n$ and obtain, using the statement for smaller values of $n$, 
	\begin{align*}
		&\int \frac{\Omega}{x_1\!\cdots x_{n-1}}
		\sum_{\emptyset \neq I \subseteq [n-1]} \!\!\! (-1)^I
			\left[ 
				\frac{p\Li_{p+1}(-z_I)}{z_I}
				+
				\sum_{k=1}^{p} \frac{\Li_k(-z_I)}{z_I}
			\right]
		\\
		&= (n-1)! \left[ 
				p \binom{p+n-2}{p}
				+\sum_{k=1}^p \binom{k + n-3}{k-1}
			\right]
		=
		(n-1)! \big[ 
		(n-1) + 1
			\big]
		 \binom{p+n-2}{p-1}.
		\qedhere
	\end{align*}
\end{proof}
Note that the parametric calculation involves polylogarithms of weight up to $n$, even though the final result is rational.\footnote{We wonder if polylogarithms could be avoided altogether in this case.} We used \eqref{eq:generalized-bubble-integral} as a test for our implementation {\HyperInt}. Furthermore we used the explicit result \eqref{eq:bubble-graphs-transcendental-period} for $\period(\BBt{n}{m})$ to check the program on
\begin{lemma}\label{lemma:BBt-projective}
	The parametric representation for the period of $\BBt{n}{m}$ can be reduced to a projective integral over variables $x_1,\ldots,x_n$ and $y_1,\ldots,y_m$ of the form
	\begin{multline}
		\period\left( \BBt{n}{m} \right)
		= \int \frac{\Omega}{x_1\!\cdots x_n y_1\!\cdots y_m} \left\{ 
				(-1)^{n+m+1}
				+ \sum_{\mathclap{
						I \times J \subsetneq [n] \times [m]
					}}
				(-1)^{I+J} \left[ 
					\Li_1\!\left( \frac{x_{\!I^c}y_{\!J^c}(x_{\!I} + y_{\!J})}{\psi(x_{[n]}+y_{[m]})} \right)
		\right.\right.\\\left.\left.
				+\frac{x_{\!I} + y_{\!J}}{x_{\!I^c}} \Li_1\!\left( -\frac{x_{\!I^c}^2}{\psi} \right)
				+\frac{x_{\!I} + y_{\!J}}{y_{\!I^c}} \Li_1\!\left( -\frac{y_{\!J^c}^2}{\psi} \right)
				\right]
		\right\}.
		\label{eq:BBt-projective}%
	\end{multline}
	Here we set $I^c \defas [n] \setminus I$, $J^c \defas [m] \setminus J$ and
	$\psi \defas x_{\!I^c} y_{\!J^c} + (x_{\!I}+y_{\!J})(x_{\!I^c}+y_{\!J^c})$. When $I^c = \emptyset$, the term $\Li_1(-x_{\!I^c}^2/\psi)/x_{\!I^c}$ is understood as zero (its limit when $x_{I^c} \rightarrow 0$). The same convention applies for $J^c = \emptyset$.
\end{lemma}
\begin{proof}
	The derivation is a straightforward extension of the arguments given in the proof of lemma~\ref{lemma:bubble-graph-rational-parametric}, so we omit it here.
\end{proof}

\subsection{Cocommutative graphs}
\label{sec:ex-cocommutative}%
The period \eqref{eq:def-period} of a graph with subdivergences usually depends on the chosen renormalization point, as we just exemplified above. But under special circumstances it may become independent of the renormalization scheme. The simplest examples where this phenomenon occurs are cocommutative graphs.
\begin{definition}
	For any $n \in \N$, the \emph{iterated coproduct}
	$\cop[n]\colon \FeynHopf \longrightarrow \FeynHopf^{\tp(n+1)}$ is defined by $\cop[1] \defas \cop$ and
	$\cop[n+1] \defas (\id^{\tp k} \tp \cop \tp \id^{\tp(n-k)}) \circ \cop[n]$ for any choice of $0 \leq k \leq n$.\footnote{This is well-defined because $\FeynHopf$ is coassociative \cite{CK:RH1}.} We write $\cop[n](x) = \sum_{(x)} x_{(1)} \tp \cdots \tp x_{(n+1)}$.

	An element $x \in \FeynHopf$ is called \emph{cocommutative} if and only if $\tau\circ \cop(x) = \cop(x)$ for the flip $\tau(a \tp b) \defas b\tp a$, which is equivalent to
	$\sum_{(x)} x_{(1)}  \tp x_{(2)} = \sum_{(x)} x_{(2)} \tp x_{(1)}$.
\end{definition}
\begin{lemma}
	\label{lemma:cocommutative-cyclic} %
	Let $x\in \FeynHopf$ be cocommutative. Then all iterated coproducts are invariant under cyclic permutations $\tau_n (a_1 \tp \cdots \tp a_n) \defas a_2 \tp \cdots \tp a_n \tp a_1$. This means that for arbitrary $n \in \N$ and $0 \leq k \leq n$ we have
	\begin{equation}
		\cop[n](x)
		= \tau_{n+1}^k \left[ \cop[n](x) \right]
		=	\sum_{(x)}
			x_{(k+1)} \tp \cdots \tp x_{(n+1)} \tp x_{(1)} \tp \cdots \tp x_{(k)}%
		.
		\label{eq:cocommutative-cyclic} %
	\end{equation}
\end{lemma}
\begin{proof}
	This is just the coassociativity
	$	\big[\id \tp \cop[n] \big]\circ \cop
	=	\big[\cop[n] \tp \id \big]\circ \cop
	= \cop[n]$:
	\begin{equation*}
		\left[ \cop[n] \tp \id \right] \circ \cop (x)
		= \left[\cop[n] \tp \id \right] \circ \tau \circ \cop(x)
		= \tau_{n+1} \circ \left[\id \tp \cop[n]\right] \circ \cop(x)
		. \qedhere
	\end{equation*}
\end{proof}
\begin{remark}
	We cannot deduce full symmetry of the iterated coproducts from cocommutativity alone. For example, the word
	$x = abc+bca+cab \in T\left( \set{a,b,c} \right)$ is cocommutative,
	but $\cop[2](x) = a \tp b \tp c + b\tp c \tp a + c \tp a \tp b + R$ is invariant only under permutations that are cyclic (all tensors in $R$ have at least one slot which is $\1$).
\end{remark}
\begin{corollary}
	If $x \in \FeynHopf$ is cocommutative, then its period $\period(x)$ is independent of the chosen renormalization point $\widetilde{\Kinematics}$.
\end{corollary}
\begin{proof}
	Changing the renormalization point from $\widetilde{\Kinematics}$ to $\widetilde{\Kinematics}'$ gives periods
	\begin{equation*}
		\period'(x)
		\urel{\eqref{eq:period-scheme-dependence}}
			\left( \Psi^{\convolution - 1} \convolution \period \convolution\Psi \right) (x)
		\urel{\eqref{eq:cocommutative-cyclic}}
			\left( \period \convolution\Psi\convolution\Psi^{\convolution-1} \right) (x)
		= \period(x)
		\quad\text{where}\quad
		\Psi = \restrict{\FRren}{\widetilde{\Kinematics}'}
		.\qedhere
	\end{equation*}
\end{proof}
It turns out that the independence on the renormalization point can be made manifest in the parametric representation. In the sequel we consider special types of cocommutative elements in the Hopf algebra of graphs (without edge labels) and lift them to cocommutative elements in the Hopf algebra of edge-labelled graphs, for example
\begin{equation*}
	\cop \left(\Graph[0.3]{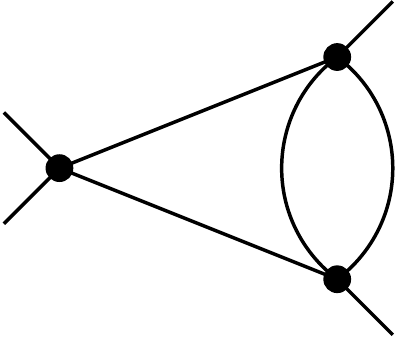}\right) = \Graph[0.25]{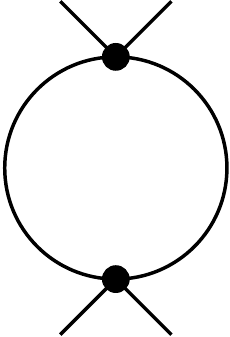} \tp \Graph[0.25]{bubble_vert}
	\quad\text{to}\quad
	\cop \left(\Graph[0.3]{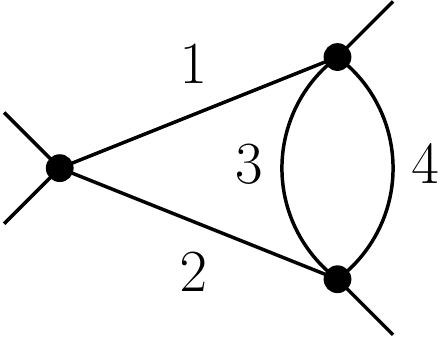}
	+\Graph[0.3]{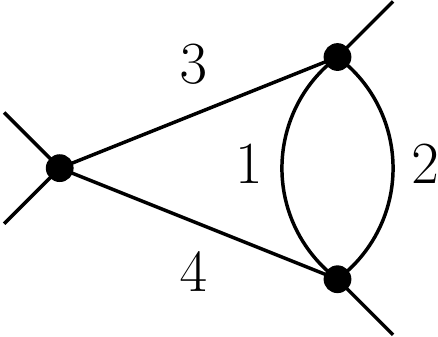} \right)
	= \Graph[0.3]{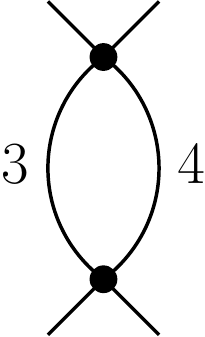} \tp \Graph[0.3]{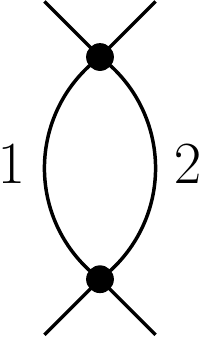} + \Graph[0.3]{bubble_12v} \tp \Graph[0.3]{bubble_34v}\ .
\end{equation*}
This allows us to completely cancel all second Symanzik polynomials $\phipol$ in the parametric representation of their period.
\begin{figure}
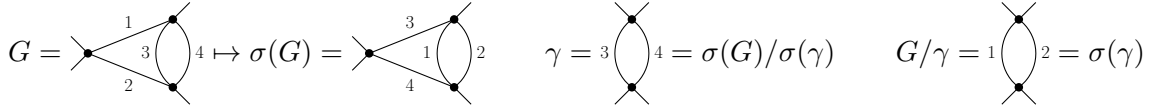

	$ G = \Graph[0.4]{dunce_numbered}
		\mapsto
		\sigma(G) = \Graph[0.4]{dunce_numbered2}
	$ \hfill
	$ \gamma = \Graph[0.4]{bubble_34v} = \sigma(G) / \sigma(\gamma)$
	\hfill
	$ G/\gamma = \Graph[0.4]{bubble_12v} = \sigma(\gamma)$
	\caption[Cocommutativity of dunce's cap]{The cocommutative dunce's cap $G$ and an isomorphism $\sigma$ to a relabelled graph $\sigma(G)$ that swaps the subdivergence $\gamma$ of $G$ with the quotient $\sigma(G)/\sigma(\gamma)$ of $\sigma(G)$ and vice versa.}%
	\label{fig:dunce-cocommutative-bijection}%
\end{figure}%
\subsubsection{Single graphs with a primitive subdivergence}
Consider a $\phi^4$ graph $G$ with a single subdivergence $\gamma$ such that both $G$ and $\gamma$ are logarithmically divergent. Then the period of $G$ can be written as
	\begin{equation}
		\period(G)
		\urel{\eqref{eq:period-projective}}
		\int \Omega \left[ \frac{1}{\psipol_G^2}
			- \frac{1}{\psipol_{\gamma}^2 \psipol_{G/\gamma}^2}
			\cdot \frac{
					\widetilde{\phipol}_{G/\gamma} \psipol_{\gamma}
				}{
					\widetilde{\phipol}_{G/\gamma} \psipol_{\gamma}
					+ \widetilde{\phipol}_{\gamma} \psipol_{G/\gamma}
				}
		\right]
		\label{eq:period-projective-singlesub}%
	\end{equation}
	and depends on $\widetilde{\Kinematics}$ through the second Symanzik polynomials $\widetilde{\phipol}_{G/\gamma}$ and $\widetilde{\phipol}_{\gamma}$.
	Now assume cocommutativity of $G$ (that is $\gamma \isomorph G/\gamma$), then we can find a relabelling (bijection)
	$
		\sigma\colon \edges(G) \longrightarrow \edges(G)
		$ of the edges of $G$ such that the subdivergence of $\sigma(G)$ is $\sigma(\gamma) = G/\gamma$ with quotient $\sigma(G)/\sigma(\gamma) = \gamma$. An example is shown in figure~\ref{fig:dunce-cocommutative-bijection}, where 
		$\sigma = \left(\begin{smallmatrix}
				1 & 2 & 3&4\\
				3 & 4 & 1 & 2\\
			\end{smallmatrix}\right)
		$. This construction interchanges the polynomials $\widetilde{\phipol}_{\sigma(\gamma)} = \widetilde{\phipol}_{G/\gamma}$ and $\widetilde{\phipol}_{\sigma(G)/\sigma(\gamma)} = \widetilde{\phipol}_{\gamma}$ of the sub- and cograph (analogously for the first Symanzik $\psipol$), when we replace $G$ with $\sigma(G)$. Thus the second Symanziks drop out in the sum
\begin{equation*}
			- \frac{1}{\psipol_{\gamma}^2 \psipol_{G/\gamma}^2}
			\frac{
					\widetilde{\phipol}_{G/\gamma} \psipol_{\gamma}
				}{
					\widetilde{\phipol}_{G/\gamma} \psipol_{\gamma}
					+ \widetilde{\phipol}_{\gamma} \psipol_{G/\gamma}
					}
			- \frac{1}{\psipol_{\sigma(\gamma)}^2 \psipol_{\sigma(G)/\sigma(\gamma)}^2}
			\frac{
					\widetilde{\phipol}_{\sigma(G)/\sigma(\gamma)} \psipol_{\sigma(\gamma)}
				}{
					\widetilde{\phipol}_{\sigma(G)/\sigma(\gamma)} \psipol_{\sigma(\gamma)}
					+ \widetilde{\phipol}_{\sigma(\gamma)} \psipol_{\sigma(G)/\sigma(\gamma)}
					}
		=
			- \frac{1}{\psipol_{\gamma}^2 \psipol_{G/\gamma}^2}
\end{equation*}
and we obtain a representation of the period that is manifestly independent of $\widetilde{\Kinematics}$:
\begin{equation}
	\period(G)
	= \period(\sigma(G))
	= \frac{\period(G) + \period(\sigma(G))}{2}
	= \frac{1}{2} \int \Omega \Bigg(
			\frac{1}{\psipol_G^2}
			+\frac{1}{\psipol_{\sigma(G)}^2}
			- \frac{1}{\psipol_{\gamma}^2 \psipol_{G/\gamma}^2}
		\Bigg).
	\label{eq:period-cocommutative-singlesub}%
\end{equation}
\begin{example}\label{ex:cocommutative-singlesub}
The simplest example in $\phi^4$ is dunce's cap (figure~\ref{fig:dunce-cocommutative-bijection}), which gives
\begin{align*}
	\period\left( \Graph[0.3]{dunce} \right)
	&=
	\int \frac{\Omega}{2} \left\{
		\frac{1}{[(\SP_1 + \SP_2)(\SP_3 + \SP_4) + \SP_3 \SP_4]^2}
		+ \frac{1}{[(\SP_1 + \SP_2)(\SP_3 + \SP_4) + \SP_1 \SP_2]^2}
		\right.\nonumber\\&\qquad\qquad
		-	\left.\frac{1}{(\SP_1 + \SP_2)^2(\SP_3 + \SP_4)^2}
			\right\}
	\nonumber\\
	&= \int \frac{\Omega}{2} \left\{
		\frac{1}{(\SP_1+\SP_2)(\SP_1 \SP_2 + \SP_1 \SP_3 + \SP_2 \SP_3)}
		- \frac{1}{(\SP_1 + \SP_2)^2(\SP_1 + \SP_2 + \SP_3)}\right\}
	\nonumber\\
	&= \int \frac{\Omega}{2(\SP_1 + \SP_2)^2} \ln \frac{(\SP_1 + \SP_2)^2}{\SP_1 \SP_2}
	= 1.
\end{align*}
\end{example}
\begin{remark}
	The representation \eqref{eq:period-cocommutative-singlesub} is very well suited for evaluation with hyperlogarithms, because only the first Symanzik polynomial occurs which gives plenty of factorization identities to aid linear reducibility. Each of its summands can be integrated separately in the sense of regularized limits: As long as we keep the same order of integration variables for each summand, the total sum of these regularized limits equals the overall (convergent) integral.
	
	In practice we can omit the term $\psipol_{\gamma}^{-2} \psipol_{G/\gamma}^{-2}$, because when we integrate the last edge $e$ of the subgraph $\gamma$, the integrand is proportional to $\dd \SP_e/\SP_e$ and integrates to $\log(\SP_e)$ which gets annihilated under $\AnaReg{\SP_e}{0,\infty}$.
\end{remark}
\begin{example}
	Since there is no primitive $\phi^4$ graph with two loops, the first interesting example appears at six loops, when the wheel $\WS{3} \isomorph \gamma \isomorph G/\gamma$ is inserted into itself. Our result, obtained with hyperlogarithms in \cite{Panzer:MasslessPropagators}, reads
\begin{equation}\label{eq:w3-in-w3}
	\period \left( \Graph[0.2]{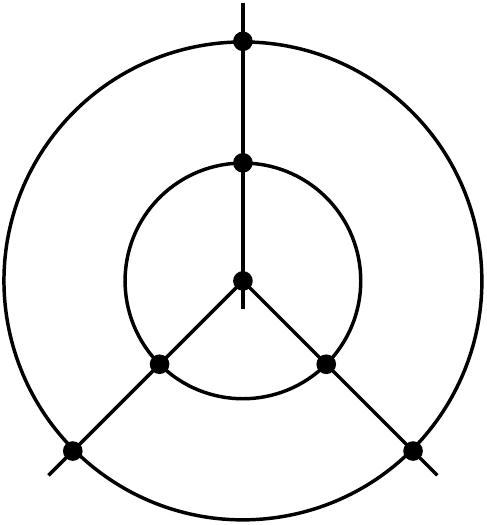} \right)
		= 72\mzv[2]{3} - \tfrac{189}{2}\mzv{7}.
\end{equation}
\end{example}
\subsubsection{Linear combinations of graphs}
	We can also construct cocommutative elements from several graphs. Let $G_1$ and $G_2$ both have one subdivergence $\gamma_i\subset G_i$ such that $\gamma_1 \isomorph G_2/\gamma_2$ and $\gamma_2 \isomorph G_1/\gamma_1$. Then $\cop(G_1 + G_2) = \gamma_1 \tp G_1/\gamma_1 + \gamma_2 \tp G_2/\gamma_2$ is cocommutative and again we can arrange for $\widetilde{\phipol}_{G_2/\gamma_2} = \widetilde{\phipol}_{\gamma_1}$ and so on with an adapted labelling of the edges such that we find
	\begin{equation}
		\period(G_1) + \period(G_2)
		= \period(G_1+G_2)
		= \int \Omega \left( \frac{1}{\psipol_{G_1}^2} + \frac{1}{\psipol_{G_2}^2} - \frac{1}{\psipol_{\gamma_1}^2\psipol_{\gamma_2}^2} \right).
		\label{eq:period-cocommutative-sum-projective}%
	\end{equation}%
\begin{figure}
		\centering
		\begin{tabular}{c@{\qquad}c@{\qquad}c@{\qquad}c}
			$ \Graph[0.5]{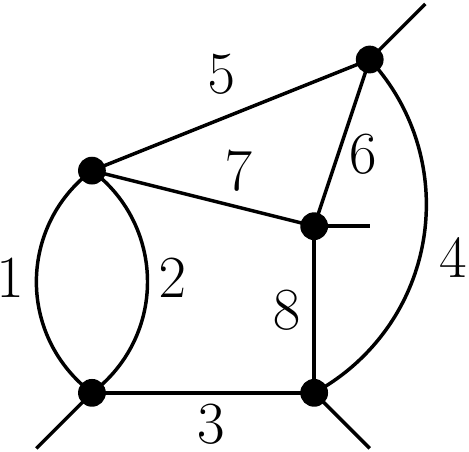} $ & $ \Graph[0.54]{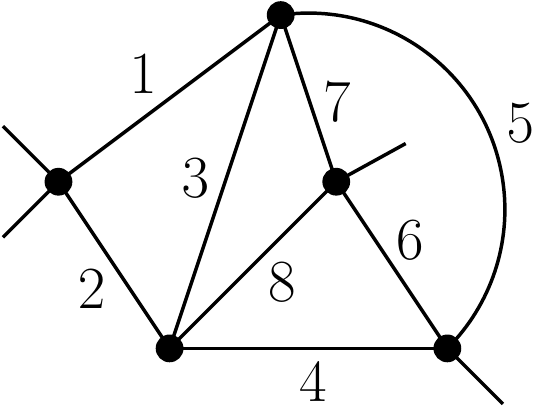}$ & $\Graph[0.5]{bubble_12v}$ & $\Graph[0.43]{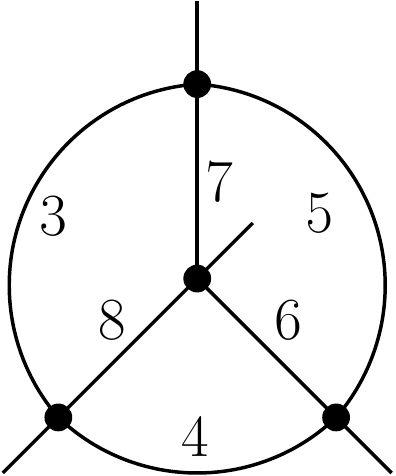}$\\
			$G_1$ & $G_2$ & $\gamma_1 = G_2/\gamma_2$ & $\gamma_2 = G_1/\gamma_1 = \WS{3}$\\
		\end{tabular}
		\caption[Insertion of the bubble into the wheel $\WS{3}$ and vice versa]{Insertions $G_1$ of the bubble $\gamma_1$ into $\WS{3}$ and $G_2$ of $\WS{3}$ into $\gamma_1$. The labelling ensures that $\gamma_1 = G_2/\WS{3}$ and $G_1/\gamma_1 = \WS{3}$ have the same induced labels.}%
		\label{fig:bubble-in-3spokes}%
\end{figure}%
\begin{example}
	The first example in $\phi^4$ theory occurs at four loops by inserting the bubble in the wheel $\WS{3}$ and vice versa, as shown in figure~\ref{fig:bubble-in-3spokes}.
	Note that for the convergence and correctness of \eqref{eq:period-cocommutative-sum-projective} it is crucial to label the edges carefully as required above. One such labelling is shown in the figure and the integration delivers
	\begin{equation*}
		\period(G_1) + \period(G_2)
		= 6\mzv{3},
		\label{eq:period-bubble-3spokes}%
	\end{equation*}
	but we will no longer indicate suitable labellings in the examples as they are straightforward to construct.
\end{example}
	Another possibility is to consider identical subdivergences $\gamma_1 = \gamma_2$ with the same cograph, then the difference $G_1-G_2$ is primitive. From \eqref{eq:period-projective-singlesub} we find
	\begin{equation}
		\period(G_1) - \period(G_2)
		= \period(G_1 - G_2)
		= \int \Omega \left( \frac{1}{\psipol_{G_1}^2} - \frac{1}{\psipol_{G_2}^2} \right).
		\label{eq:period-cocommutative-difference-projective}%
	\end{equation}
\begin{figure}
	\centering
	\begin{tabular}{ccccc@{\hspace{-1ex}}c}
		$ \Graph[0.38]{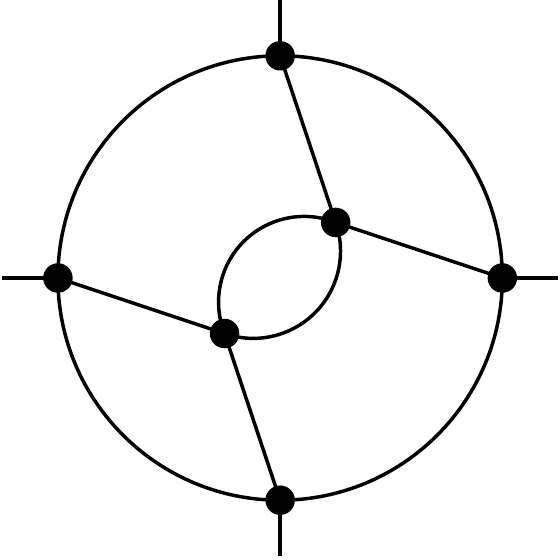} $ &
		$ \Graph[0.38]{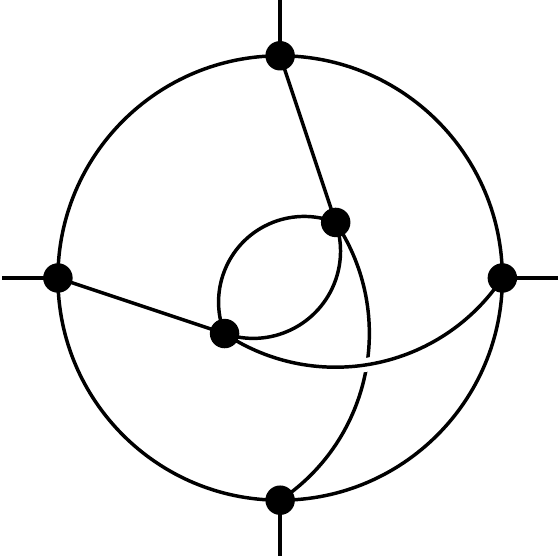} $ &
		$ \Graph[0.38]{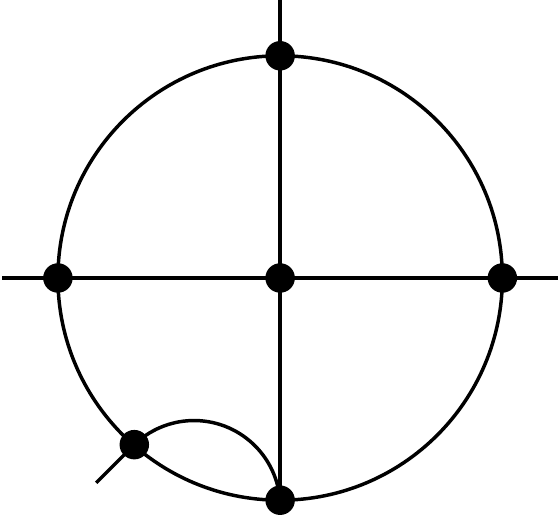} $ &
		$ \Graph[0.38]{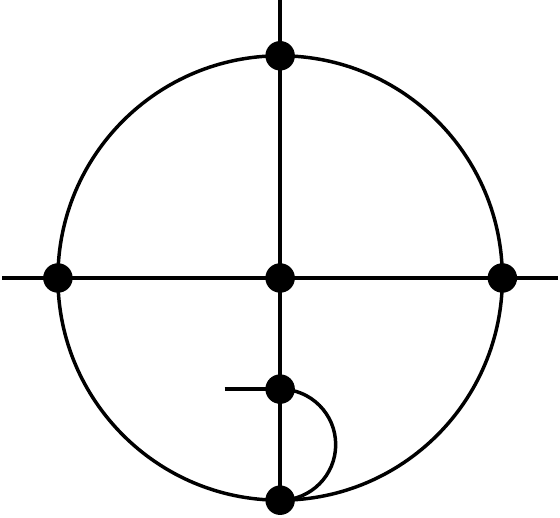} $ &
		$ \Graph[0.38]{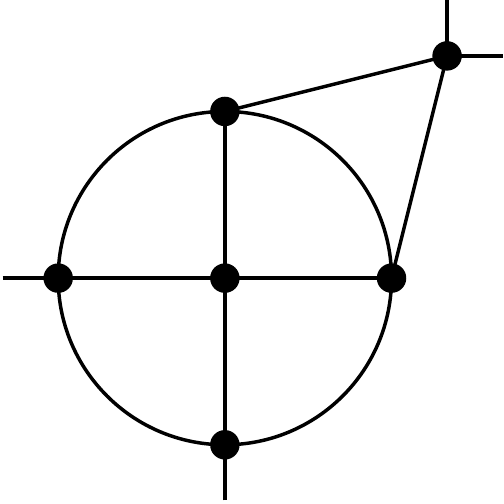} $ &
		$ \Graph[0.38]{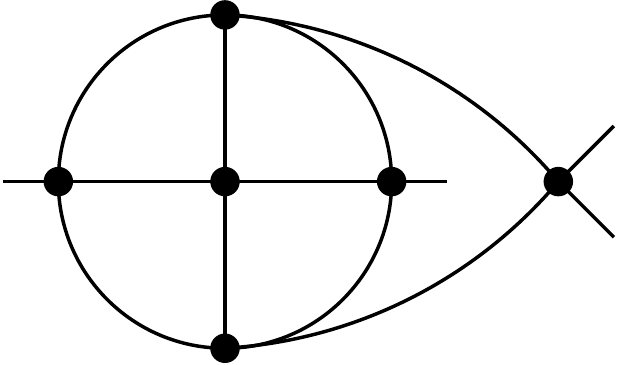} $ \\
		$ G_1$ & $G_2$ & $G_3$ & $G_4$ & $G_5$ & $G_6$ \\
	\end{tabular}
	\caption[Insertions of the bubble into $\WS{4}$ and vice versa]{Different insertions of the bubble into the wheel $\WS{4}$ with four spokes ($G_1$ to $G_4$) and insertions of $\WS{4}$ into the bubble ($G_5$ and $G_6$).}%
	\label{fig:bubble-in-4spokes}%
\end{figure}%
\begin{example}
	The first such case in $\phi^4$ theory are the four different ways to insert a bubble into $\WS{4}$. In figure~\ref{fig:bubble-in-4spokes} we also show the two different insertions of $\WS{4}$ into the bubble. Together with \eqref{eq:period-cocommutative-sum-projective}, we obtain five linearly independent relations among the periods of these six graphs. Explicitly we computed
	\begin{equation*}\begin{split}
		40 \mzv{5}
		&=\period(G_3) + \period(G_5)
		\qquad\text{and}
		\\
		6 \mzv[2]{3}
		&=\period(G_1) - \period(G_2)
		= \period(G_4) - \period(G_1)
		= \period(G_3) - \period(G_4)
		= \period(G_5) - \period(G_6).
	\end{split}\end{equation*}
\end{example}
Since a bubble is a one-scale subgraph, its insertion reduces to a period of the quotient graph (just as in example~\ref{ex:onescale-sub}). Therefore the really interesting situations are when both primitives $\gamma$ and $G/\gamma$ are different from the bubble. In $\phi^4$ theory, this requires at least seven loops.
\begin{figure}
	\centering
	$ \WS{3 \hookrightarrow 4} \defas \Graph[0.4]{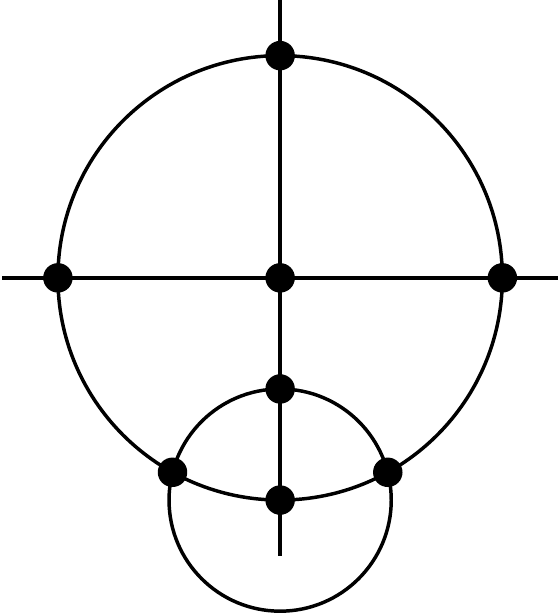} $
	\quad
	$ \WS{3 \hookrightarrow 4'} \defas \Graph[0.4]{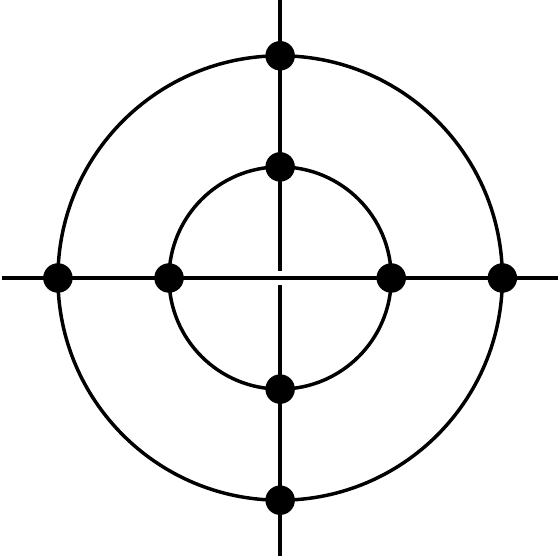} $
	\quad
	$ \WS{4 \hookrightarrow 3} \defas \Graph[0.4]{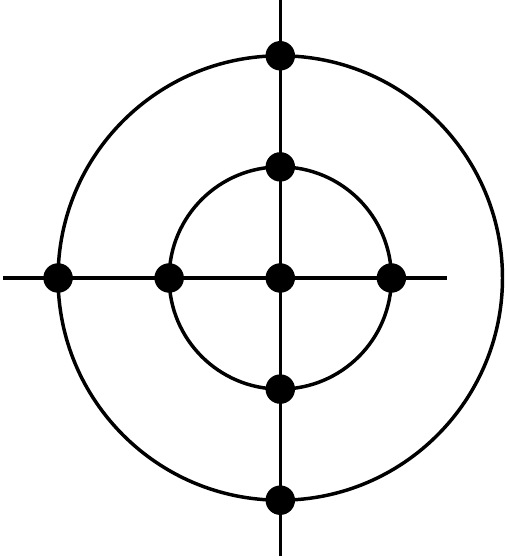} $
	\caption{Insertions of the wheels with $3$ and $4$ spokes into each other.}%
	\label{fig:ws3-in-ws4}%
\end{figure}%
\begin{example}
	The wheel $\WS{3}$ can be inserted into $\WS{4}$ in two different ways and there exists one insertion of $\WS{4}$ into $\WS{3}$, as shown in figure~\ref{fig:ws3-in-ws4}. They can be combined to define two linearly independent cocommutative elements \cite{Kreimer:WheelsInWheels} and the evaluation of their periods is of high interest \cite{Kreimer:LL2014}. Our new results read
	\begin{equation}\begin{split}\label{eq:ws3-in-ws4}%
		\period\big( \WS{3 \hookrightarrow 4} - \WS{3 \hookrightarrow 4'} \big)
		&= 72 \mzv[3]{3}
		\quad\text{and}\\
		\period\big( \WS{3 \hookrightarrow 4} + \WS{4 \hookrightarrow 3} \big)
		&= 480 \mzv{3} \mzv{5} - 40 \mzv[3]{3} - \frac{4730}{9} \mzv{9}.
	\end{split}\end{equation}
These show a double weight-drop (the generic weight for primitive $7$-loop $\phi^4$ periods is $11$) as expected by the analysis of the $c_2$ invariant \cite{BrownSchnetzYeats:PropertiesC2}. It occurs as follows: After integrating out the variables associated to the $\WS{3}$ sub- or cograph in \eqref{eq:period-cocommutative-sum-projective} or \eqref{eq:period-cocommutative-difference-projective}, the denominator of the partial integral is $\psipol_{\WS{4}}^2$.
\end{example}

\subsubsection{Iterated insertions of ladder type}
	A further source of cocommutative elements is supplied by series $(G_n)_{n \in \N}$ of \emph{ladder type}, by which we mean that their coproducts obey
	\begin{equation}
		\copred(G_n) = \sum_{k=1}^{n-1} G_k \tp G_{n-k}
		\quad\text{for all}\quad
		n \in \N.
		\label{eq:ladder-series}%
	\end{equation}
	These arise very naturally by iterated insertions of a primitive graph $\gamma \defas G_1$ into itself, such that $G_n/G_k \isomorph G_{n-k}$ for all $i<n$ and in particular we have $G_{n+1}/G_n \isomorph \gamma$. Example~\ref{ex:cocommutative-singlesub} considered just the start of such a series, which we indicate for the bubble $\gamma = \smallbubble$ in figure~\ref{fig:bubble-insertion-series}. It resembles the zigzag series (figure~\ref{fig:zigzags}) in that it defines an infinite sequence of renormalization point independent periods in $\phi^4$ theory.
\begin{figure}
	\centering
	$ G_1 = \Graph[0.3]{bubble_vert}\ 
		G_2 = \Graph[0.35]{dunce} \ 
		G_3 = \Graph[0.4]{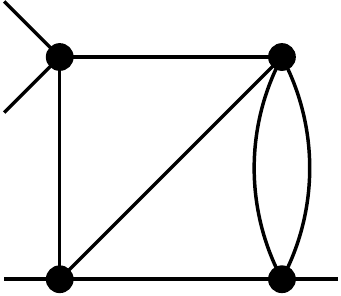} \ 
		G_4 = \Graph[0.4]{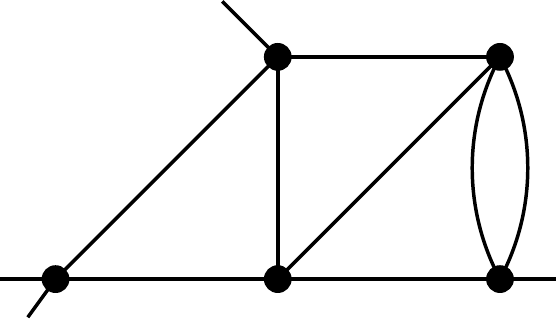} \ 
		G_5 = \Graph[0.4]{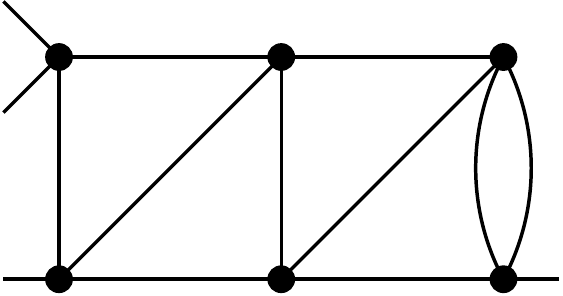} \ 
		\cdots$ 
	\caption{A series of iterated, cocommutative bubble self-insertions in $\phi^4$ theory.}%
	\label{fig:bubble-insertion-series}%
\end{figure}

	To cancel the second Symanzik polynomial in the parametric representation, we must now average over more graphs. Fix $n$ and consider a family
	$\sigma^i\colon G_n \longrightarrow G_n^i$ of isomorphisms that relabel the edges, where $0 \leq i < n$ and we set $\sigma^0 \defas \id$. We write $G_k^i \defas \sigma^i(G_k)$ for the subdivergences of $G_n^i$ (setting $G_0^i \defas \1$), so $\copred(G_n^i) = \sum_{k=1}^{n-1} G_k^i \tp G_{n-k}^i$. The cocommutativity hints that we can choose $\sigma^i$ such that
	$G^{i}_{n-k} = G_{n}^{i+k}/G_{k}^{i+k}$ and $G_n^i/G_k^i = G_{n-k}^{i+k}$ whenever $i+k < n$. The idea is that we shift the variables cyclically from one subquotient to the next:
	\begin{equation*}
		\gamma
		\isomorph
		G_{k+1}^{i} / G_k^i
		= \begin{cases}
			G_{k}^{i+1} / G_{k-1}^{i+1} &\text{for $k>0$ and}\\
			G_{n}^{i+1} / G_{n-1}^{i+1} & \text{when $k=0$}.\\
		\end{cases}
	\end{equation*}
	Since all subdivergences $\1 \neq G_1^i \subsetneq \cdots \subsetneq G_{n-1}^i \subsetneq G_n^i$ are nested, any subset
	$F^i_k = \set{G_{k_1}^i,\ldots,G_{k_r}^i}
		\in \forests\left( G_n^i \right)
	$ indexed by $1 \leq k_1 < \cdots< k_r < n$ defines a forest. By construction, the set of subquotients
	\begin{equation*}
		Q(F^i_k) \defas \set{
				G_{k_1}^i,
				G_{k_2}^i/G_{k_1}^i,
				\ldots,
				G_{k_r}^i/G_{k_{r-1}}^i,
				G_n^i / G_{k_r}^i
		}
	\end{equation*}
	that determine its contribution to the forest formula \eqref{eq:period-projective} is invariant under the shift
	\begin{equation*}
		\tau\left(F^i_k\right)
		\defas
		F^{i+k_1}_{\set{k_2-k_1,\ldots,k_r-k_1,n-k_1}}
			\quad\text{(replace $i+k_1$ with $i-n+k_1$ if $i+k_1\geq n$)}.
	\end{equation*} 
	This $\tau$ is a permutation of the set $\forests_n \defas \bigcupdot_{i=0}^{n-1} \forests(G_n^i)$ and each $F^i_k$ lies in a cycle $[F_k^i]$ of size $1 + r = \abs{Q(F_k^i)}$. We consider $\sum_{i=0}^{n-1} \period(G_n^i)$ and collect the contributions for each of the cycles $\forests_n/\tau \defas \setexp{[F]}{F \in \forests_n}$ such that the fractions with second Symanzik polynomials add up to unity. So finally,
	\begin{equation}
		\period(G_n)
		= \frac{1}{n} \sum_{k=0}^{n-1} \period\left(G_n^k\right)
		= \frac{1}{n} \int \Omega \hspace{-1ex}\sum_{[F] \in \forests_n / \tau}\hspace{-1ex} \frac{(-1)^{\abs{F}}}{\psipol_F^2}
		.
		\label{eq:ladder-period-projective-cyclic}%
	\end{equation}
\begin{figure}
	\centering
	$
		G_3^0 = \Graph[0.6]{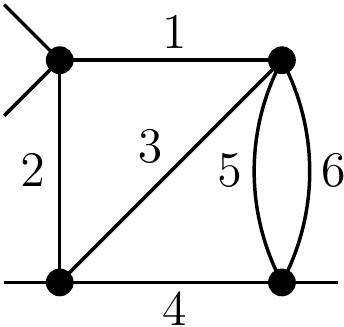} \qquad
		G_3^1 = \Graph[0.6]{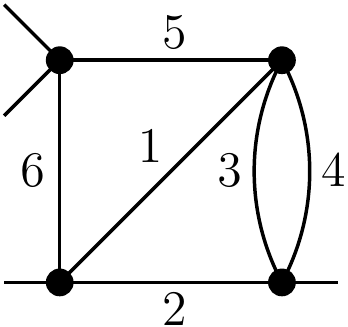} \qquad
		G_3^2 = \Graph[0.6]{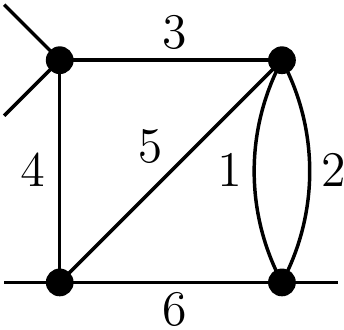} 
	$
	\caption{Cyclic relabellings $G_3^i = \sigma^i(G_3)$ of the same cocommutative graph.%
	}%
	\label{fig:dunce-ladder-3-labellings}%
\end{figure}%
\begin{example}
	For the series of bubble insertions (figure~\ref{fig:bubble-insertion-series}), explicit relabellings for $n=3$ are shown in figure~\ref{fig:dunce-ladder-3-labellings} where $\sigma = \left(\begin{smallmatrix} 1 & 2 & 3 & 4 & 5 & 6\\ 5 & 6 & 1 & 2 & 3 & 4\\ \end{smallmatrix}\right)$. The integration of
	\begin{equation*}
		\period(G_3) = \int \frac{\Omega}{3} \left( 
			\frac{1}{\psipol_{G_3^0}^2}
			+\frac{1}{\psipol_{G_3^1}^2}
			+\frac{1}{\psipol_{G_3^2}^2}
			- \frac{1}{\psipol_{G_{1}^0}^2 \psipol_{G_{2}^1}^2}
			- \frac{1}{\psipol_{G_{1}^1}^2 \psipol_{G_{2}^2}^2}
			- \frac{1}{\psipol_{G_{1}^2}^2 \psipol_{G_{2}^0}^2}
			+ \frac{1}{\psipol_{G_{1}^0}^2 \psipol_{G_{1}^1}^2 \psipol_{G_{1}^2}^2}
		\right)
	\end{equation*}
	is elementary in this case and we obtain $\period(G_3) = 2$. In fact, the bubble series evaluates to the Catalan numbers $\period(G_{n+1}) = \binom{2n}{n}/(n+1)$ for all $n\in\N$. The proof is simplest in momentum space (as in lemma~\ref{lemma:bubble-graphs-transcendental-period}) where we can exploit that in dimensional regularization, $\FRdim(G_{n+1}) = q^{-2\varepsilon} \FRdim(G_n) \onemaster{1}{1+n\varepsilon}$ if we render all graphs one-scale through nullification of the external momentum attached to the innermost subdivergence $G_1$.

If the Feynman rules are subject to such a recursion relation, the function $\onemaster{1}{1+n\varepsilon}$ is called \emph{Mellin transform} and the renormalized integrals can be computed explicitly in terms of this function. Moreover, the full generating function of all periods is related to a \emph{Dyson-Schwinger equation} and subject to a differential equation. We gave detailed expositions of these concepts in \cite{Panzer:Mellin,Panzer:Master}, where the reader will find also an essentially equivalent example resulting in the Catalan numbers as well.

Note that the parametric integration of \eqref{eq:ladder-period-projective-cyclic} is not at all trivial for higher $n$. We used this series of known periods as test cases for our program {\HyperInt}.
\end{example}
\begin{figure}
	\centering
	$
		G_1 = \Graph[0.4]{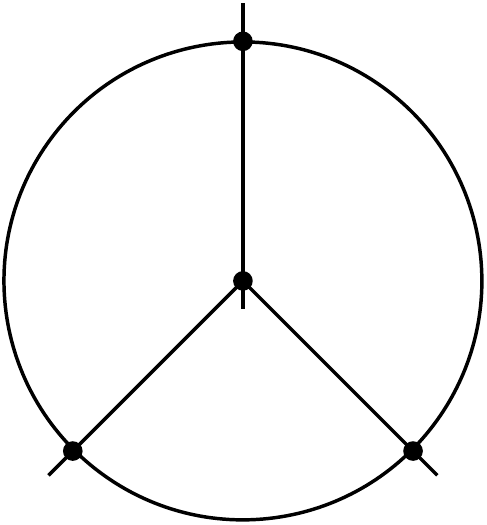}
		\quad
		G_2 = \Graph[0.4]{ws3_ws3}
		\quad
		G_3 = \Graph[0.4]{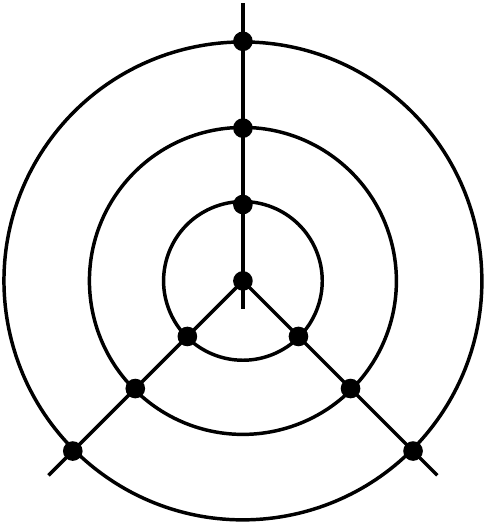}
		\quad
		G_4 = \Graph[0.4]{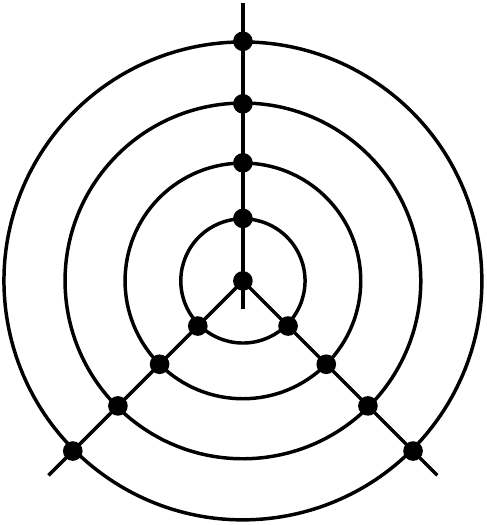}
		\quad
		\cdots
	$
	\caption{A ladder series from insertions into the wheel $G_1=\WS{3}$ with 3 spokes.}%
	\label{fig:ws3-ladder-series}%
\end{figure}%
\begin{remark}
	A much more interesting ladder series is provided by the iterated insertions of $\gamma = \WS{3}$ shown in figure~\ref{fig:ws3-ladder-series}. Apparently all these graphs have vertex-width $3$ and can thus be computed with hyperlogarithms. In particular we know $\period(G_n) \in \MZV$ for all $n \in \N$.
	Our computation \eqref{eq:w3-in-w3} of $\period(G_2)$ supplements the well-known $\period(G_1) = 6\mzv{3}$.
\end{remark}
\begin{remark}
	The coproduct \eqref{eq:ladder-series} of a ladder series implies that the graphs $G_n$ generate a Hopf subalgebra. The scaling behaviour of the renormalized Feynman rules from \eqref{eq:rge-infinitesimal} is thus completely determined by the periods of these graphs only. For example,
	\begin{align*}
		\restrict{\FRren}{\Kinematics_{\scalelog}}\left( \Graph[0.15]{ws3_ws3} \right)
		&= \frac{\scalelog^2}{2} \period^2\left( \Graph[0.15]{ws3} \right)
		- \scalelog \period\left( \Graph[0.15]{ws3_ws3} \right) 
			- \scalelog \period\left( \Graph[0.15]{ws3} \right) \FRren\left( \Graph[0.15]{ws3} \right)
			+ \FRren\left( \Graph[0.15]{ws3_ws3} \right)
		\\
		&=	18 \mzv[2]{3} \scalelog^2
				- \left( 
						72\mzv[2]{3} - \tfrac{189}{2}\mzv{7}
				\right) \scalelog
				- 6\mzv{3} \scalelog \FRren\left( \Graph[0.15]{ws3} \right)
				+ \FRren\left( \Graph[0.15]{ws3_ws3} \right)
		.
\end{align*}
\end{remark}

\section{Three-point functions}
\label{sec:ex-3pt}%
With the recursion formulas of section~\ref{sec:3pt-recursions} we provided an infinite family of massless three-point functions which can be computed with hyperlogarithms by theorem~\ref{theorem:vw3-3pt}. But there remain many graphs outside this family that are still known to evaluate to hyperlogarithms.

\subsection{Vertex-width 3 and graphical functions}
Whenever $G$ is $3$-constructible (with the external vertices active last), then so is its planar dual (remark~\ref{rem:planar-dual-construction}) and we have a complete symmetry between position and momentum space. So in particular the constructibility of $3$-loop graphs in the sense of graphical functions \cite{Schnetz:GraphicalFunctions} carries over to the position space. With our construction of forest functions, we found an explanation of these polylogarithmic results and a method to compute them in the parametric representation.

In the other direction we learn the following: By linear reducibility, we know that all coefficients in the Laurent expansion in $\varepsilon$ can be computed with hyperlogarithms. This suggests that the position space methods of single-valued integration for graphical functions should also be extendable to $\dimension=4-2\varepsilon$ dimensions.

\subsection{Reducibility up to three loops}
We tested explicitly that linear reducibility persists for all three-point functions up to three loops and computed a series of explicit examples \cite{Panzer:DivergencesManyScales}. Experiments at four loops revealed counterexamples, and also in position space (graphical functions) we studied counterexamples with Oliver Schnetz at four loops. But we first want to \emph{understand} the linear reducibility at three loops.

To classify $3$-point functions $G$ we consider their completions $\GComp{G}$, defined as the vacuum graph obtained by adding a $3$-valent vertex to $G$ which connects to the external vertices of $G$. Note that a similar construction proved very useful to study massless propagators (section~\ref{sec:propagators}).
\begin{figure}
	\centering
	\begin{tabular}{ccccc}
		$\Graph[0.5]{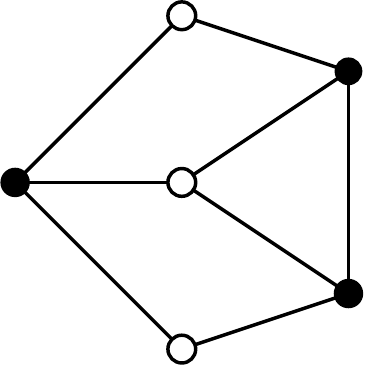}$ &
		$\Graph[0.5]{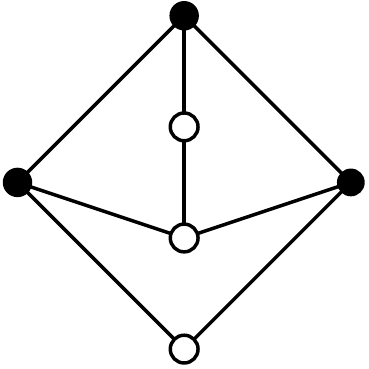}$ &
		$\Graph[0.5]{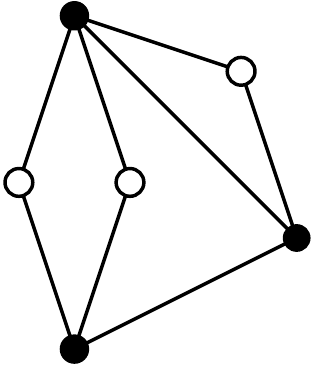}$ &
		$\Graph[0.5]{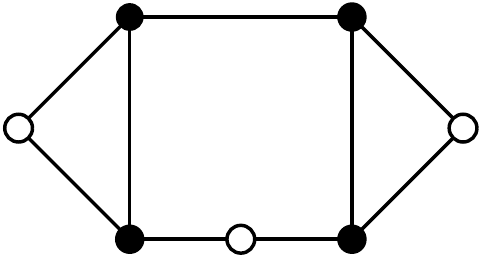}$ &
		$\Graph[0.5]{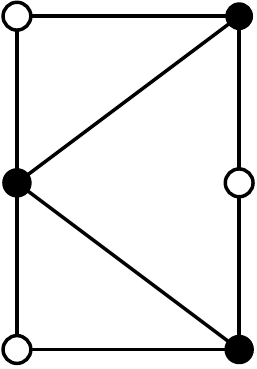}$ \\
		$\Delta_{3,3}$ &
		$\Delta_{3,4}$ &
		$\Delta_{3,5}$ &
		$\Delta_{3,13}$ &
		$\Delta_{3,18}$ \\
	\end{tabular}%
	\caption[Three-loop three-point graphs that are not $3$-constructible]{Three-loop three-point graphs that are not $3$-constructible (white vertices are external). The completions are $\GComp{\Delta_{3,3}} = \GComp{\Delta_{3,4}} = \GComp{\Delta_{3,5}} = {_5 N_2}$, $\GComp{\Delta_{3,13}} = {_5 P_6}$ and $\GComp{\Delta_{3,18}} = {_5 P_4}$ from figure~\ref{fig:5loop-vacuum}.}%
	\label{fig:3pt-nonvw3}%
\end{figure}
So we can use the list in figure~\ref{fig:5loop-vacuum} and enumerate the different three-point functions by deleting a $3$-valent vertex. This gives $23$ non-isomorphic graphs, eleven of which are $3$-constructible and thus not of interest any further. Note that $3$-constructibility of $G$ implies $\vw(\GComp{G}) \leq 3$ as well and can thus only occur if $\GComp{G}$ is one of $_5 P_1$ through $_5 P_6$ (recall that the cube $_5 P_7$ is excluded, despite being planar, by theorem~\ref{theorem:vw3-minors}). The converse is not true however: the graphs $\Delta_{3,13}$ and $\Delta_{3,18}$ from figure~\ref{fig:3pt-nonvw3} are not $3$-constructible (with the condition that the final active vertices of the construction coincide with the external vertices).

Some of the remaining graphs contained more polynomials in the final set of the reduction than $\set{1-z,1-\bar{z},z-\bar{z}}$, which indicates that indeed their expansion will involve more general polylogarithms than in the $3$-constructible case. For example, $\Delta_{3,5}$ features the polynomial $1-z\bar{z}$ \cite{Panzer:DivergencesManyScales}. This is not the case for $\Delta_{3,3}$ and $\Delta_{3,4}$ though, so we learn that for three-point functions, the class of functions appearing is not an invariant of the completion.

\section{Massless on-shell 4-point functions}\label{sec:ex-ladderboxes}
It is known that all $2$-loop massless on-shell $4$-point graphs (see figure~\ref{fig:2loop-4pt}) are linearly reducible (in Schwinger parameters), while counterexamples exist at $3$ loops \cite{BognerLueders:MasslessOnShell,Lueders:LinearReduction}. However, all known results for such graphs evaluate to harmonic polylogarithms, including also non-planar $3$-loop graphs which are not linearly reducible \cite{HennSmirnov:K4}. So while our technique for ladder boxes is already useful in itself, we will need extensions to cover more graphs.
\begin{figure}\centering
		$\Graph[0.33]{box2}$\quad
		$\Graph[0.33]{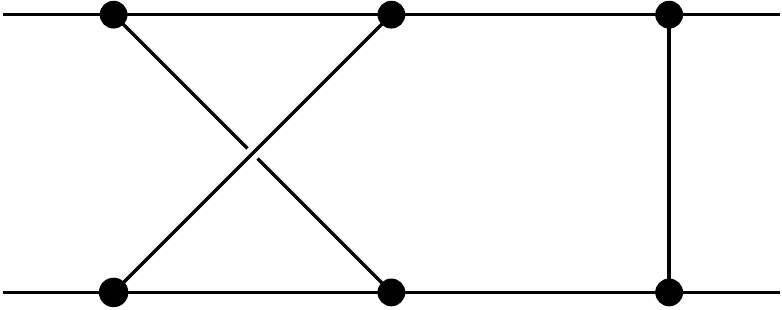}$\quad
		$\Graph[0.33]{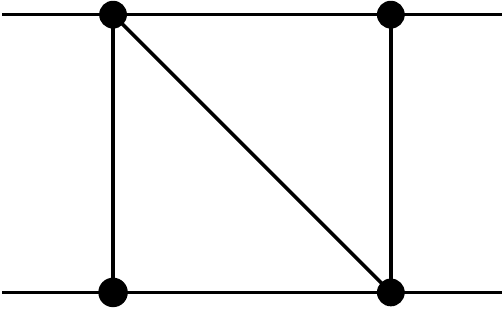}$\quad
		$\Graph[0.33]{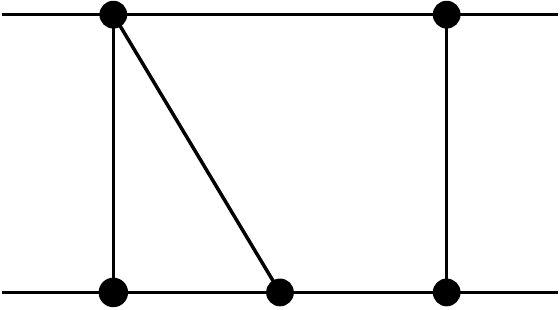}$\quad
		$\Graph[0.33]{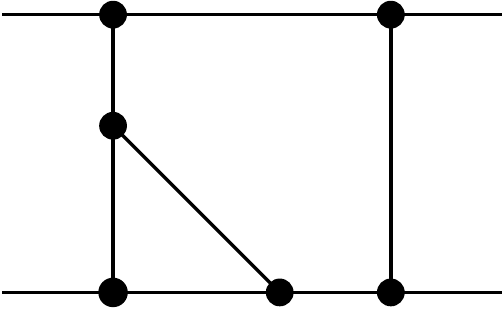}$\quad
		$\Graph[0.33]{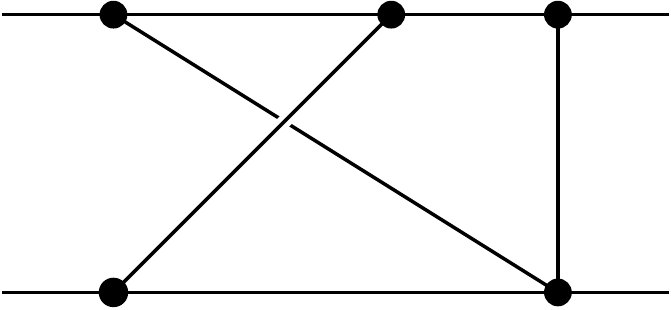}$
	\caption{Two-loop four-point graphs without self-energy (propagator) subgraphs.}%
	\label{fig:2loop-4pt}%
\end{figure}%

\subsection{Ladder boxes}
The simplest application of our recursive method from section~\ref{sec:ladderboxes} is in $\dimension=6$ dimensions with unit indices $\EP_e=1$, because then the ladder box integrals become finite (even when all $p_i^2=0$ are on-shell) and can be computed without the need of a preceding reduction to finite integrals. 

Recently there has been growing interest into precisely these integrals for the study of supersymmetric theories \cite{Kazakov:MultiBoxSix,BorkKazakovVlasenko:N11D6SYM}. From the literature results are only known for up to three loops. Our formalism turned out to be very efficient for this problem and we computed the ladder box integrals for on-shell kinematics up to $6$ loops.

For illustration we provide the full $4$-loop on-shell result ($p_1^2=\cdots=p_4^2=0$) in appendix~\ref{sec:lbox4-onshell}. However, the application in \cite{Kazakov:MultiBoxSix} requires only the special values $c_n \defas \restrict{\FR(B_n)}{s=1,t=0}$ at $s=(p_1+p_2)^2=1$ and $t=(p_1+p_4)^2 = 0$, which we list here:
\begin{align}
c_2 &= 2\mzv{2}
,\label{eq:lbox2-zero}\\
c_3 &= 4 \mzv[2]{3} 
	+ \tfrac{124}{35}\mzv[3]{2}
	- 8 \mzv{3}
	- 6\mzv{2}
,\label{eq:lbox3-zero}\displaybreak[0]\\
c_4 &=
 - 56\mzv{7}
 - 32\mzv{2}\mzv{5}
 + 32\mzv[2]{3}
 + \tfrac{8}{5}\mzv{3}
\left(
4\mzv[2]{2}
 - 15
\right)
 + \tfrac{992}{35}\mzv[3]{2}
 - 8\mzv[2]{2}
 - 18\mzv{2}
 ,\label{eq:lbox4-zero}\displaybreak[0]\\
 c_5 &= 56\mzv{7}
\left(
 \mzv{3} 
 - 5
\right)
 + 26\mzv[2]{5}
 + 4\mzv{5}
\left(
 8\mzv{2}\mzv{3}
 + 35\mzv{3}
 - 40\mzv{2}
 - 49
\right)
 + 4\mzv{3,7}
\nonumber\\&\quad
 + \tfrac{4}{5}\mzv[2]{3}
\left(
 140
 - 25\mzv{2}
 - 4\mzv[2]{2}
\right)
 + 4\mzv{2}
\left(
 2\mzv{3,5}
 - 21
\right)
 - \tfrac{1168}{385}\mzv[5]{2}
\nonumber\\&\quad
 + 20\mzv{3,5}
 - \tfrac{24}{7}\mzv[4]{2}
 + 8\mzv{3}
\left(
7\mzv{2}
 + 4\mzv[2]{2}
 - 14
\right)
 + \tfrac{496}{5}\mzv[3]{2}
\qquad\text{and}\label{eq:lbox5-zero}\displaybreak[0]\\
c_6 &=  \tfrac{\numprint{18864}}{35}\mzv[3]{2}
 + 336\mzv{3,5}
 - 12\mzv{9}
\left(
20\mzv{2}
 + 161
\right)
 + \tfrac{8}{5}\mzv{7}
\left(
104\mzv[2]{2}
 + 35\mzv{2}
 + 840\mzv{3}
 - 1120
\right)
\nonumber\\&\quad
 + 624\mzv[2]{5}
 + \tfrac{16}{35}\mzv{5}
\left(
1680\mzv{2}\mzv{3}
 - 3675
 - 12\mzv[3]{2}
 - 2240\mzv{2}
 + 490\mzv[2]{2}
 + 5145\mzv{3}
\right)
\nonumber\\&\quad
 + 96 \left( \mzv[2]{2} + \mzv{3,7} \right)
 - \tfrac{48}{5}\mzv[2]{3}
\left(
35\mzv{2}
 + 8\mzv[2]{2}
 - 60
\right)
 - \tfrac{32}{5}\mzv{3}
\left(
105
 - 32\mzv[2]{2}
 + 3\mzv[3]{2}
 - 75\mzv{2}
\right)
\nonumber\\&\quad
  +24\mzv{2}
	\left(
		 8\mzv{3,5}
		 - 21
	\right)
	- \tfrac{\numprint{28032}}{385}\mzv[5]{2}
   - \tfrac{288}{5}\mzv[4]{2}
  - 1320\mzv{11}
.\label{eq:lbox6-zero}
\end{align}
Furthermore note that phenomenological applications for ladder box integrals have recently reached the three-loop level \cite{HennSmirnov:PlanarThreeLoop} and the first computation with one leg off-shell has just been published \cite{DiVitaMastroliaSchubertYundin:ThreeLoopLadderBox}. It was obtained with the method of differential equations. Our results show that one can go much further and compute such integrals, to arbitrary loop order, also with two off-shell legs, with {\HyperInt}. 

Note that our theorem \ref{theorem:ladderbox-kinematics} proves that all integrals in the ladder box ``family'', that means including any minors with arbitrary (expanded around integer values) propagator powers, are expressible in terms of polylogarithms with alphabet \eqref{eq:ladderbox-full-kinematics}. The method of differential equations \cite{ArgeriMastrolia:FeynmanDiagramsDifferentialEquations,Henn:LecturesDE} generates a system of linear, first-order partial differential equations in terms of Feynman integrals in the family of the initial graph. Therefore our results show that this method is applicable (i.e. the differential equation can be solved in terms of polylogarithms) in principle to all ladder box graphs. However, for higher orders the derivation of the differential equation itself (which involves integration by parts reductions) is computationally very demanding, as is its subsequent solution. It therefore seems that direct parametric integration using hyperlogarithms might be a more efficient approach.\footnote{In particular, the computation of six-loop ladder boxes like \eqref{eq:lbox6-zero} currently seems to be far out of reach with the differential equations approach.}

\subsection{Extensions}
\begin{figure}
	\centering
	\begin{tabular}{c@{\qquad}c@{\qquad}c@{\qquad}c}
	$ \Graph[0.6]{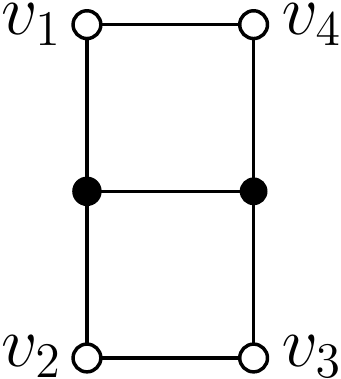} $ &
	$ \Graph[0.6]{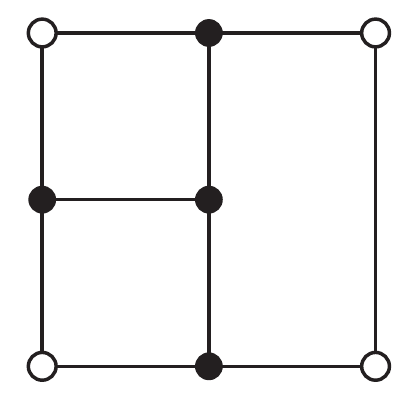} $ &
	$ \Graph[0.6]{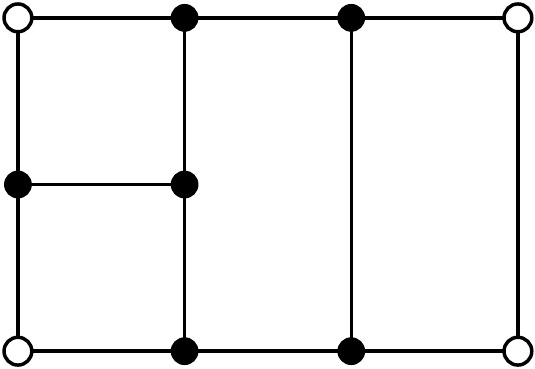} $ &
	$ \Graph[0.6]{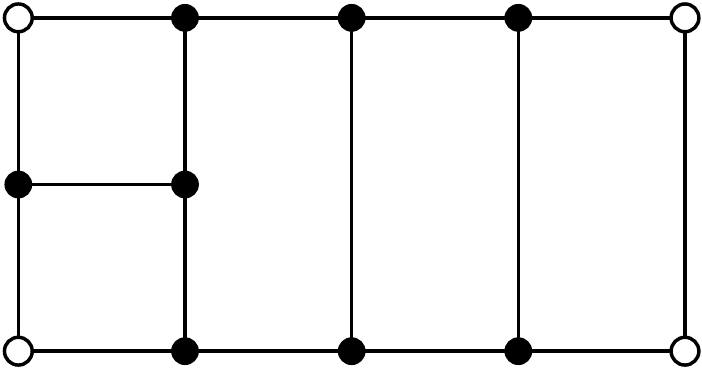} $ \\
	$ T_2 $ & $T_3$ & $T_4$ & $T_5$ \\
	\end{tabular}
	\caption{Vertical double box $T_2$, tennis court $T_3$ and generalizations.}%
	\label{fig:tennis-courts}%
\end{figure}%
\begin{example}
	We can extend theorem~\ref{theorem:ladderbox-kinematics} with every graph whose forest function has a compatibility graph bounded by $(S^{\ForestDeltaBoxSymbol}, C^{\ForestDeltaBoxSymbol})$. A classic example (apart from box ladders) is the \emph{tennis court} diagram $T_3$ shown in figure~\ref{fig:tennis-courts} for massless propagators and light-like (on-shell) external momenta.
	It was evaluated first in a very special case in \cite{BernDixonSmirnov:IterationPlanarAmplitudes} (using a Mellin-Barnes representation) in terms of harmonic polylogarithms and multiple zeta values. Recently its expansion in $\dimension = 4-2\varepsilon$ with unit indices $\EP_e=1$ was obtained to arbitrary order (with the differential equations method), also in terms of harmonic polylogarithms \cite{HennSmirnov:PlanarThreeLoop}.

	Therefore we expected it to be linear reducible, but this assumption fails in Schwinger parameters as was noted in \cite[figure~7.3 (b)]{Lueders:LinearReduction}. Our formalism of forest integrals (which uses different coordinates $f_i/\psipol$) does apply though: We could compute the forest function \eqref{eq:def-GfunForestBox} directly for the upright double box $T_2$ and obtained
	\begin{equation}
		\GfunForestBox{T_2}(z)
		= \frac{z_3 z_4}{(z_{14}+z_3 + z_4) \BoxPoly^2} \log \frac{(z_{14}+z_3)(z_{14}+z_4)}{z_3 z_4}
		\in \BarIntegralsRegulars[0]\left(S^{\ForestDeltaBoxSymbol}\right)
		\label{eq:GfunForestBox-T2}%
	\end{equation}
	in $\dimension=6$ dimensions with unit indices $\EP_e = 1$. Since all of the occurring polynomials $\set{\BoxPoly, z_{14}+z_3 + z_4,z_{14}+z_3, z_{14}+z_4}$ are mutually compatible in $C^{\ForestDeltaBoxSymbol}$, we immediately conclude that theorem~\ref{theorem:ladderbox-kinematics} extends to all graphs that we can construct from $T_2$ by iteration of the edge additions from figure~\ref{fig:boxladder-addedge}. This includes the original tennis court $T_3$ and the higher loop generalizations $T_n$. As a test we successfully calculated $\FR(T_5)$ in $\dimension=6$ dimensions.
\end{example}
\begin{figure}
	\centering
	$ G=\Graph[0.7]{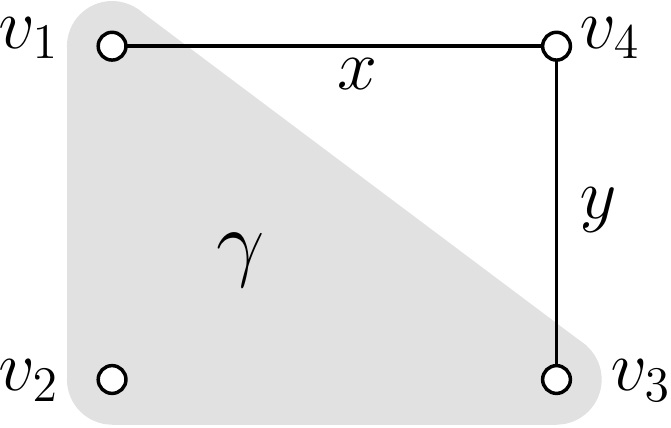} $ \qquad \qquad
	$ G_4=\Graph[0.5]{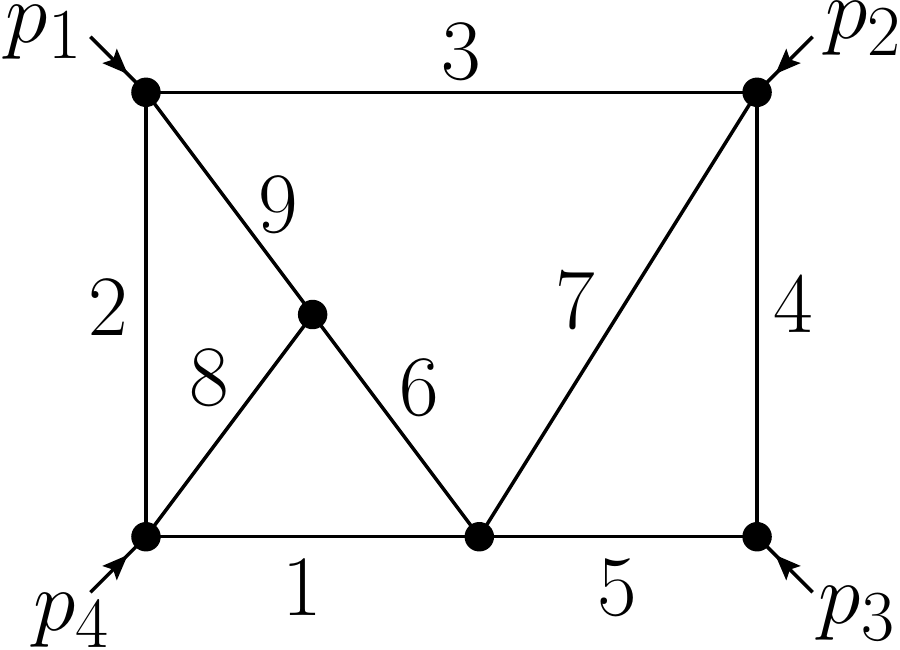} $
	\caption{Construction of a $4$-point function $G$ from a $3$-point function $\gamma$ and a linearly reducible $4$-loop $4$-point integral.}%
	\label{fig:3pt-to-4pt}%
\end{figure}%
\begin{example}\label{ex:vw3-box-reducibility}%
	In our experimental studies, we found several linearly reducible $4$-point functions which are not minors of box ladders. One example is the graph $G_4$ shown in figure~\ref{fig:3pt-to-4pt}, which we computed in \cite{Panzer:DivergencesManyScales} in terms of MZV and HPL. In $\dimension=4-2\varepsilon$ with unit indices $\EP_e=1$, the leading term\footnote{%
	Due to a subdivergence, we used a partial integration \eqref{eq:def-anapartial} to obtain a convergent parametric integral representation.}
$\FR(G_4)=f_{-1}/(s\varepsilon) + \bigo{\varepsilon^0}$ 
evaluates to the harmonic polylogarithms
	\begin{align}
		f_{-1} &=
 - \tfrac{79}{70}\mzv[3]{2}H_{{-1}}
 - \mzv{3}
\left(
15\mzv{2}H_{{-1,-1}}
 - 9\mzv{2}H_{{-1,0}}
 - H_{{-1,-2,-1}}
 + H_{{-1,-1,-2}}
 + 6H_{{-1,-1,0,0}}
\right)
\nonumber\\&\quad
 - 6\mzv[2]{3}H_{{-1}}
 - \tfrac{3}{2}\mzv{5}
\left(
11H_{{-1,-1}}
 - 5H_{{-1,0}}
\right)
 - \tfrac{3}{10}\mzv[2]{2}
\left(
H_{{-1,-2}}
 - 17H_{{-1,-1,0}}
 - 10H_{{-1,-1,-1}}
\right)
\nonumber\\&\quad
 - \mzv{2}
\Big(
 H_{{-1,-2,0,0}}
 - 2H_{{-1,-1,-2,0}}
 + 3H_{{-1,-1,-2,-1}}
 - H_{{-1,-1,-1,0,0}}
 + 6H_{{-1,-1,-3}}
\nonumber\\&\qquad\qquad
 - 3H_{{-1,-2,-1,-1}}
 - 2H_{{-1,-1,0,0,0}}
\Big)
 + H_{{-1,-2,-1,0,0,0}}
 - H_{{-1,-1,-2,-1,0,0}}
\nonumber\\&\quad
 + H_{{-1,-1,-2,0,0,0}}
 - 2H_{{-1,-1,-3,0,0}}
 + H_{{-1,-2,-1,-1,0,0}}

	\end{align}
	where $H_{\vec{n}} \defas H_{\vec{n}}(s/t)$ from \eqref{eq:compressed-Hpl-as-Hlog} with $s=(p_1+p_4)^2$ and $t=(p_1+p_4)^2$.
	We can now combine our propositions~\ref{prop:GfunStarTriangle-reduction} and \ref{prop:GfunForestBox-reduction} to extend our results on linear reducibility such that $G_4$ (and many more additional graphs) are covered.

	Consider a $3$-point graph $\gamma$ and add a fourth external vertex via edges $x = \set{v_4,v_1}$ and $y = \set{v_4,v_3}$ to define a $4$-point graph $G$ as illustrated in figure~\ref{fig:3pt-to-4pt}. Then one can show (with the same methods we used in section~\ref{sec:ladderbox-forestfunctions}) that
	\begin{equation}
		\GfunForestBox{G}(z)
		= \frac{\BoxPoly^4}{z_{12}^2 z_3^3 z_4^3} 
		\left( \frac{\BoxPoly}{z_3} \right)^{\EP_x - 1}
		\left( \frac{\BoxPoly}{z_{12}} \right)^{\EP_y - 1}
		\left( \frac{\BoxPoly^2}{z_{12} z_3 z_4} \right)^{-\dimension/2}
		\int_0^{\infty} \GfunForest{\gamma}\left( \frac{z_{14}\BoxPoly}{z_3 z_4}, u, \frac{\BoxPoly}{z_4} \right)\ \dd u.
		\label{eq:3pt-to-4pt}%
	\end{equation}
	If the inner forest function $\GfunForest{\gamma} \in \BarIntegralsRegulars(S^{\StarSymbol})$ has compatibilities in $C^{\StarSymbol}$, this formula shows via a linear reduction that $\GfunForestBox{G} \in \BarIntegralsRegulars(\set{z_{12},z_{14},z_3,z_4,\BoxPoly,z_{14}+z_3})$. In this case its compatibility graph is contained in $(S^{\ForestDeltaBoxSymbol}, C^{\ForestDeltaBoxSymbol})$ and we can append edges according to figure~\ref{fig:boxladder-addedge} without ever leaving this space of functions. If we apply this construction to $\gamma = \WS{3}^{-}$ (figure~\ref{fig:GfunForest-few}), we obtain first the subgraph $G$ of $G_4$ that consists of the edges $\set{1,2,3,6,7,8,9}$ and can then append $\set{4,5}$ to reach $G_4$.
\end{example}

\appendix
\chapter{Short reference of \HyperInt}
\label{chap:hyperint-reference}%

\section{Options and global variables}
\begin{description}
	\item[\code{\_hyper\_verbosity}]
		(default: $1$) \\
		The higher the value of this integer, the more progress information is printed during calculations. The value zero means no such output at all.

	\item[\code{\_hyper\_verbose\_frequency}]
		(default: $10$) \\
		Sets how often progress output is produced during integration or polynomial reduction.

	\item[\code{\_hyper\_return\_tables}]
		(default: \code{false}) \\
		When \code{true}, \code{integrationStep} returns a table instead of a list. This is useful for huge calculations, because {\Maple} can not work with long lists.
	\item[\code{\_hyper\_check\_divergences}]
		(default: \code{true}) \\
		When active, endpoint singularities at $z\rightarrow 0,\infty$ are detected in the computation of integrals $\int_0^{\infty} f(z)\ \dd z$.

	\item[\code{\_hyper\_abort\_on\_divergence}]
		(default: \code{true}) \\
		This option is useful when divergences are detected erroneously, as happens when periods occur for which no basis is supplied to the program.

	\item[\code{\_hyper\_divergences}]
		\hfill \\
		A table collecting all divergences that were detected.

	\item[\code{\_hyper\_max\_pole\_order}]
		(default: $10$) \\
		Sets the maximum values of $i$ and $j$ in \eqref{eq:check-divergences} for which the functions $F_{i,j}$ are computed to check for potential divergences $F_{i,j} \neq 0$.

	\item[\code{\_hyper\_splitting\_field}]
		(default: $\emptyset$) \\
		This set $R$ of radicals defines the field $\K=\Q(R)$ of constants over which all factorizations are performed.

	\item[\code{\_hyper\_algebraic\_roots}]
		(default: \code{false}) \\
		When \code{true}, all polynomials will be factored linearly by introducing algebraic functions as zeros whenever necessary. Further computations with such irrational letters are not supported. 

	\item[\code{\_hyper\_ignore\_nonlinear\_polynomials}]
		(default: \code{false}) \\
		If set to \code{true}, all non-linear polynomials (that would result in algebraic zeros as letters) will be dropped during integration. This is permissible when linear reducibility is granted.

	\item[\code{\_hyper\_restrict\_singularities}]
		(default: \code{false}) \\
		When \code{true}, the rewriting of $f$ as a hyperlogarithm in $z$ (performed during integration) projects onto the algebra $\HlogAlgebra(\Sigma)$ of letters $\Sigma$ specified by the roots of the set \code{\_hyper\_allowed\_singularities} (default: $\emptyset$) of irreducible polynomials. This can speed up the integration significantly.
\end{description}

\section{{\Maple} functions extended by {\HyperInt}}
\begin{description}
	\item[{\code{convert}$(f, \code{form})$}] with $\code{form} \in \set{\code{Hlog}, \code{Mpl}, \code{HlogRegInf}}$\\
		Rewrites polylogarithms $f$ in terms of hyper- or polylogarithms using lemma~\ref{lemma:Hlog-as-Mpl}. Choosing $\code{form}=\code{HlogRegInf}$ transforms $f$ into the list representation \eqref{eq:maple-internal-reginf}.

	\item[{\code{diff}$\left( f, z \right)$}] \hfill\\
		Computes the partial derivative $\partial_t f$ of hyperlogarithms $\code{Hlog}\left(z(t), w(t)\right)$ and multiple polylogarithms $\code{Mpl}\left( \vec{n},\vec{z}(t) \right)$ that occur in $f$, using \eqref{eq:hlog-total-differential} and \eqref{eq:def-Mpl}. Note that this works completely generally, i.e.\ also when a word $w(t)$ depends on $t$.

	\item[{\code{series}$\left(f, z=0 \right)$}] \hfill\\
		Implements the expansion of $f=\Hyper{w}(z)$ at $z\rightarrow 0$. To expand at different points, use $\code{fibrationBasis}$ first as explained in the manual.
\end{description}

\section{New functions provided by {\HyperInt}}
Note that there are further functions in the package, cf. the manual.
\begin{description}
	\item[{\code{hyperInt}$\left(f, \vec{z}\right)$}] with a list $\vec{z}=[z_1,\ldots,z_r]$ or single $\vec{z} = z_1$\\
		Computes $\int_0^{\infty} \dd z_r \ldots \int_0^{\infty} \dd z_1 f$ from right to left. Any variable can also be given as $z_i=a_i..b_i$ to specify the bounds $\int_{a_i}^{b_i} \dd z_i$ instead.

	\item[{\code{integrationStep}$\left( f, z \right)$}] \hfill\\
		Computes $\int_0^{\infty} f\ \dd z$ for $f$ in the form \eqref{eq:maple-internal-reginf}.

	\item[{\code{fibrationBasis}$\left(f, [z_1,\ldots,z_r], F, S\right)$}] \hfill\\
		Rewrites $f$ as an element of $\HlogAlgebra(\Sigma_1)(z_1) \tp \cdots \tp \HlogAlgebra(\Sigma_r)(z_r)\tp\C$ according to \eqref{eq:barintegrals-hlog-products}. Note that in general this will require algebraic alphabets $\Sigma_i \subset \overline{\C(z_{i+1},\ldots,z_r)}$ and the option \code{\_hyper\_algebraic\_roots = true} (even if in the final result all non-rational letters happen to cancel).\footnote{One can also use \code{\_hyper\_ignore\_nonlinear\_polynomials = true}, provided that one knows that only rational letters will remain.} 
		An optional table $F$ (with indexing function \code{sparsereduced}) may be supplied to store the result, otherwise \code{Hlog}-expressions are returned.
		
		If the optional fourth argument $S$ is supplied, it is assumed to be a table and for each defined key $z_i$ of $S$, the result is projected onto $\HlogAlgebra(\Sigma_i^S)(z_i)$ restricting to letters $\Sigma_i^S \defas \Sigma_i(S[z_i]) = \setexp{\text{zeros of $p(z_i)$}}{p \in S[{z_i}]}$. All words including other letters are dropped in the computation.

	\item[{\code{index/sparsereduced}}] \hfill\\
		This indexing function corresponds to {\Maple}'s \code{sparse}, but entries with value zero are removed from the table. It is used to collect coefficients of hyperlogarithms.

	\item[{\code{forgetAll}$()$}] \hfill\\
		Invalidates cache tables for internal functions and should be called whenever global options were changed.

	\item[{\code{transformWord}$(w, t)$}] \hfill\\
		Given a word $w=[\sigma_1,\ldots,\sigma_n] \in \Sigma^{\times}$ with letters $\Sigma \subset \C(t)$ that depend rationally on $t$, returns a list $[ [w_1, u_1], \ldots]$ of pairs such that
		\begin{equation*}
			\AnaReg{z}{\infty} \Hyper{w}(z)
			=
			\sum_i \Hyper{w_i}(t) \cdot \AnaReg{z}{\infty} \Hyper{u_i}(z)
		\end{equation*}
		following proposition~\ref{prop:reginf-as-hlog}. Each $u_i$ is given in the product form \eqref{eq:maple-internal-reginf}.

	\item[{\code{reglimWord}$(w, t)$}] \hfill \\
		Given a word $w=[\sigma_1,\ldots,\sigma_n] \in \Sigma^{\times}$ with rational letters  $\Sigma \subset \C(t)$ and $\sigma_n \neq 0$, it implements our algorithm from section~\ref{sec:reglim-reginf} and returns a linear combination $u$ of words in the representation \eqref{eq:maple-internal-reginf} such that
		\begin{equation*}
			\AnaReg{t}{0} \AnaReg{z}{\infty} \Hyper{w}(z)
			=
			\AnaReg{z}{\infty} \Hyper{u}(z).
		\end{equation*}

	\item[{\code{integrate}$(f, z)$}] \hfill \\
		Takes a hyperlogarithm $f(z)$ in the form \eqref{eq:maple-internal-hyperlog} and returns a primitive $F$ such that $\partial_z F(z) = f(z)$, which is computed following the proof of lemma~\ref{lemma:hlog-primitives}.

	\item[{\code{cgReduction}$\left( L, \code{todo}, d \right)$}] \hfill \\
		Computes compatibility graphs $L[I_i]=(S_{I_i},C_{I_i})$ (and stores them in the table $L$) for all sets $I_i$ of variables asked for in the list $\code{todo}=[I_1,\ldots]$. This implements the original algorithm presented in \cite{Brown:PeriodsFeynmanIntegrals} and considers only projections where each polynomial is of degree $d$ (default value $d=1$) or less in the reduction variable.
		
		One may also pass a set in the parameter \code{todo} (as opposed to a list). In this case reductions are computed for all sets $K$ that do not contain any element of $\code{todo}$. This is useful when one wants to compute all reductions of a Feynman graph with respect to Schwinger parameters (one would set $\code{todo} = \Kinematics$ to ignore all reductions which involve kinematic invariants).

	\item[{\code{checkIntegrationOrder}$\left( L, \vec{z} \right)$}] \hfill\\
		Tests whether for the order $\vec{z}=[z_1,\ldots]$ all polynomials in the reduction $L$ are linear in the corresponding $z_i$ and prints the number of polynomials.
\end{description}

\subsection{Functions related to Feynman integrals}
\label{sec:function-list-feynman}%
In the following functions, graphs $G=(\vertices,\edges)$ are always assumed to be connected and encoded only by their list $\edges = [e_1,\ldots, e_{\abs{\edges}}]$ of oriented edges $e=[\source(e),\target(e)]$ which are defined by a pair of vertices (the choice of orientation does not matter). All vertices $\vertices = \bigcup_{e\in\edges} e$ must be integers and numbered consecutively such that $\vertices=\set{1,\ldots,\abs{\vertices}}$.
\begin{description}
	\item[{\code{graphPolynomial}$(\edges)$}] \hfill\\
		Computes the first Symanzik polynomial $\psipol$ of the graph with edges $\edges$ using \eqref{eq:graph-polynomials}.

	\item[{\code{forestPolynomial}$(\edges, P)$}] \hfill\\
		Returns the spanning forest polynomial $\forestpolynom{P}$ from definition~\ref{def:forestpolynom} of the graph with edge list $\edges$. The partition $P=\set{P_1,\ldots,P_r}$ of a subset of vertices must consist of pairwise disjoint, non-empty parts $P_i$.

	\item[{\code{secondPolynomial}$(\edges, \ExMom, m)$}] \hfill\\
		Computes the second Symanzik polynomial $\phipol$ for the graph with edges $\edges$ that denote scalar propagators $P_e = k_e^2 + m_e^2$. External momenta $\ExMom(v_i)$ entering at vertex $v_i$ must be passed as a list $\ExMom=[ [v_1,\ExMom(v_1)^2],\ldots]$. The list $m=[m_1^2,\ldots,m_{\abs{\edges}}^2]$ of internal masses is optional. If it is omitted, the massless case $m_1=\cdots=m_{\abs{\edges}}=0$ will be assumed.

	\item[{\code{graphicalFunction}$(\edges,\vertices_{\text{ext}})$}] \hfill\\
		Returns the projective parametric integrand \eqref{eq:position-space-projective-dual} with \eqref{eq:phipol-graphical-function} for a \emph{graphical function} \cite{Schnetz:GraphicalFunctions} in $\dimension=4$ dimensions. The edge list $\edges=[e_1,\ldots,e_{\abs{\edges}}]$ can contain sets $e_i=\set{v_{i,1},v_{i,2}}$ to denote propagators (with $\EP_{e_i}=1$) and lists $e_i=[v_{i,1},v_{i,2}]$ for inverse propagators (in the numerator, that is $\EP_{e_i}=-1$) in compatibility with {\polylogprocedures}\footnote{This {\Maple} program for graphical functions by Oliver Schnetz is described in \cite{Schnetz:GraphicalFunctions} and can be obtained from \cite{Schnetz:ZetaProcedures}.}.

		The external vertices must be specified in the order $\vertices_{\text{ext}} = [v_z,v_0,v_1,v_{\infty}]$, where $v_{\infty}$ is optional. When present, $v_{\infty}$ is first deleted from the graph and the graphical function of the remainder is computed.

	\item[{\code{drawGraph}$(\edges, \ExMom, m, s)$}] \hfill\\
		Draws the graph defined by the edge list $\edges$. The remaining parameters are optional: $\ExMom$ and $m$ are as for \code{secondPolynomial} and highlight the external vertices and massive edges, while $s \in \set{\text{circle}, \text{tree}, \text{bipartite}, \text{spring}, \text{planar}}$ sets the style of the drawing as in \code{GraphTheory[DrawGraph]}.

	\item[{\code{findDivergences}$\left(I, \Kinematics \right)$}] \hfill\\
		For any scaling vector $\ScaleVec$ with $\ScaleVec_e \in\set{-1,0,1}$ the degree $\omega_{\ScaleVec}(I)$ of divergence is computed. The return value is a table indexed by sets $R$ and holds those $\sdd_\ScaleVec(I)$ that are $\leq 0$ when $\varepsilon=0$. Each $R$ contains only variables or their inverses and encodes the vector $\ScaleVec$ through $\ScaleVec_e = \pm 1$ when $\SP_e^{\pm 1} \in R$ and $\ScaleVec_e=0$ otherwise.

		The variables in the set $\Kinematics$ are considered fixed parameters (not to be integrated over), so only sets with $R \cap \setexp{z,z^{-1}}{z \in \Kinematics} = \emptyset$ will be considered.
		\begin{remark}
			This method is only guaranteed to detect divergences completely when $I$ is the parametric integrand of a Feynman integral with Euclidean kinematics (corollary~\ref{corollary:projective-convergence-euclidean}). Otherwise more general scaling vectors can be relevant and an algorithm as presented in \cite{PakSmirnov:GeometricApproachAsymptotic} should be used instead.
		\end{remark}
	\item[{\code{dimregPartial}$\left( I, R, \sdd \right)$}] \hfill\\
		Computes the new integrand $\anapartial{\ScaleVec}(I)$ after a partial integration according to \eqref{eq:def-anapartial}. The scaling vector is specified through the set $R$ which may contain variables and their inverses such that $\rho_e = \pm 1$ if $\SP_e^{\pm 1} \in R$ and $\ScaleVec_e = 0$ otherwise. The degree of divergence must be passed as $\sdd = \sdd_{\ScaleVec}(I)$.
\end{description}

\chapter{Explicit results}\label{chap:results}
Mainly for illustration of the different kind of results we obtained, a few explicit examples are shown in this chapter. A systematic and complete computation of massless propagators up to three loops (with some examples at four loops) was presented in \cite{Panzer:MasslessPropagators} and several multi-scale expansions are demonstrated in \cite{Panzer:DivergencesManyScales}.

\section{Integrals of the Ising class}\label{sec:IsingE}
The first values of the integrals $E_n$ from \eqref{eq:def:IsingE} are
\begin{align}
	E_2 &= 6 - 8\ln 2
,
	\label{eq:IsingE-2}\displaybreak[0]\\
	E_3 &= 32 \ln^2 2-12\mzv{2}-8 \ln 2 +10
,
	\label{eq:IsingE-3}\displaybreak[0]\\
	E_4 &= -\tfrac{256}{3} \ln^{3} 2
-82 \mzv{3}
+\left(
	96 \ln 2
	-44
\right)\mzv{2}
+176 \ln^{2} 2
-24 \ln 2
+22
,
	\label{eq:IsingE-4}\displaybreak[0]\\
	E_5 &= \tfrac{512}{3} \ln^{4} 2
-\tfrac{318}{5}\mzv[2]{2}
-992\mzv{1,-3}
+\left(
	464\mzv{3}
	-40
\right) \ln 2
\nonumber\\&\quad
+\left(
	80 \ln 2
	-124
	-256 \ln^{2} 2
\right)\mzv{2}
-74\mzv{3}
+464 \ln^{2} 2
+42
,
	\label{eq:IsingE-5}\displaybreak[0]\\
	E_6 &= -\tfrac{4096}{15} \ln^{5} 2
+768 \ln^{4} 2
+\left(
	\tfrac{1024}{3}\mzv{2}
	+\tfrac{704}{3}
\right) \ln^{3} 2
+\left(
	384\mzv{3}
	+512\mzv{2}
	+1360
\right) \ln^{2} 2
\nonumber\\&
-\left(
	\tfrac{3216}{5}\mzv[2]{2}
	-\numprint{11520}\mzv{1,-3}
	+2632\mzv{3}
	+272\mzv{2}
	88
\right)\ln 2
+\tfrac{\numprint{53775}}{2}\mzv{5}
\nonumber\\&
+830\mzv[2]{2}
-\left(
	\numprint{13964}\mzv{3}
	+348
\right)\mzv{2}
+\numprint{27904}\mzv{1,1,-3}
-6048\mzv{1,-3}
+134\mzv{3}
+86
,
	\label{eq:IsingE-6}\displaybreak[0]\\
	E_7 &= \numprint{63616}\mzv{1,1,-3}
-\numprint{575488}\mzv{1,1,1,-3}
+\tfrac{\numprint{16384}}{45} \ln^{6} 2
+\tfrac{4096}{15} \ln^{5} 2
+2432 \ln^{4} 2
\nonumber\\&
+\left(
	\tfrac{512}{3}\mzv{2}
	-\tfrac{\numprint{20992}}{3}\mzv{3}
+832\right) \ln^{3} 2
+\left(
	\tfrac{\numprint{69056}}{5}\mzv[2]{2}
	+6400\mzv{3}
	+2336\mzv{2}
	+3280
\right) \ln^{2} 2
\nonumber\\&
+\left(
	\numprint{161760}\mzv{2}\mzv{3}
	-\numprint{340588}\mzv{5}
	-688\mzv{2}
	-9304\mzv[2]{2}
	-168
	-\numprint{312320}\mzv{1,1,-3}
	-\numprint{12472}\mzv{3}
\right)\ln 2
\nonumber\\&
+\big(
	\numprint{19840} \ln 2
	-\numprint{21696}
	-\numprint{64000}\ln^{2} 2
	-8320 \mzv{2}
\big) \mzv{1,-3}
+\numprint{942624}\mzv{1,-5}
-\numprint{32624}\mzv{2}\mzv{3}
\nonumber\\&
+\tfrac{\numprint{149851}}{2}\mzv{5}
+\tfrac{\numprint{4209858}}{35}\mzv[3]{2}
+\tfrac{\numprint{18402}}{5}\mzv[2]{2}
-844\mzv{2}
-\numprint{380881}\mzv[2]{3}
+350\mzv{3}
+170
,
	\label{eq:IsingE-7}\displaybreak[0]\\
	E_8 &= \numprint{12926976} \mzv{1,1,1,1,-3} 
+\big(
	\numprint{211456} \mzv{2}
	+\numprint{1761280} \ln^{2} 2 
	-\numprint{1697792} \ln 2
	+\numprint{192128}
\big)
 \mzv{1,1,-3}
\nonumber\\&
+\Big(
	\numprint{282176} \mzv{3}
	+\left(
		\numprint{40960} \ln 2 
		+\numprint{32128}
	\right)
	\mzv{2}
	-\numprint{294912} \ln^{2} 2
	+\numprint{22656} \ln 2 
	-\numprint{84704}
\nonumber\\&\qquad
	+\tfrac{\numprint{655360}}{3} \ln^{3} 2
\Big)
 \mzv{1,-3}
-\tfrac{\numprint{62466560}}{17} \mzv{1,3,-3}
+\left(
	\numprint{7045120} \ln 2 
	-\numprint{3602432}
\right)
 \mzv{1,1,1,-3}
\nonumber\\&
+\Big(
	(
		\numprint{687888} \ln 2
		-\numprint{818624} \ln^{2} 2 
		-\numprint{62372}
	)
	 \mzv{2}
	+\tfrac{\numprint{77824}}{3} \ln^{4} 2
	-\tfrac{\numprint{8206978}}{17} \mzv[2]{2}
	-\tfrac{\numprint{210176}}{3} \ln^{3} 2
\nonumber\\&\qquad
	+\numprint{17072} \ln^{2} 2
	-\numprint{53064} \ln 2 
	+1790
\Big)
 \mzv{3}
-\tfrac{\numprint{230302165}}{136} \mzv{7}
+\tfrac{\numprint{1493504}}{17} \mzv{1,1,-5}
\nonumber\\&
+\left(
	\tfrac{\numprint{4034546}}{5}
	-\tfrac{\numprint{54575568}}{35} \ln 2
\right)
 \mzv[3]{2}
+\left(
	\numprint{4757064} \ln 2
	-\numprint{2434920}
\right)
 \mzv[2]{3}
+\numprint{6195680} \mzv{1,-5}
\nonumber\\&
+\left(
	-\tfrac{\numprint{1022464}}{15} \ln^{3} 2
	+\tfrac{\numprint{591744}}{5} \ln^{2} 2
	-\tfrac{\numprint{169624}}{5} \ln 2
	+\tfrac{\numprint{76958}}{5}
\right)
 \mzv[2]{2}
+\tfrac{\numprint{340095}}{2} \mzv{5}
+342
\nonumber\\&
+\Big(
	\tfrac{\numprint{33352925}}{17} \mzv{5}
	-\tfrac{\numprint{16384}}{15} \ln^{5} 2
	+\tfrac{\numprint{28672}}{3} \ln^{4} 2
	+256 \ln^{3} 2
	+\numprint{11424} \ln^{2} 2 
	-2960 \ln 2
	-2060
\Big)
 \mzv{2}
\nonumber\\&
-\tfrac{\numprint{131072}}{315} \ln^{7} 2
+\tfrac{\numprint{57344}}{45} \ln^{6} 2
+2048 \ln^{5} 2
+7488 \ln^{4} 2
+2752 \ln^{3} 2
\nonumber\\&
+\left(
	\numprint{1977632} \mzv{5}
	+8080
\right)
 \ln^{2} 2
-\left(
	\numprint{12015360} \mzv{1,-5} 
	+\numprint{1819522} \mzv{5}
	+344
\right)
 \ln 2
.
	\label{eq:IsingE-8}%
\end{align}

\begin{figure}\centering
	\scalebox{1.2}{$\Graph[0.4]{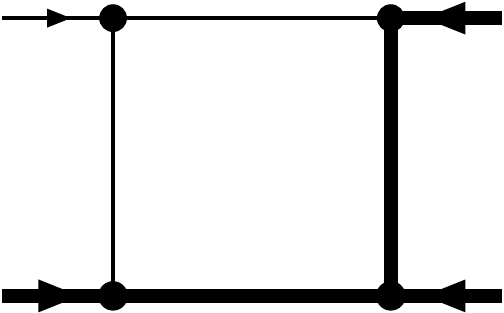}$
		\quad
		$\Graph[0.4]{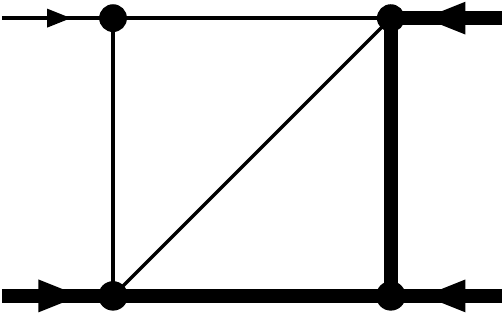}$
		\quad
		$\Graph[0.4]{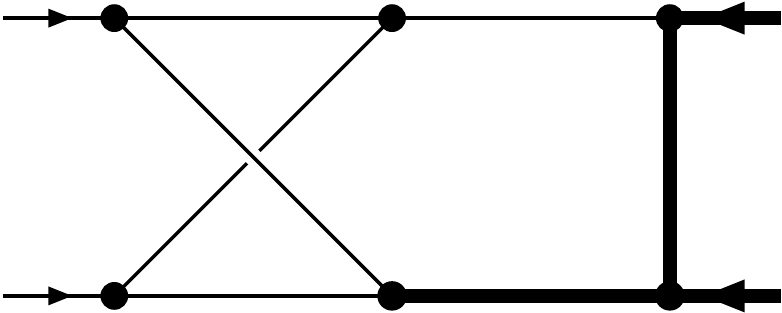}$
		\quad
		$\Graph[0.4]{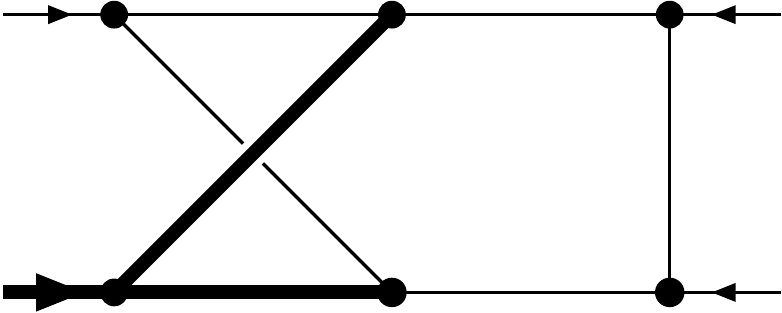}$}
		\\
		\scalebox{1.2}{$\Graph[0.3]{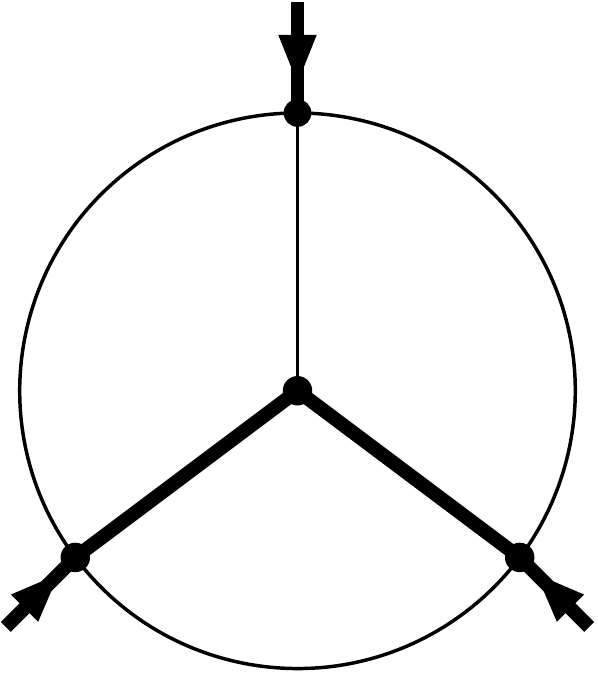}$
		\quad
		$\Graph[0.3]{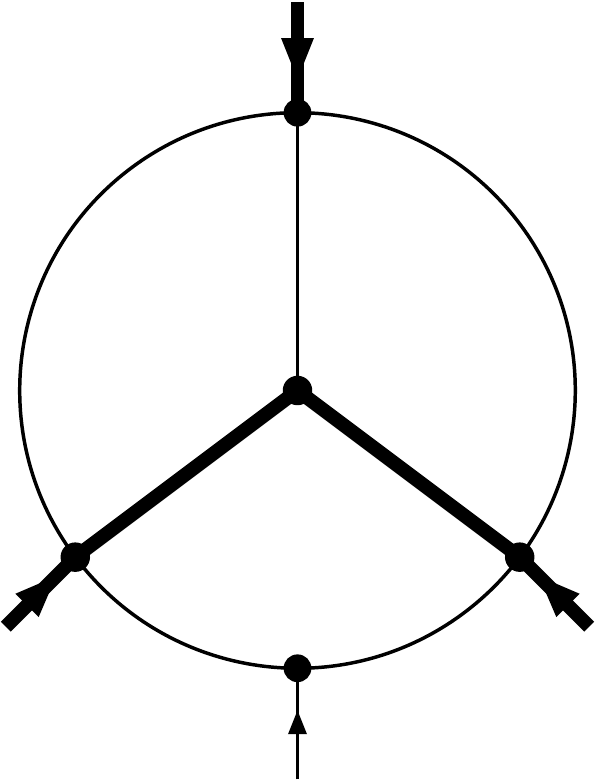}$
		\quad
		$\Graph[0.33]{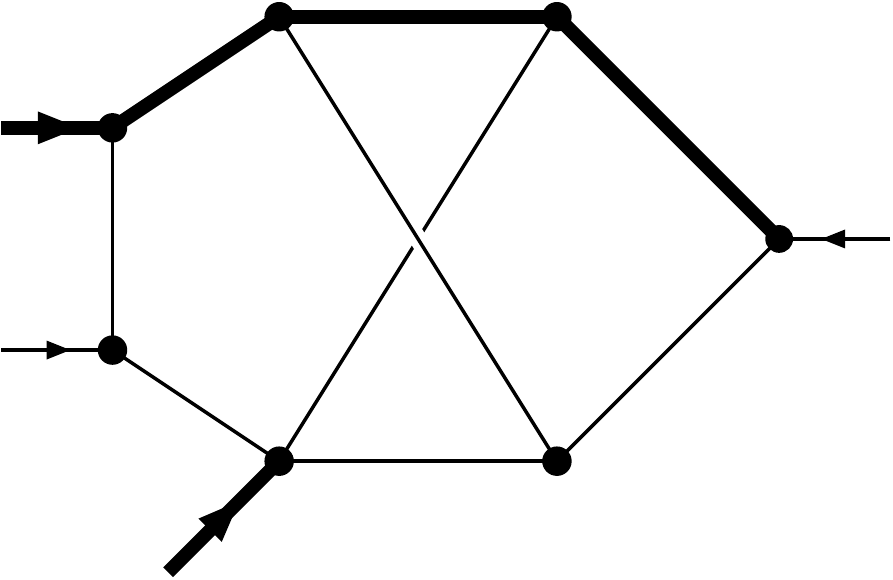}$
		\quad
		$\Graph[0.3]{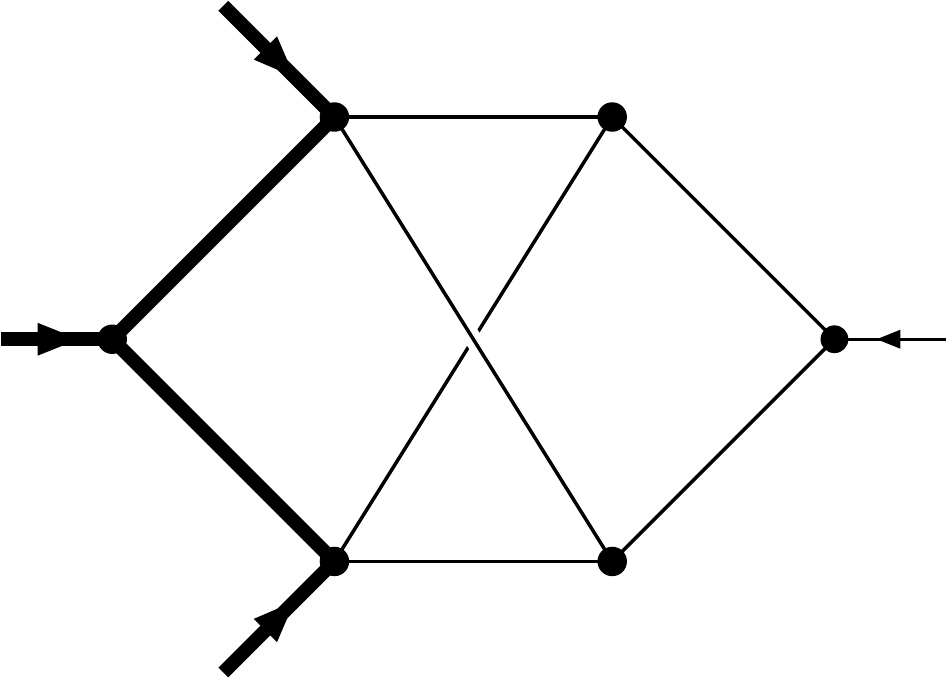}$}
		\caption[Linearly reducible graphs with masses]{%
Examples of linearly reducible graphs with some massive internal and off-shell external momenta (thick edges).}%
	\label{fig:massive}%
\end{figure}
\section{A massive 2-loop 6-scale integral}\label{sec:massive-integrals}
We found several Feynman integrals with massive internal propagators that are linearly reducible (in Schwinger parameters). Some examples are shown in figure~\ref{fig:massive} and a few explicit results were computed \cite{Panzer:DivergencesManyScales}, like the crossed box (or double-triangle)
\begin{equation}%
	\label{eq:diagbox-2mass2off-expansion}%
	\FR\left( \Graph[0.4]{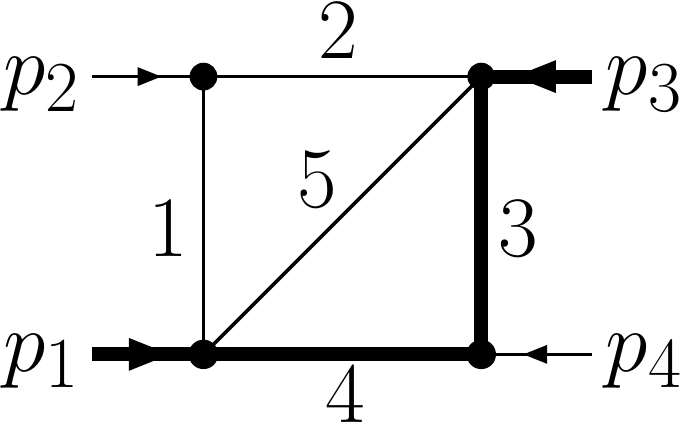} \right)
	= \frac{\Gamma(1+2\varepsilon)}{(p+q-s-u)m_3^{2+4\varepsilon}}
		\sum_{n=-1}^{\infty} f_n(p,s,u,q,m) \cdot \varepsilon^n
\end{equation}
in $\dimension=4-2\varepsilon$ dimensions with unit indices $\EP_e=1$.
It depends on the masses of propagators $3$ and $4$, the Mandelstam invariants and the off-shell momenta $p_3$ and $p_4$ (we assume $p_2^2 = p_4^2 = 0$). We scaled out $m_3^2$ and introduced the dimensionless variables
\begin{equation}%
	\label{eq:diagbox-2mass2off-variables}%
	s \defas \frac{(p_1+p_2)^2}{m_3^2},
	\quad
	u \defas \frac{(p_1+p_4)^2}{m_3^2},
	\quad
	p \defas \frac{p_1^2}{m_3^2},
	\quad
	q \defas \frac{p_3^2}{m_3^2}\
	\quad\text{and}\quad
	m \defas \frac{m_4^2}{m_3^2}.
\end{equation}
Note that \eqref{eq:diagbox-2mass2off-expansion} has a pole in $\varepsilon$ which reflects the infrared subdivergence $\gamma = \set{3,4,5}$. To compute the $\varepsilon$-expansion, we applied the corresponding operator \eqref{eq:def-anapartial} to obtain a convergent parametric integral that we expand before the integration.

The polynomial reduction shows linear reducibility leaves the singularities
\begin{align}\label{eq:diagbox-2mass2off-symbol}
	S_{\set{3,4,5,2}}
	=
	\Big\{
	&
		1-m,
		p+m,
		p-s,
		p-u,
		1+q,
		q-s,
		s+m,
		q-u,
		1+u,
		pq-us,
		s-qm,
	\nonumber\\&
		p-um,
		1-p-m+u,
		p-s-u+q,
		p-s+qm-um,
		s-pq-qm+us,
	\nonumber\\&
		1-s-m+q,
		p-us-um+pq,
		pq+p-us-s+qm-um
	\Big\}
\end{align}
which are linear in each variable, hence we can express the coefficient functions $f_n$ in terms of hyperlogarithms with rational letters. We choose the base point at $0\ll p \ll s \ll u \ll q \ll m$ and abbreviate $S_w \defas \Hyper{w} (s)$, $U_w \defas \Hyper{w} (u)$, $M_w \defas \Hyper{w} (m)$, $P_w \defas \Hyper{w} (p)$ and $Q_w \defas \Hyper{w} (q)$. Then the leading term  evaluates to
\begin{align}%
	f_{-1} &=
M_0 \left(
	Q_{{0,-1+m}}
	-P_{{u}}U_{{-1+m}}
	+S_{{q}}Q_{{-1+m}}
	-U_{{0,-1+m}}
\right)
-S_{{m \left( -u+q \right) ,qm,-m}}
+P_{{{\frac {us}{q}},s}}S_{{-m}}
\nonumber\\&\quad
-S_{{-m}}P_{{s+um-qm,s}}
+P_{{u}}U_{{-1+m,-1}}
-S_{{q}}Q_{{-1+m,-1}}
+P_{{u,-m+u+1}} \left( U_{{-1}}-M_{{0}} \right)
\nonumber\\&\quad
+U_{{-1}} \left(
	P_{{{\frac {us}{q}},um}}
	-P_{{s+um-qm,um}}
	-P_{{u,um}}
\right)
+ Q_{{-1}} \left(
	S_{{m \left( -u+q \right) ,qm}}
	-S_{{0,qm}}
	+S_{{q,qm}}
\right)
\nonumber\\&\quad
+P_{{s+um-qm,s,-m}}
-P_{{{\frac {us}{q}},um,-m}}
+S_{{0,0,-m}}
-P_{{{\frac {us}{q}},s,-m}}
+P_{{{\frac {us}{q}},0,-m}}-S_{{0,{\frac {m \left( -u+q \right) }{u+1}},-m}}
\nonumber\\&\quad
-S_{{q,qm,-m}}
+S_{{q,-m+q+1,-m}}
-P_{{u,-m+u+1,-m}}
+S_{{0,qm,-m}}
+P_{{s+um-qm,um,-m}}
\nonumber\\&\quad
+P_{{{\frac {us}{q}},-{\frac {-us-s-um+qm}{q+1}},-m}}
+U_{{0,-1+m,-1}}
+S_{{q,0,-m}}
-P_{{u,0,-m}}
-P_{{s,0,-m}}
-Q_{{0,-1+m,-1}}
\nonumber\\&\quad
-P_{{s+um-qm,-{\frac {-us-s-um+qm}{q+1}},-m}}
+S_{{m \left( -u+q \right),{\frac {m \left( -u+q \right) }{u+1}},-m}}
-S_{{m \left( -u+q \right) ,0,-m}}
+P_{{u,um,-m}}
\nonumber\\&\quad
+ \left( U_{{-1}}-Q_{{-1}} \right)\left(
	S_{{m \left( -u+q \right) ,{\frac {m \left( -u+q \right) }{u+1}}}} 
	-S_{{0,{\frac {m \left( -u+q \right) }{u+1}}}}
\right)
+S_{{m \left( -u+q \right) }} \left( -U_{{0,-1}}+Q_{{0,-1}} \right)
\nonumber\\&\quad
+\left(S_{{-m}}+U_{{-1}}-Q_{{-1}} \right) \left[
		P_{{s+um-qm,-{\frac {-us-s-um+qm}{q+1}}}} 
		-P_{{{\frac {us}{q}},-{\frac {-us-s-um+qm}{q+1}}}}
	\right]
+P_{{s}}S_{{0,-m}}
\nonumber\\&\quad
+\left( P_{{s+um-qm}} -P_{{{\frac {us}{q}}}} \right) \left[
	S_{{{\frac {m \left( -u+q \right) }{u+1}}}}
	\left( U_{{-1}}-Q_{{-1}} \right)
	-S_{{{\frac {m \left( -u+q \right) }{u+1}},-m}} 
	+S_{{qm}}Q_{{-1}}
	-S_{{qm,-m}}
\right]
\nonumber\\&\quad
+S_{{q,-m+q+1}} \left( -Q_{{-1}}+M_{{0}} \right) 
+P_{{s+um-qm}} 
\left( 
	Q_{{0,-1}}
	-S_{{0,-m}}
	-U_{{0,-1}}
\right) 
	\label{eq:diagbox-2mass2off-1}%
\end{align}
and the results for $f_0$, $f_1$ and $f_2$ are provided in \cite{Panzer:DivergencesManyScales}. To our knowledge, this is the first higher order calculation of a two-loop integral that involves $2$ masses and as many as $6$ kinematic scales in total.

\section{The 4-loop ladder box}\label{sec:lbox4-onshell}
In $\dimension=6$, the $4$-loop ladder box $B_4$ evaluates on-shell ($p_1^2=\cdots=p_4^2=0$) to harmonic polylogarithms $H_{\vec{n}} \defas H_{\vec{n}}(x)$ of the ratio $x=t/s$ of Mandelstam invariants $t = (p_1+p_4)^2$ and $s=(p_1+p_2)^2$. In the notation \eqref{eq:compressed-Hpl-as-Hlog} we find
\begin{align}
\FR(B_4) =&{} \frac{A}{s+t} + \frac{B}{t} \qquad\text{where}
\\
A=&{}
-\tfrac{8}{5} \left(
24 H_{-2}
-30 H_{-1}
+5 H_{0}
-18 H_{-2,-1}
+3 H_{-2,0}
-18 H_{-1,-2}
-6 \mzv{3}
+18\right)
 \mzv[2]{2}
\nonumber\\&{}
+ 6 H_{-2,0,0}
-\tfrac{248}{35} \left(
3 H_{-1}
-4\right)
 \mzv[3]{2}
-8 \left(
3 H_{-1}
-4\right)
 \mzv[2]{3}
-6 H_{-1,0,0}
+32 H_{-3,-1,0,0}
\nonumber\\&{}
+8 H_{-2,-2,0,0}
-10 H_{-2,-1,0,0}
-10 H_{-1,-2,0,0}
-6 H_{-2,-2,-1,0,0}
-6 H_{-2,-1,-2,0,0}
\nonumber\\&{}
+2 \Big(
9 H_{-2}
-9 H_{-1}
+3 H_{0}
+48 H_{-3,-1}
-16 H_{-3,0}
+12 H_{-2,-2}
-15 H_{-2,-1}
\nonumber\\&{}\qquad
+5 H_{-2,0}
-15 H_{-1,-2}
-6 H_{0,0}
-9 H_{-2,-2,-1}
+3 H_{-2,-2,0}
-9 H_{-2,-1,-2}
\nonumber\\&{}\qquad
-8 H_{-2,0,0}
-36 H_{-1,-3,-1}
+12 H_{-1,-3,0}
-9 H_{-1,-2,-2}
+10 H_{-1,0,0}
-24 \mzv{5}
\nonumber\\&{}\qquad
+6 H_{-2,-1,0,0}
+6 H_{-1,-2,0,0}
\Big)
 \mzv{2}
-12 \left(
3 H_{-2}
+5\right)
 \mzv{5}
-6 H_{-1,-2,-2,0,0}
\nonumber\\&{}
+2 \Big(
16 H_{-3}
-5 H_{-2}
-6 H_{0}
-3 H_{-2,-2}
-8 H_{-2,0}
-12 H_{-1,-3}
+10 H_{-1,0}
-3
\nonumber\\&{}\qquad
+6 H_{-2,-1,0}
+6 H_{-1,-2,0}
\Big)
 \mzv{3}
-84 \mzv{7}
-24 H_{-1,-3,-1,0,0}
\qquad\text{and}
\\ B =&{}
-\tfrac{8}{5} \Big(
2 H_{-1} \mzv{3}
-18 H_{-1}
+30 H_{-1,-1}
-5 H_{-1,0}
+6 H_{-1,-2,-1}
-H_{-1,-2,0}
 \nonumber\\&{}\qquad
+6 H_{-1,-1,-2}\Big)
 \mzv[2]{2}
+\tfrac{248}{35} H_{-1,-1} \mzv[3]{2}
+8 H_{-1,-1} \mzv[2]{3}
+28 H_{-1} \mzv{7}
-24 H_{-2,-1,0,0}
 \nonumber\\&{}
+18 H_{-1,-1,0,0}
+10 H_{-1,-2,-1,0,0}
+10 H_{-1,-1,-2,0,0}
+2 H_{-1,-2,-2,-1,0,0}
\nonumber\\&{}
+2 \Big(
8 H_{-1} \mzv{5}
-36 H_{-2,-1}
+12 H_{-2,0}
-9 H_{-1,-2}
+27 H_{-1,-1}
-9 H_{-1,0}
\nonumber\\&{}\qquad
+15 H_{-1,-2,-1}
-5 H_{-1,-2,0}
+15 H_{-1,-1,-2}
+6 H_{-1,0,0}
+3 H_{-1,-2,-2,-1}
\nonumber\\&{}\qquad
-H_{-1,-2,-2,0}
+3 H_{-1,-2,-1,-2}
+12 H_{-1,-1,-3,-1}
-4 H_{-1,-1,-3,0}
\nonumber\\&{}\qquad
+3 H_{-1,-1,-2,-2}
-10 H_{-1,-1,0,0}
-2 H_{-1,-2,-1,0,0}
-2 H_{-1,-1,-2,0,0}\Big)
 \mzv{2}
 \nonumber\\&{}
-2 \Big(
12 H_{-2}
-9 H_{-1}
-5 H_{-1,-2}
-6 H_{-1,0}
-H_{-1,-2,-2}
-4 H_{-1,-1,-3}
+10 H_{-1,-1,0}
\nonumber\\&{}\qquad
+2 H_{-1,-2,-1,0}
+2 H_{-1,-1,-2,0}\Big)
 \mzv{3}
+12 \left(
5 H_{-1}
+H_{-1,-2}\right)
 \mzv{5}
-6 H_{-1,-2,0,0}
 \nonumber\\&{}
+2 H_{-1,-1,-2,-2,0,0}
+2 H_{-1,-2,-1,-2,0,0}
+8 H_{-1,-1,-3,-1,0,0}
.
\end{align}

\chapter{Erratum to Lewin}\label{chap:LewinTypos}
Plenty of functional and integral equations of polylogarithms, taken from the excellent books \cite{Lewin:PolylogarithmsAssociatedFunctions,Lewin:StructuralPropertiesPolylogarithms}, were used as checks for our program {\HyperInt}. These tests revealed a very few misprints in \cite{Lewin:PolylogarithmsAssociatedFunctions}. Because this work is still frequently being referred to, we list our corrections here:
\begin{itemize}
	\item
	Equation (7.93): $-\frac{9}{4}\pi^2 \log^2(\xi)$ must be $-\frac{9}{12}\pi^2 \log^2(\xi)$.

	\item
	Equation (7.99), repeated as (44) in appendix~A.2.7: The second term $-\frac{9}{4}\pi^2 \log^3(\xi)$ of the last line must be replaced with $-\frac{3}{4}\pi^2 \log^3(\xi)$.

	\item
	Equation A.3.5. (9): The terms $-2 \Li_3\left( 1/x \right) + 2\Li_3(1)$ should read $+\Li_3(1/x) - \Li_3(1)$ instead.

	\item
	In equation (7.132), a factor $\frac{1}{2}$ in front of the second summand $D^n_{p=0} \frac{1}{p} \left\{ \cdots \right\}$ is missing (it is correctly given in 7.131).

	\item
	Equation (8.80): $(1-v)$ inside the argument of the fourth $\Li_2$-summand must be replaced by $(1+v)$, so that after including the corrections mentioned in the following paragraph, the correct identity reads
	\begin{equation}\begin{split}
		0=&{}
		\Li_2\left( \frac{(1+v)w}{1+w} \right)
		+ \Li_2\left( \frac{-(1-v)w}{1-w} \right)
		- \Li_2\left( \frac{-(1-v^2)w^2}{1-w^2} \right)
		\\
		+ &{} \Li_2\left( \frac{(1-v)w}{1+w} \right)
		+ \Li_2\left( \frac{-(1+v)w}{1-w} \right)
		+ \frac{1}{2} \log^2\left( \frac{1+w}{1-w} \right).
		\label{eq:lewin-8.80-corrected}%
	\end{split}\end{equation}

	\item
	Equation (16.46) of \cite{Lewin:StructuralPropertiesPolylogarithms}: $x^2$ must read $x^{-2}$.

	\item
	Equation (16.57) of \cite{Lewin:StructuralPropertiesPolylogarithms}: $\frac{\pi^4}{40}$ must read $\frac{\pi^4}{30}$.
\end{itemize}

\bibliography{../bib/qft}
\nomenclature[S]{$S^{\StarSymbol}, S^{\TriangleSymbol}$}{singularities of star and triangle functions, proposition~\ref{prop:GfunStarTriangle-reduction}, page~\pageref{prop:GfunStarTriangle-reduction}}%
\nomenclature[S C]{$S^{\ForestDeltaBoxSymbol}$}{singularities of ladder box forest functions, figure~\ref{fig:cg-box-ladder}, page~\pageref{fig:cg-box-ladder}}%
\nomenclature[Phi P]{$\forestpolynom[G]{P}, \forestpolynomDual[G]{P}$}{spanning forest polynomial and its dual, equation~\eqref{eq:def-forestpolynom}, page~\pageref{eq:def-forestpolynom}}%
\nomenclature[reg]{$\WordReg{A}{B}(w)$}{shuffle regularization of a word $w$, definition~\ref{def:shuffle-regularization}, page~\pageref{def:shuffle-regularization}}%
\nomenclature[reg]{$\WordReg{}{\infty}(w)$}{regularization of a word $w$ at infinity, definition~\ref{def:reginf-word}, page~\pageref{def:reginf-word}}%
\nomenclature[Li]{$\Li_{\vec{n}}(\vec{z})$}{multiple polylogarithms, equation~\eqref{eq:def-MPL-intro}, page~\pageref{eq:def-MPL-intro}}%
\nomenclature[zeta]{$\mzv{n_1,\ldots,n_r}$}{multiple zeta value, equation~\eqref{eq:def-MZV-intro}, page~\pageref{eq:def-MZV-intro}}%
\nomenclature[ZZ]{$\ZZ{n}$}{zigzag graph with $n$ loops, figure~\ref{fig:zigzags}, page~\pageref{fig:zigzags}}%
\nomenclature[WS]{$\WS{n}$}{wheel with $n$ spokes graph, figure~\ref{fig:ws3-ladder-series}, page~\pageref{fig:ws3-ladder-series} and figure~\ref{fig:4loop-vacuum}, page~\pageref{fig:4loop-vacuum}}%
\nomenclature[Z]{$\MZV[N]$}{$\Q$-algebra spanned by MPL at $N$'th roots of unity, definition~\ref{def:MZV-N}, page~\pageref{def:MZV-N}}%
\nomenclature[BbO(S)]{$\BarIntegralsRegulars[b](S)$}{rational linear combinations of iterated integrals with singularities in $S$, page~\pageref{inline:barregulars}}%
\nomenclature[S gamma w]{$\int_{\gamma} w$}{iterated integral of the word $w$ along the path $\gamma$, equation~\eqref{eq:def-II}, page~\pageref{eq:def-II}}%
\nomenclature[C]{$\RSphere$}{Riemann sphere $\RSphere = \C \cup \set{\infty}$}%
\nomenclature[Lw z]{$\Hyper{w}(z)$}{hyperlogarithm of $z$ (indexed by a word $w$), definition~\ref{def:Hlog}, page~\pageref{def:Hlog} or equation~\eqref{eq:def-Hlog-noshuffle}, page~\pageref{eq:def-Hlog-noshuffle}}%
\nomenclature[im]{$\im$}{image of a map}%
\nomenclature[Reg z t fz]{$\displaystyle\AnaReg{z}{\tau} f(z)$}{regularized limit of the function $f(z)$, definition~\ref{def:reglim}, page~\pageref{def:reglim}}%
\nomenclature[Phi f]{$\WordTransformation{f}$}{action of $f \in \Aut(\RSphere)$ on words, equation~\eqref{eq:Moebius-transformation}, page~\pageref{eq:Moebius-transformation}}%
\nomenclature[H +-]{$\Halfplane^{\pm}$}{upper and lower half-planes of $\C$}%
\nomenclature[Im]{$\Imaginaerteil$}{imaginary part}%
\nomenclature[Re]{$\Realteil$}{real part}%
\nomenclature[O]{$\regulars(\Sigma)$}{regular functions on $\C\setminus \Sigma$, definition~\ref{def:punctured-regulars}, page~\pageref{def:punctured-regulars}}%
\nomenclature[f star g]{$f\convolution g$}{convolution product of functions on $T(\Sigma)$, equation~\eqref{eq:def-convolution-product}, page~\pageref{eq:def-convolution-product}}%
\nomenclature[gamma star eta]{$\gamma \concat \eta$}{concatenation of paths from $\gamma(0)$ to $\eta(1)$}%
\nomenclature[Delta]{$\Delta(w)$}{deconcatenation coproduct of words, equation~\eqref{eq:def-deconcatenation-coproduct}, page~\pageref{eq:def-deconcatenation-coproduct}}%
\nomenclature[Delta]{$\cop(G)$}{Connes-Kreimer coproduct of Feynman graphs, equation~\eqref{eq:def-FeynCop}, page~\pageref{eq:def-FeynCop}}%
\nomenclature[Lyn Sigma]{$\Lyndons(\Sigma)$}{Lyndon words over the alphabet $\Sigma$, equation~\eqref{eq:def-lyndon-words}, page~\pageref{eq:def-lyndon-words}}%
\nomenclature[gamma eta]{$\gamma,\eta$}{piecewise differentiable paths $[0,1]\longrightarrow \C$}%
\nomenclature[P(G)]{$\period(G)$}{the period of a graph, definition~\ref{def:period}, page~\pageref{def:period}}%
\nomenclature[omega,omega(gamma)]{$\sdd, \sdd(\gamma)$}{superficial degree of divergence of $G$ or a subgraph $\gamma$, equation~\eqref{eq:def-sdd}, page~\pageref{eq:def-sdd}}%
\nomenclature[L(a1,a2)]{$\onemaster{\EP_1}{\EP_2}$}{one-loop master integral, equation~\eqref{eq:oneloop-master}, page~\pageref{eq:oneloop-master}}%
\nomenclature[Omega]{$\Omega$}{projective volume form, equation~\eqref{eq:projective-delta-form-integrand}, page~\pageref{eq:projective-delta-form-integrand}}%
\nomenclature[I G]{$I_{G}$}{parametric integrand of the graph $G$, equation~\eqref{eq:projective-delta-form-integrand}, page~\pageref{eq:projective-delta-form-integrand}}%
\nomenclature[psi,phi]{$\psipol,\phipol$}{Symanzik polynomials, equation~\eqref{eq:graph-polynomials-combinatorial}, page~\pageref{eq:graph-polynomials-combinatorial}}%
\nomenclature[vw(G)]{$\vw(G)$}{vertex-width of $G$, definition~\ref{def:vw}, page~\pageref{def:vw}}%
\nomenclature[f w,tau k (z)]{$f_{w,\tau}^{(k)}(z)$}{expansion coefficients of $\Hyper{w}(z)$ at $t$, equation~\eqref{eq:hlog-divergences}, page~\pageref{eq:hlog-divergences}}%
\nomenclature[h1G]{$\loops{G}$}{the number of independent loops of a graph $G$, equation~\eqref{eq:loop-number}, page~\pageref{eq:loop-number}}%

\printnomenclature[1.3cm]

\end{document}